\documentclass[letterpaper,11pt]{article}

\usepackage[T1]{fontenc}
\usepackage[utf8]{inputenc}
\usepackage[english]{babel}
\usepackage{csquotes}

\usepackage[margin=1in,centering]{geometry}

\usepackage[final]{microtype}

\usepackage{amsmath,amsthm,amssymb}
\numberwithin{equation}{section}

\usepackage{mathrsfs}

\usepackage{mathtools}

\usepackage[overload]{keytheorems}
\keytheoremset{restate-counters={section}}

\usepackage{mathdots}

\usepackage{cochineal}%

\usepackage{bm}

\usepackage{pifont}

\usepackage{soul}

\usepackage[page,title,titletoc]{appendix}

\usepackage[final]{graphicx}
\counterwithin{figure}{section}
\counterwithin{table}{section}
\makeatletter \let\c@table\c@figure \makeatother

\usepackage{tikz}
\usetikzlibrary{calligraphy}

\usepackage[svgnames]{xcolor}

\usepackage{colortbl}
\usepackage{booktabs}

\usepackage[subrefformat=parens]{subcaption}

\usepackage[strict]{changepage}

\usepackage[%
backend=biber,%
style=ext-alphabetic,%
innamebeforetitle=true,%
natbib=true,%
maxnames=10,%
sorting=anyt,%
giveninits=true,%
]{biblatex}
\DeclareFieldFormat[article,periodical]{number}{\mkbibparens{#1}}

\renewbibmacro*{edition}{\iffieldundef{edition}{}{\printfield{edition}\iffieldint{edition}{}{ ed}}}
\usepackage{authblk}

\usepackage{xspace}

\usepackage{multirow}

\usepackage{enumitem}

\usepackage[linesnumbered,vlined,shortend,ruled]{algorithm2e}
\makeatletter
\renewcommand{\nllabel}[1]
 {{\let\@currentlabel\algocf@currentlabel
  \let\@currentcounter\algocf@currentcounter
  \label{#1}}}%

\renewcommand{\algocf@nl@sethref}[1]{%
  \renewcommand{\theHAlgoLine}{\thealgocfproc.#1}%
  \hyper@refstepcounter{AlgoLine}%
  \gdef\algocf@currentlabel{#1}%
  \gdef\algocf@currentcounter{AlgoLine}%
 }%
\makeatother

\DontPrintSemicolon
\SetAlgoHangIndent{\algoskipindent}

\newcommand*{\AssertKeywordText}{verify}
\SetKw{Assert}{\AssertKeywordText}
\SetKwBlock{AssertBlockHelper}{\AssertKeywordText{}}{}

\newcommand*{\RefuteKeywordText}{verify not}
\SetKw{Refute}{\RefuteKeywordText}
\SetKwBlock{RefuteBlockHelper}{\RefuteKeywordText{}}{}

\SetKwInput{KwAssumption}{Assumption}
\SetKw{Continue}{continue}
\SetKw{Guess}{guess}
\SetKw{ParCalls}{parallel-oracle-calls:}

\SetKwProg{FuncHelper}{Function}{:}{end}
\newcommand{\Func}[2]{\SetAlgoNoEnd\SetAlgoNoLine{}\FuncHelper{#1}{\SetAlgoShortEnd\SetAlgoVlined{}#2}}
\SetKwProg{ProcHelper}{Procedure}{:}{end}
\newcommand{\Proc}[2]{\SetAlgoNoEnd\SetAlgoNoLine{}\ProcHelper{#1}{\SetAlgoShortEnd\SetAlgoVlined{}#2}}

\let\oldnl\nl%
\newcommand*{\nonl}{\renewcommand{\nl}{\let\nl\oldnl}}%

\SetKwFunction{AlgPow}{Pow}
\SetKwFunction{AlgFixedPolEval}{$p$}

\usepackage{orcidlink}

\usepackage{hyperref}
\hypersetup{colorlinks=false,unicode,final}
\hypersetup{%
pdfauthor={Enrico Malizia},%
pdftitle={Hausdorff Reductions and the Exponential Hierarchies},%
pdfsubject={Structural Complexity Theory},%
pdfkeywords={Computational Complexity, Turing Machines, Oracle Machines, Hausdorff Reductions, Hausdorff Complexity Classes, Oracle Complexity classes, Iterated Exponential Hierarchies, Iterated Exponential Meta-Hierarchy, intermediate hierarchy levels, Strong Exponential Hierarchy, Weak Exponential Hierarchy, Boolean Hierarchies, Hard Problems, Second Order Quantified Boolean Formulas, Datalog+/-, TGDs, Explanations, Cardinality-Minimal Explanations}%
}

\usepackage{bookmark} %

\makeatletter
\newcommand{\silentbookmark}[2][1]{%
  \bookmark[level=#1, dest=\@currentHref]{#2}%
}
\makeatother

\makeatletter
\newcommand*{\citelinktext}[2]{%
  \nocite{#1}%
  \hyper@@link[cite]{}{cite.\the\c@refsection @#1}{#2}%
}
\makeatother

\usepackage{zref-clever}
\zcsetup{countertype = {subproperty = subpropertytype}}
\zcRefTypeSetup{subpropertytype}{%
  Name-sg = Property,%
  name-sg = property,%
  Name-pl = Properties,%
  name-pl = properties,%
}
\zcsetup{countertype = {algocf = algorithm, AlgoLine = line}}
\zcLanguageSetup{english}{%
  cap = true,%
  abbrev = false,%
  type = equation,%
    abbrev = true,%
  type = line,%
    cap = false,%
}

\providecommand{\texorpdfstring}[2]{#1} %

\makeatletter
\def\mycopyright#1{%
    \protected@xdef \@thanks {\@thanks \protect \footnotetext [\the \c@footnote ]{#1}}%
}
\makeatother

\newcommand{\nocontentsline}[3]{}
\newcommand\stoptoc{\let\origcontentsline\addcontentsline\let\addcontentsline\nocontentsline}
\newcommand\resumetoc{\let\addcontentsline\origcontentsline}

\def\nbdash-{\nobreakdash-\hspace{0pt}}
\newcommand*{\hashsep}{{\#}}
\newcommand*{\dollarsep}{{\$}}
\newcommand*{\Wlog}{w.l.o.g.\@\xspace}
\newcommand*{\Wrt}{w.r.t.\@\xspace}
\newcommand*{\HausdSeqRequirementI}{h1\xspace}
\newcommand*{\HausdSeqRequirementII}{h2\xspace}

\newcommand*{\cmark}{\textnormal{\ding{51}}}
\newcommand*{\xmark}{\textnormal{\ding{55}}}

\newcommand*{\oawarelegal}{oracle-aware\xspace}
\newcommand*{\oawarelegalemph}{oracle-\emph{aware}\xspace}

\newcommand*{\ounawarelegal}{oracle-unaware\xspace}
\newcommand*{\ounawarelegalemph}{oracle-\emph{un}\/aware\xspace}
\newcommand*{\ounawarelegalemphlegal}{oracle-unaware\xspace}
\newcommand*{\ounawarelegalemphboth}{oracle-\emph{un}\/aware\xspace}
\newcommand*{\obothunawarelegal}{oracle-(un)aware\xspace}
\newcommand*{\DatalogPM}{Datalog\textsuperscript{$\pm$}\xspace}
\newcommand*{\CNF}{\textrm{CNF}\xspace}
\newcommand*{\DNF}{\textrm{DNF}\xspace}
\newcommand*{\Minex}{MinEx\xspace}
\newcommand*{\Minexs}{MinExes\xspace}

\newcommand*{\SigmaKFormulaPrefix}[1]{$\ol{\Sigma}_{#1}$}
\newcommand*{\SigmaKFormula}[1]{\SigmaKFormulaPrefix{#1}\nbdash-formula\xspace}
\newcommand*{\SigmaKFormulas}[1]{\SigmaKFormulaPrefix{#1}\nbdash-formulas\xspace}

\newcommand*{\SigmaKSentences}[1]{\SigmaKFormulaPrefix{#1}\nbdash-sentences\xspace}

\newcommand*{\SimpleSigmaKFormulaPrefix}[1]{$\Sigma_{#1}$}
\newcommand*{\SimpleSigmaKFormula}[1]{\SimpleSigmaKFormulaPrefix{#1}\nbdash-formula\xspace}
\newcommand*{\SimpleSigmaKFormulas}[1]{\SimpleSigmaKFormulaPrefix{#1}\nbdash-formulas\xspace}

\newcommand*{\SigmaKEvenFormulaPrefix}[1]{$\ol{\Sigma}_{#1{}\mhyphen{}\forall}$}
\newcommand*{\SigmaKEvenFormula}[1]{\SigmaKEvenFormulaPrefix{#1}\nbdash-formula\xspace}
\newcommand*{\SigmaKEvenFormulas}[1]{\SigmaKEvenFormulaPrefix{#1}\nbdash-formulas\xspace}

\newcommand*{\SigmaKOddFormulaPrefix}[1]{$\ol{\Sigma}_{#1{}\mhyphen{}\exists}$}
\newcommand*{\SigmaKOddFormula}[1]{\SigmaKOddFormulaPrefix{#1}\nbdash-formula\xspace}
\newcommand*{\SigmaKOddFormulas}[1]{\SigmaKOddFormulaPrefix{#1}\nbdash-formulas\xspace}

\newcommand*{\CompactSigmaKEvenFormulaPrefix}[1]{$\ol{\Sigma}_{#1{}\mhyphen{}\forall}\ol{\Sigma}_1$}
\newcommand*{\CompactSigmaKEvenFormula}[1]{\CompactSigmaKEvenFormulaPrefix{#1}\nbdash-formula\xspace}
\newcommand*{\CompactSigmaKEvenFormulas}[1]{\CompactSigmaKEvenFormulaPrefix{#1}\nbdash-formulas\xspace}

\newcommand*{\CompactSigmaKOddFormulaPrefix}[1]{$\ol{\Sigma}_{#1{}\mhyphen{}\exists}\ol{\Pi}_1$}
\newcommand*{\CompactSigmaKOddFormula}[1]{\CompactSigmaKOddFormulaPrefix{#1}\nbdash-formula\xspace}
\newcommand*{\CompactSigmaKOddFormulas}[1]{\CompactSigmaKOddFormulaPrefix{#1}\nbdash-formulas\xspace}

\newcommand*{\SimpleCompactSigmaKEvenFormulaPrefix}[1]{\ensuremath{\Sigma_{#1{}\mhyphen{}\forall}\ol{\Sigma}_1}}
\newcommand*{\SimpleCompactSigmaKEvenFormula}[1]{\SimpleCompactSigmaKEvenFormulaPrefix{#1}\nbdash-formula\xspace}
\newcommand*{\SimpleCompactSigmaKEvenFormulas}[1]{\SimpleCompactSigmaKEvenFormulaPrefix{#1}\nbdash-formulas\xspace}

\newcommand*{\SimpleCompactSigmaKEvenSentences}[1]{\SimpleCompactSigmaKEvenFormulaPrefix{#1}\nbdash-sentences\xspace}

\newcommand*{\SimpleCompactSigmaKOddFormulaPrefix}[1]{\ensuremath{\Sigma_{#1{}\mhyphen{}\exists}\ol{\Pi}_1}}
\newcommand*{\SimpleCompactSigmaKOddFormula}[1]{\SimpleCompactSigmaKOddFormulaPrefix{#1}\nbdash-formula\xspace}
\newcommand*{\SimpleCompactSigmaKOddFormulas}[1]{\SimpleCompactSigmaKOddFormulaPrefix{#1}\nbdash-formulas\xspace}

\newcommand*{\FormulaTMHausdorffParametricNumberOfBits}[1]{\Psi_{\Language{D},w}^{#1}}
\newcommand*{\FormulaTMHausdorffWithP}{\FormulaTMHausdorffParametricNumberOfBits{\PolynomialBoundingEverything(\StringLength{w})}}

\newcommand*{\PolynomialSizeCertificates}{p_u}
\newcommand*{\PolynomialTimeMachineOmega}{p_{\Omega}}
\newcommand*{\PolynomialLengthHausPredD}{p_{\Language{D}}}
\newcommand*{\PolynomialBoundingEverything}{p}

\newcommand*{\Proofsep}{%
\smallbreak
\begingroup%
\centering%
\tikz \calligraphy[copperplate] (0,0) .. controls +(0.25,-0.25) and +(-0.25,0.25) .. ++(2,0) [this stroke style={light}]
+(-0.5,0.05) .. controls +(0.25,-0.25) and +(-0.25,0.25) .. ++(0.5,-0.05) [this stroke style={light}]
+(-0.5,0.05) .. controls +(0.25,-0.25) and +(-0.25,0.25) .. ++(1.5,0.05) [this stroke style={light}];
\par%
\endgroup%
\ignorespacesafterend}

\newcommand*{\defin}[1]{\emph{#1}}
\newcommand*{\Language}[1]{\mathcal{#1}}
\newcommand*{\HausdSeq}[1]{\mathcal{#1}}
\newcommand*{\Relation}[1]{#1}%
\newcommand*{\Predicate}[1]{#1}%
\newcommand*{\Machine}[1]{#1}%
\newcommand*{\SetOfSets}[1]{\mathcal{#1}}
\newcommand*{\VarSet}[1]{\bar{#1}}
\renewcommand*{\Vec}[1]{\overline{#1}}

\newcommand*{\StringTup}[1]{\overline{#1}}
\newcommand*{\DLFrag}[1]{\mathsf{#1}}
\newcommand*{\BooleanEncoding}[1]{\bm{#1}}
\newcommand*{\Interpr}[1]{\mathcal{#1}}
\newcommand*{\EvalInterpr}[2]{{#1}^{#2}}

\newcommand*{\DB}{D}
\newcommand*{\Dep}{\Sigma}
\newcommand*{\KB}{\mi{KB}}
\newcommand*{\KBDetails}{\tup{\DB,\Dep}}
\newcommand*{\KBDef}[1]{\tup{#1}}
\newcommand*{\Query}{q}
\newcommand*{\Expl}{E}
\newcommand*{\DBvaltrue}{\mathtt{true}}
\newcommand*{\DBvalfalse}{\mathtt{false}}
\newcommand*{\AssignYFact}{$\mi{AssignY}$\kern-0.50ex\nbdash-fact\xspace}
\newcommand*{\AssignYFacts}{$\mi{AssignY}$\kern-0.50ex\nbdash-facts\xspace}

\mathchardef\mhyphen="2D
\renewcommand*{\emptyset}{\varnothing}
\newcommand*{\ProofRightarrow}{\ensuremath{(\mkern-2mu{\Rightarrow}\mkern-2.8mu)}}
\newcommand*{\ProofLeftarrow}{\ensuremath{(\mkern-2.8mu{\Leftarrow}\mkern-2mu)}}
\newcommand*{\ProofRightarrowItem}{\ensuremath{\bm{\ProofRightarrow}}\xspace}
\newcommand*{\ProofLeftarrowItem}{\ensuremath{\bm{\ProofLeftarrow}}\xspace}

\newcommand*{\PrefixMatrixSeparator}{\mskip \medmuskip}

\newcommand*{\symbolwithin}[2]{{\mathmakebox[\widthof{\ensuremath{{}#2{}}}][c]{{#1}}}}

\newcommand*{\iExp}[3][2]{{\prescript{#2}{}{#1}^{#3}}}
\newcommand*{\iExpPolFunctions}[2][2]{\iExp[#1]{#2}{\mkern-2mu\PolFunctions\mkern+1mu}}
\newcommand*{\iLog}[2]{{\prescript{#1}{}{\log{}} #2}}
\newcommand*{\mycdot}{\, \cdot \,}

\newcommand*{\SOQ}{\mathcal{Q}}
\newcommand*{\FOQ}{Q}

\newcommand*{\ol}[1]{\overline{#1}}
\newcommand*{\wt}[1]{\widetilde{#1}}
\newcommand*{\mi}[1]{\mathit{#1}}
\newcommand*{\valtrue}{\mathrm{\mathit{true}}}
\newcommand*{\valfalse}{\mathrm{\mathit{false}}}

\newcommand*{\ra}{\rightarrow}
\newcommand*{\tup}[1]{( #1 )}
\newcommand*{\set}[1]{\{ #1 \}}
\newcommand*{\pair}[1]{\tup{#1}}
\newcommand*{\setsymmdifference}{\bigtriangleup}

\newcommand*{\compl}[1]{\overline{#1}}
\newcommand*{\LanguageOf}[1]{\mathscr{L}( #1 )}
\newcommand*{\TilingSyst}[1]{\mathcal{\uppercase{#1}}}
\newcommand*{\emptystring}{\varepsilon}
\newcommand*{\alphabet}{\Sigma}
\newcommand*{\StringUniverse}{\alphabet^*}
\newcommand*{\NaturalsDomain}{\mathbb{N}}
\newcommand*{\HausdPredDomain}{\StringUniverse \! \times \NaturalsDomain}
\newcommand*{\ChiLan}[1]{\chi_{#1}}
\newcommand*{\StringLength}[1]{\| #1 \|}
\newcommand*{\SetSize}[1]{| #1 |}

\newcommand*{\querystate}{q_{\mi{?}}}

\newcommand*{\yesanswerstate}{q_{\mi{yes}}}
\newcommand*{\noanswerstate}{q_{\mi{no}}}
\newcommand*{\blanksymbol}{\mathrlap{/}{b}}
\newcommand*{\HammingWeight}[1]{\mi{hw}(#1)}

\newcommand*{\QueryInComp}[2]{{q_{#2}}(#1)}
\newcommand*{\AnsInComp}[2]{a_{#2}(#1)}

\newcommand*{\AllQueriesGenericOracle}[2]{Q^{?}_{#1}(#2)}

\newcommand*{\QueryInCompPar}[2]{Q_{#2}(#1)}
\newcommand*{\QueryYESinCompPar}[2]{Y_{#2}(#1)}
\newcommand*{\AnsInCompPar}[2]{A_{#2}(#1)}
\newcommand*{\AnsYESvectorInCompPar}[1]{\Vec{Y}(#1)}
\newcommand*{\PortionCompToQuery}[2]{\{#1(q_{#2})\}}
\newcommand*{\IDinCompQuery}[2]{#1(q_{#2})}
\newcommand*{\PortionCompToAns}[2]{\{#1(a_{#2})\}}
\newcommand*{\IDinCompAns}[2]{#1(a_{#2})}

\newcommand*{\lexsucc}{\succ}
\newcommand*{\lexsucceq}{\succcurlyeq}
\newcommand*{\lexprec}{\prec}

\newcommand*{\PredHSucc}[2][]{\Predicate{#2}^{\lexsucc}_{#1}}
\newcommand*{\PredHAccCurr}[2][]{\Predicate{#2}^{\cmark}_{#1}}
\newcommand*{\PredHRejCurr}[2][]{\Predicate{#2}^{\xmark}_{#1}}

\newcommand*{\Sat}{\textnormal{\textsc{Sat}}\xspace}
\newcommand*{\NiExpNjExpGenericProblem}[2]{\textnormal{\textsc{Accept\-Reject}}${}^{#2}_{#1}$\xspace}

\newcommand*{\BCQ}{\textrm{BCQ}\xspace}
\newcommand*{\BCQs}{\textrm{BCQs}\xspace}
\newcommand*{\UCQ}{\textrm{UCQ}\xspace}

\newcommand*{\QBSF}{\textrm{QBSF}\xspace}
\newcommand*{\QBSFs}{\textrm{QBSFs}\xspace}
\newcommand*{\MinExRelGeneric}{{\ensuremath{\textnormal{\textsc{MinEx\nbdash-Rel}}_{\preccurlyeq}}}\xspace}
\newcommand*{\MinExIrrelGeneric}{{\ensuremath{\textnormal{\textsc{MinEx\nbdash-Irrel}}_{\preccurlyeq}}}\xspace}
\newcommand*{\MinExNecGeneric}{{\ensuremath{\textnormal{\textsc{MinEx\nbdash-Nec}}_{\preccurlyeq}}}\xspace}
\newcommand*{\MinExRelMinCard}{{\ensuremath{\textnormal{\textsc{MinEx\nbdash-Rel}}_{\leq}}}\xspace}
\newcommand*{\MinExIrrelMinCard}{{\ensuremath{\textnormal{\textsc{MinEx\nbdash-Irrel}}_{\leq}}}\xspace}
\newcommand*{\MinExNecMinCard}{{\ensuremath{\textnormal{\textsc{MinEx\nbdash-Nec}}_{\leq}}}\xspace}
\newcommand*{\MinExRelMinWeight}{{\ensuremath{\textnormal{\textsc{MinEx\nbdash-Rel}}_{\leq_w}}}\xspace}
\newcommand*{\MinExIrrelMinWeight}{{\ensuremath{\textnormal{\textsc{MinEx\nbdash-Irrel}}_{\leq_w}}}\xspace}
\newcommand*{\MinExNecMinWeight}{{\ensuremath{\textnormal{\textsc{MinEx\nbdash-Nec}}_{\leq_w}}}\xspace}
\newcommand*{\ExpTilProb}{\textnormal{\textsc{Exp-Tiling}}\xspace}

\newcommand*{\MaxSatSigmaFormula}[1]{${\ol{\Sigma}_{#1}}$\nbdash-\textnormal{\textsc{\mbox{MaxSat}}}\xspace}
\newcommand*{\LexMaxSigmaFormula}[1]{${\ol{\Sigma}_{#1}}$\nbdash-\textnormal{\textsc{\mbox{LexMax}}}\xspace}
\newcommand*{\LexMaxFuncSigmaFormula}[1]{${\ol{\Sigma}_{#1}}$\nbdash-\textnormal{\textsc{\mbox{LexMaxFunc}}}\xspace}

\RequirePackage{ifthen}
\RequirePackage{mathtools}

\newcommand{\mathremovespaces}[1]{%
  \hbox{\ensuremath{%
    \medmuskip=0mu %
    #1}}}

\newcommand*{\myparallel}{{\mkern3mu\vphantom{\perp}\vrule depth 0pt\mkern2mu\vrule depth 0pt\mkern3mu}}

\newcommand*{\yeslbl}{`yes'\xspace}
\newcommand*{\nolbl}{`no'\xspace}
\newcommand*{\InstaceSuffix}{\nbdash-instance\xspace}
\newcommand*{\InstacesSuffix}{\nbdash-instances\xspace}
\newcommand*{\AnswerSuffix}{\nbdash-answer\xspace}
\newcommand*{\AnswersSuffix}{\nbdash-answers\xspace}
\newcommand*{\yesinst}{\yeslbl{}\InstaceSuffix}
\newcommand*{\noinst}{\nolbl{}\InstaceSuffix}
\newcommand*{\yesinsts}{\yeslbl{}\InstacesSuffix}
\newcommand*{\noinsts}{\nolbl{}\InstacesSuffix}
\newcommand*{\yesansw}{\yeslbl{}\AnswerSuffix}
\newcommand*{\noansw}{\nolbl{}\AnswerSuffix}
\newcommand*{\yesansws}{\yeslbl{}\AnswersSuffix}
\newcommand*{\noansws}{\nolbl{}\AnswersSuffix}

\newcommand*{\LogFunctions}{\mathit{L\mkern-.75mu og}}
\newcommand*{\PolFunctions}{{\mathchoice{\mathit{P\mkern-3.5mu ol}}{\mathit{P\mkern-3.5mu ol}}{\mathit{P\mkern-2mu ol}}{\mathit{P\mkern-2mu ol}}}}
\newcommand*{\BHText}{Boolean Hierarchy\xspace}
\newcommand*{\PHText}{Polynomial Hierarchy\xspace}
\newcommand*{\SEHText}{Strong Exponential Hierarchy\xspace}
\newcommand*{\WEHText}{Exponential Hierarchy\xspace}
\newcommand*{\iWEHText}[1]{\ensuremath{{#1}}\nbdash-Exponential Hierarchy\xspace}
\newcommand*{\iWEHsText}[1]{\ensuremath{{#1}}\nbdash-Exponential Hierarchies\xspace}
\newcommand*{\WEHStressedText}{(Weak) Exponential Hierarchy\xspace}
\newcommand*{\SEHDeltaLevel}{\ensuremath{\Delta}\nbdash-level\xspace}
\newcommand*{\SEHThetaLevel}{\ensuremath{\Theta}\nbdash-level\xspace}

\newcommand*{\ItExpHMetaText}{Iterated Exponentials Meta-Hiearchy\xspace}
\newcommand*{\iExponential}[1]{\mathremovespaces{#1}\nbdash-ex\-po\-nen\-tial\xspace}

\newcommand*{\BoundedOracle}[4][]{\textnormal{\ensuremath{{#2}^{{#3}\ifthenelse{\equal{#4}{\empty}}{\empty}{[#4]}}\hspace*{-1\scriptspace}{#1}}}\xspace}

\newcommand*{\DoubleBoundedParOracle}[5][]{\textnormal{\ensuremath{{#2}^{{#3}\ifthenelse{\equal{#4}{\empty}}{\empty}{[#4]}}_{\ifthenelse{\equal{#5}{\empty}}{\myparallel}{\parallel \mkern -1.5mu \langle #5 \rangle}}\hspace*{-1\scriptspace}{#1}}}\xspace}

\newcommand*{\DoubleBoundedPlusParOracle}[5][]{{\setbox0=\hbox{${\scriptstyle [#4]}$}\setbox1=\hbox{${\scriptscriptstyle +}$}\textnormal{\ensuremath{{#2}^{{#3}\ifthenelse{\equal{#4}{\empty}}{\empty}{{[#4]}{\raisebox{(\ht0 - \ht1 - \dp1)/2 + \dp1}{${\scriptscriptstyle +}$}}}}_{\ifthenelse{\equal{#5}{\empty}}{\myparallel}{\parallel \mkern -1.5mu \langle #5 \rangle}}\hspace*{-1\scriptspace}{#1}}}}\xspace}

\newcommand*{\Oracle}[3][]{\BoundedOracle[#1]{#2}{#3}{\empty}}

\newcommand*{\ParOracle}[3][]{\DoubleBoundedParOracle[#1]{#2}{#3}{\empty}{\empty}}

\newcommand*{\BoundedParOracle}[4][]{\DoubleBoundedParOracle[#1]{#2}{#3}{#4}{\empty}}

\newcommand*{\ParBoundedOracle}[4][]{\DoubleBoundedParOracle[#1]{#2}{#3}{\empty}{#4}}

\newcommand*{\LogOracle}[2]{\BoundedOracle{#1}{#2}{\mkern-1mu \LogFunctions}}
\newcommand*{\PolOracle}[2]{\BoundedOracle{#1}{#2}{\PolFunctions}}
\newcommand*{\ExpOracle}[2]{\BoundedOracle{#1}{#2}{2^\PolFunctions}}

\newcommand*{\ComplementPrefix}{\textnormal{co\nbdash-}}
\newcommand*{\ComplementPrefixKerned}{\textnormal{co\nbdash-}\kern-0.08em{}}
\newcommand*{\HardSuffix}{\textnormal{\nbdash-hard}\xspace}
\newcommand*{\CompleteSuffix}{\textnormal{\nbdash-complete}\xspace}

\newcommand*{\DTime}[1]{\textnormal{{DTIME}}(#1)}

\newcommand*{\DTimeOracle}[2]{\textnormal{{DTIME}}^{#2}(#1)}

\newcommand*{\NTime}[1]{\textnormal{{NTIME}}(#1)}

\newcommand*{\NTimeOracle}[2]{\textnormal{{NTIME}}^{#2}(#1)}

\newcommand*{\DSpace}[1]{\textnormal{{DSPACE}}(#1)}
\newcommand*{\DSpaceOracle}[2]{\textnormal{{DSPACE}}^{#2}(#1)}

\newcommand*{\NSpace}[1]{\textnormal{{NSPACE}}(#1)}

\newcommand*{\ComplexityClass}[1]{\mathbf{#1}}

\newcommand*{\ATime}[1]{\textnormal{{ATIME}}(#1)}
\newcommand*{\BoundedExATime}[2]{\Sigma_{#1}\mhyphen\textnormal{{TIME}}(#2)}
\newcommand*{\BoundedUnATime}[2]{\Pi_{#1}\mhyphen\textnormal{{TIME}}(#2)}
\newcommand*{\ASpace}[1]{\textnormal{{ASPACE}}(#1)}

\newcommand*{\LogSpace}{\textnormal{\textsc{LogSpace}}\xspace}

\newcommand*{\PTime}{\textnormal{P}\xspace}

\newcommand*{\NPTime}{\textnormal{NP}\xspace}
\newcommand*{\NPTimeh}{\NPTime{}\HardSuffix}

\newcommand*{\CoNPTime}{\ComplementPrefixKerned\NPTime{}\xspace}

\newcommand*{\BH}[1][]{{\ifthenelse{\equal{#1}{\empty}}{\textnormal{{BH}}}{\NPTime(#1)}}\xspace}

\newcommand*{\DP}[1][]{\textnormal{\ensuremath{\textnormal{{D}}^\mathrm{{P}}_{#1}\hspace*{-1\scriptspace}}}\xspace}

\newcommand*{\BHThree}[1][]{{\ifthenelse{\equal{#1}{\empty}}{\ensuremath{\textnormal{{BH}}_3}}{\ensuremath{\textnormal{{BH}}_3(#1)}}}\xspace}

\newcommand*{\BHGeneric}[1]{{\ensuremath{\textnormal{{BH}}(#1)}}\xspace}

\newcommand*{\PolHier}{\textnormal{PH}\xspace}

\newcommand*{\ThetaP}[1]{\textnormal{\ensuremath{\Theta^\mathrm{{P}}_{#1}\hspace*{-1\scriptspace}}}\xspace}

\newcommand*{\DeltaP}[1]{\textnormal{\ensuremath{\Delta^\mathrm{{P}}_{#1}\hspace*{-1\scriptspace}}}\xspace}

\newcommand*{\SigmaP}[1]{\textnormal{\ensuremath{\Sigma^\mathrm{{P}}_{#1}\hspace*{-1\scriptspace}}}\xspace}

\newcommand*{\PiP}[1]{\textnormal{\ensuremath{\Pi^\mathrm{{P}}_{#1}\hspace*{-1\scriptspace}}}\xspace}

\newcommand*{\PSpace}{\textnormal{\textsc{PSpace}}\xspace}

\newcommand*{\ExpTime}{\textnormal{\textsc{Exp}}\xspace}

\newcommand*{\NExpTime}{\textnormal{\textsc{NExp}}\xspace}
\newcommand*{\NExpTimeh}{\NExpTime{}\HardSuffix}
\newcommand*{\NExpTimec}{\NExpTime{}\CompleteSuffix}
\newcommand*{\CoNExpTime}{\ComplementPrefixKerned\NExpTime{}\xspace}

\newcommand*{\BHNExp}[1][]{{\ifthenelse{\equal{#1}{\empty}}{\textnormal{\textsc{ExpBH}}}{\NExpTime(#1)}}\xspace}

\newcommand*{\DExp}{\textnormal{\ensuremath{\textnormal{{D}}^\textnormal{\textsc{Exp}}\hspace*{-1\scriptspace}}}\xspace}

\newcommand*{\CoDExp}{\ComplementPrefixKerned\DExp{}\xspace}

\newcommand*{\SExpHier}{\textnormal{{SEH}}\xspace}
\newcommand*{\SigmaSExp}[1]{\textnormal{\ensuremath{\Sigma^\mathrm{SE}_{#1}\hspace*{-1\scriptspace}}}\xspace}
\newcommand*{\PiSExp}[1]{\textnormal{\ensuremath{\Pi^\mathrm{SE}_{#1}\hspace*{-1\scriptspace}}}\xspace}
\newcommand*{\DeltaSExp}[1]{\textnormal{\ensuremath{\Delta^\mathrm{SE}_{#1}\hspace*{-1\scriptspace}}}\xspace}

\newcommand*{\PNExp}{\textnormal{\Oracle{\PTime}{\NExpTime}}\xspace}
\newcommand*{\PNExph}{\PNExp{}\HardSuffix}
\newcommand*{\PNExpc}{\PNExp{}\CompleteSuffix}
\newcommand*{\PNExpPar}{\textnormal{\ParOracle{\PTime}{\NExpTime}}\xspace}

\newcommand*{\PNExpLog}{\textnormal{\LogOracle{\PTime}{\NExpTime}}\xspace}
\newcommand*{\PNExpLogh}{\PNExpLog{}\HardSuffix}
\newcommand*{\PNExpLogc}{\PNExpLog{}\CompleteSuffix}
\newcommand*{\NPNExp}{\textnormal{\Oracle{\NPTime}{\NExpTime}}\xspace}

\newcommand*{\WExpHier}{\textnormal{{EH}}\xspace}
\newcommand*{\SigmaWExp}[1]{\textnormal{\ensuremath{\Sigma^\mathrm{E}_{#1}\hspace*{-1\scriptspace}}}\xspace}
\newcommand*{\SigmaWExph}[1]{\SigmaWExp{#1}{}\HardSuffix}

\newcommand*{\PiWExp}[1]{\textnormal{\ensuremath{\Pi^\mathrm{E}_{#1}\hspace*{-1\scriptspace}}}\xspace}

\newcommand*{\DeltaWExp}[1]{\textnormal{\ensuremath{\Delta^\mathrm{E}_{#1}\hspace*{-1\scriptspace}}}\xspace}

\newcommand*{\DeltaWExpBound}[2]{\textnormal{\ensuremath{\Delta^{\mathrm{E}[#2]}_{#1}\hspace*{-1\scriptspace}}}\xspace}

\newcommand*{\ExpSpace}{\textnormal{\textsc{ExpSpace}}\xspace}

\newcommand*{\iExpTime}[1]{\textnormal{\mathremovespaces{#1}\textsc{Exp}}\xspace}

\newcommand*{\iNExpTime}[1]{\textnormal{N{}\mathremovespaces{#1}\textsc{Exp}}\xspace}

\newcommand*{\CoINExpTime}[1]{\ComplementPrefixKerned\iNExpTime{#1}\xspace}

\newcommand*{\iWExpHier}[1]{\textnormal{\ensuremath{\mathremovespaces{#1} \textnormal{{EH}}}}\xspace}
\newcommand*{\SigmaIWExp}[2]{\textnormal{\ensuremath{\Sigma^{#1 \mathrm{{E}}}_{#2}\hspace*{-1\scriptspace}}}\xspace}
\newcommand*{\PiIWExp}[2]{\textnormal{\ensuremath{\Pi^{#1 \mathrm{{E}}}_{#2}\hspace*{-1\scriptspace}}}\xspace}

\newcommand*{\DeltaIWExpBound}[3]{\textnormal{\ensuremath{\Delta^{#1 \mathrm{{E}}[#3]}_{#2}\hspace*{-1\scriptspace}}}\xspace}

\newcommand*{\iExpSpace}[1]{\textnormal{\mathremovespaces{#1}\textsc{ExpSpace}}\xspace}

\newcommand*{\iNExpSpace}[1]{\textnormal{N{}\mathremovespaces{#1}\textsc{ExpSpace}}\xspace}

\newcommand*{\Elementary}{\textnormal{\textsc{Elementary}}\xspace}
\newcommand*{\Regular}{\textnormal{\textsc{Reg}}\xspace}

\newcommand*{\KarpRed}[1][p]{\leq_{\mathrm{m}}^{#1}}
\newcommand*{\TTRed}[1][p]{\leq_{\mathrm{tt}}^{#1}}
\newcommand*{\TTRedCLASS}[2][p]{{\leq}_{\mathrm{tt}}^{#1}\mkern-1mu(#2)}
\newcommand*{\TuringRed}[1][p]{\leq_{\mathrm{T}}^{#1}}

\newcommand*{\BoundedHausdCLASS}[2]{#2{}(\mathremovespaces{#1})}
\newcommand*{\GenericLength}{\ast}%

\newcommand*{\HausdIndex}[2]{\hat{z}_{#2}(#1)}

\newcommand*{\Circuit}[1]{\mathscr{#1}}

\newcommand*{\CirctuitValue}[2]{\mi{val}_{#1}(#2)}

\theoremstyle{plain}
\newtheorem*{theorem*}{Theorem}
\newtheorem{theorem}{Theorem}[section]
\newtheorem{lemma}[theorem]{Lemma}

\newtheorem{corollary}[theorem]{Corollary}
\newtheorem{subproperty}{Property}[theorem]

\newenvironment{subproof}[1][\proofname]{%
  \renewcommand{\qedsymbol}{\(\lhd\)}%
  \begin{adjustwidth}{\parindent}{}%
  \begin{proof}[#1]%
}{%
  \end{proof}%
  \end{adjustwidth}%
  \addvspace{\topsep}
}

\theoremstyle{definition}
\newtheorem{definition}[theorem]{Definition}
\newtheorem{example}[theorem]{Example}

\newtheoremstyle{problemstyle}%
  {\topsep}%
  {\topsep}%
  {\normalfont}%
  {}%
  {\bfseries}%
  {:}%
  { }%
  {}%
\theoremstyle{problemstyle}
\newtheorem*{probenvironment}{Problem}
\newcommand{\Problem}[3]{%
\begin{samepage}
\begin{probenvironment}
#1\par
\begin{description}[nosep]
\item[\textnormal{\textit{Input:}}] #2
\item[\textnormal{\textit{Output:}}] #3
\end{description}
\end{probenvironment}
\end{samepage}}

\title{Hausdorff Reductions and the Exponential Hierarchies\thanks{This work was supported by the European Union -- NextGenerationEU programme, through the Italian Ministry of University and Research (MUR) PRIN~2022-PNRR grant P2022KHTX7 ``DISTORT''---CUP:~J53D23015000001, under the Italian ``National Recovery and Resilience Plan'' (PNRR), Mission~4 Component~1.}}

\date{}

\author{Enrico Malizia\,\orcidlink{0000-0002-6780-4711}}
\affil{DISI, University of Bologna, Italy}
\affil{enrico.malizia@unibo.it}

\begin{document}

\pagenumbering{roman}

\maketitle

\thispagestyle{empty}

\begin{abstract}
We introduce Hausdorff (complexity) classes, which yield canonical normal forms for the intermediate levels of the iterated exponential hierarchies, including the Polynomial Hierarchy, the (Weak) Exponential Hierarchy, and higher\nbdash-order exponential hierarchies.
Just as certificates capture main hierarchy levels without oracles, Hausdorff classes give an oracle\nbdash-free characterization of the intermediate hierarchy levels.

The Hausdorff perspective provides a unifying structural explanation for many known equivalences between oracle classes.
Indeed, seemingly different oracle classes corresponding to the same intermediate level are shown to arise from just three distinct yet equivalent oracle\nbdash-aided ways of deciding languages in a single Hausdorff class, replacing multiple oracle\nbdash-based views with a unique characterization.
Moreover, it explains the collapse of the \SEHText, showing that $\PNExp = \NPNExp$ because both classes coincide with the same Hausdorff class, thereby resolving a question of Hemachandra.

Finally, we define canonical complete problems for $\PNExpLog$ and $\PNExp$, yielding matching lower bounds for problems whose hardness had remained open due to the lack of suitable hard problems to reduce from.
\end{abstract}

\Proofsep

\clearpage

\tableofcontents

\clearpage

\pagenumbering{arabic}

\section{Introduction}
\label{sec_intro}

In this paper, we introduce \emph{Hausdorff (complexity) classes}, which provide canonical characterizations of the intermediate levels of the iterated exponential hierarchies, including the Polynomial Hierarchy ($\PolHier$), the (Weak) Exponential Hierarchy ($\WExpHier$), and their higher\nbdash-order generalizations. %
Assuming no hierarchy collapse, Hausdorff classes uniquely determine these levels and act as natural normal forms for them.
In particular, they capture intermediate levels without referring to oracle machines, analogously to how alternating quantifiers and certificates characterize main hierarchy levels, yielding a unified, oracle\nbdash-free view of these hierarchies.

The Hausdorff perspective sheds light on two well\nbdash-known phenomena in structural complexity theory that, while seemingly unrelated, share a common origin.

First, many complexity classes defined via seemingly different oracle mechanisms coincide, e.g.,%
\footnote{$\Oracle[\ensuremath{\langle\mkern-2mu \PolFunctions \rangle}]{\NExpTime}{\NExpTime}$ is the class of languages decided by $\NExpTime$ machines via \emph{polynomially}-long queries to $\NExpTime$ oracles.}
\begin{flalign*}
&\ParOracle{\PTime}{\NPTime} = \LogOracle{\PTime}{\NPTime} = \Oracle{\LogSpace}{\NPTime}; & \text{\cite{KoblerSW87,Wagner1987,Wagner1990,Beigel1991,Buss1991}} \\
&\mathrlap{\PNExp = \NPNExp = \ParOracle{\ExpTime}{\NPTime} = \PolOracle{\ExpTime}{\NPTime} = \Oracle{\PSpace}{\NPTime} = \Oracle[\ensuremath{\langle\mkern-2mu \PolFunctions \rangle}]{\NExpTime}{\NExpTime};} \\
& & \text{\cite{SchoningW88,Hemachandra1989,Beigel1991,Hemaspaandra1994,Gottlob1995,Mocas1996,AllenderKRR2011}} \\
&\Oracle{\ExpTime}{\NExpTime} = \Oracle{\NExpTime}{\NExpTime}. & \text{\cite{SchoningW88,Hemaspaandra1994}}
\end{flalign*}
These equivalences suggest that intermediate hierarchy levels have a deeper structure not captured by their oracle definitions.
These also raise broader questions about higher exponential oracle classes, given the many ways these classes can combine.
For instance, how do classes such as $\Oracle{\iNExpTime{2}}{\iNExpTime{3}}$, $\Oracle{\iNExpTime{3}}{\iNExpTime{2}}$, and $\Oracle{\iNExpTime{4}}{\NExpTime}$ relate?%
\footnote{(N)$\iExpTime{i}$ is the class of languages decidable in (non)deterministic \iExponential{i} time, where $\iExpTime{0} = \PTime$ and $\iNExpTime{0} = \NPTime$; $\iExpSpace{i}$ is analogously defined, with $\iExpSpace{-1} = \LogSpace$ and $\iExpSpace{0} = \PSpace$.}
Do they coincide, or even equal $\iNExpTime{5}$?

Second, \citeauthor{Hemachandra1989} showed that $\PNExp = \NPNExp$, implying that the \SEHText ($\SExpHier$) collapses to $\PNExp$~\cite{Hemachandra1987,Hemachandra1989}.
His proof relied on a census argument and did not reveal an underlying structural reason for this collapse.
Indeed, he asked whether such an explanation could be obtained via ``quantifier manipulations'' of suitable certificate\nbdash-based or alternating Turing machine characterizations of $\SExpHier$, leaving the question open.
Although $\PNExp = \NPNExp$ was later reproven via different methods~\cite{SchoningW88,Beigel1991,Gottlob1995,AllenderKRR2011}, none of these closed \citeauthor{Hemachandra1989}'s question, which has thus far remained unanswered.

The Hausdorff classes introduced here provide precisely such an explanation.
They reveal that both the coincidence of seemingly different oracle classes and the collapse of $\SExpHier$ arise from the same underlying mechanism: the oracle classes involved are captured by the same Hausdorff class.

We note that the Hausdorff characterization also yields the certificate\nbdash-based and alternating machine characterizations sought by \citeauthor{Hemachandra1989}.
However, it is the Hausdorff characterization itself that reveals the structural reason for the collapse.
The lack of such an explanation was also reflected in the historical treatment of $\WExpHier$ and $\SExpHier$.
Specifically, the collapse of $\SExpHier$ was regarded as surprising~\cite{Hemachandra1986,Hartmanis1990,Beigel1991}, and $\SExpHier$ and $\WExpHier$ were often treated as distinct hierarchies~\cite{Hemachandra1986,Hemachandra1989,Hartmanis1990},%
\footnote{\label{fn_real_exp_hierarchy}\citet{Hemachandra1986} wonders which the ``real'' exponential hierarchy is, between the strong one and the weak one, and \citet{Hartmanis1990} refers to the ``\emph{other} exponential hierarchy [...] not known to collapse'' (emphasis added).}
suggesting that the nature of $\SExpHier$ was not fully clarified.
The Hausdorff perspective clarifies these aspects as well.

\subsection{The Hausdorff perspective: notions}
\label{sec_intro_Hausdorff_notions}

The key idea of the Hausdorff perspective is that intermediate hierarchy levels above a main level $\ComplexityClass{C}$ correspond precisely to languages obtained by Boolean combinations of languages from $\ComplexityClass{C}$, where the combination length may depend on the input size.
This idea is rooted in a classical result of \citet{Hausdorff1962}, which, for complexity classes, yields a canonical characterization of Boolean combinations of languages.
Let $\ComplexityClass{C}$ be a complexity class containing the empty language and closed under disjunction and conjunction.
Every language $\Language{L}$ in the Boolean closure\footnote{The Boolean closure of $\ComplexityClass{C}$ is the class of languages obtained from those in $\ComplexityClass{C}$ via union, intersection, and difference.} of $\ComplexityClass{C}$ can be expressed as a \emph{Hausdorff summation}
\(
\Language{L} = (\Language{D}_1 \setminus \Language{D}_2) \cup (\Language{D}_3 \setminus \Language{D}_4) \cup \dots \cup (\Language{D}_{m-1} \setminus \Language{D}_m),
\)
where $\Language{D}_1 \supseteq \allowbreak \Language{D}_2 \supseteq \dots \supseteq \Language{D}_m$ are languages from $\ComplexityClass{C}$, and $m$ is a constant depending on $\Language{L}$.%
\footnote{If $m$ is odd, the summation ends with `${} \cup \Language{D}_m$'.}
Membership in $\Language{L}$ is thus determined solely by parity:
a string $w$ belongs to $\Language{L}$ iff the largest index $z$ such that $w \in \Language{D}_z$ is odd.

The key step in extending this result beyond the Boolean closure of $\ComplexityClass{C}$ to intermediate hierarchy levels above $\ComplexityClass{C}$ is to relax the constraint on the Hausdorff summation's length.
We allow it to depend on $\StringLength{w}$ rather than having it fixed.
This leads to what we call \emph{Hausdorff predicates}.
Intuitively, instead of a sequence of languages $\Language{D}_1 \supseteq \Language{D}_2 \supseteq \dots$ from $\ComplexityClass{C}$, we use a predicate $\Language{D}(w,z)$ of complexity $\ComplexityClass{C}$ such that $\Language{D}(w,z) = 1$ iff $w \in \Language{D}_z$.
A language $\Language{L}$ \emph{Hausdorff reduces} to $\Language{D}$ if membership of a string $w$ in $\Language{L}$ can be determined from the parity of indices, i.e., $w \in \Language{L}$ iff the largest $z$ such that $\Language{D}(w,z) = 1$ is odd.
\emph{Hausdorff classes} are then defined as the set of languages that Hausdorff reduce to Hausdorff predicates of a given type.%
\footnote{One might consider defining Hausdorff reductions via ``infinite Hausdorff sequences''. However, if defined naively, non\nbdash-recursively enumerable languages could be Hausdorff characterized via regular ones, flattening the analysis by masking distinctions between Hausdorff classes. By contrast, the formulation via Hausdorff predicates reveals a rich structure and enables the Hausdorff\nbdash-characterization of the intermediate levels of the iterated exponential hierarchies.}
We now formalize these notions.

\getkeytheorem{DefHausdPred}

Defining Hausdorff classes requires two properties of Hausdorff predicates:
length and complexity.

We begin with length.
A \defin{Hausdorff length function} $g\colon \NaturalsDomain \to \NaturalsDomain$ is a strictly positive nondecreasing function.
A Hausdorff predicate $\Language{D}$ is \defin{$g(n)$\nbdash-long} if $\HausdIndex{w}{\Language{D}} \le g(\StringLength{w})$ for all strings $w$.
We say that $\Language{D}$ is \defin{bounded} if it is $g(n)$\nbdash-long for some function $g(n)$.
Consider now predicate complexity.
A Hausdorff predicate $\Language{D}$ has complexity $\ComplexityClass{C}$ if deciding whether $\pair{w,z} \in \Language{D}$ is in $\ComplexityClass{C}$ with respect to the \emph{size of $w$ alone}.
This restriction reflects that Hausdorff predicates generalize finite Hausdorff sequences and prevents arbitrarily large indices $z$ from acting as padding for the pairs $\pair{w,z}$.
E.g., $\Language{D}$ is in $\Oracle{\iNExpTime{i}}{\SigmaP{c-1}}$ if a nondeterministic oracle machine with a $\SigmaP{c-1}$ oracle decides $\Language{D}(w,z)$ in \iExponential{i} time \Wrt $\StringLength{w}$~alone.

An important property of Hausdorff predicates is that bounded time or space implies bounded length.
To see this, consider the case of time\nbdash-bounded computation.
If $\pair{w,z} \in \Language{D}$ is decided by a machine $\Machine{M}$ in time $t(\StringLength{w})$, then within this bound $\Machine{M}$ can read at most $t(\StringLength{w})$ bits of the index $z$.
Hence, if $\Machine{M}$ answers correctly on every pair $\pair{w,z}$, the Hausdorff index of $w$ \Wrt $\Language{D}$ must be representable using at most $t(\StringLength{w})$ bits.
Intuitively, indices $z$ whose binary representations are too large are indistinguishable to machines that cannot read them within the available time.
Formalizing this intuition yields the following result (proof in \zcref{sec_Hausdorff_reductions_classes}).

\getkeytheorem{HausdPrefBoundedComplexityImplyBoundedLength}

The choice of measuring Hausdorff predicate complexity with respect to $w$ alone is therefore not merely technical:
it ensures that bounded predicate complexity entails bounded Hausdorff length.
Without it, Hausdorff predicates could admit arbitrarily large indices, and the resulting classes would no longer reflect the structure of the iterated exponential hierarchies (see \zcref{sec_intro_Hausdorff_structural_insights}).

With Hausdorff predicates, we define the notion of Hausdorff reduction/characterization.
Our definition of Hausdorff reduction combines, generalizes, and simplifies those in \cite{Wagner1987,Wagner1990}.

\getkeytheorem{DefHausdRed}

Hausdorff reductions then define Hausdorff complexity classes.

\getkeytheorem{DefHausdClass}

Observe that Hausdorff complexity classes are \emph{not} defined as classes of Hausdorff predicates.
For instance, the sets of Hausdorff predicates of polynomial and exponential length (see \zcref{theo_Hausd_pred_bounded_complexity_imply_bounded_length}), with the former included in the latter, both are subsets of $\NPTime$.
In contrast, the $\NPTime$ Hausdorff complexity classes of polynomial and exponential length contain $\NPTime$, as they equal $\ThetaP{2}$ and $\DeltaP{2}$, respectively.

For a family of functions $G$, we define the Hausdorff complexity class $\BoundedHausdCLASS{G}{\ComplexityClass{C}} = \bigcup_{g(n) \in G} \set{\BoundedHausdCLASS{g(n)}{\ComplexityClass{C}}}$.
E.g., $\BoundedHausdCLASS{2^{\PolFunctions}}{\Oracle{\NExpTime}{\SigmaP{c-1}}}$ is the class of all $\Oracle{\NExpTime}{\SigmaP{c-1}}$ Hausdorff languages of length $g(n)$, for some $g(n) \in 2^{\PolFunctions}$, or, more roughly, the class of all $\Oracle{\NExpTime}{\SigmaP{c-1}}$ Hausdorff languages of exponential length.

We can show that Hausdorff classes over deterministic classes do not yield new language classes.
Intuitively, the Hausdorff summation of languages from a deterministic class remains in the same class, as deterministic classes are closed under union, intersection, and complement (see \zcref{sec_Hausdorff_reductions_classes}).
Moreover, by Savitch's theorem, deterministic and nondeterministic space coincide for all space\nbdash-bounded classes from polynomial space upward.
Since we do not consider Hausdorff classes over nondeterministic logspace, we may disregard nondeterministic space classes as well.
Hence, we will focus on Hausdorff classes defined over nondeterministic $i$\nbdash-exponential time classes, with $\SigmaP{c-1}$ oracles, i.e., the main levels of the iterated exponential hierarchies.

\subsection{The Hausdorff perspective: relation to earlier notions}
\label{sec_intro_earlier_notions}

The Hausdorff notions introduced above generalize the Boolean Hierarchies (see \zcref{sec_extended_Boolean_Hierarchies}).
Early work \cite{Wagner1987,Wagner1988,Buss1988} extended the Boolean hierarchy over $\NPTime$ to polynomially\nbdash-long Boolean combinations of $\NPTime$ languages.
This was achieved via an initial notion of Hausdorff reduction, introduced by \citet{Wagner1987} as a special case of polynomial tt\mbox{-}reductions.
We briefly recall it.
Let $\ChiLan{\Language{L}}$ denote the characteristic function of a language $\Language{L} \subseteq \StringUniverse$, and let $\Language{A}, \Language{B} \subseteq \StringUniverse$ be two languages.
In \cite{Wagner1987}, $\Language{A}$ is said to \emph{polynomially Hausdorff reduce} to $\Language{B}$ if there exists a polynomial\nbdash-time computable function $f$ such that, for every $w \in \StringUniverse$, $f(w) = \tup{v_1,\dots,v_{2k}}$, where the $v_i$ are strings (and $k$ depends on $w$), and
\(
\ChiLan{\Language{A}}(w)
= (\ChiLan{\Language{B}}(v_1) \land \lnot \ChiLan{\Language{B}}(v_2))
\lor \dots \lor
(\ChiLan{\Language{B}}(v_{2k-1}) \land \lnot \ChiLan{\Language{B}}(v_{2k})),
\)
with $\ChiLan{\Language{B}}(v_1) \ge \ChiLan{\Language{B}}(v_2) \ge \dots \ge \ChiLan{\Language{B}}(v_{2k})$.
This formulation reflects Hausdorff's results.

Since Hausdorff reductions in those works were polynomial tt\nbdash-reductions, what could be investigated was significantly constrained, as such reductions are not expressive enough to capture Hausdorff classes beyond $\PolHier$ or Hausdorff summations of super\nbdash-polynomial length.
Thus, \citeauthor{Wagner1987}'s Hausdorff reductions, as defined in \cite{Wagner1987}, cannot capture Hausdorff classes beyond $\BoundedHausdCLASS{\PolFunctions}{\NPTime}$, i.e., the class $\ThetaP{2}$.
These limitations are problematic for our purposes, as we aim to apply the Hausdorff perspective to the iterated exponential hierarchies.

Our notion of Hausdorff reduction, based on Hausdorff predicates rather than polynomial reduction functions as in \cite{Wagner1987}, is both simpler and more general.
In particular, our reductions do not transform the string $w$, unlike those in \cite{Wagner1987,JennerKL1989,Buss1991,ArvindKM1993}.
This streamlined formulation also enables us to consider Hausdorff predicates of super\nbdash-polynomial complexity, making it possible to apply the Hausdorff perspective naturally to the iterated exponential hierarchies.

A Hausdorff reduction notion closer in spirit to ours, involving similar predicates, was implicitly used (though not formally defined) by \citeauthor{Wagner1990} in \cite{Wagner1990}.
That work again focused on $\NPTime$ Hausdorff languages, but the revised reduction concept made it possible to capture the Hausdorff class $\BoundedHausdCLASS{2^{\PolFunctions}}{\NPTime}$, i.e., the class $\DeltaP{2}$.
Nonetheless, the analysis remained within the \emph{polynomial} setting.
Indeed, \citeauthor{Wagner1990} observed that his results could be extended to the levels $\ThetaP{k}$ and $\DeltaP{k}$ of $\PolHier$, but he did not consider extensions to the exponential hierarchies.

This restriction explains why the Hausdorff reduction notion in \cite{Wagner1990} was less precise and less general than ours, as such precision was not required. %
There, the focus was on $\PolHier$, where we now know that the longest meaningful Hausdorff summations are exponentially\nbdash-long (see \zcref{theo_Hausd_pred_bounded_complexity_imply_bounded_length}).
Hence, there was no need to consider indices $z$ of super\nbdash-polynomial size in the pairs $\pair{w,z}$ of Hausdorff predicates.
For this reason, issues such as defining the complexity of Hausdorff predicates with respect to the size of the string $w$ alone, independently of the index $z$, did not arise.
Thus, the techniques and results developed there were difficult to extend beyond $\PolHier$, since the framework was not sufficiently general to investigate higher exponential hierarchies.

In contrast, for Hausdorff summations of (iterated) exponential length, such distinctions become unavoidable, motivating the more explicit formulation of Hausdorff predicates adopted here (see \zcref{sec_intro_Hausdorff_notions}).
In fact, our seemingly minor technical choice of defining the complexity of Hausdorff predicates with respect to the size of the string $w$ alone turns out to be crucial:
this makes Hausdorff predicates and the resulting Hausdorff classes particularly well suited to capture the intermediate levels of the iterated exponential hierarchies.

It is worth noting that, even at the time, understanding analogous notions in higher exponential hierarchies was already implicitly required.
The questions posed by \citet{Hemachandra1987} on $\SExpHier$ were already known, and addressing them would likely have required extending those Hausdorff notions beyond $\PolHier$.
Nevertheless, the possibility of extending the Hausdorff perspective in this direction appears not to have been noticed.
A plausible explanation is that the available notions and techniques were either too closely tailored to the polynomial setting \cite{Wagner1987,Wagner1988,Buss1988,Buss1991}, or not formulated with sufficient generality and precision to suggest natural extensions \cite{Wagner1990}.
Thus, the generalization developed here does not seem to be an incremental step that could have been immediately apparent from earlier formulations.
Had a suitable generalization been identified at that time, it might already have resolved the questions left open by \citeauthor{Hemachandra1987}.

Having clarified the relation to previous notions, we discuss the insights of the Hausdorff perspective.

\subsection{The Hausdorff perspective: structural insights}
\label{sec_intro_Hausdorff_structural_insights}

Our central technical result shows that the intermediate levels above a main hierarchy level $\ComplexityClass{C}$ coincide exactly with $\ComplexityClass{C}$ Hausdorff classes of increasing length.
We now discuss the structural insights provided by this Hausdorff characterization.
The formal statements of the results leading to it, together with their proof intuitions and differences from previous techniques, are deferred to \zcref{sec_intro_selected_results}.

For an integer $i \geq 0$, we define the \defin{iterated exponential} function $\smash{\iExp{i}{x} = {\underbracket[.5pt]{2^{2^{{\iddots^{\raisebox{-0.7ex}{${\scriptscriptstyle 2}$}}}}}}_{i\mhyphen\text{times}}}^{\mkern -4mu \scriptscriptstyle x}}$, where $\iExp{0}{x} = x$, and $\iExp{-i}{x} = \log^i (x)$---notation freely inspired by \cite{Ginsburg1945,Goodstein1947,Knoebel1981}.

We define the (main) levels of the \iWEHText{i} $\iWExpHier{i}$, for $c \geq 1$, by generalizing $\WExpHier$:
\[
  \SigmaIWExp{i}{c} = \Oracle{\iNExpTime{i}}{\SigmaP{c-1}},
\]
where $\SigmaIWExp{i}{0} = \iExpTime{i}$, %
and we define $\PiIWExp{i}{c} = \ComplementPrefixKerned\SigmaIWExp{i}{c}$, for $c \geq 0$;
$\iWExpHier{i}$ is then defined as $\iWExpHier{i} = \bigcup_{c \geq 0} \SigmaIWExp{i}{c}$.

The intermediate levels of these iterated exponential hierarchies are defined via oracles, analogously to those of \WExpHier.
They correspond to the classes akin to $\ThetaP{c+1}$ and $\DeltaP{c+1}$ in $\PolHier$.
The intermediate levels between the $c$\nbdash-th and $(c{+}1)$\nbdash-th main levels of $\iWExpHier{i}$ are:
\[
  \DeltaIWExpBound{i}{c+1}{\iExpPolFunctions{j}} = \BoundedOracle{\iExpTime{i}}{\SigmaP{c}}{\iExpPolFunctions{j}},
  \quad \text{for all $j$ such that $-1 \leq j \leq i$},
\]
where the notation $\DeltaIWExpBound{i}{c+1}{\, \cdot \,}$ is inspired by~\cite{Mocas1996}, $\PolFunctions$ is the set of polynomial functions~\cite{BalcazarDG1990,Mocas1996}, $\iExpPolFunctions{j}$ is the set of functions $\iExp{j}{p(n)}$ with $p(n) \in \PolFunctions$, and $\BoundedOracle{\iExpTime{i}}{\SigmaP{c}}{\iExpPolFunctions{j}}$ is the class of languages decided in \iExponential{i} time by deterministic oracle machines issuing at most $\iExpPolFunctions{j}$ many (adaptive) queries to a $\SigmaP{c}$ oracle;%
\footnote{For now, we are intentionally not considering the case of parallel queries, as the latter will also be reconciled into one cohesive big picture obtained via the Hausdorff perspective.}
we let $\iExpPolFunctions{-1} = \LogFunctions$, i.e., the set of logarithmic functions~\cite{BalcazarDG1990,Mocas1996}.
A hierarchy's \defin{$c$\nbdash-th step} is the union of the $c$\nbdash-th main level and all the intermediate ones between the $c$\nbdash-th and $(c{+}1)$\nbdash-th main levels.
Notice that $\iWExpHier{0} = \PolHier$ and $\iWExpHier{1} = \WExpHier$, level by level.

A key result linking Hausdorff classes to intermediate hierarchy levels states that:
\begin{equation}\label{eq_first_insight}
    \BoundedOracle{\iExpTime{i}}{\SigmaP{c}}{\iExpPolFunctions{j}} = \BoundedHausdCLASS{\iExpPolFunctions{j+1}}{\Oracle{\iNExpTime{i}}{\SigmaP{c-1}}},
    \quad \text{for $i \geq 0$, $j \geq -1$, and $j \leq i$}.
\end{equation}

Thus, each intermediate level between the $c$\nbdash-th and $(c{+}1)$\nbdash-th main levels of $\iWExpHier{i}$, where at most $\iExpPolFunctions{j}$ queries are allowed, coincides with a Hausdorff class.
Specifically, it is the class of $\Oracle{\iNExpTime{i}}{\SigmaP{c-1}}$ Hausdorff languages of $\iExpPolFunctions{j+1}$ length.
This Hausdorff characterization reveals a uniform underlying structure in the iterated exponential hierarchies.
Each step consists of a main level (Hausdorff classes of length~1), its Boolean closure (Hausdorff classes of constant length), and a sequence of intermediate levels given by Hausdorff classes of increasing length.
Moreover, higher\nbdash-order exponential hierarchies admit more intermediate levels, as $\iExpTime{i}$ oracle machines may issue up to \iExponential{i}{}ly many queries.
This is captured by \zcref{theo_Hausd_pred_bounded_complexity_imply_bounded_length}, which implies that the complexity of a Hausdorff predicate by itself bounds its possible Hausdorff length.
Specifically, $\Oracle{\iNExpTime{i}}{\SigmaP{c-1}}$ predicates can have length at most \iExponential{(i+1)}.
The Hausdorff length bound thus reflects that the number of intermediate levels is bounded, and that larger bounds yield more such levels.

This uniformity suggests viewing the iterated exponential hierarchies as a hierarchy of hierarchies, which we may call the \defin{Iterated Exponentials (Meta\nbdash-){}Hierarchy}.
Its (meta-)levels correspond to hierarchies of increasing orders of exponentiation.
This meta\nbdash-hierarchy is precisely $\Elementary = \bigcup_{i \geq 1} \DTime{\iExp{i}{n}}$~\cite{Simon1975,Papadimitriou1994}.
\zcref{fig_iterated_exponentials_meta-hierarchy} 
illustrates several of its levels together with their corresponding Hausdorff classes.

Like main hierarchy levels, defined via alternating quantifiers and certificates without resorting to oracles, Hausdorff classes provide oracle\nbdash-free definitions of intermediate levels.
Main hierarchy levels are characterized by two ingredients:
the size of the certificates and the number of quantifier alternations.
The Hausdorff characterization of intermediate levels likewise depends on two ingredients:
the main level on which they sit (whose definition does not require oracles) and the length of the Hausdorff languages in the class.
For this reason, unlike oracle\nbdash-based definitions, where seemingly different oracle classes may coincide (e.g., $\ParOracle{\PTime}{\NPTime} = \LogOracle{\PTime}{\NPTime} = \Oracle{\LogSpace}{\NPTime}$), Hausdorff classes identify these levels uniquely (assuming no hierarchy collapse).

The Hausdorff characterization reveals a unifying insight.
If $\Language{L}$ is a language from an intermediate hierarchy level above a main level $\ComplexityClass{C}$, then $\Language{L}$ belongs to a Hausdorff class $\BoundedHausdCLASS{g(n)}{\ComplexityClass{C}}$, for some length function $g(n)$.
Thus, deciding $\Language{L}$ reduces to determining the Hausdorff index $\HausdIndex{w}{\Language{D}}$ of an input $w$ with respect to a suitable Hausdorff predicate $\Language{D}$, and checking its parity.
This yields three natural oracle\nbdash-based approaches to determine $\HausdIndex{w}{\Language{D}}$:
\begin{enumerate}[noitemsep,label=(\roman*)]
  \item querying an oracle for $\Language{D}$ in parallel on all predicates $\Language{D}(w,1),\dots,\Language{D}(w,g(\StringLength{w}))$;
  \item computing $\HausdIndex{w}{\Language{D}}$ via a binary search over $[1,g(\StringLength{w})]$ using an oracle for $\Language{D}$; and
  \item guessing $\HausdIndex{w}{\Language{D}}$ and verifying the guess with \emph{two} (parallel) oracle queries to an oracle for $\Language{D}$.
\end{enumerate}

This explains many equivalences between apparently different oracle classes that correspond to the same intermediate levels:
they correspond to different ways of deciding the languages of the same Hausdorff class, namely by determining the Hausdorff index in different ways.

For example, $\ParOracle{\PTime}{\NPTime} = \LogOracle{\PTime}{\NPTime}$ follows because both classes equal $\BoundedHausdCLASS{\PolFunctions}{\NPTime}$ and they are its oracle classes of type~(i) and type~(ii), respectively.
Similarly, $\PNExp = \NPNExp$, underlying the collapse of $\SExpHier$, can be explained via the Hausdorff perspective: both classes equal $\BoundedHausdCLASS{2^{\PolFunctions}}{\NExpTime}$ and they are its type~(ii) and type~(iii) oracle classes, respectively---%
for more on what the Hausdorff perspective tells us about $\SExpHier$ see \zcref{sec_intro_SEH}.

Another key result enabling us to chart the iterated exponential hierarchies is:
\begin{equation}\label{eq_second_insight}
  \Oracle{\iNExpTime{i}}{\Oracle{\iNExpTime{j}}{\SigmaP{c-1}}}
  = \BoundedHausdCLASS{\iExpPolFunctions{i+1}}{\Oracle{\iNExpTime{(i+j)}}{\SigmaP{c-1}}},
  \quad \text{for $i \geq 0$  and  $j \geq 1$.}
\end{equation}

This result shows that such ``doubly\nbdash-nondeterministic'' oracle classes $\Oracle{\iNExpTime{i}}{\Oracle{\iNExpTime{j}}{\SigmaP{c-1}}}$ coincide with Hausdorff classes and hence, together with \zcref{eq_first_insight}, correspond to specific intermediate levels. %
To illustrate, consider $\Oracle{\iNExpTime{3}}{\iNExpTime{2}}$, $\Oracle{\iNExpTime{2}}{\iNExpTime{3}}$, and $\Oracle{\iNExpTime{4}}{\NExpTime}$.
All three involve nondeterministic oracles running in \iExponential{5} time (since the caller may issue long queries), but differ in the running time of the caller machines.
More precisely, their different running times yield different guessing capabilities, suggesting that these classes are distinct.
This intuition is confirmed by our results, which show that they form three distinct intermediate levels of the first step of $\iWExpHier{5}$.
From the Hausdorff perspective, this difference corresponds precisely to the ability of the caller machine to guess Hausdorff indices of different sizes, i.e., to act as a type~(iii) oracle, and hence to decide Hausdorff languages of different lengths.
In particular (see \zcref{sec_discussion_Hausdorff_perspective}):
\begin{gather*}
\begin{aligned}
  \Oracle{\iNExpTime{3}}{\iNExpTime{2}} = \BoundedHausdCLASS{\iExpPolFunctions{4}}{\iNExpTime{5}} = \BoundedOracle{\iExpTime{5}}{\NPTime}{\iExpPolFunctions{3}}; & \quad &
  \Oracle{\iNExpTime{2}}{\iNExpTime{3}} = \BoundedHausdCLASS{\iExpPolFunctions{3}}{\iNExpTime{5}} = \BoundedOracle{\iExpTime{5}}{\NPTime}{\iExpPolFunctions{2}};
\end{aligned}\\
  \Oracle{\iNExpTime{4}}{\NExpTime} = \BoundedHausdCLASS{\iExpPolFunctions{5}}{\iNExpTime{5}} = \BoundedOracle{\iExpTime{5}}{\NPTime}{\iExpPolFunctions{4}}.
\end{gather*}

Altogether, the Hausdorff perspective shows that the structure of the intermediate hierarchy levels is governed by the length of the Hausdorff classes that characterize them, while different oracle models correspond to different ways of determining the Hausdorff index.

We next present selected results, formalizing these structural insights and outlining their main proof ideas.

\subsection{A selection of the results obtained}
\label{sec_intro_selected_results}

In this section, we present selected results, focusing on their consequences and their relation to previous work.
After stating them, we highlight the key ideas underlying their proofs, especially where new techniques are required; full details are deferred to subsequent sections.

\subsubsection{Intermediate levels are indeed Hausdorff classes}
\label{sec_intro_intermediate_levels_are_Hausdorff_classes}

The Hausdorff characterization of the intermediate hierarchy levels, anticipated in \zcref{eq_first_insight}, follows as a corollary of the theorem below.
In the statement, $\DoubleBoundedParOracle{\ComplexityClass{X}}{\ComplexityClass{Y}}{r(n)}{s(n)}$ denotes the class of languages decided by $\ComplexityClass{X}$ oracle machines that issue at most $s(n)$ rounds of queries, each with at most $r(n)$ parallel queries, to an oracle in $\ComplexityClass{Y}$. 
For $\DoubleBoundedPlusParOracle{\ComplexityClass{X}}{\ComplexityClass{Y}}{r(n)}{s(n)}$, an additional query is allowed in the first round.

\getkeytheorem{ENcontainment}

\zcref[S]{theo_exp_nexp_containment} generalizes the classical equalities
$\LogOracle{\PTime}{\NPTime} = \BoundedHausdCLASS{\PolFunctions}{\NPTime}$
and
$\Oracle{\PTime}{\NPTime} = \BoundedHausdCLASS{2^\PolFunctions}{\NPTime}$~\cite{Buss1988,Wagner1990} to all levels of the \iWEHsText{i}.
\zcref[S]{theo_exp_nexp_containment} comprises two inclusions.
The proof of the second uses a ``generalized'' binary search to compute the Hausdorff index of an input string.
At each round, $r(n)$ parallel queries (instead of~$1$) are issued to a Hausdorff predicate oracle.
This is similar to the simpler binary\nbdash-search arguments used to show
$\BoundedHausdCLASS{\PolFunctions}{\NPTime} \subseteq \LogOracle{\PTime}{\NPTime}$
and
$\BoundedHausdCLASS{2^{\PolFunctions}}{\NPTime} \subseteq \Oracle{\PTime}{\NPTime}$~\cite{Buss1988,Buss1991,Wagner1988,Wagner1990}.

For the first inclusion, instead, the techniques in the literature proving
$\LogOracle{\PTime}{\NPTime} \subseteq \BoundedHausdCLASS{\PolFunctions}{\NPTime}$
and
$\Oracle{\PTime}{\NPTime} \subseteq \BoundedHausdCLASS{2^{\PolFunctions}}{\NPTime}$ do not extend beyond the \emph{polynomial} setting, as they rely on two key features:
Hausdorff reductions are formulated as \emph{polynomial} tt\nbdash-reductions, and both the caller machine and the oracle run in \emph{polynomial} time.
This allows one to encode each oracle computation into a Boolean formula, combine these formulas into a larger one, and analyze its ``mind changes'' as the number of positive oracle answers increases~\cite{Buss1988,Buss1991,Wagner1988,Wagner1990}.
This proof strategy works because oracle computations can be encoded into Boolean formulas of polynomial length (via \citeauthor{Cook1971}'s theorem~\cite{Cook1971}), and the number of oracle queries are polynomially bounded.
For this reason, this approach does not readily generalize to higher exponential hierarchies.
First, exponential\nbdash-time oracle computations cannot be encoded into polynomially\nbdash-long Boolean formulas.
Second, exponential\nbdash-time caller machines may issue more than polynomially\nbdash-many queries, so the assembled formula would be too large.
Thus, the earlier approach is too closely tied to the polynomial setting to capture the intermediate levels of the iterated exponential hierarchies or the collapse of $\SExpHier$.

Our proof takes a different approach.
Since our reductions are no longer tt\nbdash-reductions, we cannot rely on Boolean encodings of oracle computations, and thus cannot appeal to mind changes or directly extend previous techniques.
Instead, we work directly with the oracle computations.

At a high level, the goal is as follows.
For a language
\(
\Language{L} \in \DoubleBoundedParOracle{\iExpTime{i}}{\Oracle{\iNExpTime{j}}{\SigmaP{c-1}}}{r(n)}{s(n)},
\)
let $\Machine{M}$ be an oracle machine and $\Omega$ an oracle such that $\LanguageOf{\Oracle{\Machine{M}}{\Omega}} = \Language{L}$.
We want to show that
\(
\Language{L} \in \BoundedHausdCLASS{{(r(n)+1)}^{s(n)}}{\Oracle{\iNExpTime{(i+j)}}{\SigmaP{c-1}}}.
\)
A natural approach would be to transform the computation of $\Oracle{\Machine{M}}{\Omega}$ into a single $\Oracle{\iNExpTime{(i+j)}}{\SigmaP{c-1}}$ predicate.
This predicate could then be decided by a nondeterministic machine $\Machine{N}$.
However, a direct simulation of $\Oracle{\Machine{M}}{\Omega}$ by $\Machine{N}$ encounters a fundamental obstacle:
$\Machine{N}$ can verify \emph{positive} oracle answers by guessing accepting computations of $\Omega$ on queries issued by $\Machine{M}$, but \emph{negative} oracle answers lack equally direct witnesses.
Thus, the main obstacle is not the oracle's running time, but the asymmetry between positive and negative answers.

To overcome this difficulty, we exploit the fact that our goal is to characterize membership of $w$ in $\Language{L}$ via the parity of a Hausdorff index with respect to a suitable Hausdorff predicate.
This allows us to avoid reconstructing the oracle computation exactly in a single predicate.

The key idea is to consider \emph{all} computations consistent with the caller machine's transition function, regardless of the correctness of their oracle answers, and then identify the correct one.
To this end, we introduce \emph{\ounawarelegal} computations, which are sequences of configurations consistent with the transition function of the caller machine, whose oracle answers need not be correct.

Observe that all \ounawarelegal computations of $\Machine{M}$ on input $w$ can be totally ordered and thus enumerated.
The order is induced by a vector that records, for each round of parallel queries, the number of \yesansws supposedly returned to $\Machine{M}$ by its oracle in the \ounawarelegal computation.
The lexicographic order of these vectors yields a total order on all such computations.

Let us call an \ounawarelegal computation $\pi$ \emph{plausible} if every oracle \yesansw in $\pi$ is correct.
Importantly, plausibility can be checked within
$\Oracle{\iNExpTime{(i+j)}}{\SigmaP{c-1}}$,
as it suffices to verify only positive oracle answers. %

Crucially, under the above lexicographic order, the real computation of $\Oracle{\Machine{M}}{\Omega}$ on input $w$ can be identified purely combinatorially:
\emph{the real computation of $\Machine{M}$ on $w$, using the actual oracle answers of $\Omega$, is the maximum plausible computation.}
Indeed, any \ounawarelegal computation that assumes more positive oracle
answers than the real one is implausible, whereas the real computation remains plausible and is the last such computation.

This characterization allows us to define the Hausdorff predicate.
The predicate $\Language{D}(w,z)$ detects whether the $z$-th computation in the ordering is the real computation of $\Oracle{\Machine{M}}{\Omega}$ on $w$.
Specifically, $\Language{D}(w,z)$ checks whether a plausible \ounawarelegal computation with index greater than $z$ exists.
This test can be carried out within $\Oracle{\iNExpTime{(i+j)}}{\SigmaP{c-1}}$.
If no such computation exists, the $z$\nbdash-th computation must be the real computation of $\Oracle{\Machine{M}}{\Omega}$ on input $w$, and we set $\Language{D}(w,z)$ to reflect whether $w \in \Language{L}$, up to a parity adjustment.
Thus, the actual oracle computation on $w$ is encoded in the Hausdorff index, and deciding whether $w \in \Language{L}$ reduces to computing this index and checking its parity.
This yields the desired Hausdorff characterization.

\subsubsection{Charting the exponential hierarchies}
\label{intro_charting_exp_hierarchies}

An important consequence of the Hausdorff characterization is that it provides a uniform and unambiguous way to identify the intermediate levels of the iterated exponential hierarchies within \Elementary.
This allows us to study these hierarchies without becoming entangled in the intricate web of equivalences between oracle classes exhibited in previous work (see, e.g., \cite{SchoningW88,Hemachandra1989,Hemaspaandra1994} and references therein).
The results ``charting'' these hierarchies are discussed in detail in \zcref{sec_charting_top}.
A corollary of the following theorem yields the characterization anticipated in \zcref{eq_second_insight}, as well as our proof of the equivalence $\PNExp = \NPNExp$.

\getkeytheorem{NNcontainment}

The proof of the first inclusion faces the same obstacle as that of \zcref{theo_exp_nexp_containment}, and uses ideas reminiscent of the pseudo\nbdash-complement technique~\cite{Mahaney82,Kadin1989} to define the Hausdorff predicate.

The goal is to show that if
\(
\Language{L} \in \Oracle{\iNExpTime{i}}{\Oracle{\iNExpTime{j}}{\SigmaP{c-1}}}
\)
then
\(
\BoundedHausdCLASS{\iExpPolFunctions{i+1}}{\Oracle{\iNExpTime{(i+j)}}{\SigmaP{c-1}}}.
\)
Let $\Machine{M}$ be an oracle machine and $\Omega$ be an oracle such that $\LanguageOf{\Oracle{\Machine{M}}{\Omega}} = \Language{L}$.
As before, a direct simulation of $\Oracle{\Machine{M}}{\Omega}$ by a single $\Oracle{\iNExpTime{(i+j)}}{\SigmaP{c-1}}$ predicate fails due to the asymmetry between positive and negative oracle answers.
Again, we avoid exact simulation by reducing membership of $w$ in $\Language{L}$ to the parity of the Hausdorff index of $w$ with respect to a suitable Hausdorff predicate.
This allows us to reason about all \ounawarelegal computations of $\Machine{M}$ on $w$.

Unlike in \zcref{theo_exp_nexp_containment}, where the real computation is identified via an ordering of \ounawarelegal ones, we instead single it out by validating oracle answers.
While positive answers can be verified within $\Oracle{\iNExpTime{(i+j)}}{\SigmaP{c-1}}$, negative ones cannot, which prevents direct validation.

We overcome this as follows.
Since $\Machine{M}$ runs in time $\iExp{i}{p(\StringLength{w})}$, for some polynomial $p(n)$, all its queries have length at most $\iExp{i}{p(\StringLength{w})}$.
Thus, it suffices to determine the behavior of $\Omega$ on queries of bounded length.
The key idea is to reason about \emph{sets of queries} answered positively by the oracle, thereby avoiding the need to validate negative answers.
Let $Y$ be a set of queries of length at most $\iExp{i}{p(\StringLength{w})}$.
Within $\Oracle{\iNExpTime{(i+j)}}{\SigmaP{c-1}}$, we can check whether all queries in $Y$ are answered positively by $\Omega$.
This allows us to define a Hausdorff predicate that identifies the real computation via the Hausdorff index.

The predicate $\Language{D}(w,z)$ is defined as follows.
First, we check whether there exists a set of more than $z$ queries of length at most $\iExp{i}{p(\StringLength{w})}$ answered positively by $\Omega$.
If no such set exists, then there are at most $z$ such queries, and hence exactly $z$.
We then guess a set $Y$ containing exactly these queries and check its correctness.
Given $Y$, we guess an \ounawarelegal computation $\pi$ of $\Machine{M}$ on $w$ and verify that all oracle answers in $\pi$ are consistent with $Y$.
This validates both positive and negative answers, since $Y$ captures all queries answered positively by $\Omega$, and $\Machine{M}$ does not ask longer queries.
If successful, $\pi$ is the real computation of $\Machine{M}$ on $w$ with oracle $\Omega$.
We then set $\Language{D}(w,z)$ according to whether $w \in \Language{L}$, up to a parity adjustment.

The remaining inclusion relationships of \zcref{theo_nexp_nexp_containment} follow more or less directly from the definition of Hausdorff languages via type~(iii) oracle procedures.

The same approach extends to space\nbdash-bounded oracle classes, and Hausdorff classes can be related to
$\Oracle{\iExpSpace{(i-1)}}{\Oracle{\iNExpTime{j}}{\SigmaP{c-1}}}$.
We adopt the deterministic query model~\cite{RuzzoST84}, which is equivalent to the unrestricted model when deterministic space oracle classes are considered~\cite{LadnerL76}.
The unrestricted query model has also been used beyond \LogSpace (see, e.g., \cite{Hemaspaandra1994,Gottlob1995,Dawar1998}).

\getkeytheorem{SNcontainment}

At first glance, one might expect the argument for \zcref{theo_nexp_nexp_containment} to apply directly and yield the same bound.
Indeed, an $\iExpSpace{(i-1)}$ machine can run for \iExponential{i} time, issue queries of up to \iExponential{i} length, and hence there are \iExponential{(i+1)}{}ly\nbdash-many distinct such queries.
However, this would instead yield the weaker bound
\(
\Oracle{\iExpSpace{(i-1)}}{\Oracle{\iNExpTime{j}}{\SigmaP{c-1}}}
\subseteq
\BoundedHausdCLASS{\iExpPolFunctions{i+1}}{\Oracle{\iNExpTime{(i+j)}}{\SigmaP{c-1}}}.
\)

The stronger bound follows from a key property of deterministic space\nbdash-bounded computations:
across all \ounawarelegal computations, the number of \emph{distinct} queries is only \iExponential{i}, even though queries may have \iExponential{i} length.
Thus, the construction used for \zcref{theo_nexp_nexp_containment} can be refined by restricting the guessed set $Y$ to queries that actually occur in \ounawarelegal computations.
Hence, $Y$ ranges over only \iExponential{i}{}ly\nbdash-many queries, yielding the tighter Hausdorff length bound;
see \zcref{sec_expspace_nexp_oracle_classes} for details.

\subsubsection{The Strong Exponential Hierarchy}
\label{sec_intro_SEH}

The Hausdorff perspective also resolves the questions about $\SExpHier$ raised by \citet{Hemachandra1989}.
Indeed, the general results above yield Hausdorff characterizations of the $\SExpHier$ levels (see \zcref{sec_top_seh} for details).
Similarly to $\LogOracle{\PTime}{\NPTime}$, which is considered part of $\PolHier$ by \citet{Wagner1990}, $\PNExpLog$ is here regarded as a level of $\SExpHier$.

\getkeytheorem{ThetaLevelPolHausdorff}

\getkeytheorem{DeltaLevelExpHausdorff}

These characterizations yield a structural explanation for $\PNExp = \NPNExp$ and hence for the collapse of $\SExpHier$.
Indeed, both classes coincide with the Hausdorff class $\BoundedHausdCLASS{2^{\PolFunctions}}{\NExpTime}$, and thus correspond to different ways of deciding languages in that class:
$\PNExp$ and $\NPNExp$ are type~(ii) and type~(iii) oracle classes, respectively.

More broadly, the Hausdorff perspective shows that $\SExpHier$ is not a hierarchy of main levels, but of \emph{intermediate} ones.
Specifically, the classes
$\BoundedHausdCLASS{\PolFunctions}{\NExpTime}$,
$\BoundedHausdCLASS{2^{\PolFunctions}}{\NExpTime}$, and
$\BoundedHausdCLASS{2^{2^{\PolFunctions}}}{\NExpTime}$
are the intermediate levels of the first step of $\WExpHier$.
Thus, $\SExpHier$ corresponds to a portion of that first step (see \zcref{fig_iterated_exponentials_meta-hierarchy}).
Although \citet{Hemachandra1989} observed that higher levels of $\SExpHier$ are no harder than lower levels of $\WExpHier$, this connection was not explicitly recognized.
This is likely due to the lack of a notion, such as Hausdorff classes, that reveals the equivalence between $\SExpHier$ levels and the intermediate levels of $\WExpHier$.
Notice that \citeauthor{Hemachandra1989} did consider the class $\PNExpLog$, but only showed that it coincides with languages decided by $\NPTime$ oracle machines with a $\mi{sparse}\mhyphen\NExpTime$ oracle~\cite[Theorem~4.11, Part~2]{Hemachandra1989}.

Since $\NPNExp = \BoundedHausdCLASS{2^\PolFunctions}{\NExpTime}$, the certificate characterization of $\SExpHier$ sought by \citet{Hemachandra1989} follows directly.
This characterization is also implicit in the $\PNExph${}ness proof for the Extended Tiling Problem in~\cite{EiterTechRep2016}, although their argument relies on $\PNExp = \NPNExp$.
Here, it follows from the fact that a language in $\NPNExp$ can be decided by an $\NPTime$ oracle machine that guesses the Hausdorff index of the input and verifies it via two parallel queries to a $\NExpTime$ oracle.

\getkeytheorem{NPNexpCertificates}

An alternating-machine characterization of $\NPNExp$ appears less natural, since alternation is better suited to capture main hierarchy levels rather than intermediate ones (see \zcref{sec_delta_level_hausdorff}).
In contrast, the Hausdorff characterization directly reflects the intrinsic structure of the class.

Finally, several previously known characterizations of $\PNExp$ (see the first page of the introduction) follow directly from the Hausdorff perspective developed here (see \zcref{sec_delta_level_hausdorff}).

\subsubsection{Hard problems}
\label{sec_intro_hard_problems}

In \zcref{sec_canonical_hard_problems}, we use Hausdorff classes to obtain canonical hard problems for the intermediate levels of the iterated exponential hierarchies.
We give canonical complete problems for the first intermediate level of each step, and a family of canonical problems for all intermediate levels except the first and the last of the first step.

In \zcref{sec_qbsf_hard_problems}, for all the intermediate levels of $\WExpHier$, we exhibit hard problems based on Quantified Boolean Second\nbdash-order Formulas~(QBSFs)~\cite{Luck2016-techrep,Jiang2023}, which are quantified Boolean formulas where Boolean functions can be quantified as well.
We consider formulas of the following form:
\[
  \Phi(\VarSet{g},\VarSet{y}) =
    (\SOQ_1 \VarSet{f}^1) \cdots (\SOQ_n \VarSet{f}^n)
    (\FOQ_1 \VarSet{x}^1) \cdots (\FOQ_m \VarSet{x}^m)\;
      \phi(\VarSet{f}^1,\dots,\VarSet{f}^n,\VarSet{g},\VarSet{x}^1,\dots,\VarSet{x}^m,\VarSet{y}),
\]
where $\VarSet{f}^i$ and $\VarSet{g}$ are Boolean function variables, $\VarSet{x}^j$ and $\VarSet{y}$ are propositional variables, $\SOQ_i,\FOQ_j \in \set{\exists,\forall}$ are second- and first-order alternating quantifiers, and $\phi$ is quantifier\nbdash-free.
An example is:
\[
\Phi(h,z) = (\exists f,g) (\forall x) (\exists y)\; (y \lor \lnot f(z,x) \land \lnot z \leftrightarrow h(y,0,x)) \rightarrow (g(1,y,z) \land \lnot x \lor f(y,y) \land 1).
\]

\SigmaKFormulas{c} are \QBSFs with $c$ alternating second\nbdash-order quantifiers, starting with $\exists$.
Deciding their validity is complete for $\SigmaWExp{c} = \Oracle{\NExpTime}{\SigmaP{c-1}}$~\cite{Lohrey2012,Luck2016-techrep}.
We define the following problems:

\Problem{\MaxSatSigmaFormula{c}}%
{A satisfiable \SigmaKFormula{c} $\Phi(\VarSet{y})$, where $\VarSet{y} = \set{y_1,\dots,y_n}$ is a set of propositional variables.}%
{Is the Hamming weight of a maximum-Hamming-weight model of $\Phi(\VarSet{y})$ odd/even?}

\Problem{\LexMaxSigmaFormula{c}}%
{A satisfiable \SigmaKFormula{c} $\Phi(\VarSet{y})$, where $\VarSet{y} = \set{y_n,\dots,y_1}$ is an ordered set of propositional variables.}%
{Does the lexicographic-maximum model of $\Phi(\VarSet{y})$ assign $\valtrue$/$\valfalse$ to $y_1$?} 

\Problem{\LexMaxFuncSigmaFormula{c}}%
{A satisfiable \SigmaKFormula{c} $\Phi(\VarSet{g})$, where $\VarSet{g} = \set{g_n,\dots,g_1}$ is an ordered set of function variables.}%
{Does the lexicographic-maximum model $\Interpr{I}$ of $\Phi(\VarSet{g})$  satisfy $\EvalInterpr{g_1}{\Interpr{I}}(0,\dots,0) = \valtrue$/$\valfalse$?}

These problems are shown complete for $\BoundedOracle{\ExpTime}{\SigmaP{c}}{\LogFunctions}$, $\BoundedOracle{\ExpTime}{\SigmaP{c}}{\PolFunctions}$, and $\Oracle{\ExpTime}{\SigmaP{c}}$, respectively (see \zcref{sec_qbsf_hard_problems}).
Hardness follows from Hausdorff characterizations and \QBSF encodings of $\Oracle{\NExpTime}{\SigmaP{c-1}}$ computations.

\medskip

In \zcref{sec_datalog_hard_problems}, we obtain matching lower bounds for problems in $\PNExpLog$ and $\PNExp$~\cite{CeylanLMMV2021} whose hardness remained open due to the lack of suitable hard problems to reduce from.
In particular, no \PNExpLogc{} problems were known at all, while the known \PNExpc{} ones had an ill\nbdash-suited \NPNExp{} flavor.

These problems concern \DatalogPM knowledge bases~\cite{CaliGL2012}, i.e., pairs $\KB = \KBDetails$ of a relational database $\DB$ and a set of rules $\Dep$.
A Boolean conjunctive query (\BCQ) $\Query$ over $\KB$ is evaluated over $\DB$ extended with the consequences of $\Dep$.
A $\preccurlyeq$-minimal explanation ($\preccurlyeq$-\Minex) for $\Query$ is a subset of $\DB$ that entails $\Query$ via $\Dep$ and is minimal according to some criterion $\preccurlyeq$~\cite{CeylanLMV2019,CeylanLMMV2021}, such as cardinality ($\leq$) or weight of the facts ($\leq_w$).

We consider three problems over $\preccurlyeq$-\Minexs~\cite{CeylanLMMV2021}:
(i)~\textsc{MinEx-Rel}, asking whether some $\preccurlyeq$\nbdash-\Minex contains a fact $f$;
(ii)~\textsc{MinEx-Nec}, asking whether all $\preccurlyeq$\nbdash-\Minexs contain $f$;
(iii)~\textsc{MinEx-Irrel}, asking whether some $\preccurlyeq$\nbdash-\Minex avoids all facts in given forbidden sets.
We show that these problems are \PNExpLog-complete and \PNExp-complete, in the $\mi{ba}$\nbdash-combined complexity, for $\leq$\nbdash-\Minexs and $\leq_w$\nbdash-\Minexs, respectively (see \zcref{sec_datalog_hard_problems}).

\subsection{Organization of the paper}

We start by providing preliminaries in \zcref{sec_specific_preliminaries}, where the iterated exponential hierarchies are introduced;
more detailed background on notation and basic notions used throughout the paper is given in \zcref{sec_general_preliminaries}.
We then introduce Hausdorff reductions and classes in \zcref[S]{sec_Hausdorff_reductions_classes}.
Next, we study the iterated exponential hierarchies in \zcref{sec_charting_top}, presenting the main results linking oracle classes to Hausdorff classes and analyzing $\SExpHier$, thereby resolving \citeauthor{Hemachandra1989}'s open questions.
Finally, \zcref{sec_hard_problems} presents canonical complete problems for the intermediate hierarchy levels and establishes the $\PNExpLog$-hardness and the $\PNExp$-hardness results.

\begingroup %

\section{Preliminaries}
\label{sec_specific_preliminaries}

In this section, we first define the iterated exponential functions.
We next recall the complexity classes here considered and two classical complexity hierarchies. These hierarchies will then be generalized to obtain the iterated exponential hierarchies.
The families of logarithmic and polynomial functions are~\cite{BalcazarDG1990,Mocas1996}:
\begin{align*}
  \LogFunctions & = \set{f\colon \NaturalsDomain \to \NaturalsDomain \mid f(n) = c \cdot \lceil \log_2 n \rceil, \text{ for some integer constant }c \geq 1} \\
  \PolFunctions & = \set{f\colon \NaturalsDomain \to \NaturalsDomain \mid f(n) = c \cdot n^k, \text{ for some integer constants }c,k \geq 1}.  
\end{align*}

For two arbitrary function $g, h \colon \NaturalsDomain \to \NaturalsDomain$, when we write $g(n) \leq h(n)$ (resp., $g(n) < h(n)$) we mean that, for all integers $n \geq 1$, it holds $g(n) \leq h(n)$ (resp., $g(n) < h(n)$).
Let $w$ be a string, $\StringLength{w}$ is its length.
By the \defin{(canonical) binary representation} of an integer $n \geq 0$ we mean its binary form whose most significant bit is `$1$' when $n \geq 1$, i.e., there are \emph{no} leading `$0$'s in this representation, and is just ``0'' (one bit) when $n = 0$~\cite{Kobayashi1985}.
Unless stated otherwise, non\nbdash-negative integers~$n$ are assumed to be represented in their canonical binary form starting with the most significant bits, i.e., most significant bits appear on the left of the tapes;
$\StringLength{n}$ is the length of this representation.
Analogously, unless stated otherwise, generic lists/sequences of symbols/numbers are lexicographically ordered by looking at left\nbdash-most elements first.

\subsection{Iterated Exponentials}
\label{sec_prelim_iterated_exponentials}

For an integer $i$, let the \defin{iterated exponential} function $\iExp{i}{x}$ be (notation freely inspired by \cite{Ginsburg1945,Goodstein1947,Knoebel1981}):
\[
\iExp{i}{x} = 
\begin{cases}
  {\underbracket[.5pt]{2^{2^{2^{\iddots^{\raisebox{-0.7ex}{${\scriptscriptstyle 2}$}}}}}}_{i\mhyphen\text{times}}}^{\mkern -4mu \scriptscriptstyle x}, & \text{if } i > 0 \\
  x, & \text{if } i = 0 \\
  \underbracket[.5pt]{\log \cdots \log}_{i\mhyphen\text{times}} x, & \text{if } i < 0.
\end{cases}
\]

For notational convenience, for an integer $i \geq 1$, in lieu of $\iExp{-i}{x}$, we may write $\smash{\iLog{i}{x} = \overbracket[.5pt]{\log \cdots \log}^{{i\mhyphen\text{times}}} x}$, i.e., an \defin{iterated logarithm}.
In this paper, the iterated exponential functions $\iExp{i}{f(x)}$ of interest have $i \geq -1$.

We denote by $\iExpPolFunctions{j}$ the set of functions $\iExp{j}{p(n)}$ with $p(n) \in \PolFunctions$, and we let $\iExpPolFunctions{-1} = \LogFunctions$.

Since iterated exponentials characterize the complexity class definitions below, we highlight few inclusion relationships between some asymptotic classes involving them.
The results below are not necessarily tight, but are those useful for our subsequent discussion.
The proofs are deferred to \zcref{sec_detailed_proofs_properties_iterated_exponentials}.

\begin{lemma}[store=asymptoticClassPropertiesIteratedExponentialsNonContainment]
\label{theo_asymptotic_class_properties_iterated_exponentials_non_containment}
Let $i \geq 1$, $j \geq 0$, and $k \geq 1$, be integers.
Then, $\smash{\iExp{i}{O(\iExp{j}{n^k})} \nsubseteq O(\iExp{i+j}{n^{k}})}$.
\end{lemma}

\begin{lemma}[store=asymptoticClassPropertiesIteratedExponentialsContainment]
\label{theo_asymptotic_class_properties_iterated_exponentials_containment}
Let $i, j, k \geq 0$ be integers.
Then, $\smash{\iExp{i}{O(\iExp{j}{n^k})} \subseteq O(\iExp{i+j}{n^{k+1}})}$.
\end{lemma}

\begin{lemma}[store=operationsBetweenIteratedExponentials]
\label{theo_operations_between_iterated_exponentials}
\label{theo_polynomial_of_iterated_exponentials}
\label{theo_multiplication_between_iterated_exponentials}
\label{theo_addition_between_iterated_exponentials}
\label{theo_exponentiation_between_iterated_exponentials}
\label{theo_composition_iterated_expentials}
Let $i, j \geq -1$ be integers, and let $\smash{f(n) \in O(\iExpPolFunctions{i})}$ and $\smash{g(n) \in O(\iExpPolFunctions{j})}$ be functions.
Then,
\begin{enumerate}[nosep,label=\arabic*)]
  \item $O(c \cdot \! {f(n)}^k) \subseteq O(\iExpPolFunctions{\max\set{0,i}})$, for integer constants $c,k \geq 1$;
  \item $O(f(n) + g(n)) \subseteq O(\iExpPolFunctions{\max\set{i,j}})$;
  \item $O(f(n) \cdot g(n)) \subseteq O(\iExpPolFunctions{\max\set{0,i,j}})$;
  \item $O({f(n)}^{g(n)}) \subseteq O(\iExpPolFunctions{\max\set{1,i,j+1}})$; and
  \item $O(f(g(n))) \subseteq O(\iExpPolFunctions{\max\set{i,i+j}})$ (by which, for $i,j \geq 0$, it holds $\iExp{i}{O(\iExpPolFunctions{j})} \subseteq O(\iExpPolFunctions{i+j})$).
\end{enumerate}
\end{lemma}

\subsection{Turing Machines and Central Complexity Classes}
\label{sec_prelim_central_complexity_classes}

Unless differently stated, we consider multi\nbdash-tape Turing machines with bidirectional semi\nbdash-infinite tapes with a \emph{read-only} input tape and (possibly several) \emph{read/write} work tapes.
As we consider only \emph{recursive decision} problems, here machines accept, or reject, their input by \emph{halting} in an accepting, or non\nbdash-accepting, state, respectively.
A Turing machine can be represented via a binary string encoding its transition function (see, e.g., \cite{Hopcroft1979}).

Remember that the computation of a deterministic machine is always characterized by a single next step at any given moment, whereas the computation of a \emph{non}\/deterministic machine may have several possible next steps at some points.
In this paper, all machines are deterministic, unless otherwise specified.
The computation that a deterministic (resp., a nondeterministic) machine~$\Machine{M}$ performs (resp., may perform) over an input string~$w$ can be described by the sequence of configurations that~$\Machine{M}$ traverses (resp., may traverse) when executing on~$w$.

A \defin{configuration}, or \defin{instantaneous description (ID)}, of a (non)deterministic machine $\Machine{M}$ is a comprehensive snapshot of $\Machine{M}$'s execution at a particular moment.
An ID of $\Machine{M}$ includes the current (control) state, the current content of all tapes, and the current positions of all heads on the tapes.
From an ID and the knowledge of the machine's transition function, for a deterministic (resp., nondeterministic) machine $\Machine{M}$ it is possible to know what the next action performed by $\Machine{M}$ is (resp., what the possible next actions available to $\Machine{M}$ are).

The IDs that a ({non})deterministic machine $\Machine{M}$ may traverse when executing on input~$w$ can be arranged into %
the \defin{computation tree} of $\Machine{M}$ on $w$, or of $\Machine{M}(w)$.
The root of this tree is the starting ID of $\Machine{M}$ on $w$, and in this tree there is an edge from an ID (a tree node) $\alpha$ to an ID $\beta$ iff $\beta$ is, according to $\Machine{M}$'s transition function, a legal next ID of $\alpha$.
The computation tree of a deterministic machine is a linear path.
A \defin{partial computation} of $\Machine{M}$ on~$w$, or of $\Machine{M}(w)$, is a path within the computation tree of $\Machine{M}(w)$.
A \defin{computation} of $\Machine{M}$ on~$w$, or of $\Machine{M}(w)$, is partial computation of $\Machine{M}(w)$ whose first ID is the root of the computation tree of $\Machine{M}(w)$, and its last ID is \defin{final}, i.e., it is a leaf of the tree.
A computation is \defin{accepting} or \defin{rejecting} if its last ID has an accepting or non\nbdash-accepting state, respectively.
Since (partial) computations of $\Machine{M}(w)$ are sequences of IDs, they can be represented into strings over the binary alphabet via a suitable encoding;
such an encoding has only linear overhead (see, e.g., \cite{Hopcroft1979}).

For a (non)deterministic machine $\Machine{M}$, with a slight overload of notation we denote by $\Machine{M}(w)$ the \defin{output of the machine $\Machine{M}$ on input $w$}, where $\Machine{M}(w) = 1$ if $\Machine{M}$ accepts $w$, otherwise $\Machine{M}(w) = 0$;
we denote by $\LanguageOf{\Machine{M}}$ the \defin{language decided by the machine $\Machine{M}$};
we say that $\Machine{M}$ decides/solves a language/problem $\Language{L}$ iff $\Language{L} = \LanguageOf{\Machine{M}}$.

The \defin{computation time} (resp., \defin{computation space}) \defin{of a machine $\Machine{M}$ on input $w$} is the length of the longest computation (resp., is the maximum number of distinct cells, on all $\Machine{M}$'s \emph{work} tapes, scanned in any computation) in the computation tree of $\Machine{M}(w)$.
For a strictly positive~\cite{BalcazarDG1995} nondecreasing~\cite{LewisSH1965,HartmanisS1965,StearnsHL1965} function $t\colon \NaturalsDomain \to \NaturalsDomain$ (resp., $s\colon \NaturalsDomain\to\NaturalsDomain$), the machine $\Machine{M}$ has \defin{running time} $t(n)$ (resp., \defin{running space} $s(n)$) iff, on all but finitely\nbdash-many inputs $w$, the computation time (resp., space) of $\Machine{M}$ on $w$ does not exceed $t(\StringLength{w})$ (resp., $s(\StringLength{w})$)~\cite{Kozen2006};
$t(n)$ (resp., $s(n)$) is also named \defin{time function} (resp., \defin{space function}).
A time function $t(n)$ (resp., space function $s(n)$) is \defin{time constructible} (resp., \defin{space constructible}) iff there is a Turing machine $\Machine{M}$ that, on every input of \emph{length} $n$, halts in exactly $t(n)$ steps (resp., marks off $s(n)$ work tape cells and halts, never using more than $s(n)$ space in the process)~\cite{HartmanisS1965,StearnsHL1965,BalcazarDG1995,Kozen2006}.
An equivalent characterization of time and space constructibility is as follows.
Let $t(n)$ be a time function such that $t(n) \geq (1 {+} \epsilon) \cdot n$ for some $\epsilon > 0$ and for all but finitely many $n$, and let $s(n)$ be a space function.
The function $t(n)$ (resp., $s(n)$) is \emph{time-constructible} (resp., \emph{space-constructible}) iff there is a Turing machine that on evry input of length $n$ outputs the binary representation of $t(n)$ (resp., $s(n)$) in $O(t(n))$ time (resp., $O(s(n))$ space) \cite{Kannan1982,Kobayashi1985,BalcazarDG1995}.
A machine \defin{$\Machine{M}$ decides a language $\Language{L}$ in time $t(n)$ (resp., space $s(n)$)} iff $\Language{L} = \LanguageOf{\Machine{M}}$ and $\Machine{M}$ has running time $t(n)$ (resp., space $s(n)$).

By $\DTime{t(n)}$ and $\DSpace{s(n)}$ (resp., $\NTime{t(n)}$ and $\NSpace{s(n)}$) we denote the class of languages decided by deterministic (resp., nondeterministic) machines in time $t(n)$ and space $s(n)$, respectively.
The functions $t(n)$ and $s(n)$ are generally assumed to be time and space constructible, respectively, as the complexity classes defined via these kinds of functions enjoy some desirable properties (see, e.g., \cite{Hopcroft1979,BalcazarDG1995}).

We now introduce the complexity classes that we will deal with.
By the definition of the iterated exponentials and the relationships between their asymptotic classes, for $i \geq 0$ and $j \geq -1$, we define the complexity classes:
\begin{align*}
  \iExpTime{i} &= \bigcup_{k \geq 1} \DTime{\iExp{i}{n^k}}   &   \iExpSpace{j} &= \bigcup_{k \geq 1} \DSpace{\iExp{j}{n^k}} \\
  \iNExpTime{i} &= \bigcup_{k \geq 1} \NTime{\iExp{i}{n^k}}   &   \iNExpSpace{j} &= \bigcup_{k \geq 1} \NSpace{\iExp{j}{n^k}}.
\end{align*}
These definitions follow and generalize those in~\cite{Simon1975,Hartmanis1985}.%
\footnote{The class $\iNExpTime{i}$ is denoted $S_i$ by~\citet{Simon1975}, who defines these classes via ``higher order'' objects, i.e., starting from strings, the author considers sets of strings, sets of sets of strings, and so on.}
Observe that, for $i \geq 0$ and $j \geq -1$, functions $\smash{\iExp{i}{n^k}}$ and $\smash{\iExp{j}{n^k}}$ are time and space constructible, respectively (see the property above and \zcref{sec_maths_complexity}).
The reader can easily check that, by the equivalence with the standard definitions, $\LogSpace = \iExpSpace{-1}$, $\PTime = \iExpTime{0}$, $\NPTime = \iNExpTime{0}$, $\PSpace = \iExpSpace{0}$, $\ExpTime = \iExpTime{1}$, $\NExpTime = \iNExpTime{1}$, and $\ExpSpace = \iExpSpace{1}$.%
\footnote{Other exponential-time classes considered in the literature are the deterministic and the nondeterministic exponential-time with linear exponent classes. These classes are defined as $\mathrm{E} = \textsc{ETime} = \bigcup_{c \geq 1}\DTime{2^{c \cdot n}}$ and $\mathrm{NE} = \textsc{NETime} = \bigcup_{c \geq 1}\NTime{2^{c \cdot n}}$, respectively (see, e.g.,~\cite{Hartmanis1985,Hemachandra1989,Mocas1996,Dawar1998}). We will not however consider them in this paper.}
In what follows, we use the standard notation for these complexity classes, unless we provide general results.
Below, when we say that a language, a problem, a relation, or a predicate, can be decided within a given time or space bound, without mentioning whether deterministic or not, we mean that it can be decided within that resource bound by a \emph{deterministic} machine.
When we will refer to \emph{non}\/deterministic classes, we will explicitly mention it.

By \citeauthor{Savitch1970}'s theorem~\cite{Savitch1970}, for all $i \geq 0$, $\iExpSpace{i} = \iNExpSpace{i}$.
As we will not consider nondeterministic logspace in this paper, we can here avoid to explicitly deal with nondeterministic space classes.

By the linear speed-up theorem~\cite{HartmanisS1965}, for $i \geq 0$, it holds that $\DTime{\iExp{i}{n^k}} = \DTime{O(\iExp{i}{n^k})}$ and $\NTime{\iExp{i}{n^k}} = \NTime{O(\iExp{i}{n^k})}$, and, by the space compression theorem~\cite{StearnsHL1965}, for $j \geq -1$, it holds that $\DSpace{\iExp{j}{n^k}} = \DSpace{O(\iExp{j}{n^k})}$ (see also \cite{Hopcroft1979,Papadimitriou1994,BalcazarDG1995,Kozen2006}).
Therefore, if $\Language{L} \in (\mathrm{N})\iExpTime{i}$, there is a (non)deterministic machine deciding $\Language{L}$ in time $\iExp{i}{n^k}$.
By $\iExp{i}{O(\iExp{j}{n^k})} \subseteq O(\iExp{i+j}{n^{k+1}})$ (see \zcref{theo_asymptotic_class_properties_iterated_exponentials_containment}), if $\Language{L} \in \DTime{\iExp{i}{O(\iExp{j}{n^k})}}$, then $\Language{L} \in \DTime{O(\iExp{i+j}{n^{k+1}})} = \linebreak[0] \DTime{\iExp{i+j}{n^{k+1}}}$, and hence we have $\Language{L} \in \iExpTime{(i+j)}$, as expected.
Similar remarks hold for $\iNExpTime{i}$ and $\iExpSpace{i}$ as well.

Languages in \NPTime (resp., \NExpTime) enjoy this property:
each \yesinst $w$ of $\Language{L}$ has a \emph{certificate}~$u$ of polynomial (resp., exponential) size in~$\StringLength{w}$, which witnesses $w \in \Language{L}$ and that can be checked in \emph{polynomial} time \Wrt $\StringLength{w} + \StringLength{u}$.
Languages in \CoNPTime (resp., \CoNExpTime) have instead polynomial (resp., exponential) certificates for \noinsts.
By this, a language $\Language{L}$ is in \NPTime, \CoNPTime, \NExpTime, or \CoNExpTime (see, e.g.,~\cite{Karp1972,Hartmanis1985,Hemachandra1989,Goldreich2008}) iff there is a polynomial $p(n)$ and a \emph{deterministic polynomial-time} binary predicate $R$, whose time complexity is measured \Wrt the combined size of its two arguments, such that, for~every~string~$w$,
\begin{flalign*}
    \NPTime &\colon w \in \Language{L} \Leftrightarrow (\exists u \in \alphabet^{\leq p(\StringLength{w})}) \PrefixMatrixSeparator R(w,u) = 1; & \CoNPTime &\colon w \in \Language{L} \Leftrightarrow (\forall u \in \alphabet^{\leq p(\StringLength{w})}) \PrefixMatrixSeparator R(w,u) = 1; \\
    \NExpTime &\colon w \in \Language{L} \Leftrightarrow (\exists u \in \alphabet^{\leq 2^{p(\StringLength{w})}}) \PrefixMatrixSeparator R(w,u) = 1; & \CoNExpTime &\colon w \in \Language{L} \Leftrightarrow (\forall u \in \alphabet^{\leq 2^{p(\StringLength{w})}}) \PrefixMatrixSeparator R(w,u) = 1.
\end{flalign*}

Standard techniques (see, e.g., \cite{Wrathall1976,ChandraKS81,BalcazarDG1995,Goldreich2008,Arora2009}) generalize to yield certificate-based characterizations for higher\nbdash-order exponential time classes. %
A language $\Language{L}$ belongs to $\iNExpTime{i}$ (resp., \CoINExpTime{i}), with $i \geq 0$, iff there exist a polynomial $p(n)$ and a \emph{deterministic polynomial-time} binary predicate $R$, whose time complexity is measured \Wrt the combined size of its two arguments, such that, for every string $w$,
\begin{flalign*}
\iNExpTime{i} &\colon w \in \Language{L} \Leftrightarrow (\exists u \in \alphabet^{\leq \iExp{i}{p(\StringLength{w})}}) \mskip 1mu R(w,u) = 1; & \CoINExpTime{i} &\colon w \in \Language{L} \Leftrightarrow (\forall u \in \alphabet^{\leq \iExp{i}{p(\StringLength{w})}}) \mskip 1mu R(w,u) = 1.
\end{flalign*}

The \emph{certificate-based} characterization of $\iNExpTime{i}$, for $i \geq 0$, is rather interesting, as it highlights that languages $\Language{L} \in \iNExpTime{i}$ can be decided by machines working in two phases:
first an \iExponential{i}{}ly-long certificate is \emph{nondeterministically} guessed, and then the validity of the guessed certificate is \emph{deterministically} checked in \emph{polynomial time} \Wrt the \emph{combined} size of the input string and the guessed certificate.

The complexity classes $\iNExpTime{i}$ and $\CoINExpTime{i}$, for $i \geq 0$, are closed under conjunction and disjunction;
deterministic complexity classes are also closed under complement.

Standard results in complexity theory (see, e.g., \cite{BalcazarDG1995}), show that, for all $i \geq 0$, the inclusion relationships (which are all currently
believed to be strict) between these complexity classes are:%
\footnote{For these inclusion relationships, the notation ``\!$A\subseteq B, C \subseteq D$'' is a shorthand for: $A \subseteq B$; $A\subseteq C$; $B\subseteq D$; and $C\subseteq D$. Moreover, no inclusion relationships are currently known between $B$ and $C$, and it is currently believed that $B \not\subseteq C$ and $C \not\subseteq B$.}%
\begin{equation*}
  \cdots \subseteq \iExpSpace{(i-1)} \subseteq \iExpTime{i} \subseteq \iNExpTime{i}, \CoINExpTime{i} \subseteq \iExpSpace{i} \subseteq \cdots.
\end{equation*}

\subsection{Oracle Machines and Oracle Complexity Classes}
\label{sec_prelim_Oracle_Classes}

An \defin{oracle Turing machine $\Oracle{\Machine{M}}{?}$}, is a (non)deterministic Turing machine $\Machine{M}$ that, during its computation, may ask to an oracle to answer membership queries at \emph{unit} cost. %
For this reason, an oracle machine $\Oracle{\Machine{M}}{?}$ is equipped with an additional \emph{write-only} \emph{unidirectional} work tape, called the \defin{query tape}, through which the machine passes its queries to the oracle.
The definition of $\Oracle{\Machine{M}}{?}$ is \emph{independent} from its oracle, and the symbol ``$?$'' indicates that oracles for different languages can be ``attached'' to $\Machine{M}$~\cite{Papadimitriou1994}.
For a language $\Language{A}$, by $\Oracle{\Machine{M}}{\Language{A}}$ we mean that the oracle attached to $\Oracle{\Machine{M}}{?}$ decides~$\Language{A}$.
For an oracle machine $\Oracle{\Machine{M}}{?}$ and a language $\Language{A}$, $\Oracle{\Machine{M}}{\Language{A}}(w)$ denotes the output of the oracle machine on input $w$ when the oracle decides $\Language{A}$;
$\LanguageOf{\Oracle{\Machine{M}}{\Language{A}}}$ denotes the language decided by the oracle machine $\Oracle{\Machine{M}}{?}$ with an oracle for~$\Language{A}$.
Like for non\nbdash-oracle machines, oracle machine computations can be described via sequences of IDs.
Oracle machine IDs are very similar to the IDs of non-oracle machines:
the query tape is simply an additional tape, whose content and head's position are part of the ID.
Notice that this makes the oracle machine IDs here considered different from those sometimes considered in the literature, where the query tape is \emph{not} included in the IDs (see, e.g., \cite{LadnerL76,Hemaspaandra1994}).
For presentation purposes, we could sometimes refer to the oracle as to an additional machine which the oracle machine (i.e., the caller) can ask questions to.
Hence, we might sometimes refer to the computation carried out by the oracle;
nonetheless, the time cost paid by the caller for the computation carried out by the oracle is always \emph{one} step.
If it is clear from the context that $\Oracle{\Machine{M}}{?}$ is actually an oracle Turing machine, to streamline the notation, we may refer to it just by~$\Machine{M}$.

For a strictly positive nondecreasing function $f\colon \NaturalsDomain\to\NaturalsDomain$ (resp., for a constant integer $k$), $\BoundedOracle{\Machine{M}}{?}{f(n)}$ (resp., $\BoundedOracle{\Machine{M}}{?}{k}$) denotes an oracle machine allowed to issue at most $f(n)$ (resp., at most $k$) queries to its oracle, where $n$ is the input string size.
The notation $\ParOracle{\Machine{M}}{?}$ state that the oracle machine is allowed to issue a \emph{single} round of \defin{parallel queries}, i.e.,
the queries that $\ParOracle{\Machine{M}}{?}$ submits to its oracle are collected and asked all at once.
This way of querying the oracle is also called \defin{nonadaptive}, in contrast with the standard sequential way of asking queries, which instead can be \defin{adaptive}~\cite{Book1988}.
For parallel queries, we might relax the constraint on the \emph{single} round of parallel queries;
for a function $f\colon \NaturalsDomain\to\NaturalsDomain$ (resp., a constant integer $k$), $\ParBoundedOracle{\Machine{M}}{?}{f(n)}$ (resp., $\ParBoundedOracle{\Machine{M}}{?}{k}$) denotes an oracle machine allowed to issue at most $f(n)$ (resp., at most $k$) rounds of parallel queries to its oracle, where $n$ is the input string size.
We can also combine the constraints:
for two strictly positive nondecreasing functions $f(n)$ and $g(n)$, $\DoubleBoundedParOracle{\Machine{M}}{?}{f(n)}{g(n)}$ denotes an oracle machine allowed to perform $g(n)$ rounds of parallel queries, and in each round at most $f(n)$ queries can be asked;
the notation $\BoundedParOracle{\Machine{M}}{?}{f(n)}$ has the natural meaning, and the functions $f(n)$ and $g(n)$ can be replaced in the notation by constant values, with the natural meanings.

If $\ComplexityClass{C}$ is a (non)deterministic \emph{time}\nbdash-bounded complexity class and $\Language{A}$ is a language, $\ComplexityClass{C}^\Language{A}$ is the class of languages decided by oracle machines in $\ComplexityClass{C}$ querying an oracle for $\Language{A}$; %
for a complexity class $\ComplexityClass{D}$, we denote by $\ComplexityClass{C}^\ComplexityClass{D}$ the class of languages decided by oracle machines in $\ComplexityClass{C}$ querying an oracle for a language in~$\ComplexityClass{D}$.
In the following, when we say that an oracle machine $\Oracle{\Machine{M}}{?}$ queries an oracle in~$\ComplexityClass{D}$, 
we mean that $\Machine{M}$ queries an oracle for a $\ComplexityClass{D}$\CompleteSuffix language.
The notation introduced to denote constraints over the queries is naturally extended to oracle complexity classes, e.g., $\DoubleBoundedParOracle{\ComplexityClass{C}}{\ComplexityClass{D}}{f(n)}{g(n)}$.
For a family $F$ of functions, $\BoundedOracle{\ComplexityClass{C}}{\ComplexityClass{D}}{F}$ (resp., $\BoundedParOracle{\ComplexityClass{C}}{\ComplexityClass{D}}{F}$) is the class of languages decided by oracle machines in $\ComplexityClass{C}$ querying an oracle in $\ComplexityClass{D}$ at most $f(n)$\nbdash-many times (resp., with at most $f(n)$\nbdash-many parallel queries), with $f(n) \in F$;
more formally, $\BoundedOracle{\ComplexityClass{C}}{\ComplexityClass{D}}{F} = \bigcup_{f(n) \in F} \BoundedOracle{\ComplexityClass{C}}{\ComplexityClass{D}}{f(n)}$ and $\BoundedParOracle{\ComplexityClass{C}}{\ComplexityClass{D}}{F} = \bigcup_{f(n) \in F} \BoundedParOracle{\ComplexityClass{C}}{\ComplexityClass{D}}{f(n)}$.

Defining \emph{space}\nbdash-bounded oracle machines is not so straightforward~\cite{Hartmanis1988}, as the space used on the query tape might, or might not, contribute toward the computation space of the oracle machine.
Different definitions have been proposed in the literature, giving rise to different space\nbdash-bounded oracle complexity classes (see, e.g., \cite{Michel1992}).
In this paper, for space\nbdash-bounded oracle classes we adopt the \emph{deterministic query model}.
This query model does \emph{not} impose space constraints over the query tape, but it requires oracle machines to \emph{always} act deterministically from the moment they start to write a query on the tape, until the query is issued to the oracle.
Nonetheless, since our results regarding space\nbdash-bounded oracle machines will only be about deterministic ones, for what concerns us here, the deterministic and the unrestricted query model, in which no restriction at all is imposed, are equivalent.
Notice that the unrestricted query model has been adopted multiple times in the literature also for space\nbdash-bounded oracle classes above \LogSpace (see, e.g., \cite{Hemaspaandra1994,Gottlob1995,Dawar1998}).

\subsection{The Polynomial and the (Weak) Exponential Hierarchies}
\label{sec_prelim_PH_ExpH}

The Polynomial and the (Weak) Exponential Hierarchies are two well known complexity hierarchies.
Below, we simultaneously introduce them, so to highlight their similarities.
We start by defining the hierarchy ``main levels'', and afterwards their ``intermediate levels''.
We will generalize these definitions in \zcref{sec_def_iterated_exp_hierarchy}.

For $c \geq 1$, the (main) $c$\nbdash-th levels of the Polynomial and the (Weak) Exponential Hierarchies are
\begin{alignat*}{3}
  \SigmaP{c} &{}= \Oracle{\NPTime}{\SigmaP{c-1}} & \qquad \text{and} \qquad & \SigmaWExp{c} &{}= \Oracle{\NExpTime}{\SigmaP{c-1}},
\end{alignat*}
respectively, where we define $\SigmaP{0} = \PTime$ and $\SigmaWExp{0} = \ExpTime$ as the ground levels of the respective hierarchies.
For $c \geq 0$, define $\PiP{c} = \ComplementPrefixKerned\SigmaP{c}$ and $\PiWExp{c} = \ComplementPrefixKerned\SigmaWExp{c}$.
Notice that $\SigmaP{c},\PiP{c} \subseteq \SigmaP{c+1},\PiP{c+1}$ and $\SigmaWExp{c},\PiWExp{c} \subseteq \SigmaWExp{c+1},\PiWExp{c+1}$.
The Polynomial and the (Weak) Exponential Hierarchies are $\PolHier = \bigcup_{c \geq 0} \SigmaP{c}$ and $\WExpHier = \bigcup_{c \geq 0} \SigmaWExp{c}$, respectively.%
\footnote{Classically, $\WExpHier$ denoted the Exponential-Time with linear exponent Hierarchy, i.e., the Exponential-Time Hierarchy built over $\mathrm{E}$, rather than over $\ExpTime$.
Since we do not consider $\mathrm{E}$ here, for notational convenience we use the notation $\SigmaWExp{c}$, $\PiWExp{c}$, $\DeltaWExp{c}$, and $\WExpHier$, in lieu of the longer $\Sigma^{\textsc{Exp}}_c$, $\Pi^{\textsc{Exp}}_c$, $\Delta^{\textsc{Exp}}_c$, and $\textnormal{\textsc{ExpH}}$, respectively (see, e.g., \cite{Hartmanis1985,Hemachandra1989,Mocas1996}).}

The (\emph{Weak}) Exponential Hierarchy was named in this way to distinguish it from the \emph{Strong} Exponential Hierarchy, which was differently defined and separately introduced also at the time~\cite{Hemachandra1989}.
We will provide additional details on the \SEHText when needed in \zcref{sec_top_seh}.

Observe that the levels of $\WExpHier$ are defined as $\Oracle{\NExpTime}{\SigmaP{c-1}}$ and \emph{not} as $\Oracle{\NExpTime}{\SigmaWExp{c-1}}$.
For example, the third level of $\WExpHier$ is $\Oracle{\NExpTime}{\Oracle{\NPTime}{\NPTime}}$ and \emph{not} $\Oracle{\NExpTime}{\Oracle{\NExpTime}{\NExpTime}}$.
The reason behind this specific definition will be explained below.

Like $\NPTime$ and $\CoNPTime$, also $\SigmaP{c}$ and $\PiP{c}$~\cite{Stockmeyer1976,Wrathall1976} and $\SigmaWExp{c}$ and $\PiWExp{c}$~\cite{Hartmanis1985,Hemachandra1989} have cer\-tif\-i\-cate-based characterizations.
A language $\Language{L}$ is in $\SigmaP{c}$, $\SigmaWExp{c}$ (resp., $\PiP{c}$, $\PiWExp{c}$) iff there exists a polynomial $p(n)$ and a \emph{deterministic polynomial-time $(c{+}1)$\nbdash-ary predicate $R$, whose time complexity is measured \Wrt the combined size of all its arguments, such that, for every string $w$},
\begin{alignat*}{2}
    \SigmaP{c}&\colon &\quad&w \in \Language{L} \Leftrightarrow (\exists u_1 \in \alphabet^{\leq p(\StringLength{w})}) (\forall u_2 \in \alphabet^{\leq p(\StringLength{w})}) \cdots (Q_c \: u_c \in \alphabet^{\leq p(\StringLength{w})}) \: R(w,u_1,\dots,u_c) = 1 \\
    \SigmaWExp{c}&\colon &\quad&w \in \Language{L} \Leftrightarrow (\exists u_1 \in \alphabet^{\leq 2^{p(\StringLength{w})}}) (\forall u_2 \in \alphabet^{\leq 2^{p(\StringLength{w})}}) \cdots (Q_c \: u_c\in \alphabet^{\leq 2^{p(\StringLength{w})}}) \: R(w,u_1,\dots,u_c) = 1 \\
    \PiP{c}&\colon &\quad&w \in \Language{L} \Leftrightarrow (\forall u_1 \in \alphabet^{\leq p(\StringLength{w})}) (\exists u_2 \in \alphabet^{\leq p(\StringLength{w})}) \cdots (Q_c \: u_c \in \alphabet^{\leq p(\StringLength{w})}) \: R(w,u_1,\dots,u_c) = 1 \\
    \PiWExp{c}&\colon &\quad&w \in \Language{L} \Leftrightarrow (\forall u_1 \in \alphabet^{\leq 2^{p(\StringLength{w})}}) (\exists u_2 \in \alphabet^{\leq 2^{p(\StringLength{w})}}) \cdots (Q_c \: u_c\in \alphabet^{\leq 2^{p(\StringLength{w})}}) \: R(w,u_1,\dots,u_c) = 1,
\end{alignat*}
where $Q_c$ is $\forall$ or $\exists$ depending on $c$ being even or odd  (resp., odd or even), respectively.

By this characterization, $\Oracle{\NPTime}{\SigmaP{c}} = \BoundedOracle{\NPTime}{\SigmaP{c}}{1}$ (resp., $\Oracle{\NExpTime}{\SigmaP{c}} = \BoundedOracle{\NExpTime}{\SigmaP{c}}{1}$), for all $c \geq 1$, i.e., allowing an $\NPTime$ (resp., a $\NExpTime$) oracle machine to query more than once its $\SigmaP{c}$ oracle makes no difference~\cite{Wagner1990,BalcazarDG1995}.

Besides the certificate-based characterizations, the Polynomial and the (Weak) Exponential Hierarchies can be defined via alternating machines.
Indeed, $\SigmaP{c}$ and $\PiP{c}$ (resp., $\SigmaWExp{c}$ and $\PiWExp{c}$) can equivalently be defined as the classes of languages that can be decided by $\Sigma_c$\nbdash-alternating and $\Pi_c$\nbdash-alternating machines, respectively, in polynomial~\cite{ChandraKS81} (resp., exponential~\cite{Mocas1996}) time. 

The levels of $\WExpHier$ were defined as $\NExpTime \subseteq \Oracle{\NExpTime}{\NPTime} \subseteq \Oracle{\NExpTime}{\Oracle{\NPTime}{\NPTime}} \subseteq {} \cdots$, rather than $\NExpTime \subseteq \Oracle{\NExpTime}{\NExpTime} \subseteq \Oracle{\NExpTime}{\Oracle{\NExpTime}{\NExpTime}} \subseteq {} \cdots$,
because $\Oracle{\NExpTime}{\NExpTime}$ contains $\iExpTime{2}$ and $\ExpSpace$.
Instead, the intent was to obtain a $\NExpTime$-based hierarchy analogue to $\PolHier$, with alike certificate-based and alternating machines characterizations~\cite{Hemachandra1989,Mocas1996}, and where $\PolHier$ is contained in $\PSpace$ and not vice-versa.

The Polynomial and the (Weak) Exponential Hierarchies were defined to include some intermediate classes.
The intermediate levels \emph{above} the $c$\nbdash-th main level and \emph{beneath} the $(c{+}1)$\nbdash-th main level were defined as
\begin{alignat*}{3}
  \DeltaP{c+1} &{}= \Oracle{\PTime}{\SigmaP{c}} & \qquad \text{and} \qquad & \DeltaWExp{c+1} &{}= \Oracle{\ExpTime}{\SigmaP{c}},
\end{alignat*}
for \PolHier and \WExpHier, respectively.
Notice that $\Oracle{\PTime}{\SigmaP{c}} = \BoundedOracle{\PTime}{\SigmaP{c}}{\PolFunctions}$ and $\Oracle{\ExpTime}{\SigmaP{c}} = \BoundedOracle{\ExpTime}{\SigmaP{c}}{2^{\PolFunctions}}$.

Additional intermediate levels
of $\PolHier$ and $\WExpHier$ were studied by \citeauthor{Wagner1990} and \citeauthor{Mocas1996}, respectively, revealing a richer structure between the main levels.
They defined the classes $\ThetaP{c+1} = \LogOracle{\PTime}{\SigmaP{c}}$~\cite{Wagner1990} and $\DeltaWExpBound{c+1}{\PolFunctions} = \PolOracle{\ExpTime}{\SigmaP{c}}$~\cite{Mocas1996}, between the $c$\nbdash-th and the $(c{+}1)$\nbdash-th main levels of \PolHier and \WExpHier, respectively.
For $c \geq 1$, the inclusion relationships, all currently believed to be strict, between the mentioned hierarchy classes are:
\begin{alignat}{4}
\PiP{c},\SigmaP{c} &= \Oracle{\NPTime}{\SigmaP{c-1}} & \subseteq \LogOracle{\PTime}{\SigmaP{c}} & {} \subseteq \PolOracle{\PTime}{\SigmaP{c}} & & {} \subseteq \Oracle{\NPTime}{\SigmaP{c}} & & {} = \SigmaP{c+1},\PiP{c+1} \quad \text{and} \nonumber \\
\PiWExp{c},\SigmaWExp{c} &= \Oracle{\NExpTime}{\SigmaP{c-1}} & & {} \subseteq \PolOracle{\ExpTime}{\SigmaP{c}} \subseteq \ExpOracle{\ExpTime}{\SigmaP{c}} & & {} \subseteq \Oracle{\NExpTime}{\SigmaP{c}} & & {} = \SigmaWExp{c+1},\PiWExp{c+1}. \label{eq_hint_missing_intermediate_level}
\end{alignat}

The $\WExpHier$ intermediate levels mentioned in \zcref{eq_hint_missing_intermediate_level} and analyzed by \citet{Mocas1996} are not, in some way, ``exhaustive'', as
we can also define $\LogOracle{\ExpTime}{\SigmaP{c}}$, and this might have already caught the reader's eye from the way in which we have purposefully arranged the classes in \zcref{eq_hint_missing_intermediate_level}.%
\footnote{\citet{Mocas1996} defined $\DeltaWExpBound{c+1}{F} = \BoundedOracle{\ExpTime}{\SigmaP{c}}{F}$, for a family $F$ of functions, but only $\DeltaWExpBound{c+1}{\PolFunctions}$ was then investigated.
\citet{EiterGL1997} reckoned $\DeltaWExpBound{c+1}{\PolFunctions} = \PolOracle{\ExpTime}{\SigmaP{c}}$ as the analogue of $\LogOracle{\PTime}{\SigmaP{c}}$, whereas we stress that $\LogOracle{\ExpTime}{\SigmaP{c}}$ can be defined as well.}
Observe that the number of the intermediate levels could undoubtedly be inflated by progressively limiting the amount of queries available to the oracle machines.
For example, we could define the classes $\BoundedOracle{\PTime}{\NPTime}{O(\log n)} \supseteq \BoundedOracle{\PTime}{\NPTime}{O(\log \log n)} \supseteq \BoundedOracle{\PTime}{\NPTime}{O(\log \log \log n)} \supseteq \cdots$, and so on,%
\footnote{For the iterated logarithm, we might envisage a logarithmic function like $\max\set{1,\lceil \log n \rceil}$, to guarantee the possibility of repeatedly evaluating the function even on relatively small numbers $n$.}
and these classes would still be supersets of the classes imposing a constant number of possible queries---conversely, the amount of allowed queries cannot be increased at will, because, e.g., a $\PTime$ oracle machine cannot issue more than polynomially-many queries.
The bounded-query oracle classes considered in this paper, besides those with a constant number of queries, are those admitting at least logarithmically-many queries.

The \BHText over $\SigmaP{c}$ (resp., $\SigmaWExp{c}$), for $c \geq 1$, which equals the class of languages in $\BoundedOracle{\PTime}{\SigmaP{c}}{O(1)}$ (resp., $\BoundedOracle{\ExpTime}{\SigmaP{c}}{O(1)}$)~\cite{Beigel1991}, can be seen as an additional intermediate level sitting between $\SigmaP{c}$ (resp., $\SigmaWExp{c}$) and $\LogOracle{\PTime}{\SigmaP{c}}$ (resp., $\LogOracle{\ExpTime}{\SigmaP{c}}$)---see \zcref{sec_extended_Boolean_Hierarchies} for more details on the Boolean Hierarchies;
notice, however, that some parts of this appendix rely on material introduced in \zcref{sec_Hausdorff_reductions_classes}.
The \BHText over \NPTime was extensively studied in various papers (see the works cited in \zcref{sec_extended_Boolean_Hierarchies}), and the \BHText over $\SigmaP{2}$ was investigated in~\cite{ChangK1996}, whereas the \BHText over \NExpTime was considered in~\cite{Dawar1998}.

It is well known that the Polynomial and the (Weak) Exponential Hierarchies are contained within \PSpace and \ExpSpace, respectively~\cite{Stockmeyer1976,Wrathall1976,Mocas1996},
which, in turn, equal alternating polynomial and exponential time (with \emph{no} bounds on the number of alterations), respectively~\cite{ChandraKS81}.

\subsection{Iterated Exponential Hierarchies}
\label{sec_def_iterated_exp_hierarchy}

By generalizing \PolHier and \WExpHier, we now introduce general complexity hierarchies.
These will be the double-, the triple-\WEHText, and so on.
We will generically refer to them as the iterated exponential hierarchies.
Besides defining the main levels of these generalized hierarchies, we also define their intermediate levels.

\begin{definition}
For $c \geq 1$, the (main) $c$\nbdash-th level of the \iWEHText{i} is defined as
\[
  \SigmaIWExp{i}{c} = \Oracle{\iNExpTime{i}}{\SigmaP{c-1}},
\]
where we define $\SigmaIWExp{i}{0} = \iExpTime{i}$ as the hierarchy's ground level, and we define $\PiIWExp{i}{c} = \ComplementPrefixKerned\SigmaIWExp{i}{c}$, for $c \geq 0$.
The \iWEHText{i} $\iWExpHier{i}$ is defined as $\iWExpHier{i} = \bigcup_{c \geq 0} \SigmaIWExp{i}{c}$.
\end{definition}

As for the Polynomial and the Exponential Hierarchies, for all $i \geq 0$ and all $c \geq 0$, $\SigmaIWExp{i}{c}, \PiIWExp{i}{c} \subseteq \SigmaIWExp{i}{c+1},\PiIWExp{i}{c+1}$.

Generalizations of known techniques (see, e.g., \cite{Wrathall1976,ChandraKS81,BalcazarDG1995,Goldreich2008,Arora2009}) show that $\SigmaIWExp{i}{c}$ and $\PiIWExp{i}{c}$ can equally be defined via certificates and alternating machines.
A language $\Language{L}$ is in $\SigmaIWExp{i}{c}$ (resp., $\PiIWExp{i}{c}$) iff there exist a polynomial $p(n)$ and a \emph{deterministic polynomial-time} $(c{+}1)$\nbdash-ary predicate $R$, whose time complexity is measured \Wrt the combined size of all its arguments, such that, for every string $w$,
\begin{equation*}
\begin{alignedat}{2}
    \SigmaIWExp{i}{c}&\colon &\quad&w \in \Language{L} \Leftrightarrow (\exists u_1 \in \alphabet^{\leq \iExp{i}{p(\StringLength{w})}}) (\forall u_2 \in \alphabet^{\leq \iExp{i}{p(\StringLength{w})}}) \cdots (Q_c \: u_c\in \alphabet^{\leq \iExp{i}{p(\StringLength{w})}}) \: R(w,u_1,\dots,u_c) = 1 \\
    \PiIWExp{i}{c}&\colon &\quad&w \in \Language{L} \Leftrightarrow (\forall u_1 \in \alphabet^{\leq \iExp{i}{p(\StringLength{w})}}) (\exists u_2 \in \alphabet^{\leq \iExp{i}{p(\StringLength{w})}}) \cdots (Q_c \: u_c\in \alphabet^{\leq \iExp{i}{p(\StringLength{w})}}) \: R(w,u_1,\dots,u_c) = 1,
\end{alignedat}
\end{equation*}
where $Q_c$ is $\forall$ or $\exists$ if $c$ is even or odd (resp., odd or even), respectively.
The classes $\SigmaIWExp{i}{c}$ and $\PiIWExp{i}{c}$ are also the classes of languages decided by $\Sigma_c$\nbdash-alternating and $\Pi_c$\nbdash-alternating machines, respectively, in $O(\iExpPolFunctions{i})$ time.

Since, for a space-constructible function $f(n)$, $\ATime{f(n)} = \DSpace{f(n)^2}$~\cite{ChandraKS81},
we have $\ATime{\iExpPolFunctions{i}} = \iExpSpace{i}$, because all functions in $\iExpPolFunctions{i}$ are space-constructible.
Therefore, by $\SigmaIWExp{i}{c} = \BoundedExATime{c}{\iExpPolFunctions{i}} \subseteq  \ATime{\iExpPolFunctions{i}}$, $\iWExpHier{i}$ is contained in $\iExpSpace{i}$, similarly to what happens for $\PolHier$ and $\WExpHier$, which are contained in $\PSpace$ and $\ExpSpace$, respectively.

The following definition generalizes the notion of intermediate levels of $\PolHier$ and $\WExpHier$ to $\iWExpHier{i}$. 

\begin{samepage}
\begin{definition}
For $c \geq 1$, the intermediate levels between the $c$\nbdash-th and the $(c{+}1)$\nbdash-th main levels of $\iWExpHier{i}$ are:
\begin{itemize}[nosep,label=--]
  \item $\DeltaIWExpBound{i}{c+1}{O(1)} = \BoundedOracle{\iExpTime{i}}{\SigmaP{c}}{O(1)}$, equalling the Boolean Hierarchy over $\SigmaIWExp{i}{c}$, denoted by $\BHGeneric{\SigmaIWExp{i}{c}}$;%
      \footnote{The equivalence follows from a result in~\cite{Beigel1991}, and it can also be obtained from our results in \zcref{sec_iEXP_jNEXP}. For more on the Boolean Hierarchies see \zcref{sec_extended_Boolean_Hierarchies};
      remember that parts of this appendix rely on concepts introduced in \zcref{sec_Hausdorff_reductions_classes}.}
      and
  \item $\DeltaIWExpBound{i}{c+1}{\iExpPolFunctions{j}} = \BoundedOracle{\iExpTime{i}}{\SigmaP{c}}{\iExpPolFunctions{j}}$, for all $j$ such that $-1 \leq j \leq i$.
\end{itemize}
\end{definition}
\end{samepage}

By definition, $\iWExpHier{0} = \PolHier$ and $\iWExpHier{1} = \WExpHier$, with all their main and intermediate levels.
The inclusion relationships, which we can believe to be strict, between the mentioned classes, for all $i \geq 0$ and $c \geq 1$, are:
\begin{equation*}
  \cdots \subseteq \SigmaIWExp{i}{c-1},\PiIWExp{i}{c-1} \subseteq \DeltaIWExpBound{i}{c}{\LogFunctions} \subseteq \DeltaIWExpBound{i}{c}{\PolFunctions} \subseteq \DeltaIWExpBound{i}{c}{2^{\PolFunctions}} \subseteq \DeltaIWExpBound{i}{c}{\iExpPolFunctions{2}} \subseteq \cdots \subseteq \DeltaIWExpBound{i}{c}{\iExpPolFunctions{i}} \subseteq \SigmaIWExp{i}{c},\PiIWExp{i}{c} \subseteq \cdots.
\end{equation*}

We define the \defin{$c$\nbdash-th step} of a complexity hierarchy as the union of the $c$\nbdash-th main level and all the intermediate levels \emph{above} the $c$\nbdash-th main level and beneath the $(c{+}1)$\nbdash-th main level.
For example, $\SigmaP{2} \cup \LogOracle{\PTime}{\SigmaP{2}} \cup \Oracle{\PTime}{\SigmaP{2}}$ is the second step of $\PolHier$, whereas $\NExpTime \cup \LogOracle{\ExpTime}{\NPTime} \cup \PolOracle{\ExpTime}{\NPTime} \cup \Oracle{\ExpTime}{\NPTime}$ is the first step of $\WExpHier$.
In these examples notice that, although the (classically\nbdash-defined) level indices would suggest grouping together $\ThetaP{2}$, $\DeltaP{2}$, and $\SigmaP{2}$, we instead define the second step of the \PHText to consist of the classes $\SigmaP{2}$, $\ThetaP{3}$, and $\DeltaP{3}$.
We adopt this grouping because, as we will show below, $\ThetaP{3}$ and $\DeltaP{3}$ can be defined as Hausdorff complexity classes built upon $\SigmaP{2}$.
Needless to say, the union of a step's classes is \emph{not} larger than its highest intermediate class, as the classes within a step are linearly ordered by inclusion.
We introduce the notion of hierarchy steps as a concept through which we look at the hierarchies and organize their investigation, rather than as a structural definition.

Since $\iExpSpace{i} \subseteq \iExpTime{(i+1)}$ (see \zcref{sec_prelim_central_complexity_classes}), and the latter is the ground level of $\iWExpHier{(i+1)}$, the iterated exponential hierarchies are ``stacked'', i.e., for all $i \geq 0$,
\[
  \cdots \subseteq \iExpSpace{(i-1)} \subseteq \iExpTime{i} \subseteq \iWExpHier{i} \subseteq \iExpSpace{i} \subseteq \iExpTime{(i+1)} \subseteq \iWExpHier{(i+1)} \subseteq \iExpSpace{(i+1)} \subseteq \cdots.
\]

Moreover, the iterated exponential hierarchies have very similar structures, which replicate with minimal differences from one hierarchy to the one right above.
The feature changing between these hierarchies is the number of intermediate levels, which are $i+2$ in $\iWExpHier{i}$---in this count, the \BHText over a main level is not considered as an intermediate level.
We could hence think of defining an \defin{Iterated Exponentials Meta-Hierarchy}, which is a kind of ``hierarchy of hierarchies'', whose meta-levels and meta-steps are the classes $\iExpTime{i}$ and the hierarchies $\iWExpHier{i}$s, respectively.
This meta-hierarchy equals the class $\Elementary = \bigcup_{i \geq 1} \DTime{\iExp{i}{n}}$~\cite{Simon1975,Papadimitriou1994}.
In \zcref{fig_iterated_exponentials_meta-hierarchy}, some levels and hierarchies of the meta-hierarchy are depicted---in the figure are also reported the ``Hausdorff complexity classes'', introduced below.
One of the main goals of this paper will be showing that the intermediate levels of the iterated exponential hierarchies can equivalently be characterized via Hausdorff complexity classes, which we will define and extensively discuss in the next section.

\subsection{Reductions and Hard Problems}
\label{sec_prelim_Reductions_Hard_Problems}

A language $\Language{A}$ is \defin{(Karp) reducible}, or \defin{many\nbdash-one reducible}, to a language $\Language{B}$, denoted by $\Language{A} \KarpRed[] \Language{B}$, if there exists a computable function~$f\colon \StringUniverse \to \StringUniverse$, called \emph{(Karp) reduction}, such that, for every string $w$, $f(w)$ is defined, and $w$ is a \yesinst of $\Language{A}$ iff $f(w)$ is a \yesinst of $\Language{B}$. %
A language $\Language{A}$ is \emph{polynomially (Karp) reducible} to a language $\Language{B}$, denoted by $\Language{A} \KarpRed \Language{B}$, if there exists a polynomial\nbdash-time computable (Karp) reduction from $\Language{A}$ to $\Language{B}$.

A language $\Language{A}$ is \defin{hard} for a complexity class $\ComplexityClass{C} \supseteq \PTime$ (but for the class $\PTime$), or $\ComplexityClass{C}$\HardSuffix, if every language in $\ComplexityClass{C}$ is polynomially reducible to $\Language{A}$;
if $\Language{A}$ is hard for $\ComplexityClass{C}$ and belongs to $\ComplexityClass{C}$, then $\Language{A}$ is \defin{complete} for~$\ComplexityClass{C}$, or $\ComplexityClass{C}$\CompleteSuffix.
Thus, problems that are complete for $\ComplexityClass{C}$ are the most difficult problems in $\ComplexityClass{C}$.
In particular, they do not belong to some lower class in the hierarchies, unless some collapse of the hierarchy's levels occurs.

\section{Hausdorff Reductions and Classes}
\label{sec_Hausdorff_reductions_classes}

In this section, we introduce the main notions of this paper:
we first define Hausdorff predicates and reductions, and then, through the latter, we define Hausdorff complexity classes.

\Citet{Wagner1987} was the first to introduce a notion of Hausdorff reduction, defined as a special case of \emph{polynomial} tt\nbdash-reductions.
For this reason, \citet{Buss1991} named \citeauthor{Wagner1987}'s Hausdorff reductions as ``normalized tt-reductions''.
This definition was however actually proposed \emph{only} for the polynomial case, and these authors did not consider a general notion going beyond the polynomial setting.
Since we aim at providing a Hausdorff-inspired characterization of the intermediate levels of the \emph{iterated exponential} hierarchies, we here need a different definition.
Interestingly enough, the definition that we here propose, although more general, is actually simpler than the one in the literature.
Both \citeauthor{Wagner1987}'s and our Hausdorff reduction notions
stem from a result due to \citeauthor{Hausdorff1962} regarding Boolean algebras over rings of sets, which we now recall.

A \defin{ring} $\SetOfSets{R}$ (of sets) is a collection of sets that is closed under set\nbdash-theoretic union and intersection;
the collection $\SetOfSets{R}$ may be infinite, and its elements may be sets of infinite cardinality \cite{Hausdorff1962}.
The \defin{Boolean closure of}, or \defin{Boolean algebra over}, $\SetOfSets{R}$ is the closure of $\SetOfSets{R}$ under set-theoretic difference~\cite{Hausdorff1962}.
This entity is named Boolean closure/algebra because it can equivalently be defined as the least superset of $\SetOfSets{R}$ closed under union, intersection, and complement in the \emph{universe} of~$\SetOfSets{R}$, i.e., the set of all elements from the sets of~$\SetOfSets{R}$.
A \defin{Hausdorff sequence} over $\SetOfSets{R}$ is a sequence $\set{S_1,\dots,S_m}$ of sets from $\SetOfSets{R}$ such that $S_1 \supseteq S_2 \supseteq\linebreak[0] \dots\linebreak[0] \supseteq\linebreak[0] S_m$.
The \defin{Hausdorff summation}, or \defin{combination}, of $\set{S_1,\dots,S_m}$ is defined as $(S_1 \setminus S_2) \cup (S_3 \setminus S_4) \cup \dots \cup (S_{m-1} \setminus S_m)$, where, if $m$ is an odd number, then this summation ends with `${} \cup S_m$', rather than with `${} \cup (S_{m-1} \setminus S_m)$'.
The terms Hausdorff sequence and Hausdorff summation are not standard and are introduced here for convenience.

Let $\SetOfSets{R}$ now be a ring including the empty set.
\citet{Hausdorff1962} proved that, for every set $A$ in the Boolean closure of such a ring $\SetOfSets{R}$, there exists a Hausdorff sequence $\set{S_1,\dots,S_{m}}$ over $\SetOfSets{R}$, where $m$ depends on $A$, whose Hausdorff summation equals $A$.
By this, an element $x$ from the universe of $\SetOfSets{R}$ belongs to $A$ iff $(x \in S_1 \land x \notin S_2) \lor \dots \lor (x \in S_{m-1} \land x \notin S_{m})$.
Furthermore, since $S_1 \supseteq S_2 \supseteq \dots \supseteq S_{m}$, we also have that $x \in A$ iff the number of true statements in the sequence `$x \in S_1$', `$x \in S_2$', \dots, `$x \in S_{m}$' is odd---in this sequence, if $x \in S_i$ then $x \in S_{i-1}$, for all $i > 1$, or, equivalently, if $x \notin S_i$ then $x \notin S_{i+1}$, for all $i \geq 1$.

The characterization individuated by \citeauthor{Hausdorff1962} can naturally be carried over to classes of languages, and hence to complexity classes.
Let $\ComplexityClass{C}$ be a language class closed under union and intersection, and including the empty language. 
Remember that a language is just a set of strings, therefore such a class $\ComplexityClass{C}$ is by assumption a ring including the empty set.
By \Citeauthor{Hausdorff1962}'s result, for every language $\Language{L}$ in the Boolean closure $\BHGeneric{\ComplexityClass{C}}$ of~$\ComplexityClass{C}$,%
\footnote{Classically, $\BHGeneric{\ComplexityClass{C}}$ denotes the Boolean closure of $\ComplexityClass{C}$, as the latter equals the \emph{B}\/oolean \emph{H}\/ierarchy over $\ComplexityClass{C}$ (see \zcref{sec_extended_Boolean_Hierarchies} for more).}
there exists a Hausdorff sequence $\set{\Language{D}_1,\dots,\Language{D}_m}$ of languages from $\ComplexityClass{C}$ (with $\Language{D}_1 \supseteq \Language{D}_2 \supseteq \dots \supseteq \Language{D}_m$)
whose Hausdorff summation equals $\Language{L}$.
By this, a string $w$ belongs to $\Language{L}$ iff $\SetSize{\set{z \mid 1 \leq z \leq m \land w \in \Language{D}_z}}$ is odd.
We say that $\set{\Language{D}_1,\dots,\Language{D}_m}$ is a \defin{Hausdorff sequence of length $m$ (over $\ComplexityClass{C}$) characterizing~$\Language{L}$}.

Our discussion now departs from the concept of Hausdorff reduction introduced by~\Citeauthor{Wagner1987} in~\cite{Wagner1987}.%
\footnote{\label{footnote_wagner_hausdorff_reductions_limited_to_polynomial_case}%
A Hausdorff reduction notion closer in spirit to ours was employed by \Citeauthor{Wagner1990} in~\cite{Wagner1990}.
However, the notion in~\cite{Wagner1990} was not formally defined and was limited to the polynomial case, making it difficult to generalize to the (iterated) exponential complexity classes.
Had a suitable generalization been identified, we believe it could have answered the questions about the Strong Exponential Hierarchy that had been left open by \citet{Hemachandra1987} a few years earlier.
Nevertheless, this possibility remained unnoticed, probably because the Hausdorff reduction notion in \cite{Wagner1990} was not precise or general enough to inspire further investigation.}
Our Hausdorff reductions generalize the property above, and transform the task of deciding whether a string $w$ belongs to a language $\Language{L}$ into a task akin to deciding whether $w$ belongs to the Hausdorff summation of a Hausdorff sequence whose length might depend on $\StringLength{w}$.
We start by introducing Hausdorff predicates.%
\footnote{One might consider defining generalized Hausdorff reductions via ``infinite Hausdorff sequences''. However, without careful thought, such a definition can become overly broad, even allowing non recursively enumerable languages to be Hausdorff characterized by regular languages, due to the infinitely\nbdash-long Hausdorff summations involved (see \zcref{sec_infinite_hausdorff_sequences} for more).
This would flatten the analysis by masking distinctions between Hausdorff classes. The definition via Hausdorff predicates here proposed reveals instead a rich structure of Hausdorff classes, and enables the Hausdorff\nbdash-characterization of the intermediate levels of the iterated exponential hierarchies.}

\begin{definition}[store=DefHausdPred]
\label{def_Hausdorff_predicate}
A \defin{Hausdorff predicate} $\Language{D}$ is a binary relation over $\HausdPredDomain$ such that:
\begin{itemize}[nosep]
  \item[\textbf{(\HausdSeqRequirementI)}] $\Language{D}(w,z) \geq \Language{D}(w,z+1)$, for every string $w$ and every integer $z \geq 1$; and
  \item[\textbf{(\HausdSeqRequirementII)}] for every string $w$, there exists a non\nbdash-negative integer $z$ such that $\Language{D}(w, z+1) = 0$; the \defin{Hausdorff index} of $w$ \Wrt $\Language{D}$, denoted $\HausdIndex{w}{\Language{D}}$, is the minimum integer $z \geq 0$ with this property.
\end{itemize}
\end{definition}
Notice that, by requirements~(\HausdSeqRequirementI) and (\HausdSeqRequirementII), the Hausdorff index $\HausdIndex{w}{\Language{D}}$ is also the maximum index $z \geq 0$ at which $\Language{D}(w,z) = 1$---for uniformity, for all Hausdorff predicates $\Language{D}$ and all strings $w$, we assume $\Language{D}(w,0) = 1$.

Via Hausdorff predicates, we will define Hausdorff reductions as the basis of Hausdorff complexity classes.
This requires us to specify two characteristics of Hausdorff predicates:
length/boundedness and complexity.

Let us start with the former.
From the notion of Hausdorff index stems that of bounded Hausdorff predicate. 
Intuitively, a Hausdorff predicate $\HausdSeq{D}$ is \defin{bounded} iff there is some function bounding every string's Hausdorff index \Wrt $\HausdSeq{D}$.
More precisely, like time and space functions,
a \defin{Hausdorff length function} $g\colon \NaturalsDomain \to \NaturalsDomain$ is a strictly positive nondecreasing
function.
A Hausdorff predicate $\HausdSeq{D}$ is \defin{$g(n)$\nbdash-long} iff $g(n)$ is a Hausdorff length function such that, for all strings~$w$, it holds that $\HausdIndex{w}{\HausdSeq{D}} \leq g(\StringLength{w})$.
Notice that the boundedness of a 
Hausdorff predicate $\Language{D}$ does \emph{not} imply the existence of a constant value $U$ such that $\HausdIndex{w}{\Language{D}} \leq U$ for all strings~$w$.

Let us now focus on Hausdorff predicates' complexity.
We start by introducing a tailored Turing machine model, with \emph{two} input tapes, deciding Hausdorff predicates.
This model easily lends itself to the definition of suitable time and space complexity notions for Hausdorff predicates.
Toward the end of this section, based on a few properties of Hausdorff predicates that we will highlight, we will see that for our purposes it is enough to rely on a more standard machine model with just one input tape, which will be adopted in the rest of the paper.

A Turing machine $\Machine{M}$ deciding a Hausdorff predicate is a machine with \emph{two} (read\nbdash-only) input tapes and possibly many (read/write) work tapes.
Such a machine $\Machine{M}$ receives its input, which is constituted by pairs $\pair{w,z} \in \HausdPredDomain$, on the two input tapes:
on the first tape the string $w$, and on the second tape the integer $z$ in its canonical binary representation.
The machine $\Machine{M}$ decides the Hausdorff predicate $\Language{D} \subseteq \HausdPredDomain$ iff, for every pair $\pair{w,z} \in \HausdPredDomain$, when $\Language{D}(w,z) = 1$ (resp., when $\Language{D}(w,z) = 0$), $\Machine{M}$ accepts (resp., rejects) in \emph{finite} time the pair $\pair{w,z}$.
The running time and space of $\Machine{M}$ is measured \Wrt to the size of $w$ alone, i.e., \Wrt to the size of the input on the first tape alone;
this means that, for the running space, the second input tape is regarded as a work tape.
The complexity bound is imposed \Wrt $\StringLength{w}$ alone to abide by the spirit of finite Hausdorff sequences (see above), and hence to avoid that hugely inflated indices~$z$ act as padding for the input pairs $\pair{w,z}$.
A Hausdorff predicate $\Language{D} \subseteq \HausdPredDomain$ can be decided in (non)deterministic $f(n)$ time (resp., space) iff pairs $\tup{w,z} \in \HausdPredDomain$ can be decided to belong to $\Language{D}$ or not in time (resp., space) $f(\StringLength{w})$ by a suitable (non)deterministic machine.
By this, we say that a Hausdorff predicate $\Language{D}$ belongs to a(n oracle) complexity class $\ComplexityClass{C}$, or that $\Language{D}$ is a $\ComplexityClass{C}$ Hausdorff predicate, if $\Language{D}$'s pairs $\tup{w,z}$ can be decided by a $\ComplexityClass{C}$ (oracle) machine whose running time or space is measured \Wrt \emph{the size of~$w$ alone}.
For example, a Hausdorff predicate $\Language{D}$ is in $\Oracle{\iNExpTime{i}}{\SigmaP{c-1}}$ if there is a nondeterministic oracle machine that, aided by a $\SigmaP{c-1}$ oracle, decides $\Language{D}(w,z)$ in time $\iExp{i}{p(\StringLength{w})}$, where $p(n) \in \PolFunctions$ is a polynomial.

Below we highlight two properties of Hausdorff predicates.
An intuitive first property is that if a Hausdorff predicate's length is bounded by some function $g(n)$, then it is also bounded by any greater function $h(n)$.
A less evident second property is that every Hausdorff predicate decidable in bounded time or space has bounded length.
We provide an intuition for the latter by considering bounded time (the rationale behind bounded space is similar).
If $\Language{D}$ is a Hausdorff predicate whose pairs $\pair{w,z}$ can be decided by a machine $\Machine{M}$ in $t(\StringLength{w})$ time,
then $t(\StringLength{w})$ also bounds the amount of bits of $z$ that $\Machine{M}$ can read from the second input tape within the allowed time.
By this, if $\Machine{M}$ returns the correct answer on \emph{every} pair $\pair{w,z}$, then the Hausdorff index of $w$ \Wrt $\Language{D}$ can be represented within $t(\StringLength{w})$ bits.
The two \zcref*[typeset=name,nocap]{theo_Hausd_pred_extendable,theo_Hausd_pred_bounded_complexity_imply_bounded_length} below formalize these properies.
Since the first property directly stems from the definition of Hausdorff predicates' length, its straightforward formal proof is omitted.

\begin{lemma}[store=HausdPredExtendable]
\label{theo_Hausd_pred_extendable}
Let $\Language{D} \subseteq \HausdPredDomain$ be a Hausdorff predicate, and let $g$ and $h$ be two Hausdorff length functions such that $g(n) \leq h(n)$, for all $n \geq 1$.
If $\Language{D}$ is $g(n)$\nbdash-long, then $\Language{D}$ is also $h(n)$\nbdash-long.
\end{lemma}

\begin{lemma}[store=HausdPrefBoundedComplexityImplyBoundedLength]
\label{theo_Hausd_pred_bounded_complexity_imply_bounded_length}
Let $f\colon \NaturalsDomain \to \NaturalsDomain$ be a time\nbdash-constructible time (resp., a space\nbdash-constructible space) function, and let $\Language{A}$ be a language.
If $\Language{D} \subseteq \HausdPredDomain$ is a Hausdorff predicate decidable in (non)deterministic $f(n)$ time (resp., space), possibly with the aid of an oracle for $\Language{A}$, then $\Language{D}$ is $(2^{f(n)}{-}1)$\nbdash-long.
\end{lemma}

\begin{proof}
Because $\Language{D}$ can be decided in $f(n)$ time (resp., space), with possibly the aid of an oracle for $\Language{A}$, there exists a(n oracle) Turing machine $\Machine{M}$ deciding, for every pair $\pair{w,z} \in \HausdPredDomain$, whether $\pair{w,z} \in \Language{D}$ or not in time (resp., space) $f(\StringLength{w})$.
Assume by contradiction that $\Language{D}$ is \emph{not} $(2^{f(n)}{-}1)$\nbdash-long.

Since $\Language{D}$ is not $(2^{f(n)}{-}1)$\nbdash-long, there exists a string $s$ whose Hausdorff index $\HausdIndex{s}{\Language{D}}$ \Wrt $\Language{D}$ is strictly greater than $2^{f(\StringLength{s})}-1$.
This implies that $\HausdIndex{s}{\Language{D}} > 0$, as $f$ is strictly positive.
Consider the value $\tilde z = 2 \cdot \HausdIndex{s}{\Language{D}}$.
By definition, $\tilde z > \HausdIndex{s}{\Language{D}}$ as $\HausdIndex{s}{\Language{D}} > 0$.
The canonical binary representation of $\tilde z$, compared to that of $\HausdIndex{s}{\Language{D}}$, has an extra digit `$0$' as the least significant bit;
the rest of their canonical binary representations are the same.
By this, $\StringLength{\tilde z} > \StringLength{\HausdIndex{s}{\Language{D}}}$.
Since $2^{f(\StringLength{s})}-1$ is the greatest value representable in binary over $f(\StringLength{s})$ bits, and we are assuming $\HausdIndex{s}{\Language{D}} > 2^{f(\StringLength{s})}-1$, we have $\StringLength{\tilde z} > \StringLength{\HausdIndex{s}{\Language{D}}} > f(\StringLength{s})$.
By definition of $\HausdIndex{s}{\Language{D}}$ and $\tilde z$, it holds that ${\Language{D}}(s,\HausdIndex{s}{\Language{D}}) = 1$ and ${\Language{D}}(s,\tilde z) = 0$, because $\HausdIndex{s}{\Language{D}} < \tilde z$.
Hence, the machine $\Machine{M}$ must answer \yeslbl on $\pair{s,\HausdIndex{s}{\Language{D}}}$ \emph{and} \nolbl on $\pair{s,\tilde z}$.
However, since $\Machine{M}$ performs at most $f(\StringLength{s})$ computation steps (resp., $\Machine{M}$ may scan at most $f(\StringLength{s})$ tape cells besides those on the first input tape),
$\Machine{M}$ does not have enough time (resp., enough allowed tape cells to scan) to read from the second input tape the least significant bit of $\HausdIndex{s}{\Language{D}}$ and the two least significant bits of $\tilde z$---remember that the least significant bits of $\tilde z$ and $\HausdIndex{s}{\Language{D}}$ are written on the right\nbdash-hand side of the second input tape, and the machine $\Machine{M}$ starts its computation by having the head of the second input tape over the left\nbdash-most bit.
Therefore, from a time\nbdash-bounded (resp., space\nbdash-bounded) computation perspective, the input pairs $\pair{s,\HausdIndex{s}{\Language{D}}}$ and $\pair{s,\tilde z}$ are completely indistinguishable to $\Machine{M}$.
Nevertheless, $\Machine{M}$ must provide two different answers on these two inputs:
a contradiction.
Thus, the Hausdorff predicate $\Language{D}$ must be $(2^{f(n)}{-}1)$\nbdash-long.
\end{proof}

Via Hausdorff predicates we define the notion of Hausdorff reduction/characterization.
Our definition of Hausdorff reduction combines, generalizes, and simplifies, those in \cite{Wagner1987,Wagner1990}.%
\footnote{A different (polynomial) Hausdorff reduction notion whose translation to Hausdorff predicates would not require for them to satisfy (\HausdSeqRequirementI) was proposed by \citet{ArvindKM1993}. Also their definition was anyway based on polynomial tt-reductions.}

\begin{definition}[store=DefHausdRed]
\label{def_Hausdorff_reduction}
A language $\Language{L} \subseteq \StringUniverse$ \defin{Hausdorff reduces} to, or is \defin{Hausdorff characterized} by, a Hausdorff predicate $\Language{D} \subseteq \HausdPredDomain$ if and only if
$w \in \Language{L} \Leftrightarrow \HausdIndex{w}{\Language{D}}$ is odd, for every string $w$.
If $\Language{D}$ is $g(n)$\nbdash-long, 
we equivalently say that
$\Language{L}$ has a $g(n)$\nbdash-long Hausdorff characterization, or that
$\Language{L}$ is $g(n)$\nbdash-Hausdorff reducible to $\Language{D}$, or that
there is a Hausdorff reduction of length $g(n)$ from $\Language{L}$ to $\Language{D}$.  
\end{definition}

Notice that our notion of Hausdorff reduction does not involve any transformation of the string $w$ when it is ``passed'' from $\Language{L}$ to $\Language{D}$, on the contrary of what happens 
in~\cite{Wagner1987,JennerKL1989,Buss1991,ArvindKM1993}.
In the present paper, a Hausdorff reduction is a very simple mapping from $\Language{L}$ to $\Language{D}$.

Using the concept of Hausdorff reduction, we can now introduce the notion of Hausdorff complexity class.

\begin{definition}[store=DefHausdClass]
Let $g\colon \NaturalsDomain \to \NaturalsDomain$ be a 
Hausdorff length function, and let $\ComplexityClass{C}$ be a(n oracle) complexity class.
The $\ComplexityClass{C}$ \defin{Hausdorff (complexity) class of length $g(n)$}, denoted by $\BoundedHausdCLASS{g(n)}{\ComplexityClass{C}}$, is the class of all languages that are Hausdorff reducible to some $\ComplexityClass{C}$ Hausdorff predicate of length~$g(n)$.
\end{definition}

Notice that Hausdorff complexity classes are \emph{not} defined as classes of Hausdorff predicates.
To make this clearer for the reader, let us consider the following example involving $\NPTime$. By \zcref{theo_Hausd_pred_bounded_complexity_imply_bounded_length}, there are $\NPTime$ Hausdorff predicates of polynomial and exponential length.
The class of $\NPTime$ Hausdorff predicates of polynomial (resp., exponential) length is a set of Hausdorff predicates, i.e., a set of peculiar languages with a specific structure (see \zcref{def_Hausdorff_predicate}) within the class $\NPTime$.
Since $\NPTime$ Hausdorff predicates of polynomial length are also $\NPTime$ Hausdorff predicates of exponential length (see \zcref{theo_Hausd_pred_extendable}), we have that the set of $\NPTime$ Hausdorff predicates of polynomial length is contained in the set of $\NPTime$ Hausdorff predicates of exponential length, which in turn is contained in $\NPTime$.
On the other hand, the $\NPTime$ Hausdorff complexity classes of polynomial and exponential length contain $\NPTime$, as these classes will be shown to equal $\ThetaP{2}$ and $\DeltaP{2}$, respectively.

If $\Language{L}$ is a language characterized by some ($g(n)$\nbdash-long) Hausdorff predicate of complexity $\ComplexityClass{C}$, we will also say that $\Language{L}$ is a \defin{$\ComplexityClass{C}$ Hausdorff language (of length $g(n)$)}---notice the difference between a Hausdorff predicate of length $g(n)$ and a Hausdorff language of length $g(n)$.
By this, $\BoundedHausdCLASS{g(n)}{\ComplexityClass{C}}$ is also the class of all $\ComplexityClass{C}$ Hausdorff languages of length~$g(n)$.
Observe that, if $\Language{L}$ is a language $g(n)$\nbdash-Hausdorff reducible to a Hausdorff predicate $\Language{D}$, for every string $w$, it holds that $w \in \Language{L} \Leftrightarrow \SetSize{\set{z \mid 1 \leq z \leq g(\StringLength{w}) \land \Language{D}(w,z) = 1}}$ is odd.
For this reason, we can just focus on the $g(\StringLength{w})$\nbdash-long prefix of the sequence $\Language{D}(w,1), \dots, \Language{D}(w,z), \dots$, to decide whether $w \in \Language{L}$.

If $\BoundedHausdCLASS{g(n)}{\ComplexityClass{C}}$ is a Hausdorff class, $\ComplementPrefix\BoundedHausdCLASS{g(n)}{\ComplexityClass{C}}$ is the class of languages whose complements are in $\BoundedHausdCLASS{g(n)}{\ComplexityClass{C}}$.
By this, a language $\Language{L}$ belongs to $\ComplementPrefix\BoundedHausdCLASS{g(n)}{\ComplexityClass{C}}$ iff there is a $\ComplexityClass{C}$ Hausdorff predicate $\Language{D}$ of length $g(n)$ such that, for every string $w$, it holds that $w \in \Language{L}$ iff $\SetSize{\set{z \mid z \geq 1 \land \Language{D}(w,z) = 1}} \text{ is \emph{even}}$, or, equivalently, iff $\HausdIndex{w}{\Language{D}}$ is \emph{even}.

For a family of functions $G$, we define the Hausdorff complexity class $\BoundedHausdCLASS{G}{\ComplexityClass{C}} = \bigcup_{g(n) \in G} \set{\BoundedHausdCLASS{g(n)}{\ComplexityClass{C}}}$.
For example, $\BoundedHausdCLASS{2^{\PolFunctions}}{\Oracle{\NExpTime}{\SigmaP{c-1}}}$ is the class of all $\Oracle{\NExpTime}{\SigmaP{c-1}}$ Hausdorff languages of length $g(n)$, for some $g(n) \in 2^{\PolFunctions}$, or, more roughly, the class of all $\Oracle{\NExpTime}{\SigmaP{c-1}}$ Hausdorff languages of exponential length.
By $\BoundedHausdCLASS{\GenericLength}{\ComplexityClass{C}}$ we denote the class of $\ComplexityClass{C}$ Hausdorff languages of \emph{any} length---notice that $\BoundedHausdCLASS{\GenericLength}{\ComplexityClass{C}}$ contains all $\ComplexityClass{C}$ Hausdorff predicates of \emph{bounded} length, and \emph{not} $\ComplexityClass{C}$ Hausdorff predicates that are \emph{un}\/bounded.
The class $\ComplementPrefix\BoundedHausdCLASS{G}{\ComplexityClass{C}}$ is defined in the natural way.

We now show that considering $\ComplexityClass{C}$ Hausdorff classes where $\ComplexityClass{C}$ is deterministic does not individuate any new interesting class of languages, and hence we can focus on $\ComplexityClass{C}$ Hausdorff classes with $\ComplexityClass{C}$ nondeterministic.
The intuition behind this is that the Hausdorff summation of languages from a deterministic class like $\iExpTime{i}$ is still a language in $\iExpTime{i}$, because $\iExpTime{i}$ is closed under union, intersection, and complement.

\begin{lemma}\label{theo_hausdorff_useless_for_deterministic_classes}
Let $f\colon \NaturalsDomain \to \NaturalsDomain$ be a time\nbdash-constructible time function (resp., space\nbdash-constructible space function), such that $f(n)$ is at least linear, and let $\Language{A}$ be a decidable language.
Then,
\begin{align*}
\BoundedHausdCLASS{\GenericLength}{\Big(\DTimeOracle{f(n)}{\Language{A}}\Big)} & \subseteq \DTimeOracle{O(f(n)^2)}{\Language{A}} \\
\BoundedHausdCLASS{\GenericLength}{\Big(\DSpaceOracle{f(n)}{\Language{A}}\Big)} & \subseteq \DSpaceOracle{O(f(n))}{\Language{A}}.
\end{align*}
\end{lemma}

\begin{proof}
Let $\Language{L} \in \BoundedHausdCLASS{\GenericLength}{\Big(\DTimeOracle{f(n)}{\Language{A}}\Big)}$ be a language.
We prove $\Language{L} \in \DTimeOracle{O(f(n)^2)}{\Language{A}}$.

Since $\Language{L} \in \BoundedHausdCLASS{\GenericLength}{\Big(\DTimeOracle{f(n)}{\Language{A}}\Big)}$, $\Language{L}$ Hausdorff reduces to a $\DTimeOracle{f(n)}{\Language{A}}$ Hausdorff predicate $\Language{D}$ which, by \zcref{theo_Hausd_pred_bounded_complexity_imply_bounded_length}, is $(2^{f(n)}{-}1)$\nbdash-long.
Hence, for every string $w$, it holds $w \in \Language{L} \Leftrightarrow \HausdIndex{w}{\Language{D}}$ is odd, and $\HausdIndex{w}{\Language{D}} \leq 2^{f(\StringLength{w})} - 1$.
By this, we can decide whether $w \in \Language{L}$ by computing the Hausdorff index $\HausdIndex{w}{\Language{D}}$ of $w$ \Wrt $\Language{D}$ and answer \yeslbl iff $\HausdIndex{w}{\Language{D}}$ is odd.
Because we have $\HausdIndex{w}{\Language{D}} \leq 2^{f(\StringLength{w})} - 1$, to compute $\HausdIndex{w}{\Language{D}}$ it suffices to carry out a binary search within the domain $[0, 2^{f(\StringLength{w})}-1]$, and at each step test whether $\Language{D}(w,z) = 1$ or not.
Observe that a binary search in that domain requires at most ${f(\StringLength{w})}$ tests, and each test can be carried out in \emph{deterministic} $f(\StringLength{w})$ time (\Wrt the size of $w$ only), with possibly the aid of an oracle for $\Language{A}$.
Since testing $\Language{D}(w,z) = 1$ is feasible in deterministic time, the very same machine progressing in the binary search can carry out the test as well, without the need of calling an additional oracle, besides the possible calls to $\Language{A}$.
Once the Hausdorff index of $w$ \Wrt $\Language{D}$ is computed, we can answer \yeslbl according to its parity.
Notice that the entire procedure is feasible in $O(f(n)^2)$ deterministic time, with possibly the aid of an oracle for $\Language{A}$.

Proving $\BoundedHausdCLASS{\GenericLength}{\Big(\DSpaceOracle{f(n)}{\Language{A}}\Big)}  \subseteq \DSpaceOracle{O(f(n))}{\Language{A}}$ is similar.
Just consider the procedure outlined above, and notice that the space required amounts to that needed to store the indices $z$ for the binary search, which are $f(n)$ bits, plus the space needed to carry out the tests $\Language{D}(w,z) = 1$, which by assumption is $f(n)$.
\end{proof}

From the previous \zcref*[typeset=name,nocap]{theo_hausdorff_useless_for_deterministic_classes}, the next corollary easily follows (the straightforward formal proof is omitted).

\begin{corollary}
Let $i \geq 0$ be an integer and let $\ComplexityClass{D}$ be a complexity class.
Then,
\begin{align*}
  \BoundedHausdCLASS{\GenericLength}{\Oracle{\iExpTime{i}}{\ComplexityClass{D}}} & = \Oracle{\iExpTime{i}}{\ComplexityClass{D}} \\
  \BoundedHausdCLASS{\GenericLength}{\Oracle{\iExpSpace{i}}{\ComplexityClass{D}}} & = \Oracle{\iExpSpace{i}}{\ComplexityClass{D}}.
\end{align*}
\end{corollary}

The above results state that Hausdorff classes over deterministic time and space classes do not add much to the plain classes.
Therefore, it makes sense to focus our analysis on Hausdorff classes over nondeterministic classes.
However, by Savitch's theorem, for all integers $i \geq 0$, it holds that $\iExpSpace{i} = \iNExpSpace{i}$.
Hence, since in this paper we will not deal with Hausdorff classes defined over nondeterministic logspace, we can also avoid to deal with Hausdorff classes defined over nondeterministic space classes.
By all these considerations, in this paper we will look at Hausdorff classes defined over nondeterministic time classes.
More specifically, we will consider the Hausdorff complexity classes $\BoundedHausdCLASS{\iExpPolFunctions{j}}{\Oracle{\iNExpTime{i}}{\SigmaP{c-1}}}$,
with $i,j \geq 0$ and $c \geq 1$.

To conclude this section, we observe that from the above two properties of Hausdorff predicates, stated in \zcref{theo_Hausd_pred_extendable,theo_Hausd_pred_bounded_complexity_imply_bounded_length} directly follow two corresponding properties of Hausdorff complexity classes.
The first, descending from \zcref{theo_Hausd_pred_extendable}, is pretty intuitive and it suggests that ``longer'' Hausdorff classes are at least as expressive as shorter ones;
this is needed for some of the proofs in the rest of our discussion.
The rather interesting second property, stemming from \zcref{theo_Hausd_pred_bounded_complexity_imply_bounded_length}, highlights that there is a threshold above which having longer Hausdorff classes does not add anything.
The two \zcref*[typeset=name,nocap]{theo_longer_Hausdorff_as_expressive,theo_longer_Hausdorff_bound} below formalize the two properties.

The first property can easily be proven by noticing that, if $\Language{L}$ is a language in $\BoundedHausdCLASS{g(n)}{\Big(\NTimeOracle{t(n)}{\Language{A}}\Big)}$, there must be a $\NTimeOracle{t(n)}{\Language{A}}$ Hausdorff predicate $\Language{D}$ of length $g(n)$ which $\Language{L}$ Hausdorff reduces to.
Since, by \zcref{theo_Hausd_pred_extendable}, $\Language{D}$ has also length $h(n)$, if $g(n) \leq h(n)$, the language $\Language{L}$ also belongs to $\BoundedHausdCLASS{h(n)}{\Big(\NTimeOracle{t(n)}{\Language{A}}\Big)}$.

\begin{lemma}[store=longerHausdorffAsExpressive]
\label{theo_longer_Hausdorff_as_expressive}
Let $t\colon \NaturalsDomain \to \NaturalsDomain$ be a time\nbdash-constructible time function, let $\Language{A}$ be a decidable language, and let $g$ and $h$ be two Hausdorff length functions such that $g(n) \leq h(n)$, for all $n \geq 1$.
Then, \[\BoundedHausdCLASS{g(n)}{\Big(\NTimeOracle{t(n)}{\Language{A}}\Big)} \subseteq \BoundedHausdCLASS{h(n)}{\Big(\NTimeOracle{t(n)}{\Language{A}}\Big)}.\]
\end{lemma}

From the above result, the next corollary easily follows (the straightforward formal proof is omitted).

\begin{corollary}
\label{theo_restricted_main_levels_longer_Hausdorff_as_expressive}
Let $i \geq 0$ be an integer, let $\ComplexityClass{D}$ be a complexity class, and let $g$ and $h$ be two Hausdorff length functions such that $g(n) \leq h(n)$ for all $n \geq 1$.
Then,
\[
\BoundedHausdCLASS{g(n)}{\Oracle{\iNExpTime{i}}{\ComplexityClass{D}}} \subseteq \BoundedHausdCLASS{h(n)}{\Oracle{\iNExpTime{i}}{\ComplexityClass{D}}}.
\]
\end{corollary}

Also the second property is easily proven.
If $\Language{L}$ is a language in $\BoundedHausdCLASS{\GenericLength}{\Big(\NTimeOracle{t(n)}{\Language{A}}\Big)}$, then $\Language{L}$ Hausdorff reduces to a $\NTimeOracle{t(n)}{\Language{A}}$ Hausdorff predicate $\Language{D}$.
Since, by \zcref{theo_Hausd_pred_bounded_complexity_imply_bounded_length}, $\Language{D}$ has length $2^{t(n)}{-}1$, $\Language{L}$ belongs to $\BoundedHausdCLASS{2^{t(n)}}{\Big(\NTimeOracle{t(n)}{\Language{A}}\Big)}$, too;
clearly, it also holds $\BoundedHausdCLASS{2^{t(n)}}{\Big(\NTimeOracle{t(n)}{\Language{A}}\Big)} \subseteq \BoundedHausdCLASS{\GenericLength}{\Big(\NTimeOracle{t(n)}{\Language{A}}\Big)}$.

\begin{lemma}[store=longerHausdorffBound]
\label{theo_longer_Hausdorff_bound}
Let $t\colon \NaturalsDomain \to \NaturalsDomain$ be a time\nbdash-constructible time function and let $\Language{A}$ be a decidable language.
Then, \[\BoundedHausdCLASS{\GenericLength}{\Big(\NTimeOracle{t(n)}{\Language{A}}\Big)} = \BoundedHausdCLASS{2^{t(n)}}{\Big(\NTimeOracle{t(n)}{\Language{A}}\Big)}.\]
\end{lemma}

Thus, by the last lemma, Hausdorff reductions asymptotically longer than $2^{t(n)}$ to Hausdorff predicates in $\NTimeOracle{t(n)}{\Language{A}}$ are not more powerful than $2^{t(n)}$\nbdash-long Hausdorff reductions to the same Hausdorff predicates.

From the above result, the next corollary easily follows (the straightforward formal proof is omitted).

\begin{corollary}
\label{theo_restricted_main_levels_longer_Hausdorff_bound}
Let $i \geq 0$ be an integer and let $\ComplexityClass{D}$ be a complexity class.
Then,
\[
\BoundedHausdCLASS{\GenericLength}{\Oracle{\iNExpTime{i}}{\ComplexityClass{D}}} = \BoundedHausdCLASS{\iExpPolFunctions{i+1}}{\Oracle{\iNExpTime{i}}{\ComplexityClass{D}}}.
\]
\end{corollary}

Since we will only deal with Hausdorff predicates decidable in nondeterministic polynomial time, and up, the above two-input-tapes machine can conveniently be replaced by a more standard one-input-tape machine.
More specifically, in what follows we will have that a machine $\Machine{M}$ decides a Hausdorff predicate $\Language{D} \subseteq \HausdPredDomain$ iff, for every pair $\pair{w,z} \in \HausdPredDomain$, when $\Language{D}(w,z) = 1$ (resp., when $\Language{D}(w,z) = 0$), $\Machine{M}$ accepts (resp., rejects) in \emph{finite} time the input string ``$w \hashsep z$'', where `$\hashsep$' is a tape symbol neither appearing in $w$ nor in $z$, and the integer $z$ is written on tape in canonical binary form.
In accordance with the definition above, the running time of such a machine is measured \Wrt the size of the $w$ part alone of the input string $w \hashsep z$.
It is not hard to see that this different machine model does not impact the study of the asymptotic time complexity of Hausdorff predicates.

\section{Charting the Iterated Exponential Hierarchies via Hausdorff Classes}
\label{sec_charting_top}

One of this paper's aims is showing that the intermediate levels of the iterated exponential hierarchies can equivalently be characterized via classes of Hausdorff languages of increasing lengths.
This provides us a unifying and elegant perspective on the levels of these hierarchies.
More specifically, we will show that:
\begin{alignat*}{3}
  \SigmaIWExp{i}{c} & \qquad {} = {} \qquad & & \BoundedHausdCLASS{1}{\Oracle{\iNExpTime{i}}{\SigmaP{c-1}}} \\
  \BHGeneric{\SigmaIWExp{i}{c}} & \qquad {} = {} \qquad & & \BoundedHausdCLASS{O(1)}{\Oracle{\iNExpTime{i}}{\SigmaP{c-1}}} \\
  \BoundedOracle{\iExpTime{i}}{\SigmaP{c}}{\LogFunctions} & \qquad = \qquad & & \BoundedHausdCLASS{\PolFunctions}{\Oracle{\iNExpTime{i}}{\SigmaP{c-1}}} \\
  \BoundedOracle{\iExpTime{i}}{\SigmaP{c}}{\PolFunctions} & \qquad = \qquad & & \BoundedHausdCLASS{2^\PolFunctions}{\Oracle{\iNExpTime{i}}{\SigmaP{c-1}}} \\
  & \qquad \symbolwithin{\vdots}{=} \qquad & & \\
  \BoundedOracle{\iExpTime{i}}{\SigmaP{c}}{\iExpPolFunctions{i-1}} & \qquad = \qquad & & \BoundedHausdCLASS{\iExpPolFunctions{i}}{\Oracle{\iNExpTime{i}}{\SigmaP{c-1}}} \\
  \Oracle{\iExpTime{i}}{\SigmaP{c}} = \BoundedOracle{\iExpTime{i}}{\SigmaP{c}}{\iExpPolFunctions{i}} & \qquad = \qquad & & \BoundedHausdCLASS{\iExpPolFunctions{i+1}}{\Oracle{\iNExpTime{i}}{\SigmaP{c-1}}}.
\end{alignat*}

Besides providing a novel characterization of the intermediate levels of the iterated exponential hierarchies, Hausdorff classes are also a powerful tool to chart these hierarchies.
Indeed, Hausdorff classes offer a simple way to individuate equivalent oracle classes spanning the \ItExpHMetaText.
More specifically, via the results in this section, for $i,j \geq 0$ and a class $\ComplexityClass{C} = \Oracle{\iNExpTime{i}}{\iNExpTime{j}}$, we will be able for example to answer questions like:
What is $\ComplexityClass{C}$'s relationship with $\Oracle{\iNExpTime{k}}{\iNExpTime{\ell}}$, when $k,\ell \geq 0$?
Where is $\ComplexityClass{C}$ located within the meta-hierarchy?
Is $\ComplexityClass{C}$ a main or an intermediate level of one of the iterated exponential hierarchies?
Which level and which hierarchy, more precisely?
We can hence build a precise map of the \ItExpHMetaText, where the hierarchy levels, both main and intermediate, are linked to the respective Hausdorff classes.
In \zcref{fig_iterated_exponentials_meta-hierarchy}, we depict some hierarchies and classes of the meta-hierarchy, together with the equivalent Hausdorff classes;
in the figure, we can see how the \iWEHsText{i} are stacked one over the other, and how the number of their intermediate levels increase from one hierarchy to the one above.

\begin{figure}
  \centering
  \includegraphics[width=.98\textwidth]{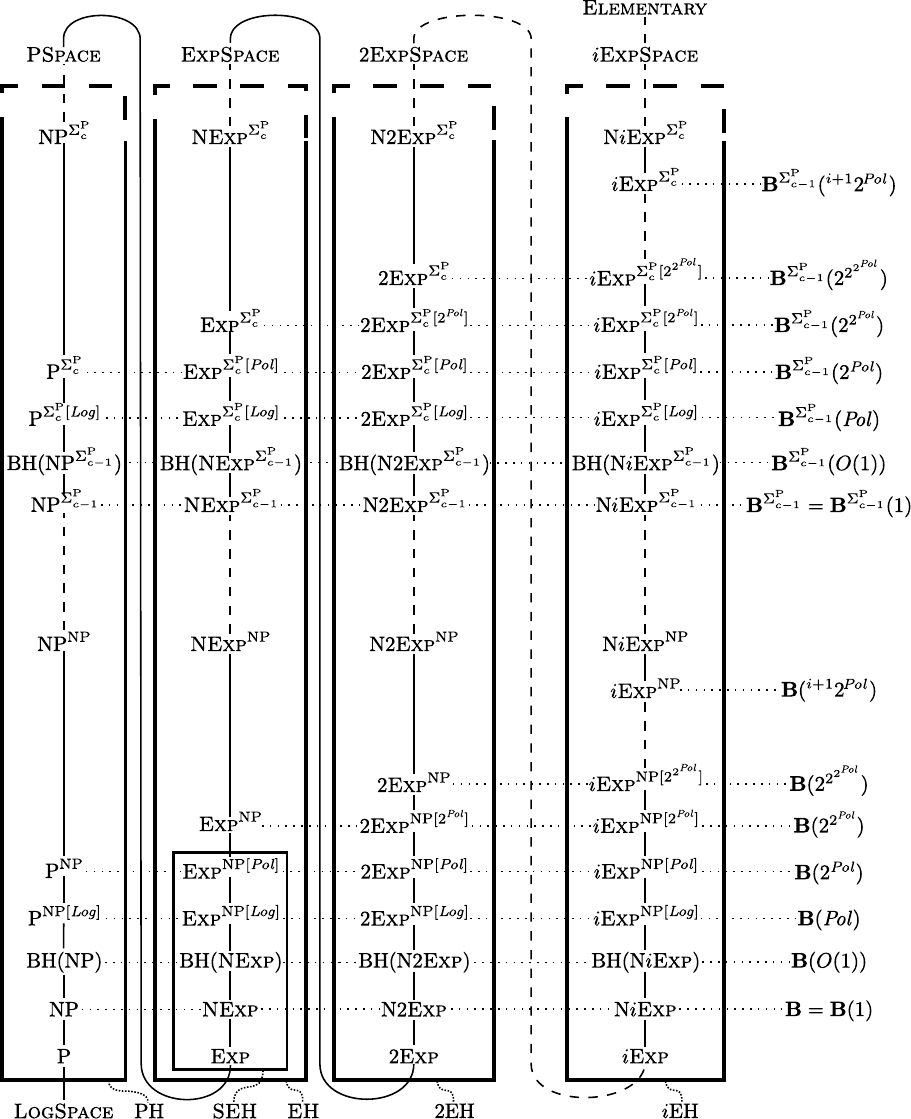}
  \caption{Some hierarchies of the Iterated Exponentials Meta-Hierarchy. The two on the left are the Polynomial ($\PolHier$) and the (Weak) Exponential ($\WExpHier$) Hierarchies. Observe how the Strong Exponential Hierarchy ($\SExpHier$) is a portion of $\WExpHier$'s first step.
  The other two hierarchies depicted are the $2$\nbdash-Exponential ($\iWExpHier{2}$) and the $i$\nbdash-Exponential ($\iWExpHier{i}$) Hierarchy. The symbol $\ComplexityClass{B}$ refers to the first main level of the \iWEHText{i} considered.}
  \label{fig_iterated_exponentials_meta-hierarchy}
\end{figure}

To chart the meta-hierarchy, in this section we investigate the oracle classes built via combinations of $\iExpTime{i}$, $\iNExpTime{i}$, and $\iExpSpace{i}$, oracle machines and oracles.
In \zcref{sec_charting_exp_oracles}, we look at the oracle complexity classes $\Oracle{\iExpTime{i}}{\Oracle{\iExpTime{j}}{\SigmaP{c}}}$, $\Oracle{\iNExpTime{i}}{\Oracle{\iExpTime{j}}{\SigmaP{c}}}$, and $\Oracle{\iExpSpace{i}}{\Oracle{\iExpTime{j}}{\SigmaP{c}}}$.
Then, in \zcref{sec_charting_nexp_oracles},
we analyze 
$\Oracle{\iExpTime{i}}{\Oracle{\iNExpTime{j}}{\SigmaP{c-1}}}$~(\zcref{sec_iEXP_jNEXP}), $\Oracle{\iNExpTime{i}}{\Oracle{\iNExpTime{j}}{\SigmaP{c-1}}}$~(\zcref{sec_details_NEXP_NEXP}), and $\Oracle{\iExpSpace{i}}{\Oracle{\iNExpTime{j}}{\SigmaP{c-1}}}$~(\zcref{sec_expspace_nexp_oracle_classes}).
We continue by discussing in \zcref{sec_discussion_Hausdorff_perspective} how the Hausdorff classes characterization provides structural insights into the equivalences between oracle classes.
We conclude with \zcref{sec_top_seh}, where we analyze the \SEHText via the Hausdorff classes perspective.
Through the latter we close \citeauthor{Hemachandra1989}'s~\cite{Hemachandra1989} open questions, and we see how the \SEHText is actually a portion of the \WEHStressedText{}'s first step.

\subsection{Deterministic Exponential-Time Oracles}
\label{sec_charting_exp_oracles}

In this section we look at the classes $\Oracle{\iExpTime{i}}{\Oracle{\iExpTime{j}}{\SigmaP{c}}}$, $\Oracle{\iNExpTime{i}}{\Oracle{\iExpTime{j}}{\SigmaP{c}}}$, and $\Oracle{\iExpSpace{i}}{\Oracle{\iExpTime{j}}{\SigmaP{c}}}$.
The obtained inclusion and the equivalence relationships are summarized in \zcref{fig_subset_relation_classes_exp_oracle} and \zcref{fig_equivalence_relation_classes_exp_oracle}, respectively.
These results' proofs, which are deferred to \zcref{sec_detailed_proofs_charting_exp_oracles}, rely on the analysis of machine computations and padding arguments.

\begin{figure}
  \centering
  {%
  \renewcommand{\arraystretch}{1.75}
  \begin{tabular}{|r|c|l|}
    \hline
    $\Oracle{\iExpTime{i}}{\Oracle{\iExpTime{j}}{\SigmaP{c}}}$ & \multirow{3}{*}{${} \subseteq \Oracle{\iExpTime{(i+j)}}{\SigmaP{c}} \subseteq {}$} & $\BoundedOracle{\iExpTime{k}}{\Oracle{\iExpTime{(i+j-k)}}{\SigmaP{c}}}{1}$ \\
    \cline{1-1}\cline{3-3}
    \raisebox{1ex}{*} \hspace{3em} $\Oracle{\iNExpTime{i}}{\Oracle{\iExpTime{j}}{\SigmaP{c}}}$ &  & $\BoundedOracle{\iNExpTime{k}}{\Oracle{\iExpTime{(i+j-k)}}{\SigmaP{c}}}{1}$ \hspace{3em} \raisebox{1ex}{*} \\
    \cline{1-1}\cline{3-3}
    $\Oracle{\iExpSpace{(i-1)}}{\Oracle{\iExpTime{j}}{\SigmaP{c}}}$ &  & $\BoundedOracle{\iExpSpace{(k-1)}}{\Oracle{\iExpTime{(i+j-k)}}{\SigmaP{c}}}{1}$ \\
    \hline
  \end{tabular}%
  }
  \caption{Inclusion relationships between oracle classes. In the table, $i,j,c \geq 0$ and $0 \leq k \leq i+j$. {*}In this row, $j$ is assumed to be such that $j \geq 1$; for $j = 0$, the equivalence $\Oracle{\iNExpTime{i}}{\Oracle{\iExpTime{j}}{\SigmaP{c}}} = \Oracle{\iNExpTime{i}}{\SigmaP{c}}$ holds.}
  \label{fig_subset_relation_classes_exp_oracle}
\end{figure}

\begin{figure}
  \centering
  {%
  \renewcommand{\arraystretch}{1.75}
  \begin{tabular}{|r|c|l|}
    \hline
    $\Oracle{\iExpTime{i}}{\Oracle{\iExpTime{j}}{\SigmaP{c}}}$ & \multirow{3}{*}{${} = \Oracle{\iExpTime{\ell}}{\SigmaP{c}} = {}$} & $\BoundedOracle{\iExpTime{i'}}{\Oracle{\iExpTime{j'}}{\SigmaP{c}}}{1}$ \\
    \cline{1-1}\cline{3-3}
    \raisebox{1ex}{*} \hspace{3em} $\Oracle{\iNExpTime{i}}{\Oracle{\iExpTime{j}}{\SigmaP{c}}}$ &  & $\BoundedOracle{\iNExpTime{i'}}{\Oracle{\iExpTime{j'}}{\SigmaP{c}}}{1}$ \hspace{3em} \raisebox{1ex}{*} \\
    \cline{1-1}\cline{3-3}
    $\Oracle{\iExpSpace{(i-1)}}{\Oracle{\iExpTime{j}}{\SigmaP{c}}}$ &  & $\BoundedOracle{\iExpSpace{(i'-1)}}{\Oracle{\iExpTime{j'}}{\SigmaP{c}}}{1}$ \\
    \hline
  \end{tabular}%
  }
  \caption{Equivalence relationships between oracle classes. In the table, $i,i',j,j',c \geq 0$ and $i+j = \ell = i'+j'$. {*}In this row, $j,j'$ are assumed to be such that $j,j' \geq 1$.}
  \label{fig_equivalence_relation_classes_exp_oracle}
\end{figure}

We start with the classes $\Oracle{\iExpTime{i}}{\Oracle{\iExpTime{j}}{\SigmaP{c}}}$.
Intuitively, since two \emph{deterministic} $i$- and \iExponential{j}{}\nbdash-time machines are chained, they give rise to a deterministic class where the exponential orders $i$ and $j$ are combined.

\begin{theorem}[store=EEcontainment]
\label{theo_exp_exp_containment}
Let $i,j,k \geq 0$ be integers with $k \leq i+j$.
Then, for all integers $c \geq 0$, \[\Oracle{\iExpTime{i}}{\Oracle{\iExpTime{j}}{\SigmaP{c}}} \subseteq \Oracle{\iExpTime{(i+j)}}{\SigmaP{c}} \subseteq \BoundedOracle{\iExpTime{k}}{\Oracle{\iExpTime{(i+j-k)}}{\SigmaP{c}}}{1}.\]
\end{theorem}

\begin{corollary}[store=EEequivalence]
Let $i,i',j,j' \geq 0$, and $\ell$, be integers with $i + j = \ell = i' + j'$.
Then, for all integers $c \geq 0$,
\[\Oracle{\iExpTime{i}}{\Oracle{\iExpTime{j}}{\SigmaP{c}}} = \Oracle{\iExpTime{\ell}}{\SigmaP{c}} =  \BoundedOracle{\iExpTime{i'}}{\Oracle{\iExpTime{j'}}{\SigmaP{c}}}{1}.\]
\end{corollary}

We now look at the complexity classes
$\Oracle{\iNExpTime{i}}{\Oracle{\iExpTime{j}}{\SigmaP{c}}}$.
For $i \geq 0$ and $j = 0$, the equivalence $\Oracle{\iNExpTime{i}}{\Oracle{\iExpTime{j}}{\SigmaP{c}}} = \Oracle{\iNExpTime{i}}{\SigmaP{c}}$ holds, for all integers $c \geq 0$.
The \zcref*[typeset=name,nocap]{theo_nexp_exp_containment} below focuses on the cases when $i \geq 0$ and $j \geq 1$.

\begin{theorem}[store=NEcontainment]
\label{theo_nexp_exp_containment}
Let $i,k \geq 0$ and $j \geq 1$ be integers with $k \leq i+j$.
Then, for all integers $c \geq 0$,
\[\Oracle{\iNExpTime{i}}{\Oracle{\iExpTime{j}}{\SigmaP{c}}} \subseteq \Oracle{\iExpTime{(i+j)}}{\SigmaP{c}} \subseteq \BoundedOracle{\iNExpTime{k}}{\Oracle{\iExpTime{(i+j-k)}}{\SigmaP{c}}}{1}.\]
\end{theorem}

\begin{corollary}[store=NEequivalence]
\label{theo_nexp_exp_equivalence}
Let $i,i' \geq 0$, $j,j' \geq 1$, and $\ell$, be integers with $i + j = \ell = i' + j'$.
Then, for all integers $c \geq 0$,
\[\Oracle{\iNExpTime{i}}{\Oracle{\iExpTime{j}}{\SigmaP{c}}} = \Oracle{\iExpTime{\ell}}{\SigmaP{c}} =  \BoundedOracle{\iNExpTime{i'}}{\Oracle{\iExpTime{j'}}{\SigmaP{c}}}{1}.\]
\end{corollary}

We conclude this section by looking at the complexity classes $\Oracle{\iExpSpace{i}}{\Oracle{\iExpTime{j}}{\SigmaP{c}}}$.

\begin{theorem}[store=SEcontainment]
\label{theo_expspace_exp_containment}
Let $i, j, k \geq 0$ be integers with $k \leq i + j$.
Then, for all integers $c \geq 0$,
\[\Oracle{\iExpSpace{(i-1)}}{\Oracle{\iExpTime{j}}{\SigmaP{c}}} \subseteq \Oracle{\iExpTime{(i+j)}}{\SigmaP{c}} \subseteq \BoundedOracle{\iExpSpace{(k-1)}}{\Oracle{\iExpTime{(i+j-k)}}{\SigmaP{c}}}{1}.\]
\end{theorem}

\begin{corollary}[store=SEequivalence]
Let $i,i',j, j' \geq 0$, and $\ell$, be integers with $i + j = \ell = i' + j'$.
Then, for all integers $c \geq 0$,
\[\Oracle{\iExpSpace{(i-1)}}{\Oracle{\iExpTime{j}}{\SigmaP{c}}} = \Oracle{\iExpTime{\ell}}{\SigmaP{c}} =  \BoundedOracle{\iExpSpace{(i'-1)}}{\Oracle{\iExpTime{j'}}{\SigmaP{c}}}{1}.\]
\end{corollary}

\subsection{Nondeterministic Exponential-Time Oracles}
\label{sec_charting_nexp_oracles}

In this section, we analyze the classes $\Oracle{\iExpTime{i}}{\Oracle{\iNExpTime{j}}{\SigmaP{c-1}}}$, $\Oracle{\iNExpTime{i}}{\Oracle{\iNExpTime{j}}{\SigmaP{c-1}}}$, and $\Oracle{\iExpSpace{i}}{\Oracle{\iNExpTime{j}}{\SigmaP{c-1}}}$, and some of their bounded-query variants.
These classes will be shown to be related to the intermediate levels of the iterated exponential hierarchies, and we will prove that these levels can uniformly 
be defined via Hausdorff classes.

It was shown that $\LogOracle{\PTime}{\NPTime} = \BoundedHausdCLASS{\PolFunctions}{\NPTime}$ and $\Oracle{\PTime}{\NPTime} = \BoundedHausdCLASS{2^\PolFunctions}{\NPTime}$~\cite{Buss1988,Wagner1990}.
Nonetheless, these equivalences have not previously been generalized to higher\nbdash-order exponential hierarchies.
Arguably, this generalization has been out of reach, because the results in the literature rely on proofs highly tailored for the \PHText, and hence hardly generalizable (see \zcref{footnote_wagner_hausdorff_reductions_limited_to_polynomial_case}, and \zcref{sec_intro_intermediate_levels_are_Hausdorff_classes} for additional details).
In contrast, our results are obtained via the notion of ``\ounawarelegal computations'', which enables an analysis based on machine computations, yielding proofs that readily generalize to higher\nbdash-order exponential hierarchies.

Intuitively, an \emph{\ounawarelegal computation} $\pi$ for an oracle machine $\Oracle{\Machine{M}}{?}$ on input $w$ is a sequence of IDs for $\Machine{M}$ that is legal \Wrt the transition function of $\Machine{M}$, but that does \emph{not} necessarily report the correct oracle answers to the queries by $\Machine{M}$ appearing in $\pi$;
i.e., $\pi$ can be seen as an oracle-\emph{agnostic} plausible computation for $\Oracle{\Machine{M}}{?}(w)$.

Embryonic ideas of this concept appeared in a work by \citet{LadnerL76}, where they defined the \emph{ID graph for an input string $w$ and an oracle machine $\Oracle{\Machine{M}}{?}$}.
Via this notion they aimed at analyzing all the queries that $\Oracle{\Machine{M}}{?}$, when executing on~$w$, may ever ask to an \emph{arbitrary} oracle.
Intuitively, an ID graph is an extended computation tree representing all the possible computations of $\Oracle{\Machine{M}}{?}$ over $w$, for all the possible oracle answers.

Our notion of \ounawarelegal (partial) computation for $\Oracle{\Machine{M}}{?}(w)$ is essentially a computation path within \citeauthor{LadnerL76}'s ID graph for $w$ and $\Oracle{\Machine{M}}{?}$.
Nonetheless, despite they were just a step away from singling out the concept of \ounawarelegal computation, \citeauthor{LadnerL76} missed the opportunity to define, and subsequently pivot on, this notion.
Indeed, to support their arguments, they then introduced the \emph{query graph for $w$, $\Oracle{\Machine{M}}{?}$, and an (oracle) language $\Language{A}$}.
Roughly, their query graph is a ``summary'' of their ID graph, in which only the start ID, the answers IDs, 
the accepting IDs, and the links between them, are retained from the original ID graph.\footnote{Despite their name, ``query graphs'' do \emph{not} explicitly report the queries generated by the oracle machine. These graphs' name stems from their vertices being IDs (\emph{excluding} the content of the query tape) associated with ``important'' events for the queries.}
However, the key aspect of query graphs is that these are defined \Wrt to a \emph{specific} oracle $\Language{A}$.
For this reason, paths associated with computations of $\Oracle{\Machine{M}}{?}$ on $w$ \emph{in}\/compatible with oracle answers for $\Language{A}$ are not in the query graph for $\Oracle{\Machine{M}}{?}$, $w$, and $\Language{A}$. 
We have to say though that \citeauthor{LadnerL76} adopt a different approach in one of their algorithms.
Indeed, in that algorithm, they resort to a complete ID graph and explore its key IDs
as the ID graph were a query graph for an unknown oracle language.
However, they assume in input to this algorithm also the oracle answers for a \emph{specific} oracle language.
Altogether, it does not seem that \citeauthor{LadnerL76} raise the notion of \ounawarelegal computation to the status of a primary object, or tool, of investigation.

A later paper by \citet{Kadin1989} introduced the \emph{query tree of an oracle machine $\Oracle{\Machine{M}}{?}$ on an input string $w$}.
This notion is tightly linked to \citeauthor{LadnerL76}'s query graphs, which remain however unmentioned in \citeauthor{Kadin1989}'s work.
A \citeauthor{Kadin1989}'s query tree connects all the successive queries, and not IDs, that $\Oracle{\Machine{M}}{?}$ might pose to its \emph{unspecified} oracle, based on every possible sequence of oracle answers.
In a query tree, a \emph{valid path} \Wrt an oracle $\Language{A}$ is the path associated with the sequence of queries issued to the oracle when the correct answers for $\Language{A}$ are considered.
Our \emph{\oawarelegal computation} notion, that is an oracle-machine computation in which query answers are correct \Wrt a specified oracle, relates to \citeauthor{Kadin1989}'s valid path notion.
Nonetheless, we stress that our \ounawarelegalemph computation notion is more general than \citeauthor{Kadin1989}'s on a couple of aspects.
First, our \ounawarelegal computations are sequences of IDs, and not just sequence of queries.
By this, we can more easily reason and operate on \ounawarelegal computations, as all the relevant information is represented in the sequence of IDs.
Second, our notion also applies to the case of (multiple rounds) of nonadaptive queries, which makes it more versatile.

A paper by \citet{Hemaspaandra1994} is another work where ideas similar to our \obothunawarelegal computation notions are adopted.
In an algorithm proposed in~\cite{Hemaspaandra1994}, \citeauthor{Hemaspaandra1994} employs the technique of first guessing an oracle-machine computation, together with witnesses for the oracle answers, and then checking the correctness of the guessed computation.
In our terminology, this approach consists in first guessing an \ounawarelegal computation and then checking whether the guessed computation is actually \oawarelegal.
In the same paper~\cite{Hemaspaandra1994}, a technique more similar to \citeauthor{Kadin1989}'s query tree is used in a different algorithm.
Indeed, in this second algorithm \citeauthor{Hemaspaandra1994} focuses on pairs of IDs individuating portions of computations delimited by the events of when an answer is received and of when the next query is submitted.

Although these authors proposed notions not too far from ours \ounawarelegal computation, they did not significantly pivot on this concept.
The extent to which we employ the notion of \ounawarelegal computation is considerably broader than what these authors did.
In fact, we resort to this concept in various contexts, and not just to refer to guessed computations that must later be verified to contain the correct oracle answers.
For example, we will define Hausdorff predicates revolving around checking the existence of \ounawarelegal computations meeting different criteria, before selecting the \oawarelegal computation of interest.

We also emphasize that, unlike what \citet{LadnerL76} and \citet{Hemaspaandra1994} did, our definition of IDs for oracle machines \emph{includes the query tape}.
This distinction enhances even more the effectiveness of our \ounawarelegal computation notion in simplifying proof arguments.
For example, if the query tape were not explicitly represented in the IDs, referencing queries in a proof, or accessing them in an algorithm, would require more complex reasoning or additional algorithmic steps.
Indeed, without the query tape in the IDs, queries would need to be pieced together indirectly from other elements, making the process more complicated.

Let us now more formally define the above intuitively introduced concepts.

\begin{definition}[store=OracleUnawareDef]
\label{def_oracle-unaware-legal-comp}
For an oracle machine $\Oracle{\Machine{M}}{?}$, a language $\Language{A}$, and a string $w$, a sequence $\pi$ of IDs is:
\begin{itemize}[nosep,label=--,left=0pt]
  \item
    an \oawarelegal (partial) computation for $\Oracle{\Machine{M}}{\Language{A}}$ over $w$, or for $\Oracle{\Machine{M}}{\Language{A}}(w)$, if $\pi$ is a (partial) computation in the computation tree of $\Oracle{\Machine{M}}{\Language{A}}$ over~$w$; and
  \item
    an \ounawarelegalemph (partial) computation for $\Oracle{\Machine{M}}{?}$ over $w$, or for $\Oracle{\Machine{M}}{?}(w)$, if there exists an oracle $\Language{B}$ such that $\pi$ is a (partial) computation in the computation tree of $\Oracle{\Machine{M}}{\Language{B}}$ over~$w$.
\end{itemize}
\end{definition}

Notice that, by its definition, $\pi$ is an \ounawarelegalemph computation for $\Oracle{\Machine{M}}{?}$ over $w$ iff $\pi$ is ``oracle\nbdash-consistent'', i.e., for every pair of oracle queries appearing in $\pi$, whenever the two queries are the same, then they receive the same answer in $\pi$, and, for every two IDs $\alpha$ and $\beta$ in $\pi$, with $\beta$ immediately following $\alpha$, it holds that:
\begin{itemize}[nosep,label=--,left=0pt]
  \item if $\alpha$'s control state is \emph{not} the query state, then $\beta$ is a legal $\alpha$'s next ID according to $\Machine{M}$'s transition function; and
  \item  if $\alpha$'s control state \emph{is} the query state, then $\beta$ is just a meaningful $\alpha$'s next ID, i.e., $\beta$ is either a \yeslbl-successor or \nolbl-successor of~$\alpha$ (the ID $\beta$ is only required to make sense as an $\alpha$'s next ID when the oracle machine receives the answer; the answer itself is not important).
\end{itemize}
These definitions are extended in the natural way to the scenario of parallel queries.

\subsubsection%
[\texorpdfstring{${i}$}{i}E\texorpdfstring{{\smaller XP}}{XP} Oracle Machines with N\texorpdfstring{${j}$}{j}E\texorpdfstring{{\smaller XP}}{XP} Oracles]%
{\texorpdfstring{$\boldsymbol{i}$}{i}E\texorpdfstring{{\smaller XP}}{XP} Oracle Machines with N\texorpdfstring{$\boldsymbol{j}$}{j}E\texorpdfstring{{\smaller XP}}{XP} Oracles}
\label{sec_iEXP_jNEXP}

We start our analysis with the classes $\Oracle{\iExpTime{i}}{\Oracle{\iNExpTime{j}}{\SigmaP{c-1}}}$.
The main result of this section is the next \zcref*[typeset=name,nocap]{theo_exp_nexp_containment}, which links the intermediate levels of the iterated exponential hierarchies to the Hausdorff classes defined over the main hierarchy levels.
We obtain this by showing that, given functions $r(n)$ and $s(n)$, if a language $\Language{L}$ belongs to $\DoubleBoundedParOracle{\iExpTime{i}}{\Oracle{\iNExpTime{j}}{\SigmaP{c-1}}}{r(n)}{s(n)}$, then there is a $\mathremovespaces{{(r(n) + 1)}^{s(n)}}$\nbdash-long Hausdorff reduction from $\Language{L}$ to an $\Oracle{\iNExpTime{(i+j)}}{\SigmaP{c-1}}$ Hausdorff predicate, and (almost) vice-versa---we say ``almost'' because an extra query to the oracle is needed.

Below, $\DoubleBoundedPlusParOracle{\ComplexityClass{X}}{\ComplexityClass{Y}}{r(n)}{s(n)}$ is the class of languages decided by oracle machines $\ParOracle{\Machine{M}}{?} \in \ComplexityClass{X}$ issuing to an oracle in $\ComplexityClass{Y}$ at most $s(n)$ rounds of at most $r(n)$ parallel queries, but for the first round when an additional query is allowed.

\begin{theorem}[store=ENcontainment]
\label{theo_exp_nexp_containment}
Let $i,j,k \geq 0$ and $g,h \geq -1$ be integers, with $\max\set{g,h} \leq i$ and $\max\set{g,h} \leq k \leq i+j$.
Let $r(n) \in O(\iExpPolFunctions{g})$ and $s(n) \in O(\iExpPolFunctions{h})$ be strictly positive nondecreasing functions computable within \iExponential{\max \set{0,g}} time and \iExponential{\max \set{0,h}} time, respectively.
Then, for all integers $c \geq 1$,
\[\DoubleBoundedParOracle{\iExpTime{i}}{\Oracle{\iNExpTime{j}}{\SigmaP{c-1}}}{r(n)}{s(n)} \subseteq \BoundedHausdCLASS{{(r(n)+1)}^{s(n)}}{\Oracle{\iNExpTime{(i+j)}}{\SigmaP{c-1}}} \subseteq \DoubleBoundedPlusParOracle{\iExpTime{k}}{\Oracle{\iNExpTime{(i+j-k)}}{\SigmaP{c-1}}}{r(n)}{s(n)}.\]
\end{theorem}

\begin{proof}
Before proceeding with the proof, we summarize in the \zcref*[typeset=name,nocap]{subprop_sizes_rp1_s_rp1TOs} below the representation sizes of some values that will be referred to in the rest of this proof.

\setcounter{subproperty}{-1}

\begin{subproperty}
\label{subprop_sizes_rp1_s_rp1TOs}
The values $(r(n) + 1)$, $s(n)$, and ${(r(n) + 1)}^{s(n)}$, have binary representation sizes which are $O(\iExpPolFunctions{g-1})$, $O(\iExpPolFunctions{h-1})$, and $O(\iExpPolFunctions{\max \set{0,g-1,h}})$, respectively, and so at most $i$- and \iExponential{k} \Wrt~$\StringLength{w}$.
\end{subproperty}

\begin{subproof}
Remember that $r(n)$ is assumed to be $O(\iExpPolFunctions{g})$, hence also $(r(n) + 1)$ is $O(\iExpPolFunctions{g})$, whereas $s(n)$ is assumed to be $O(\iExpPolFunctions{h})$.
Thus, the binary representation \emph{sizes} of $(r(n) + 1)$ and $s(n)$ are bounded by functions $x(n) \in O(\iExpPolFunctions{g-1})$ and $y(n) \in O(\iExpPolFunctions{h-1})$, respectively.
The representation size of ${(r(n) + 1)}^{s(n)}$ is $O(2^{y(n)}\cdot x(n))$ (see \zcref{sec_maths_complexity}, Exponentiation), where $2^{y(n)}$ is $O(\iExpPolFunctions{h})$ (see \zcref{theo_composition_iterated_expentials}.5).
Hence, the representation size of ${(r(n) + 1)}^{s(n)}$, which is $O(2^{y(n)}\cdot x(n))$, is $O(\iExpPolFunctions{\max \set{0,g-1,h}})$ (see \zcref{theo_multiplication_between_iterated_exponentials}.3).
To conclude, since we are assuming $i \geq 0$, $k \geq 0$, $\max \set{g,h} \leq i$, and $\max \set{g,h} \leq k$, we have that all these sizes are at most \iExponential{i} and \iExponential{k} \Wrt~$\StringLength{w}$.
\end{subproof}

\Proofsep

We first prove that $\DoubleBoundedParOracle{\iExpTime{i}}{\Oracle{\iNExpTime{j}}{\SigmaP{c-1}}}{r(n)}{s(n)} \subseteq \BoundedHausdCLASS{{(r(n)+1)}^{s(n)}}{\Oracle{\iNExpTime{(i+j)}}{\SigmaP{c-1}}}$.

Let $\Language{L} \in \DoubleBoundedParOracle{\iExpTime{i}}{\Oracle{\iNExpTime{j}}{\SigmaP{c-1}}}{r(n)}{s(n)}$ be a language.
We prove $\Language{L} \in \BoundedHausdCLASS{{(r(n)+1)}^{s(n)}}{\Oracle{\iNExpTime{(i+j)}}{\SigmaP{c-1}}}$ by showing that $\Language{L}$ can be characterized by an $\Oracle{\iNExpTime{(i+j)}}{\SigmaP{c-1}}$ Hausdorff predicate~$\Language{D}$ of length ${(r(n)+1)}^{s(n)}$.

Since $\Language{L} \in \DoubleBoundedParOracle{\iExpTime{i}}{\Oracle{\iNExpTime{j}}{\SigmaP{c-1}}}{r(n)}{s(n)}$, there are an $\iExpTime{i}$ oracle machine $\DoubleBoundedParOracle{\Machine{M}}{?}{r(n)}{s(n)}$ and an $\Oracle{\iNExpTime{j}}{\SigmaP{c-1}}$ oracle $\Omega$ such that $\Language{L} = \LanguageOf{\DoubleBoundedParOracle{\Machine{M}}{\Omega}{r(n)}{s(n)}}$.
Assume \Wlog that $\Machine{M}$ issues exactly $s(n)$ rounds of $r(n)$ parallel queries.
Indeed, if $M$ needed to ask strictly fewer than $r(n)$ questions in a round, and/or needed to perform strictly fewer than $s(n)$ rounds of queries, then there would be a machine equivalent to $\Machine{M}$ that can perform dummy rounds of queries and issue dummy queries to the oracle, just to precisely hit the values $r(n)$ and $s(n)$.

For an integer $a \geq 1$, let $V_{s,r}^a$ be the space of vectors of $s(a)$\nbdash-many integers with values between $0$ and $r(a)$, i.e., $V_{s,r}^a = \set{0,1,\dots,r(a)}^{s(a)}$.
We impose over $V_{s,r}^a$ the usual lexicographic order:
if $\Vec{m}, \Vec{n} \in V_{s,r}^a$ are two vectors, $\Vec{m}$ (lexicographically) dominates $\Vec{n}$, denoted $\Vec{m} \lexsucc \Vec{n}$, iff the left\nbdash-most position at which they differ has the element at that position in $\Vec{m}$ (strictly) greater than that in $\Vec{n}$;
relations $\lexsucceq$ and $\lexprec$ are defined in the natural way.
Let $\pi$ be an \ounawarelegal (partial) computation for an oracle machine issuing multiple rounds of parallel queries, $\AnsYESvectorInCompPar{\pi}$ is the vector of the \emph{numbers} of \yesansws appearing in $\pi$ for the respective rounds of queries.

Let us define the following three predicates.
These will be combined to define a predicate that serves as the basis for the Hausdorff predicate characterizing~$\Language{L}$.
Below, $w$ is a string, and $\Vec{m} \in V_{s,r}^{\StringLength{w}}$ is a vector of integers.

\begin{itemize}[noitemsep]
  \item $\PredHSucc[\Machine{M},\Omega]{A}(w,\Vec{m})$: $\valtrue$ iff
      there exists an \ounawarelegalemph computation $\pi$ for $\DoubleBoundedParOracle{\Machine{M}}{?}{r(n)}{s(n)}(w)$ such that $\AnsYESvectorInCompPar{\pi} \lexsucc \Vec{m}$ and, for all queries $q$ receiving in $\pi$ a \yesansw, $\Omega(q) = 1$;

  \item $\PredHAccCurr[\Machine{M},\Omega]{A}(w,\Vec{m})$: $\valtrue$ iff
      there exists an \emph{accepting} \ounawarelegalemph computation~$\pi$ for $\DoubleBoundedParOracle{\Machine{M}}{?}{r(n)}{s(n)}(w)$ such that $\AnsYESvectorInCompPar{\pi} = \Vec{m}$ and, for all queries $q$ receiving in $\pi$ a \yesansw, $\Omega(q) = 1$;
      and

  \item $\PredHRejCurr[\Machine{M},\Omega]{A}(w,\Vec{m})$: $\valtrue$ iff
      there exists a \emph{rejecting} \ounawarelegalemph computation $\pi$ for $\DoubleBoundedParOracle{\Machine{M}}{?}{r(n)}{s(n)}(w)$ such that $\AnsYESvectorInCompPar{\pi} = \Vec{m}$ and, for all queries $q$ receiving in $\pi$ a \yesansw, $\Omega(q) = 1$.
\end{itemize}

We show that the three above predicates can be decided in $\Oracle{\iNExpTime{(i+j)}}{\SigmaP{c-1}}$ \Wrt~$\StringLength{w}$ only.
First notice that, since $\Vec{m} \in V_{s,r}^{\StringLength{w}}$, the vector $\Vec{m}$ contains \iExponential{h}{}ly-many (\Wrt $\StringLength{w}$) components, whose \emph{values} are \iExponential{g}{}ly bounded (\Wrt $\StringLength{w}$).
Hence, each $\Vec{m}$'s component can be represented in \iExponential{(g-1)} space in~$\StringLength{w}$.
Because $\max \set{g,h} \leq i$, the representation size of $\Vec{m}$ is \iExponential{i}{}ly bounded (\Wrt $\StringLength{w}$).

Let us now focus on the complexity of deciding $\PredHSucc[\Machine{M},\Omega]{A}(w,\Vec{m})$.
Observe that, since $\Machine{M}$ is an \iExponential{i}-time (\Wrt $\StringLength{w}$) oracle machine, $\Machine{M}$ cannot issue more than \iExponential{i}{}ly-many (\Wrt $\StringLength{w}$) queries to its oracle.
To answer \yeslbl on $\PredHSucc[\Machine{M},\Omega]{A}(w,\Vec{m})$, we first guess an \iExponential{i}{}ly-long (\Wrt $\StringLength{w}$) \ounawarelegal computation $\pi$ for $\DoubleBoundedParOracle{\Machine{M}}{?}{r(n)}{s(n)}(w)$, together with the \iExponential{i}{}ly-many 
certificates witnessing $\Omega(q) = 1$ for all the queries $q$ receiving a \yesansw in $\pi$.
These certificates are accepting \oawarelegalemph computations for the $\iNExpTime{j}$ ``part'' of $\Omega$---remember that $\Omega \in \Oracle{\iNExpTime{j}}{\SigmaP{c-1}}$;
i.e., we leave out from the guess the part of computation associated with the $\SigmaP{c-1}$ oracle.
These accepting computations for $\Omega$ are \iExponential{(i+j)}{}ly-long, because $\Omega$ may receive \iExponential{i}{}-long queries from $\Machine{M}$.
Hence, these are \iExponential{i}{}ly-many \iExponential{(i+j)}{}ly-long certificates.
This guess can therefore be carried out in nondeterministic \iExponential{(i+j)} time.

After the guess, we check three conditions:
(i)~that $\pi$ is an \ounawarelegal computation for $\DoubleBoundedParOracle{\Machine{M}}{?}{r(n)}{s(n)}(w)$, which is feasible in \iExponential{i} time \Wrt $\StringLength{w}$;
(ii)~that $\AnsYESvectorInCompPar{\pi} \lexsucc \Vec{m}$, which is feasible in \iExponential{i} time \Wrt $\StringLength{w}$, because the representation sizes of $\AnsYESvectorInCompPar{\pi}$ and $\Vec{m}$ are \iExponential{i}{}ly bounded \Wrt $\StringLength{w}$ (see above); and
(iii) that the certificates for $\Omega(q) = 1$, for all queries $q$, are valid, which is feasible in \iExponential{(i+j)} time \Wrt $\StringLength{w}$ with the aid of a $\SigmaP{c-1}$ oracle.
Hence, the overall procedure is in $\Oracle{\iNExpTime{(i+j)}}{\SigmaP{c-1}}$, and the time bound is \Wrt $\StringLength{w}$ only.
Remember that the predicate $\PredHSucc[\Machine{M},\Omega]{A}(w,\Vec{m})$ is defined over the two arguments $w$ and $Vec{m}$, and the representation size of $\Vec{m}$ is \iExponential{i}{}ly bounded (\Wrt $\StringLength{w}$) (see above).
Therefore, a procedure running in \iExponential{(i+j)} time (\Wrt $\StringLength{w}$) only has enough time to entirely read~$\Vec{m}$ as well.

The intuition why this predicate is in $\Oracle{\iNExpTime{(i+j)}}{\SigmaP{c-1}}$ is that the guessed $\pi$ needs to be an \ounawarelegalemph computation and \emph{not} \oawarelegalemph.
Thus, apart from the queries positively answered in~$\pi$, for which we are required to provide witnesses for $\Omega$ accepting them, it is \emph{not} required that $\Omega$ actually rejects the queries negatively answered in~$\pi$.
This is why $\PredHSucc[\Machine{M},\Omega]{A}(w,\Vec{m})$ is in $\Oracle{\iNExpTime{(i+j)}}{\SigmaP{c-1}}$, as we do not need to carry out the $\ComplementPrefixKerned\Oracle{\iNExpTime{(i+j)}}{\SigmaP{c-1}}$ tasks of checking that the queries receiving \noansws in~$\pi$ are actually rejected by $\Omega$.

Similarly, to decide $\PredHAccCurr[\Machine{M},\Omega]{A}(w,\Vec{m})$ (resp., $\PredHRejCurr[\Machine{M},\Omega]{A}(w,\Vec{m})$), we first guess an \ounawarelegal computation $\pi$ for $\DoubleBoundedParOracle{\Machine{M}}{?}{r(n)}{s(n)}(w)$ and the certificates witnessing $\Omega(q) = 1$ for all the queries receiving a \yesansw in $\pi$.
This guess is carried out in nondeterministic \iExponential{(i+j)} time (see above).
After the guess, we check:
(i)~that $\pi$ is an accepting (resp., a rejecting) \ounawarelegal computation for $\DoubleBoundedParOracle{\Machine{M}}{?}{r(n)}{s(n)}(w)$;
(ii)~that $\AnsYESvectorInCompPar{\pi} = \Vec{m}$; and
(iii)~that the certificates are valid.
The overall check is feasible in \iExponential{(i+j)} time with the aid of a $\SigmaP{c-1}$ oracle (see above).
Hence, also these predicates are in $\Oracle{\iNExpTime{(i+j)}}{\SigmaP{c-1}}$, and the time bound is \Wrt $\StringLength{w}$ only.

\medbreak

We now show that, via the three predicates above, we can define a binary predicate $\Predicate{B}_{\Machine{M},\Omega}(w,z) \subseteq \HausdPredDomain$,
from which we will obtain a Hausdorff predicate $\Language{D}(w,z)$ of length ${(r(n)+1)}^{s(n)}$ characterizing~$\Language{L}$.
Similarly to Hausdorff predicates, $\Predicate{B}_{\Machine{M},\Omega}(w,z)$ will be such that $\Predicate{B}_{\Machine{M},\Omega}(w,z) \geq \Predicate{B}_{\Machine{M},\Omega}(w,z + 1)$, but in $\Predicate{B}_{\Machine{M},\Omega}(w,z)$ we will have $z \geq 0$, instead of $z \geq 1$ like in Hausdorff predicates (see \zcref{def_Hausdorff_predicate}).
We focus on $\Predicate{B}_{\Machine{M},\Omega}(w,z)$ just to streamline the presentation.
The Hausdorff predicate $\Language{D}(w,z)$ obtained from $\Predicate{B}_{\Machine{M},\Omega}(w,z)$ will then just be a simple ``shift'' of $\Predicate{B}_{\Machine{M},\Omega}$, i.e., we will have $\Language{D}(w,z) = \Predicate{B}_{\Machine{M},\Omega}(w,z-1)$.
For this reason, we will need to show that $w \in \Language{L}$ iff the highest value of $z$ at which $\Predicate{B}_{\Machine{M},\Omega}(w,z) = 1$ is \emph{even}.

To define $\Predicate{B}_{\Machine{M},\Omega}(w,z)$, we need to refer to the $z$\nbdash-th vector $V_{s,r}^{\StringLength{w}}[z]$ in the \emph{lexicographically ordered} space $V_{s,r}^{\StringLength{w}}$ (see above).
The first vector $V_{s,r}^{\StringLength{w}}[0]$ in $V_{s,r}^{\StringLength{w}}$ is the vector of all zeroes.
Notice that the vector $V_{s,r}^{\StringLength{w}}[z]$ is essentially the representation of the number $z$ in base $(r(\StringLength{w})+1)$ over $s(\StringLength{w})$ digits.
We claim (see below) that the vector $V_{s,r}^{\StringLength{w}}[z]$ can be computed from~$z$ in \iExponential{i} time \Wrt $\StringLength{w}$.

First, observe that ${(r(\StringLength{w}) + 1)}^{s(\StringLength{w})} - 1$ is the highest value of $z$ associated with a vector $V_{s,r}^{\StringLength{w}}[z]$ in $V_{s,r}^{\StringLength{w}}$.
Hence, by \zcref{subprop_sizes_rp1_s_rp1TOs}, the representation sizes of $z$ and of $r(\StringLength{w}) + 1$ are \iExponential{i} \Wrt~$\StringLength{w}$.
To obtain $V_{s,r}^{\StringLength{w}}[z]$ from~$z$, we repeatedly divide $z$, and the successive quotients, by $r(\StringLength{w}) + 1$ and take the remainders of the integer divisions.
Because $r(n)$ is assumed to be computable in \iExponential{\max \set{0,g}} time, and we assume $i \geq 0$ and $g \leq i$, the value $r(\StringLength{w}) + 1$ can be computed in \iExponential{i} time.
The divisions to obtain $V_{s,r}^{\StringLength{w}}[z]$ from~$z$ can be carried out in time that is polynomial in the operands' sizes (see \zcref{sec_maths_complexity}, Arithmetic).
Since the biggest operands have \iExponential{i} sizes (see above), each division can be evaluated in \iExponential{i} time.
We perform $s(\StringLength{w}) - 1$ divisions, which are \iExponential{h}{}ly-many, and hence also \iExponential{i}{}ly-many, as $h \leq i$ by assumption.
By this, $V_{s,r}^{\StringLength{w}}[z]$ can be obtained from~$z$ in \iExponential{i} time \Wrt $\StringLength{w}$---observe that, since the size of $z$ is \iExponential{i} in $\StringLength{w}$ (see above), a procedure running in \iExponential{i} time in $\StringLength{w}$ has enough time to entirely read $z$ from tape.

Via the previous three predicates, we define the predicate $\Predicate{B}_{\Machine{M},\Omega}(w,z)$ as follows:
\[
\Predicate{B}_{\Machine{M},\Omega}(w,z) =
\begin{cases}
  \lnot \PredHSucc[\Machine{M},\Omega]{A}(w,V_{s,r}^{\StringLength{w}}[z]) \to \PredHAccCurr[\Machine{M},\Omega]{A}(w,V_{s,r}^{\StringLength{w}}[z]),
  & {\text{if } z < {(r(\StringLength{w})+1)}^{s(\StringLength{w})} \text{ and } z \text{ is even}} \\
  \lnot \PredHSucc[\Machine{M},\Omega]{A}(w,V_{s,r}^{\StringLength{w}}[z]) \to \PredHRejCurr[\Machine{M},\Omega]{A}(w,V_{s,r}^{\StringLength{w}}[z]),
  & {\text{if } z < {(r(\StringLength{w})+1)}^{s(\StringLength{w})} \text{ and } z \text{ is odd}} \\
  \valfalse, & \text{if } z \geq {(r(\StringLength{w})+1)}^{s(\StringLength{w})}.
\end{cases}
\]
We claim that $\Predicate{B}_{\Machine{M},\Omega}$ is in $\Oracle{\iNExpTime{(i+j)}}{\SigmaP{c-1}}$ \Wrt $\StringLength{w}$ only, as the following procedure shows.
First, we compute $\mi{max} = {(r(\StringLength{w})+1)}^{s(\StringLength{w})}$, which can be obtained in time that is polynomial in the size of the result, which is \iExponential{i} \Wrt $\StringLength{w}$ (see \zcref{subprop_sizes_rp1_s_rp1TOs}, and Exponentiation in \zcref{sec_maths_complexity}).
Then, we compare $z$ with $\mi{max}$.
Also this comparison is feasible in \iExponential{i} time \Wrt $\StringLength{w}$, because, by \zcref{subprop_sizes_rp1_s_rp1TOs}, the size of ${(r(\StringLength{w})+1)}^{s(\StringLength{w})}$ is \iExponential{i} \Wrt $\StringLength{w}$, hence in \iExponential{i} time it is possible to read from tape the relevant part of $z$ and perform the comparison.
In fact, if $z$ is too big, there is no need to read it entirely, as we are assuming that the most significant bit of $z$ on tape is `$1$', hence, if the representation of $z$ is longer than the representation of $\mi{max}$, then $z$ is surely greater than $\mi{max}$.
After this, if $z$ is not less than $\mi{max}$, we answer $\valfalse$.
On the other hand, if $z$ is less than $\mi{max}$, then, by $\lnot a \to b \equiv a \lor b$, we answer $\PredHSucc[\Machine{M},\Omega]{A}(w,V_{s,r}^{\StringLength{w}}[z]) \lor \PredHAccCurr[\Machine{M},\Omega]{A}(w,V_{s,r}^{\StringLength{w}}[z])$ if $z$ is even, and we answer $\PredHSucc[\Machine{M},\Omega]{A}(w,V_{s,r}^{\StringLength{w}}[z]) \lor \PredHRejCurr[\Machine{M},\Omega]{A}(w,V_{s,r}^{\StringLength{w}}[z])$ if $z$ is odd.
In both cases we need to compute $V_{s,r}^{\StringLength{w}}[z]$ from $z$, which can be obtained in \iExponential{i} time \Wrt $\StringLength{w}$ (see above).
Lastly notice that, since the answer is the result of the disjunction of two $\Oracle{\iNExpTime{(i+j)}}{\SigmaP{c-1}}$ predicates and $\Oracle{\iNExpTime{(i+j)}}{\SigmaP{c-1}}$ is closed under disjunction, the answer can be obtained in $\Oracle{\iNExpTime{(i+j)}}{\SigmaP{c-1}}$.

\medbreak

We now prove that, for every string $w$, it holds that $\Predicate{B}_{\Machine{M},\Omega}(w,z) \geq \Predicate{B}_{\Machine{M},\Omega}(w,z+1)$, for all $z \geq 0$.
To this aim, we state some intermediate properties, and we start by introducing some notation.

Let $\pi$ be an \ounawarelegal (partial) computation for an oracle machine $\ParOracle{\Machine{M}}{?}$ issuing multiple rounds of parallel queries.
We define the following functions:
$\QueryInCompPar{\pi}{u}$ is the set of queries that, according to $\pi$, are issued in the $u$\nbdash-th round by $\ParOracle{\Machine{M}}{?}$;
and $\AnsInCompPar{\pi}{u}$ is the set of answers that, according to $\pi$, $\ParOracle{\Machine{M}}{?}$ receives from its oracle to its $u$\nbdash-th round of queries---we assume that we can read from $\AnsInCompPar{\pi}{u}$ the information associating an answer with the respective query.
Remember that $\AnsYESvectorInCompPar{\pi}$ is the vector of the \emph{numbers} of \yesansws appearing in $\pi$ for the respective rounds of queries;
$\QueryYESinCompPar{\pi}{u}$ is the \emph{number} of \yesansws appearing in $\pi$ for the $u$\nbdash-th round queries.

We denote by $\PortionCompToQuery{\pi}{u}$ and $\PortionCompToAns{\pi}{u}$ the initial portion of $\pi$ up to when the $u$\nbdash-th round of queries is issued and answered, respectively;
$\IDinCompQuery{\pi}{u}$ and $\IDinCompAns{\pi}{u}$ denote the ID in $\pi$ when the $u$\nbdash-th round of queries is issued and answered, respectively, and they are the last ID of $\PortionCompToQuery{\pi}{u}$ and $\PortionCompToAns{\pi}{u}$, respectively.

In the following, let $\hat \pi$ be \emph{the} \oawarelegalemph computation for $\DoubleBoundedParOracle{\Machine{M}}{\Omega}{r(n)}{s(n)}(w)$---remember that $\Machine{M}$ is deterministic.

The following \zcref*[typeset=name,nocap]{subprop_reconstruction_of_oracle_aware_from_counting} states that, if $\pi$ is an \ounawarelegalemph computation for $\DoubleBoundedParOracle{\Machine{M}}{?}{r(n)}{s(n)}(w)$ with $\QueryYESinCompPar{\pi}{1} = \QueryYESinCompPar{\hat \pi}{1}, \dots,\linebreak[0] \QueryYESinCompPar{\pi}{u} = \QueryYESinCompPar{\hat \pi}{u}$, and all queries $q$ receiving a \yesansw in $\pi$ are such that $\Omega(q) = 1$,
then $\pi$ and $\hat \pi$ coincide from the initial ID to when the $(u{+}1)$\nbdash-th round of queries is issued.
Observe here that we are imposing the equality on the \emph{numbers} of \yesansws per round, not on the specific \emph{queries} that are positively answered.
Nonetheless, we can prove that this is enough to reconstruct the \oawarelegalemph computation $\hat \pi$.

\begin{subproperty}
\label{subprop_reconstruction_of_oracle_aware_from_counting}
Let $\pi$ be an \ounawarelegal computation for $\DoubleBoundedParOracle{\Machine{M}}{?}{r(n)}{s(n)}(w)$.
For every $u$ with $0 \leq u \leq s(\StringLength{w})-1$ (resp., for $u = s(\StringLength{w})$),
if $\QueryYESinCompPar{\pi}{1} = \QueryYESinCompPar{\hat \pi}{1}, \dots, \QueryYESinCompPar{\pi}{u} = \QueryYESinCompPar{\hat \pi}{u}$, and all queries $q$ receiving a \yesansw in $\pi$ are such that $\Omega(q) = 1$,
then $\PortionCompToQuery{\pi}{u+1} = \PortionCompToQuery{\hat \pi}{u+1}$ (resp., then $\pi = \hat \pi$).
\end{subproperty}

\begin{subproof}
We first prove by induction that the \zcref*[typeset=name,nocap]{subprop_reconstruction_of_oracle_aware_from_counting} holds for all values of $u$ between $0$ and $s(\StringLength{w}) - 1$.
A conclusive remark will deal with the case in which $u = s(\StringLength{w})$.

\medskip

(Base case).
For $u = 0$, we show $\PortionCompToQuery{\pi}{1} = \PortionCompToQuery{\hat \pi}{1}$.
First, observe that $\Machine{M}$ is deterministic and does not issue any oracle call before the first-round queries.
By this, $\Machine{M}$'s computation up to that moment cannot depend on oracle answers, but depends only on the input string $w$ and on $\Machine{M}$'s transition function.
Hence, every \ounawarelegalemphboth computation $\pi$ for $\DoubleBoundedParOracle{\Machine{M}}{?}{r(n)}{s(n)}(w)$ must be such that $\PortionCompToQuery{\pi}{1} = \PortionCompToQuery{\hat \pi}{1}$.

\medskip

(Inductive hypothesis).
Assume the \zcref*[typeset=name,nocap]{subprop_reconstruction_of_oracle_aware_from_counting} holds for all values of $u$ up to $v$, for some $v \leq s(\StringLength{w})-2$.

\medskip

(Inductive step).
We now prove that the \zcref*[typeset=name,nocap]{subprop_reconstruction_of_oracle_aware_from_counting} holds also for $v + 1$.
Let $\pi$ be an \ounawarelegal computation for $\DoubleBoundedParOracle{\Machine{M}}{?}{r(n)}{s(n)}(w)$ such that $\QueryYESinCompPar{\pi}{1} = \QueryYESinCompPar{\hat \pi}{1}, \dots, \QueryYESinCompPar{\pi}{v} = \QueryYESinCompPar{\hat \pi}{v},\linebreak[0] \QueryYESinCompPar{\pi}{v+1} = \QueryYESinCompPar{\hat \pi}{v+1}$, and all queries $q$ receiving a \yesansw in $\pi$ are such that $\Omega(q) = 1$.
We show that $\PortionCompToQuery{\pi}{v+2} = \PortionCompToQuery{\hat \pi}{v+2}$.

By inductive hypothesis, since $\QueryYESinCompPar{\pi}{1} = \QueryYESinCompPar{\hat \pi}{1}, \dots, \QueryYESinCompPar{\pi}{v} = \QueryYESinCompPar{\hat \pi}{v},$ it holds that $\PortionCompToQuery{\pi}{v+1} = \PortionCompToQuery{\hat \pi}{v+1}$, which means that $\QueryInCompPar{\pi}{v+1} = \QueryInCompPar{\hat \pi}{v+1}$.
Let us consider $\PortionCompToAns{\pi}{v+1}$.
The computation portions $\PortionCompToQuery{\pi}{v+1}$ and $\PortionCompToAns{\pi}{v+1}$ differ only for the additional last ID in $\PortionCompToAns{\pi}{v+1}$, which is $\IDinCompAns{\pi}{v+1}$.
The latter is the ID following $\IDinCompQuery{\pi}{v+1}$, and contains the answers to the queries in $\IDinCompQuery{\pi}{v+1}$, which we know to be $\QueryInCompPar{\pi}{v+1}$.
Now notice that it must be the case that $\AnsInCompPar{\pi}{v+1} = \AnsInCompPar{\hat \pi}{v+1}$ for the following reasons:
(i)~$\QueryInCompPar{\pi}{v+1} = \QueryInCompPar{\hat \pi}{v+1}$;
(ii)~$\QueryYESinCompPar{\pi}{v+1} = \QueryYESinCompPar{\hat \pi}{v+1}$ (by assumption on $\pi$); and
(iii)~all queries $q \in \QueryYESinCompPar{\pi}{v+1}$ are such that $\Omega(q) = 1$.
This hence proves that $\PortionCompToAns{\pi}{v+1} = \PortionCompToAns{\hat \pi}{v+1}$.

Consider now the $(v{+}2)$\nbdash-th round queries $\QueryInCompPar{\hat \pi}{v+2}$.
These queries deterministically depend, as $\Machine{M}$ is deterministic, only on the input string $w$ and on the answers $\AnsInCompPar{\hat \pi}{1}, \dots, \AnsInCompPar{\hat \pi}{v+1}$ that $\Machine{M}$ receives from $\Omega$.
Since $\pi$ is an \ounawarelegalemphlegal computation for $\DoubleBoundedParOracle{\Machine{M}}{?}{r(n)}{s(n)}(w)$ and $\AnsInCompPar{\hat \pi}{1}, \dots, \AnsInCompPar{\hat \pi}{v+1}$ are in $\PortionCompToAns{\pi}{v+1}$ (see above), the portion of $\pi$ between $\IDinCompAns{\pi}{v+1}$ and $\IDinCompQuery{\pi}{v+2}$ must be a (legal) partial computation for $\DoubleBoundedParOracle{\Machine{M}}{\Omega}{r(n)}{s(n)}(w)$.
This shows that $\PortionCompToQuery{\pi}{v+2} = \PortionCompToQuery{\hat \pi}{v+2}$, and closes the inductive argument.

\medskip

The induction proves the \zcref*[typeset=name,nocap]{subprop_reconstruction_of_oracle_aware_from_counting} for $u$ up to $s(\StringLength{w})-1$.
Consider now $u = s(\StringLength{w})$.
Let $\pi$ be an \ounawarelegal computation for $\DoubleBoundedParOracle{\Machine{M}}{?}{r(n)}{s(n)}(w)$ such that $\QueryYESinCompPar{\pi}{1} = \QueryYESinCompPar{\hat \pi}{1}, \dots,\linebreak[0] \QueryYESinCompPar{\pi}{s(\StringLength{w}-1)} = \QueryYESinCompPar{\hat \pi}{s(\StringLength{w}-1)},\linebreak[0] \QueryYESinCompPar{\pi}{s(\StringLength{w})} = \QueryYESinCompPar{\hat \pi}{s(\StringLength{w})}$, and all queries $q$ receiving a \yesansw in $\pi$ are such that $\Omega(q) = 1$.

By the inductive argument above, from $\QueryYESinCompPar{\pi}{1} = \QueryYESinCompPar{\hat \pi}{1}, \dots, \QueryYESinCompPar{\pi}{s(\StringLength{w}-1)} = \QueryYESinCompPar{\hat \pi}{s(\StringLength{w}-1)}$ follows $\PortionCompToQuery{\pi}{s(\StringLength{w})} = \PortionCompToQuery{\hat \pi}{s(\StringLength{w})}$.
Since $\QueryYESinCompPar{\pi}{s(\StringLength{w})} = \QueryYESinCompPar{\hat \pi}{s(\StringLength{w})}$ by assumption, an argument similar to the one in the inductive step above shows that $\AnsInCompPar{\pi}{s(\StringLength{w})} = \AnsInCompPar{\hat \pi}{s(\StringLength{w})}$.
Hence, $\PortionCompToAns{\pi}{s(\StringLength{w})} = \PortionCompToAns{\hat \pi}{s(\StringLength{w})}$.
Notice that $\pi$ is an \ounawarelegalemphlegal computation and in its portion $\PortionCompToAns{\pi}{s(\StringLength{w})}$ contains precisely all the actual answers $\AnsInCompPar{\hat \pi}{1}, \dots, \AnsInCompPar{\hat \pi}{s(\StringLength{w})}$ to the actual queries $\QueryInCompPar{\hat \pi}{1}, \dots, \QueryInCompPar{\hat \pi}{s(\StringLength{w})}$ submitted to $\Omega$ by $\Machine{M}$ when executing on $w$.
Therefore, the remaining part of $\pi$ following $\PortionCompToAns{\pi}{s(\StringLength{w})}$ must be a (legal) partial computation for $\DoubleBoundedParOracle{\Machine{M}}{\Omega}{r(n)}{s(n)}(w)$.
Thus, $\pi$ is an \oawarelegal computation for $\DoubleBoundedParOracle{\Machine{M}}{\Omega}{r(n)}{s(n)}(w)$.
\end{subproof}

\begin{subproperty}
\label{subprop_no_unaware_comp_lex_greater_hat_pi}
There is no \ounawarelegal computation $\pi$ for $\DoubleBoundedParOracle{\Machine{M}}{?}{r(n)}{s(n)}(w)$ with $\AnsYESvectorInCompPar{\pi} \lexsucc \AnsYESvectorInCompPar{\hat \pi}$ and such that, for all queries $q$ receiving in $\pi$ a \yesansw, $\Omega(q) = 1$.
\end{subproperty}

\begin{subproof}
Assume by contradiction that there is an \ounawarelegal computation $\pi$ for $\DoubleBoundedParOracle{\Machine{M}}{?}{r(n)}{s(n)}(w)$ with $\AnsYESvectorInCompPar{\pi} \lexsucc \linebreak[0] \AnsYESvectorInCompPar{\hat \pi}$ and such that, for all queries $q$ receiving a \yesansw in $\pi$, $\Omega(q) = 1$.
Let $u$ be the lowest index at which $\QueryYESinCompPar{\pi}{u} \neq \QueryYESinCompPar{\hat \pi}{u}$---%
we are here assuming \Wlog that positions of vector elements are indexed left to right with increasing values.
By \zcref{subprop_reconstruction_of_oracle_aware_from_counting}, $\PortionCompToQuery{\pi}{u} = \PortionCompToQuery{\hat \pi}{u}$, implying that $\QueryInCompPar{\pi}{u} = \QueryInCompPar{\hat \pi}{u}$.
From $\AnsYESvectorInCompPar{\pi} \lexsucc \AnsYESvectorInCompPar{\hat \pi}$ follows $\QueryYESinCompPar{\pi}{u} > \QueryYESinCompPar{\hat \pi}{u}$.
Since by assumption all queries receiving a \yesansw in $\pi$ are accepted by $\Omega$, all the queries $s \in \QueryInCompPar{\pi}{u}$ receiving a \yesansw in $\pi$ are such that $\Omega(s) = 1$.
However, by $\QueryInCompPar{\hat \pi}{u} = \QueryInCompPar{\pi}{u}$ and $\QueryYESinCompPar{\hat \pi}{u} < \QueryYESinCompPar{\pi}{u}$, there are queries in $\QueryInCompPar{\hat \pi}{u}$ receiving a \noansw in $\hat \pi$ which are instead accepted by $\Omega$:
a contradiction, because $\hat \pi$ is an \oawarelegalemph computation.
\end{subproof}

\begin{subproperty}
\label{subprop_A_lex_greater_true_iff_preceeds_hat_pi}
For every vector $\Vec{m} \in V_{s,r}^{\StringLength{w}}$ of integers, $\PredHSucc[\Machine{M},\Omega]{A}(w,\Vec{m}) = 1 \Leftrightarrow (\Vec{m} \lexprec \AnsYESvectorInCompPar{\hat \pi})$.
\end{subproperty}

\begin{subproof}
\ProofLeftarrowItem
Let $\Vec{m} \in V_{s,r}^{\StringLength{w}}$ be a vector such that $\Vec{m} \lexprec \AnsYESvectorInCompPar{\hat \pi}$.
By definition of $\PredHSucc[\Machine{M},\Omega]{A}(w,\Vec{m})$, the computation $\hat \pi$ is actually a witness for $\PredHSucc[\Machine{M},\Omega]{A}(w,\Vec{m}) = 1$.

\ProofRightarrowItem
Let $\Vec{m} \in V_{s,r}^{\StringLength{w}}$ be a vector such that $\Vec{m} \lexsucceq \AnsYESvectorInCompPar{\hat \pi}$.
In order for $\PredHSucc[\Machine{M},\Omega]{A}(w,\Vec{m}) = 1$ to hold true, there must exist an \ounawarelegal computation $\tilde \pi$ for $\DoubleBoundedParOracle{\Machine{M}}{?}{r(n)}{s(n)}(w)$ with $\AnsYESvectorInCompPar{\tilde \pi} \lexsucc \Vec{m}$ such that, for all queries $q$ receiving a \yesansw in $\tilde \pi$, $\Omega(q) = 1$.
For such a computation $\tilde \pi$, we would have $\AnsYESvectorInCompPar{\tilde \pi} \lexsucc \Vec{m} \lexsucceq \AnsYESvectorInCompPar{\hat \pi}$.
However, this contradicts \zcref{subprop_no_unaware_comp_lex_greater_hat_pi}.
\end{subproof}

\noindent
From \zcref{subprop_no_unaware_comp_lex_greater_hat_pi}, and by the definitions of $\PredHAccCurr[\Machine{M},\Omega]{A}(w,\Vec{m})$ and $\PredHRejCurr[\Machine{M},\Omega]{A}(w,\Vec{m})$, the following is easily proven.

\begin{subproperty}
\label{subprop_m_greater_than_hat_pi_THEN_A_accept_reject_false}
For every vector $\Vec{m} \in V_{s,r}^{\StringLength{w}}$ of integers, if $\Vec{m} \lexsucc \AnsYESvectorInCompPar{\hat \pi}$, then $\PredHAccCurr[\Machine{M},\Omega]{A}(w,\Vec{m}) = \PredHRejCurr[\Machine{M},\Omega]{A}(w,\Vec{m}) = 0$.
\end{subproperty}

\noindent
From \zcref{subprop_A_lex_greater_true_iff_preceeds_hat_pi}, we have $\Predicate{B}_{\Machine{M},\Omega}(w,z) = 1$, for all $z \geq 0$ such that $V_{s,r}^{\StringLength{w}}[z] \lexprec \AnsYESvectorInCompPar{\hat \pi}$.
Moreover, by \zcref{subprop_A_lex_greater_true_iff_preceeds_hat_pi,subprop_m_greater_than_hat_pi_THEN_A_accept_reject_false}, we have $\Predicate{B}_{\Machine{M},\Omega}(w,z) = 0$, for all $z \geq 0$ such that $V_{s,r}^{\StringLength{w}}[z] \lexsucc \AnsYESvectorInCompPar{\hat \pi}$.
Therefore, if we denote by $\hat z$ the index such that $V_{s,r}^{\StringLength{w}}[\hat z] = \AnsYESvectorInCompPar{\hat \pi}$, irrespective of whether $\Predicate{B}_{\Machine{M},\Omega}(w, \hat z)$ is actually true or false, it holds that $\Predicate{B}_{\Machine{M},\Omega}(w,z) \geq \Predicate{B}_{\Machine{M},\Omega}(w,z+1)$, for all $z \geq 0$.

\medbreak

We now show that $w \in \Language{L}$ iff the highest value of $z$ at which $\Predicate{B}_{\Machine{M},\Omega}(w,z) = 1$ is \emph{even}---remember that $\Language{L}$ will be characterized by a Hausdorff predicate $\Language{D}(w,z) = \Predicate{B}_{\Machine{M},\Omega}(w,z-1)$.

Let us focus on the truth value of $\Predicate{B}_{\Machine{M},\Omega}(w,\hat z)$,
and consider the truth values of 
$\PredHAccCurr[\Machine{M},\Omega]{A}(w,V_{s,r}^{\StringLength{w}}[\hat z])$ and $\PredHRejCurr[\Machine{M},\Omega]{A}(w,V_{s,r}^{\StringLength{w}}[\hat z])$---remember that $V_{s,r}^{\StringLength{w}}[\hat z]$ is the vector of the \emph{numbers} of queries receiving a \yesansw in $\hat \pi$, for each round of queries;
$V_{s,r}^{\StringLength{w}}[\hat z]$ does \emph{not} contain information on \emph{which} queries receive a \yesansw in~$\hat \pi$.

\begin{subproperty}
\label{subprop_A_accept_reject_correct_at_hat_z}
$\PredHAccCurr[\Machine{M},\Omega]{A}(w,V_{s,r}^{\StringLength{w}}[\hat z]) = \DoubleBoundedParOracle{\Machine{M}}{\Omega}{r(n)}{s(n)}(w)$ and $\PredHRejCurr[\Machine{M},\Omega]{A}(w,V_{s,r}^{\StringLength{w}}[\hat z]) = 1 - \DoubleBoundedParOracle{\Machine{M}}{\Omega}{r(n)}{s(n)}(w)$.
\end{subproperty}

\begin{subproof}
We show $\PredHAccCurr[\Machine{M},\Omega]{A}(w,V_{s,r}^{\StringLength{w}}[\hat z]) = \DoubleBoundedParOracle{\Machine{M}}{\Omega}{r(n)}{s(n)}(w)$.
{Proving $\PredHRejCurr[\Machine{M},\Omega]{A}(w,V_{s,r}^{\StringLength{w}}[\hat z]) = 1 - \DoubleBoundedParOracle{\Machine{M}}{\Omega}{r(n)}{s(n)}(w)$ is similar.}

\ProofRightarrowItem
If $\PredHAccCurr[\Machine{M},\Omega]{A}(w, V_{s,r}^{\StringLength{w}}[\hat z]) = 1$, then there exists an \emph{accepting} \ounawarelegal computation $\pi$ for $\DoubleBoundedParOracle{\Machine{M}}{?}{r(n)}{s(n)}(w)$ with $\AnsYESvectorInCompPar{\pi} = V_{s,r}^{\StringLength{w}}[\hat z] = \AnsYESvectorInCompPar{\hat \pi}$ and such that, for all queries $q$ receiving a \yesansw in $\pi$, $\Omega(q) = 1$.
By \zcref{subprop_reconstruction_of_oracle_aware_from_counting}, such a computation $\pi$ is actually an \oawarelegal computation for $\DoubleBoundedParOracle{\Machine{M}}{\Omega}{r(n)}{s(n)}(w)$.
Since $\pi$ is accepting, it holds that $\DoubleBoundedParOracle{\Machine{M}}{\Omega}{r(n)}{s(n)}(w) = 1$.

\ProofLeftarrowItem
If $\DoubleBoundedParOracle{\Machine{M}}{\Omega}{r(n)}{s(n)}(w) = 1$, then the \oawarelegalemph computation $\hat \pi$ for $\DoubleBoundedParOracle{\Machine{M}}{\Omega}{r(n)}{s(n)}(w)$ is \emph{accepting}.
Hence, $\hat \pi$ is actually a witness for $\PredHAccCurr[\Machine{M},\Omega]{A}(w, V_{s,r}^{\StringLength{w}}[\hat z]) = 1$, as $\AnsYESvectorInCompPar{\hat \pi} = V_{s,r}^{\StringLength{w}}[\hat z]$.
\end{subproof}

\noindent
By \zcref{subprop_A_lex_greater_true_iff_preceeds_hat_pi}, $\PredHSucc[\Machine{M},\Omega]{A}(w,V_{s,r}^{\StringLength{w}}[\hat z]) = 0$, and hence, if $\hat z$ is even, it holds $\Predicate{B}_{\Machine{M},\Omega}(w,\hat z) = \PredHAccCurr[\Machine{M},\Omega]{A}(w,V_{s,r}^{\StringLength{w}}[\hat z])$, whereas, if $\hat z$ is odd, it holds $\Predicate{B}_{\Machine{M},\Omega}(w,\hat z) = \PredHRejCurr[\Machine{M},\Omega]{A}(w,V_{s,r}^{\StringLength{w}}[\hat z])$.
By \zcref{subprop_A_accept_reject_correct_at_hat_z}, and by considering all the possible combinations of cases for $\hat z$ odd/even and for $\PredHAccCurr[\Machine{M},\Omega]{A}(w,V_{s,r}^{\StringLength{w}}[\hat z])$ and $\PredHRejCurr[\Machine{M},\Omega]{A}(w,V_{s,r}^{\StringLength{w}}[\hat z])$ $\valtrue$/$\valfalse$, it is not hard to check that $\DoubleBoundedParOracle{\Machine{M}}{\Omega}{r(n)}{s(n)}(w) = 1$ iff the highest value of $z$ at which $\Predicate{B}_{\Machine{M},\Omega}(w,z) = 1$ is \emph{even}.

\medbreak

To conclude, consider the binary predicate $\Language{D}(w,z) = \Predicate{B}_{\Machine{M},\Omega}(w,z-1)$---remember that indices for $\Predicate{B}_{\Machine{M},\Omega}(w,z)$ start at~$0$.
By the properties of $\Predicate{B}_{\Machine{M},\Omega}$ shown above, $\Language{D}$ is a Hausdorff predicate characterizing $\Language{L}$.
Indeed, we have $\Language{D}(w,z) \geq \Language{D}(w,z+1)$ for all strings $w$ and integers $z \geq 1$, i.e., the requirement (\HausdSeqRequirementI) is met.
Moreover, by the definition of $\Predicate{B}_{\Machine{M},\Omega}$, for every string $w$, it holds that $\Language{D}(w,{(r(\StringLength{w})+1)}^{s(\StringLength{w})} + 1) = \Predicate{B}_{\Machine{M},\Omega}(w,{(r(\StringLength{w})+1)}^{s(\StringLength{w})}) = 0$.
The latter implies that both the requirement (\HausdSeqRequirementII) is met and that the length of the Hausdorff predicate is ${(r(n)+1)}^{s(n)}$.
Lastly, for every string $w$, it holds that $w \in \Language{L}$ iff $\HausdIndex{w}{\Language{D}}$ is odd.

\Proofsep

We now prove $\BoundedHausdCLASS{{(r(n)+1)}^{s(n)}}{\Oracle{\iNExpTime{(i+j)}}{\SigmaP{c-1}}} \subseteq \DoubleBoundedPlusParOracle{\iExpTime{k}}{\Oracle{\iNExpTime{(i+j-k)}}{\SigmaP{c-1}}}{r(n)}{s(n)}$.

Let $\Language{L}$ be a language characterized by an $\Oracle{\iNExpTime{(i+j)}}{\SigmaP{c-1}}$ Hausdorff predicate $\Language{D}$ of length ${(r(n) + 1)}^{s(n)}$.
We show that $\Language{L} \in \DoubleBoundedPlusParOracle{\iExpTime{k}}{\Oracle{\iNExpTime{(i+j-k)}}{\SigmaP{c-1}}}{r(n)}{s(n)}$ by exhibiting a $\iExpTime{k}$ oracle machine $\DoubleBoundedPlusParOracle{\Machine{M}}{?}{r(n)}{s(n)}$ and an oracle $\Omega \in \Oracle{\iNExpTime{(i+j-k)}}{\SigmaP{c-1}}$ such that $\Language{L} = \LanguageOf{\DoubleBoundedPlusParOracle{\Machine{M}}{\Omega}{r(n)}{s(n)}}$.

Consider first the oracle $\Omega$.
The machine $\Omega$ is designed to receive from $\Machine{M}$ pairs $\pair{\wt{w},z}$, where $\wt{w}$ is a padded version of the string~$w$ in input to $\Machine{M}$ to reach \iExponential{k} length (\Wrt~$\StringLength{w}$).
By this, since $\Omega \in \Oracle{\iNExpTime{(i+j-k)}}{\SigmaP{c-1}}$, the machine $\Omega$ can run for \iExponential{(i+j)} time \Wrt~$\StringLength{w}$.
Upon reception of the query $\pair{\wt{w},z}$, $\Omega$ simply ignores the padding of $\wt{w}$ and decides whether $\pair{w,z} \in \Language{D}$.
This can be done by $\Omega$, because $\Language{D} \in \Oracle{\iNExpTime{(i+j)}}{\SigmaP{c-1}}$.
For this reason, below we will regard $\Omega$ as an oracle for $\Language{D}$.

We now focus on the algorithm carried out by $\Machine{M}$.
Since $\Language{D}$ is such that $w \in \Language{L}$ iff the Hausdorff index $\HausdIndex{w}{\Language{D}}$ of $w$ \Wrt $\Language{D}$ is odd, $\Machine{M}$ can decide $\Language{L}$ via a procedure akin to a ``generalized'' binary search with the aid of the oracle $\Omega$ for $\Language{D}$.
We first give an intuition on the algorithm (see \zcref{fig_generalized_binary_search} for an exemplification), then we provide details and show that it can actually be executed in \iExponential{k} time by~$\Machine{M}$.

\begin{figure}[t]
  \centering
  \includegraphics[width=.98\textwidth]{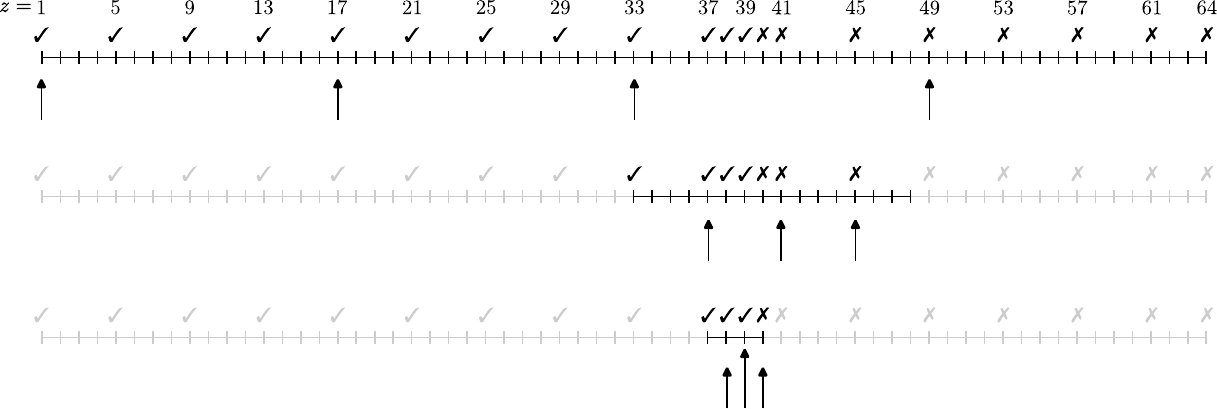}
  \caption{An example of the ``generalized'' binary search used in the proof of \zcref{theo_exp_nexp_containment}, where $r(n) = 3$, and $s(n) = 3$.
  The interval has $(3+1)^3 = 64$ points.
  Symbols `\cmark' and `\xmark' mean that the Hausdorff predicate $\Language{D}(w,z)$ at those positions is $\valtrue$ and $\valfalse$, respectively (to avoid a too cumbersome drawing not all symbols are reported); $39$ is the Hausdorff index $\HausdIndex{w}{\Language{D}}$ of $w$ \Wrt $\Language{D}$.
  The three levels of the drawing relate, from top to bottom, with the three rounds of parallel queries. The arrows individuate the ``samples'' taken at each round---notice the extra query in the first round. The greyed-out areas in the bottom two levels are the interval portions recognized as not relevant. In the bottom level, the longer arrow is the sample individuating~$\HausdIndex{w}{\Language{D}}$.}\label{fig_generalized_binary_search}
\end{figure}

The machine $\Machine{M}$ starts by ``sampling'', via parallel queries to its oracle, the truth value of $\Language{D}(w,\cdot)$ at $(r(\StringLength{w}) + 1)$\nbdash-many uniformly-spaced points within the interval $[1,{(r(\StringLength{w}) + 1)}^{s(\StringLength{w})}]$.
The first sample is at index $1$.
If $\Language{D}(w,1) = 0$, then $\Machine{M}$ can answer \nolbl, as $\HausdIndex{w}{\Language{D}} = 0$ is even.
If $\Language{D}(w,1) = 1$, by knowing the indices of the samples at which the value of $\Language{D}(w,\cdot)$ turns from $\valtrue$ to $\valfalse$, $\Machine{M}$ can delimit a smaller portion of the interval containing $\HausdIndex{w}{\Language{D}}$.
Observe that this smaller portion's length is $\frac{1}{r(\StringLength{w})+1}$ of the entire interval's length.

Within this smaller portion, $\Machine{M}$ can sample $r(\StringLength{w})$\nbdash-many uniformly-spaced points and ask to the oracle the truth value of $\Language{D}(w,\cdot)$ at the sampled points.
This allows $\Machine{M}$ to narrow down even more the portion of interest of the interval.
By doing this $s(\StringLength{w}) - 1$ times (including the first round), $\Machine{M}$ obtains a residual portion of the interval encompassing $\HausdIndex{w}{\Language{D}}$ and containing only $r(\StringLength{w})$\nbdash-many points that have not been sampled yet.
In the last round of queries, $\Machine{M}$ samples the remaining points in the last portion considered and individuates~$\HausdIndex{w}{\Language{D}}$.

\begin{algorithm}[!t]
\caption{The procedure discussed in the proof of \zcref{theo_exp_nexp_containment} carried out by the oracle machine $\DoubleBoundedPlusParOracle{\Machine{M}}{?}{r(n)}{s(n)}$. The machine's oracle is assumed to decide a Hausdorff predicate. The algorithm computes the Hausdorff index of the input string $w$ \Wrt the Hausdorff predicate decided by the oracle, and returns ``$\mathsf{accept}$'' if and only if the computed Hausdorff index is odd. By this, $\DoubleBoundedPlusParOracle{\Machine{M}}{?}{r(n)}{s(n)}$ decides the language $\Language{L}$ characterized by the Hausdorff predicate of length $\mathremovespaces{{(r(n)+1)}^{s(n)}}$ decided by its oracle.}
\label{alg_simulation_by_machine_from_Hausdorff_reduction}

\BlankLine

\KwAssumption{The oracle queried by $\DoubleBoundedPlusParOracle{\Machine{M}}{?}{r(n)}{s(n)}$ decides a Hausdorff predicate of length $\mathremovespaces{{(r(n)+1)}^{s(n)}}$.}

\BlankLine

\KwIn{A string $w$.}
\KwOut{$\mathsf{accept}$, if $w$'s Hausdorff index \Wrt the Hausdorff predicate decided by the oracle is odd;\newline $\mathsf{reject}$, otherwise.}

\BlankLine

\tcc*[h]{At \zcref{alg_generalized_binary_search_input_padding}, the value $k$ is a constant directly encoded in the algorithm depending on the time-bound of the machine meant to run the algorithm}

\BlankLine

\nonl \Proc{{$\DoubleBoundedPlusParOracle{\Machine{M}}{?}{r(n)}{s(n)}\!(w)$}}{

    $\mi{rw} \leftarrow r(\StringLength{w})$\;
    $\mi{sw} \leftarrow s(\StringLength{w})$\;
    $\wt{w} \leftarrow w$ padded to reach $\iExp{k}{\StringLength{w}}$ symbols \nllabel{alg_generalized_binary_search_input_padding} \;
    $\mi{sample}[0] \leftarrow 1$\;
    $\mi{step} \leftarrow {(\mi{rw} + 1)}^{\mi{sw} - 1}$ \nllabel{alg_generalized_binary_search_first_computation_of_step} \;
    
    \lFor{$v \leftarrow 1$ \KwTo $\mi{rw}$}{$\mi{sample}[v] \leftarrow \mi{sample}[v-1] + \mi{step}$}
    
    $\mi{ans}[0, \dots, \mi{rw}] \leftarrow$ \ParCalls ${\tup{\wt{w},\mi{sample}[0]},\tup{\wt{w},\mi{sample}[1]},\dots,\tup{\wt{w},\mi{sample}[\mi{rw}-1]},\tup{\wt{w},\mi{sample}[\mi{rw}]}}$ \nllabel{alg_generalized_binary_search_first_parallel_calls} \;
    
    \lIf(\tcc*[f]{The Hausdorff index of $w$ is $0$, i.e., even}){$\mi{ans}[0] = 0$}{\Return{$\mathsf{reject}$}}
    \lElse{$t\mhyphen{}max \leftarrow$ the maximum index $t$ at which $\mi{ans}[t] = 1$}
    
    \For{$u \leftarrow 2$ \KwTo $\mi{sw}$}{
        $\mi{sample}[0] \leftarrow \mi{sample}[t\mhyphen{}max]$\;
        $\mi{step} \leftarrow {(\mi{rw} + 1)}^{\mi{sw} - u}$ \nllabel{alg_generalized_binary_search_further_computation_of_step} \;
        \lFor{$v \leftarrow 1$ \KwTo $\mi{rw}$}{$\mi{sample}[v] \leftarrow \mi{sample}[v-1] + \mi{step}$}
        $\mi{ans}[1, \dots, \mi{rw}] \leftarrow$ \ParCalls ${\tup{\wt{w},\mi{sample}[1]},\tup{\wt{w},\mi{sample}[2]},\dots,\tup{\wt{w},\mi{sample}[\mi{rw}-1]},\tup{\wt{w},\mi{sample}[\mi{rw}]}}$ \nllabel{alg_generalized_binary_search_further_parallel_calls} \;
        
        \lIf{$\mi{ans}[1] = 0$}{$t\mhyphen{}max = 0$}
        \lElse{$t\mhyphen{}max \leftarrow$ the maximum index $t$ at which $\mi{ans}[t] = 1$}
    }
    
    \tcc*[h]{At this point, $\mi{sample}[t\mhyphen{}max]$ contains the Hausdorff index of $w$}
    
    \lIf(\tcc*[f]{The Hausdorff index of $w$ is odd}){$\mi{sample}[t\mhyphen{}max]$ is odd}{\Return{$\mathsf{accept}$}}
    \lElse(\tcc*[f]{The Hausdorff index of $w$ is even}){\Return{$\mathsf{reject}$}}
}
\end{algorithm}

\zcref[S]{alg_simulation_by_machine_from_Hausdorff_reduction} provides a more precise definition of the procedure carried out by $\Machine{M}$.
Below, we will show that \zcref{alg_simulation_by_machine_from_Hausdorff_reduction} is actually computable in \iExponential{k} time by $\DoubleBoundedPlusParOracle{\Machine{M}}{?}{r(n)}{s(n)}$.

By the definition of \zcref{alg_simulation_by_machine_from_Hausdorff_reduction}, when executing on~$w$, the machine $\Machine{M}$ performs $s(\StringLength{w})$ rounds of $r(\StringLength{w})$ parallel queries, but for the first round in which $r(\StringLength{w})+1$ queries are issued to the oracle.
Therefore, the constraints on the number of queries are met by $\Machine{M}$.
Notice that, the number of queries at each round is $O(\iExpPolFunctions{g})$, and the number of rounds is $O(\iExpPolFunctions{h})$.
Hence, the total number of queries issued by $\Machine{M}$, which is bounded by $(r(n) + 1) \cdot s(n)$, is $O(\iExpPolFunctions{\max \set{0,g,h}})$ (see \zcref{theo_multiplication_between_iterated_exponentials}.3).
Since $k \geq 0$ and $\max\set{g,h} \leq k$ by assumption, the \emph{number} of required queries can be issued by $\Machine{M}$ in \iExponential{k} time.
Thus, \zcref{alg_simulation_by_machine_from_Hausdorff_reduction} runs in \iExponential{k} time as long as the rest of its steps are feasible in \iExponential{k} time \Wrt~$\StringLength{w}$.
The algorithm performs two kinds of operations:
string manipulations and arithmetic operations.

The first kind of operations includes steps such as writing the queries on the query tape and checking the answers from the oracle (to find the maximum index at which a sample has received a \yesansw).
The complexity of these operations is basically related to the sizes of the strings to write and read.
In a round of queries (\zcref{alg_generalized_binary_search_first_parallel_calls,alg_generalized_binary_search_further_parallel_calls} of \zcref{alg_simulation_by_machine_from_Hausdorff_reduction}),
each query contains the string~$\wt{w}$ and the index of the sample.
Remember that~$\wt{w}$ is of \iExponential{k} size.
Also the size of the sample index is \iExponential{k}, because the biggest index is ${(r(\StringLength{w}) + 1)}^{s(\StringLength{w})}-1$, whose size, by \zcref{subprop_sizes_rp1_s_rp1TOs}, is 
\iExponential{k}.
Therefore, the size of a single query is bounded by a function $a(n) \in O(\iExpPolFunctions{k})$.
Since the algorithm writes in each round at most $r(n) + 1$ queries, which is $O(\iExpPolFunctions{g})$, the overall size of all the parallel queries submitted in a single round is bounded by $(r(n) + 1) \cdot a(n)$.
The latter is $O(\iExpPolFunctions{\max \set{0,g,k}})$ by \zcref{theo_multiplication_between_iterated_exponentials}.3, and so $O(\iExpPolFunctions{k})$ by $k \geq 0$ and $g \leq k$.
Thus, writing all the queries of a round is feasible in \iExponential{k} time.
Also reading from tape the oracle answers can be carried out in \iExponential{k} time, because $r(n) + 1$ answers are \iExponential{g}{}ly-many, and $g \leq k$.

We now show that the arithmetic operations to obtain the sample indices are feasible in \iExponential{k} time.

We have observed that the \emph{size} of a sample index is $O(\iExpPolFunctions{k})$.
Thus, simple arithmetic operations over them, such as assignments and sums, are feasible in time that is polynomial in the size of these values (see \zcref{sec_maths_complexity}, Arithmetic), and hence in \iExponential{k} time.
We are left to show that the most involved arithmetic operations, namely the exponentiations at \zcref{alg_generalized_binary_search_first_computation_of_step,alg_generalized_binary_search_further_computation_of_step} of \zcref{alg_simulation_by_machine_from_Hausdorff_reduction}, 
can also be carried out in \iExponential{k} time.

To this aim, consider the value ${(r(\StringLength{w}) + 1)}^{s(\StringLength{w}) - u}$.
By assumption, $r(\StringLength{w})$, and hence $r(\StringLength{w}) + 1$, can be computed in \iExponential{\max \set{0,g}} time, whereas $s(\StringLength{w})$, and hence $s(\StringLength{w}) - u$, for $1 \leq u \leq \mi{sw}$, can be computed in \iExponential{\max \set{0,h}} time.
By $k \geq 0$ and $\max\set{g,h} \leq k$, these can be computed in \iExponential{k} time.
Once $r(\StringLength{w}) + 1$ and $s(\StringLength{w}) - u$ have been computed, ${(r(\StringLength{w}) + 1)}^{s(\StringLength{w}) - u}$ can be computed in time that is polynomial in the result's size (see \zcref{sec_maths_complexity}, Exponentiation), which is \iExponential{k} by \zcref{subprop_sizes_rp1_s_rp1TOs}.
Thus, computing ${(r(\StringLength{w}) + 1)}^{s(\StringLength{w}) - u}$ is feasible in 
\iExponential{k} time \Wrt~$\StringLength{w}$.
\end{proof}

From the previous \zcref*[typeset=name,nocap]{theo_exp_nexp_containment}, we state the following \zcref*[typeset=name,nocap]{theo_general_chain_parallel_hausdorff,theo_general_chain_adaptive_hausdorff,theo_summary_generalized_equivalence_intermediate_levels_Hausdorff} individuating the Hausdorff characterizations of the intermediate levels of the iterated exponential hierarchies.
The proofs are deferred to \zcref{sec_detailed_proofs_charting_nexp_oracles}.

\begin{corollary}[store=ENContainmentParallelCalls]
\label{theo_general_chain_parallel_hausdorff}
Let $i,i',j,j' \geq 0$, $g \geq -1$, and $\ell$, be integers, with $i + j = \ell = i' + j'$ and $g \leq i, i'$.
Let $f(n) \in O(\iExpPolFunctions{g})$ be a function computable in \iExponential{\max \set{0,g}} time.
Then, for all integers $c \geq 1$,
\[
\BoundedHausdCLASS{f(n)}{\Oracle{\iNExpTime{\ell}}{\SigmaP{c-1}}} \subseteq
\BoundedParOracle{\iExpTime{i}}{\Oracle{\iNExpTime{j}}{\SigmaP{c-1}}}{f(n)} \subseteq
\BoundedHausdCLASS{f(n) + 1}{\Oracle{\iNExpTime{\ell}}{\SigmaP{c-1}}} \subseteq
\BoundedParOracle{\iExpTime{i'}}{\Oracle{\iNExpTime{j'}}{\SigmaP{c-1}}}{f(n) + 1}.
\]
\end{corollary}

\begin{corollary}[store=ENContainmentAdaptiveCalls]
\label{theo_general_chain_adaptive_hausdorff}
Let $i,i',j,j' \geq 0$, $g \geq -1$, and $\ell$, be integers, with $i + j = \ell = i' + j'$ and $g \leq i, i'$.
Let $f(n) \in O(\iExpPolFunctions{g})$ be a function computable in \iExponential{\max \set{0,g}} time.
Then, for all integers $c \geq 1$,
\[
\BoundedHausdCLASS{2^{f(n)-1}}{\Oracle{\iNExpTime{\ell}}{\SigmaP{c-1}}} \subseteq
\BoundedOracle{\iExpTime{i}}{\Oracle{\iNExpTime{j}}{\SigmaP{c-1}}}{f(n)} \subseteq
\BoundedHausdCLASS{2^{f(n)}}{\Oracle{\iNExpTime{\ell}}{\SigmaP{c-1}}} \subseteq
\BoundedOracle{\iExpTime{i'}}{\Oracle{\iNExpTime{j'}}{\SigmaP{c-1}}}{f(n)+1}.
\]
\end{corollary}

The following result is immediately proven from the two \zcref*[typeset=name,nocap]{theo_general_chain_parallel_hausdorff,theo_general_chain_adaptive_hausdorff} above and \zcref{theo_restricted_main_levels_longer_Hausdorff_as_expressive}.

\begin{corollary}
\label{theo_summary_generalized_equivalence_intermediate_levels_Hausdorff}
Let $i,i',j,j' \geq 0$, $g \geq -1$, and $\ell$, be integers, with $i + j = \ell = i' + j'$ and $g \leq i, i'$.
Then, for all integers $c \geq 1$,
\begin{gather*}
\SwapAboveDisplaySkip
  \BoundedParOracle{\iExpTime{i}}{\Oracle{\iNExpTime{j}}{\SigmaP{c-1}}}{\iExpPolFunctions{g}} = \BoundedHausdCLASS{\iExpPolFunctions{g}}{\Oracle{\iNExpTime{\ell}}{\SigmaP{c-1}}} = \BoundedParOracle{\iExpTime{i'}}{\Oracle{\iNExpTime{j'}}{\SigmaP{c-1}}}{\iExpPolFunctions{g}} \\
  \BoundedOracle{\iExpTime{i}}{\Oracle{\iNExpTime{j}}{\SigmaP{c-1}}}{\iExpPolFunctions{g}} = \BoundedHausdCLASS{\iExpPolFunctions{g+1}}{\Oracle{\iNExpTime{\ell}}{\SigmaP{c-1}}} = \BoundedOracle{\iExpTime{i'}}{\Oracle{\iNExpTime{j'}}{\SigmaP{c-1}}}{\iExpPolFunctions{g}} \\
  \BoundedOracle{\iExpTime{i}}{\Oracle{\iNExpTime{j}}{\SigmaP{c-1}}}{\iExpPolFunctions{g}} = \BoundedParOracle{\iExpTime{i'}}{\Oracle{\iNExpTime{j'}}{\SigmaP{c-1}}}{\iExpPolFunctions{g+1}} \rlap{\quad\text{(here, $g \leq i$ and $g \leq i' - 1$)}.}
\end{gather*}
\end{corollary}

We now show that deterministic oracle machines issuing a \emph{single} round of (at least polynomially-many) parallel queries are as powerful as deterministic oracle machines issuing a \emph{constant} number of \emph{multiple} rounds of (at least polynomially-many) parallel queries.
An intuition why at least polynomially-many queries are needed for the property to hold is provided after the corollary statement.
The proof is in \zcref{sec_detailed_proofs_charting_nexp_oracles}.
The analogue result restricted to $\PolHier$, i.e., $\ParBoundedOracle{\PTime}{\NPTime}{k} = \ParOracle{\PTime}{\NPTime}$ for every fixed integer $k \geq 1$, was proven by \citet{Buss1991}.

\begin{corollary}[store=ConstantRoundsParallelEqualsSingleRound]
\label{theo_general_constant_rounds_parallel_calls_equals_single_round}
Let $i,i',j,j',g \geq 0$, and $k \geq 1$, be integers, with $i + j = i' + j'$ and $g \leq i, i'$.
Then, for all integers $c \geq 1$,
\[\DoubleBoundedParOracle{\iExpTime{i}}{\Oracle{\iNExpTime{j}}{\SigmaP{c-1}}}{\iExpPolFunctions{g}}{k} = \BoundedParOracle{\iExpTime{i'}}{\Oracle{\iNExpTime{j'}}{\SigmaP{c-1}}}{\iExpPolFunctions{g}}.\]
\end{corollary}

Intuitively, in the statement of the previous \zcref*[typeset=name,nocap]{theo_general_constant_rounds_parallel_calls_equals_single_round}, the integer $g$ of $\iExpPolFunctions{g}$ is required to be at least $0$ for this reason.
Let $\Language{L} \in \DoubleBoundedParOracle{\iExpTime{i}}{\Oracle{\iNExpTime{j}}{\SigmaP{c-1}}}{\iExpPolFunctions{g}}{k}$ be a language.
By this, $\Language{L} \in \DoubleBoundedParOracle{\iExpTime{i}}{\Oracle{\iNExpTime{j}}{\SigmaP{c-1}}}{\iExp{g}{p(n)}}{k}$, for some polynomial $p(n)$, and hence
by \zcref{theo_exp_nexp_containment}, $\Language{L} \in \BoundedHausdCLASS{{(\iExp{g}{p(n)}+1)}^{k}}{\Oracle{\iNExpTime{(i+j)}}{\SigmaP{c-1}}}$.
Now observe that, by \zcref{theo_polynomial_of_iterated_exponentials}, when $g \geq 0$ the function $q(n) = {(\iExp{g}{p(n)}+1)}^{k}$ is still in $\iExpPolFunctions{g}$, that is, $q(n)$ does not grow asymptotically faster then $p(n)$.
On the contrary, if $g = -1$, and hence the bound on the queries is indeed $\log p(n)$, we have that $q(n) = {(\log p(n) + 1)^{k}}$ is actually polylogarithmic, which grows asymptotically faster than $\log p(n)$.

From the previous four \zcref*[typeset=name,nocap]{theo_general_chain_parallel_hausdorff,theo_general_chain_adaptive_hausdorff,theo_summary_generalized_equivalence_intermediate_levels_Hausdorff,theo_general_constant_rounds_parallel_calls_equals_single_round}, the following four are immediately proven by considering $j = j' = 0$.
By this, the statements are more closely related to the intermediate levels of the iterated exponential hierarchies.

\begin{corollary}
\label{theo_specific_chain_parallel_hausdorff}
Let $i \geq 0$ and $g \geq -1$ be integers, with $g \leq i$, and let $f(n) \in O(\iExpPolFunctions{g})$ be a function computable in \iExponential{\max \set{0,g}} time.
Then, for all integers $c \geq 1$,
\[
\BoundedHausdCLASS{f(n)}{\Oracle{\iNExpTime{i}}{\SigmaP{c-1}}} \subseteq
\BoundedParOracle{\iExpTime{i}}{\SigmaP{c}}{f(n)} \subseteq
\BoundedHausdCLASS{f(n) + 1}{\Oracle{\iNExpTime{i}}{\SigmaP{c-1}}} \subseteq
\BoundedParOracle{\iExpTime{i}}{\SigmaP{c}}{f(n) + 1}.
\]
\end{corollary}

\begin{corollary}
Let $i \geq 0$ and $g \geq -1$ be integers, with $g \leq i$, and let $f(n) \in O(\iExpPolFunctions{g})$ be a function computable in \iExponential{\max \set{0,g}} time.
Then, for all integers $c \geq 1$,
\[\BoundedHausdCLASS{2^{f(n)-1}}{\Oracle{\iNExpTime{i}}{\SigmaP{c-1}}} \subseteq
\BoundedOracle{\iExpTime{i}}{\SigmaP{c}}{f(n)} \subseteq
\BoundedHausdCLASS{2^{f(n)}}{\Oracle{\iNExpTime{i}}{\SigmaP{c-1}}} \subseteq
\BoundedOracle{\iExpTime{i}}{\SigmaP{c}}{f(n) + 1}.\]
\end{corollary}

\begin{corollary}[store=SummaryEqIntermLevelsHausdorff]
\label{theo_summary_equivalence_intermediate_levels_Hausdorff}
Let $i \geq 0$ and $g \geq -1$ be integers, with $g \leq i$.
Then, for all integers $c \geq 1$,
\begin{align*}
  \BoundedParOracle{\iExpTime{i}}{\SigmaP{c}}{\iExpPolFunctions{g}} &= \BoundedHausdCLASS{\iExpPolFunctions{g}}{\Oracle{\iNExpTime{i}}{\SigmaP{c-1}}} \\
  \BoundedOracle{\iExpTime{i}}{\SigmaP{c}}{\iExpPolFunctions{g}} &= \BoundedHausdCLASS{\iExpPolFunctions{g+1}}{\Oracle{\iNExpTime{i}}{\SigmaP{c-1}}} \\
  \BoundedOracle{\iExpTime{i}}{\SigmaP{c}}{\iExpPolFunctions{g}} &= \BoundedParOracle{\iExpTime{i}}{\SigmaP{c}}{\iExpPolFunctions{g+1}} \rlap{\qquad\text{(here, $g \leq i - 1$).}}
\end{align*}
\end{corollary}

\begin{corollary}
Let $i, g \geq 0$, and $k \geq 1$, be integers, with $g \leq i$.
Then, for all integers $c \geq 1$,
\[\DoubleBoundedParOracle{\iExpTime{i}}{\SigmaP{c}}{\iExpPolFunctions{g}}{k} = \BoundedParOracle{\iExpTime{i}}{\SigmaP{c}}{\iExpPolFunctions{g}}.\]
\end{corollary}

By \zcref{theo_specific_chain_parallel_hausdorff}, we obtain the following relationship between successive levels of the \BHText over $\Oracle{\iNExpTime{i}}{\SigmaP{c-1}}$.
The analogue result restricted to the \BHText over \NPTime is in \cite{KoblerSW87,Wagner1990,Beigel1991}.
Their proof of the restricted result relies on the mind changes technique.
Our general result stems from \zcref{theo_exp_nexp_containment}, whose proof is entirely based on machine computation analysis.
The proof is in \zcref{sec_detailed_proofs_charting_nexp_oracles}.

\begin{theorem}[store=ChainBooleanHierarchy]
Let $i \geq 0$ and $k \geq 1$ be integers.
Then, for all integers $c \geq 1$,
\[(\BoundedHausdCLASS{k}{\Oracle{\iNExpTime{i}}{\SigmaP{c-1}}} \cup \ComplementPrefixKerned\BoundedHausdCLASS{k}{\Oracle{\iNExpTime{i}}{\SigmaP{c-1}}}) \subseteq \BoundedParOracle{\iExpTime{i}}{\SigmaP{c}}{k} \subseteq (\BoundedHausdCLASS{k+1}{\Oracle{\iNExpTime{i}}{\SigmaP{c-1}}} \cap \ComplementPrefixKerned\BoundedHausdCLASS{k+1}{\Oracle{\iNExpTime{i}}{\SigmaP{c-1}}}).\]
\end{theorem}

\subsubsection%
[N\texorpdfstring{${i}$}{i}E\texorpdfstring{{\smaller XP}}{XP} Oracle Machines with N\texorpdfstring{${j}$}{j}E\texorpdfstring{{\smaller XP}}{XP} Oracles]%
{N\texorpdfstring{$\boldsymbol{i}$}{i}E\texorpdfstring{{\smaller XP}}{XP} Oracle Machines with N\texorpdfstring{$\boldsymbol{j}$}{j}E\texorpdfstring{{\smaller XP}}{XP} Oracles}
\label{sec_details_NEXP_NEXP}

We now focus on the classes $\Oracle{\iNExpTime{i}}{\Oracle{\iNExpTime{j}}{\SigmaP{c-1}}}$.
For $i \geq 0$ and $j = 0$, we have $\Oracle{\iNExpTime{i}}{\Oracle{\iNExpTime{j}}{\SigmaP{c-1}}} = \Oracle{\iNExpTime{i}}{\SigmaP{c}}$, for all integers $c \geq 1$.
Hence, let us now consider the case in which $i \geq 0$ and $j \geq 1$.
The result below links classes such as $\Oracle{\iNExpTime{i}}{\Oracle{\iNExpTime{j}}{\SigmaP{c-1}}}$ to Hausdorff classes, and hence to the intermediate levels of the iterated exponential hierarchies (see \zcref{theo_exp_nexp_containment}).
In fact, the equivalence between $\PNExp$ and $\NPNExp$, and the resulting collapse of the \SEHText, stems from \zcref{theo_exp_nexp_containment,theo_nexp_nexp_containment} (see \zcref{sec_top_seh} for details).

\begin{theorem}[store=NNcontainment]
\label{theo_nexp_nexp_containment}
Let $i \geq 0$, $j \geq 1$, and $k$, be integers, with $i \leq k \leq i + j$.
Then, for all integers $c \geq 1$,
\[\Oracle{\iNExpTime{i}}{\Oracle{\iNExpTime{j}}{\SigmaP{c-1}}} \subseteq \BoundedHausdCLASS{\iExpPolFunctions{i+1}}{\Oracle{\iNExpTime{(i+j)}}{\SigmaP{c-1}}} \subseteq \BoundedOracle{\iNExpTime{k}}{\bigl(\BoundedHausdCLASS{{\scriptstyle 2}}{\Oracle{\iNExpTime{(i+j-k)}}{\SigmaP{c-1}}}\bigr)}{1} \subseteq \BoundedParOracle{\iNExpTime{k}}{\Oracle{\iNExpTime{(i+j-k)}}{\SigmaP{c-1}}}{2}.\]
\end{theorem}

\begin{proof}
We start by proving $\Oracle{\iNExpTime{i}}{\Oracle{\iNExpTime{j}}{\SigmaP{c-1}}} \subseteq \BoundedHausdCLASS{\iExpPolFunctions{i+1}}{\Oracle{\iNExpTime{(i+j)}}{\SigmaP{c-1}}}$.

Let $\Language{L} \in \Oracle{\iNExpTime{i}}{\Oracle{\iNExpTime{j}}{\SigmaP{c-1}}}$ be a language.
We prove $\Language{L} \in \BoundedHausdCLASS{\iExpPolFunctions{i+1}}{\Oracle{\iNExpTime{(i+j)}}{\SigmaP{c-1}}}$ by showing that $\Language{L}$ can be characterized by an $\Oracle{\iNExpTime{(i+j)}}{\SigmaP{c-1}}$ Hausdorff predicate of length $\iExp{i+1}{p(n)}$ for some polynomial $p(n) \in \PolFunctions$.

Since $\Language{L} \in \Oracle{\iNExpTime{i}}{\Oracle{\iNExpTime{j}}{\SigmaP{c-1}}}$, there are a $\iNExpTime{i}$ oracle machine $\Oracle{\Machine{M}}{?}$ and an $\Oracle{\iNExpTime{j}}{\SigmaP{c-1}}$ oracle $\Omega$ such that $\Language{L} = \LanguageOf{\Oracle{\Machine{M}}{\Omega}}$.
Because $\Machine{M} \in \iNExpTime{i}$, the running time of $\Machine{M}$ is bounded by $\iExp{i}{p_{\Machine{M}}(n)}$, for some polynomial $p_{\Machine{M}}(n)$.
By this, the \emph{size} of the queries issued by $\Machine{M}$ is bounded by $\iExp{i}{p_{\Machine{M}}(n)}$, as well.
Therefore, the \emph{number} of distinct queries (over binary strings) of size at most $\iExp{i}{p_{\Machine{M}}(n)}$ that $\Machine{M}$ might ever ask 
is $\sum_{\ell=0}^{\iExp{i}{p_{\Machine{M}}(n)}} 2^\ell = 2 \cdot \iExp{(i+1)}{p_{\Machine{M}}(n)} - 1$.

Let us define the following three predicates, which will be combined to define a predicate that serves as the basis for the Hausdorff predicate characterizing~$\Language{L}$.
Below, $w$ is a string and $z$ is a non\nbdash-negative integer.

\begin{itemize}[noitemsep]
  \item $\PredHSucc[\Machine{M},\Omega]{C}(w,z)$: $\valtrue$ iff there are \emph{more} than $z$ distinct strings, of length at most $\iExp{i}{p_{\Machine{M}}(\StringLength{w})}$, accepted by $\Omega$;

  \item $\PredHAccCurr[\Machine{M},\Omega]{A}(w,z)$: $\valtrue$ iff there is a set $Y$ of $z$ distinct strings, of length at most $\iExp{i}{p_{\Machine{M}}(\StringLength{w})}$, accepted by $\Omega$, and there is an \emph{accepting} \ounawarelegalemph computation $\pi$ for $\Oracle{\Machine{M}}{?}(w)$ such that every query receiving in $\pi$ a \yesansw belongs to $Y$, and every query receiving a \noansw in $\pi$ does not belong to $Y$;
      and

  \item $\PredHRejCurr[\Machine{M},\Omega]{A}(w,z)$: $\valtrue$ iff there is a set $Y$ of $z$ distinct strings, of length at most $\iExp{i}{p_{\Machine{M}}(\StringLength{w})}$, accepted by $\Omega$, and there is a \emph{rejecting} \ounawarelegalemph computation $\pi$ for $\Oracle{\Machine{M}}{?}(w)$ such that every query receiving in $\pi$ a \yesansw belongs to $Y$, and every query receiving a \noansw in $\pi$ does not belong to $Y$.
\end{itemize}

The three above predicates can be shown in $\Oracle{\iNExpTime{(i+j)}}{\SigmaP{c-1}}$ \Wrt~$\StringLength{w}$ only.
First, remember  
that the number of distinct (binary) strings of length at most $\iExp{i}{p_{\Machine{M}}(\StringLength{w})}$ is $2 \cdot \iExp{(i+1)}{p_{\Machine{M}}(\StringLength{w})} - 1$.
The latter, for $\StringLength{w} \geq 1$, is strictly less than $\iExp{i+1}{2 \cdot p_{\Machine{M}}(\StringLength{w})}$. 
Hence, predicates $\PredHSucc[\Machine{M},\Omega]{C}(w,z)$, $\PredHAccCurr[\Machine{M},\Omega]{A}(w,z)$, and $\PredHRejCurr[\Machine{M},\Omega]{A}(w,z)$, are trivially false when $z \geq \iExp{i+1}{2 \cdot p_{\Machine{M}}(\StringLength{w})}$. 
By this, we only need to consider pairs $\pair{w,z}$ where the \emph{value} of $z$ is \iExponential{(i+1)}{}ly bounded in~$\StringLength{w}$.
Such values of $z$ have a binary representation whose \emph{size} is \iExponential{i} in~$\StringLength{w}$.

Let us now consider $\PredHSucc[\Machine{M},\Omega]{C}(w,z)$.
To answer \yeslbl, we can proceed as follows.
First, we guess $z+1$ distinct strings $q$ of length at most $\iExp{i}{p_{\Machine{M}}(\StringLength{w})}$, i.e., \iExponential{(i+1)}{}ly-many \iExponential{i}{}ly-long strings, together with the respective certificates witnessing $\Omega(q) = 1$.
These certificates are accepting \oawarelegalemph computations for the $\iNExpTime{j}$ ``part'' of $\Omega$---remember that $\Omega \in \Oracle{\iNExpTime{j}}{\SigmaP{c-1}}$;
i.e., we leave out from the guess the part of computation of $\Omega$ associated with the $\SigmaP{c-1}$ oracle.
These accepting computations for $\Omega$ are \iExponential{(i+j)}{}ly-long, because $\Omega$ may receive \iExponential{i}{}-long queries from $\Machine{M}$.
Hence, these are \iExponential{(i+1)}{}ly-many \iExponential{(i+j)}{}ly-long certificates.
Because $j \geq 1$, the overall guess is feasible in nondeterministic \iExponential{(i+j)} time.
We conclude by checking that the guessed strings $q$ are distinct (feasible in \iExponential{(i+1)} time) and the guessed certificates are indeed valid (feasible in \iExponential{(i+j)} time with the aid of an oracle in $\SigmaP{c-1}$).

To decide $\PredHAccCurr[\Machine{M},\Omega]{A}(w,z)$ (resp., $\PredHRejCurr[\Machine{M},\Omega]{A}(w,z)$),
we guess a set $Y$ of $z$ distinct strings $q$ of length at most $\iExp{i}{p_{\Machine{M}}(\StringLength{w})}$, i.e., \iExponential{(i+1)}{}ly-many \iExponential{i}{}ly-long strings, together with the respective \iExponential{(i+j)}{}ly-long certificates witnessing $\Omega(q) = 1$ (see above).
We also guess an accepting (resp., a rejecting) \ounawarelegal computation $\pi$ for $\Oracle{\Machine{M}}{?}(w)$, which is feasible in nondeterministic \iExponential{i} time, as $\Machine{M} \in \iNExpTime{i}$.
The entire guess can hence be carried out in nondeterministic \iExponential{(i+j)} time, because $j \geq 1$ (see above).
Then, we check these four conditions:
(i)~that the guessed strings in $Y$ are distinct (feasible in \iExponential{(i+1)} time);
(ii)~that all the certificates for the strings in $Y$ are valid (feasible in \iExponential{(i+j)} time with the aid of a $\SigmaP{c-1}$ oracle);
(iii)~that $\pi$ is actually an accepting (resp., a rejecting) \ounawarelegal computation $\pi$ for $\Oracle{\Machine{M}}{?}(w)$ (feasible in \iExponential{i} time); and
(iv)~that all queries receiving a \yesansw in $\pi$ belong to $Y$, and all queries receiving a \noansw in $\pi$ do not belong to $Y$ (feasible in \iExponential{(i+1)} time).
Observe that the entire computation can hence be carried out in \iExponential{(i+j)} time, because $j \geq 1$.

\medbreak

We now show that, via the three predicates above, we can define a binary predicate $\Predicate{B}_{\Machine{M},\Omega}(w,z) \subseteq{\HausdPredDomain}$, 
from which we will obtain a Hausdorff predicate $\Language{D}(w,z)$ of length $\iExp{i+1}{2 \cdot p_{\Machine{M}}(n)}$.
Similarly to Hausdorff predicates, $\Predicate{B}_{\Machine{M},\Omega}(w,z)$ will be such that $\Predicate{B}_{\Machine{M},\Omega}(w,z) \geq \Predicate{B}_{\Machine{M},\Omega}(w,z + 1)$, but in $\Predicate{B}_{\Machine{M},\Omega}(w,z)$ we will have $z \geq 0$, instead of $z \geq 1$ like in Hausdorff predicates (see \zcref{def_Hausdorff_predicate}).
The Hausdorff predicate $\Language{D}(w,z)$ obtained from $\Predicate{B}_{\Machine{M},\Omega}(w,z)$ will just be a simple ``shift'' of $\Predicate{B}_{\Machine{M},\Omega}$, i.e., we will have $\Language{D}(w,z) = \Predicate{B}_{\Machine{M},\Omega}(w,z-1)$.
For this reason, we well need to show that $w \in \Language{L}$ iff the highest value of $z$ at which $\Predicate{B}_{\Machine{M},\Omega}(w,z) = 1$ is \emph{even}.
We focus on $\Predicate{B}_{\Machine{M},\Omega}(w,z)$ just to streamline the presentation.

Via the previous three predicates, we define the predicate $\Predicate{B}_{\Machine{M},\Omega}(w,z)$ as follows: 
\[
\Predicate{B}_{\Machine{M},\Omega}(w,z) =
\begin{cases}
  \lnot \PredHSucc[\Machine{M},\Omega]{C}(w,z) \to \PredHAccCurr[\Machine{M},\Omega]{A}(w,z), 
  & \text{if } z < \iExp{i+1}{2 \cdot p_{\Machine{M}}(\StringLength{w})} \text{ and } z \text{ is even} \\
  \lnot \PredHSucc[\Machine{M},\Omega]{C}(w,z) \to \PredHRejCurr[\Machine{M},\Omega]{A}(w,z), 
  & \text{if } z < \iExp{i+1}{2 \cdot p_{\Machine{M}}(\StringLength{w})} \text{ and } z \text{ is odd} \\
  \valfalse,
  & \text{if } z \geq \iExp{i+1}{2 \cdot p_{\Machine{M}}(\StringLength{w})}.
\end{cases}
\]
We claim that $\Predicate{B}_{\Machine{M},\Omega}$ is in $\Oracle{\iNExpTime{(i+j)}}{\SigmaP{c-1}}$ \Wrt to $\StringLength{w}$ only, as the following procedure shows.
First, we compute $\mi{max} = \iExp{i+1}{2 \cdot p_{\Machine{M}}(\StringLength{w})}$ in \iExponential{i} time (see \zcref{sec_maths_complexity}, Iterated exponentials of polynomials), and hence in \iExponential{(i+j)} time, as $j \geq 1$.
We show that comparing $z$ with $\mi{max}$ can be done in \iExponential{i} time.
Indeed, since $z$ and $\mi{max}$ are represented in their canonical binary form, if the size of the representation of $z$ is greater than that of $\mi{max}$, then surely $z > \mi{max}$---%
notice here that it is enough to verify whether the representation of $z$ has even just a single additional bit compared to that of $\mi{max}$.
Conversely, if the representation of $z$ is not longer than that of $\mi{max}$, then we actually compare their values.
Because the representation of $\mi{max}$ has an \iExponential{i} size, comparing $z$ with $\mi{max}$ is feasible in \iExponential{i} time \Wrt $\StringLength{w}$ (and hence in \iExponential{(i+j)} time, as $j \geq 1$).
Regarding the answer to provide, if $z$ is not less than $\mi{max}$, we answer $\valfalse$.
Otherwise, by $\lnot a \to b \equiv a \lor b$, we answer $\PredHSucc[\Machine{M},\Omega]{C}(w,z) \lor \PredHAccCurr[\Machine{M},\Omega]{A}(w,z)$ if $z$ is even, and we answer $\PredHSucc[\Machine{M},\Omega]{C}(w,z) \lor \PredHRejCurr[\Machine{M},\Omega]{A}(w,z)$ if $z$ is odd.
To conclude, since the answer is the disjunction of two predicates in $\Oracle{\iNExpTime{(i+j)}}{\SigmaP{c-1}}$, which is closed under disjunction, the answer can be obtained in $\Oracle{\iNExpTime{(i+j)}}{\SigmaP{c-1}}$.

\medbreak

We start by proving that $\Predicate{B}_{\Machine{M},\Omega}(w,z) \geq \Predicate{B}_{\Machine{M},\Omega}(w,z+1)$, for every string $w$ and $z \geq 0$.
Let $Q$ be the set of all (binary) strings of length at most $\iExp{i}{p_{\Machine{M}}(\StringLength{w})}$.
Let $\hat z \geq 0$ be the \emph{number} of strings in $Q$ accepted by $\Omega$.
By definition of $\PredHSucc[\Machine{M},\Omega]{C}(w,z)$, $\PredHAccCurr[\Machine{M},\Omega]{A}(w,z)$, and $\PredHRejCurr[\Machine{M},\Omega]{A}(w,z)$, the next two properties are immediately proven.

\begin{subproperty}
\label{subprop_C_true_iff_z_smaller_hat_z}
For every integer $z \geq 0$, $\PredHSucc[\Machine{M},\Omega]{C}(w,z) = 1 \Leftrightarrow (z < \hat z)$.
\end{subproperty}

\begin{subproperty}
\label{subprop_if_z_greater_hat_z_A_accept_A_reject_false}
For every integer $z > \hat z$, $\PredHAccCurr[\Machine{M},\Omega]{A}(w,z) = \PredHRejCurr[\Machine{M},\Omega]{A}(w,z) = 0$.
\end{subproperty}

\noindent
By \zcref{subprop_C_true_iff_z_smaller_hat_z}, for all $z$ with $0 \leq z < \hat z$, we have $\Predicate{B}_{\Machine{M},\Omega}(w,z) = 1$.
Furthermore, by \zcref{subprop_C_true_iff_z_smaller_hat_z,subprop_if_z_greater_hat_z_A_accept_A_reject_false}, for all $z$ with $z > \hat z$, we have $\Predicate{B}_{\Machine{M},\Omega}(w,z) = 0$.
Thus, irrespective of the actual truth value of $\Predicate{B}_{\Machine{M},\Omega}(w, \hat z)$, for all $z \geq 0$, it holds that $\Predicate{B}_{\Machine{M},\Omega}(w,z) \geq \Predicate{B}_{\Machine{M},\Omega}(w,z+1)$.

\medbreak

We now show that $w \in \Language{L}$ iff the highest value of $z$ at which $\Predicate{B}_{\Machine{M},\Omega}(w,z) = 1$ is \emph{even}---remember that $\Language{L}$ will be characterized by a Hausdorff predicate $\Language{D}(w,z) = \Predicate{B}_{\Machine{M},\Omega}(w,z-1)$.

We focus on the truth value of $\Predicate{B}_{\Machine{M},\Omega}(w,\hat z)$. 
By \zcref{subprop_C_true_iff_z_smaller_hat_z}, $\PredHSucc[\Machine{M},\Omega]{C}(w,\hat z) = 0$.
Therefore, $\Predicate{B}_{\Machine{M},\Omega}(w,\hat z)$ equals either $\PredHAccCurr[\Machine{M},\Omega]{A}(w,\hat z)$ or $\PredHRejCurr[\Machine{M},\Omega]{A}(w,\hat z)$, depending on $\hat z$ being even or odd, respectively.
We hence analyze the truth value of $\PredHAccCurr[\Machine{M},\Omega]{A}(w,\hat z)$ and $\PredHRejCurr[\Machine{M},\Omega]{A}(w,\hat z)$.
In the proof of the \zcref*[typeset=name,nocap]{subprop_A_accept_reject_correct_at_hat_z_NEXP} below, we use a technique reminiscent of the census and the pseudo-complement adopted by \citet{Mahaney82} and \citet{Kadin1989}, respectively.

\begin{subproperty}
\label{subprop_A_accept_reject_correct_at_hat_z_NEXP}
$\PredHAccCurr[\Machine{M},\Omega]{A}(w,\hat z) = \Oracle{\Machine{M}}{\Omega}(w)$ and $\PredHRejCurr[\Machine{M},\Omega]{A}(w,\hat z) = 1 - \Oracle{\Machine{M}}{\Omega}(w)$.
\end{subproperty}

\begin{subproof}
We show $\PredHAccCurr[\Machine{M},\Omega]{A}(w,\hat z) = \Oracle{\Machine{M}}{\Omega}(w)$.
Proving $\PredHRejCurr[\Machine{M},\Omega]{A}(w,\hat z) = 1 - \Oracle{\Machine{M}}{\Omega}(w)$ is similar.

\ProofRightarrowItem
Assume $\PredHAccCurr[\Machine{M},\Omega]{A}(w,\hat z) = 1$, we prove $\Oracle{\Machine{M}}{\Omega}(w) = 1$.
By $\PredHAccCurr[\Machine{M},\Omega]{A}(w,\hat z) = 1$, there are a subset $Y \subseteq Q$ of strings accepted by $\Omega$, with $\SetSize{Y} = \hat z$, and an accepting \ounawarelegal computation $\pi$ for $\Oracle{\Machine{M}}{?}(w)$, such that every query $q$ appearing in $\pi$ receives in $\pi$ a \yesansw iff $q$ belongs to $Y$.
Remember that $\hat z$ is the exact number of strings in $Q$ accepted by $\Omega$.
Since $\SetSize{Y} = \hat z$, the set $Y$ contains \emph{all and only} the strings from $Q$ accepted by $\Omega$.
Because every query $q$ appearing in $\pi$ receives in $\pi$ a \yesansw iff $q$ belongs to $Y$, we have that $\pi$ is an \oawarelegalemph computation for $\Oracle{\Machine{M}}{\Omega}(w)$.
Since $\pi$ is an accepting computation, $\Oracle{\Machine{M}}{\Omega}(w) = 1$.

\ProofLeftarrowItem
Assume $\Oracle{\Machine{M}}{\Omega}(w) = 1$, we prove $\PredHAccCurr[\Machine{M},\Omega]{A}(w,\hat z) = 1$.
Since $\hat z$ is assumed to be the number of strings in $Q$ accepted by $\Omega$, there must exist a subset $Y \subseteq Q$ of strings accepted by $\Omega$ with $\SetSize{Y} = \hat z$.
The accepting \oawarelegalemph computation $\hat \pi$ for $\Oracle{\Machine{M}}{\Omega}(w)$ witnesses the existence of an accepting \ounawarelegalemph computation for $\Oracle{\Machine{M}}{?}(w)$ such that every query $q$ appearing in $\hat \pi$ receives in $\hat \pi$ a \yesansw iff $q$ is in $Y$.
\end{subproof}

\noindent
By \zcref{subprop_A_accept_reject_correct_at_hat_z_NEXP} and by looking at all possible combinations of $\hat z$ odd/even and $\PredHAccCurr[\Machine{M},\Omega]{A}(w,\hat z)$/$\PredHRejCurr[\Machine{M},\Omega]{A}(w,\hat z)$ $\valtrue$/$\valfalse$, it can be checked that $\Oracle{\Machine{M}}{\Omega}(w) = 1$ iff the highest value of $z$ at which $\Predicate{B}_{\Machine{M},\Omega}(w,z) = 1$ is \emph{even}.

\medbreak

To conclude, consider the binary predicate $\Language{D}(w,z) = \Predicate{B}_{\Machine{M},\Omega}(w,z-1)$---remember that indices for $\Predicate{B}_{\Machine{M},\Omega}$ start at~$0$.
By the properties of $\Predicate{B}_{\Machine{M},\Omega}$ shown above, $\Language{D}$ is a Hausdorff predicate characterizing $\Language{L}$.
Indeed, we have $\Language{D}(w,z) \geq \Language{D}(w,z+1)$ for all strings $w$ and integers $z \geq 1$, i.e., the requirement (\HausdSeqRequirementI) is met.
Moreover, by the definition of $\Predicate{B}_{\Machine{M},\Omega}$, for every string $w$, it holds that $\Language{D}(w,\iExp{i+1}{2 \cdot p_{\Machine{M}}(\StringLength{w})} + 1) = \Predicate{B}_{\Machine{M},\Omega}(w,\iExp{i+1}{2 \cdot p_{\Machine{M}}(\StringLength{w})}) = 0$.
The latter implies that both the requirement (\HausdSeqRequirementII) is met and that the length of the Hausdorff predicate is $\iExp{i+1}{2 \cdot p_{\Machine{M}}(n)}$.
Lastly, for every string $w$, it holds that $w \in \Language{L}$ iff $\HausdIndex{w}{\Language{D}}$ is odd.

\Proofsep

We now show $\BoundedHausdCLASS{\iExpPolFunctions{i+1}}{\Oracle{\iNExpTime{(i+j)}}{\SigmaP{c-1}}} \subseteq \BoundedOracle{\iNExpTime{k}}{\bigl(\BoundedHausdCLASS{{\scriptstyle 2}}{\Oracle{\iNExpTime{(i+j-k)}}{\SigmaP{c-1}}}\bigr)}{1}$.

Let $\Language{L} \in \BoundedHausdCLASS{\iExpPolFunctions{i+1}}{\Oracle{\iNExpTime{(i+j)}}{\SigmaP{c-1}}}$ be a language.
We prove $\Language{L} \in \BoundedOracle{\iNExpTime{k}}{\bigl(\BoundedHausdCLASS{{\scriptstyle 2}}{\Oracle{\iNExpTime{(i+j-k)}}{\SigmaP{c-1}}}\bigr)}{1}$ by exhibiting a $\iNExpTime{k}$ oracle machine $\BoundedOracle{\Machine{M}}{?}{1}$ and a language $\Language{A} \in \BoundedHausdCLASS{2}{\Oracle{\iNExpTime{(i+j-k)}}{\SigmaP{c-1}}}$ such that $\Language{L} = \LanguageOf{\BoundedOracle{\Machine{M}}{\Language{A}}{1}}$.

By $\Language{L} \in \BoundedHausdCLASS{\iExpPolFunctions{i+1}}{\Oracle{\iNExpTime{(i+j)}}{\SigmaP{c-1}}}$, there is an $\Oracle{\iNExpTime{(i+j)}}{\SigmaP{c-1}}$ Hausdorff predicate $\Language{D}$ of length $\iExp{i+1}{p(n)}$, for some polynomial $p(n) \in \PolFunctions$, characterizing $\Language{L}$.
Hence, for every string~$w$, the following relation holds:
\begin{equation}\label{eq_property_for_language_A}
w \in \Language{L} \Leftrightarrow (\exists z \in \NaturalsDomain) \: (1 \leq z \leq \iExp{i+1}{p(\StringLength{w})} \land \text{ $z$ is odd } \land \Language{D}(w,z) = 1 \land \Language{D}(w,z+1) = 0).
\end{equation}

Via \zcref{eq_property_for_language_A}, we define a language $\Language{A}$ so that an oracle machine $\BoundedOracle{\Machine{M}}{?}{1}$ can decide $\Language{L}$ by querying $\Language{A}$~once:
\[\Language{A} = \set{w \dollarsep^m z \mid w,z \in \alphabet^+ \land (z = \text{``}0\text{''} \lor z\text{ starts with `1'}) \land m = \iExp{k}{\StringLength{w}} \land \Language{D}(w,z) = 1 \land {}\linebreak[0] \Language{D}(w,z+1) = 0}.\]
By its definition, $\Language{A}$ contains the strings ``$w \dollarsep^m z$'', where $w$ is an arbitrary non\nbdash-empty binary string, $m = \iExp{k}{\StringLength{w}}$, and $z$ is the canonical binary representation of the \emph{Hausdorff index $\HausdIndex{w}{\Language{D}}$ of $w$ \Wrt $\Language{D}$}.

Let us start by showing that $\Language{L}$ can be decided by a $\iNExpTime{k}$ oracle machine $\BoundedOracle{\Machine{M}}{?}{1}$ issuing a single query to an oracle for $\Language{A}$.
Let $w$ be the input string to $\Machine{M}$.
First, $\Machine{M}$ guesses a binary string $z$ of length \emph{at most} $\iExp{i}{p(\StringLength{w})} + 1$, starting and ending with `$1$' (so the guessed binary string represents an odd number).
With this amount of bits, the guessed string $z$ can represent the value $\iExp{i+1}{p(\StringLength{w})}$ (actually, even greater values).
Next, $\Machine{M}$ issues the query ``$w {\dollarsep}^{m} z$'', with $m = \iExp{k}{\StringLength{w}}$, to its oracle.
To conclude, $\Machine{M}$ returns the very same answer of its oracle.
By the definition of $\Language{A}$, it holds that $\Language{L} = \LanguageOf{\BoundedOracle{\Machine{M}}{\Language{A}}{1}}$.
Notice that $\Machine{M}$, to answer \yeslbl, does \emph{not} need to guess longer strings $z$.
Indeed, since $\Machine{M}$ answers \yeslbl iff it receives a \yesansw from its oracle, and since $\Language{D}$ is a $\iExp{i+1}{p(\StringLength{w})}$\nbdash-long Hausdorff predicate, a necessary condition for $\Machine{M}$ to receive a \yesansw from its oracle is guessing a string $z$ representing a value not exceeding $\iExp{i+1}{p(\StringLength{w})}$.
Observe that the computation of $\Machine{M}$ run within \iExponential{k} time.
Indeed, $\Machine{M}$ guesses a string $z$ of length at most $\iExp{i}{p(\StringLength{w})} + 1$.
Since $k \geq i$, this guess can be carried out in \iExponential{k} time.
Moreover, the query ``$w \dollarsep^m z$'', with $m = \iExp{k}{\StringLength{w}}$, can be written on the query tape in \iExponential{k} time in the following way.
First, we copy the string $w$ from the input tape to the query tape (in linear time \Wrt~$\StringLength{w}$).
Next, we write ``$\dollarsep^m$'' on the query tape.
To this aim, we compute $m = \iExp{k}{\StringLength{w}}$, which can be done in \iExponential{k-1} time \Wrt~$\StringLength{w}$ (see \zcref{sec_maths_complexity}, Iterated exponentials of polynomials).
Observe that the \emph{value} $m$ is \iExponential{k} \Wrt~$\StringLength{w}$.
Once $m$ is computed, we write $m$\nbdash-many `$\dollarsep$' symbols on the query tape.
Hence, ``$\dollarsep^m$'' can be written in \iExponential{k} time.
Then, we write on the query tape the guessed string $z$, and also this can be done in \iExponential{k} time \Wrt~$\StringLength{w}$ (see above).

We are left to prove that $\Language{A} \in \BoundedHausdCLASS{2}{\Oracle{\iNExpTime{(i+j-k)}}{\SigmaP{c-1}}}$.
To this aim, we show that $\Language{A}$ can be characterized by an $\Oracle{\iNExpTime{(i+j-k)}}{\SigmaP{c-1}}$ Hausdorff predicate $\Language{E}$ of length~$2$.
We define $\Language{E}$ as follows:
\begin{equation*}
\Language{E}(s,\ell) = 
\begin{cases}
  \Language{D}(w,z+\ell-1),
  & \text{if }s = \text{``}w \dollarsep^m z\text{''} \land w,z \in \alphabet^+ \land {} \\
  & \qquad (z = \text{``}0\text{''} \lor z\text{ starts with `1'}) \land m = \iExp{k}{\StringLength{w}} \land \ell \leq 2 \\
  \valfalse,
  & \text{otherwise.}
\end{cases}
\end{equation*}

Let us first show that $\Language{E}$ is indeed a Hausdorff predicate of length~$2$.
Let $s$ be an arbitrary string.
There are two cases:
either $s$ is in the form ``$w \dollarsep^m z$'', for some $w,z \in \alphabet^+$, with $z$ either equal to ``$0$'' or starting with `$1$', and $m = \iExp{k}{\StringLength{w}}$,
or $s$ is not in that form.
If $s$ is in that form, the predicate $\Language{E}$ is defined over the Hausdorff predicate $\Language{D}$, and hence $\Language{E}(s,\ell) \geq \Language{E}(s,\ell + 1)$ for all $\ell \geq 1$.
If $s$ is not in that form, $\Language{E}(s,\ell) = 0$ for all $\ell$, and thus $\Language{E}(s,\ell) \geq \Language{E}(s,\ell + 1)$ for all $\ell \geq 1$.
Therefore, the requirement (\HausdSeqRequirementI) is met by~$\Language{E}$.
Notice also that, for $\ell > 2$, the Hausdorff predicate $\Language{E}$ is $\valfalse$ for all strings~$s$.
This, not only implies that $\Language{E}$ fulfils the requirement (\HausdSeqRequirementII), by which we can say that $\Language{E}$ is a Hausdorff predicate, but it also implies that $\Language{E}$ has length~$2$.

We now prove that the Hausdorff predicate $\Language{E}$ characterizes $\Language{A}$.
To this aim, we show that, for every string $s$, it holds that $s \in \Language{A}$ iff the Hausdorff index $\HausdIndex{s}{\Language{E}}$ of $s$ \Wrt $\Language{E}$ is odd.

\smallskip

\ProofRightarrowItem
Assume that $s \in \Language{A}$.
By this, $s$ is a string in the form ``$w \dollarsep^m z$'', for some $w,z \in \alphabet^+$, with $z$ either equal to ``$0$'' or starting with `$1$', and $m = \iExp{k}{\StringLength{w}}$.
By definition of $\Language{A}$ and $s \in \Language{A}$, we have $\Language{D}(w,z) = 1$ and $\Language{D}(w,z+1) = 0$.
Moreover, by definition of $\Language{E}$, it also holds $\Language{E}(s,1) = 1$ and $\Language{E}(s,2) = 0$.
Hence, $\HausdIndex{s}{\Language{E}}$ is odd.

\ProofLeftarrowItem
Assume that $\HausdIndex{s}{\Language{E}}$ is odd.
Since $\Language{E}$ is a Hausdorff predicate of length~$2$, it must be the case that  $\HausdIndex{s}{\Language{E}} = 1$.
Hence, by definition of $\Language{E}$, the string $s$ is in the form ``$w \dollarsep^m z$'', for some $w,z \in \alphabet^+$, with $z$ either equal to ``$0$'' or starting with `$1$', and $m = \iExp{k}{\StringLength{w}}$. 
Moreover, since $\HausdIndex{s}{\Language{E}} = 1$, we also have $\Language{E}(s,1) = 1$ and $\Language{E}(s,2) = 0$, by which $\Language{D}(w,z) = 1$ and $\Language{D}(w,z+1) = 0$.
Altogether, these properties of $s$ imply that $s \in \Language{A}$.

\smallbreak

Lastly, we prove that $\Language{E}$ is in $\Oracle{\iNExpTime{(i+j-k)}}{\SigmaP{c-1}}$, despite being defined over the $\Oracle{\iNExpTime{(i+j)}}{\SigmaP{c-1}}$ Hausdorff predicate~$\Language{D}$.
To this aim, we show that $\Language{E}(s,\ell)$ can be decided in \iExponential{(i+j-k)} time \Wrt the size of~$s$ by a nondeterministic oracle machine $\Oracle{\Machine{N}}{?}$ querying an oracle in $\SigmaP{c-1}$.
The machine $\Machine{N}$ may proceed as follows---remember that the input to $\Machine{N}$ is a string ``$s \hashsep \ell$'', where $\ell$ is a number in its canonical binary from.
\begin{itemize}[nosep,label=--]
\item
First, $\Machine{N}$ checks that $s$ is in the form ``$w \dollarsep^+ z$'', with $w,z \in \alphabet^+$ and $z$ either equals ``$0$'' or starts with `$1$'.
This can be done in linear time \Wrt~$\StringLength{s}$, as we simply need to scan~$s$---we are just checking that $s$ belongs to the regular language $(0|1)^+ \dollarsep^+ ( 0 | ( 1 (0|1)^* ) )$, and we are \emph{not} checking yet that the amount of `$\dollarsep$' symbols in $s$ is correct.
After having scanned $s$, the head of the input tape is on the `$\hashsep$' symbol between $s$ and $\ell$.

\item
Second, $\Machine{N}$ checks whether $\ell > 2$.
If this is the case, $\Machine{N}$ rejects the input. 
This can be done in constant time, as $\Machine{N}$, after having scanned $s$, just needs to look at the first three symbols (of $\ell$) after the `$\hashsep$' symbol.
If all these three symbols are binary digits, then, since $\ell$ is in canonical binary form, $\ell$ is surely strictly greater than~$2$, and $\Machine{N}$ does not need to check the rest of the representation of $\ell$.
If the third symbol after `$\hashsep$' is blank, then $\Machine{N}$ checks whether the two symbols after `$\hashsep$' represent a number greater than~$2$ or not.
Next, $\Machine{N}$ moves the input tape head at the beginning of $s$ (clearly in linear time \Wrt $\StringLength{s}$).

\item
Third, $\Machine{N}$ copies on the second tape the portion $w$ of $s$ from the input tape (in linear time \Wrt~$\StringLength{s}$).

\item
Fourth, $\Machine{N}$ checks that the number of `$\dollarsep$' symbols in~$s$ is $m = \iExp{k}{\StringLength{w}}$.
This task is a bit tricky, as we \emph{cannot} simply compute the value $m$ to then check that the amount of `$\dollarsep$' is correct.
In fact, if there were only ``few'' `$\dollarsep$' symbols on tape and the portion $z$ of $s = \text{``}w \dollarsep^m z\text{''}$ were ``short'', computing $m$ could require \iExponential{(k-1)} time in the size of $s$ (see \zcref{sec_maths_complexity}, Iterated exponentials of polynomials).
However, we are allowed to carry out only a \iExponential{(i+j-k)} time computation in the size of~$s$---observe that it may well be the case that $k - 1 > i + j - k$.
We hence need to employ a different strategy.
Remember that $\iExp{k}{\StringLength{w}}$ is a time constructible function (see \zcref{sec_prelim_central_complexity_classes}), therefore we can start a ``clock'' which ``ticks'' $m = \iExp{k}{\StringLength{w}}$ times (see, e.g., \cite{BalcazarDG1995}).
This ``clock'' is a portion of the transition function of $\Machine{N}$ which forces the machine to perform \emph{exactly} $m$ steps by considering the string $w$ copied on the second tape.
While the ``clock ticks'', we advance on the input tape, one `$\dollarsep$' symbol at a time.
We stop at the earliest of these two events happening:
either when the clock reaches $m$, or when there are no more `$\dollarsep$' symbols on tape.
If we run out of `$\dollarsep$' symbols on tape before the clock reaches $m$, than the amount of `$\dollarsep$' symbols on tape is less than~$m$, and $\Machine{N}$ can reject the input. 
On the other hand, if the clock reaches $m$ and there are more `$\dollarsep$' symbols on tape, then there are too many `$\dollarsep$' symbols on tape, and $\Machine{N}$ can reject the input. 
The key point here is that, with this computation trick, $\Machine{N}$ never performs more computation steps than the amount of `$\dollarsep$' symbols on tape.
Therefore, devised in this way, $\Machine{N}$ checks whether the number of `$\dollarsep$' symbols on tape is correct in linear time \Wrt the size of~$s$.

\item
Fifth, on the second tape $\Machine{N}$ writes `$\hashsep$' at the end of the string $w$ and then copies the portion $z$ of $s$ from the input tape (this can be done in linear time \Wrt $\StringLength{s}$).

\item
Sixth, $\Machine{N}$ adds the value $\ell-1$ to $z$ on the second tape.
This can be done in linear time \Wrt $\StringLength{s}$, as the value to be added to $z$ is either~$0$ or~$1$ (remember that, if $\ell > 2$, the computation halts before this point).

\item
Seventh, to conclude, $\Machine{N}$ decides whether $\pair{w,z+\ell-1} \in \Language{D}$.
Since $\Language{D}$ is an $\Oracle{\iNExpTime{(i+j)}}{\SigmaP{c-1}}$ Hausdorff predicate, there is a nondeterministic oracle machine $\Oracle{\Machine{N}_{\Language{D}}}{?}$ that, aided by a $\SigmaP{c-1}$ oracle $\Omega$, decides whether an arbitrary pair $\pair{v,t}$ belongs to $\Language{D}$ in \iExponential{(i+j)} time \Wrt the size of~$v$ only.
The machine $\Machine{N}$ can simply act as $\Machine{N}_{\Language{D}}$, using the same oracle $\Omega$, by considering the content of the second tape as the input string.
Observe that this computation by $\Machine{N}$ can actually be carried out in \iExponential{(i+j-k)} time \Wrt~$\StringLength{s}$.
Indeed, since $s = w \dollarsep^m z$ and $m = \iExp{k}{\StringLength{w}}$, a run of $\Machine{N}$ that takes \iExponential{(i+j-k)} time \Wrt~$\StringLength{s}$ may actually comprise \iExponential{(i+j)}{}ly\nbdash-many steps \Wrt~$\StringLength{w}$, which is only a portion of $s$.
\end{itemize}

\Proofsep

We now show that $\BoundedOracle{\iNExpTime{k}}{\bigl(\BoundedHausdCLASS{{\scriptstyle 2}}{\Oracle{\iNExpTime{(i+j-k)}}{\SigmaP{c-1}}}\bigr)}{1} \subseteq \BoundedParOracle{\iNExpTime{k}}{\Oracle{\iNExpTime{(i+j-k)}}{\SigmaP{c-1}}}{2}$.

Let $\Language{L} \in \BoundedOracle{\iNExpTime{k}}{\bigl(\BoundedHausdCLASS{{\scriptstyle 2}}{\Oracle{\iNExpTime{(i+j-k)}}{\SigmaP{c-1}}}\bigr)}{1}$ be a language.
By this, there are a $\iNExpTime{k}$ oracle machine $\BoundedOracle{\Machine{M}}{?}{1}$ and a language $\Language{A} \in \BoundedHausdCLASS{2}{\Oracle{\iNExpTime{(i+j-k)}}{\SigmaP{c-1}}}$ such that $\Language{L} = \LanguageOf{\BoundedOracle{\Machine{M}}{\Language{A}}{1}}$.
By $\Language{A} \in \BoundedHausdCLASS{2}{\Oracle{\iNExpTime{(i+j-k)}}{\SigmaP{c-1}}}$, there is an $\Oracle{\iNExpTime{(i+j-k)}}{\SigmaP{c-1}}$ Hausdorff predicate $\Language{D}$ of length~$2$ characterizing~$\Language{A}$. 

We now prove $\Language{L} \in \BoundedParOracle{\iNExpTime{k}}{\Oracle{\iNExpTime{(i+j-k)}}{\SigmaP{c-1}}}{2}$ by exhibiting a $\iNExpTime{k}$ oracle machine $\BoundedParOracle{\Machine{N}}{?}{2}$ such that $\Language{L} = \LanguageOf{\BoundedParOracle{N}{\Language{D}}{2}}$---remember that $\Language{D}$ is an $\Oracle{\iNExpTime{(i+j-k)}}{\SigmaP{c-1}}$ Hausdorff predicate.
The machine $\BoundedParOracle{\Machine{N}}{?}{2}$ decide $\Language{L}$ as follows.
On input~$w$, the machine $\Machine{N}$ first guesses an accepting \ounawarelegal computation $\pi$ for $\BoundedOracle{\Machine{M}}{?}{1}(w)$ (feasible in nondeterministic \iExponential{k} time).
Then, $\Machine{N}$ checks that $\pi$ is legal and accepting (feasible in \iExponential{k} time).
Next, $\Machine{N}$ submits two parallel queries to its oracle $\Language{D}$, namely ``$\tup{\QueryInComp{\pi}{},1}$'' and ``$\tup{\QueryInComp{\pi}{},2}$'', where $\QueryInComp{\pi}{}$ denotes the single query appearing in $\pi$.
To conclude, $\Machine{N}$ verifies that the answer in $\pi$ to $\QueryInComp{\pi}{}$ is \yeslbl iff the answers from its $\Language{D}$ oracle are \yeslbl and \nolbl to ``$\tup{\QueryInComp{\pi}{},1}$'' and ``$\tup{\QueryInComp{\pi}{},2}$'', respectively. 
If this is the case, $\Machine{N}$ answers \yeslbl.
The correctness of the last step stems from $\Language{D}$ being a Hausdorff predicate of length~$2$.
\end{proof}

From the result above we obtain the following corollary.
The proof is in \zcref{sec_detailed_proofs_charting_nexp_oracles}.

\begin{corollary}[store=SummaryNNHausdorff]
\label{theo_summary_nexp_nexp_Hausdorff}
Let $i \geq 0$ and $j \geq 1$ be integers.
Then, for all integers $c \geq 1$,
\[\Oracle{\iNExpTime{i}}{\Oracle{\iNExpTime{j}}{\SigmaP{c-1}}} = \BoundedHausdCLASS{\iExpPolFunctions{i+1}}{\Oracle{\iNExpTime{(i+j)}}{\SigmaP{c-1}}} = \BoundedOracle{\iNExpTime{i}}{\bigl(\BoundedHausdCLASS{{\scriptstyle 2}}{\Oracle{\iNExpTime{j}}{\SigmaP{c-1}}}\bigr)}{1} = \BoundedParOracle{\iNExpTime{i}}{\Oracle{\iNExpTime{j}}{\SigmaP{c-1}}}{2}.\]
\end{corollary}

\zcref[S]{theo_nexp_nexp_containment,theo_summary_nexp_nexp_Hausdorff} can be generalized by also considering a constraint on the size of the queries issued to the oracle (in the spirit of $\PNExp = \Oracle[\ensuremath{\langle\mkern-2mu \PolFunctions \rangle}]{\NExpTime}{\NExpTime}$~\cite[Corollary~4]{SchoningW88} and~\cite[Theorem~24]{AllenderKRR2011}).
Via these generalizations, we can also (re)obtain $\PNExp = \Oracle[\ensuremath{\langle\mkern-2mu \PolFunctions \rangle}]{\NExpTime}{\NExpTime}$ (see \zcref{sec_delta_level_hausdorff} for more).

One of the results of the above corollary, namely $\Oracle{\iNExpTime{i}}{\Oracle{\iNExpTime{j}}{\SigmaP{c-1}}} = \BoundedParOracle{\iNExpTime{i}}{\Oracle{\iNExpTime{j}}{\SigmaP{c-1}}}{2}$, is reminiscent of $\Oracle{\NPTime}{\NPTime} = \BoundedOracle{\NPTime}{\NPTime}{1}$~\cite{Wagner1990}, and more in general of $\Oracle{\iNExpTime{i}}{\Oracle{\iExpTime{j}}{\SigmaP{c-1}}} = \BoundedOracle{\iNExpTime{i}}{\Oracle{\iExpTime{j}}{\SigmaP{c-1}}}{1}$ (see \zcref{theo_nexp_exp_equivalence}).
The fact that in $\Oracle{\iNExpTime{i}}{\Oracle{\iNExpTime{j}}{\SigmaP{c-1}}} = \BoundedParOracle{\iNExpTime{i}}{\Oracle{\iNExpTime{j}}{\SigmaP{c-1}}}{2}$ two parallel queries are needed, and not just one, may be viewed as evidence that the Hausdorff characterization of the intermediate hierarchy levels is indeed natural.
Indeed, two queries, and not just one, are needed to identify the Hausdorff index of the input string.
This is even more supported by the following two results, as they imply that it is rather unlikely that $\Oracle{\iNExpTime{i}}{\Oracle{\iNExpTime{j}}{\SigmaP{c-1}}}$ equals $\BoundedOracle{\iNExpTime{i}}{\Oracle{\iNExpTime{j}}{\SigmaP{c-1}}}{1}$.
We indeed show that, if this were the case, there would be a collapse of the intermediate levels $\BoundedOracle{\iExpTime{(i+j)}}{\SigmaP{c}}{\iExpPolFunctions{g}}$, for all $g$ with $0 \leq g \leq i+j-1$ (i.e., the last intermediate level of the step is \emph{not} involved), to the second level of the \BHText over $\Oracle{\iNExpTime{(i+j)}}{\SigmaP{c-1}}$, which is not expected to happen (see~\cite{Dawar1998} and references therein);
we provide an intuition on this after stating \zcref{theo_NN1_equivalence_BH}.

\begin{theorem}
\label{theo_NN1_containment_BH}
Let $i,k \geq 0$ and $j \geq 1$ be integers, with $1 \leq k \leq i + j$.
Then, for all integers $c \geq 1$,
\[\BoundedOracle{\iNExpTime{i}}{\Oracle{\iNExpTime{j}}{\SigmaP{c-1}}}{1} \subseteq \ComplementPrefixKerned{}\BoundedHausdCLASS{2}{\Oracle{\iNExpTime{(i+j)}}{\SigmaP{c-1}}} \subseteq \BoundedOracle{\iNExpTime{k}}{\Oracle{\iNExpTime{(i+j-k)}}{\SigmaP{c-1}}}{1}.\]
\end{theorem}

\begin{proof}
We start by showing that $\BoundedOracle{\iNExpTime{i}}{\Oracle{\iNExpTime{j}}{\SigmaP{c-1}}}{1} \subseteq \ComplementPrefixKerned{}\BoundedHausdCLASS{2}{\Oracle{\iNExpTime{(i+j)}}{\SigmaP{c-1}}}$.

Let $\Language{L} \in \BoundedOracle{\iNExpTime{i}}{\Oracle{\iNExpTime{j}}{\SigmaP{c-1}}}{1}$ be a language.
We prove that $\Language{L} \in \ComplementPrefixKerned{}\BoundedHausdCLASS{2}{\Oracle{\iNExpTime{(i+j)}}{\SigmaP{c-1}}}$ by exhibiting an $\Oracle{\iNExpTime{(i+j)}}{\SigmaP{c-1}}$ Hausdorff predicate $\Language{D}$ of length~$2$ 
such that, for all strings $w$, $w \in \Language{L}$ iff $\HausdIndex{w}{\Language{D}}$ is \emph{even}.

By $\Language{L} \in \BoundedOracle{\iNExpTime{i}}{\Oracle{\iNExpTime{j}}{\SigmaP{c-1}}}{1}$, there are a $\iNExpTime{i}$ oracle machine $\BoundedOracle{\Machine{M}}{?}{1}$ and an $\Oracle{\iNExpTime{j}}{\SigmaP{c-1}}$ oracle $\Omega$ such that $\Language{L} = \LanguageOf{\BoundedOracle{\Machine{M}}{\Omega}{1}}$.
We define $\Language{D}$ via the predicates $\Predicate{A}_{\Machine{M},\Omega}$ and $\Predicate{B}_{\Machine{M},\Omega}$ below. 
In the following, given an \ounawarelegal computation $\pi$ for the oracle machine $\BoundedOracle{\Machine{M}}{?}{1}$, we denote by $\QueryInComp{\pi}{}$ and $\AnsInComp{\pi}{}$ the single query and the single oracle answer appearing in $\pi$, respectively.
In the definitions of the predicates below, $w$ is a string.

\begin{itemize}[noitemsep]
  \item $\Predicate{A}_{\Machine{M},\Omega}(w)$: $\valtrue$ iff, \emph{all} accepting \ounawarelegal computations $\pi$ for $\Oracle{\Machine{M}}{?}(w)$ with $\AnsInComp{\pi}{} = 0$ are such that $\Omega(\QueryInComp{\pi}{}) = 1$; and
  \item $\Predicate{B}_{\Machine{M},\Omega}(w)$: $\valtrue$ iff \emph{there exists} an accepting \oawarelegal computation $\pi$ for $\Oracle{\Machine{M}}{\Omega}(w)$ with $\Omega(\QueryInComp{\pi}{}) = 1$.
\end{itemize}

We claim that the two above predicates are in $\Oracle{\iNExpTime{(i+j)}}{\SigmaP{c-1}}$.
First, observe the following.
Since $\Machine{M}$ is an \iExponential{i} time nondeterministic machine, we can pick a polynomial $p_{\Machine{M}}(n)$ such that every \ounawarelegal computation $\pi$ for $\Machine{M}$ consists of no more than $\iExp{i}{p_\Machine{M}(n)}$ IDs, \emph{and} each ID of $\pi$ can be described with no more than $\iExp{i}{p_\Machine{M}(n)}$ symbols (see \zcref{sec_proving_qbsf-problems_hard}).
Remember that an ID can be encoded into a binary string with only linear overhead \Wrt to its non\nbdash-binary representation (see, e.g., \cite{Hopcroft1979}).
Hence, the size of a single ID of $\pi$ is bounded by $d \cdot \iExp{i}{p_\Machine{M}(n)}$, for some constant $d$.
By this, $\pi$ can be encoded into a binary string of size at most $d \cdot {\bigl( \iExp{i}{p_\Machine{M}(n)} \bigr)}^{2}$.
For $i \geq 0$, by \zcref{theo_polynomial_of_iterated_exponentials}.1, we have that $d \cdot {\bigl( \iExp{i}{p_\Machine{M}(n)} \bigr)}^{2}$ is \iExponential{i}.

Let us now consider $\Predicate{A}_{\Machine{M},\Omega}(w)$.
To answer true, we can proceed as follows.
We generate in turn all the possible binary strings of length $d \cdot {\bigl( \iExp{i}{p_\Machine{M}(\StringLength{w})} \bigr)}^{2}$.
From the observation above, writing on tape one of these strings takes \iExponential{i} time.
For each of these strings $s$, we check whether $s$ is an accepting \ounawarelegal computation for $\BoundedOracle{\Machine{M}}{?}{1}(w)$.
On a single string $s$, this check can be carried out in \iExponential{i} time.
However, all the possible binary strings of length $d \cdot {\bigl( \iExp{i}{p_\Machine{M}(\StringLength{w})} \bigr)}^{2}$ are \iExponential{(i+1)}{}ly-many.
This is not an obstacle, because we are assuming $j \geq 1$, and hence we can go through all of them in \iExponential{(i+j)} time.
If $s$ is indeed an accepting \ounawarelegal computation for $\BoundedOracle{\Machine{M}}{?}{1}(w)$, we check whether $\AnsInComp{s}{} = 0$ (feasible in \iExponential{i} time).
If this is the case, we check that $\Omega(\QueryInComp{s}{}) = 1$, by guessing a witnessing certificate.
This certificate is an accepting \oawarelegalemph computation $\pi$ for $\Omega(\QueryInComp{s}{})$ of the $\iNExpTime{j}$ ``part'' of $\Omega$---remember that $\Omega \in \Oracle{\iNExpTime{j}}{\SigmaP{c-1}}$;
i.e., we leave out from the guess the part of computation associated with the $\SigmaP{c-1}$ oracle.
Notice that $\pi$ can be \iExponential{(i+j)}{}ly long, as $\Omega$ might receive from $\Machine{M}$ \iExponential{i}{}ly-long queries.
The guess of $\pi$ can hence be carried out in nondeterministic \iExponential{(i+j)} time.
Checking that $\pi$ is an accepting \oawarelegal computation $\pi$ for $\Omega(\QueryInComp{s}{})$ can be done in \iExponential{(i+j)} time with the aid of an oracle in $\SigmaP{c-1}$.
If for all the strings $s$ the checks are passed, we answer \yeslbl.
This procedure is in $\Oracle{\iNExpTime{(i+j)}}{\SigmaP{c-1}}$.

Consider $\Predicate{B}_{\Machine{M},\Omega}(w)$.
Answering true on $\Predicate{B}_{\Machine{M},\Omega}(w)$ is straightforward. 
Simply, we guess an \iExponential{i}{}ly-long accepting \ounawarelegal computation $\pi$ for $\BoundedOracle{\Machine{M}}{?}{1}(w)$ such that $\AnsInComp{\pi}{} = 1$, together with a certificate for $\Omega(\QueryInComp{\pi}{}) = 1$.
As above, this certificate is an accepting \oawarelegalemph computation $\sigma$ for $\Omega(\QueryInComp{\pi}{})$ of the $\iNExpTime{j}$ ``part'' of $\Omega$.
The certificate $\sigma$ is \iExponential{(i+j)}{}ly-long (see above), and hence the guess can be carried out in nondeterministic \iExponential{(i+j)} time.
Checking that $\pi$ is an accepting \ounawarelegal computation for $\BoundedOracle{\Machine{M}}{?}{1}(w)$ such that $\AnsInComp{\pi}{} = 1$ can be done in \iExponential{i} time.
Checking that $\sigma$ is an accepting \oawarelegal computation for $\Omega(\QueryInComp{\pi}{})$ can be carried out in \iExponential{(i+j)} time with the aid of a $\SigmaP{c-1}$ oracle.

We define the next predicate based on the previous two:
\[
\Language{D}(w,z) =
\begin{cases}
  \Predicate{A}_{\Machine{M},\Omega}(w) \lor \Predicate{B}_{\Machine{M},\Omega}(w),
  & \text{if } z = 1 \\
  \Predicate{B}_{\Machine{M},\Omega}(w),
  & \text{if } z = 2 \\
  \valfalse,
  & \text{if } z > 2.
\end{cases}
\]

By definition of $\Language{D}$, it holds that $\Language{D}(w,z) \geq \Language{D}(w,z+1)$, for every string $w$ and for all $z \geq 1$, hence the requirement~(\HausdSeqRequirementI) is met.
Moreover, again by definition of $\Language{D}$, for every string $w$ and every integer $z > 2$, it holds that $\Language{D}(w,z) = 0$.
Hence, the requirement~(\HausdSeqRequirementII) is met, by which $\Language{D}$ is a Hausdorff predicate, and its length is~2.

The predicate $\Language{D}$ can also be shown in $\Oracle{\iNExpTime{(i+j)}}{\SigmaP{c-1}}$ \Wrt~$\StringLength{w}$ only.
For an input string ``$w \hashsep z$'', checking whether $z > 2$ can be done by looking at the first three symbols on the input tape after the `$\hashsep$' symbol: 
If all these three symbols are binary digits, then, since $z$ is in canonical binary form, $z$ is surely strictly greater than~$2$, and there is no need to check the rest of the representation of $z$.
If the third symbol after `$\hashsep$' is blank, then we simply check whether the two symbols after `$\hashsep$' represent a number greater than~$2$ or not. 

After this initial check, if $z = 2$, we answer $\Predicate{B}_{\Machine{M},\Omega}(w)$, which can be done in $\Oracle{\iNExpTime{(i+j)}}{\SigmaP{c-1}}$ (see above).
On the other hand, if $z = 1$, we answer $\Predicate{A}_{\Machine{M},\Omega}(w) \lor \Predicate{B}_{\Machine{M},\Omega}(w)$, which can be done in $\Oracle{\iNExpTime{(i+j)}}{\SigmaP{c-1}}$, as both predicate are in $\Oracle{\iNExpTime{(i+j)}}{\SigmaP{c-1}}$ (see above), and the latter is closed under disjunction.

We now show that, for every string $w$, $w \in \Language{L}$ iff the Hausdorff index $\HausdIndex{w}{\Language{D}}$ of $w$ \Wrt $\Language{D}$ is \emph{even}. 

\smallbreak

\ProofRightarrowItem
Assume that $w \in \Language{L}$.
There are two cases:
either (i) there is an accepting \oawarelegal computation $\pi$ for $\BoundedOracle{\Machine{M}}{\Omega}{1}(w)$ with $\Omega(\QueryInComp{\pi}{}) = 1$,
or (ii) there is not.
Consider case~(i).
In this case, $\Predicate{B}_{\Machine{M},\Omega}(w)$ is true, by which $\Language{D}(w,1) = \Language{D}(w,2) = 1$, and hence $\HausdIndex{w}{\Language{D}} = 2$, which is even.
Consider now case~(ii).
In this case, $\Predicate{B}_{\Machine{M},\Omega}(w)$ is false, by which $\Language{D}(w,2) = 0$ and $\Language{D}(w,1) = \Predicate{A}_{\Machine{M},\Omega}(w)$.
Since there is \emph{no} accepting \oawarelegal computation for $\BoundedOracle{\Machine{M}}{\Omega}{1}(w)$ with $\Omega(\QueryInComp{\pi}{}) = 1$, but we are assuming $w \in \Language{L}$, there must be an accepting \oawarelegal computation $\pi$ for $\BoundedOracle{\Machine{M}}{\Omega}{1}(w)$ with $\Omega(\QueryInComp{\pi}{}) = 0$.
Observe that such an \oawarelegal computation $\pi$ is \emph{also} an \ounawarelegalemph computation for $\BoundedOracle{\Machine{M}}{?}{1}(w)$ with $\AnsInComp{\pi}{} = 0$ such that $\Omega(\QueryInComp{\pi}{}) = 0$.
By this, also $\Predicate{A}_{\Machine{M},\Omega}(w)$ is false, and hence $\Language{D}(w,1) = 0$.
Therefore, $\HausdIndex{w}{\Language{D}} = 0$, which is also in this case even.

\ProofLeftarrowItem
Assume that $\HausdIndex{w}{\Language{D}}$ is even.
Since $\Language{D}$ is of length~$2$, we have that either $\HausdIndex{w}{\Language{D}} = 0$ or $\HausdIndex{w}{\Language{D}} = 2$.
If $\HausdIndex{w}{\Language{D}} = 2$, then $\Predicate{B}_{\Machine{M},\Omega}(w)$ is true, by which there is an accepting \oawarelegal computation for $\BoundedOracle{\Machine{M}}{\Omega}{1}(w)$ (with $\Omega(\QueryInComp{\pi}{}) = 1$, which is not relevant now).
Hence, $w \in \Language{L}$.
Assume now that $\HausdIndex{w}{\Language{D}} = 0$.
In this case $\Predicate{B}_{\Machine{M},\Omega}(w)$ is false, from which it follows that $\Language{D}(w,1) = \Predicate{A}_{\Machine{M},\Omega}(w)$.
Since $\HausdIndex{w}{\Language{D}} = 0$, also $\Predicate{A}_{\Machine{M},\Omega}(w)$ is false, implying that there exists an accepting \ounawarelegal computation $\pi$ for $\BoundedOracle{\Machine{M}}{?}{1}(w)$ such that $\AnsInComp{\pi}{} = 0$ and $\Omega(\QueryInComp{\pi}{}) = 0$.
Observe that $\pi$ is actually an accepting \oawarelegal computation for $\BoundedOracle{\Machine{M}}{\Omega}{1}$.
Therefore, $w \in \Language{L}$.

\Proofsep

We now prove $\ComplementPrefixKerned{}\BoundedHausdCLASS{2}{\Oracle{\iNExpTime{(i+j)}}{\SigmaP{c-1}}} \subseteq \BoundedOracle{\iNExpTime{k}}{\Oracle{\iNExpTime{(i+j-k)}}{\SigmaP{c-1}}}{1}$.

Let $\Language{L} \in \ComplementPrefixKerned{}\BoundedHausdCLASS{2}{\Oracle{\iNExpTime{(i+j)}}{\SigmaP{c-1}}}$ be a language.
We show $\Language{L} \in \BoundedOracle{\iNExpTime{k}}{\Oracle{\iNExpTime{(i+j-k)}}{\SigmaP{c-1}}}{1}$ by exhibiting a $\iNExpTime{k}$ oracle machine $\BoundedOracle{\Machine{M}}{?}{1}$ and an $\Oracle{\iNExpTime{(i+j-k)}}{\SigmaP{c-1}}$ oracle $\Omega$ such that $\Language{L} = \LanguageOf{\BoundedOracle{\Machine{M}}{\Omega}{1}}$.

Since $\Language{L} \in \ComplementPrefixKerned{}\BoundedHausdCLASS{2}{\Oracle{\iNExpTime{(i+j)}}{\SigmaP{c-1}}}$, there exists an $\Oracle{\iNExpTime{(i+j)}}{\SigmaP{c-1}}$ Hausdorff predicate $\Language{D}$ of length~$2$ such that, for every string $w$, $w \in \Language{L}$ iff the Hausdorff index $\HausdIndex{w}{\Language{D}}$ of $w$ \Wrt $\Language{D}$ is \emph{even}.
By this, we have that: 
\begin{equation}\label{eq_property_for_predicates_A_and_B}
w \in \Language{L} \Leftrightarrow \underbracket[.5pt]{(\Language{D}(w,1) = 0 \land \Language{D}(w,2) = 0)}_{{} \Leftrightarrow \Predicate{A}(w) = 1} \lor \underbracket[.5pt]{(\Language{D}(w,1) = 1 \land \Language{D}(w,2) = 1)}_{{} \Leftrightarrow \Predicate{B}(w) = 1}.
\end{equation}

Because $\smash{\Language{D} \in \Oracle{\iNExpTime{(i+j)}}{\SigmaP{c-1}}}$, the predicate $\Predicate{B}(w)$ and the \emph{negation} of the predicate $\Predicate{A}(w)$, that is, $\compl{\Predicate{A}}(w) \equiv (\Language{D}(w,1) = 1 \lor \Language{D}(w,2) = 1)$, are in $\Oracle{\iNExpTime{(i+j)}}{\SigmaP{c-1}}$, as the latter is closed under conjunction and disjunction.

Let us now define the oracle machine $\BoundedOracle{\Machine{M}}{?}{1}$ and the oracle $\Omega$.
We start by looking at $\Omega$.
The machine $\Omega$ is designed to receive from $\Machine{M}$ pairs $\pair{\wt{w},\alpha}$ from the domain $\StringUniverse \! \times \set{\text{`}\compl{a}\text{'},\text{`}b\text{'}}$, where $\wt{w}$ is a padded version of the string~$w$ in input to $\Machine{M}$ to reach \iExponential{k} length (\Wrt~$\StringLength{w}$), and $\set{\text{`}\compl{a}\text{'},\text{`}b\text{'}}$ are two symbols.
By this, since $\Omega \in \Oracle{\iNExpTime{(i+j-k)}}{\SigmaP{c-1}}$, the machine $\Omega$ can run for \iExponential{(i+j)} time \Wrt~$\StringLength{w}$.
Upon reception of the query $\pair{\wt{w},\alpha}$, $\Omega$ simply ignores the padding of $\wt{w}$ and decides the language below:
\[
\LanguageOf{\Omega}(\wt{w},\alpha) =
\begin{cases}
  \compl{\Predicate{A}}(w),
  & \text{if } \alpha = \text{`}\compl{a}\text{'} \\
  \Predicate{B}(w),
  & \text{if } \alpha = \text{`}b\text{'}.
\end{cases}
\]

The oracle machine $\BoundedOracle{\Machine{M}}{?}{1}$ works as follows:
first, $\Machine{M}$ guesses a symbol $\hat \alpha$ from $\set{\text{`}\compl{a}\text{'},\text{`}b\text{'}}$.
Then, $\Machine{M}$ submits to $\Omega$ the query $\tup{\wt{w},\hat \alpha}$, where $\wt{w}$ is an \iExponential{k}{}ly padded version of~$w$.
To conclude, $\Machine{M}$ returns the answer of $\Omega$ if $\hat \alpha = \text{`}b\text{'}$, and the opposite answer of $\Omega$ if $\hat \alpha = \text{`}\compl{a}\text{'}$.
This procedure is feasible in \iExponential{k} time.

We show that $\Language{L} = \LanguageOf{\BoundedOracle{\Machine{M}}{\Omega}{1}}$.
\ProofRightarrowItem
Let $w \in \Language{L}$ be a string.
We prove $\BoundedOracle{\Machine{M}}{\Omega}{1}(w) = 1$.
By $w \in \Language{L}$, we have $\Predicate{A}(w) = 1$ \emph{or} $\Predicate{B}(w) = 1$ (see \zcref{eq_property_for_predicates_A_and_B}).
To accept $w$, $\Machine{M}$ guess $\hat \alpha = \text{`}\compl{a}\text{'}$ if $\Predicate{A}(w) = 1$ or guess $\hat \alpha = \text{`}b\text{'}$ if $\Predicate{B}(w) = 1$.
The right guess of $\hat \alpha$ allows $\Machine{M}$ to answer \yeslbl.
\ProofLeftarrowItem
Le $w \notin \Language{L}$ be a string.
We prove $\BoundedOracle{\Machine{M}}{\Omega}{1}(w) = 0$.
Since $w \notin \Language{L}$, $\Predicate{A}(w) = 0$ \emph{and} $\Predicate{B}(w) = 0$ (see \zcref{eq_property_for_predicates_A_and_B}).
Thus, there is no guess of $\hat a$ enabling $\Machine{M}$ to answer \yeslbl.
\end{proof}

From the above \zcref*[typeset=name,nocap]{theo_NN1_containment_BH} stems the following rather interesting result, which states that for bounded oracle classes $\BoundedOracle{\iNExpTime{i}}{\Oracle{\iNExpTime{j}}{\SigmaP{c-1}}}{1}$, when $j \geq 1$, the exact values of $i$ and $j$ are not important, as all those oracle classes equal $\ComplementPrefixKerned\BoundedHausdCLASS{2}{\Oracle{\iNExpTime{(i+j)}}{\SigmaP{c-1}}}$.  
The proof is in \zcref{sec_detailed_proofs_charting_nexp_oracles}.

\begin{corollary}[store=NN1EquivalenceBH]
\label{theo_NN1_equivalence_BH}
Let $i,i' \geq 0$, $j,j' \geq 1$, and $\ell$, be integers with $i + j = \ell = i' + j'$.
Then, for all integers $c \geq 1$,
\[\BoundedOracle{\iNExpTime{i}}{\Oracle{\iNExpTime{j}}{\SigmaP{c-1}}}{1} = \ComplementPrefixKerned\BoundedHausdCLASS{2}{\Oracle{\iNExpTime{\ell}}{\SigmaP{c-1}}} = \BoundedOracle{\iNExpTime{i'}}{\Oracle{\iNExpTime{j'}}{\SigmaP{c-1}}}{1}.\]
\end{corollary}

As anticipated above, \zcref{theo_summary_equivalence_intermediate_levels_Hausdorff,theo_summary_nexp_nexp_Hausdorff,theo_NN1_equivalence_BH} support why $\Oracle{\iNExpTime{i}}{\Oracle{\iNExpTime{j}}{\SigmaP{c-1}}} = \BoundedParOracle{\iNExpTime{i}}{\Oracle{\iNExpTime{j}}{\SigmaP{c-1}}}{2}$, but we do not expect $\Oracle{\iNExpTime{i}}{\Oracle{\iNExpTime{j}}{\SigmaP{c-1}}}$ to equal $\BoundedOracle{\iNExpTime{i}}{\Oracle{\iNExpTime{j}}{\SigmaP{c-1}}}{1}$.
By \zcref{theo_summary_equivalence_intermediate_levels_Hausdorff,theo_summary_nexp_nexp_Hausdorff}, we have that $\Oracle{\iNExpTime{i}}{\Oracle{\iNExpTime{j}}{\SigmaP{c-1}}} = \BoundedHausdCLASS{\iExpPolFunctions{i+1}}{\Oracle{\iNExpTime{(i+j)}}{\SigmaP{c-1}}} = \BoundedOracle{\iExpTime{(i+j)}}{\SigmaP{c}}{\iExpPolFunctions{i}}$.
This means that, for a fixed sum $i + j$, with varying values of $i$ and $j$ the oracle classes $\Oracle{\iNExpTime{i}}{\Oracle{\iNExpTime{j}}{\SigmaP{c-1}}}$ characterize the various intermediate levels above the main level $\Oracle{\iNExpTime{(i+j)}}{\SigmaP{c-1}}$ of the \iWEHText{(i{+}j)}.
On the other hand, by \zcref{theo_NN1_equivalence_BH}, for a fixed sum $i + j$, with varying values of $i$ and $j$ (when $j \geq 1$) the oracle classes $\BoundedOracle{\iNExpTime{i}}{\Oracle{\iNExpTime{j}}{\SigmaP{c-1}}}{1}$ all equal the same Hausdorff class $\ComplementPrefix\BoundedHausdCLASS{2}{\Oracle{\iNExpTime{(i+j)}}{\SigmaP{c-1}}}$.
By this, if $\Oracle{\iNExpTime{i}}{\Oracle{\iNExpTime{j}}{\SigmaP{c-1}}}$ were equal to $\BoundedOracle{\iNExpTime{i}}{\Oracle{\iNExpTime{j}}{\SigmaP{c-1}}}{1}$, we would have that all the intermediate levels above $\Oracle{\iNExpTime{(i+j)}}{\SigmaP{c-1}}$, but the last one as $j \geq 1$, would be equal to $\ComplementPrefix\BoundedHausdCLASS{2}{\Oracle{\iNExpTime{(i+j)}}{\SigmaP{c-1}}}$, which is at the second level of the \BHText over $\Oracle{\iNExpTime{(i+j)}}{\SigmaP{c-1}}$.

\subsubsection%
[\texorpdfstring{${i}$}{i}E\texorpdfstring{{\smaller XP}}{XP}S\texorpdfstring{{\smaller PACE}}{PACE} Oracle Machines with N\texorpdfstring{${j}$}{j}E\texorpdfstring{{\smaller XP}}{XP} Oracles]%
{\texorpdfstring{$\boldsymbol{i}$}{i}E\texorpdfstring{{\smaller XP}}{XP}S\texorpdfstring{{\smaller PACE}}{PACE} Oracle Machines with N\texorpdfstring{$\boldsymbol{j}$}{j}E\texorpdfstring{{\smaller XP}}{XP} Oracles}
\label{sec_expspace_nexp_oracle_classes}

In this section, we look at the classes $\Oracle{\iExpSpace{i}}{\Oracle{\iNExpTime{j}}{\SigmaP{c-1}}}$.

For our first result, namely for $\Oracle{\iExpSpace{(i-1)}}{\Oracle{\iNExpTime{j}}{\SigmaP{c-1}}} \subseteq \BoundedHausdCLASS{\iExpPolFunctions{i}}{\Oracle{\iNExpTime{(i+j)}}{\SigmaP{c-1}}}$, one may think that this can easily be obtained by noticing that $\Oracle{\iExpSpace{(i-1)}}{\Oracle{\iNExpTime{j}}{\SigmaP{c-1}}} \subseteq \Oracle{\iExpTime{i}}{\Oracle{\iNExpTime{j}}{\SigmaP{c-1}}} = \BoundedOracle{\iExpTime{i}}{\Oracle{\iNExpTime{j}}{\SigmaP{c-1}}}{\iExpPolFunctions{i}}$, and then \zcref{theo_summary_generalized_equivalence_intermediate_levels_Hausdorff} applies.
However, this argument implies only that
$\Language{L} \in \BoundedHausdCLASS{\iExpPolFunctions{(i+1)}}{\Oracle{\iNExpTime{(i+j)}}{\SigmaP{c-1}}}$.

Interestingly, showing  $\Oracle{\iExpSpace{(i-1)}}{\Oracle{\iNExpTime{j}}{\SigmaP{c-1}}} \subseteq \BoundedHausdCLASS{\iExpPolFunctions{i}}{\Oracle{\iNExpTime{(i+j)}}{\SigmaP{c-1}}}$ requires a tailored proof.
The intuition why this is required is as follows.
The different queries in all the possible \ounawarelegal computations of a $\Oracle{\iExpTime{i}}{?}$ machine are \iExponential{(i+1)}{}ly-many.
Whereas, the different queries in all the possible \ounawarelegal computations of a $\Oracle{\iExpSpace{(i-1)}}{?}$ machine are only \iExponential{i}{}ly-many (for more details, see the proof of \zcref{theo_expspace_nexp_containment}), and this is a key point of the proof.
In fact, the proof will resemble the technique used to prove \zcref{theo_nexp_nexp_containment} and will rest on the fact that an $\iExpSpace{(i-1)}$ oracle machine cannot generate too many different queries over all its possible \ounawarelegal computations for a given input.

The analogue result restricted to $\PolHier$, i.e., $\ParOracle{\PTime}{\NPTime} = \Oracle{\LogSpace}{\NPTime}$, was independently shown by \citet{Wagner1990} and \citet{Buss1991}.
Both works show that polynomial truth-table reductions to \NPTime (i.e., $\ParOracle{\PTime}{\NPTime}$) coincide with logspace truth-table reductions to \NPTime, and then invoke a result by \citet{LadnerL76} establishing that logspace truth-table and logspace Turing reductions to \NPTime are equivalent.
Our proof, on the other hand, directly show that languages in $\Oracle{\iExpSpace{(i-1)}}{\Oracle{\iNExpTime{j}}{\SigmaP{c-1}}}$ can be characterized by $\Oracle{\iNExpTime{(i+j)}}{\SigmaP{c-1}}$ Hausdorff predicate of \iExponential{i} length, and only marginally shares ideas with \citet{LadnerL76}'s proof. 

As \citet{LadnerL76} did, we base our argument on the notion of an \emph{ID graph for an input string~$x$ and an oracle machine $\Oracle{\Machine{M}}{?}$}.
They introduced this concept to analyze the queries that an oracle machine $\Oracle{\Machine{M}}{?}$, \emph{when executed on~$x$}, may ever ask to a \emph{generic} oracle.
Intuitively, the ID graph is an extended computation tree for $\Oracle{\Machine{M}}{?}(x)$ representing all possible computations on input $x$ under every possible pattern of oracle answers.
Our use of this notion, however, differs substantially from theirs.
\Citeauthor{LadnerL76} construct an explicit algorithm that explores the ID graph as a data structure, following the correct computation path guided by the oracle's answers, which are assumed to be available to the algorithm.
In contrast, we employ the ID graph only as a conceptual tool to reason about the set of queries that an $\iExpSpace{(i-1)}$ oracle machine may issue to an unknown oracle.
In fact, we cannot rely on an explicit exploration of the ID graph, because our proof relies on Hausdorff reducing languages to suitable Hausdorff predicates.
Hence, we need to resort to different methods.

We employ a census technique combined with a variant of the pseudo\nbdash-complement method adopted by \citet{Mahaney82} and \citet{Kadin1989}.
These tools allow us to determine, for a given input, all queries that an $\iExpSpace{(i-1)}$ oracle machine may issue and that are accepted by the oracle.
This approach avoids the need for the ``tt\nbdash-condition generator'' and ``tt\nbdash-condition evaluator'' used in the proof of \citet{LadnerL76}.

\begin{theorem}[store=SNcontainment]
\label{theo_expspace_nexp_containment}
Let $i,j \geq 0$ and $k$ be integers with $i \leq k \leq i + j$.
Then, for all integers $c \geq 1$,
\[\Oracle{\iExpSpace{(i-1)}}{\Oracle{\iNExpTime{j}}{\SigmaP{c-1}}} \subseteq \BoundedHausdCLASS{\iExpPolFunctions{i}}{\Oracle{\iNExpTime{(i+j)}}{\SigmaP{c-1}}} \subseteq \mkern -2mu
\begin{cases}
  \BoundedOracle{\iExpSpace{(k-1)}}{\Oracle{\iNExpTime{(i+j-k)}}{\SigmaP{c-1}}}{\iExpPolFunctions{(i-1)}} \\
  \BoundedParOracle{\iExpSpace{(k-1)}}{\Oracle{\iNExpTime{(i+j-k)}}{\SigmaP{c-1}}}{\iExpPolFunctions{i}}.
\end{cases}\]
\end{theorem}

\begin{proof}
We start by proving $\Oracle{\iExpSpace{(i-1)}}{\Oracle{\iNExpTime{j}}{\SigmaP{c-1}}} \subseteq \BoundedHausdCLASS{\iExpPolFunctions{i}}{\Oracle{\iNExpTime{(i+j)}}{\SigmaP{c-1}}}$.

Let $\Language{L} \in \Oracle{\iExpSpace{(i-1)}}{\Oracle{\iNExpTime{j}}{\SigmaP{c-1}}}$ be a language.
We prove $\Language{L} \in \BoundedHausdCLASS{\iExpPolFunctions{i}}{\Oracle{\iNExpTime{(i+j)}}{\SigmaP{c-1}}}$.
Since $\Language{L} \in \Oracle{\iExpSpace{(i-1)}}{\Oracle{\iNExpTime{j}}{\SigmaP{c-1}}}$, there are an $\iExpSpace{(i-1)}$ oracle machine $\Oracle{\Machine{M}}{?}$ and a $\Oracle{\iNExpTime{j}}{\SigmaP{c-1}}$ oracle $\Omega$ such that $\Language{L} = \LanguageOf{\Oracle{\Machine{M}}{\Omega}}$.
Let $\AllQueriesGenericOracle{\Machine{M}}{w}$ be the set of all queries that $\Oracle{\Machine{M}}{?}$, when executing on~$w$, may issue to its (unknown) oracle---this is equivalent to \citeauthor{LadnerL76}'s set of \emph{queries generated by $\Oracle{\Machine{M}}{?}$ on input~$w$}~\cite{LadnerL76}.
It is known that $\AllQueriesGenericOracle{\Machine{M}}{w}$ contains at most \iExponential{i}{}ly-many queries~\cite{LadnerL76,Hemaspaandra1994}---we provide an intuition below.

The rationale is as follows.
Because the query tape of $\Oracle{\Machine{M}}{?}$ is \emph{write\nbdash-only}, neither its content nor the position of its head affects the machine's subsequent transitions.
In fact, to reconstruct the entire computation of $\Oracle{\Machine{M}}{?}$, it is sufficient to record the content and the position of the head of the tapes that can be read, hence the input and work tapes.
Since $\Oracle{\Machine{M}}{?}$ may write \iExponential{(i-1)}{}ly-many symbols on its work tapes, the total number of distinct IDs of this form, i.e., IDs not including the query tape, is \iExponential{i}.
Consider now the ID graph for $w$ and $\Oracle{\Machine{M}}{?}$ whose nodes are the IDs of this form.
Following \citet{LadnerL76}, call \emph{begin nodes} those corresponding to the initial ID and to IDs at which $\Oracle{\Machine{M}}{?}$ receives an oracle answer.
We introduce the term \emph{end nodes} for those corresponding to IDs at which $\Oracle{\Machine{M}}{?}$ issues a query or halts.
Since $\Oracle{\Machine{M}}{?}$ is deterministic, starting from each begin node there is a \emph{unique} partial computation terminating at an end node.
Hence, every query generated during the execution on $w$ is associated with some node of this ID graph.
Since the number of distinct IDs of this form is \iExponential{i}{}, it follows that $\AllQueriesGenericOracle{\Machine{M}}{w}$ contains at most \iExponential{i}{}ly\nbdash-many distinct queries.

Let $p(n)$ hence be a polynomial such that $\SetSize{\AllQueriesGenericOracle{\Machine{M}}{w}} \leq \iExp{i}{p(\StringLength{w})}$.

Let us define the following three predicates, which will be combined to define a predicate that serves as the basis for the Hausdorff predicate characterizing~$\Language{L}$.
Below, $w$ is a string and $z$ is a non\nbdash-negative integer.

\begin{itemize}[noitemsep]
  \item $\PredHSucc[\Machine{M},\Omega]{C}(w,z)$: $\valtrue$ iff there exist (strictly) more than $z$ distinct queries $q \in \AllQueriesGenericOracle{\Machine{M}}{w}$ such that $\Omega(q) = 1$;

  \item $\PredHAccCurr[\Machine{M},\Omega]{A}(w,z)$: $\valtrue$ iff there exists a subset $Y \subseteq \AllQueriesGenericOracle{\Machine{M}}{w}$ with $\SetSize{Y} = z$ and, for all $q \in Y$, $\Omega(q) = 1$, and there exists an \emph{accepting} \ounawarelegalemph computation $\pi$ for $\Oracle{\Machine{M}}{?}(w)$ such that every query receiving in $\pi$ a \yesansw belongs to $Y$, and every query receiving a \noansw in $\pi$ does not belong to $Y$;
      and

  \item $\PredHRejCurr[\Machine{M},\Omega]{A}(w,z)$: $\valtrue$ iff there exists a subset $Y \subseteq \AllQueriesGenericOracle{\Machine{M}}{w}$ with $\SetSize{Y} = z$ and, for all $q \in Y$, $\Omega(q) = 1$, and there exists a \emph{rejecting} \ounawarelegalemph computation $\pi$ for $\Oracle{\Machine{M}}{?}(w)$ such that every query receiving in $\pi$ a \yesansw belongs to $Y$, and every query receiving a \noansw in $\pi$ does not belong to $Y$.
\end{itemize}

The three above predicates can be shown in $\Oracle{\iNExpTime{(i+j)}}{\SigmaP{c-1}}$ \Wrt $\StringLength{w}$ only.
First, observe the following.
We know that the number of distinct queries that $\Oracle{\Machine{M}}{?}$ may generate on input~$w$ is bounded by $\iExp{i}{p(\StringLength{w})}$.
Therefore, predicates $\PredHSucc[\Machine{M},\Omega]{C}(w,z)$ and  $\PredHAccCurr[\Machine{M},\Omega]{A}(w,z)$/$\PredHRejCurr[\Machine{M},\Omega]{A}(w,z)$ are trivially false when $z \geq \iExp{i}{p(\StringLength{w})}$ and $z > \iExp{i}{p(\StringLength{w})}$, respectively.
By this, we need to only consider pairs $\pair{w,z}$ where the value of $z$ is \iExponential{i}{}ly-bounded in the size of~$w$.
Such values of $z$ have a binary representation of \iExponential{(i-1)} size in the size of~$w$.

Let us now consider $\PredHSucc[\Machine{M},\Omega]{C}(w,z)$.
To answer \yeslbl, we can proceed as follows.
First we guess a set $Y$ of $z + 1$ \iExponential{i}{}ly-long strings (since $\Oracle{\Machine{M}}{?}$ is an \iExponential{(i-1)}-space machine, it cannot generate longer queries).
Remember that $z$'s value is \iExponential{i} in the size of~$w$.
Hence, we need to guess \iExponential{i}{}ly-many \iExponential{i}{}ly-long strings.
We also guess enough \ounawarelegal computations $\Pi$ for $\Oracle{\Machine{M}}{?}(w)$ such that each string in $Y$ appears as a query at least once in the guessed computations.
Observe that we do not need to guess more than \iExponential{i}{}ly-many computations, as in the worst case scenario there is one string from $Y$ in each computation, and each guessed computation is \iExponential{i}{}ly-long, as $\Oracle{\Machine{M}}{?}$ is a \iExponential{(i-1)}-space machine and hence its running time is \iExponential{i}{}ly-bounded.
Furthermore, we guess the certificates witnessing that each string in $Y$ is actually accepted by $\Omega$---remember that $\Omega \in \Oracle{\iNExpTime{j}}{\SigmaP{c-1}}$.
These certificates are accepting computations for the $\iNExpTime{j}$ ``part'' of $\Omega$, i.e., we leave out from the guess the part of computation associated with the $\SigmaP{c-1}$ oracle.
These accepting computations for $\Omega$ are \iExponential{(i+j)}{}ly-long, because $\Omega$ may receive \iExponential{i}{}-long queries from $\Machine{M}$.
Hence, these are \iExponential{i}{}ly-many \iExponential{(i+j)}{}ly-long certificates.
Thus, the nondeterministic guess phase can overall be carried out in \iExponential{(i+j)} time.
The check phase proceeds as follows.
We need to check that the guessed computations in $\Pi$ are indeed \ounawarelegal computations for $\Oracle{\Machine{M}}{?}(w)$ (feasible in \iExponential{i} time).
Then, we check that each string in $Y$ appears in at least one of the guessed computations in $\Pi$ (feasible in \iExponential{i} time).
We conclude by checking the validity of the certificates witnessing that the guessed queries are actually accepted by $\Omega$ (feasible in \iExponential{(i+j)} time with the aid of a $\SigmaP{c-1}$ oracle).

To decide $\PredHAccCurr[\Machine{M},\Omega]{A}(w,z)$ (resp., $\PredHRejCurr[\Machine{M},\Omega]{A}(w,z)$), we proceed in a similar way.
We first guess a set $Y$ of $z$ \iExponential{i}{}ly-long strings (feasible in \iExponential{i} time, see above) together with \ounawarelegal computations $\Pi$ for $\Oracle{\Machine{M}}{?}(w)$ such that each string in $Y$ appears at least once as a query in the guessed computations (feasible in \iExponential{i} time; see above).
We moreover guess the certificates witnessing that each string in $Y$ is actually accepted by $\Omega$ (feasible in \iExponential{(i+j)} time; see above).
We also guess an accepting (resp., a rejecting) \ounawarelegal computation $\pi$ for $\Oracle{\Machine{M}}{?}(w)$ (feasible in \iExponential{i} time).
The entire guess phase can be carried out in nondeterministic \iExponential{(i+j)} time.
The check phase proceeds as follows.
We need to check that the guessed computations $\Pi$ are indeed \ounawarelegal computations for $\Oracle{\Machine{M}}{?}(w)$ (feasible in \iExponential{i} time).
Then, we check that each string in $Y$ appears in at least one of the guessed computations (feasible in \iExponential{i} time).
After this, we check the validity of the certificates witnessing that the guessed queries are actually accepted by $\Omega$ (feasible in \iExponential{(i+j)} time with the aid of a $\SigmaP{c-1}$ oracle).
We continue by checking that $\pi$ is actually an accepting (resp., a rejecting) \ounawarelegal computation $\pi$ for $\Oracle{\Machine{M}}{?}(w)$ (feasible in \iExponential{i} time).
We conclude by checking that all queries receiving in $\pi$ a \yesansw belong to $Y$, and all queries receiving in $\pi$ a \noansw do not belong to $Y$ (feasible in \iExponential{i} time).

The rest of the proof follows that of \zcref{theo_nexp_nexp_containment}, by defining the same predicate $\Predicate{B}_{\Machine{M},\Omega}(w,z)$.

\Proofsep

We now prove \[\BoundedHausdCLASS{\iExpPolFunctions{i}}{\Oracle{\iNExpTime{(i+j)}}{\SigmaP{c-1}}} \subseteq \mkern -2mu
\begin{cases}
  \BoundedOracle{\iExpSpace{(k-1)}}{\Oracle{\iNExpTime{(i+j-k)}}{\SigmaP{c-1}}}{\iExpPolFunctions{(i-1)}} \\
  \BoundedParOracle{\iExpSpace{(k-1)}}{\Oracle{\iNExpTime{(i+j-k)}}{\SigmaP{c-1}}}{\iExpPolFunctions{i}}.
\end{cases}\]

Let $\Language{L} \in \BoundedHausdCLASS{\iExpPolFunctions{i}}{\Oracle{\iNExpTime{(i+j)}}{\SigmaP{c-1}{}}}$ be a language.
We prove $\Language{L} \in \ParOracle{\iExpSpace{(k-1)}}{\Oracle{\iNExpTime{(i+j-k)}}{\SigmaP{c-1}}}$ and $\Language{L} \in \BoundedOracle{\iExpSpace{(k-1)}}{\Oracle{\iNExpTime{(i+j-k)}}{\SigmaP{c-1}}}{\iExpPolFunctions{(i-1)}}$.
Since $\Language{L} \in \BoundedHausdCLASS{\iExpPolFunctions{i}}{\Oracle{\iNExpTime{(i+j)}}{\SigmaP{c-1}}}$, there is a $\Oracle{\iNExpTime{(i+j)}}{\SigmaP{c-1}}$ Hausdorff predicate of length $\iExp{i}{p(n)}$, for some polynomial $p(n)$, such that, for every string $w$, we have $w \in \Language{L}$ iff The Hausdorff index $\HausdIndex{w}{\Language{D}}$ is odd.
We claim that a $\iExpSpace{(k-1)}$ oracle machine $\Oracle{\Machine{M}}{?}$ can decide $\Language{L}$ with the aid of a $\Oracle{\iNExpTime{(i+j-k)}}{\SigmaP{c-1}}$ oracle $\Omega$ for (a slight variation of)~$\Language{D}$ via either adaptive or nonadaptive queries.

More specifically, the machine $\Omega$ is designed to receive from $\Oracle{\Machine{M}}{?}$ pairs $\pair{\wt{w},z}$, where $\wt{w}$ is a \iExponential{k}{}ly padded version of the input string~$w$ (remember that $\Oracle{\Machine{M}}{?} \in \iExpSpace{(k-1)}$, and hence it can run for \iExponential{k} time).
By this, $\Omega$ can run for \iExponential{(i+j)} time \Wrt to the size of~$w$.
Upon reception of the query $\pair{\wt{w},z}$, $\Omega$ simply ignores the padding of $\wt{w}$ and decides whether $\pair{w,z} \in \Language{D}$ or not.
Clearly, this can be done by $\Omega$, because $\Language{D} \in \Oracle{\iNExpTime{(i+j)}}{\SigmaP{c-1}}$.
For this reason, below we will regard $\Omega$ as an oracle for $\Language{D}$.

The machine $\Oracle{\Machine{M}}{?}$ can compute the Hausdorff index of $w$ as follows:
either $\Oracle{\Machine{M}}{?}$ issues in parallel the queries $(\wt{w},z)$ for all values of $z$ with $1 \leq z \leq \iExp{i}{p(\StringLength{w})}$, or $\Oracle{\Machine{M}}{?}$ performs a binary search within the domain $[0,\iExp{i}{p(\StringLength{w})}]$.
To implement these approaches, $\Machine{M}$ must write \iExponential{k}{}ly padded queries, compute the value $\iExp{i}{p(\StringLength{w})}$, operate on integers bounded by this quantity, and issue $\iExp{i}{p(\StringLength{w})}$\nbdash-many queries at most.
We now show that these do not pose any difficulty for $\Machine{M}$.

The value $\iExp{i}{p(\StringLength{w})}$ can be computed in $\iExpSpace{(i-1)}$ (see \zcref{sec_maths_complexity}; Iterated exponentials of polynomials), and hence in $\iExpSpace{(k-1)}$, as $i \leq k$.

The queries $\tup{\wt{w}, z}$, where $\wt{w}$ is a \iExponential{k}{}ly padded version of~$w$, can be asked by $\Machine{M}$ as its query tape has no space constraints, and $\Machine{M}$ can run for \iExponential{k} time.
Whether the machine $\Machine{M}$ issues its queries in parallel or issues its queries sequentially to implement the binary search, $\Machine{M}$ needs to keep track of the indices $z$. 
If $\Machine{M}$ asks all its queries in parallel, the machine requires to store on its work tape a counter needing only \iExponential{(k-1)} space, as $z \leq \iExp{i}{p(\StringLength{w})}$ and $i \leq k$, to be represented (see above).
If $\Machine{M}$ uses a binary search, the machine needs also to perform an integer division by~$2$, which in binary can simply be carried out by dropping the least significant bit of the number, and compare numbers represented in \iExponential{(k-1)} space.

Let us now consider the amount of queries issued.
If $\Machine{M}$ asks all its queries in parallel, $\Machine{M}$ needs to issue $\iExp{i}{p(\StringLength{w})}$ queries.
All these queries can actually be submitted by $\Machine{M}$, because $\Machine{M}$ can run for \iExponential{k} time and we are assuming $i \leq k$.
If $\Machine{M}$ performs a binary search, $\Machine{M}$ issues only \iExponential{(i-1)}{}ly-many queries as the search space is $[0,\iExp{i}{p(\StringLength{w})}]$, and hence again it can be done by $\Machine{M}$. 

Once $\Machine{M}$ has computed the Hausdorff index $\HausdIndex{w}{\Language{D}}$ of $w$, the machine $\Machine{M}$ answers \yeslbl iff $\HausdIndex{w}{\Language{D}}$ is odd.
\end{proof}

From the theorem above, we can derive the following corollary.
A result with a flavor of the corollary below, but not as general as the one here stated, was reported in~\cite{Gottlob1995} and was limited only to $\Oracle{\PSpace}{\SigmaP{c}} = \BoundedOracle{\ExpTime}{\SigmaP{c}}{\PolFunctions}$, which here descends from \zcref{theo_general_expspace_nexp_various_flavors,theo_summary_equivalence_intermediate_levels_Hausdorff}.
The proof is in \zcref{sec_detailed_proofs_charting_nexp_oracles}.

\begin{corollary}[store=SummarySNHausdorff]
\label{theo_general_expspace_nexp_various_flavors}
Let $i,j \geq 0$ be integers.
Then, for all integers $c \geq 1$, 
\begin{multline*}
    \Oracle{\iExpSpace{(i-1)}}{\Oracle{\iNExpTime{j}}{\SigmaP{c-1}}} = \BoundedHausdCLASS{\iExpPolFunctions{i}}{\Oracle{\iNExpTime{(i+j)}}{\SigmaP{c-1}}} = {} \\ \BoundedOracle{\iExpSpace{(i-1)}}{\Oracle{\iNExpTime{j}}{\SigmaP{c-1}}}{\iExpPolFunctions{(i-1)}} = \ParOracle{\iExpSpace{(i-1)}}{\Oracle{\iNExpTime{j}}{\SigmaP{c-1}}}.
\end{multline*}
\end{corollary}

\subsection{Looking through the Lens of Hausdorff Classes}
\label{sec_discussion_Hausdorff_perspective}

The above ``charting'' results, obtained via Hausdorff classes, provide us an interesting tool to easily uncover relationships between classes of the iterated exponential hierarchies.
Consider, for example, the class $\Oracle{\iNExpTime{3}}{\iNExpTime{2}}$.
We can now easily answer questions such as: What is its relation with $\Oracle{\iNExpTime{2}}{\iNExpTime{3}}$? And with $\Oracle{\iNExpTime{4}}{\NExpTime}$?

These three classes share the feature that their oracle can run in nondeterministic \iExponential{5} time (as the oracle may receive long queries).
However, the calling machines differ in running time and thus in the strength of their guesses.
Our results confirm this intuition by placing the three classes at distinct intermediate levels of the first step of the \iWEHText{5}.
Indeed, by \zcref{theo_summary_equivalence_intermediate_levels_Hausdorff,theo_summary_nexp_nexp_Hausdorff}, we have:
\begin{alignat*}{2}
  \Oracle{\iNExpTime{3}}{\iNExpTime{2}} &= \BoundedHausdCLASS{\iExpPolFunctions{4}}{\iNExpTime{5}} & &= \BoundedOracle{\iExpTime{5}}{\NPTime}{\iExpPolFunctions{3}} \\
  \Oracle{\iNExpTime{2}}{\iNExpTime{3}} &= \BoundedHausdCLASS{\iExpPolFunctions{3}}{\iNExpTime{5}} & &= \BoundedOracle{\iExpTime{5}}{\NPTime}{\iExpPolFunctions{2}} \\
  \Oracle{\iNExpTime{4}}{\NExpTime} &= \BoundedHausdCLASS{\iExpPolFunctions{5}}{\iNExpTime{5}} & &= \BoundedOracle{\iExpTime{5}}{\NPTime}{\iExpPolFunctions{4}}. 
\end{alignat*}

Our Hausdorff characterization of the intermediate levels of the iterated exponential hierarchies exhibit an interesting common thread.
If $\BoundedHausdCLASS{g(n)}{\ComplexityClass{C}}$ is the class of Hausdorff languages of length $g(n)$, all its languages can be decided by checking whether the Hausdorff index, for a suitable Hausdorff predicate $\Language{D}$, of an input string $w$ is odd.
There are hence at least three kinds of oracle machines deciding the languages in $\BoundedHausdCLASS{g(n)}{\ComplexityClass{C}}$, and relating to three different ways of identifying the Hausdorff index $\HausdIndex{w}{\Language{D}}$ of the input string $w$:
\begin{enumerate}[nosep,label=(\roman*)]
  \item those finding $\HausdIndex{w}{\Language{D}}$ by asking in parallel to the $\Language{D}$ oracle the truth value of all predicates $\Language{D}(w,1),\dots,\linebreak[0]\Language{D}(w,g(\StringLength{w}))$;
  \item those computing $\HausdIndex{w}{\Language{D}}$ via a binary search aided by the $\Language{D}$ oracle deciding whether the predicates $\Language{D}(w,z)$ are true or not for the ``sampled'' values $z$; and,
  \item those guessing $\HausdIndex{w}{\Language{D}}$ and checking via (two parallel queries to) the $\Language{D}$ oracle whether the guessed $\HausdIndex{w}{\Language{D}}$ is actually the maximum index $z$ of the true predicates $\Language{D}(w,z)$.
\end{enumerate}

The existence of these three approaches to decide the languages in $\ComplexityClass{C}$ does not necessarily entail that the three oracle complexity classes characterized by these three kinds of computation are equivalent.
These three oracle complexity classes simply contain $\ComplexityClass{C}$.
Showing these classes equivalent requires to prove that the oracle machines of these classes cannot decide languages outside $\ComplexityClass{C}$.

For example, let us consider the class $\BoundedHausdCLASS{2^{2^{\PolFunctions}}}{\iNExpTime{2}}$ of languages admitting double-exponentially-long Hausdorff reductions to $\iNExpTime{2}$ languages.
By our results, we have
\begin{alignat*}{2}
  \BoundedHausdCLASS{2^{2^{\PolFunctions}}}{\iNExpTime{2}} &= \ParOracle{\iExpTime{2}}{\NPTime} & & \text{(by \zcref{theo_summary_equivalence_intermediate_levels_Hausdorff})} \\
  \BoundedHausdCLASS{2^{2^{\PolFunctions}}}{\iNExpTime{2}} &= \Oracle{\ExpTime}{\NExpTime} & & \text{(by \zcref{theo_summary_generalized_equivalence_intermediate_levels_Hausdorff})} \\
  \BoundedHausdCLASS{2^{2^{\PolFunctions}}}{\iNExpTime{2}} &= \Oracle{\NExpTime}{\NExpTime} \qquad & & \text{(by \zcref{theo_summary_nexp_nexp_Hausdorff})}. 
\end{alignat*}

The equivalences above easily (re)prove the equivalence $\Oracle{\ExpTime}{\NExpTime} = \Oracle{\NExpTime}{\NExpTime}$~\cite{SchoningW88,Hemaspaandra1994}, and establish, moreover, the equivalence of the latter with $\ParOracle{\iExpTime{2}}{\NPTime}$, which is new.
In this example, the classes $\ParOracle{\iExpTime{2}}{\NPTime}$, $\Oracle{\ExpTime}{\NExpTime}$, and $\Oracle{\NExpTime}{\NExpTime}$, are type~(i), (ii), and~(iii), oracle complexity classes, respectively, associated with $\BoundedHausdCLASS{2^{2^{\PolFunctions}}}{\iNExpTime{2}}$.

However, it is not always the case that three oracle complexity classes associated with the three types above are actually equivalent.
Let us consider for example the class $\BoundedHausdCLASS{2^{2^{\PolFunctions}}}{\NExpTime}$.
On the one hand, we have that the class $\Oracle{\ExpTime}{\NPTime}$, equalling $\BoundedHausdCLASS{2^{2^{\PolFunctions}}}{\NExpTime}$ by \zcref{theo_summary_equivalence_intermediate_levels_Hausdorff}, is a type~(ii) oracle complexity class associated with $\BoundedHausdCLASS{2^{2^{\PolFunctions}}}{\NExpTime}$.
On the other hand, $\Oracle{\NExpTime}{\NPTime}$ seems to be a type~(iii) oracle complexity class associated with $\BoundedHausdCLASS{2^{2^{\PolFunctions}}}{\NExpTime}$.
In fact, we have $\BoundedHausdCLASS{2^{2^{\PolFunctions}}}{\NExpTime} \subseteq \Oracle{\NExpTime}{\NPTime}$, because a \NExpTime oracle machine can guess an integer $\hat z$ of double-exponential \emph{value}, and hence of exponential \emph{size}, and ask the oracle to check the validity of the guess.
Nonetheless, $\Oracle{\NExpTime}{\NPTime}$ does not seem to be subset of $\BoundedHausdCLASS{2^{2^{\PolFunctions}}}{\NExpTime}$, and most likely it is not, because $\Oracle{\NExpTime}{\NPTime}$ is the second (main) level of the \WEHStressedText, whereas $\BoundedHausdCLASS{2^{2^{\PolFunctions}}}{\NExpTime}$ characterizes the highest intermediate level in the first step of the \WEHStressedText.

\subsection{The Strong Exponential Hierarchy}
\label{sec_top_seh}

In this section we investigate the \SEHText via Hausdorff classes.
By looking at the \SEHText in this way, we can answer the questions left open by \citet{Hemachandra1989}.

The \SEHText ($\SExpHier$) was introduced at the time in which also the \WEHStressedText ($\WExpHier$) was investigated.
The definition of the $\SExpHier$ levels was ``specular'' to that of $\WExpHier$.
If the levels of $\WExpHier$ were defined as $\Oracle{\NExpTime}{\SigmaP{k}}$, the levels of $\SExpHier$ were defined as $\Oracle{(\SigmaP{k})}{\NExpTime}$.
More precisely, the complexity classes $\SigmaSExp{k}$, $\PiSExp{k}$, and $\DeltaSExp{k}$, constituting the \defin{\SEHText}~\cite{Hemachandra1989} are defined as:
\begin{alignat*}{3}
  \DeltaSExp{0} &= \ExpTime; & \qquad \qquad \DeltaSExp{1} &= \ExpTime;  & \qquad \qquad \DeltaSExp{k} &= \Oracle{\PTime}{\SigmaSExp{k-1}} = \Oracle{(\DeltaP{k-1})}{\NExpTime}, \quad \text{for }k \geq 2; \\
  \SigmaSExp{0} &= \ExpTime; &\SigmaSExp{1} &= \NExpTime;  &  \SigmaSExp{k} &= \Oracle{\NPTime}{\SigmaSExp{k-1}} = \Oracle{(\SigmaP{k-1})}{\NExpTime}, \quad \text{for }k \geq 2; \\
  \PiSExp{0} &= \ExpTime; & \PiSExp{1} &= \CoNExpTime;  &  \PiSExp{k} &= \ComplementPrefix\SigmaSExp{k} = \Oracle{(\PiP{k-1})}{\NExpTime}, \quad \text{for }k \geq 2.
\end{alignat*}
The \SEHText is defined as $\SExpHier = \bigcup_{k \geq 0} \SigmaSExp{k}$.
\Citet{Hemachandra1989} showed that $\PNExp = \NPNExp$, and hence that $\SExpHier$ collapses to its second level, i.e., $\SExpHier = \PNExp$.
The inclusion relationships, which are all currently believed to be strict, between the \SExpHier levels are:
$\ExpTime \subseteq \NExpTime \subseteq \PNExp = \NPNExp = \SExpHier$.

Interestingly, the \SEHText was classically regarded, and hence investigated, as an entity independent from the \WEHStressedText.
The collapse of the \SEHText was greeted as something ``surprising''~\cite{Hemachandra1986,Hartmanis1990,Beigel1991}, and it was wondered whether the \PHText would have collapsed for similar reasons~\cite{Hemachandra1989}.
This was also due to the fact that the collapse of \SExpHier was obtained by \citet{Hemachandra1989} via a census argument, and hence no deeper insights on the reason for the collapse were provided.
\Citet{Hemachandra1989} himself left in his work some open questions on the \SEHText hinting at the need of bringing to light a structural reason for its collapse.

In this respect, Hausdorff classes come to the rescue.
We will see that the reason why $\PNExp$ equals $\NPNExp$ is simply that $\PNExp$ and $\NPNExp$ are type~(ii) and~(iii) oracle complexity classes associated with the Hausdorff class $\BoundedHausdCLASS{2^\PolFunctions}{\NExpTime}$ of $\NExpTime$ Hausdorff languages of exponential length (see \zcref{sec_discussion_Hausdorff_perspective}).
We will also unveil that, although $\SExpHier$ and $\WExpHier$ were treated as two independent objects, $\SExpHier$ is actually a portion of the first step of $\WExpHier$ (see \zcref{fig_iterated_exponentials_meta-hierarchy})---in particular, an equivalent of the class $\Oracle{\ExpTime}{\NPTime}$, which is the highest intermediate level in the first step of $\WExpHier$, does not belong to $\SExpHier$.
For this reason, the collapse of $\SExpHier$ must not be surprising, given that the levels of $\SExpHier$ are precisely characterized by $\NExpTime$ Hausdorff classes of increasing lengths, and hence they are intermediate levels of a step in a hierarchy.

In what follows, we will look first at the \SEHThetaLevel (i.e., the classes $\PNExpLog = \PNExpPar$) and then at the \SEHDeltaLevel (i.e., the classes $\PNExp = \NPNExp$) of \SExpHier.
Similarly to the classes $\ThetaP{c+1} = \LogOracle{\PTime}{\SigmaP{c}}$, which were considered part of $\PolHier$ by \citet{Wagner1990}, we here consider $\PNExpLog$ as a level of $\SExpHier$.
In these subsections, we will characterize the levels of the \SEHText in different ways and we will provide additional results, among which the answers to \citeauthor{Hemachandra1989}'s~\cite{Hemachandra1989} open questions, which are: 
What is a certificate-based characterization for \SExpHier?
Is there an alternating Turing machine characterization for \SExpHier?

\stoptoc

\subsubsection{The \texorpdfstring{$\Theta$}{\textbackslash{}Theta}-level of the \SEHText}
\label{sec_theta_level_hausdorff}
\silentbookmark[3]{The \textbackslash{}Theta-level of the \SEHText}

In this section we deal with the \SEHThetaLevel of $\SExpHier$.
We will see that this hierarchy level is precisely characterized by $\NExpTime$ Hausdorff languages of polynomial length, and hence by many different, but equivalent, oracle classes.

These additional characterizations of the \SEHThetaLevel of \SExpHier stem from the general results on the iterated exponential hierarchies in \zcref{sec_charting_nexp_oracles}.
Below, we will also provide references to the analogue results for $\PolHier$.

From \zcref{theo_general_constant_rounds_parallel_calls_equals_single_round}, any language $\Language{L}$ decidable by a polynomial\nbdash-time oracle Turing machine using a fixed number of rounds of parallel queries to a \NExpTime oracle is also decidable by such a machine using a \emph{single} round of parallel queries to a \NExpTime oracle. 
Combined with \zcref{theo_summary_generalized_equivalence_intermediate_levels_Hausdorff}, we obtain the following.

\begin{corollary}
$\BoundedHausdCLASS{\PolFunctions}{\NExpTime} = \ParBoundedOracle{\PTime}{\NExpTime}{k} = \PNExpPar$, for every fixed integer $k \geq 1$.
\end{corollary}

We can also obtain that $\Oracle{\LogSpace}{\NExpTime}$ equals the \SEHThetaLevel of \SExpHier.
The next result follows from \zcref{theo_general_expspace_nexp_various_flavors}.

\begin{corollary}%
$\BoundedHausdCLASS{\PolFunctions}{\NExpTime} = \Oracle{\LogSpace}{\NExpTime} = \ParOracle{\LogSpace}{\NExpTime} = \LogOracle{\LogSpace}{\NExpTime}$.
\end{corollary}

From \zcref{theo_summary_generalized_equivalence_intermediate_levels_Hausdorff}, the equivalence between $\PNExpPar$, $\PNExpLog$, and $\BoundedHausdCLASS{\PolFunctions}{\NExpTime}$ follows.

\begin{corollary}
$\BoundedHausdCLASS{\PolFunctions}{\NExpTime} = \PNExpPar = \PNExpLog$.
\end{corollary}

From \zcref{theo_summary_generalized_equivalence_intermediate_levels_Hausdorff,theo_general_constant_rounds_parallel_calls_equals_single_round}, it follows that $\DoubleBoundedParOracle{\ExpTime}{\NPTime}{\PolFunctions}{k}$, $\LogOracle{\ExpTime}{\NPTime}$, and $\BoundedHausdCLASS{\PolFunctions}{\NExpTime}$ are equivalent.

\begin{corollary}
$\BoundedHausdCLASS{\PolFunctions}{\NExpTime} = \DoubleBoundedParOracle{\ExpTime}{\NPTime}{\PolFunctions}{k} = \LogOracle{\ExpTime}{\NPTime}$, for every fixed integer $k \geq 1$.
\end{corollary}

From the results above, we have that the \SEHThetaLevel of $\SExpHier$ can be defined in different ways, all equivalent.

\begin{theorem}[store=ThetaLevelPolHausdorff]
\label{theo_theta_level_equals_polHausdorff}
The following classes are equivalent and characterize the~\SEHThetaLevel~of~the~\SEHText:
\begin{enumerate}[nosep,label=(\roman*)]
  \item $\BoundedHausdCLASS{\PolFunctions}{\NExpTime}$
  \item $\TTRedCLASS{\NExpTime}$
  \item $\PNExpLog = \ParBoundedOracle{\PTime}{\NExpTime}{k}$, for every fixed integer $k \geq 1$
  \item $\LogOracle{\ExpTime}{\NPTime} = \DoubleBoundedParOracle{\ExpTime}{\NPTime}{\PolFunctions}{k}$, for every fixed integer $k \geq 1$
  \item $\Oracle{\LogSpace}{\NExpTime} = \ParOracle{\LogSpace}{\NExpTime} = \LogOracle{\LogSpace}{\NExpTime}$
\end{enumerate}
\end{theorem}

These results tell us that the \SEHThetaLevel of $\SExpHier$ is \emph{precisely} characterized by $\NExpTime$ Hausdorff languages of polynomial length. 
The oracle complexity classes $\ParOracle{\LogSpace}{\NExpTime}$, $\PNExpPar$, $\BoundedParOracle{\ExpTime}{\NPTime}{\PolFunctions}$, $\LogOracle{\LogSpace}{\NExpTime}$, $\Oracle{\LogSpace}{\NExpTime}$, $\PNExpLog$, and $\LogOracle{\ExpTime}{\NPTime}$, essentially refers to different approaches to decide languages in $\BoundedHausdCLASS{\PolFunctions}{\NExpTime}$.
We know that every language $\Language{L} \in \BoundedHausdCLASS{\PolFunctions}{\NExpTime}$ is characterized by some $\NExpTime$ Hausdorff predicate $\Language{D}_{\Language{L}}$ of polynomial length.
Hence $\Language{L}$ can be decided by looking at the parity of the Hausdorff indices \Wrt $\Language{D}_{\Language{L}}$ of the strings.
With this in mind, machines in $\ParOracle{\LogSpace}{\NExpTime}$, $\PNExpPar$, and $\BoundedParOracle{\ExpTime}{\NPTime}{\PolFunctions}$, decide $\Language{L}$ by issuing in parallel all the polynomially\nbdash-many queries to the $\NExpTime$ oracle for $\Language{D}_{\Language{L}}$, sufficient to individuate the Hausdorff index (i.e., type~(i) oracle complexity classes associated with $\BoundedHausdCLASS{\PolFunctions}{\NExpTime}$; see \zcref{sec_discussion_Hausdorff_perspective}).
Machines in $\LogOracle{\LogSpace}{\NExpTime}$, $\Oracle{\LogSpace}{\NExpTime}$, $\PNExpLog$, and $\LogOracle{\ExpTime}{\NPTime}$, decide $\Language{L}$ by computing the Hausdorff index via a binary search with the aid of a \NExpTime oracle for $\Language{D}_{\Language{L}}$ (i.e., type~(ii) oracle complexity classes associated with $\BoundedHausdCLASS{\PolFunctions}{\NExpTime}$; see \zcref{sec_discussion_Hausdorff_perspective}).

\subsubsection{The \texorpdfstring{$\Delta$}{\textbackslash{}Delta}-level of the \SEHText}
\label{sec_delta_level_hausdorff}
\silentbookmark[3]{The \textbackslash{}Delta-level of the \SEHText}

In this section, we deal with the \SEHDeltaLevel of $\SExpHier$.
Below, we will refer to the second level of the \BHText over \NExpTime, i.e., $\NExpTime \land \CoNExpTime$, which is the analogue of $\DP = \NPTime \land \CoNPTime$ located in $\PolHier$.
By analogy with $\DP$, we denote $\NExpTime \land \CoNExpTime$ as \DExp;
its complement is $\NExpTime \lor \CoNExpTime$, denoted as $\CoDExp$.

We start by showing that the \SEHDeltaLevel of \SExpHier is precisely characterized by the $\NExpTime$ Hausdorff languages of exponential length.
By this, many different, but equivalent, oracle classes characterize this level of \SExpHier.
Among these classes there are $\PNExp$ and $\NPNExp$, which are then proven equal, and from this follows the collapse of the \SEHText, that was already obtained by \citet{Hemachandra1989} in a different way.

We then focus on the certificate\nbdash-based characterization of $\NPNExp$, left as an open problem by \citet{Hemachandra1987,Hemachandra1989}.
We obtain this characterization via the equivalence $\NPNExp = \BoundedHausdCLASS{2^\PolFunctions}{\NExpTime}$.
\Citet{Hemachandra1987,Hemachandra1989} also asked for an alternating Turing machine characterization of $\NPNExp$.
Our results show that, although such a characterization exists, it is a ``tailored'' one, and most likely we cannot obtain one that is as natural as those for the \PHText~\cite{ChandraKS81} or the \WEHStressedText~\cite{Hemachandra1989,Mocas1996}.

We now start by looking at the relationship between \PNExp and \NPNExp, and at the characterization of the \SEHDeltaLevel of \SExpHier via Hausdorff classes.
By \zcref{theo_summary_generalized_equivalence_intermediate_levels_Hausdorff} we have the following.

\begin{corollary}
  $\BoundedHausdCLASS{2^{\PolFunctions}}{\NExpTime} = \PNExp$.
\end{corollary}

Moreover, by \zcref{theo_summary_nexp_nexp_Hausdorff} we obtain the following.

\begin{corollary}
\label{theo_NP_NEXP_equals_NEXP_2Pol}
  $\BoundedHausdCLASS{2^{\PolFunctions}}{\NExpTime} = \NPNExp = \BoundedOracle{\NPTime}{\DExp}{1} = \BoundedParOracle{\NPTime}{\NExpTime}{2}$.
\end{corollary}

The intuition behind $\BoundedParOracle{\NPTime}{\NExpTime}{2} = \NPNExp$, and \emph{not} $\BoundedOracle{\NPTime}{\NExpTime}{1} = \NPNExp$, as for the specular relation $\BoundedOracle{\NPTime}{\NPTime}{1} = \Oracle{\NPTime}{\NPTime}$ in the \PHText, is as follows.

Let $\Oracle{\Machine{M}}{?}$ be an $\NPTime$ oracle machine and let $\Omega$ be an \NPTime oracle, hence $\LanguageOf{\Oracle{\Machine{M}}{\Omega}} \in \Oracle{\NPTime}{\NPTime}$.
Intuitively, an \NPTime oracle machine $\BoundedOracle{\Machine{N}}{?}{1}$ with an \NPTime oracle $\Gamma$ can decide the same language of $\Oracle{\Machine{M}}{\Omega}$ by simulating $\Oracle{\Machine{M}}{\Omega}(x)$ in five phases:
(1)~$\Machine{N}$ guesses a sequence $\pi$ of IDs; then,
(2)~$\Machine{N}$ checks that $\pi$ is a \ounawarelegalemph computation for $\Oracle{\Machine{M}}{?}(x)$;
(3)~the crucial point now is that $\Machine{N}$ \emph{can} check by itself that all the positive answers to the queries in $\pi$ are actually correct, indeed $\Machine{N}$ can guess and check certificates for the respective $\Machine{M}$'s queries being accepted by $\Omega$; on the other hand,
(4)~$\Machine{N}$ \emph{cannot} check by itself that the negative answers to the queries in $\pi$ are correct, nevertheless $\Machine{N}$ can collect all these queries and with just one query to its \NPTime oracle $\Gamma$ check that them all are rejected by $\Omega$ (\NPTime and \CoNPTime are closed under conjunction); to conclude,
(5)~$\Machine{N}$ answers \yeslbl iff $\pi$ is an accepting computation.

This simulation approach however does not work when $\Oracle{\Machine{M}}{?}$ is a $\NExpTime$ oracle machine, that is $\LanguageOf{\Oracle{\Machine{M}}{\Omega}} \in \NPNExp$.
Indeed, to simulate $\Oracle{\Machine{M}}{\Omega}(x)$, an \NPTime oracle machine $\BoundedOracle{\Machine{N}}{?}{1}$ with an \NExpTime oracle in this case \emph{cannot} perform the phase~(3) above by itself, because now the certificates witnessing the positive answers by $\Omega$ are exponentially long, hence there is the need for $\Machine{N}$ to issue an extra call to its \NExpTime oracle.

Beside the above intuition, $\BoundedOracle{\NPTime}{\NExpTime}{1} \neq \NPNExp$ is supported by the fact that if $\BoundedOracle{\NPTime}{\NExpTime}{1}$ and $\NPNExp$ were equal, then the \BHText over \NExpTime would collapse to its second level, which is not expected to happen~\cite{Dawar1998}.
To achieve this, an intermediate result, consequence of \zcref{theo_NN1_equivalence_BH}, is useful.

\begin{corollary}
\label{theo_np_nexp_1_equals_coDExp}
$\BoundedOracle{\NPTime}{\NExpTime}{1} = \CoDExp$ (and, consequently, $\BoundedOracle{\CoNPTime}{\NExpTime}{1} = \DExp$).
\end{corollary}

Observe that \zcref{theo_np_nexp_1_equals_coDExp} implies these containment relationships:
$\BoundedOracle{\NPTime}{\NExpTime}{1} = \CoDExp \subseteq \BHGeneric{\NExpTime} \subseteq \PNExp = \NPNExp$.
Hence, if it were the case that $\BoundedOracle{\NPTime}{\NExpTime}{1} = \NPNExp$, then the \BHText over \NExpTime, sitting in between $\BoundedOracle{\NPTime}{\NExpTime}{1}$ and $\PNExp$, would be squashed to \CoDExp.
We can therefore state the following.

\begin{theorem}[store=NPNexpSingleCallCollapseBH]
\label{theo_np_nexp_1_not_equal_np_nexp}
If $\BoundedOracle{\NPTime}{\NExpTime}{1} = \NPNExp$, then the \BHText over \NExpTime collapses to its second level.
\end{theorem}

Again from \zcref{theo_summary_generalized_equivalence_intermediate_levels_Hausdorff}, combined with \zcref{theo_general_constant_rounds_parallel_calls_equals_single_round}, the next result follows.

\begin{corollary}
  $\BoundedHausdCLASS{2^{\PolFunctions}}{\NExpTime} = \PolOracle{\ExpTime}{\NPTime} = \ParBoundedOracle{\ExpTime}{\NPTime}{k}$, for every fixed integer $k \geq 1$.
\end{corollary}

By \zcref{theo_general_expspace_nexp_various_flavors}, we obtain the following;
remember that we adopt for space-bounded oracles the deterministic query model~\cite{RuzzoST84}, coinciding with the unrestricted query model for deterministic space classes~\cite{LadnerL76}.

\begin{corollary}
  $\BoundedHausdCLASS{2^{\PolFunctions}}{\NExpTime} = \Oracle{\PSpace}{\NPTime} = \ParOracle{\PSpace}{\NPTime} = \PolOracle{\PSpace}{\NPTime}$.
\end{corollary}

The reader should be aware of 
the result $\Oracle{\PTime}{\textnormal{NE}} = \Oracle{\PSpace}{\textnormal{NE}}$ (implying $\PNExp = \Oracle{\PSpace}{\NExpTime}$) reported in Theorem~24 and Corollary~25 of~\cite{AllenderKRR2011}.
Their result assume the bounded query model for space-bounded oracle machines, and \emph{not} the deterministic query model  
as it is instead done in \cite{Hemaspaandra1994,Gottlob1995} and here. 

From the results above, we have that the \SEHDeltaLevel of $\SExpHier$ can be defined in different ways, all equivalent.

\begin{theorem}[store=DeltaLevelExpHausdorff]
\label{theo_delta_level_equals_expHausdorff}
The following classes are equivalent and characterize the \SEHDeltaLevel of the \SEHText:
\begin{enumerate}[nosep,label=(\roman*)]
  \item $\BoundedHausdCLASS{2^{\PolFunctions}}{\NExpTime}$
  \item $\PNExp$
  \item $\NPNExp = \BoundedOracle{\NPTime}{\DExp}{1} = \BoundedParOracle{\NPTime}{\NExpTime}{2}$
  \item $\PolOracle{\ExpTime}{\NPTime} = \ParBoundedOracle{\ExpTime}{\NPTime}{k}$, for every fixed integer $k \geq 1$
  \item $\Oracle{\PSpace}{\NPTime} = \ParOracle{\PSpace}{\NPTime} = \PolOracle{\PSpace}{\NPTime}$
\end{enumerate}
\end{theorem}

As a corollary of the previous \zcref*[typeset=name,nocap]{theo_delta_level_equals_expHausdorff}, we obtain a different proof of \citeauthor{Hemachandra1989}'s $\PNExp = \NPNExp$ result \cite{Hemachandra1989}, and the consequent collapse of the \SEHText.

\begin{corollary}[\cite{SchoningW88,Hemachandra1989,Beigel1991,Gottlob1995,AllenderKRR2011}]
\label{theo_p_nexp_equals_np_nexp}
$\PNExp = \NPNExp$.
\end{corollary}

An interesting aspect of our rederivation of \citeauthor{Hemachandra1989}'s result is the structural link it reveals between $\PNExp$ and $\NPNExp$.
Both classes are exactly characterized by $\NExpTime$ Hausdorff languages of exponential length.
Deciding such languages reduces to checking whether the Hausdorff index, for a suitable Hausdorff predicate, of the input string is odd.
The classes $\PNExp$ and $\NPNExp$ simply correspond to two broad families of oracle machines that decide languages in the \SEHDeltaLevel of $\SExpHier$ via two different strategies.
Machines from $\PNExp$ are deterministic polynomial-time oracle machines that compute the Hausdorff index via a polynomial-time binary search, whereas machine from $\NPNExp$ are nondeterministic polynomial-time oracle machines that guess the Hausdorff index and then check it.
Thus, $\PNExp$ is a type~(ii) oracle class associated with $\BoundedHausdCLASS{2^\PolFunctions}{\NExpTime}$, while $\NPNExp$ is a type~(iii) oracle class associated with $\BoundedHausdCLASS{2^\PolFunctions}{\NExpTime}$ (see \zcref{sec_discussion_Hausdorff_perspective} for more on oracle classes ``types'').

Some of the characterizations of the \SEHDeltaLevel of \SExpHier in \zcref{theo_delta_level_equals_expHausdorff} were already known, e.g.:
\begin{itemize}[nosep,label=--]
  \item $\PNExp = \ParOracle{\ExpTime}{\NPTime}$~\cite[Theorem~4.10, Part~2]{Hemachandra1989}, \cite[Lemma~3.1]{Mocas1996}, and~\cite[Corollary~25]{AllenderKRR2011};
  \item $\PNExp = \BoundedOracle{\ExpTime}{\NPTime}{\mkern-2mu\PolFunctions}$~\cite[Theorem~8, Parts~3 \&~4]{Gottlob1995} and~\cite[Lemma~3.1]{Mocas1996}; and
  \item $\PNExp = \Oracle{\PSpace}{\NPTime}$~\cite[Corollary~2.2]{Hemaspaandra1994} and~\cite[Theorem~8, Parts~3 \&~4]{Gottlob1995}.
\end{itemize}
These results were often shown in the literature via tailored techniques, whereas we obtain them in a uniform and rather simple manner, via the notion of Hausdorff classes. 
In fact, $\ParOracle{\ExpTime}{\NPTime}$ and $\ParOracle{\PSpace}{\NPTime}$ are type~(i) oracle classes associated with $\BoundedHausdCLASS{2^{\PolFunctions}}{\NExpTime}$;
$\PNExp$, $\PolOracle{\ExpTime}{\NPTime}$, $\Oracle{\PSpace}{\NPTime}$, and $\PolOracle{\PSpace}{\NPTime}$, are type~(ii) oracle classes associated with $\BoundedHausdCLASS{2^{\PolFunctions}}{\NExpTime}$;
and $\NPNExp$ and is a type~(iii) oracle class associated with $\BoundedHausdCLASS{2^{\PolFunctions}}{\NExpTime}$.

An additional interesting equivalence result is
\begin{itemize}[nosep,label=--]
  \item $\PNExp = \Oracle[\ensuremath{\langle\mkern-2mu \PolFunctions \rangle}]{\NExpTime}{\NExpTime}$~\cite[Corollary~4]{SchoningW88} and~\cite[Theorem~24]{AllenderKRR2011},
\end{itemize}
where $\Oracle[\ensuremath{\langle\mkern-2mu \PolFunctions \rangle}]{\NExpTime}{\NExpTime}$ is the class of languages decided by $\NExpTime$ oracle machines querying an $\NExpTime$ oracle with queries whose size is polynomially bounded.
This result also admits a Hausdorff class interpretation:
machines from this oracle class can guess, and subsequently check, the Hausdorff index, making $\Oracle[\ensuremath{\langle\mkern-2mu \PolFunctions \rangle}]{\NExpTime}{\NExpTime}$ a type~(iii) oracle class.
Since both the oracle machine and the oracle are in $\NExpTime$, the query size of the caller must be restricted;
otherwise, the oracle could run for double\nbdash-exponential time in the input length.
This idea of bounding query size readily generalizes to the statements of \zcref{theo_nexp_nexp_containment,theo_summary_nexp_nexp_Hausdorff}.

\medskip

We have been showing that the \SEHText levels are exactly characterized by $\NExpTime$ Hausdorff classes of increasing lengths.
At the base level, the class $\NExpTime$ itself is the class of $\NExpTime$ Hausdorff languages of length~$1$.
Above \NExpTime there is the \BHText over $\NExpTime$, which is the class of Hausdorff languages of constant length (greater than~$1$).
Above this, there is the \SEHThetaLevel of $\SExpHier$, where we have $\PNExpPar = \PNExpLog$, characterized by $\NExpTime$ Hausdorff languages of polynomial length.
And then, there is the \SEHDeltaLevel of $\SExpHier$, where $\PNExp = \NPNExp$, characterized by $\NExpTime$ Hausdorff languages of exponential length.
For this reason, we have that the levels of the \SEHText are actually (some of) the intermediate levels of the first step of the \WEHStressedText (see \zcref{fig_iterated_exponentials_meta-hierarchy})---an equivalent of $\Oracle{\ExpTime}{\NPTime}$ does not belong to \SExpHier.
Therefore, these two hierarchies are not unrelated entities.

Understanding whether $\PNExpPar$ is strictly contained in $\PNExp$ is therefore tantamount to understanding whether exponentially\nbdash-long Hausdorff reductions to $\NExpTime$ Hausdorff predicates are strictly more powerful than polynomially\nbdash-long ones.
Given that these two classes are equivalent to two intermediate levels in the first step of the \WEHStressedText, we conjecture that this is the case, and hence that $\PNExpPar \neq \PNExp$, like it was conjectured that $\ParOracle{\PTime}{\NPTime} \neq \Oracle{\PTime}{\NPTime}$~(see, e.g.,~\cite{Krentel1988,Wagner1990,Beigel1991}).

The following downward separation result however implies that proving $\PNExpPar \neq \PNExp$ will unlikely be easy, as this would cause the separation of \PTime from \NPTime.
Indeed, if $\PTime$ were equal to $\NPTime$, then $\ExpTime$ would equal $\NExpTime$~\cite{Hartmanis1985}.
From this, it would follow a collapse of the \SEHText to \ExpTime~\cite{Hemachandra1989} involving \PNExpPar as well.
Below, we provide also a different proof.

\begin{theorem}
\label{theo_downward_separation}
If $\PNExpPar \neq \PNExp$, then $\PTime \neq \NPTime$.
\end{theorem}

\begin{proof}
Prove the contrapositive: $\PTime = \NPTime \Rightarrow \PNExpPar = \PNExp$.
By \zcref{theo_NP_NEXP_equals_NEXP_2Pol,theo_p_nexp_equals_np_nexp}, $\ParOracle{\NPTime}{\NExpTime} = \NPNExp = \PNExp$.
Clearly, $\PNExpPar \subseteq \ParOracle{\NPTime}{\NExpTime}$.
If $\PTime = \NPTime$ were the case, we would have $\PNExpPar = \ParOracle{\NPTime}{\NExpTime} = \PNExp$. 
\end{proof}

We now give a certificate\nbdash-based characterization of \NPNExp, closing an open problem posed by \citet{Hemachandra1987,Hemachandra1989}.
The characterization follows directly from the equivalence $\NPNExp = \BoundedHausdCLASS{2^\PolFunctions}{\NExpTime}$ established in \zcref{theo_NP_NEXP_equals_NEXP_2Pol}.
The following characterization is also implicit in the proof of the \PNExph{}ness of the Extended Tiling Problem in~\cite{EiterLP2016,EiterTechRep2016}, though that argument relies on the equality $\NPNExp = \PNExp$.

\begin{theorem}[store=NPNexpCertificates]
\label{theo_np_nexp_certificate_characterization}
Let $\Language{L}$ be a language.
Then, $\Language{L} \in \NPNExp$ if and only if there exist
a polynomial $p(n)$ and polynomial-time predicates $R$ and $S$ such that, for every string $w$,
\begin{equation}\label{eq_np_nexp_certificate_characterization}
w \in \Language{L} \Leftrightarrow (\exists t \in \alphabet^{\leq p(\StringLength{w})}) \PrefixMatrixSeparator
((\exists u \in \alphabet^{\leq 2^{p(\StringLength{w})}}) \PrefixMatrixSeparator R(w,t,u) = 1 \land
(\forall v \in \alphabet^{\leq 2^{p(\StringLength{w})}}) \PrefixMatrixSeparator S(w,t,v) = 1).
\end{equation}
\end{theorem}

\begin{proof}
\ProofRightarrowItem
Let $\Language{L}$ be a $\NPNExp$ language.
By \zcref{theo_NP_NEXP_equals_NEXP_2Pol}, $\Language{L} \in \BoundedHausdCLASS{2^\PolFunctions}{\NExpTime}$.
Therefore, there is a $\NExpTime$ Hausdorff predicate $\Language{D}$ of length $2^{q(n)} - 1$, for some polynomial $q(n)$, such that, for every string $w$, $w \in \Language{L}$ iff $\HausdIndex{w}{\Language{D}}$ is odd.
Let us define $\Language{D}'(w,z) = \compl{\Language{D}}(w,z+1)$.
For every string~$w$, we hence have that
\[w \in \Language{L} \Leftrightarrow (\exists z \in \NaturalsDomain) \PrefixMatrixSeparator (1 \leq z \leq 2^{q(\StringLength{w})} - 1 \land \text{ $z$ is odd } \land \Language{D}(w,z) = 1 \land \Language{D}'(w,z) = 1).\]

Observe that the fact that $\Language{D}$ and $\Language{D}'$ are decidable in exponential time with respect to $\StringLength{w}$ only plays no role here.
Indeed, any value of $z$ with $z \leq 2^{q(\StringLength{w})} - 1$ can be represented in canonical binary form with $q(\StringLength{w})$ bits at most.
Therefore, since $\Language{D} \in \NExpTime$ and $\Language{D}' \in \CoNExpTime$, there exist polynomials $q_{\Language{D}}(n)$ and $q_{\Language{D}'}(n)$ and polynomial\nbdash-time predicates $T_{\Language{D}}$ and $T_{\Language{D}'}$ such that, for every string~$w$,
\begin{align*}
  \tup{w,z} \in \Language{D} & \Leftrightarrow (\exists u \in \alphabet^{\leq 2^{q_{\Language{D}}(\StringLength{w})}}) \PrefixMatrixSeparator T_{\Language{D}}(w,z,u) = 1 \\
  \tup{w,z} \in \Language{D}' & \Leftrightarrow (\forall u \in \alphabet^{\leq 2^{q_{\Language{D}'}(\StringLength{w})}}) \PrefixMatrixSeparator T_{\Language{D}'}(w,z,u) = 1.
\end{align*}

By combining the expressions above, we obtain:
\begin{align*}
w \in \Language{L}
{} \Leftrightarrow
    (\exists z \in \alphabet^{\leq q(\StringLength{w})}) \PrefixMatrixSeparator (
        & 1 \leq z \leq 2^{q(\StringLength{w})} - 1 \land \text{ $z$ is odd } \land {}\\
        &(\exists u \in \alphabet^{\leq 2^{q_{\Language{D}}(\StringLength{w})}}) \PrefixMatrixSeparator (
            T_{\Language{D}}(w, z, u) = 1
        ) \land {}\\
        &(\forall v \in \alphabet^{\leq 2^{q_{\Language{D}'}(\StringLength{w})}}) \PrefixMatrixSeparator (
            T_{\Language{D}'}(w, z, v) = 1
        )
    ).
\end{align*}

A rewriting of the above expression allows us to more explicitly map it onto~Eq.~\eqref{eq_np_nexp_certificate_characterization}.
Let $p(n)$ be a polynomial such that $p(n) \geq q(n), q_{\Language{D}}(n), q_{\Language{D}'}(n)$.
We obtain:%
\footnote{%
In the rewriting, we use these equivalences:
For two polynomials $q(n)$ and $p(n)$ such that $q(n) \leq p(n)$ it holds that:
\begin{itemize}[label=--,nosep]
  \item $(\exists u \in \alphabet^{\leq 2{q(\StringLength{w})}}) \PrefixMatrixSeparator T(w,u) \Leftrightarrow (\exists u \in \alphabet^{\leq 2^{p(\StringLength{w})}}) \PrefixMatrixSeparator (\StringLength{u} \leq 2^{q(\StringLength{w})} \land T(w,u))$;
  \item $(\forall v \in \alphabet^{\leq 2^{q(\StringLength{w})}}) \PrefixMatrixSeparator T(w,v) \Leftrightarrow (\forall v \in \alphabet^{\leq 2^{p(\StringLength{w})}}) \PrefixMatrixSeparator (\StringLength{v} \leq 2^{q(\StringLength{w})} \rightarrow T(w,v))$.
\end{itemize}}%
\begin{align*}
w \in \Language{L}
    &\begin{aligned}[t]
        {} \Leftrightarrow
        (\exists z \in \alphabet^{\leq p(\StringLength{w})}) \PrefixMatrixSeparator (
            &1 \leq z \leq 2^{q(\StringLength{w})} - 1 \land \text{ $z$ is odd } \land {}\\
            &(\exists u \in \alphabet^{\leq 2^{p(\StringLength{w})}}) \PrefixMatrixSeparator (
                \StringLength{u} \leq 2^{q_{\Language{D}}(\StringLength{w})} \land T_{\Language{D}}(w, z, u) = 1
            ) \land {}\\
            &(\forall v \in \alphabet^{\leq 2^{p(\StringLength{w})}}) \PrefixMatrixSeparator (
                \StringLength{v} \leq 2^{q_{\Language{D}'}(\StringLength{w})} \rightarrow T_{\Language{D}'}(w, z, v) = 1
            )
        )
     \end{aligned} \displaybreak[0] \\
    &\begin{aligned}[t]
        {} \Leftrightarrow
        (\exists z \in \alphabet^{\leq p(\StringLength{x})}) \PrefixMatrixSeparator
        (
            &(\exists u \in \alphabet^{\leq 2^{p(\StringLength{w})}}) \PrefixMatrixSeparator (
                    1 \leq z \leq 2^{q(\StringLength{w})} - 1 \land \text{ $z$ is odd } \land {} \\
            &\phantom{(\exists u \in \alphabet^{\leq 2^{p(\StringLength{w})}}) \PrefixMatrixSeparator (}
                    \underbracket[.5pt]{\qquad\quad \StringLength{u} \leq 2^{q_{\Language{D}}(\StringLength{w})} \land T_{\Language{D}}(w, z, u) = 1}_{\Leftrightarrow R(w,z,u) = 1}
            ) {} \land {} \\
            &(\forall v \in \alphabet^{\leq 2^{p(\StringLength{w})}}) \PrefixMatrixSeparator
            (
                \underbracket[.5pt]{\StringLength{v} \leq 2^{q_{\Language{D}'}(\StringLength{w})} \rightarrow T_{\Language{D}'}(w, z, v) = 1}_{\Leftrightarrow S(w,z,v) = 1}
            )
        ).
     \end{aligned}
\end{align*}

To conclude, we need to show that $R(w,z,u)$ and $S(w,z,u)$ are deterministic polynomial-time predicates.
Let us focus on $R(w,z,u)$.
Checking whether $z$ is odd is clearly feasible in polynomial time, and checking whether $T_{\Language{D}}(w, z, u) = 1$ can be done in polynomial time, because $T_{\Language{D}}$ is assumed to be polynomial.
We are left to show that checking $z \leq 2^{q(\StringLength{w})} - 1$ and $\StringLength{u} \leq 2^{q_{\Language{D}}(\StringLength{w})}$ can be carried out in polynomial time.

Consider first checking whether $z \leq 2^{q(\StringLength{w})} - 1$.
Notice that $z \leq 2^{q(\StringLength{w})} - 1$ holds whenever $\StringLength{z} \leq q(\StringLength{w})$ (irrespective of whether the string $z$ is actually an integer in canonical form or not, because it is not possible to represent a number greater than $2^{q(\StringLength{w})} - 1$ with ${q(\StringLength{w})}$ bits only).
By this, we simply need to compute the value $q(\StringLength{w})$, which can be carried out in polynomial time (see \zcref{sec_maths_complexity}; Polynomials), and then check that $\StringLength{z}$ does not exceed that value (feasible in polynomial time as well).

Similarly, to check that $\StringLength{u} \leq 2^{q_{\Language{D}}(\StringLength{w})}$ we need to compute $2^{q_{\Language{D}}(\StringLength{w})}$ and compare the length of $\StringLength{u}$ with that value.
This can be done in polynomial time because $2^{q_{\Language{D}}(\StringLength{w})}$ can be computed in polynomial time (see \zcref{sec_maths_complexity}; Iterated exponentials of polynomials).

Similarly, we can show that $S(w,z,u)$ is a deterministic polynomial-time predicate.

\medbreak

\ProofLeftarrowItem
Assume that the language $\Language{L}$ satisfies the certificate-based characterization of~Eq.~\eqref{eq_np_nexp_certificate_characterization}.
We show that $\Language{L}$ can be decided by an $\NPTime$ oracle Turing machine $\Oracle{\Machine{M}}{?}$ querying a $\NExpTime$ oracle.

To decide $w \in \Language{L}$, the machine $\Machine{M}$ simply guesses the polynomially-long string $t$, and then checks, via two oracle calls, that $((\exists u \in \alphabet^{\leq 2^{p(\StringLength{x})}}) \PrefixMatrixSeparator R(w,t,u) = 1)$ holds, which is a $\NExpTime$ task, and that $((\forall v \in \alphabet^{\leq 2^{p(\StringLength{x})}}) \PrefixMatrixSeparator S(w,t,v) = 1)$ holds, which is a $\CoNExpTime$ task.
\end{proof}

Besides asking for a certificate\nbdash-based characterization of $\NPNExp$, \citet{Hemachandra1989} also posed the open problem of finding an alternating Turing machine characterization of $\NPNExp$.
\zcref[S]{theo_NP_NEXP_equals_NEXP_2Pol,theo_np_nexp_1_not_equal_np_nexp}, showing that 
$\BoundedOracle{\NPTime}{\DExp}{1} = \NPNExp$ and suggesting that 
$\BoundedOracle{\NPTime}{\NExpTime}{1} \neq \NPNExp$, give us evidence that such a characterization, although possible, may be less ``natural'' than those for the main levels of $\PolHier$ \cite{ChandraKS81} or $\WExpHier$ \cite{Hemachandra1989,Mocas1996}.
Indeed, accommodating the $\DExp = (\NExpTime \land \CoNExpTime)$ oracle access may require a tailored alternating computation.
By the certificate\nbdash-based characterization of $\NPNExp$ in \zcref{theo_np_nexp_certificate_characterization}, an alternating Turing machine could begin with an existential polynomial\nbdash-time computation guessing the certificate ``$t$''.
Since $\BoundedOracle{\NPTime}{\NExpTime}{1} = \NPNExp$ is unlikely to hold, the machine cannot simply switch to a universal phase for verification.
Instead, using $\BoundedOracle{\NPTime}{\DExp}{1} = \NPNExp$, the machine may enter a universal state with two branches:
one initiating an existential exponential\nbdash-time computation verifying the 
$\NExpTime$ condition, and the other remaining universal while performing an exponential\nbdash-time computation verifying the $\CoNExpTime$ condition.

We stress here that $\NPNExp$, by our Hausdorff characterization, is an intermediate level of \WExpHier, and not a main level.
Main levels of the iterated exponential hierarchies can easily be characterized via alternating machines (see \zcref{sec_def_iterated_exp_hierarchy}), whereas the most natural characterization for intermediate levels we have shown to be Hausdorff classes of increasing lengths (see \zcref{sec_charting_top}).

\resumetoc

\section{Hard Problems}
\label{sec_hard_problems}

In this section, via the Hausdorff characterization, we obtain hard problems for the intermediate levels of the iterated exponential hierarchies.
In the first subsection, for all the iterated exponential hierarchies, we provide:
\begin{itemize}[nosep,label=--]
  \item for all steps, canonical complete problems for the first intermediate level; and
  \item for the first steps, canonical complete problems for all the intermediate levels, but the first and the last.
\end{itemize}
In the second subsection, we consider hard problems over Quantified Boolean Second Order formulas.
These enable us to exhibit complete problems for all the intermediate levels of all the steps of the \WEHStressedText.
In the last subsection, we obtain matching lower bounds for $\PNExpLog$ and $\PNExp$ problems whose hardness was left open in the literature due to the lack of known suitable complete problems.

\subsection{Hard Problems for Some of the Intermediate Levels}
\label{sec_canonical_hard_problems}

In this section, we provide canonical complete problems for some of the intermediate levels of the iterated exponential hierarchies.
We first look at problems complete for the first intermediate level in every step. 
Then, we define a family of problems complete for all the intermediate levels, but the first and the last, in the first steps. 

The following is a problem complete for $\BoundedParOracle{\iExpTime{i}}{\SigmaP{c}}{\PolFunctions} = \BoundedOracle{\iExpTime{i}}{\SigmaP{c}}{\LogFunctions} = \BoundedHausdCLASS{\iExpPolFunctions{0}}{\Oracle{\iNExpTime{i}}{\SigmaP{c-1}}}$.
Similar problems restricted to \PolHier were provided by \citet{Wagner1987,Wagner1990}.
The membership of this problem is obtained via a simple counting algorithm invoking an oracle, and the hardness is obtained via a ``template'' reduction pivoting on the notion of Hausdorff language. 
Details of the proof are deferred to \zcref{sec_detailed_proofs_canonical_hard_problems}.

\begin{theorem}[store=HardnessOddityGeneral]
\label{theo_hardness_oddity_general}
Let $i \geq 0$ and $c \geq 1$ be integers, and let $\Language{A}$ be a language complete for $\Oracle{\iNExpTime{i}}{\SigmaP{c-1}}$ (resp., $\ComplementPrefixKerned\Oracle{\iNExpTime{i}}{\SigmaP{c-1}}$).
Then, for a tuple $\StringTup{w} = \tup{w_1,\dots,w_n}$ of strings, deciding whether the number of \yesinsts of $\Language{A}$ in~$\StringTup{w}$ is odd is complete for $\BoundedOracle{\iExpTime{i}}{\SigmaP{c}}{\LogFunctions} = \BoundedParOracle{\iExpTime{i}}{\SigmaP{c}}{\PolFunctions} = \BoundedHausdCLASS{\PolFunctions}{\Oracle{\iNExpTime{i}}{\SigmaP{c-1}}}$.
Hardness holds even if the tuples $\StringTup{w} = \tup{w_1,\dots,w_n}$ are such that ${\Language{A}}(w_1) \geq \dots \geq {\Language{A}}(w_n)$, and $n$ is an even integer.
\end{theorem}

Inspired by a problem in \cite{LukasiewiczM17}, from that in \zcref{theo_hardness_oddity_general} we derive the counting-and-comparison problem in \zcref{theo_hardness_count_comp_general_2L_2S}, complete for
$\BoundedOracle{\iExpTime{i}}{\SigmaP{c}}{\LogFunctions}
= \BoundedParOracle{\iExpTime{i}}{\SigmaP{c}}{\PolFunctions}
= \BoundedHausdCLASS{\iExpPolFunctions{0}}{\Oracle{\iNExpTime{i}}{\SigmaP{c-1}}}$.
From it, we derive the related complete problem in \zcref{theo_hardness_count_comp_general_2L_1S}.
Membership follows by simple counting algorithms, while hardness follows through sequences of reductions;
details are provided in \zcref{sec_detailed_proofs_canonical_hard_problems}.

\begin{theorem}[store=HardnessCountCompGeneral]
\label{theo_hardness_count_comp_general_2L_2S}
Let $i \geq 0$ and $c \geq 1$ be integers, and let $\Language{A}$ and $\Language{B}$ be two languages complete for $\Oracle{\iNExpTime{i}}{\SigmaP{c-1}}$ (resp., $\ComplementPrefixKerned\Oracle{\iNExpTime{i}}{\SigmaP{c-1}}$).
Then, for two tuples $\StringTup{w} = \tup{w_1,\dots,w_n}$ and $\StringTup{v} = \tup{v_1,\dots,v_m}$ of strings,
deciding whether the number of \yesinsts of $\Language{A}$ in $\StringTup{w}$ is greater than the number of \yesinsts of $\Language{B}$ in $\StringTup{v}$ is complete for $\BoundedOracle{\iExpTime{i}}{\SigmaP{c}}{\LogFunctions} = \BoundedParOracle{\iExpTime{i}}{\SigmaP{c}}{\PolFunctions} = \BoundedHausdCLASS{\PolFunctions}{\Oracle{\iNExpTime{i}}{\SigmaP{c-1}}}$.
Hardness holds even if $\Language{A} = \Language{B}$, the tuples $\StringTup{w}= \tup{w_1,\dots,w_n}$ and $\StringTup{v} = \tup{v_1,\dots,v_m}$ are such that $\Language{A}(w_1) \geq \dots \geq \Language{A}(w_n)$ and $\Language{B}(v_1) \geq \dots \geq \Language{B}(v_m)$, and $n = m$.
\end{theorem}

\begin{theorem}[store=HardnessCountCompTwoLangOneSet]
\label{theo_hardness_count_comp_general_2L_1S}
Let $i \geq 0$ and $c \geq 1$ be integers, and let $\Language{A}$ and $\Language{B}$ be two (distinct) languages complete for $\Oracle{\iNExpTime{i}}{\SigmaP{c-1}}$ (resp., $\ComplementPrefixKerned\Oracle{\iNExpTime{i}}{\SigmaP{c-1}}$).
Then, for a tuple $\StringTup{w} = \tup{w_1,\dots,w_n}$ of strings, deciding whether $\StringTup{w}$ contains more \yesinsts of $\Language{A}$ than of $\Language{B}$ is complete for $\BoundedOracle{\iExpTime{i}}{\SigmaP{c}}{\LogFunctions} = \BoundedParOracle{\iExpTime{i}}{\SigmaP{c}}{\PolFunctions} = \BoundedHausdCLASS{\PolFunctions}{\Oracle{\iNExpTime{i}}{\SigmaP{c-1}}}$.
\end{theorem}

To conclude, for integers $i \geq 0$ and $j \geq 1$, we define a family of problems complete for $\Oracle{\iNExpTime{i}}{\iNExpTime{j}} = \BoundedHausdCLASS{\iExpPolFunctions{i+1}}{\iNExpTime{(i+j)}} = \BoundedParOracle{\iExpTime{(i+j)}}{\NPTime}{\iExpPolFunctions{i+1}} = \BoundedOracle{\iExpTime{(i+j)}}{\NPTime}{\iExpPolFunctions{i}}$ (see \zcref{theo_summary_equivalence_intermediate_levels_Hausdorff,theo_summary_nexp_nexp_Hausdorff});
observe that the integers $i$ and $j$ mentioned in the statement of the \zcref*[typeset=name,nocap]{theo_canonical_np_nexp_hard_problem} below are fixed, i.e., they are \emph{not} part of the input.

\begin{theorem}[store=NPNexpHardnessTwoMachines]
\label{theo_canonical_np_nexp_hard_problem}
Let $i \geq 0$ and $j \geq 1$ be fixed integers, let $\Machine{M}_\alpha$ and $\Machine{M}_\beta$ be (encodings of) two Turing machines, and let $r,t$ be two integers represented in unary, with $\StringLength{t}$ polynomially\nbdash-bounded in~$\StringLength{r}$.
Then, deciding whether there exists a string $v$ %
of length at most $\iExp{i}{r}$ such that $\Machine{M}_\alpha(v)$ has an accepting computation of at most $\iExp{i+j}{t}$ steps, while $\Machine{M}_\beta(v)$ has no accepting computation of at most $\iExp{i+j}{t}$ steps, is complete for
$\Oracle{\iNExpTime{i}}{\iNExpTime{j}} = \BoundedHausdCLASS{\iExpPolFunctions{i+1}}{\iNExpTime{(i+j)}} = \BoundedParOracle{\iExpTime{(i+j)}}{\NPTime}{\iExpPolFunctions{i+1}} = \BoundedOracle{\iExpTime{(i+j)}}{\NPTime}{\iExpPolFunctions{i}}$.
Hardness holds even if the length of $v$ must be exactly~$\iExp{i}{r}$.
\end{theorem}

\subsection{Hard Problems over QBSFs for the Intermediate Levels of EH}
\label{sec_qbsf_hard_problems}

In this section, we prove complete for the intermediate levels of the \WEHStressedText some problems over Quantified Boolean Second-order Formulas (QBSFs).
We first provide preliminaries on \QBSFs, then we introduce the problems, and we conclude by proving them complete for the respective complexity classes.
The proofs will rely on the notions of Hausdorff languages and classes introduced in this paper.

\stoptoc

\subsubsection{Preliminaries on QBSFs}
\silentbookmark[3]{Preliminaries on QBSFs}

Quantified Boolean Second-order Formulas~(QBSFs)~\cite{Luck2016-techrep}, or Second-Order Quantified Boolean Formulas~(SOQBFs)~\cite{Jiang2023}, are defined starting from usual quantified Boolean propositional formulas by adding the possibility to quantify also Boolean function variables.
We here focus on \emph{prenex simple} formulas, which are formulas whose quantifiers appear at the beginning, i.e., prenex, and no (proper) function appear as argument of another function, i.e., simple~\cite{Luck2021}.
The formulas that we deal with are hence \QBSFs of the form:
\[
  \Phi(\VarSet{g},\VarSet{y}) =
    (\SOQ_1 \VarSet{f}^1) \cdots (\SOQ_n \VarSet{f}^n)
    (\FOQ_1 \VarSet{x}^1) \cdots (Q_m \VarSet{x}^m) \PrefixMatrixSeparator
    \phi(\VarSet{f}^1,\dots,\VarSet{f}^n,\VarSet{g},\VarSet{x}^1,\dots,\VarSet{x}^m,\VarSet{y}),
\]
where:
$\VarSet{f}^i$, for all $i$, and $\VarSet{g}$ are tuples of distinct Boolean \emph{function} variables;
$\VarSet{x}^j$, for all $j$, and $\VarSet{y}$ are tuples of distinct Boolean \emph{propositional} variables;%
\footnote{In some works, see, e.g., \cite{Luck2016-techrep,Luck2021}, \emph{all} Boolean variables in \QBSFs are \emph{function} variables, with function variables of arity~$0$ called \emph{propositional} variables, and the others called (proper) function variables.}
$\SOQ_1,\dots,\SOQ_n \in \set{\exists,\forall}$ are $n \geq 0$ alternating \emph{second}\nbdash-order quantifiers; and
$\FOQ_1,\dots,\FOQ_m \in \set{\exists,\forall}$ are $m \geq 0$ alternating \emph{first}\nbdash-order quantifiers.
The function and propositional variables in $\VarSet{g}$ and in $\VarSet{y}$ are \defin{free} as they are not bound by quantifiers.
The quantifier part of $\Phi$ is the \defin{prefix} of $\Phi$, and the quantifier\nbdash-free formula $\phi$ is the \defin{matrix} of $\Phi$.
In $\phi$, the usual Boolean connectives link the following elements:
\begin{itemize}[nosep]
  \item Boolean values $1$ and $0$ for $\valtrue$ and $\valfalse$, respectively;
  \item \defin{literals}, which are positive or negated occurrences of
  \begin{itemize}[nosep]
    \item Boolean propositional variables; or
    \item expressions of the form $f(z_1,\dots,z_a)$, where $f$ is a Boolean function variable of arity $a$, and $z_1,\dots,z_a$ are Boolean values or Boolean propositional variables (or function variables of arity~$0$);
  \end{itemize}
  \item subformulas obtained by linking via Boolean connectives the elements above.
\end{itemize}
\emph{Henceforth, by Boolean formulas we mean formulas over Boolean propositions \emph{and} functions, unless stated otherwise.}
An example of a prenex simple \QBSF, where the function $h$ and the proposition $z$ are free variables,~is:
\[\Phi(h,z) = (\exists f,g) (\forall x) (\exists y) \; (y \lor \lnot f(z,x) \land \lnot z \leftrightarrow h(y,0,x)) \rightarrow (g(1,y,z) \land \lnot x \lor f(y,y) \land 1).\]

\defin{Clauses} and \defin{terms} are disjunctions and conjunctions of literals, respectively.
A prenex \QBSF is in \CNF or in \DNF if its matrix is a conjunction of clauses or a disjunction of terms, respectively.

The truth value of a \QBSF $\Phi$ depends on the interpretation of its \emph{free} function and propositional variables.
Let $x$ be a Boolean propositional variable.
An interpretation $\Interpr{I}$ of $x$ is a substitution of $x$ with a Boolean value.
We call interpretations of propositional variables also (\defin{truth}) \defin{assignments}.
The notation $\EvalInterpr{x}{\Interpr{I}}$ means that $x$ is substituted by the Boolean value that $\Interpr{I}$ assigns to $x$.
Let $f$ be a Boolean function variable of arity~$n$.
An interpretation $\Interpr{I}$ of $f$ is a substitution of $f$ with a (specific) Boolean function $\set{\valtrue,\valfalse}^n \to \set{\valtrue,\valfalse}$ mapping $n$\nbdash-tuples of Boolean values to a Boolean value.
We call interpretations of function variables also (\defin{function}) \defin{instantiations}.
The notation $\EvalInterpr{f}{\Interpr{I}}$ means that $f$ is substituted by the Boolean function that is the instantiation of $f$ in $\Interpr{I}$.
If $\VarSet{f}$ and $\VarSet{x}$ are a set of Boolean function and propositional variables, respectively, an interpretation for $\VarSet{f}$ and $\VarSet{x}$ instantiates all the functions in $\VarSet{f}$ and assign truth-values to all the propositions in $\VarSet{x}$, respectively.
The truth value of $\Phi$ \Wrt $\Interpr{I}$, denoted by $\EvalInterpr{\Phi}{\Interpr{I}}$, is obtained in the natural way (see, e.g., \cite{Luck2016-conference,Luck2016-techrep,Luck2021,Jiang2023}).

An interpretation $\Interpr{I}$ \defin{satisfies} a \QBSF $\Phi$ if $\Interpr{I}$ makes $\Phi$ $\valtrue$;
in which case, $\Interpr{I}$ is a \defin{model} of $\Phi$, and we denote this by $\Interpr{I} \models \Phi$. 
A \QBSF $\Phi$ is \defin{satisfiable} iff it admits a model, otherwise $\Phi$ is \defin{unsatisfiable};
the formula $\Phi$ is \defin{valid} iff all interpretations satisfy $\Phi$.
Let $\Phi$ and $\Psi$ be two \QBSFs, then $\Phi$ \defin{entails} $\Psi$, denoted $\Phi \models \Psi$, iff, for all interpretations $\Interpr{I}$ of the free variables in $\Phi$ \emph{and} $\Psi$, if $\Interpr{I} \models \Phi$ then $\Interpr{I} \models \Psi$.
The formulas $\Phi$ and $\Psi$ are \defin{equivalent} iff $\Phi \models \Psi$ \emph{and} $\Psi \models \Phi$, i.e., each model of $\Phi$ is a model of $\Psi$, and viceversa.

A \QBSF is a \defin{sentence} iff it has no free variables.
A sentence $\Phi$ has a \emph{single} interpretation, which is the empty one that does not assign values to propositional variables or instantiate function variables, as there are no free variables in $\Phi$.
By this, a sentence is valid iff it is satisfied by the empty interpretation, and hence sentences are satisfiable iff they are valid.
Thus, we may interchangeably use the terms validity and satisfiability for sentences.

Since Boolean functions of arity~$0$ can be seen as Boolean propositions~\cite{Luck2016-techrep,Luck2021}, prenex \QBSFs can be assumed \Wlog to have $\SOQ_n \neq \FOQ_1$ (if they have quantifiers over propositional variables at all).

\SigmaKFormulas{c} are \QBSFs whose alternating second\nbdash-order quantifiers start with $\SOQ_1 = \exists$ and are $c$ in total---as a mnemonic, the bar over the `$\Sigma$' symbol highlights that each second\nbdash-order quantifier may quantify \emph{sets} of Boolean function variables.
\SigmaKOddFormulas{c} (resp., \SigmaKEvenFormulas{c}) are \SigmaKFormulas{c} in which $c$ is odd (resp., even), and hence the last second\nbdash-order quantifier in the formulas' prefix is $\exists$ (resp., $\forall$)---as a mnemonic, this is highlighted in the subscript.
\CompactSigmaKOddFormulas{c} (resp., \CompactSigmaKEvenFormulas{c}) are \SigmaKOddFormulas{c} (resp., \SigmaKEvenFormulas{c}) in which there is a \emph{single} first\nbdash-order quantifier in the formulas' prefix---as a mnemonic, the bar over the symbols $\Pi_1$ and $\Sigma_1$ highlights that the first order quantification can be over a set of propositional variables.
Although \SigmaKOddFormulas{c} (resp., \SigmaKEvenFormulas{c}) are a superset of \CompactSigmaKOddFormulas{c} (resp., \CompactSigmaKEvenFormulas{c}), they are semantically equivalent, i.e., for every \SigmaKOddFormula{c} (resp., \SigmaKEvenFormula{c}) there exists an equivalent \CompactSigmaKOddFormula{c} (resp., \CompactSigmaKEvenFormula{c})~\cite{Jiang2023}.

When we remove the bar above the `$\Sigma$' symbol, we mean that the second order quantifiers are over \emph{single} Boolean function variables.
Therefore, for example,
\SimpleCompactSigmaKOddFormulas{c} and \SimpleCompactSigmaKEvenFormulas{c} are subclasses of \CompactSigmaKOddFormulas{c} and \CompactSigmaKEvenFormulas{c}, respectively, of the form:
\begin{align*}
    \SimpleCompactSigmaKOddFormulaPrefix{c}\text{-formulas} &\colon \quad \Phi(\VarSet{g},\VarSet{y}) = (\exists f^1)(\forall f^2) \cdots (\exists f^c) (\forall \VarSet{x})\; \phi(f^1,f^2,\dots,f^c,\VarSet{g},\VarSet{x},\VarSet{y}); \\ 
    \SimpleCompactSigmaKEvenFormulaPrefix{c}\text{-formulas} &\colon \quad \Phi(\VarSet{g},\VarSet{y}) = (\exists f^1)(\forall f^2) \cdots (\forall f^c) (\exists \VarSet{x})\; \phi(f^1,f^2,\dots,f^c,\VarSet{g},\VarSet{x},\VarSet{y}).
\end{align*}

Deciding the validity of \SigmaKSentences{c} is in $\SigmaWExp{c} = \Oracle{\NExpTime}{\SigmaP{c-1}}$, and is \SigmaWExph{c} even if the formulas are restricted to be \CNF{} 
\SimpleCompactSigmaKOddFormulaPrefix{c}- and \SimpleCompactSigmaKEvenSentences{c} (see~\cite{Lohrey2012,Luck2016-techrep}, and references therein).

\subsubsection{Definition of the problems}
\silentbookmark[3]{Definition of the problems}

For two binary strings/sequences $s$ and $t$ of \emph{equal} length, $s$ is \defin{lexicographically greater} than $t$ iff $s$ has $\valtrue$/$1$ and $t$ has $\valfalse$/$0$ in the left-most position at which $s$ and $t$ differ.
The \defin{unfolding} of a Boolean function $f$ is the binary string obtained by concatenating the values of $f$ for all possible argument combinations, from the lexicographically greatest to the least.
For example, if $f(0,0) = f(1,0) = \valfalse$ and $f(0,1) = f(1,1) = \valtrue$, the unfolding of $f$ is ``1010''.
The \defin{Hamming weight} of a truth assignment is the number of propositional variables receiving  $\valtrue$ in the assignment~\cite{CreignouV2015,ChenF2019}.
We can extend this to function variables, where the \defin{Hamming weight} of a function instantiation is the number of $\valtrue$ values in the unfolding of the instantiated function.

Let $\VarSet{f}$ and $\VarSet{x}$ be sets of function and propositional variables, respectively, and let $\Interpr{I}$ be an interpretation for both.
The \defin{Hamming weight} $\HammingWeight{\Interpr{I}}$ of $\Interpr{I}$ is the sum of the Hamming weights of its function instantiations and truth assignments.
Given a joint order $o$ over $\VarSet{f}$ and $\VarSet{x}$, the \defin{unfolding} of $\Interpr{I}$ (\Wrt $o$) is the binary string obtained by concatenating the unfolding of its function instantiations and truth assignments from most to least significant according to $o$.
For two interpretations $\Interpr{I}$ and $\Interpr{J}$ of the \emph{same} variables, we say that $\Interpr{I}$ is \defin{lexicographically greater (\Wrt $o$)} than $\Interpr{J}$, denoted by $\Interpr{I} \lexsucc_o \Interpr{J}$, if the unfolding (\Wrt $o$) of $\Interpr{I}$ is lexicographically greater than that of~$\Interpr{J}$.
We simply write $\Interpr{I} \lexsucc \Interpr{J}$ when $o$ is clear from context.
Below, an ordered set $\VarSet{a} = \set{a_1, \dots, a_n}$ of propositional or function variables is, unless otherwise stated, ordered by decreasing subscript, i.e., $a_j$ is more significant than $a_i$ iff $j > i$.
Accordingly, ordered sets will often be written as $\VarSet{a} = \set{a_n, \dots, a_1}$ to emphasize this order.

Given the above concepts, we can now define the problems over \QBSFs that we will investigate.

\Problem{\MaxSatSigmaFormula{c}}%
{A satisfiable \SigmaKFormula{c} $\Phi(\VarSet{y})$, where $\VarSet{y} = \set{y_1,\dots,y_n}$ is a set of propositional variables.}%
{Is the Hamming weight of a maximum-Hamming-weight model of $\Phi(\VarSet{y})$ odd/even?}

\Problem{\LexMaxSigmaFormula{c}}%
{A satisfiable \SigmaKFormula{c} $\Phi(\VarSet{y})$, where $\VarSet{y} = \set{y_n,\dots,y_1}$ is an ordered set of propositional variables.}%
{Does the lexicographic-maximum model of $\Phi(\VarSet{y})$ assign $\valtrue$/$\valfalse$ to $y_1$?} 

\Problem{\LexMaxFuncSigmaFormula{c}}%
{A satisfiable \SigmaKFormula{c} $\Phi(\VarSet{g})$, where $\VarSet{g} = \set{g_n,\dots,g_1}$ is an ordered set of function variables.}%
{Does the lexicographic-maximum model $\Interpr{I}$ of $\Phi(\VarSet{g})$  satisfy $\EvalInterpr{g_1}{\Interpr{I}}(0,\dots,0) = \valtrue$/$\valfalse$?}

We will show below that \MaxSatSigmaFormula{c}, \LexMaxSigmaFormula{c}, and \LexMaxFuncSigmaFormula{c}, are complete for $\BoundedOracle{\ExpTime}{\SigmaP{c}}{\LogFunctions} = \BoundedHausdCLASS{\PolFunctions}{\Oracle{\NExpTime}{\SigmaP{c-1}}}$, $\BoundedOracle{\ExpTime}{\SigmaP{c}}{\PolFunctions} = \BoundedHausdCLASS{2^{\PolFunctions}}{\Oracle{\NExpTime}{\SigmaP{c-1}}}$, and $\Oracle{\ExpTime}{\SigmaP{c}} = \BoundedHausdCLASS{2^{2^{\PolFunctions}}}{\Oracle{\NExpTime}{\SigmaP{c-1}}}$, respectively, and they remain hard even over \SimpleSigmaKFormulas{c}.
These problems 
share the flavor of the classical complete problems of the \PHText intermediate levels (see, e.g., \cite{Wagner1987,Buss1988,Krentel1988,Wagner1988,Krentel92}).
Compared to those of \PolHier, we here have an additional complete problem, as the steps of \WExpHier have three intermediate levels rather than two (above the \BHText of the step).
When $c = 1$, \MaxSatSigmaFormula{1}, \LexMaxSigmaFormula{1}, and \LexMaxFuncSigmaFormula{1}, are hence complete for $\BoundedOracle{\ExpTime}{\NPTime}{\LogFunctions} = \BoundedHausdCLASS{\PolFunctions}{\NExpTime}$, $\BoundedOracle{\ExpTime}{\NPTime}{\PolFunctions} = \BoundedHausdCLASS{2^\PolFunctions}{\NExpTime}$, and $\Oracle{\ExpTime}{\NPTime} = \BoundedHausdCLASS{2^{2^\PolFunctions}}{\NExpTime}$, respectively.
Thus, by \zcref{theo_theta_level_equals_polHausdorff,theo_delta_level_equals_expHausdorff}, the first two problems are also complete for the classes $\PNExpPar = \PNExpLog$ and $\PNExp = \NPNExp$ of the \SEHText, respectively.

\subsubsection{Proving the problems complete}
\label{sec_proving_qbsf-problems_hard}
\silentbookmark[3]{Proving the problems complete}

Proving the hardness of the problems above will rely on reductions from Hausdorff languages (of the intermediate levels of $\WExpHier$) to tasks involving the satisfiability/validity of suitable \SigmaKFormulas{c}.
Similarly to \citeauthor{Cook1971}'s~\cite{Cook1971} reduction for \Sat, parts of these formulas will encode Turing machine computations deciding $\Oracle{\NExpTime}{\SigmaP{c-1}}$ Hausdorff predicates.
However, these computations take exponential time and cannot be encoded using only propositional variables without producing exponentially\nbdash-long formulas.
To maintain polynomial reductions, we instead make extensive use of Boolean function variables as the main building blocks of our construction.
Before giving the formal details, we comment on the notation used throughout the section.

We partition the arguments of Boolean functions into groups, typically encoding integers in binary.
E.g., if $f$ relates non\nbdash-negative integers $a$, $a'$, and $a''$, represented over $r$, $s$, and $t$ bits, respectively, we use ordered sets of propositional variables
$\VarSet{a} = \set{a_{r},\dots,a_1}$,
$\VarSet{a}' = \set{a'_{s},\dots,a'_1}$, and
$\VarSet{a}'' = \set{a''_{t},\dots,a''_1}$,
and have
\[
f(\underbracket[.5pt]{\overbracket[.5pt]{\bullet,\dots,\bullet}^{a}}_{r\text{-many}};
  \underbracket[.5pt]{\overbracket[.5pt]{\circ,\dots,\circ}^{a'}}_{s\text{-many}};
  \underbracket[.5pt]{\overbracket[.5pt]{\bullet,\dots,\bullet}^{a''}}_{t\text{-many}}),
\]
where semicolons separate argument groups, $f$ has arity $r + s + t$, and it may, e.g., encode the relation $a + a' = a''$.

We adopt the following typographical convention for function arguments.
A set of propositional variables occurring as arguments is denoted by its name, whereas the binary encoding of a value is written in bold.
Thus,
\[
f(a_{r},a_{r-1},0,\dots,0,a_3,a_2,a_1;
  \BooleanEncoding{\alpha};\VarSet{a}'')
\]
has, as its three argument groups, the displayed variables and constants, the binary encoding of $\alpha$ (or $\alpha$ itself if it is a binary string), and all variables of $\VarSet{a}''$ in the order $\set{a''_{t},\dots,a''_1}$, respectively.

For the numerical arguments of functions we use \emph{fixed\nbdash-length} binary representations, with a predetermined number of bits and hence possibly leading zeros, rather than the canonical binary representations used above.

To encode below the computations of machines deciding Hausdorff predicates, we need Boolean functions for successor, majority, inequality, and addition.
The proof details of the following \zcref*[typeset=name,nocap]{theo_SigmaFormula_arithmetic} are in \zcref{sec_detailed_proofs_qbsf_hard_problems}.

\begin{lemma}[store=ArithmeticEncodedAsBooleanFunctions]
\label{theo_SigmaFormula_arithmetic}
\label{theo_SigmaFormula_majority}
\label{theo_SigmaFormula_successor}
\label{theo_SigmaFormula_inequality}
\label{theo_SigmaFormula_addition}
Let $u > 0$ be an integer,
let $\mi{succ}$, $\mi{less}$, and $\mi{neq}$, 
be Boolean functions of arity $2u$,
and let $\mi{add}$ be a Boolean function of arity $3u$. 
Additionally, let
$\VarSet{a}$, $\VarSet{a}'$, $\VarSet{a}''$, $\VarSet{b}$, and $\VarSet{b}'$
be ordered sets of $u$ propositional variables interpreted as encoding the $u$\nbdash-bit integers
$a$, $a'$, $a''$, $b$, and $b'$, respectively.
Then, there exists a quantifier-free \CNF Boolean formula
$\mu^u(\mi{succ},\mi{less},\mi{neq},\mi{add},
\VarSet{a},\VarSet{a}',\VarSet{a}'',\VarSet{b},\VarSet{b}')$
of size polynomial in $u$ such that the quantified formula
\[
(\forall \VarSet{a},\VarSet{a}',\VarSet{a}'',\VarSet{b},\VarSet{b}')
\PrefixMatrixSeparator
\mu^u(\mi{succ},\mi{less},\mi{neq},\mi{add},
\VarSet{a},\VarSet{a}',\VarSet{a}'',\VarSet{b},\VarSet{b}')
\]
has a unique model $\Interpr{I}$, and $\Interpr{I}$ instantiates $\mi{succ}$, $\mi{less}$, $\mi{neq}$, and $\mi{add}$, such that:
\begin{itemize}[nosep,label=--,left=0pt]
  \item $\EvalInterpr{\mi{succ}}{\Interpr{I}}(\VarSet{a};\VarSet{b}) = \valtrue$ iff $a + 1 = b$;
  \item $\EvalInterpr{\mi{less}}{\Interpr{I}}(\VarSet{a};\VarSet{b}) = \valtrue$ iff $a < b$;
  \item $\EvalInterpr{\mi{neq}}{\Interpr{I}}(\VarSet{a};\VarSet{b}) = \valtrue$ iff $a \neq b$;
  \item $\EvalInterpr{\mi{add}}{\Interpr{I}}(\VarSet{a};\VarSet{a}',\VarSet{b}) = \valtrue$ iff $a + a' = b$.
\end{itemize}

\end{lemma}

To prove the hardness of \MaxSatSigmaFormula{c}, \LexMaxSigmaFormula{c}, and \LexMaxFuncSigmaFormula{c}, we construct \SigmaKFormulas{c} incorporating a \QBSF $\FormulaTMHausdorffParametricNumberOfBits{u}$.
For a \emph{fixed} Hausdorff predicate $\Language{D}$ and a \emph{fixed} string $w$, this formula states whether $\Language{D}(w,z)=1$ as $z$ \emph{varies}.
Thus, $w$ is hard\nbdash-coded in $\FormulaTMHausdorffParametricNumberOfBits{u}$, whereas $z$ is represented by Boolean variables.

For our purposes, $\Language{D}$ may have doubly exponential length, so the relevant values of $z$ are bounded by $2^{2^{\PolynomialLengthHausPredD(\StringLength{w})}}-1$, for some polynomial $\PolynomialLengthHausPredD$.
Having propositional variables representing the bits of $z$ would require $2^{\PolynomialLengthHausPredD(\StringLength{w})}$ variables, precluding polynomial reductions.
Instead, the positions of these bits require only $\PolynomialLengthHausPredD(\StringLength{w})$ bits to index.
Hence, the formula $\FormulaTMHausdorffParametricNumberOfBits{u}$ is designed to ``read'' the representation of $z$ over $2^{\PolynomialLengthHausPredD(\StringLength{w})}$ from a \emph{bit\nbdash-indexing function} variable $\mi{ind}^{\PolynomialLengthHausPredD(\StringLength{w})}$ of arity $\PolynomialLengthHausPredD(\StringLength{w})$, appearing free in $\FormulaTMHausdorffParametricNumberOfBits{u}(\mi{ind}^{\PolynomialLengthHausPredD(\StringLength{w})})$, with:
\begin{center}
\begin{tabular}{l p{7.5cm}}
  $\mi{ind}^{\PolynomialLengthHausPredD(\StringLength{w})}(\underbracket[.5pt]{\overbracket[.5pt]{\bullet,\bullet,\dots,\bullet,\bullet}^{j}}_{{\PolynomialLengthHausPredD(\StringLength{w})}\text{-many}})$ &
  stating whether the $(j{+}1)$-th least significant bit of the $2^{\PolynomialLengthHausPredD(\StringLength{w})}$-bit representation of $z$ is `$1$', or not.\\
\end{tabular}
\end{center}

The superscript $u$ of $\FormulaTMHausdorffParametricNumberOfBits{u}$ indicates that the arithmetic formula $\mu$ of \zcref{theo_SigmaFormula_arithmetic}, used within $\FormulaTMHausdorffParametricNumberOfBits{u}$ to have functions to compare and sum integers, is instantiated as $\mu^u$, for integers represented over $u$ bits.

The properties of $\FormulaTMHausdorffParametricNumberOfBits{u}$ are stated in the following \zcref*[typeset=name,nocap]{theo_coding_nexp_machine_for Hausdorff_predicate}, whose proof is deferred to \zcref{sec_detailed_proofs_qbsf_hard_problems} and details its construction.

\begin{lemma}[store=CodingNEXPMachineHausdorffPredicate]
\label{theo_coding_nexp_machine_for Hausdorff_predicate}
Let $\Language{D}$ be a $\Oracle{\NExpTime}{\SigmaP{c-1}}$ Hausdorff predicate, where $c \geq 1$ is an integer.
If $\Language{D}$ has doubly exponential (resp., exponential, polynomial) length,
then there exist polynomials $\PolynomialBoundingEverything(n)$ and $\PolynomialLengthHausPredD(n)$, with  $\PolynomialBoundingEverything(n) \geq \PolynomialLengthHausPredD(n)$, such that
$2^{2^{\PolynomialLengthHausPredD(n)}} - 1$ (resp., ${2^{\PolynomialLengthHausPredD(n)}} - 1$, ${{\PolynomialLengthHausPredD(n)}} - 1$) bounds the length of $\Language{D}$
and,
from every string $w$,
one can build in polynomial time \Wrt~$\StringLength{w}$ a \CNF \SigmaKFormula{c} $\FormulaTMHausdorffWithP(\mi{ind}^{\PolynomialLengthHausPredD(\StringLength{w})})$
having the following properties:
\begin{itemize}[nosep,label=--,left=0pt]
  \item for every integer $z \geq 0$,
      if $\Interpr{I}_z$ interprets $\mi{ind}^{\PolynomialLengthHausPredD(\StringLength{w})}$ as the bit\nbdash-indexing function of $z$ over $2^{\PolynomialLengthHausPredD(\StringLength{w})}$ bits, then $\Interpr{I}_z \models \FormulaTMHausdorffWithP(\mi{ind}^{\PolynomialLengthHausPredD(\StringLength{w})}) \Leftrightarrow \Language{D}(w,z) = 1$;
  \item $\FormulaTMHausdorffWithP(\mi{ind}^{\PolynomialLengthHausPredD(\StringLength{w})})$ includes the formula $\mu^{\PolynomialBoundingEverything(\StringLength{w})}$ (from \zcref{theo_SigmaFormula_arithmetic}) by existentially quantifying its function variables and universally quantifying its propositional variables; and
  \item $\FormulaTMHausdorffWithP(\mi{ind}^{\PolynomialLengthHausPredD(\StringLength{w})})$ is a \CNF \CompactSigmaKOddFormula{c} (resp., a \CNF \CompactSigmaKEvenFormula{c}) when $c$ is odd (resp., even).
\end{itemize}
\end{lemma}

The complexity of the problems introduced above are summarized in the following three results.
The proofs are deferred to \zcref{sec_detailed_proofs_qbsf_hard_problems}.
We start with the complexity of \LexMaxFuncSigmaFormula{c}, shown to be $\Oracle{\ExpTime}{\SigmaP{c}}$\CompleteSuffix.

\begin{theorem}[store=ComplexityLexMaxFunc]
\label{theo_complexity_LexMaxFunc}
Let $\Phi(\VarSet{g})$ be a satisfiable \SigmaKFormula{c}, where $\VarSet{g} = \set{g_n,\dots,g_1}$ is an ordered set of Boolean function variables, and $c \geq 1$ is an integer.
Then, deciding whether the lexicographic\nbdash-maximum model of $\Phi(\VarSet{g})$ instantiates $g_1$ so that $\EvalInterpr{g_1}{\Interpr{I}}(0,\dots,0)$ is $\valtrue$/$\valfalse$ is $\Oracle{\ExpTime}{\SigmaP{c}}$\CompleteSuffix.
Hardness holds even if $\Phi(\VarSet{g})$ has the following properties:
\begin{itemize}[nosep,label=--,left=0pt]
  \item $\VarSet{g}$ contains a single Boolean function variable, i.e., $\VarSet{g} = \set{g}$;
  \item $\Phi(g)$ is a \CNF \SimpleCompactSigmaKOddFormula{c} (resp., a \CNF \SimpleCompactSigmaKEvenFormula{c}) when $c$ is odd (resp., even);
  \item if $\Interpr{I}$ and $\Interpr{J}$ are two interpretations for $g$ such that $\Interpr{I} \lexsucc \Interpr{J}$, then $\Interpr{I} \models \Phi(g) \Rightarrow \Interpr{J} \models \Phi(g)$; and
  \item the interpretation instantiating $g$ as the function that evaluates to $\valfalse$ on every input is a model of $\Phi(g)$.
\end{itemize}
\end{theorem}

The following \zcref*[typeset=name,nocap]{theo_complexity_LexMax} shows that \LexMaxSigmaFormula{c} is $\BoundedOracle{\ExpTime}{\SigmaP{c}}{\PolFunctions}$\CompleteSuffix.

\begin{theorem}[store=ComplexityLexMax]
\label{theo_complexity_LexMax}
Let $\Phi(\VarSet{y})$ be a satisfiable \SigmaKFormula{c}, where $\VarSet{y} = \set{y_{n},\dots,y_1}$ is an ordered set of Boolean propositional variables, and $c \geq 1$ is an integer.
Then, deciding whether the lexicographic\nbdash-maximum model of $\Phi(\VarSet{y})$ assigns $\valtrue$/$\valfalse$ to $y_1$ is $\BoundedOracle{\ExpTime}{\SigmaP{c}}{\PolFunctions}$\CompleteSuffix.
Hardness holds even if $\Phi(\VarSet{y})$ has the following properties:
\begin{itemize}[nosep,label=--,left=0pt]
  \item $\Phi(\VarSet{y})$ is a \CNF \SimpleCompactSigmaKOddFormula{c} (resp., a \CNF \SimpleCompactSigmaKEvenFormula{c}) when $c$ is odd (resp., even);
  \item the propositional variables $\VarSet{y}$ never occur as arguments of function variables in $\Phi(\VarSet{y})$;
  \item if $\Interpr{I}$ and $\Interpr{J}$ are two interpretations for $\VarSet{y}$ such that $\Interpr{I} \lexsucc \Interpr{J}$, then $\Interpr{I} \models \Phi(\VarSet{y}) \Rightarrow \Interpr{J} \models \Phi(\VarSet{y})$; and
  \item the interpretation assigning $\valfalse$ to all variables in $\VarSet{y}$ is a model of $\Phi(\VarSet{y})$.
\end{itemize}
\end{theorem}

The next \zcref*[typeset=name,nocap]{theo_complexity_MaxSat} shows that \MaxSatSigmaFormula{c} is $\BoundedOracle{\ExpTime}{\SigmaP{c}}{\LogFunctions}$\CompleteSuffix.

\begin{theorem}[store=ComplexityMaxSat]
\label{theo_complexity_MaxSat}
Let $\Phi(\VarSet{y})$ be a satisfiable \SigmaKFormula{c}, where $\VarSet{y} = \set{y_1,\dots,y_n}$ is a set of Boolean propositional variables, and $c \geq 1$ is an integer.
Then, deciding whether the Hamming weight of a maximum\nbdash-Hamming\nbdash-weight model of $\Phi(\VarSet{y})$ is odd/even is $\BoundedOracle{\ExpTime}{\SigmaP{c}}{\LogFunctions}$\CompleteSuffix.
Hardness holds even if $\Phi(\VarSet{y})$ has the following properties:
\begin{itemize}[nosep,label=--,left=0pt]
  \item $\Phi(\VarSet{y})$ is a \CNF \SimpleCompactSigmaKOddFormula{c} (resp., a \CNF \SimpleCompactSigmaKEvenFormula{c}) when $c$ is odd (resp., even);
  \item the propositional variables $\VarSet{y}$ never occur as arguments of function variables in $\Phi(\VarSet{y})$;
  \item if $\Interpr{I}$ is a model of $\Phi(\VarSet{y})$, then $\EvalInterpr{y_i}{\Interpr{I}} = \valtrue \Rightarrow \EvalInterpr{y_{i-1}}{\Interpr{I}} = \valtrue$, for all $i$ with $1 < i \leq n$;
  \item if $\Interpr{I}$ and $\Interpr{J}$ are two interpretations for $\VarSet{y}$ such that $\HammingWeight{\Interpr{I}} > \HammingWeight{\Interpr{J}}$, then $\Interpr{I} \models \Phi(\VarSet{y}) \Rightarrow \Interpr{J} \models \Phi(\VarSet{y})$; and
  \item the interpretation assigning $\valfalse$ to all variables in $\VarSet{y}$ is a model of $\Phi(\VarSet{y})$.
\end{itemize}
\end{theorem}

\resumetoc

\subsection%
[Natural Hard Problems for \texorpdfstring{$\mathrm{P}^{\mathrm{{NE{\scriptscriptstyle XP}}}[\LogFunctions]}$}{P\textasciicircum{}NEXP[Log]} and \texorpdfstring{$\mathrm{P}^{\mathrm{{NE{\scriptscriptstyle XP}}}}$}{P\textasciicircum{}NEXP}]%
{Natural Hard Problems for \texorpdfstring{$\mathbf{P}^{\mathbf{{NE{\scriptscriptstyle XP}}}\boldsymbol{[\LogFunctions]}}$}{} and \texorpdfstring{$\mathbf{P}^{\mathbf{{NE{\scriptscriptstyle XP}}}}$}}
\label{sec_datalog_hard_problems}

In this section, we establish tight \PNExpLogh{}ness and \PNExph{}ness results for problems for which hardness was left open in \cite{CeylanLMMV2021} due to the lack of suitable source problems to reduce from.
While no \PNExpLogc{} problems were known at all, the known \PNExpc{} problems~\cite{EiterLP2016,EiterTechRep2016,LukasiewiczMMMP2022} had an \NPNExp{} flavor (recall that $\PNExp = \NPNExp$), making them ill\nbdash-suited as sources of hardness for problems with a \PNExp{} flavor (see \zcref{sec_discussion_Hausdorff_perspective}).
The problems analyzed in this section concern reasoning over \DatalogPM ontologies, which are ontologies expressed via existential rules.

\stoptoc

\subsubsection{Preliminaries on \DatalogPM and \Minexs}
\silentbookmark[3]{Preliminaries on Datalog+/- and MinExes}

A \DatalogPM knowledge base~\cite{CaliGL2012}, or ontology, is a pair $\KB = \KBDetails$, where $\DB$ is a relational database, and $\Dep$ is a set of existential rules, or tuple generating dependencies (TGDs)~\cite{BeeriV1984}.
TGDs are first\nbdash-order logic formulas of the form
\(\forall \VarSet{X} \forall \VarSet{Y} \PrefixMatrixSeparator (\phi(\VarSet{X}, \VarSet{Y}) \ra \exists \VarSet{Z} \PrefixMatrixSeparator \psi(\VarSet{X}, \VarSet{Z})),\)
where $\VarSet{X}$, $\VarSet{Y}$, and $\VarSet{Z}$, are pairwise disjoint tuples of first\nbdash-order variables, and $\phi(\VarSet{X},\VarSet{Y})$ and $\psi(\VarSet{X}, \VarSet{Z})$ are (non-empty) conjunctions of atoms.
For brevity, we will write TGDs as `$\phi(\VarSet{X},\VarSet{Y}) \ra \exists \VarSet{Z} \PrefixMatrixSeparator \psi(\VarSet{X},\VarSet{Z})$'.
Intuitively, TGDs allow us to derive from facts in the database, and from atoms already derived via previous TGD applications, additional true facts enriching the knowledge available in the database.
This is achieved via this simple step:
if the TGD body $\phi(\VarSet{X},\VarSet{Y})$ is satisfied by a particular instantiation of the variables over the set of atoms currently known to be true (tuples in the database are assumed to be true atoms), then the TGD head $\exists \VarSet{Z} \PrefixMatrixSeparator \psi(\VarSet{X},\VarSet{Z})$ must be satisfied as well.
In particular, if in the set of the true atoms there is no atom satisfying the TGD head, then a suitable atom is ``generated'' and added to the set of true atoms;
if needed, new first\nbdash-order objects are ``created'' as instantiations of the existentially quantified FO variables $\VarSet{Z}$ appearing in the rule head.
These freshly created objects are called \emph{null} values, intuitively because we do not know who they precisely are, but we do know that they must exist.

As an example, consider the following very simple TGD stating that if $D$ is a department in a university $U$, then there must be a head of the department who is a professor $P$ affiliated with the department:
\[
  \mi{dep}(D,U) \ra (\exists P) \PrefixMatrixSeparator \mi{prof}(P) \land \mi{affil}(P,D) \land \mi{head\mhyphen{}of\mhyphen{}dep}(P,D).
\]

Assume that, in the process of applying TGDs, at some point we do know that ``computer science'' is a department of a university, i.e., that the atom $\mi{dep}(\mi{cs},u)$ is $\valtrue$, but there is no atom stating who the head of the computer science department is.
Since the head of the TGD has to be satisfied, a fresh null value, say $\nu_1$, is created, and the facts $\mi{prof}(\nu_1), \mi{affil}(\nu_1,cs),\mi{head\mhyphen{}of\mhyphen{}dep}(\nu_1,cs)$ are added to the set of true facts.

Querying a knowledge base $\KB = \KBDetails$ via a Boolean Conjunctive Query (\BCQ) $\Query$ intuitively means querying via $\Query$ all the information that is stored in the database $\DB$ augmented with all the information that can be inferred from $\DB$ via the TGDs in $\Dep$.
We denote by $\KBDetails \models \Query$ the fact that the Boolean query $\Query$ is entailed by the database $\DB$ via the rules in $\Dep$.
A Union of Boolean Conjunctive queries (\UCQ) is a query $\Query = \Query_1 \lor \dots \lor \Query_n$ which is the disjunction of multiple \BCQs.
A $\KB = \KBDetails$ entails such a \UCQ $\Query = \Query_1 \lor \dots \lor \Query_n$ iff $\KBDetails \models \Query_i$ for some $i$.
For more comprehensive preliminaries and additional references on this topic, see \cite{TsamouraCMU2021,LukasiewiczMMMP2022}.

An explanation $\Expl$ for a \BCQ $\Query$ \Wrt the knowledge base $\KB$ is a set of facts from $\DB$ that is sufficient to entail the query via the TGDs in $\Dep$.
An explanation $\Expl$ for $\Query$ is $\preccurlyeq$\nbdash-minimal, for some preorder $\preccurlyeq$ over $2^{\DB}$, if there is no other explanation $\Expl'$ for $\Query$ such that $\Expl' \prec \Expl$ \cite{CeylanLMV2019,CeylanLMMV2021}.
Among these preorders we can find, e.g., subset\nbdash-inclusion ($\subseteq$), cardinality minimality ($\leq$), and weight minimality ($\leq_w$)---in case of weight minimality, each database fact $f$ is associated with a weight $w(f)$, and the weight $w(S)$ of a set $S$ of facts is the sum of their weights.
If $\Expl$ is a $\preccurlyeq$\nbdash-minimal explanation, we say that $\Expl$ is a $\preccurlyeq$\nbdash-\Minex.
Hence, an explanation $\Expl$ is:
\begin{itemize}[nosep,label=--]
  \item a $\subseteq$\nbdash-\Minex if there is no other explanation $\Expl'$ such that $\Expl' \subset \Expl$ \cite{CeylanLMV2019};
  \item a $\leq$\nbdash-\Minex if there is no other explanation $\Expl'$ such that $\SetSize{\Expl'} < \SetSize{\Expl}$ \cite{CeylanLMMV2021}; and
  \item a $\leq_w$\nbdash-\Minex if there is no other explanation $\Expl'$ such that $w(\Expl') < w(\Expl)$ \cite{CeylanLMMV2021}.
\end{itemize}

Over minimal explanations various problems can be defined~\cite{CeylanLMV2019,CeylanLMMV2021}, and among the problems there presented, here we focus on the following ones;
below, the symbol $\DLFrag{A}$ means that the TGDs of the input knowledge base have to fulfill an ``acyclicity'' condition (see, e.g.,~\cite{CalauttiGMT22,LukasiewiczMMMP2022}):
\begin{itemize}[nosep,label=--,left=0pt]
  \item \MinExRelGeneric{}($\DLFrag{A}$): for a knowledge base $\KB = \KBDetails$, a \UCQ $\Query$, and a fact $f \in \DB$, decide whether there exists a $\preccurlyeq$\nbdash-\Minex for $\Query$ including $f$;
  \item \MinExNecGeneric{}($\DLFrag{A}$): for a knowledge base $\KB = \KBDetails$, a \UCQ $\Query$, and a fact $f \in \DB$, decide whether every $\preccurlyeq$\nbdash-\Minex for $\Query$ includes $f$;
  \item \MinExIrrelGeneric{}($\DLFrag{A}$): for a knowledge base $\KB = \KBDetails$, a \UCQ $\Query$, and a family of sets of database facts $\SetOfSets{F} = \set{F_1,\dots,F_n}$, decide whether there exists a $\preccurlyeq$\nbdash-\Minex $\Expl$ for $\Query$ such that $F_i \not\subseteq \Expl$, for all $1 \leq i \leq n$.
\end{itemize}

The complexity of these problems can be analyzed within different ``complexity settings'', depending on which parts of the input are kept constant.
For example, if we study the complexity of these problems by considering the query and the TGDs fixed, then we say that we are analyzing the data complexity of the problem~\cite{Vardi1982}.
On the other hand, if also the TGDs and the query may vary, then we are analyzing the combined complexity of the problem~\cite{Vardi1982}.
The $\mi{ba}$\nbdash-combined complexity setting is a variant of the combined complexity one, in which the query and the TGDs may vary, but the arity of the predicates of the relational schema is bounded by a constant.
For the above three problems, membership in \PNExpLog{} and \PNExp{}, for the cardinality- and weight\nbdash-minimality criteria, respectively, was shown in the $\mi{ba}$\nbdash-combined and the combined complexity~\cite{CeylanLMMV2021}.
However, all the corresponding hardness results have remained open.

\subsubsection{Closing the hardness of the problems}
\silentbookmark[3]{Closing the hardness of the problems}

To obtain these hardness results, we give reductions establishing the $\PNExpLogh${}ness of \MinExRelMinCard{}($\DLFrag{A}$) and the $\PNExph${}ness of \MinExRelMinWeight{}($\DLFrag{A}$); the remaining hardness results follow as corollaries.
The reductions are from \MaxSatSigmaFormula{1} and \LexMaxSigmaFormula{1}, respectively, introduced above.
These problems concern models of \SigmaKFormulas{1} and remain hard for \CNF \SimpleCompactSigmaKOddFormulas{1} of the form
$\Phi(\VarSet{y}) = (\exists f)(\forall \VarSet{x}) \PrefixMatrixSeparator \phi(f,\VarSet{x},\VarSet{y})$.
Hence, both reductions share an encoding of the satisfiability of $\Phi(\VarSet{y})$ into a \DatalogPM ontology $\mathcal{M}\Phi = \KBDef{\DB\Phi,\Dep_\Phi}$.

Since the variables in $\VarSet{y}$ are the only free ones of $\Phi(\VarSet{y})$, a model of the formula $\Phi(\VarSet{y})$ is an interpretation $\Interpr{I}$ of $\VarSet{y}$ such that
$(\exists f)(\forall \VarSet{x}) \PrefixMatrixSeparator \phi(f,\VarSet{x}, \EvalInterpr{\VarSet{y}}{\Interpr{I}})$
is valid.
We design $\mathcal{M}_\Phi$ so that, when augmented with facts encoding a set $\mathscr{I} = \set{\Interpr{I}_1,\dots,\Interpr{I}_p}$ of such interpretations, it entails facts identifying which members of $\mathscr{I}$ are models of $\Phi(\VarSet{y})$.
In particular, a set $E_{\mathscr{I}}$ of facts is said to ``properly encode'' $\mathscr{I} = \set{\Interpr{I}_1,\dots,\Interpr{I}_p}$ iff:
\begin{itemize}[nosep,label=--]
\item for each interpretation $\Interpr{I}_j \in \mathscr{I}$, a fact $\mi{AssignY}(\mi{id}_j)$ belongs to $E_{\mathscr{I}}$, where $\mi{id}_j$ is an object ``identifying'' the interpretation $\Interpr{I}_j$, and no additional \AssignYFact is in $E_{\mathscr{I}}$;
\item for each interpretation $\Interpr{I}_j \in \mathscr{I}$, and for each variable $y_i \in \VarSet{y}$, the fact $\mi{YVal}(\mi{id}_j,y_i,\DBvaltrue)$ belongs to $E_{\mathscr{I}}$ iff $\EvalInterpr{y_i}{\Interpr{I}_j} = \valtrue$, and the fact $\mi{YVal}(\mi{id}_j,y_i,\DBvalfalse)$ belongs to $E_{\mathscr{I}}$ iff $\EvalInterpr{y_i}{\Interpr{I}_j} = \valfalse$, where $\DBvaltrue$ and $\DBvalfalse$ are the constants used in $\DB_\Phi$ associated with the Boolean values $\valtrue$ and $\valfalse$, respectively.
\end{itemize}
This naturally extends to a (single) interpretation $\Interpr{I}$ of $\VarSet{y}$, and we say that a set of facts $E_{\Interpr{I}}$ ``properly encodes'' $\Interpr{I}$.

Additionally, when augmented with $E_{\mathscr{I}}$, $\mathcal{M}\Phi$ entails the following facts, where $f$ has arity $a \leq n$:
\begin{itemize}[nosep,label=--]
\item for every binary sequence $\VarSet{s}$ of length $a$, a fact $\mi{InputFunc}(\mi{id}_{\VarSet{s}})$ stating that $\mi{id}_{\VarSet{s}}$ is an object ``identifying'' the sequence $\VarSet{s}$, which can hence be regarded as an input for $f$ (the object $\mi{id}_{\VarSet{s}}$ will actually be a null);
\item for each binary sequence $\VarSet{s}$ input for $f$, and for each $1 \leq i \leq a$, either a fact $\mi{InputFuncVal}(\mi{id}_{\VarSet{s}},p_i,\DBvaltrue)$ or a fact $\mi{InputFuncVal}(\mi{id}_{\VarSet{s}},p_i,\DBvalfalse)$ stating that in the binary sequence identified by $\mi{id}_{\VarSet{s}}$, in $i$-th position from the left there is $\valtrue$ or $\valfalse$, respectively; here, $p_i$ is a constant identifying the $i$-th position;
\item for each instantiation $\Interpr{F}$ of the function $f$, a fact $\mi{Func}(\mi{id}_{\Interpr{F}})$ stating that $\mi{id}_{\Interpr{F}}$ is an object ``identifying'' the instantiation $\Interpr{F}$ of $f$ (the object $\mi{id}_{\Interpr{F}}$ will actually be a null);
\item for each instantiation $\Interpr{F}$ of the function $f$, and for each binary sequence $\VarSet{s}$ input for $f$, either a fact $\mi{OutputFunc}(\mi{id}_{\Interpr{F}},\mi{id}_{\VarSet{s}},\DBvaltrue)$ or $\mi{OutputFunc}(\mi{id}_{\Interpr{F}},\mi{id}_{\VarSet{s}},\DBvaltrue)$ stating that the instantiation of $f$ identified by $\mi{id}_{\Interpr{F}}$, over the input sequence identified by $\mi{id}_{\VarSet{s}}$, outputs $\valtrue$ or $\valfalse$, respectively; and
\item for each interpretation $\Interpr{I}_j$ for the variables $\VarSet{y}$, identified by $\mi{id}_j$, and for each instantiation $\Interpr{F}$ of the function $f$, identified by $\mi{id}_{\Interpr{F}}$, the fact $\mi{ModelWitness}(\mi{id}_j,\mi{id}_{\Interpr{F}})$ iff the formula $(\forall \VarSet{x}) \PrefixMatrixSeparator \phi(\EvalInterpr{f}{\Interpr{F}},\VarSet{x},\EvalInterpr{\VarSet{y}}{\Interpr{I}_j})$ is valid.
\end{itemize}
Further facts, e.g., concerning $\VarSet{x}$, are entailed as well; 
see the proof of \zcref{theo_encoding_qbsf_satisfaction} for the complete construction.

Notice that the \zcref*[typeset=name,nocap]{theo_encoding_qbsf_satisfaction} below also provides a (superset) of a construction showing the \NExpTimeh{}ness of entailment over bounded\nbdash-arity acyclic \DatalogPM knowledge bases---the actual construction for the latter is obtained from the one below by removing all references to the variables $\VarSet{y}$.
The construction we propose is a novelty compared to those in the literature, as ours is based on the satisfiability of \QBSFs, whereas the others, to the best of our knowledge, are all based on the exponential\nbdash-square tiling problem~\cite{DantsinV1997,EiterLP2016,EiterTechRep2016,LukasiewiczMMMP2022}.
The proof details of the following \zcref*[typeset=name,nocap]{theo_encoding_qbsf_satisfaction} are in \zcref{sec_detailed_proofs_datalog_hard_problems}.

\begin{lemma}[store=EncodingQBSFSatisfactionByTGDs]
\label{theo_encoding_qbsf_satisfaction}
Let $\Phi(\VarSet{y}) = (\exists f)(\forall \VarSet{x}) \PrefixMatrixSeparator \phi(f,\VarSet{x},\VarSet{y})$ be a \CNF \SimpleCompactSigmaKOddFormula{1}, where $f$ is a Boolean function variable of arity $a$, and $\VarSet{x} = \set{x_1,\dots,x_n}$ and $\VarSet{y}=\set{y_1,\dots,y_m}$ are sets of Boolean propositional variables, where \Wlog $a \leq n$.
Then, one can build in polynomial time an acyclic \DatalogPM ontology $\mathcal{M}_{\Phi} = \KBDef{\DB_\Phi,\Dep_\Phi}$, whose predicates have bounded arity, such that, if $E_\mathscr{I}$ is a set of facts ``properly encoding'' a set $\mathscr{I} = \set{\Interpr{I}_1,\dots,\Interpr{I}_p}$ of interpretations for $\VarSet{y}$, then $\KBDef{\DB_\Phi \cup E_{\mathscr{I}},\Dep_\Phi}$ entails a fact $\mi{ModelWitness}(\mi{id}_j,\mi{id}_{\Interpr{F}})$ for some objects $\mi{id}_j$ and $\mi{id}_{\Interpr{F}}$ iff:
\begin{itemize}[nosep,label=--,left=0pt]
  \item $\mi{id}_j$ identifies an interpretation $\Interpr{I}_j \in \mathscr{I}$ such that $\Interpr{I}_j \models \Phi(\VarSet{y})$; and
  \item $\mi{id}_{\Interpr{F}}$ identifies an instantiation $\Interpr{F}$ of $f$ witnessing the validity of $\Phi(\EvalInterpr{\VarSet{y}}{\Interpr{I}_j})$, i.e., $(\forall \VarSet{x}) \PrefixMatrixSeparator \phi(\EvalInterpr{f}{\Interpr{F}},\VarSet{x},\EvalInterpr{\VarSet{y}}{\Interpr{I}_j})$ is valid.
\end{itemize}
\end{lemma}

The following result shows the \PNExpLogh{}ness of \MinExRelMinCard{}($\DLFrag{A}$).%
\footnote{An different proof based on \zcref{theo_hardness_oddity_general} and the \emph{Exponential Tiling Problem} \cite{EiterLP2016,LukasiewiczMMMP2022} is in \zcref{sec_detailed_proofs_datalog_hard_problems}.}

\begin{theorem}[store=thmMinexRelACardBAHardness]
\label{thm_minex-rel-card_A_ba_hardness}
\MinExRelMinCard{}($\DLFrag{A}$) is \PNExpLogh{} in the $\mi{ba}$\nbdash-combined (resp., combined) complexity.
Hardness holds even on instances whose queries are \BCQs.
\end{theorem}

\begin{proof}
We exhibit a polynomial reduction from the problem \MaxSatSigmaFormula{1}.
In particular, from an instance $\Phi(\VarSet{y}) = \linebreak[0] (\exists f)(\forall \VarSet{x}) \PrefixMatrixSeparator \phi(f,\VarSet{x},\VarSet{y})$ of \MaxSatSigmaFormula{1}, with $\VarSet{y} = \set{y_1,\dots,y_m}$ a set of propositional variables, we build an instance $\tup{\KB,\Query,f^\star}$ of \MinExRelMinCard{}($\DLFrag{A}$), with $\KB = \KBDef{\DB,\Dep}$ a knowledge base, $\Query$ a \BCQ, and $f^\star \in \DB$ a fact.
By \zcref{theo_complexity_MaxSat}, we assume $\Phi(\VarSet{y})$ to be a satisfiable \CNF \SimpleCompactSigmaKOddFormula{1} whose models satisfy $y_i \rightarrow y_{i-1}$, for all $1 < i \leq m$.
Hence, an interpretation $\Interpr{I}$ of $\VarSet{y}$ can be a model of $\Phi(\VarSet{y})$ only if $\Interpr{I}$ is such that 
\[
  \EvalInterpr{y_1}{\Interpr{I}} = \dots = \EvalInterpr{y_k}{\Interpr{I}} = \valtrue
  \quad\text{and}\quad
  \EvalInterpr{y_{k+1}}{\Interpr{I}} = \dots = \EvalInterpr{y_m}{\Interpr{I}} = \valfalse
\]
for some $0 \leq k \leq m$.
Therefore, below we can focus only on interpretations $\Interpr{I}_k$ assigning $\valtrue$ to all $y_i$ with $1 \leq i \leq k$, and $\valfalse$ to the remaining $y_j$ with $k < j \leq m$;
observe that the Hamming weight of $\Interpr{I}_k$ is exactly $k$.

Let $\mathcal{M}_\Phi = \KBDef{\DB_\Phi,\Dep_\Phi}$ be the ontology of \zcref{theo_encoding_qbsf_satisfaction}.
The reduction in this proof extends $\mathcal{M}_\Phi$.

\medbreak

\noindent
\emph{(The database).}
The database $\DB$ is obtained by augmenting the database $\DB_\Phi$ with the following facts.

There are facts that we will associate with the parity of the size of minimal explanations:
\[
  \mi{Odd}()
  \qquad\text{and}\qquad
  \mi{Even}().
\]

Moreover, for every $0 \leq k \leq m$, there are \emph{multiple} facts $\mi{Cut}_k$ allowing us to select a specific interpretation $\Interpr{I}_k$ of $\VarSet{y}$---remember that the interpretations $\Interpr{I}_k$ of interest have a very specific form, and hence, to individuate $\Interpr{I}_k$, we simply need to choose the cut $k$ above which the variables $y_i$ are evaluated $\valfalse$ in $\Interpr{I}_k$:
\[
  \mi{Cut}_k(1),\dots,\mi{Cut}_k(m-k+1).
\]
Intuitively, for a given $k$, we introduce this specific multiplicity of $\mi{Cut}_k$\nbdash-facts so that, for increasing values of $k$, the explanations associated with $\Interpr{I}_k$ have decreasing sizes.

\smallskip

\noindent
\emph{(The program).}
The program $\Dep$ is obtained by augmenting the program $\Dep_\Phi$ with the following rules.

For every odd $k$, and $0 \leq k \leq m$, we add the following rule whose aim is translating the choice of the ``cut'' $k$ individuating the interpretation $\Interpr{I}_k$ into the specific assignment for $\VarSet{y}$ (we generate also an \AssignYFact, so to generate a set of facts ``properly encoding'' this assignment; see above):
\begin{align*}
&
\mi{Odd}(),
\mi{Cut}_k(1),\dots,\mi{Cut}_k(m-k+1)
\\
&\rightarrow
\mi{AssignY}(\mi{id}_k),
\mi{YVal}(\mi{id}_k,y_1,\DBvaltrue),
\dots,
\mi{YVal}(\mi{id}_k,y_k,\DBvaltrue),
\\
&\hspace{2cm}
\mi{YVal}(\mi{id}_k,y_{k+1},\DBvalfalse),
\dots,
\mi{YVal}(\mi{id}_k,y_m,\DBvalfalse).
\end{align*}
For every even $k$, and $0 \leq k \leq m$, we add the analogous rule:
\begin{align*}
&
\mi{Even}(),
\mi{Cut}_k(1),\dots,\mi{Cut}_k(m-k+1)
\\
&\rightarrow
\mi{AssignY}(\mi{id}_k),
\mi{YVal}(\mi{id}_k,y_1,\DBvaltrue),
\dots,
\mi{YVal}(\mi{id}_k,y_k,\DBvaltrue),
\\
&\hspace{2cm}
\mi{YVal}(\mi{id}_k,y_{k+1},\DBvalfalse),
\dots,
\mi{YVal}(\mi{id}_k,y_m,\DBvalfalse).
\end{align*}

\smallskip

\noindent
\emph{(The query).}
The query $\Query$ is the \BCQ \[\Query = \mi{AssignY}(\mi{IdY}), \mi{YVal}(\mi{IdY},y_1,T_1), \dots, \mi{YVal}(\mi{IdY},y_m,T_m), \DB_\Phi, \mi{ModelWitness}(\mi{IdY},\mi{IdF}),\]
where with $\DB_\Phi$ in the query we mean the conjunction of all facts of $\DB_\Phi$.

\smallskip

\noindent
\emph{(The distinguished fact).}
Finally, the distinguished fact is $f^\star = \mi{Odd}()$.

\smallskip

Observe that the construction can be computed in polynomial time in the size of $\Phi(\VarSet{y})$.
The predicates have bounded arity, and acyclicity follows from \zcref{theo_encoding_qbsf_satisfaction} together with the fact that the added rules are non\nbdash-recursive. %
Furthermore, the query is entailed by $\KB$, and this follows from \zcref{theo_encoding_qbsf_satisfaction}, the fact that we are assuming $\Phi(\VarSet{y})$ satisfiable, and the presence of all the $\mi{Cut}_k$\nbdash-facts in the database.

\medskip

\noindent
We now prove the transformation correctness by showing that the Hamming weight of the maximum-Hamming-weight model of $\Phi(\VarSet{y}) = (\exists f)(\forall \VarSet{x}) \PrefixMatrixSeparator \phi(f,\VarSet{x},\VarSet{y})$ is odd iff there exists a $\leq$\nbdash-\Minex for $\Query$ containing $f^\star$.

We start by highlighting a couple of properties of the $\leq$\nbdash-minimal explanations for $\Query$.

\begin{subproperty}
\label{subprop_shape_explanations_min_card}
Let $\Expl \subseteq \DB$ be a $\leq$\nbdash-\Minex for $\Query$.
Then, there exists exactly one $k$ such that $\Expl = \DB_\Phi \cup \set{\mi{Cut}_k(1),\dots,\linebreak[0]\mi{Cut}_k(m-k+1)} \cup \set{\mi{Odd}()\mid k \text{ is odd}} \cup \set{\mi{Even}()\mid k \text{ is even}}$.
\end{subproperty}

\begin{subproof}
Since all facts from $\DB_\Phi$ are atoms of $\Query$, and $\Expl$ is an explanation for $\Query$, $\Expl$ must contain all these facts.

We now show that $\Expl$ contains \emph{all} $\mi{Cut}_k$\nbdash-facts for exactly one $k$, and \emph{no} $\mi{Cut}_{k'}$\nbdash-fact for every $k' \neq k$.

Assume by contradiction that $\Expl$ contains, for some $\ell$, at least one but not all the $\mi{Cut}_\ell$\nbdash-facts.
Since some of the $\mi{Cut}_\ell$\nbdash-facts are not in $\Expl$, the related TGD is not triggered by the facts in $\Expl$.
Hence, we can remove from $\Expl$ all the $\mi{Cut}_\ell$\nbdash-facts and obtain a smaller explanation for $\Query$, as $\Expl$ is an explanation and the $\mi{Cut}_\ell$\nbdash-facts are not used in the derivation satisfying the query.
However, this contradicts the $\leq$\nbdash-minimality of $\Expl$.

Moreover, assume now by contradiction that $\Expl$ contain all $\mi{Cut}_k$\nbdash-facts and $\mi{Cut}_{k'}$\nbdash-facts, for some $k$ and $k'$ with $k \neq k'$.
Since $\Expl$ is an explanation for $\Query$, the query atom $\mi{ModelWitness}(\mi{IdY},\mi{IdF})$ must be entailed by $\Expl$ for some instantiation
of $\mi{IdY}$ and $\mi{IdF}$.
By construction, this can happen only through the cut rules above.
Assume \Wlog that $k$ is one cut whose associated facts, i.e., the $\mi{Cut}_k$\nbdash-facts and $\mi{Odd}()$ or $\mi{Even}()$ depending on $k$'s parity, are used to derive such a
$\mi{ModelWitness}$\nbdash-fact.
Then, all $\mi{Cut}_{k'}$\nbdash-facts with $k' \neq k$ are not needed for this
derivation.
Removing them from $\Expl$ leaves all atoms of $\DB_\Phi$ in the query
satisfied, preserves the derivation of the same $\mi{AssignY}$-, $\mi{YVal}$-,
and $\mi{ModelWitness}$\nbdash-facts, and therefore still gives an explanation
for $\Query$.
This explanation is strictly smaller than $\Expl$, contradicting $\Expl$'s $\leq$\nbdash-minimality.

Since by these two arguments $\Expl$ contains all $\mi{Cut}_k$\nbdash-facts for exactly one $k$, in order for $\Expl$ to be a $\leq$\nbdash-minimal explanation for $\Query$, it must be that $\Expl$ contains exactly one fact between $\mi{Odd}()$ or $\mi{Even}()$, so to trigger the TGD associated with the $\mi{Cut}_k$\nbdash-facts, depending on $k$'s parity.
Therefore, there must be a $k$ such that $\Expl$ contains $\set{\mi{Cut}_k(1),\dots,\linebreak[0]\mi{Cut}_k(m-k+1)} \cup \set{\mi{Odd}()\mid k \text{ is odd}} \cup \set{\mi{Even}()\mid k \text{ is even}}$.

As there are no other database facts besides those from $\DB_\Phi$, the $\mi{Cut}_\ell$\nbdash-facts, and the facts $\mi{Odd}()$ and $\mi{Even}()$, we have that, for some value $k$, $\Expl = \DB_\Phi \cup \set{\mi{Cut}_k(1),\dots,\linebreak[0]\mi{Cut}_k(m-k+1)} \cup \set{\mi{Odd}()\mid k \text{ is odd}} \cup \set{\mi{Even}()\mid k \text{ is even}}$.
\end{subproof}

By \zcref{subprop_shape_explanations_min_card}, the only possible sets of facts that are candidates for being $\leq$\nbdash-\Minexs for $\Query$ are in one-to-one correspondence with the interpretations $\Interpr{I}_k$ of $\VarSet{y}$, and are defined as follows:
$\Expl_k = \DB_\Phi \cup \set{\mi{Cut}_k(1),\dots,\linebreak[0]\mi{Cut}_k(m-k+1)} \cup \set{\mi{Odd}()\mid k \text{ is odd}} \cup \set{\mi{Even}()\mid k \text{ is even}}$.
Notice that the sizes of the sets $\Expl_k$ differ only for the $\mi{Cut}_k$\nbdash-facts.
Given that the number of $\mi{Cut}_k$\nbdash-facts is $m-k+1$, we can easily state the following.

\begin{subproperty}
\label{subprop_comparison_sizes_explanations_min_card}
Let $\Expl_k$ and $\Expl_{k'}$ be two sets of facts associated with the interpretations $\Interpr{I}_k$ and $\Interpr{I}_{k'}$ of $\VarSet{y}$, respectively.
Then, $\SetSize{\Expl_{k'}} < \SetSize{\Expl_k}$ if and only if $k' > k$.
\end{subproperty} 

We now prove that the construction is indeed a reduction from
\MaxSatSigmaFormula{1} to \MinExRelMinCard{}($\DLFrag{A}$).

\smallbreak

\ProofRightarrowItem
Assume that $\Phi(\VarSet{y})$ is a \yesinst of \MaxSatSigmaFormula{1}, and hence that the Hamming weight of a maximum-Hamming-weight model of $\Phi(\VarSet{y})$ is odd.
We show that there exists a $\leq$\nbdash-\Minex for $\Query$ containing $f^\star$.

By assumption on the structure of $\Phi(\VarSet{y})$, we have that its models have a specific shape (see above).
Let $\Interpr{I}_k$ be the maximum-Hamming-weight model of $\Phi(\VarSet{y})$; 
it holds that $\HammingWeight{\Interpr{I}_k} = k$ and $k$ is odd.

We claim that the set $\Expl_k$ associated with $\Interpr{I}_k$ is a $\leq$\nbdash-\Minex for $\Query$ containing $f^\star$.

Clearly, by the definition of $\Expl_k$, the TGDs added to $\Dep_\Phi$ in this construction, the fact that $\Interpr{I}_k \models \Phi(\VarSet{y})$, and \zcref{theo_encoding_qbsf_satisfaction}, we have that $\KBDef{\Expl_k,\Dep} \models \Query$, and hence $\Expl_k$ is an explanation for $\Query$ containing $f^\star$.

Assume by contradiction that $\Expl_k$ is not $\leq$\nbdash-minimal.
By \zcref{subprop_shape_explanations_min_card}, the only possible different $\leq$\nbdash-\Minex for $\Query$ must be a set $\Expl_{k'}$ for some $k' \neq k$, such that $\SetSize{\Expl_{k'}} < \SetSize{\Expl_k}$.
By \zcref{subprop_comparison_sizes_explanations_min_card}, it must be $k' > k$.
Since $\Expl_{k'}$ is an explanation, by \zcref{theo_encoding_qbsf_satisfaction} and the TGDs added to $\Dep_\Phi$ in this construction, the interpretation $\Interpr{I}_{k'}$ associated with $\Expl_{k'}$ is a model of $\Phi(\VarSet{y})$.
Observe however that $\HammingWeight{\Interpr{I}_{k'}} = k' > k = \HammingWeight{\Interpr{I}_k}$, which contradicts $\Interpr{I}_k$ being a maximum-Hamming-weight model of $\Phi(\VarSet{y})$.
Thus, $\Expl_k$ is a $\leq$\nbdash-\Minex for $\Query$ containing $f^\star = \mi{Odd}()$.

\smallbreak

\ProofLeftarrowItem
Assume that $\tup{\KB,\Query,f^\star}$ is a \yesinst of \MinExRelMinCard{}($\DLFrag{A}$), hence there exists a $\leq$\nbdash-\Minex $\Expl$ for $\Query$ containing $f^\star$.
We show that a maximum-Hamming-weight model of $\Phi(\VarSet{y})$ has odd weight.

By \zcref{subprop_shape_explanations_min_card}, $\Expl$ is of specific shape.
More specifically, $\Expl$ must equal one of those explanations $\Expl_k$ characterized above, for some odd $k$.
Consider the interpretation $\Interpr{I}_k$ of $\VarSet{y}$ associated with $\Expl_k$.
By its definition, the Hamming weight of $\Interpr{I}_k$ is $k$ and is odd.
We claim that $\Interpr{I}_k$ is a maximum-Hamming-weight model of $\Phi(\VarSet{y})$.

Since $\KBDef{\Expl_k,\Dep} \models \Query$, by \zcref{theo_encoding_qbsf_satisfaction} and the TGDs added to $\Dep_\Phi$ in this construction, we have $\Interpr{I}_k \models \Phi(\VarSet{y})$.

Assume by contradiction that $\Interpr{I}_k$ is not a maximum-Hamming-weight model of $\Phi(\VarSet{y})$.
If $\Interpr{I}_k$ is not a maximum-Hamming-weight model of $\Phi(\VarSet{y})$, %
by assumption on the structure of $\Phi(\VarSet{y})$ there must be another interpretation $\Interpr{I}_{k'}$ that is a model of $\Phi(\VarSet{y})$ and $\HammingWeight{\Interpr{I}_{k'}} = k' > k = \HammingWeight{\Interpr{I}_k}$.
Consider now the set of facts $\Expl_{k'}$ associated with $\Interpr{I}_{k'}$.
Since $\Interpr{I}_{k'} \models \Phi(\VarSet{y})$, by \zcref{theo_encoding_qbsf_satisfaction} and the TGDs added to $\Dep_\Phi$ in this construction, %
$\Expl_{k'}$ is an explanation of $\Query$.
However, by $k' > k$ and \zcref{subprop_comparison_sizes_explanations_min_card}, $\Expl_{k'}$ is an explanation such that $\SetSize{\Expl_{k'}} < \SetSize{\Expl_k}$, contradicting $\Expl_k$ being a $\leq$\nbdash-\Minex.
Thus, $\Interpr{I}_k$ is a maximum-Hamming-weight model of $\Phi(\VarSet{y})$ with odd weight.
\end{proof}

From the previous theorem, we can obtain the hardness results for the other two problems.

\begin{corollary}
\MinExNecMinCard{}($\DLFrag{A})$ and \MinExIrrelMinCard{}($\DLFrag{A}$) are \PNExpLogh in the $\mi{ba}$\nbdash-com\-bined (resp., combined) complexity.
Hardness holds even on instances whose queries are \BCQs.
\end{corollary}

\begin{proof}
By inspection of the proof of \zcref{thm_minex-rel-card_A_ba_hardness}, \MinExNecMinCard{}($\DLFrag{A}$) is proven \PNExpLogh, as the reduction in the proof is such that only one cardinality\nbdash-minimal explanation exists.
From the hardness of \MinExNecMinCard{}($\DLFrag{A}$) follows the \PNExpLogh{}ness of \MinExIrrelMinCard{}($\DLFrag{A}$) as well, because \PNExpLog is closed under complement.
\end{proof}

The following result shows the \PNExph{}ness of \MinExRelMinWeight{}($\DLFrag{A}$).

\begin{theorem}
\label{theo_minex_rel_min_weight_hard}
\MinExRelMinWeight{}($\DLFrag{A}$) is \PNExph in the $\mi{ba}$\nbdash-combined (resp., combined) complexity.
Hardness holds even on instances whose queries are \BCQs.
\end{theorem}

\begin{proof}
We exhibit a polynomial reduction from the problem \LexMaxSigmaFormula{1}.
In particular, from an instance $\Phi(\VarSet{y}) = \linebreak[0] (\exists f)(\forall \VarSet{x}) \PrefixMatrixSeparator \phi(f,\VarSet{x},\VarSet{y})$ of \LexMaxSigmaFormula{1}, with $\VarSet{y} = \set{y_m,\dots,y_1}$ a set of propositional variables ordered from the most significant to the least significant one, we build an instance $\tup{\KB,\Query,f^\star}$ of \MinExRelMinWeight{}($\DLFrag{A}$), with $\KB = \KBDef{\DB,\Dep}$ a knowledge base with weighted database facts, $\Query$ a \BCQ, and $f^\star \in \DB$ a fact.
By \zcref{theo_complexity_LexMax}, we assume $\Phi(\VarSet{y})$ to be a satisfiable \CNF \SimpleCompactSigmaKOddFormula{1}.

Let $\mathcal{M}_\Phi = \KBDef{\DB_\Phi,\Dep_\Phi}$ be the ontology of \zcref{theo_encoding_qbsf_satisfaction}.
The reduction in this proof extends $\mathcal{M}_\Phi$.

\medbreak

\noindent
\emph{(The database).}
The weighted database $\DB$ is obtained by augmenting the database $\DB_\Phi$ as follows.

All database facts of $\DB_\Phi$ have weight $0$.
Furthermore, we add to the database the weighted facts
\begin{gather*}
  \mi{AssignY}(\mi{id}),\\
  \mi{YVal}(\mi{id},y_j,\DBvaltrue)
  \qquad\text{and}\qquad
  \mi{YVal}(\mi{id},y_j,\DBvalfalse),
  \rlap{\quad\text{for every $1 \leq j \leq m$},}
\end{gather*}
where `$\mi{id}$' is the only object identifying
an interpretation of $\VarSet{y}$, and whose weights are:
$w(\mi{AssignY}(\mi{id})) = 0$, $w(\mi{YVal}(\mi{id},y_j,\DBvaltrue))=2^m$, and $w(\mi{YVal}(\mi{id},y_j,\DBvalfalse))=2^m+2^{j-1}$---notice the inversion, that is, we assign higher weights to facts associated with $\valfalse$ values rather than to facts associated with $\valtrue$ values.

\smallskip

\noindent
\emph{(The program).}
The program $\Dep$ is the same as the program $\Dep_\Phi$.

\smallskip

\noindent
\emph{(The query).}
The query $\Query$ is the \BCQ \[\Query = \mi{AssignY}(\mi{id}), \mi{YVal}(\mi{id},y_1,T_1), \dots, \mi{YVal}(\mi{id},y_m,T_m), \DB_\Phi, \mi{ModelWitness}(\mi{id},\mi{IdF}),\]
where with $\DB_\Phi$ in the query we mean the conjunction of all facts of $\DB_\Phi$.

\smallskip

\noindent
\emph{(The distinguished fact).}
Finally, the distinguished fact is $f^\star = \mi{YVal}(\mi{id},y_1,\DBvaltrue)$.

\smallskip

Observe that the construction can be computed in polynomial time in the size of $\Phi(\VarSet{y})$.
Although the largest weight is $2^m+2^{m-1}$, all weights are representable with $O(m)$ bits, and hence the weighting function is produced in polynomial time.
The predicates have bounded arity, and acyclicity follows from \zcref{theo_encoding_qbsf_satisfaction} since no new TGDs are added.
Furthermore, the query is entailed by $\KB$, and this follows from \zcref{theo_encoding_qbsf_satisfaction}, the fact that we are assuming $\Phi(\VarSet{y})$ satisfiable, and the presence of the \AssignYFact and the $\mi{YVal}$\nbdash-facts in the database.

\medskip

\noindent
We now prove the transformation correctness by showing that the lexicographic\nbdash-maximum model $\Interpr{I}$ of $\Phi(\VarSet{y})$ is such that $\EvalInterpr{y_1}{\Interpr{I}} = \valtrue$ iff there exists a $\leq_w$\nbdash-\Minex for $\Query$ containing $f^\star$.

We start by highlighting a couple of properties of the $\leq_w$\nbdash-minimal explanations of $\Query$.

\begin{subproperty}
\label{subprop_shape_explanations_min_weight}
Let $\Expl \subseteq \DB$ be a $\leq_w$\nbdash-\Minex for $\Query$.
Then, there exists an interpretation $\Interpr{I}$ of $\VarSet{y}$ such that $\Expl = \DB_\Phi \cup \set{\mi{AssignY}(\mi{id})} \cup \set{\mi{YVal}(\mi{id},y_j,\DBvalfalse) \mid \EvalInterpr{y_j}{\Interpr{I}} = \valfalse} \cup \set{\mi{YVal}(\mi{id},y_j,\DBvaltrue) \mid \EvalInterpr{y_j}{\Interpr{I}} = \valtrue}$.
\end{subproperty}

\begin{subproof}
Since $\Expl$ is an explanation for $\Query$, $\Expl$ must contain all facts from $\DB_\Phi$ and the fact $\set{\mi{AssignY}(\mi{id})}$.
Furthermore, in $\Query$ there are the atoms $\mi{YVal}(\mi{id},y_j,T_j)$, for every $1 \leq j \leq m$.
Hence, $\Expl$ must contain at least one fact among $\mi{YVal}(\mi{id},y_j,\DBvalfalse)$ and $\mi{YVal}(\mi{id},y_j,\DBvaltrue)$, for every $y_j$.
However, we can actually show that $\Expl$ contains exactly one fact among $\mi{YVal}(\mi{id},y_j,\DBvalfalse)$ and $\mi{YVal}(\mi{id},y_j,\DBvaltrue)$, for every $y_j$.

Assume by contradiction that this is not the case, and hence there is a variable $y_\ell$ such that both $\mi{YVal}(\mi{id},y_\ell,\DBvalfalse)$ and $\mi{YVal}(\mi{id},y_\ell,\DBvaltrue)$ belong to $\Expl$. 
Observe the following.
Since $\Phi(\VarSet{y})$ is assumed to be satisfiable, there is an interpretation $\Interpr{J}$ of $\VarSet{y}$ such that $\Interpr{J} \models \Phi(\VarSet{y})$.
Consider the set of facts $\Expl_{\Interpr{J}} = \DB_\Phi \cup \set{\mi{AssignY}(\mi{id})} \cup \set{\mi{YVal}(\mi{id},y_j,\DBvalfalse) \mid \EvalInterpr{y_j}{\Interpr{J}} = \valfalse} \cup \set{\mi{YVal}(\mi{id},y_j,\DBvaltrue) \mid \EvalInterpr{y_j}{\Interpr{J}} = \valtrue}$.
By \zcref{theo_encoding_qbsf_satisfaction}, $\KBDef{\Expl_{\Interpr{J}},\Dep_\Phi} \models \Query$, hence $\Expl_{\Interpr{J}}$ is an explanation for $\Query$.
Observe that the weight of $\Expl_{\Interpr{J}}$ is:
\[
  w(\Expl_{\Interpr{J}})
  =
  \sum_{j : \EvalInterpr{y_j}{\Interpr{J}} = \valtrue} 2^m
  +
  \sum_{j : \EvalInterpr{y_j}{\Interpr{J}} = \valfalse} (2^m + 2^{j-1}) 
  =
  m \cdot 2^m
  +
  \sum_{j : \EvalInterpr{y_j}{\Interpr{J}} = \valfalse} 2^{j-1} ,
\]
as only the $\mi{YVal}$\nbdash-facts in $\Expl_{\Interpr{J}}$ have non\nbdash-zero weight.
Notice that $\sum_{j : \EvalInterpr{y_j}{\Interpr{J}} = \valfalse} 2^{j-1}$ is \emph{strictly} smaller than $2^m$, as $\sum_{j=1}^{m}2^{j-1}=2^m-1$.
Hence, we have that $w(\Expl_{\Interpr{J}}) < (m+1)\cdot 2^m$.

Instead, the weight of $\Expl$, which contains both $\mi{YVal}(\mi{id},y_\ell,\DBvalfalse)$ and $\mi{YVal}(\mi{id},y_\ell,\DBvaltrue)$, is:
\begin{multline*}
  w(\Expl)
  =
  \sum_{j : \mi{YVal}(\mi{id},y_j,\DBvaltrue) \in \Expl} 2^m
  +
  \sum_{j : \mi{YVal}(\mi{id},y_j,\DBvalfalse) \in \Expl} (2^m + 2^{j-1})
  \geq {} \\
  (m+1) \cdot 2^m
  +
  \sum_{j : \mi{YVal}(\mi{id},y_j,\DBvalfalse) \in \Expl} 2^{j-1}
  \geq
  (m+1) \cdot 2^m.
\end{multline*}
Therefore, $w(\Expl_{\Interpr{J}}) < w(\Expl)$.
However, $\Expl_{\Interpr{J}}$ is an explanation, hence $\Expl$ is not $\leq_w$\nbdash-minimal: a contradiction.

As the are no other database facts besides those from $\DB_\Phi$, the $\mi{YVal}$\nbdash-facts, and the fact $\mi{AssignY}(\mi{id})$, there must be an interpretation $\Interpr{I}$ of $\VarSet{y}$ such that $\Expl = \DB_\Phi \cup \set{\mi{AssignY}(\mi{id})} \cup \set{\mi{YVal}(\mi{id},y_j,\DBvalfalse) \mid \EvalInterpr{y_j}{\Interpr{I}} = \valfalse} \cup \set{\mi{YVal}(\mi{id},y_j,\DBvaltrue) \mid \EvalInterpr{y_j}{\Interpr{I}} = \valtrue}$.
\end{subproof}

By \zcref{subprop_shape_explanations_min_weight}, the only possible sets of facts that are candidates for being $\leq_w$\nbdash-\Minexs for $\Query$ are in one-to-one correspondence with the interpretations $\Interpr{I}$ of $\VarSet{y}$, and are defined as follows:
$\Expl_{\Interpr{I}} = \DB_\Phi \cup \set{\mi{AssignY}(\mi{id})} \cup \set{\mi{YVal}(\mi{id},y_j,\DBvalfalse) \mid \EvalInterpr{y_j}{\Interpr{I}} = \valfalse} \cup \set{\mi{YVal}(\mi{id},y_j,\DBvaltrue) \mid \EvalInterpr{y_j}{\Interpr{I}} = \valtrue}$.

We can now prove that there is a strong relationship between the lexicographic order of two interpretations $\Interpr{I}$ and $\Interpr{J}$ of $\VarSet{y}$ and the weight of the sets of facts $\Expl_{\Interpr{I}}$ and $\Expl_{\Interpr{J}}$.

\begin{subproperty}
\label{subprop_comparison_weights_explanations_min_weight}
Let $\Expl_{\Interpr{I}}$ and $\Expl_{\Interpr{J}}$ be two sets of facts associated with the interpretations $\Interpr{I}$ and $\Interpr{J}$ of $\VarSet{y}$, respectively.
Then, $\Interpr{I} \lexsucc \Interpr{J}$ if and only if $w(\Expl_{\Interpr{I}}) < w(\Expl_{\Interpr{J}})$.
\end{subproperty} 

\begin{subproof}
\ProofRightarrowItem
Assume that $\Interpr{I} \lexsucc \Interpr{J}$.
Hence, there exists an index $i$ such that $\EvalInterpr{y_j}{\Interpr{I}} = \EvalInterpr{y_j}{\Interpr{J}}$, for all $m \geq j > i$, $\EvalInterpr{y_i}{\Interpr{I}} = \valtrue$ and $\EvalInterpr{y_i}{\Interpr{J}} = \valfalse$.

Consider the interpretation $\Interpr{I}'$ defined as follows:
$\EvalInterpr{y_j}{\Interpr{I}} = \EvalInterpr{y_j}{\Interpr{I}'}$, for all $m \geq j \geq i$, and $\EvalInterpr{y_j}{\Interpr{I}'} = \valfalse$ for all $i > j \geq 1$.
Consider also the interpretation $\Interpr{J}'$ defined as follows:
$\EvalInterpr{y_j}{\Interpr{J}} = \EvalInterpr{y_j}{\Interpr{J}'}$, for all $m \geq j \geq i$, and $\EvalInterpr{y_j}{\Interpr{J}'} = \valtrue$ for all $i > j \geq 1$.
Notice that $\Interpr{I} \lexsucceq \Interpr{I}' \lexsucc \Interpr{J}' \lexsucceq \Interpr{J}$.

We now focus on the weights of the four associated set of facts $\Expl_{\Interpr{I}}$, $\Expl_{\Interpr{I}'}$, $\Expl_{\Interpr{J}}$, $\Expl_{\Interpr{J}'}$.

Consider first $\Expl_{\Interpr{I}}$ and $\Expl_{\Interpr{I}'}$.
Since $\Expl_{\Interpr{I}}$ and $\Expl_{\Interpr{I}'}$ contain the same $\mi{YVal}$\nbdash-facts for $y_j$ with $m \geq j \geq i$, and $\mi{YVal}(\mi{id},y_j,\DBvalfalse) \in \Expl_{\Interpr{I}'}$, for all $i > j \geq 1$, by the fact weights, we have $w(\Expl_{\Interpr{I}}) \leq w(\Expl_{\Interpr{I}'})$.

Consider now $\Expl_{\Interpr{J}}$ and $\Expl_{\Interpr{J}'}$.
Since $\Expl_{\Interpr{J}}$ and $\Expl_{\Interpr{J}'}$ contain the same $\mi{YVal}$\nbdash-facts for $y_j$ with $m \geq j \geq i$, and $\mi{YVal}(\mi{id},y_j,\DBvaltrue) \in \Expl_{\Interpr{J}'}$, for all $i > j \geq 1$, by the fact weights, we have $w(\Expl_{\Interpr{J}'}) \leq w(\Expl_{\Interpr{J}})$.

Consider now $\Expl_{\Interpr{I}'}$ and $\Expl_{\Interpr{J}'}$.
We have that $\Expl_{\Interpr{I}'}$ and $\Expl_{\Interpr{J}'}$ contain the same $\mi{YVal}$\nbdash-facts for $y_j$ with $m \geq j > i$.
Hence, these facts do not contribute to make the weights of the two sets different.
Let us hence focus only on the variables $y_j$, with $i \geq j \geq 1$.
In $\Expl_{\Interpr{I}'}$ there are the facts $\mi{YVal}(\mi{id},y_j,\DBvalfalse)$, for all $i > j \geq 1$, whose weight is $\sum_{j = 1}^{i-1} (2^m + 2^{j-1}) = (i-1)\cdot 2^m + 2^{i-1}-1$.
Moreover, $\Expl_{\Interpr{I}'}$ contains $\mi{YVal}(\mi{id},y_i,\DBvaltrue)$, whose weight is $2^m$.
In $\Expl_{\Interpr{J}'}$ there are the facts $\mi{YVal}(\mi{id},y_j,\DBvaltrue)$, for all $i > j \geq 1$, whose weight is $\sum_{j = 1}^{i-1} 2^m = (i-1)\cdot 2^m$.
Furthermore, $\Expl_{\Interpr{J}'}$ contains $\mi{YVal}(\mi{id},y_i,\DBvalfalse)$, whose weight is $2^m + 2^{i-1}$.
Therefore, over the $\mi{YVal}$\nbdash-facts for the variables $y_j$ with $i \geq j \geq 1$, the two sets $\Expl_{\Interpr{I}'}$ and $\Expl_{\Interpr{J}'}$ have weights
$i\cdot 2^m + 2^{i-1}-1$ and  $i\cdot 2^m + 2^{i-1}$, respectively.
Hence, $w(\Expl_{\Interpr{I}'}) < w(\Expl_{\Interpr{J}'})$.

By putting all together we have $w(\Expl_{\Interpr{I}}) \leq w(\Expl_{\Interpr{I}'}) < w(\Expl_{\Interpr{J}'}) \leq w(\Expl_{\Interpr{J}})$. 

\ProofLeftarrowItem
Let us assume that $\Interpr{I} \not \lexsucc \Interpr{J}$.
By this, either $\Interpr{I} = \Interpr{J}$ or $\Interpr{I} \lexprec \Interpr{J}$.

If $\Interpr{I} = \Interpr{J}$, then we have $w(\Expl_{\Interpr{I}}) = w(\Expl_{\Interpr{J}})$ (i.e., $w(\Expl_{\Interpr{I}}) \not < w(\Expl_{\Interpr{J}})$).
If $\Interpr{I} \lexprec \Interpr{J}$, an argument symmetric to the one above shows that $w(\Expl_{\Interpr{I}}) > w(\Expl_{\Interpr{J}})$ (i.e., $w(\Expl_{\Interpr{I}}) \not < w(\Expl_{\Interpr{J}})$).
\end{subproof}

We now prove that the construction is indeed a reduction from
\LexMaxSigmaFormula{1} to \MinExRelMinWeight{}($\DLFrag{A}$).

\smallbreak

\ProofRightarrowItem
Assume that $\Phi(\VarSet{y})$ is a \yesinst of \LexMaxSigmaFormula{1}, and hence that the lexicographic\nbdash-maximum model $\Interpr{I}$ of $\Phi(\VarSet{y})$ assigns $\valtrue$ to $y_1$.
We show that there exists a $\leq_w$\nbdash-\Minex for $\Query$ containing $f^\star$.

Consider the set $\Expl_{\Interpr{I}}$ associated with $\Interpr{I}$.
Because $\EvalInterpr{y_1}{\Interpr{I}} = \valtrue$, we have that $f^\star = \mi{YVal}(\mi{id},y_1,\DBvaltrue) \in \Expl_{\Interpr{I}}$.
Moreover, by \zcref{theo_encoding_qbsf_satisfaction}, $\KBDef{\Expl_{\Interpr{I}},\Dep_\Phi} \models \Query$, hence $\Expl_{\Interpr{I}}$ is an explanation for $\Query$.

Assume by contradiction that $\Expl_{\Interpr{I}}$ is not $\leq_w$\nbdash-minimal.
Since $\Phi(\VarSet{y})$ is satisfiable, there is a different $\leq_w$\nbdash-\Minex $\Expl$ for $\Query$ such that $w(\Expl) < w(\Expl_{\Interpr{I}})$.
By \zcref{subprop_shape_explanations_min_weight}, if $\Expl$ is a $\leq_w$\nbdash-\Minex, there must be an interpretation $\Interpr{J}$ such that $\Expl = \Expl_{\Interpr{J}}$.
By \zcref{theo_encoding_qbsf_satisfaction}, since $\Expl_{\Interpr{J}}$ is an explanation, we have that $\Interpr{J} \models \Phi(\VarSet{y})$.
From \zcref{subprop_comparison_weights_explanations_min_weight} and $w(\Expl_{\Interpr{J}}) < w(\Expl_{\Interpr{I}})$, it follows that $\Interpr{J} \lexsucc \Interpr{I}$, contradicting that $\Interpr{I}$ is a lexicographic\nbdash-maximum model of $\Phi(\VarSet{y})$.
Thus, $\Expl_{\Interpr{I}}$ is a $\leq_w$\nbdash-\Minex for $\Query$ containing $f^\star$.

\smallbreak

\ProofLeftarrowItem
Assume that $\tup{\KB,\Query,f^\star}$ is a \yesinst of \MinExRelMinWeight{}($\DLFrag{A}$), hence there is a $\leq_w$\nbdash-\Minex $\Expl$ for $\Query$ containing $f^\star$.
We show that the lexicographic\nbdash-maximum model of $\Phi(\VarSet{y})$ assigns $\valtrue$ to $y_1$.

By \zcref{subprop_shape_explanations_min_weight}, since $\Expl$ is a $\leq_w$\nbdash-\Minex, there exists an interpretation $\Interpr{I}$ of $\VarSet{y}$ such that $\Expl = \Expl_{\Interpr{I}}$.
Since $\Expl_{\Interpr{I}}$ is an explanation, by \zcref{theo_encoding_qbsf_satisfaction} we have that $\Interpr{I}$ is a model of $\Phi(\VarSet{y})$.
Notice moreover that, because $f^\star = \mi{YVal}(\mi{id},y_1,\DBvaltrue) \in \Expl_{\Interpr{I}}$, we have that $\Interpr{I}$ assigns $\valtrue$ to $y_1$.

Let us now assume by contradiction that $\Interpr{I}$ is not lexicographic\nbdash-maximum.
There hence exists an interpretation $\Interpr{J}$ such that $\Interpr{J} \lexsucc \Interpr{I}$ and $\Interpr{J} \models \Phi(\VarSet{y})$.
Consider the associated set of facts $\Expl_{\Interpr{J}}$.
Since $\Interpr{J} \models \Phi(\VarSet{y})$, by \zcref{theo_encoding_qbsf_satisfaction} we have that $\Expl_{\Interpr{J}}$ is an explanation.
However, since $\Interpr{J} \lexsucc \Interpr{I}$, by \zcref{subprop_comparison_weights_explanations_min_weight} we have that $w(\Expl_{\Interpr{J}}) < w(\Expl_{\Interpr{I}})$, contradicting the $\leq_w$\nbdash-minimality of $\Expl_{\Interpr{I}}$.
Thus, $\Interpr{I}$ is the lexicographic\nbdash-maximum model of $\Phi(\VarSet{y})$, and we have that $\EvalInterpr{y_1}{\Interpr{I}} = \valtrue$.
\end{proof}

Again, from the previous theorem, we can obtain the hardness results for the other two problems.

\begin{corollary}
\MinExNecMinWeight{}($\DLFrag{A})$ and \MinExIrrelMinWeight{}($\DLFrag{A}$) are \PNExph in the $\mi{ba}$\nbdash-com\-bined (resp., combined) complexity.
Hardness holds even on instances whose queries are \BCQs.
\end{corollary}

\begin{proof}
By inspection of the proof of \zcref{theo_minex_rel_min_weight_hard}, \MinExNecMinWeight{}($\DLFrag{A}$) is proven \PNExph, as the reduction in the proof is such that only one weight\nbdash-minimal explanation exists.
From the hardness of \MinExNecMinWeight{}($\DLFrag{A}$) follows the \PNExph{}ness of \MinExIrrelMinWeight{}($\DLFrag{A}$) as well, because \PNExp is closed under complement.
\end{proof}

\resumetoc

\begin{appendices}

\section{Simple Maths Complexity}
\label{sec_maths_complexity}

We review some results on the complexity of evaluating a few mathematical functions over input \emph{integers} represented in binary or unary;
the output is always represented in binary.
Results are summarized in \zcref{tab_maths_complexity}.

{%
\renewcommand{\arraystretch}{1.5}
\begin{table}[t]
\centering%
\begin{tabular}{ l m{3.16cm} l l l }
\toprule
  \textbf{Function} & \textbf{Input} & \textbf{Output size} & \textbf{Space compl.} & \textbf{Time compl.} \\
  \midrule
  \rowcolor[gray]{0.925}
  $a + b$ & Integers $a$ and $b$ & $\max \set{\StringLength{a}{,}\StringLength{b}} + 1$ & $O(\log (\StringLength{a} {+} \StringLength{b}))$ & $O(\mi{poly}(\StringLength{a} {+} \StringLength{b}))$ \\
  \rowcolor[gray]{0.975}
  $a - b$ & Integers $a$ and $b$\;\!\raisebox{.25ex}{*} & $\StringLength{a}$ & $O(\log (\StringLength{a} {+} \StringLength{b}))$ & $O(\mi{poly}(\StringLength{a} {+} \StringLength{b}))$ 
  \\
  \rowcolor[gray]{0.925}
  $a \cdot b$ & Integers $a$ and $b$ & $\StringLength{a} + \StringLength{b}$ & $O(\log (\StringLength{a} {+} \StringLength{b}))$ & $O(\mi{poly}(\StringLength{a} {+} \StringLength{b}))$ \\
  \rowcolor[gray]{0.975}
  $\lfloor a / b \rfloor$ & Integers $a$ and $b$ & $\max \set{0{,}\StringLength{a} {-} \StringLength{b}} + 1$  & $O(\log (\StringLength{a} {+} \StringLength{b}))$ & $O(\mi{poly}(\StringLength{a} {+} \StringLength{b}))$ \\
  \rowcolor[gray]{0.925}
  $a^b$ & Integers $a$ and $b$ & $O(2^{\StringLength{b}}\cdot \StringLength{a})$ & $O(\StringLength{a^b})$ & $O(\mi{poly}({\StringLength{a^b}}))$ \\
  \rowcolor[gray]{0.975}
  $\lfloor \log a \rfloor + 1$ & An integer $a$ & $\lfloor \log \StringLength{a} \rfloor + 1$ & $O(\log \StringLength{a})$ & $O(\StringLength{a})$ \\
  \rowcolor[gray]{0.925}
  $p(\StringLength{w})$ & An integer $n$ encoded in unary as the length of the input string $w$\;\!\raisebox{.25ex}{\textdagger} & $O(\log \StringLength{w})$ & $O(\StringLength{p(\StringLength{w})})$ & $O(\mi{poly}({\StringLength{p(\StringLength{w})}}))$ \\
  \rowcolor[gray]{0.975}
  $\iExp{i}{p(\StringLength{w})}$ & An integer $n$ encoded in unary as the length of the input string $w$\;\!\raisebox{.25ex}{\textdaggerdbl} & $\iExp{i-1}{p(\StringLength{w})} + 1$ & $O(\iExp{i-2}{p(\StringLength{w})})$ & $O(\iExp{i-1}{p(\StringLength{w})})$ \\
  \bottomrule
\end{tabular}
\caption{%
Complexity of evaluating some mathematical functions.
Input integers are non-negative and represented in binary, unless stated otherwise.
The intermediate results and integer variables in the algorithms, likewise their output, are represented in binary.
For integers $a$ and $b$, we denote by $\StringLength{a}$ and $\StringLength{b}$ the size of their binary representations, whereas, for a string $w$, we denote by $\StringLength{w}$ its length.
By $p(n)$ we mean the \emph{value} of the polynomial $p$ at $n$, hence $p(\StringLength{w})$ is the value of $p$ at $n = \StringLength{w}$.
For a function $f$, we use $\StringLength{f(n)}$ to refer to the size of the binary representation of the \emph{value} of $f$ at $n$.
For example, for the time complexity of exponentiation, $O(\mi{poly}({\StringLength{a^b}}))$ means that it is polynomial in the binary representation size of the value $a^b$. 
Complexity results for $a^b$ assume the use of \zcref{alg_exponentiation_by_squaring}.
Complexity results involving evaluations of polynomials assume the use of \zcref{alg_polynomial_evaluation} (Horner's rule for polynomials' evaluation would not asymptotically improve over this).
\raisebox{.25ex}{*}Integers $a$ and $b$ are such that $a \geq b$.
\raisebox{.25ex}{\textdagger}The polynomial $p$ is fixed.
\raisebox{.25ex}{\textdaggerdbl}The polynomial $p$ and the integer $i \geq 1$ are fixed.}
\label{tab_maths_complexity}
\end{table}
}

\paragraph{Arithmetic.}\hspace{0pt}\newline
(Iterated) addition, subtraction, (iterated) multiplication, and division, are computable in logspace, and so in polynomial time, 
in the \emph{size} of the \emph{binary}\nbdash-represented input numbers \mbox{(see \cite{Immerman1999,HesseAB2002}, and references therein)}.

\begin{algorithm}[t]
\caption{An exponentiation by squaring algorithm.}
\label{alg_exponentiation_by_squaring}

\BlankLine

\KwIn{Integers $a$ and $b$ (in binary).}
\KwOut{The value $a^b$ (in binary).}

\BlankLine

\nonl \Func{\AlgPow{$a,b$}}{

    $\mi{result} \leftarrow 1$\;
    $\mi{multiplier} \leftarrow a$\;
    
    \While{$b > 0$}{
        \If{$b \bmod 2 = 1$}{\nllabel{alg_exp_mod_div_by_two}
            $\mi{result} \leftarrow \mi{result} \cdot \mi{multiplier}$\;
        }
    
        $b \leftarrow \lfloor b/2 \rfloor$ \nllabel{alg_exp_int_div_by_two} \;
        $\mi{multiplier} \leftarrow \mi{multiplier} \cdot \mi{multiplier}$\;
    }
    
    \Return{$\mi{result}$}\;
}
\end{algorithm}

\paragraph{Exponentiation.}\hspace{0pt}\newline
Let us now focus on the following problem:
given binary-represented integers $a$ and $b$, compute $a^b$.
Let $\StringLength{a}$ and $\StringLength{b}$ denote the size of the binary representation of $a$ and $b$, respectively.
For the reader's convenience, \zcref{alg_exponentiation_by_squaring} reports an exponentiation by squaring algorithm;
remember that `$b \bmod 2$' (\zcref{alg_exp_mod_div_by_two}) and `$\lfloor b/2 \rfloor$' (\zcref{alg_exp_int_div_by_two}) can be evaluated in linear time, as these can be carried out by looking at, and by dropping, the least significant bit of $b$, respectively.
An exponentiation by squaring algorithm evaluates $a^b$ via $2{\cdot}\StringLength{b}$ multiplications.
The multiplications can be evaluated in logarithmic space and in polynomial time (see the paragraph ``Arithmetic'') in the size of the intermediate results, which are the variables $\mi{result}$ and $\mi{multiplier}$ in \zcref{alg_exponentiation_by_squaring}.
The intermediate results, however, can be considerably bigger than the input during the exponentiation procedure.

The size of the intermediate results is bounded by the output size, which is $O(2^{\StringLength{b}}\cdot\StringLength{a})$, i.e., exponential in the input size---simply observe that $a^b$ is defined as the multiplication of $a$ by itself $b$ times, each multiplication adds $\StringLength{a}$ bits to the result, and the value $b$ is $O(2^{\StringLength{b}})$.
Besides the space to store the operands, extra space is needed to carry out the multiplications.
This extra work space is logarithmic in the exponentially sized operands, and hence polynomial in the input size.
Hence, the size of the intermediate results is dominant, and the space complexity of this algorithm is $O(2^{\StringLength{b}}\cdot\StringLength{a})$.
The algorithm's time complexity is polynomial in the size of the intermediate results, as it performs a linear number of multiplications, each taking polynomial time \Wrt the operands' size.
By this, the time complexity of this algorithm is a polynomial of $O(2^{\StringLength{b}}\cdot\StringLength{a})$.

\paragraph{Logarithms.}\hspace{0pt}\newline
We now focus on the computation of
the integer logarithm to the base~$2$ of a binary\nbdash-represented integer $a$.
We consider the function $\lfloor \log a \rfloor + 1 = \lceil \log (a+1) \rceil$, as this is just the number of bits needed to represent $a$.
This can easily be evaluated by counting the number of bits of $a$'s binary representations excluding all the leading zeroes.
This can be carried out in linear time, and in logarithmic space (to store the counter), which is also the output size (see, e.g., \cite[Chapter~16, ``Incrementing a binary counter'']{CormenLRS2022}).

\begin{algorithm}[t]
\caption{A simple algorithm to evaluate (fixed) polynomials.}
\label{alg_polynomial_evaluation}

\BlankLine

\KwIn{A string $w$ whose length $\StringLength{w}$ represents in unary notation the number $n = \StringLength{w}$.}
\KwOut{The value $p(n)$, represented in binary, of the (fixed) polynomial $p$ at $n$.}

\BlankLine

\tcc*[h]{The polynomial $p$ is fixed; $k$ is its degree, and ${\set{a_i}}_{0 \leq i \leq k}$ are its $i$\nbdash-th degree coefficients; these values are directly encoded in the algorithm}

\BlankLine

\nonl \Func{\AlgFixedPolEval{$w$}}{
    
    $n \leftarrow \StringLength{w}$ \nllabel{alg_pol_from_unary_to_binary}\;
    $\mi{result} \leftarrow a_0$\;
    $\mi{power}\_\mi{of}\_n \leftarrow 1$\;
    
    \For{$i \leftarrow 1$ \KwTo $k$}{
        $\mi{power}\_\mi{of}\_n \leftarrow \mi{power}\_\mi{of}\_n \cdot n$\;
        $\mi{result} \leftarrow \mi{result} + a_i \cdot \mi{power}\_\mi{of}\_n$\;
    }
    
    \Return{$\mi{result}$}\;
}
\end{algorithm}

\paragraph{Polynomials.}\hspace{0pt}\newline
Let $p(n) = a_0 + a_1 {\cdot} n + a_2 {\cdot} n^2 + \dots + a_k {\cdot} n^k$ be a univariate polynomial of degree~$k$, where ${\set{a_i}}_{0 \leq i \leq k}$ are its $i$\nbdash-th degree coefficients.
We assume $p$ to be fixed, that is, $p$ is not part of the input to the algorithm.
Let us consider the complexity of computing the value $p(n)$, i.e., the value of $p$ at $n$, where $n$ is the length of an input string $w$.
Consider the (not so) naive \zcref{alg_polynomial_evaluation} for polynomials' evaluation---Horner's rule for polynomials' evaluation does not asymptotically improve over this.
\zcref[S]{alg_polynomial_evaluation} first converts the input from unary to binary (\zcref{alg_pol_from_unary_to_binary});
this can be carried out in linear time and logarithmic space (see, e.g., \cite[Chapter~16, ``Incrementing a binary counter'']{CormenLRS2022}).
Then, \zcref{alg_polynomial_evaluation} evaluates $p(n)$ via $2 {\cdot} k$ multiplications and $k$ additions over the intermediate results, which are the variables $\mi{result}$ and $\mi{power}\_\mi{of}\_n$.
These operations are computable in logarithmic space and polynomial time in the size of the intermediate results (see the paragraph ``Arithmetic'').

Let $a = \max \set{a_0,a_1,\dots,a_k}$.
We have $p(n) = \sum_{i = 0}^{k} a_i \cdot n^i \leq \sum_{i = 0}^{k} a \cdot n^i \leq \sum_{i = 0}^{k} a \cdot n^k = (k+1) \cdot a \cdot n^k$.
By this, if $\StringLength{a}$, $\StringLength{k+1}$, and $\StringLength{n}$, denote the binary\nbdash-representation size of the values $a$, $k+1$, and $n$, respectively, the size of $(k+1) \cdot a$ is $\StringLength{k + 1} + \StringLength{a}$ and the size of $n^k$ is $k \cdot \StringLength{n}$, by which the overall output size $\StringLength{p(n)}$ is bounded by $\StringLength{k+1} + \StringLength{a} + k \cdot \StringLength{n}$ (see \zcref{tab_maths_complexity} for the output size of a multiplication, and remember that $n^k$ is the multiplication of $n$ by itself $k$ times).
Since $n = \StringLength{w}$ and hence $\StringLength{n} = \lfloor\log \StringLength{w}\rfloor + 1$, the output size $\StringLength{p(n)}$ is $O(\log \StringLength{w})$, as $k$ and $a$ are constants due to $p$ being a fixed polynomial.

The space complexity of the algorithm is characterized by the space needed to store the intermediate results, plus the extra work space needed to carry out the arithmetic operations over them.
Clearly, the output size bounds the size of the intermediate results.
The extra space needed to carry out the operations, on the other hand, is logarithmic in the logarithmically sized operands, and hence doubly\nbdash-logarithmic in the input size.
Thus, the size of the intermediate results is dominant, and the space complexity of this algorithm is $O(\log \StringLength{w})$.

The algorithm's time complexity is polynomial in the intermediate results' size, as we perform a constant number of operations over the intermediate results (remember that the polynomial $p$ is fixed), each of them taking polynomial time in the operands' size.
Since the operands' size is bounded by the size of the output, the time complexity of this algorithm is $O(\mi{poly}({\StringLength{p(\StringLength{w})}}))$.

\paragraph{Iterated exponentials of polynomials.}\hspace{0pt}\newline
Consider the computation of $\iExp{i}{p(n)}$, for a fixed integer $i \geq 1$ and a fixed polynomial $p$, and $n$ is the length of an input string $w$---see \zcref{sec_prelim_iterated_exponentials} for the definition of iterated exponentials.
We proceed by induction on~$i$.

Let us first focus on the base case, i.e., when $i = 1$ and hence when computing $2^{p(n)}$.
The value $2^{p(n)}$ is easily represented in binary by a bit `$1$' followed by a number of bits `$0$' equaling the \emph{value} $p(n)$.
Thus, since $n = \StringLength{w}$, the output size of $2^{p(n)}$ is $p(\StringLength{w}) + 1$, which is $\iExp{0}{p(\StringLength{w})} + 1$, and hence $\iExp{i-1}{p(\StringLength{w})} + 1$ when $i = 1$.

The output can easily be produced by first computing $p(n) = p(\StringLength{w})$, then writing one bit `$1$', and then writing the required number of bits `$0$' by incrementing a counter from $1$ to $p(\StringLength{w})$.
The space complexity of this approach amounts to the space needed to evaluate $p(\StringLength{w})$, to store the value $p(\StringLength{w})$, and to store the counter (in binary).
Evaluating $p(\StringLength{w})$, representing $p(\StringLength{w})$, and storing the counter, all require $O(\log \StringLength{w})$ space (see the paragraph ``Polynomials'' above),
which is $\iExp{-1}{p(\StringLength{w})}$, and hence $O(\iExp{i-2}{p(\StringLength{w})})$ when $i = 1$.
The execution time of the outlined procedure consists of the time needed to evaluate $p(\StringLength{w})$, and to count from $1$ to $p(\StringLength{w})$.
The former is polynomial in the \emph{size} of $p(\StringLength{w})$ (see the paragraph ``Polynomials'' above), whereas the latter is linear in the \emph{value} of $p(\StringLength{w})$ (see, e.g., \cite[Chapter~16, ``Incrementing a binary counter'']{CormenLRS2022}).
Since the latter is dominating, the time complexity is $O(\iExp{0}{p(\StringLength{w})})$ and hence $O(\iExp{i-1}{p(\StringLength{w})})$ when $i = 1$. 

Let us now assume that the results hold for up to a generic $j > 1$.
We show that the results hold for $i = j + 1$ as well.
Since $\iExp{i}{p(\StringLength{w})} = \iExp{j+1}{p(\StringLength{w})} = 2^{(\iExp{j}{p(\StringLength{w})})}$, the procedure to compute $\iExp{j+1}{p(\StringLength{w})}$ is similar to the one outlined in the base case:
we first compute $\iExp{j}{p(\StringLength{w})}$, then we write a bit `$1$', followed by a number of bits `$0$' equalling $\iExp{j}{p(\StringLength{w})}$.
By this, the output size is $\iExp{j}{p(\StringLength{w})} + 1 = \iExp{i-1}{p(\StringLength{w})} + 1$.
Regarding the space complexity of the procedure, we need the space to compute $\iExp{j}{p(\StringLength{w})}$, which by induction is $O(\iExp{j-2}{p(\StringLength{w})})$.
Once computed, we need the space to store the value $\iExp{j}{p(\StringLength{w})}$, which is $\iExp{j-1}{p(\StringLength{w})}+1$.
We also need to count from $1$ to $\iExp{j}{p(\StringLength{w})}$, which requires a counter of size $\iExp{j-1}{p(\StringLength{w})}+1$.
Hence, the space complexity is $O(\iExp{j-1}{p(\StringLength{w})})$ that is $O(\iExp{i-2}{p(\StringLength{w})})$.
Regarding the time complexity of the procedure, we need the time to compute $\iExp{j}{p(\StringLength{w})}$, which by induction is $O(\iExp{j-1}{p(\StringLength{w})})$.
Then, we need to count from $1$ to $\iExp{j}{p(\StringLength{w})}$, which takes $O(\iExp{j}{p(\StringLength{w})})$ time and is dominating---remember, again, that the overall time spent to increment the counter is linear in the value of $\iExp{j}{p(\StringLength{w})}$; see \cite[Chapter~16, ``Incrementing a binary counter'']{CormenLRS2022}.
Therefore, the time complexity is $O(\iExp{j}{p(\StringLength{w})})$, which is $O(\iExp{i-1}{p(\StringLength{w})})$.

\medbreak

The complexity of computing $\smash{{\left(\iExp{i}{p(n)}\right)}^{\iExp{j}{q(n)}}}$, for fixed $i, j \geq 0$ and fixed polynomials $p$ and $q$, where $n$ is the length of an input string $w$, can be obtained by combining the complexity of the evaluation of $a^b$ and of $\iExp{i}{p(n)}$.

\section{The Boolean Hierarchies}
\label{sec_extended_Boolean_Hierarchies}

\citet{PapadimitriouY1984} introduced the complexity class \DP to characterize the complexity of several natural problems, such as deciding whether the size of a maximum clique in a graph is precisely a given integer.
Languages in \DP can be defined as the intersection of a language in \NPTime and one in \CoNPTime.
For this reason, \DP has often been identified as $\DP = \NPTime \land \CoNPTime$.
This notation does not mean $\DP = \NPTime \cap \CoNPTime$;
instead, \DP is a class of languages that are an intersection of languages, i.e., $\DP = \set{\Language{A} \cap \Language{B} \mid \Language{A} \in \NPTime \land \Language{B} \in \CoNPTime}$.

Boolean hierarchies over \NPTime were obtained by generalizing the definition of \DP.
Instead of the intersection of just two languages, longer ``Boolean'' combinations 
were considered, i.e., union, intersection, and complement (or, equivalently, difference) of multiple \NPTime languages.
From this perspective, \DP can also be defined as $\DP = \set{\Language{A} \cap \compl{\Language{B}} \mid \Language{A}, \Language{B} \in \NPTime} = \set{\Language{A} \setminus \Language{B} \mid \Language{A}, \Language{B} \in \NPTime}$.
Investigating Boolean combinations of \NPTime languages was therefore tantamount to studying the Boolean closure of \NPTime---observe that \NPTime is a ring of sets, as it is closed under union and intersection, and clearly contains the empty language.

Different authors investigated the Boolean closure of \NPTime by layering this closure into the levels of (apparently) different ``Boolean hierarchies'', which only afterwards were proven equivalent.
These hierarchies are each defined by a specific Boolean-combination of \NPTime languages.
Languages belonging to higher levels of the hierarchies can be defined by increasing-size language combinations of the specific kind.
These Boolean hierarchies have \NPTime 
at the base level.
Then, the $k$\nbdash-th level of the hierarchies is obtained by iterating the closure of \NPTime under~$k$ operations involving complement/difference~\cite{BertoniBJSY1989,BruschiJY1990}.
The Boolean closure of \NPTime is obtained as the union of all the infinite-many levels of these hierarchies.
Among these hierarchies there are:
\begin{itemize}[nosep,label=--]
  \item the (nested) differences hierarchy~\cite{CaiH1985,Beigel1991} (inspired by the (nested) differences hierarchy of generic sets~\cite{Addison1965} and of recursively enumerable sets~\cite{Ershov1968});
  \item the (symmetric) differences hierarchy~\cite{KoblerSW87} (inspired by the (symmetric) differences hierarchy of recursively enumerable sets~\cite{Posner1980});
  \item the (union of) differences hierarchy~\cite{CaiH1985};
  \item the (Hausdorff) differences hierarchy~\citep{CaiH1985,Wechsung1985,WechsungW1985,Wagner1987,Wagner1990}; and
  \item the Boolean (``alternating sums'') hierarchy itself~\cite{CaiH1985,CaiH1986};
\end{itemize}
defined as follows (see also~\cite{HemaspaandraR1997,BertoniBJSY1989,CaiGHHS1988}).
For two language classes $\ComplexityClass{X}$ and $\ComplexityClass{Y}$ we have:
\begin{align*}
  \ComplexityClass{X} \land \ComplexityClass{Y} &= \set{\Language{A} \cap \Language{B} \mid \Language{A} \in \ComplexityClass{X} \land \Language{B} \in \ComplexityClass{Y}} &  \ComplexityClass{X} \setsymmdifference \ComplexityClass{Y} &= \set{(\Language{A} \setminus \Language{B}) \cup (\Language{B} \setminus \Language{A}) \mid \Language{A} \in \ComplexityClass{X} \land \Language{B} \in \ComplexityClass{Y}} \\
  \ComplexityClass{X} \lor \ComplexityClass{Y} &= \set{\Language{A} \cup \Language{B} \mid \Language{A} \in \ComplexityClass{X} \land \Language{B} \in \ComplexityClass{Y}} &  \ComplexityClass{X} \setminus \ComplexityClass{Y} &= \set{\Language{A} \setminus \Language{B} \mid \Language{A} \in \ComplexityClass{X} \land \Language{B} \in \ComplexityClass{Y}}.
\end{align*}

\begin{enumerate}[label=(\roman*)]
  \item The levels $\ComplexityClass{C}_k$ of the Boolean (``alternating sums'') hierarchy over \NPTime are defined as:
  \begin{flalign*}
  & \ComplexityClass{C}_k =
  \begin{cases}
    \NPTime, & \text{if } k = 1 \\
    \ComplexityClass{C}_{k-1} \land \CoNPTime, & \text{if } k \geq 2 \text{ and } k \text{ is even} \\
    \ComplexityClass{C}_{k-1} \lor \NPTime, & \text{if } k \geq 2 \text{ and } k \text{ is odd}.
  \end{cases} &
  \end{flalign*}

  \item The levels $\ComplexityClass{D}_k$ of the (nested) differences hierarchy over \NPTime are defined as:
  \begin{flalign*}
  & \ComplexityClass{D}_k =
  \begin{cases}
    \NPTime, & \text{if } k = 1 \\
    \NPTime \setminus \ComplexityClass{D}_{k-1}, & \text{if } k \geq 2.
  \end{cases} &
  \end{flalign*}

  \item The levels $\ComplexityClass{DIFF}_k$ of the (symmetric) differences hierarchy over \NPTime are defined as:%
      \footnote{The notation $\ComplexityClass{DIFF}$ was used by \citet{Beigel1991} and by \citet{KoblerSW87}; however, the former referred to the \emph{nested} differences hierarchy whereas the latter referred to the \emph{symmetric} differences hierarchy. The notation $\ComplexityClass{D}$ for the nested differences hierarchy was introduced by \Citet{CaiH1985}.}
  \begin{flalign*}
  & \ComplexityClass{DIFF}_k =
  \begin{cases}
    \NPTime, & \text{if } k = 1 \\
    \ComplexityClass{DIFF}_{k-1} \setsymmdifference \NPTime, & \text{if } k \geq 2.
  \end{cases} &
  \end{flalign*}

  \item The levels $\ComplexityClass{E}_k$ of the (union of) differences hierarchy over \NPTime are defined as:
  \begin{flalign*}
  & \ComplexityClass{E}_k =
  \begin{cases}
    \NPTime, & \text{if } k = 1 \\
    \underbracket[.5pt]{\NPTime \setminus \NPTime}_{= \DP}, & \text{if } k = 2 \\
    \ComplexityClass{E}_{k-1} \lor \NPTime, & \text{if } k > 2 \text{ and } k \text{ is odd} \\
    \ComplexityClass{E}_{k-2} \lor (\underbracket[.5pt]{\NPTime \setminus \NPTime}_{= \DP}), & \text{if } k > 2 \text{ and } k \text{ is even}.
  \end{cases} &
  \end{flalign*}

  \item The levels $\ComplexityClass{E}_k$ of the (Hausdorff) differences hierarchy over \NPTime are defined similarly to the levels of the union of differences hierarchy above, with the additional constraint that languages in the considered Boolean combinations constitute a Hausdorff sequence.
      This means that languages $\Language{L} \in \ComplexityClass{E}_{k}$ are characterized by \NPTime languages $\Language{L}_1\supseteq \dots \supseteq \Language{L}_k$ such that $\Language{L} = (\Language{L}_1 \setminus \Language{L}_2) \cup (\Language{L}_3 \setminus \Language{L}_4) \cup \dots \cup (\Language{L}_{k-1} \setminus \Language{L}_k)$.
\end{enumerate}

The definitions above were shown all equivalent level by level~\cite{Addison1965,KoblerSW87,CaiGHHS1988,Wagner1988,BertoniBJSY1989,HemaspaandraR1997}.
By this, authors have generically referred to them as the \BHText (over \NPTime), even if they were differently defined.
By the ``union of differences'' definition of the \BHText follows that every language in the \BHText can be characterized by a finite union of \DP languages~\cite{CaiGHHS1988}.
The notation $\BH[k]$ denoted the $k$\nbdash-th level of the \BHText.
Notice that the levels of the Hausdorff differences hierarchy are tightly linked to Hausdorff reductions of increasing constant lengths.
Indeed, \citet{Wagner1988,Wagner1990} defined the $k$\nbdash-th level of the Boolean hierarchy over \NPTime as the class of $\NPTime$ Hausdorff languages of length $k$, to use the nomenclature introduced in this paper. 
By this, $\DP$ is the class of $\NPTime$ Hausdorff languages of length~$2$.

The equivalence of the different \BHText definitions was proven to hold, not only for the hierarchies over $\NPTime$, but whenever these hierarchies are built over a complexity class containing $\alphabet^*$ and $\emptyset$, and is closed under union and intersection~\cite{Addison1965,CaiGHHS1988,BertoniBJSY1989,HemaspaandraR1997}.
Therefore, since the levels of the above hierarchies only refers to different ways of combining \NPTime languages, the very same definitions can be adopted for the Boolean Hierarchies over other complexity classes~\cite{BertoniBJSY1989,HemaspaandraR1997}.
In fact, for example, the Boolean Hierarchies over $\SigmaP{2}$, \textnormal{RP}, \textnormal{UP}, and \NExpTime, appeared in~\cite{ChangK1996}, \cite{BertoniBJSY1989}, \cite{HemaspaandraR1997}, and~\cite{Dawar1998}, respectively.

\citet{Wagner1988,Wagner1990} also defined the \defin{Extended \BHText over $\NPTime$} as follows.
For a strictly positive nondecreasing computable function $r\colon \NaturalsDomain \to \NaturalsDomain$, $\BoundedHausdCLASS{r(n)}{\NPTime}$ is the ``$r(n)$\nbdash-level'' of the (Extended) Boolean Hierarchy over $\NPTime$, which, in our terms, is the class of $\NPTime$ Hausdorff languages of length $r(n)$.
Although \citeauthor{Wagner1990} introduced this notion, he did not deepen its analysis to extended it to higher level exponential time classes.
This most likely happened because his Extended Boolean Hierarchy definition was tightly linked to the Hausdorff reduction notion introduced in \cite{Wagner1990}, which is highly tailored for the polynomial case (see also \zcref{footnote_wagner_hausdorff_reductions_limited_to_polynomial_case}).

\section{Infinite Hausdorff Sequences}
\label{sec_infinite_hausdorff_sequences}

In this section, we outline a possible way to define infinite Hausdorff sequences.
We will see that such a definition is too powerful, as non recursively enumerable languages can be Hausdorff characterized over regular ones.
This could spur a more precise analysis to understand how this definition could be refined and become viable.

Let $\ComplexityClass{C}$ be a language class closed under union and intersection, and including the empty language.
The (infinite) sequence $\HausdSeq{H} = {\set{\Language{D}_z}}_{z \geq 1}$ of languages from $\ComplexityClass{C}$ is an \defin{(infinite) Hausdorff sequence (over $\ComplexityClass{C}$)} iff
\textbf{(\HausdSeqRequirementI)}~$\Language{D}_z \supseteq \Language{D}_{z+1}$ for all $z \geq 1$, and \textbf{(\HausdSeqRequirementII)}~for every string $w$, there is a $\Language{D}_{z_w}$ such that $w \notin \Language{D}_{z_w}$.
We define the (\defin{infinite}) \defin{Hausdorff summation}, or \defin{combination}, of $\HausdSeq{H}$ as $\bigcup_{\text{odd }z \geq 1} (\Language{D}_{z} \setminus \Language{D}_{z+1})$. 
A language $\Language{L}$ is \defin{Hausdorff characterized} by the infinite sequence $\HausdSeq{H}$ iff the infinite Hausdorff summation of $\HausdSeq{H}$ equals $\Language{L}$.
By this, for every string $w$, it holds that $w \in \Language{L} \Leftrightarrow\linebreak[0] \SetSize{\set{z \mid z \geq 1 \land w \in \Language{D}_z}}$ is odd.
Observe that (\HausdSeqRequirementI) and~(\HausdSeqRequirementII) 
imply that, for every string $w$, the quantity $\SetSize{\set{z \mid z \geq 1 \land w \in \Language{D}_z}}$ is finite, although defined over an infinite sequence, and hence always unambiguously determined.
Similarly to Hausdorff predicates (see \zcref{sec_Hausdorff_reductions_classes}), we can define the Hausdorff index of a string \Wrt an infinite Hausdorff sequence $\HausdSeq{H}$, and the length of an infinite Hausdorff sequence.
By (\HausdSeqRequirementI) and~(\HausdSeqRequirementII), for every string $w$ there is an index $z$ at which $w \in \Language{D}_{z}$;
if $w \notin \Language{D}_1$, and hence $w \notin \Language{D}_z$ for all $z \geq 1$, we let this maximum index to be~$0$.
We call this value the \defin{Hausdorff index} of $w$ \Wrt $\HausdSeq{H}$, and we formally define it as $\HausdIndex{w}{\HausdSeq{H}} = \max (\set{0, \sup (\set{z \mid z \geq 1 \land w \in \Language{D}_{z}})})$.
From the notion of Hausdorff index stems that of bounded Hausdorff sequence. 
Intuitively, a Hausdorff sequence $\HausdSeq{H}$ is bounded if there is some function bounding every string's Hausdorff index \Wrt $\HausdSeq{H}$.
More precisely, if $g\colon \NaturalsDomain \to \NaturalsDomain$ is a Hausdorff length function, we say that a Hausdorff sequence $\HausdSeq{H}$ is \defin{$g(n)$\nbdash-long} iff $g(n)$ is such that $\HausdIndex{w}{\HausdSeq{H}} \leq g(\StringLength{w})$, for all strings $w$.

Although these definitions might seem completely reasonable, we can show that they are overly powerful, as non recursively enumerable languages can be Hausdorff characterized over the set of regular languages.

\begin{example}
\label{example_hausdorff_sequence}
Let $\Regular$ be the class of all regular languages over $\alphabet = \set{0,1}$, and let $\Language{L} \subseteq \StringUniverse$ be a non recursively enumerable language.
We can assume \Wlog that $\Language{L}$ does not include the empty string---indeed, if $\Language{L}'$ were a generic non recursively enumerable language, we could define a language $\Language{L}$ containing all and only the strings in $\Language{L}'$ but padded with a very same extra symbol.
Remember that $\Regular$ is closed under union, intersection, and difference, and contains the empty language (see, e.g., \cite{Hopcroft1979}).
Hence, every \emph{finite} Hausdorff summation of regular languages yields a regular language.
Since $\Language{L}$ is not recursively enumerable, and hence non-regular, if a Hausdorff sequence characterizing $\Language{L}$ exists, it must be \emph{infinite}.
We show that $\Language{L}$ can actually be characterized over $\Regular$ by an infinite Hausdorff sequence $\HausdSeq{H} = {\set{\Language{D}_z}}_{z \geq 1}$ of bounded length. 
For each integer $z \geq 1$, let:
\begin{equation*}
  \Language{D}_z = %
    \begin{cases}
      (\Language{L} \cap \alphabet^{z}) \cup \alphabet^{\geq z+1}, & \text{if $z$ is odd}\\
      \alphabet^{\geq z} \setminus (\Language{L} \cap \alphabet^{z}), & \text{if $z$ is even}.
    \end{cases}
\end{equation*}
We claim that $\HausdSeq{H}$ is a Hausdorff sequence over $\Regular$.
We first show that $\HausdSeq{H}$ is over $\Regular$.
Notice that $(\Language{L} \cap \alphabet^{z})$ is a regular language.
Indeed, although $\Language{L}$ is non\nbdash-regular, the set $(\Language{L} \cap \alphabet^{z})$ contains finitely\nbdash-many strings, and every set $S$ containing only a finite number of strings is a regular language, as it can be expressed by the regular expression that is the (finite) disjunction of all the strings belonging to~$S$. 
By this, the languages $\Language{D}_z$ are indeed regular, as they are defined via union, intersection, and difference, of regular languages, and such combinations yield regular languages.
Besides this, $\HausdSeq{H}$ fulfills~(\HausdSeqRequirementI), as it is easy to verify that $\Language{D}_z \supseteq \Language{D}_{z+1}$ for all $z \geq 1$.

We are left to show that $\HausdSeq{H}$ fulfills~(\HausdSeqRequirementII), that is, for every string $w \in \StringUniverse$, there is a point in the Hausdorff sequence beyond which $w$ does not belong to the languages of the sequence. 
Let $w \in \alphabet^*$ be a string.
There are two case:
either (a)~$\StringLength{w}$ is odd,
or (b)~$\StringLength{w}$ is even.
Consider first case~(a).
Since $\StringLength{w}$ is odd, 
$\Language{D}_{\StringLength{w} + 1} = \alphabet^{\geq \StringLength{w} + 1} \setminus (\Language{L} \cap \alphabet^{\StringLength{w} + 1})$.
By this, 
$w \notin \Language{D}_{\StringLength{w} + 1}$, because $w \notin \alphabet^{\geq \StringLength{w} + 1}$.
Consider now case~(b).
By $\StringLength{w}$ being even, 
$\Language{D}_{\StringLength{w} + 1} = (\Language{L} \cap \alphabet^{\StringLength{w} + 1}) \cup \alphabet^{\geq \StringLength{w} + 2}$.
Hence, 
$w \notin \Language{D}_{\StringLength{w} + 1}$, because $w \notin \alphabet^{\StringLength{w} + 1}$ and $w \notin \alphabet^{\geq \StringLength{w} + 2}$.

Furthermore, from what we have highlighted above, the Hausdorff index \Wrt $\HausdSeq{H}$ of a generic string $w$ does not exceed $\StringLength{w}$.
Therefore, $\HausdSeq{H}$ is a bounded Hausdorff sequence of length $g(n) = n$.

Lastly, notice that, for each \emph{odd} integer $z \geq 1$, we have $(\Language{D}_z \setminus \Language{D}_{z+1}) = (\Language{L} \cap \alphabet^{z}) \cup (\Language{L} \cap \alphabet^{z+1})$.
By this, $\Language{L}$ equals the infinite Hausdorff summation of $\HausdSeq{H}$, and hence $\Language{L}$ is Hausdorff characterized by $\HausdSeq{H}$.
\hfill~$\blacksquare$
\end{example}

In the example above, we have seen that, if not carefully defined, an infinite Hausdorff sequence over the regular languages can characterize even a non recursively enumerable language.
The tricky point in the example above seems to be the definition of the languages $(\Language{L} \cap \alphabet^{z})$.
These languages, since they contain only finitely\nbdash-many strings, are regular.
However, at the moment in which these languages are defined, it might be \emph{unknown} which strings belong to them.
Perhaps it is on this point that one may carry out a finer analysis to obtain a clearer separation between languages that can be considered and languages that cannot.
So that the notion of infinite Hausdorff sequence introduced at the beginning of this section could become viable.

\section{Deferred Proofs}
\label{sec_deferred_proofs_top}

\stoptoc

\subsection%
{Proofs for \zcref{sec_prelim_iterated_exponentials}}
\label{sec_detailed_proofs_properties_iterated_exponentials}
\silentbookmark[2]{Proofs for Section~\ref*{sec_prelim_iterated_exponentials}}

\getkeytheorem{asymptoticClassPropertiesIteratedExponentialsNonContainment}

\begin{proof}
We prove $\iExp{i}{O(\iExp{j}{n^k})} \not\subseteq O(\iExp{i+j}{n^{k}})$ by exhibiting a function $f(n)$ which is $\iExp{i}{O(\iExp{j}{n^k})}$ but not $O(\iExp{i+j}{n^{k}})$.
Let $f(n) = \iExp{i}{c \mycdot \iExp{j}{n^k}} \in \iExp{i}{O(\iExp{j}{n^k})}$ be a function, for some constant $c \geq \iExp{j}{1} + 1$.
If $f(n) \in O(\iExp{i+j}{n^{k}})$ were the case, there would be constants $d$ (\Wlog, assume $d \geq \iExp{i+j}{1}$) and $n_0$ such that, for all $n \geq n_0$, $f(n) = \iExp{i}{c \mycdot \iExp{j}{n^k}} \leq d \cdot \iExp{i+j}{n^k}$.
Let us rewrite this inequality (below, $\iLog{i}{(\cdot)}$ is a(n iterated) logarithm, as $i \geq 1$):
\begin{multline*}
\iExp{i}{c \mycdot \iExp{j}{n^k}} \leq d \cdot \iExp{i+j}{n^k} {} \Leftrightarrow {} 
    \iLog{i}{(\iExp{i}{c \mycdot \iExp{j}{n^k}})} \leq \iLog{i}{(d \cdot \iExp{i+j}{n^k})} \Leftrightarrow {} \\
    c \mycdot \iExp{j}{n^k} \leq \iExp{j}{n^k} + \iLog{i}{d} \Leftrightarrow 
\iExp{j}{n^k}\cdot (c-1) \leq \iLog{i}{d}.
\end{multline*}

Now, there are two cases:
either $j = 0$, or $j \geq 1$.
When $j = 0$, we have:
\begin{gather}
\iExp{j}{n^k}\cdot (c-1) \leq \iLog{i}{d} \Leftrightarrow n^k \cdot (c-1) \leq \iLog{i}{d} \Leftrightarrow {} \nonumber \displaybreak[0] \\
n^k \leq \frac{\iLog{i}{d}}{c-1}. \label{eq_not_containment_iterated_exp_j0}
\end{gather}

When $j \geq 1$, we have:
\begin{gather}
\iExp{j}{n^k}\cdot (c-1) \leq \iLog{i}{d} \Leftrightarrow
\iLog{j}{(\iExp{j}{n^k} \cdot (c-1))} \leq \iLog{j}{(\iLog{i}{d})} \Leftrightarrow
n^k + \iLog{j}{(c-1)} \leq \iLog{i+j}{d} \Leftrightarrow {} \nonumber \displaybreak[0] \\
n^k \leq \iLog{i+j}{d} - \iLog{j}{(c-1)}. \label{eq_not_containment_iterated_exp_j1}
\end{gather}

Since $k \geq 1$, the left\nbdash-hand side values of Eqs.~(\ref{eq_not_containment_iterated_exp_j0}, \ref{eq_not_containment_iterated_exp_j1}) indefinitely grow as $n$ grows, whereas the right\nbdash-hand sides of Eqs.~(\ref{eq_not_containment_iterated_exp_j0}, \ref{eq_not_containment_iterated_exp_j1}) are constants (remember that $d$ and $c$ are assumed to be such that $d \geq \iExp{i+j}{1}$ and $c \geq \iExp{j}{1} + 1$, so $\iLog{i+j}{d}$, $\iLog{i}{d}$, and $\iLog{j}{(c-1)}$, are defined).
Thus, there are no constants $d$ and $n_0$ such that the inequalities \eqref{eq_not_containment_iterated_exp_j0} and \eqref{eq_not_containment_iterated_exp_j1} hold for all $n \geq n_0$.
Hence, $f(n) \notin O(\iExp{i+j}{n^{k}})$, and so $\iExp{i}{O(\iExp{j}{n^k})} \nsubseteq O(\iExp{i+j}{n^{k}})$. 
\end{proof}

\getkeytheorem{asymptoticClassPropertiesIteratedExponentialsContainment}

\begin{proof}
The property clearly holds when $i = 0$, as $O(\iExp{j}{n^k}) \subseteq O(\iExp{j}{n^{k+1}})$ when $j,k \geq 0$.

Assume now $i \geq 1$.
Let $f(n) \in \iExp{i}{O(\iExp{j}{n^k})}$ be a function.
There hence are constants $c$ and $n_a$ (\Wlog, assume $c, n_a \geq 1$) such that, for all $n \geq n_a$, $f(n) \leq \iExp{i}{c \mycdot \iExp{j}{n^k}}$.
We prove $f(n) \in O(\iExp{i+j}{n^{k+1}})$ by exhibiting constants $d$ and $n_b$ such that, for all $n \geq n_b$, $f(n) \leq d \cdot \iExp{i+j}{n^{k+1}}$;
the following proof holds for $j = 0$ \emph{and} $j \geq 1$.

When $n \geq n_a$, by $j \geq 0$ and $c \geq 1$ we have $f(n) \leq \iExp{i}{c \mycdot \iExp{j}{n^k}} \leq \iExp{i}{\iExp{j}{c \mycdot n^k}} = \iExp{i+j}{c \mycdot n^k}$.
Thus, if there existed $d$ and $n_b$ such that $n_b \geq n_a$ and $\iExp{i+j}{c \mycdot n^k} \leq d \cdot \iExp{i+j}{n^{k+1}}$ , for all $n \geq n_b$, then $f(n) \in O(\iExp{i+j}{n^{k+1}})$ would be proven.
We rewrite $\iExp{i+j}{c \mycdot n^k} \leq d \cdot \iExp{i+j}{n^{k+1}}$ as (below, $\iLog{i+j}{(\cdot)}$ is a(n iterated) logarithm, as $i \geq 1$ and $j \geq 0$):
\begin{gather}
  \iExp{i+j}{c \mycdot n^k} \leq d \cdot \iExp{i+j}{n^{k+1}} \Leftrightarrow
    \iLog{i+j}{(\iExp{i+j}{c \mycdot n^k})} \leq \iLog{i+j}{(d \cdot \iExp{i+j}{n^{k+1}})} \Leftrightarrow
    c \cdot n^k \leq n^{k+1} + \iLog{i+j}{d} \Leftrightarrow {} \nonumber \\
    n^k \cdot (c - n) \leq \iLog{i+j}{d}. \label{eq_containment_iterated_exp}
\end{gather}

For $d \geq \iExp{i+j}{1}$, the right-hand side of Eq.~\eqref{eq_containment_iterated_exp} is a positive constant.
By letting $n_b = \max \set{n_a,c+1}$, the quantity $(c - n)$ in the left-hand side of Eq.~\eqref{eq_containment_iterated_exp} is negative for all $n \geq n_b > c$.
By this and $n_b \geq n_a$, for all $n \geq n_b$, it holds $f(n) \leq \iExp{i}{c \mycdot \iExp{j}{n^k}} \leq d \cdot \iExp{i+j}{n^{k + 1}}$.
Hence, $f(n) \in O(\iExp{i+j}{n^{k+1}})$, and $\iExp{i}{O(\iExp{j}{n^k})} \subseteq O(\iExp{i+j}{n^{k+1}})$.
\end{proof}

\getkeytheorem{operationsBetweenIteratedExponentials}

\begin{proof}\hspace{0pt}
For functions $x(n)$ and $y(n)$, by ``$y(n)$ eventually exceeds $x(n)$'' we mean that there exists $n_0 \geq 1$ such that, for all $n \geq n_0$, $x(n) \leq y(n)$.
\emph{All inequalities $x(n) \leq y(n)$ in this proof mean ``$y(n)$ eventually exceeds $x(n)$''.}

The inclusion relationships in the statement of the \zcref*[typeset=name,nocap]{theo_operations_between_iterated_exponentials} are obtained by the arguments below and by the fact that, for functions $x(n)$ and $y(n)$, it is easy to see that $O(x(n)) \subseteq O(y(n))$ iff $x(n) \in O(y(n))$.
\begin{enumerate}[left=0pt,label=\arabic*)]
  \item
    Since $f(n) \in O(\iExpPolFunctions{i})$, there are integers $d$, $d'$, and~$\ell$, such that  $f(n) \leq d' \cdot \iExp{\max\set{0,i}}{d \cdot n^\ell}$.
    For this reason \[c \cdot \! {f(n)}^k \leq c \cdot {\bigl( d' \cdot \iExp{\max{\set{0,i}}}{d \cdot n^\ell} \bigr)}^k \leq c \cdot {(d')}^k \cdot \iExp{\max\set{0,i}}{k \cdot d^k \cdot n^{\ell \cdot k}},\]
    where the last inequality is obtained by combining ${\bigl( \iExp{0}{d \cdot n^\ell} \bigr)}^k = {(d \cdot n^\ell)}^k = d^k \cdot n^{\ell \cdot k}$, and, when $i \geq 1$, ${\bigl( \iExp{i}{d \cdot n^\ell} \bigr)}^k \leq \iExp{i}{k \cdot d \cdot n^\ell}$.
    To conclude, notice that $c \cdot {(d')}^k \cdot \iExp{\max\set{0,i}}{k \cdot d^k \cdot n^{\ell \cdot k}}$ belongs to $O(\iExpPolFunctions{\max\set{0,i}})$.

  \item
    Since $f(n) \in O(\iExpPolFunctions{i})$ and $g(n) \in O(\iExpPolFunctions{j})$, there are integers $c$, $c'$, $k$, $d$, $d'$, and~$\ell$, such that $f(n) \leq c' \cdot \iExp{i}{c \cdot n^k}$ and $g(n) \leq d' \cdot \iExp{j}{d \cdot n^\ell}$.
    Hence,
    \begin{multline*}
      f(n) + g(n) \leq c' \cdot \iExp{i}{c \cdot n^k} + d' \cdot \iExp{j}{d \cdot n^\ell} \leq {} \\
      \max\set{c',d'} \cdot \iExp{\max\set{i,j}}{\max\set{c,d}\cdot n^{\max\set{k,\ell}}} + \max\set{c',d'} \cdot \iExp{\max\set{i,j}}{\max\set{c,d}\cdot n^{\max\set{k,\ell}}} \leq {} \\
      2 \cdot \max\set{c',d'} \cdot \iExp{\max\set{i,j}}{\max\set{c,d}\cdot n^{\max\set{k,\ell}}}.
    \end{multline*}
    Observe that $2 \cdot \max\set{c',d'} \cdot \iExp{\max\set{i,j}}{\max\set{c,d}\cdot n^{\max\set{k,\ell}}}$ belongs to $O(\iExpPolFunctions{\max\set{i,j}})$.

  \item
    Since $f(n) \in O(\iExpPolFunctions{i})$ and $g(n) \in O(\iExpPolFunctions{j})$, there are integers $c$, $c'$, $k$, $d$, $d'$, and~$\ell$, such that $f(n) \leq c' \cdot \iExp{\max\set{0,i}}{c \cdot n^k}$ and $g(n) \leq d' \cdot \iExp{\max\set{0,j}}{d \cdot n^\ell}$.
    Hence,
    \begin{multline*}
      f(n) \cdot g(n) \leq c' \cdot \iExp{\max\set{0,i}}{c \cdot n^k} \! \cdot d' \cdot \iExp{\max\set{0,j}}{d \cdot n^\ell} \leq {} \\
      c' \cdot d' \cdot \iExp{\max\set{0,i,j}}{\max\set{c,d} \cdot n^{\max\set{k,\ell}}} \! \cdot \iExp{\max\set{0,i,j}}{\max\set{c,d} \cdot n^{\max\set{k,\ell}}} = {} \\ 
      c' \cdot d' \cdot {\left(\iExp{\max\set{0,i,j}}{\max\set{c,d} \cdot n^{\max\set{k,\ell}}}\right)}^{2} \leq c' \cdot d' \cdot \iExp{\max\set{0,i,j}}{2 \cdot \max\set{c^2,d^2} \cdot n^{\max\set{2k,2\ell}}},
    \end{multline*}
    where the last inequality is obtained in a similar way to what we did in Point~1 above.
    Observe that $c' \cdot d' \cdot \iExp{\max\set{0,i,j}}{2 \cdot \max\set{c^2,d^2} \cdot n^{\max\set{2k,2\ell}}}$ belongs to $O(\iExpPolFunctions{\max\set{0,i,j}})$.

  \item
    Since $f(n) \in O(\iExpPolFunctions{i})$ and $g(n) \in O(\iExpPolFunctions{j})$, there are integers $c$, $c'$, $k$, $d$, $d'$, and~$\ell$, such that $f(n) \leq c' \cdot \iExp{\max\set{1,i}}{c \cdot n^k}$ and $g(n) \leq d' \cdot \iExp{\max\set{0,j}}{d \cdot n^\ell}$.
    Therefore,
    \begin{multline*}
        f(n)^{g(n)} \leq {\left(c' \cdot \iExp{\max\set{1,i}}{c \cdot n^k}\right)}^{d' \mycdot \iExp{\max\set{0,j}}{d \cdot n^\ell}} = {(c')}^{d' \mycdot \iExp{\max\set{0,j}}{d \cdot n^\ell}} \cdot {\left( \iExp{\max\set{1,i}}{c \cdot n^k} \right)}^{d' \mycdot \iExp{\max\set{0,j}}{d \cdot n^\ell}} \leq {} \\
        {\left( \iExp{\max\set{1,i}}{c \cdot n^k} \right)}^{d' \mycdot \iExp{\max\set{0,j}}{d \cdot n^\ell}} \cdot {\left( \iExp{\max\set{1,i}}{c \cdot n^k} \right)}^{d' \mycdot \iExp{\max\set{0,j}}{d \cdot n^\ell}} = {\left( {\left( \iExp{\max\set{1,i}}{c \cdot n^k} \right)}^{d' \mycdot \iExp{\max\set{0,j}}{d \cdot n^\ell}} \right)}^{2} = {} \\
        {\left({2}^{\bigl(\iExp{\max\set{0,i-1}}{c \cdot n^k} \mycdot d' \mycdot \iExp{\max\set{0,j}}{d \cdot n^\ell}\bigr)}\right)}^{2} = {\left({2}^{2 \mycdot \bigl(\iExp{\max\set{0,i-1}}{c \cdot n^k} \mycdot d' \mycdot \iExp{\max\set{0,j}}{d \cdot n^\ell}\bigr)}\right)}
    \end{multline*}
    
    By Point~3 above, there exist integers $x$, $y$, and $z$, such that 
    \begin{multline*}
        {\left({2}^{2 \mycdot \iExp{\max\set{0,i-1}}{c \cdot n^k} \mycdot d' \mycdot \iExp{\max\set{0,j}}{d \cdot n^\ell}}\right)} \leq 2^{\left(x \mycdot \iExp{\max\set{0,i-1,j}}{y \cdot n^{z}}\right)} \leq 2^{\left(\iExp{\max\set{0,i-1,j}}{x \cdot y \cdot n^{z}}\right)} = {} \\
        \iExp{\max\set{1,i,j+1}}{x \cdot y \cdot n^{z}}.
    \end{multline*}
    Notice that $\iExp{\max\set{1,i,j+1}}{x \cdot y \cdot n^{z}}$ belongs to $O(\iExpPolFunctions{\max\set{1,i,j+1}})$.

\item
    We consider two cases:
    first $i \geq 0$, and then $i = -1$.
    Let us focus on $i \geq 0$.
    Since $f(n) \in O(\iExpPolFunctions{i})$, there are integers $c$, $c'$, and $k$, such that $f(n) \leq c' \cdot \iExp{i}{c \cdot n^k}$.
    We hence have $f(g(n)) \leq c' \cdot \iExp{i}{c \cdot {g(n)}^k}$.
    By Point~1 above, there are integers $d$, $d'$, and~$\ell$, such that $c \cdot {g(n)}^k \leq d' \cdot \iExp{\max\set{0,j}}{d \cdot n^\ell}$.
    Therefore,
    \[
        f(g(n)) \leq c' \cdot \iExp{i}{d' \cdot \iExp{\max\set{0,j}}{d \cdot n^\ell}} \leq \iExp{i}{c' \cdot d' \cdot \iExp{\max\set{0,j}}{d \cdot n^\ell}} \leq \iExp{i}{\iExp{\max\set{0,j}}{c' \cdot d' \cdot d \cdot n^\ell}} \leq \iExp{\max\set{i,i+j}}{c' \cdot d' \cdot d \cdot n^\ell}.
    \]
    Clearly, the latter is $O(\iExpPolFunctions{\max\set{i,i+j}})$.
    
    Let us now consider $i = -1$.
    By $f(n) \in O(\iExpPolFunctions{-1})$, there is an integer $c$ such that $f(n) \leq c \cdot \log n$.
    We consider three cases:
    $j \geq 1$, $j = 0$, and $j = -1$.
    We start with $j \geq 1$.
    Since $g(n) \in O(\iExpPolFunctions{j})$, there are integers $d$, $d'$, and $\ell$, such that $g(n) \leq d' \cdot \iExp{j}{d \cdot n^\ell}$.
    By this, we have:
    \[
        f(g(n)) \leq c \cdot \log (g(n)) \leq c \cdot \log (d' \cdot \iExp{j}{d \cdot n^\ell}) \leq c \cdot \log (\iExp{j}{d' \cdot d \cdot n^\ell}) \leq c \cdot \iExp{j-1}{d' \cdot d \cdot n^\ell} \leq \iExp{j-1}{c' \cdot d' \cdot d \cdot n^\ell}. 
    \]
    The latter is $O(\iExpPolFunctions{j-1})$, and, since $i = -1$ and $j \geq 1$, it is also $O(\iExpPolFunctions{\max\set{i,i+j}})$.
      
    Consider now $j = 0$.
    By $g(n) \in O(\iExpPolFunctions{0})$, there are integers $d$ and $\ell$ such that $g(n) \leq d \cdot n^\ell$.
    By this:
    \[
        f(g(n)) \leq c \cdot \log (g(n)) \leq c \cdot \log (d \cdot n^\ell) \leq c \cdot \log d + c \cdot \ell \cdot \log n.
    \]
    The latter is $O(\iExpPolFunctions{-1})$, and, since $i = -1$ and $j = 0$, it is also $O(\iExpPolFunctions{\max\set{i,i+j}})$.
    
    We conclude with $j = -1$.
    By $g(n) \in O(\iExpPolFunctions{-1})$, there is an integer $d$ such that $g(n) \leq d \cdot \log n$.
    By this:
    \[
        f(g(n)) \leq c \cdot \log (g(n)) \leq c \cdot \log (d \cdot \log n) \leq c \cdot \log d + c \cdot \log \log n.
    \]
    The latter is $O(\iExpPolFunctions{-2})$, and so $O(\iExpPolFunctions{-1})$.
    By $i = -1$ and $j = -1$, this is also $O(\iExpPolFunctions{\max\set{i,i+j}})$.
    \qedhere    
\end{enumerate}
\end{proof}

\subsection%
{Proofs for \zcref{sec_charting_exp_oracles}}
\label{sec_detailed_proofs_charting_exp_oracles}
\silentbookmark[2]{Proofs for Section~\ref*{sec_charting_exp_oracles}}

\getkeytheorem{EEcontainment}

\begin{proof}
We first prove that $\Oracle{\iExpTime{i}}{\Oracle{\iExpTime{j}}{\SigmaP{c}}} \subseteq \Oracle{\iExpTime{(i+j)}}{\SigmaP{c}}$.

Let $\Language{L} \in \Oracle{\iExpTime{i}}{\Oracle{\iExpTime{j}}{\SigmaP{c}}}$ be a language.
There are oracle machines $\Oracle{\Machine{M}'}{?} \in \iExpTime{i}$ and $\Oracle{\Machine{M}''}{?} \in \iExpTime{j}$, and a $\SigmaP{c}$ oracle $\Omega$, such that $\Language{L} = \LanguageOf{\Oracle{\Machine{M}'}{\Oracle{\Machine{M}''}{\Omega}}}$.
We claim that there is an oracle machine $\Oracle{\Machine{N}}{?} \in \iExpTime{(i+j)}$ such that $\Language{L} = \LanguageOf{\Oracle{\Machine{N}}{\Omega}}$.
Observe that $\Machine{M}'$ is an \iExponential{i} oracle machine issuing queries to the \iExponential{j} oracle machine $\Machine{M}''$.
The latter may hence carry out \iExponential{(i+j)} computations in the size of the string in input to $\Machine{M}'$, as $\Machine{M}'$ may issue \iExponential{i}{}ly\nbdash-long queries to $\Machine{M}''$.
Therefore, an \iExponential{(i+j)} oracle machine $\Machine{N}$ can actually simulate, via the aid of the oracle $\Omega$, the working of $\Oracle{\Machine{M}'}{\Oracle{\Machine{M}''}{\Omega}}$.

\Proofsep

We now prove that $\Oracle{\iExpTime{(i+j)}}{\SigmaP{c}} \subseteq \BoundedOracle{\iExpTime{k}}{\Oracle{\iExpTime{(i+j-k)}}{\SigmaP{c}}}{1}$.

Let $\Language{L} \in \Oracle{\iExpTime{(i+j)}}{\SigmaP{c}}$ be a language.
There are an $\iExpTime{(i+j)}$ oracle machine $\Oracle{\Machine{M}}{?}$ and a $\SigmaP{c}$ oracle $\Omega$ such that $\Language{L} = \LanguageOf{\Oracle{\Machine{M}}{\Omega}}$.
We claim that there are oracle machines $\BoundedOracle{\Machine{N}'}{?}{1} \in \iExpTime{k}$ and $\Oracle{\Machine{N}''}{?} \in \iExpTime{(i+j-k)}$ such that $\Language{L} = \LanguageOf{\Oracle{\Machine{N}'}{\Oracle{\Machine{N''}}{\Omega}}}$. 
These machines may work as follows:
$\Machine{N}'$ just passes to $\Machine{N}''$ a \iExponential{k}{}ly padded version $\wt{w}$ of the input string~$w$.
Upon reception of $\wt{w}$, $\Machine{N}''$ can carry out a computation that is \iExponential{(i+j)} in the size of~$w$.
Therefore, with the aid of the oracle $\Omega$, $\Machine{N}''$ can simulate on $\wt{w}$ the computation of $\Oracle{\Machine{M}}{\Omega}$ on~$w$.
\end{proof}

\getkeytheorem{EEequivalence}

\begin{proof}
By \zcref{theo_exp_exp_containment}, $\Oracle{\iExpTime{i}}{\Oracle{\iExpTime{j}}{\SigmaP{c}}} \subseteq \Oracle{\iExpTime{\ell}}{\SigmaP{c}} \subseteq \BoundedOracle{\iExpTime{i'}}{\Oracle{\iExpTime{j'}}{\SigmaP{c}}}{1}$.
Clearly, $\BoundedOracle{\iExpTime{i'}}{\Oracle{\iExpTime{j'}}{\SigmaP{c}}}{1} \subseteq \Oracle{\iExpTime{i'}}{\Oracle{\iExpTime{j'}}{\SigmaP{c}}}$.
Again by \zcref{theo_exp_exp_containment}, $\Oracle{\iExpTime{i'}}{\Oracle{\iExpTime{j'}}{\SigmaP{c}}} \subseteq \Oracle{\iExpTime{\ell}}{\SigmaP{c}} \subseteq \BoundedOracle{\iExpTime{i}}{\Oracle{\iExpTime{j}}{\SigmaP{c}}}{1}$.
To conclude, by $\BoundedOracle{\iExpTime{i}}{\Oracle{\iExpTime{j}}{\SigmaP{c}}}{1} \subseteq \Oracle{\iExpTime{i}}{\Oracle{\iExpTime{j}}{\SigmaP{c}}}$, the statement follows.
\end{proof}

\getkeytheorem{NEcontainment}

\begin{proof}
We start by showing that $\Oracle{\iNExpTime{i}}{\Oracle{\iExpTime{j}}{\SigmaP{c}}} \subseteq \Oracle{\iExpTime{(i+j)}}{\SigmaP{c}}$.

Let $\Language{L} \in \Oracle{\iNExpTime{i}}{\Oracle{\iExpTime{j}}{\SigmaP{c}}}$ be a language.
There are oracle machines $\Oracle{\Machine{M}'}{?} \in \iNExpTime{i}$ and $\Oracle{\Machine{M}''}{?} \in \iExpTime{j}$, and a $\SigmaP{c}$ oracle $\Omega$, such that $\Language{L} = \LanguageOf{\Oracle{\Machine{M}'}{\Oracle{\Machine{M}''}{\Omega}}}$.
We exhibit an oracle machine $\Oracle{\Machine{N}}{?} \in \iExpTime{(i+j)}$ such that $\Language{L} = \LanguageOf{\Oracle{\Machine{N}}{\Omega}}$.

Remember that the computation of a \emph{non}\/deterministic \iExponential{i}{}\nbdash-time machine $\Machine{M}$ can be simulated in deterministic \iExponential{(i+1)} time by a depth\nbdash-first exploration of the computation tree of $\Machine{M}$.

The oracle machine $\Machine{N}$ can decide $\Language{L}$ by simulating the working of $\Oracle{\Machine{M}'}{\Oracle{\Machine{M}''}{\Omega}}$ as follows.
Since $j \geq 1$, $\Machine{N}$ can simulate within \iExponential{(i+j)} time the working of $\Machine{M}'$, as $\Machine{M}'$ works in \iExponential{i} time.
When $\Machine{M}'$ calls its \iExponential{j}{}\nbdash-time oracle $\Machine{M}''$, we claim the machine $\Machine{N}$ can simulate the working of $\Machine{M}''$ as well.
Indeed, $\Machine{M}''$ is a \emph{deterministic} machine working in at most \iExponential{(i+j)} time, as $\Machine{M}''$ may receive from $\Machine{M}'$ \iExponential{i}{}ly-long queries, and $\Machine{N}$ by assumption can work in \iExponential{(i+j)} time.
When $\Machine{M}''$ calls the oracle $\Omega$, then $\Machine{N}$ calls $\Omega$ as well and progresses in the simulation.

\Proofsep

The inclusion $\Oracle{\iExpTime{(i+j)}}{\SigmaP{c}} \subseteq \BoundedOracle{\iNExpTime{k}}{\Oracle{\iExpTime{(i+j-k)}}{\SigmaP{c}}}{1}$ can be obtained by an argument very similar to the one showing $\Oracle{\iExpTime{(i+j)}}{\SigmaP{c}} \subseteq \BoundedOracle{\iExpTime{k}}{\Oracle{\iExpTime{(i+j-k)}}{\SigmaP{c}}}{1}$ in the proof of~\zcref{theo_exp_exp_containment}.
The only difference is that the oracle machine $\BoundedOracle{\Machine{N}'}{?}{1}$ there mentioned is here a $\iNExpTime{k}$ machine rather than a $\iExpTime{k}$ one.
\end{proof}

\getkeytheorem{NEequivalence}

\begin{proof}
By \zcref{theo_nexp_exp_containment}, $\Oracle{\iNExpTime{i}}{\Oracle{\iExpTime{j}}{\SigmaP{c}}} \subseteq \Oracle{\iExpTime{\ell}}{\SigmaP{c}} \subseteq \BoundedOracle{\iNExpTime{i'}}{\Oracle{\iExpTime{j'}}{\SigmaP{c}}}{1}$.
Clearly, the inclusion  $\BoundedOracle{\iNExpTime{i'}}{\Oracle{\iExpTime{j'}}{\SigmaP{c}}}{1} \subseteq \Oracle{\iNExpTime{i'}}{\Oracle{\iExpTime{j'}}{\SigmaP{c}}}$ holds.
Again, by \zcref{theo_nexp_exp_containment}, $\Oracle{\iNExpTime{i'}}{\Oracle{\iExpTime{j'}}{\SigmaP{c}}} \subseteq \Oracle{\iExpTime{\ell}}{\SigmaP{c}} \subseteq \BoundedOracle{\iNExpTime{i}}{\Oracle{\iExpTime{j}}{\SigmaP{c}}}{1}$.
We also have that $\BoundedOracle{\iNExpTime{i}}{\Oracle{\iExpTime{j}}{\SigmaP{c}}}{1} \subseteq \Oracle{\iNExpTime{i}}{\Oracle{\iExpTime{j}}{\SigmaP{c}}}$, from which the statement follows.
\end{proof}

\getkeytheorem{SEcontainment}

\begin{proof}
We start by showing that $\Oracle{\iExpSpace{(i-1)}}{\Oracle{\iExpTime{j}}{\SigmaP{c}}} \subseteq \Oracle{\iExpTime{(i+j)}}{\SigmaP{c}}$.

Let $\Language{L} \in \Oracle{\iExpSpace{(i-1)}}{\Oracle{\iExpTime{j}}{\SigmaP{c}}}$ be a language.
There are oracle machines $\Oracle{\Machine{M}'}{?} \in \iExpSpace{(i-1)}$ and $\Oracle{\Machine{M}''}{?} \in \iExpTime{j}$, and a \SigmaP{c} oracle $\Omega$, such that $\Language{L} = \LanguageOf{\Oracle{\Machine{M}'}{\Oracle{\Machine{M}''}{\Omega}}}$.
We exhibit a $\iExpTime{(i+j)}$ oracle machine $\Oracle{\Machine{N}}{?}$ such that $\Language{L} = \LanguageOf{\Oracle{\Machine{N}}{\Omega}}$.
The machine $\Machine{N}$ can decide $\Language{L}$ by simulating $\Oracle{\Machine{M}'}{\Oracle{\Machine{M}''}{\Omega}}$ as follows.
Since $j \geq 0$, $\Machine{N}$ can simulate within \iExponential{(i+j)} time the working of $\Machine{M}'$, as the $\iExpSpace{(i-1)}$ machine $\Machine{M}'$ runs for at most \iExponential{i} time.
When $\Machine{M}'$ calls its \iExponential{j}{}\nbdash-time oracle $\Machine{M}''$, we claim that $\Machine{N}$ can simulate the working of $\Machine{M}''$ as well.
Indeed, $\Machine{M}''$ is a \emph{deterministic} machine working in at most \iExponential{(i+j)} time, as $\Machine{M}''$ may receive from $\Machine{M}'$ \iExponential{i}{}ly-long queries, and $\Machine{N}$ by assumption can work in \iExponential{(i+j)} time.
When $\Machine{M}''$ calls $\Omega$, then $\Machine{N}$ calls $\Omega$ as well, and the simulation proceeds.

\Proofsep

The inclusion $\Oracle{\iExpTime{(i+j)}}{\SigmaP{c}} \subseteq \BoundedOracle{\iExpSpace{(k-1)}}{\Oracle{\iExpTime{(i+j-k)}}{\SigmaP{c}}}{1}$ can be obtained by an argument very similar to the one showing $\Oracle{\iExpTime{(i+j)}}{\SigmaP{c}} \subseteq \BoundedOracle{\iExpTime{k}}{\Oracle{\iExpTime{(i+j-k)}}{\SigmaP{c}}}{1}$ in the proof of~\zcref{theo_exp_exp_containment}.
The only difference is that the oracle machine $\BoundedOracle{\Machine{N}'}{?}{1}$ there mentioned is here a $\iExpSpace{(k-1)}$ machine rather than a $\iExpTime{k}$ one.
Remember that a $\iExpSpace{(k-1)}$ machine can run for \iExponential{k} time, and can hence produce the \iExponential{k}{}ly padded query required in the proof mentioned.
\end{proof}

\getkeytheorem{SEequivalence}

\begin{proof}
By \zcref{theo_expspace_exp_containment}, $\Oracle{\iExpSpace{(i-1)}}{\Oracle{\iExpTime{j}}{\SigmaP{c}}} \subseteq \Oracle{\iExpTime{\ell}}{\SigmaP{c}} \subseteq \BoundedOracle{\iExpSpace{(i'-1)}}{\Oracle{\iExpTime{j'}}{\SigmaP{c}}}{1}$.
Clearly, we have $\BoundedOracle{\iExpSpace{(i'-1)}}{\Oracle{\iExpTime{j'}}{\SigmaP{c}}}{1} \subseteq \Oracle{\iExpSpace{(i'-1)}}{\Oracle{\iExpTime{j'}}{\SigmaP{c}}}$.
By \zcref{theo_expspace_exp_containment}, $\Oracle{\iExpSpace{(i'-1)}}{\Oracle{\iExpTime{j'}}{\SigmaP{c}}} \subseteq \Oracle{\iExpTime{\ell}}{\SigmaP{c}} \subseteq \BoundedOracle{\iExpSpace{(i-1)}}{\Oracle{\iExpTime{j}}{\SigmaP{c}}}{1}$.
From $\BoundedOracle{\iExpSpace{(i-1)}}{\Oracle{\iExpTime{j}}{\SigmaP{c}}}{1} \subseteq \Oracle{\iExpSpace{(i-1)}}{\Oracle{\iExpTime{j}}{\SigmaP{c}}}$ the statement follows.
\end{proof}

\subsection%
{Proofs for \zcref{sec_charting_nexp_oracles}}
\label{sec_detailed_proofs_charting_nexp_oracles}
\silentbookmark[2]{Proofs for Section~\ref*{sec_charting_nexp_oracles}}

\getkeytheorem{ENContainmentParallelCalls}

\begin{proof}
By definition, $\BoundedParOracle{\iExpTime{i}}{\Oracle{\iNExpTime{j}}{\SigmaP{c-1}}}{f(n)} = \DoubleBoundedParOracle{\iExpTime{i}}{\Oracle{\iNExpTime{j}}{\SigmaP{c-1}}}{f(n)}{1}$.
By \zcref{theo_exp_nexp_containment}, we have
\(\DoubleBoundedParOracle{\iExpTime{i}}{\Oracle{\iNExpTime{j}}{\SigmaP{c-1}}}{f(n)}{1} \subseteq \BoundedHausdCLASS{f(n)+1}{\Oracle{\iNExpTime{\ell}}{\SigmaP{c-1}}} \subseteq \DoubleBoundedPlusParOracle{\iExpTime{i'}}{\iNExpTime{j'}}{f(n)}{1}\).
Notice that
$\DoubleBoundedPlusParOracle{\iExpTime{i'}}{\iNExpTime{j'}}{f(n)}{1} \subseteq \BoundedParOracle{\iExpTime{i'}}{\Oracle{\iNExpTime{j'}}{\SigmaP{c-1}}}{f(n)+1}$, from which the statement follows.
\end{proof}

\getkeytheorem{ENContainmentAdaptiveCalls}

\begin{proof}
Notice that $\BoundedOracle{\iExpTime{i}}{\Oracle{\iNExpTime{j}}{\SigmaP{c-1}}}{f(n)} = \DoubleBoundedParOracle{\iExpTime{i}}{\Oracle{\iNExpTime{j}}{\SigmaP{c-1}}}{1}{f(n)}$.
By \zcref{theo_exp_nexp_containment}, we have
\(\DoubleBoundedParOracle{\iExpTime{i}}{\Oracle{\iNExpTime{j}}{\SigmaP{c-1}}}{1}{f(n)} \subseteq \BoundedHausdCLASS{2^{f(n)}}{\Oracle{\iNExpTime{\ell}}{\SigmaP{c-1}}} \subseteq \DoubleBoundedPlusParOracle{\iExpTime{i'}}{\Oracle{\iNExpTime{j'}}{\SigmaP{c-1}}}{1}{f(n)}\).
Observe now that
$\DoubleBoundedPlusParOracle{\iExpTime{i'}}{\Oracle{\iNExpTime{j'}}{\SigmaP{c-1}}}{1}{f(n)} \subseteq \BoundedOracle{\iExpTime{i'}}{\Oracle{\iNExpTime{j'}}{\SigmaP{c-1}}}{f(n)+1}$, and the statement follows.
\end{proof}

\getkeytheorem{ConstantRoundsParallelEqualsSingleRound}

\begin{proof}
Clearly, $\DoubleBoundedParOracle{\iExpTime{i}}{\Oracle{\iNExpTime{j}}{\SigmaP{c-1}}}{\iExpPolFunctions{g}}{k} \supseteq \BoundedParOracle{\iExpTime{i}}{\Oracle{\iNExpTime{j}}{\SigmaP{c-1}}}{\iExpPolFunctions{g}}$. By \zcref{theo_summary_generalized_equivalence_intermediate_levels_Hausdorff}, it holds $\BoundedParOracle{\iExpTime{i}}{\Oracle{\iNExpTime{j}}{\SigmaP{c-1}}}{\iExpPolFunctions{g}} = \BoundedParOracle{\iExpTime{i'}}{\Oracle{\iNExpTime{j'}}{\SigmaP{c-1}}}{\iExpPolFunctions{g}}$.
Thus, $\DoubleBoundedParOracle{\iExpTime{i}}{\Oracle{\iNExpTime{j}}{\SigmaP{c-1}}}{\iExpPolFunctions{g}}{k} \supseteq \BoundedParOracle{\iExpTime{i'}}{\Oracle{\iNExpTime{j'}}{\SigmaP{c-1}}}{\iExpPolFunctions{g}}$.

We now prove that $\DoubleBoundedParOracle{\iExpTime{i}}{\Oracle{\iNExpTime{j}}{\SigmaP{c-1}}}{\iExpPolFunctions{g}}{k} \subseteq \BoundedParOracle{\iExpTime{i'}}{\Oracle{\iNExpTime{j'}}{\SigmaP{c-1}}}{\iExpPolFunctions{g}}$.
Let $\Language{L} \in \DoubleBoundedParOracle{\iExpTime{i}}{\Oracle{\iNExpTime{j}}{\SigmaP{c-1}}}{\iExpPolFunctions{g}}{k}$ be a language.
We show that $\Language{L} \in \BoundedParOracle{\iExpTime{i'}}{\Oracle{\iNExpTime{j'}}{\SigmaP{c-1}}}{\iExpPolFunctions{g}}$ as well.
Since $\Language{L} \in \DoubleBoundedParOracle{\iExpTime{i}}{\Oracle{\iNExpTime{j}}{\SigmaP{c-1}}}{\iExpPolFunctions{g}}{k}$, there is a polynomial $p(n) \in \PolFunctions$ such that $\Language{L} \in \DoubleBoundedParOracle{\iExpTime{i}}{\Oracle{\iNExpTime{j}}{\SigmaP{c-1}}}{\iExp{g}{p(n)}}{k}$.
By \zcref{theo_exp_nexp_containment}, $\Language{L} \in \BoundedHausdCLASS{{(\iExp{g}{p(n)} + 1)}^{k}}{\Oracle{\iNExpTime{(i+j)}}{\SigmaP{c-1}}}$.
Observe that, for a fixed integer $k$, ${(\iExp{g}{p(n)} + 1)}^{k}$ is bounded by a function in $\iExpPolFunctions{g}$ (see \zcref{theo_polynomial_of_iterated_exponentials}.2 and \zcref{theo_polynomial_of_iterated_exponentials}.1), let us say $\iExp{g}{p'(n)}$, where $p'(n) \in \PolFunctions$.
Hence, by \zcref{theo_restricted_main_levels_longer_Hausdorff_as_expressive}, $\Language{L} \in \BoundedHausdCLASS{\iExp{g}{p'(n)}}{\Oracle{\iNExpTime{(i+j)}}{\SigmaP{c-1}}}$.
Since $\BoundedHausdCLASS{\iExp{g}{p'(n)}}{\Oracle{\iNExpTime{(i+j)}}{\SigmaP{c-1}}} \subseteq \BoundedHausdCLASS{\iExpPolFunctions{g}}{\Oracle{\iNExpTime{(i+j)}}{\SigmaP{c-1}}}$, we have $\Language{L} \in \BoundedHausdCLASS{\iExpPolFunctions{g}}{\Oracle{\iNExpTime{(i+j)}}{\SigmaP{c-1}}}$.
By \zcref{theo_summary_generalized_equivalence_intermediate_levels_Hausdorff}, $\BoundedHausdCLASS{\iExpPolFunctions{g}}{\Oracle{\iNExpTime{(i+j)}}{\SigmaP{c-1}}} = \BoundedParOracle{\iExpTime{i'}}{\Oracle{\iNExpTime{j'}}{\SigmaP{c-1}}}{\iExpPolFunctions{g}}$, by which $\Language{L} \in \BoundedParOracle{\iExpTime{i'}}{\Oracle{\iNExpTime{j'}}{\SigmaP{c-1}}}{\iExpPolFunctions{g}}$ follows.
\end{proof}

\getkeytheorem{ChainBooleanHierarchy}

\begin{proof}
By \zcref{theo_specific_chain_parallel_hausdorff}, for an integer constant $k$, we have $\BoundedHausdCLASS{k}{\Oracle{\iNExpTime{i}}{\SigmaP{c-1}}} \subseteq \BoundedParOracle{\iExpTime{i}}{\SigmaP{c}}{k} \subseteq \BoundedHausdCLASS{k+1}{\Oracle{\iNExpTime{i}}{\SigmaP{c-1}}}$.
On the other hand, if we consider the complement classes,\footnote{We here use the following property: if $\ComplexityClass{C}$ and $\ComplexityClass{D}$ are two language classes, then $\ComplexityClass{C} \subseteq \ComplexityClass{D} \Leftrightarrow \ComplementPrefix\ComplexityClass{C} \subseteq \ComplementPrefix\ComplexityClass{D}$ \cite[Lemma~1.15]{BalcazarDG1995}.} we obtain that $\ComplementPrefixKerned\BoundedHausdCLASS{k}{\Oracle{\iNExpTime{i}}{\SigmaP{c-1}}} \subseteq \ComplementPrefix\BoundedParOracle{\iExpTime{i}}{\SigmaP{c}}{k} \subseteq \ComplementPrefixKerned\BoundedHausdCLASS{k+1}{\Oracle{\iNExpTime{i}}{\SigmaP{c-1}}}$.
Since $\ComplementPrefix\BoundedParOracle{\iExpTime{i}}{\SigmaP{c}}{k} = \BoundedParOracle{\iExpTime{i}}{\SigmaP{c}}{k}$, by the relationships just highlighted, it follows that $(\BoundedHausdCLASS{k}{\Oracle{\iNExpTime{i}}{\SigmaP{c-1}}} \cup \ComplementPrefixKerned\BoundedHausdCLASS{k}{\Oracle{\iNExpTime{i}}{\SigmaP{c-1}}}) \subseteq \BoundedParOracle{\iExpTime{i}}{\SigmaP{c}}{k} \subseteq (\BoundedHausdCLASS{k+1}{\Oracle{\iNExpTime{i}}{\SigmaP{c-1}}} \cap \ComplementPrefixKerned\BoundedHausdCLASS{k+1}{\Oracle{\iNExpTime{i}}{\SigmaP{c-1}}})$.
\end{proof}

\getkeytheorem{SummaryNNHausdorff}

\begin{proof}
By \zcref{theo_nexp_nexp_containment}, we have $\Oracle{\iNExpTime{i}}{\Oracle{\iNExpTime{j}}{\SigmaP{c-1}}} \subseteq \BoundedHausdCLASS{\iExpPolFunctions{i+1}}{\Oracle{\iNExpTime{(i+j)}}{\SigmaP{c-1}}} \subseteq \BoundedOracle{\iNExpTime{i}}{\bigl(\BoundedHausdCLASS{{\scriptstyle 2}}{\Oracle{\iNExpTime{j}}{\SigmaP{c-1}}}\bigr)}{1} \subseteq \BoundedParOracle{\iNExpTime{i}}{\Oracle{\iNExpTime{j}}{\SigmaP{c-1}}}{2}$.
By $\BoundedParOracle{\iNExpTime{i}}{\Oracle{\iNExpTime{j}}{\SigmaP{c-1}}}{2} \subseteq \Oracle{\iNExpTime{i}}{\Oracle{\iNExpTime{j}}{\SigmaP{c-1}}}$, the statement follows.
\end{proof}

\getkeytheorem{NN1EquivalenceBH}

\begin{proof}
By \zcref{theo_NN1_containment_BH}, $\BoundedOracle{\iNExpTime{i}}{\Oracle{\iNExpTime{j}}{\SigmaP{c-1}}}{1} \subseteq \ComplementPrefixKerned\BoundedHausdCLASS{2}{\Oracle{\iNExpTime{\ell}}{\SigmaP{c-1}}} \subseteq \BoundedOracle{\iNExpTime{i'}}{\Oracle{\iNExpTime{j'}}{\SigmaP{c-1}}}{1}$.
Again by \zcref{theo_NN1_containment_BH}, $\BoundedOracle{\iNExpTime{i'}}{\Oracle{\iNExpTime{j'}}{\SigmaP{c-1}}}{1} \subseteq \ComplementPrefixKerned\BoundedHausdCLASS{2}{\Oracle{\iNExpTime{\ell}}{\SigmaP{c-1}}} \subseteq \BoundedOracle{\iNExpTime{i}}{\Oracle{\iNExpTime{j}}{\SigmaP{c-1}}}{1}$, and the statement follows.
\end{proof}

\getkeytheorem{SummarySNHausdorff}

\begin{proof}
\zcref[S]{theo_expspace_nexp_containment} implies that
\[\Oracle{\iExpSpace{(i-1)}}{\Oracle{\iNExpTime{j}}{\SigmaP{c-1}}} \subseteq \BoundedHausdCLASS{\iExpPolFunctions{i}}{\Oracle{\iNExpTime{(i+j)}}{\SigmaP{c-1}}} \subseteq \mkern -2mu
\begin{cases}
  \BoundedOracle{\iExpSpace{(i-1)}}{\Oracle{\iNExpTime{j}}{\SigmaP{c-1}}}{\iExpPolFunctions{(i-1)}} \\
  \BoundedParOracle{\iExpSpace{(i-1)}}{\Oracle{\iNExpTime{j}}{\SigmaP{c-1}}}{\iExpPolFunctions{i}}.
\end{cases}\]
From the fact that
\[
\left.\begin{aligned}
\BoundedOracle{\iExpSpace{(i-1)}}{\Oracle{\iNExpTime{j}}{\SigmaP{c-1}}}{\iExpPolFunctions{(i-1)}}& \\
\BoundedParOracle{\iExpSpace{(i-1)}}{\Oracle{\iNExpTime{j}}{\SigmaP{c-1}}}{\iExpPolFunctions{i}} = \ParOracle{\iExpSpace{(i-1)}}{\Oracle{\iNExpTime{j}}{\SigmaP{c-1}}}&
\end{aligned}
\right\}
\mkern -2mu \subseteq \Oracle{\iExpSpace{(i-1)}}{\Oracle{\iNExpTime{j}}{\SigmaP{c-1}}}
\]
also hold, the statement follows.
\end{proof}

\subsection%
{Proofs for \zcref{sec_canonical_hard_problems}}
\label{sec_detailed_proofs_canonical_hard_problems}
\silentbookmark[2]{Proofs for Section~\ref*{sec_canonical_hard_problems}}

\getkeytheorem{HardnessOddityGeneral}

\begin{proof}
We start by considering the case in which the language $\Language{A}$ is complete for $\Oracle{\iNExpTime{i}}{\SigmaP{c-1}}$.

(\emph{Membership}).
The task can easily be shown in $\BoundedParOracle{\iExpTime{i}}{\SigmaP{c}}{\PolFunctions}$.
Indeed, we first write on the query tape, for each string $w_z \in \StringTup{w}$, a string $\wt{w}_z$, where $\wt{w}_z$ is an \iExponential{i}{}ly-padded version of $w_z$.
Then, we ask the $\SigmaP{c}$ oracle to decide all the queries in parallel.
Notice that the oracle can actually decide $\Language{A} \in \Oracle{\iNExpTime{i}}{\SigmaP{c-1}}$, because the $\SigmaP{c} = \Oracle{\NPTime}{\SigmaP{c-1}}$ oracle receives in input \iExponential{i}{}ly-padded strings;
hence, the $\NPTime$ part of $\Oracle{\NPTime}{\SigmaP{c-1}}$ can act as an $\iNExpTime{i}$ oracle machine.
We answer \yeslbl iff the number of \yesansws obtained from the oracle is odd.

(\emph{Hardness}).
Let $\Language{L}$ be an arbitrary language in $\BoundedOracle{\iExpTime{i}}{\SigmaP{c}}{\LogFunctions} = \BoundedParOracle{\iExpTime{i}}{\SigmaP{c}}{\PolFunctions} = \BoundedParOracle{\iExpTime{i}}{\SigmaP{c}}{\PolFunctions}$.
We exhibit a reduction from $\Language{L}$ to the task in the statement of this \zcref*[typeset=name,nocap]{theo_hardness_oddity_general}.
This reduction builds from a string $w$ a tuple $\StringTup{w} = \tup{w_1,\dots,w_n}$ of strings such that $w \in \Language{L}$ iff the number of \yesinsts of $\Language{A}$ in $\StringTup{w}$ is odd.

Since $\Language{L} \in \BoundedOracle{\iExpTime{i}}{\SigmaP{c}}{\LogFunctions} = \BoundedParOracle{\iExpTime{i}}{\SigmaP{c}}{\PolFunctions} = \BoundedHausdCLASS{\PolFunctions}{\Oracle{\iNExpTime{i}}{\SigmaP{c-1}}}$ (see \zcref{theo_summary_equivalence_intermediate_levels_Hausdorff}), $\Language{L}$ is an $\Oracle{\iNExpTime{i}}{\SigmaP{c-1}}$ Hausdorff language of length $p(n)$, for some polynomial $p(n)$.
By this, there is an $\Oracle{\iNExpTime{i}}{\SigmaP{c-1}}$ Hausdorff predicate $\Language{D}$ of length $p(n)$ characterizing $\Language{L}$, which means that, for every string $w$, $w \in \Language{L} \Leftrightarrow \SetSize{\set{z \mid 1 \leq z \leq p(\StringLength{w}) \land \linebreak[0] \Language{D}(w,z) = 1}}$ is odd.
Notice that the relevant values of $z$ are polynomially bounded in $\StringLength{w}$ (by $p(\StringLength{w})$; see above).
Hence, the representation size of $z$ would not affect the complexity of the Hausdorff predicate $\Language{D}$ even if this were measured with respect to the whole pair $\tup{w,z}$, rather than $\StringLength{w}$ alone.

Therefore, since $\Language{A}$ is complete for $\Oracle{\iNExpTime{i}}{\SigmaP{c-1}}$, there is a polynomial reduction $g$ from $\Language{D}$ to $\Language{A}$.

Let $w_z = g(w,z)$, and consider the tuple $\StringTup{w} = \tup{w_1,\dots,w_{p(\StringLength{w})}}$ of strings.
Assume \Wlog that $p(\StringLength{w})$ is an even value, as, if this were not the case, we can append to the tuple a last (fixed) string that trivially is a \noinst for $\Language{A}$.
By construction, ${\Language{A}}(w_1) \geq \dots \geq {\Language{A}}(w_{p(\StringLength{w})})$, and the number of \yesinsts of $\Language{A}$ in $\StringTup{w}$ is odd iff $w \in \Language{L}$.
Hence, via the function $g$ we can build a polynomial reduction from $\Language{L}$ %
to the task of the statement, which is thus proven hard for $\BoundedOracle{\iExpTime{i}}{\SigmaP{c}}{\LogFunctions} = \BoundedParOracle{\iExpTime{i}}{\SigmaP{c}}{\PolFunctions} = \BoundedHausdCLASS{\PolFunctions}{\Oracle{\iNExpTime{i}}{\SigmaP{c-1}}}$.

\Proofsep

\noindent
Let us now consider the case in which the language $\Language{A}$ is complete for $\ComplementPrefixKerned\Oracle{\iNExpTime{i}}{\SigmaP{c-1}}$.

(\emph{Membership}).
Proving the membership in $\BoundedOracle{\iExpTime{i}}{\SigmaP{c}}{\LogFunctions} = \BoundedParOracle{\iExpTime{i}}{\SigmaP{c}}{\PolFunctions} = \BoundedHausdCLASS{\PolFunctions}{\Oracle{\iNExpTime{i}}{\SigmaP{c-1}}}$ for this case 
is very similar to the case above;
simply, the oracle now decides the $\Oracle{\iNExpTime{i}}{\SigmaP{c-1}}$\CompleteSuffix language~$\compl{\Language{A}}$ complement to $\Language{A}$, and we answer \yeslbl iff $n$ minus the number of \yesansws from \mbox{the oracle for $\compl{\Language{A}}$ is odd}.

(\emph{Hardness}).
Consider the $\Oracle{\iNExpTime{i}}{\SigmaP{c-1}}$\CompleteSuffix language $\compl{\Language{A}}$ complement to $\Language{A}$.
By the discussion above, deciding, for a tuple $\StringTup{w} = \tup{w_1,\dots,w_n}$ of strings, with $\compl{\Language{A}}(w_1) \geq \dots \geq \compl{\Language{A}}(w_n)$, whether the number of \yeslbl-instances of $\compl{\Language{A}}$ in $\StringTup{w}$ is odd is hard for $\BoundedOracle{\iExpTime{i}}{\SigmaP{c}}{\LogFunctions} = \BoundedParOracle{\iExpTime{i}}{\SigmaP{c}}{\PolFunctions} = \BoundedHausdCLASS{\PolFunctions}{\Oracle{\iNExpTime{i}}{\SigmaP{c-1}}}$, even if $n$ is restricted to be even.
We reduce this task to that of deciding, for a tuple $\StringTup{v} = \tup{v_1,\dots,v_m}$ of strings, with $\Language{A}(v_1) \geq \dots \geq \Language{A}(v_m)$, whether the number of \yesinsts of $\Language{A}$ in $\StringTup{v}$ is odd.
This shows that the task of the statement is hard for $\BoundedOracle{\iExpTime{i}}{\SigmaP{c}}{\LogFunctions} = \BoundedParOracle{\iExpTime{i}}{\SigmaP{c}}{\PolFunctions} = \BoundedHausdCLASS{\PolFunctions}{\Oracle{\iNExpTime{i}}{\SigmaP{c-1}}}$ also when $\Language{A}$ is $\ComplementPrefixKerned\Oracle{\iNExpTime{i}}{\SigmaP{c-1}}$\CompleteSuffix.

The tuple $\StringTup{v}$ is simply obtained from $\StringTup{w}$ by reverting the order of the strings in $\StringTup{w}$.
More precisely, $\StringTup{v}$ contains as many strings as those in $\StringTup{w}$, and, for each $1\leq z\leq n$, we let $v_z = w_{(n+1)-z}$.
Observe that $\compl{\Language{A}}(w_1) \geq \dots \geq \compl{\Language{A}}(w_n)$ implies that $\Language{A}(w_1) \leq \dots \leq \Language{A}(w_n)$, and, by the definition of the strings $v_z$, it follows that $\Language{A}(v_1) \geq \dots \geq \Language{A}(v_n)$---hence, the constraint on the input tuple is met.

We now show that $\StringTup{w}$ contains an odd number of \yeslbl-instances of $\compl{\Language{A}}$ if and only if $\StringTup{v}$ contains an odd number of \yeslbl-instances of $\Language{A}$.
If there is an odd number of \yesinsts of $\compl{\Language{A}}$ in~$\StringTup{w}$, then, since $n$ is even, there is also an odd number of \noinsts of $\compl{\Language{A}}$ in~$\StringTup{w}$.
By the fact that the strings $v_z \in \StringTup{v}$ are the same strings of those in~$\StringTup{w}$, just in a different order, we have that the number of \yesinsts of $\Language{A}$ in~$\StringTup{v}$ is odd as well.
By symmetry, if the number of \yesinsts of $\compl{\Language{A}}$ in~$\StringTup{w}$ is even, then the number of \yesinsts of $\Language{A}$ in~$\StringTup{v}$ is even as well.
{By the fact that $n$ is an even number, it is not difficult to show that $\StringTup{w}$ contains an odd number of \yeslbl-instances of $\compl{\Language{A}}$ if and only if $\StringTup{v}$ contains an odd number of \yeslbl-instances of $\Language{A}$.}
\end{proof}

\getkeytheorem{HardnessCountCompGeneral}

\begin{proof}
(\emph{Membership}).
A $\BoundedParOracle{\iExpTime{i}}{\SigmaP{c}}{\PolFunctions}$ procedure deciding the task of the statement can be as follows.
Let $\Language{C}$ be a language complete for $\Oracle{\iNExpTime{i}}{\SigmaP{c-1}}$ (resp., $\ComplementPrefixKerned\Oracle{\iNExpTime{i}}{\SigmaP{c-1}}$).
Since $\Language{A}$ and $\Language{B}$ are in $\Oracle{\iNExpTime{i}}{\SigmaP{c-1}}$ (resp., $\ComplementPrefixKerned\Oracle{\iNExpTime{i}}{\SigmaP{c-1}}$), there exist polynomial reductions $g$ and $h$ from $\Language{A}$ and $\Language{B}$ to $\Language{C}$, respectively.
We prepare the queries $g(w_1), \dots, g(w_n),\linebreak[0] h(v_1), \dots, h(v_m)$, each of them \iExponential{i}{}lly-padded, and then we ask them in parallel to an oracle for $\Language{C}$---an argument similar to that in the proof of \zcref{theo_hardness_oddity_general} shows that an oracle in $\SigmaP{c}$ can decide the language $\Language{C} \in \Oracle{\iNExpTime{i}}{\SigmaP{c-1}}$ when receiving \iExponential{i}{}ly-padded queries.
We conclude by counting the number of \yesansws in the two groups, and answer accordingly (see the proof of \zcref{theo_hardness_oddity_general}).

(\emph{Hardness}).
Let $\Language{A}$ be a language complete for $\Oracle{\iNExpTime{i}}{\SigmaP{c-1}}$ (resp., $\ComplementPrefixKerned\Oracle{\iNExpTime{i}}{\SigmaP{c-1}}$).
Let $\StringTup{u} = \tup{u_1,\dots,u_\ell}$ be a tuple of strings, with $\ell$ even and $\Language{A}(u_1) \geq \dots \geq \Language{A}(u_\ell)$.
By \zcref{theo_hardness_oddity_general}, deciding whether the number of \yesinsts of $\Language{A}$ in $\StringTup{u}$ is odd is complete for $\BoundedOracle{\iExpTime{i}}{\SigmaP{c}}{\LogFunctions} = \BoundedParOracle{\iExpTime{i}}{\SigmaP{c}}{\PolFunctions} = \BoundedHausdCLASS{\PolFunctions}{\Oracle{\iNExpTime{i}}{\SigmaP{c-1}}}$.
Starting from $\StringTup{u} = \tup{u_1,\dots,u_\ell}$, the reduction builds two tuples of strings $\StringTup{w} = \tup{w_1,\dots,w_n}$ and $\StringTup{v} = \tup{v_1,\dots,v_m}$ such that the number of \yesinsts of $\Language{A}$ in $\StringTup{u}$ is odd iff the number of \yesinsts of $\Language{A}$ in $\StringTup{w}$ is greater than the number of \yesinsts of $\Language{B}$(${} = \Language{A}$) in $\StringTup{v}$.

Intuitively, we build the tuples $\StringTup{w}$ and $\StringTup{v}$ by filling $\StringTup{w}$ with the odd\nbdash-indexed strings from $\StringTup{u}$, and $\StringTup{v}$ with the even\nbdash-indexed strings from $\StringTup{u}$.
More precisely, we define the tuples $\StringTup{w} = \tup{w_1,\dots,w_n}$ and $\StringTup{v} = \tup{v_1,\dots,v_m}$ as follows:
$w_z = u_{2z-1}$ and $v_z = u_{2z}$, for all $1 \leq z \leq \ell / 2$.
By definition, we have $\StringLength{\StringTup{w}} = \StringLength{\StringTup{v}}$ (because $\ell$ is even), $\Language{A}(w_1) \geq \dots \geq \Language{A}(w_n)$, and  $\Language{B}(v_1) \geq \dots \geq \Language{B}(v_m)$ (remember that we have chosen $\Language{B} = \Language{A}$).

It is not hard to verify that the number of \yesinsts of $\Language{A}$ in $\StringTup{u}$ is odd iff the number of \yesinsts of $\Language{A}$ in $\StringTup{w}$ is greater than the number of \yesinsts of $\Language{B}$(${} = \Language{A}$) in $\StringTup{v}$.
\end{proof}

\getkeytheorem{HardnessCountCompTwoLangOneSet}

\begin{proof}
(\emph{Membership}).
Proven by a procedure similar to that in the proof of \zcref{theo_hardness_count_comp_general_2L_2S}.
The difference here is that the parallel queries submitted to the oracle are $g(w_1), \dots, g(w_n),\linebreak[0] h(w_1), \dots, h(w_n)$.
Notice that some of the strings of $\StringTup{w}$ are instances of $\Language{A}$ and some others are instances of $\Language{B}$, and hence they might be over different alphabets---this does not happen when $\Language{A}$ and $\Language{B}$ are defined over the same alphabet, e.g., the binary one.
In the case of different alphabets, the reduction functions $g$ and $h$, since they both are run over all strings of $\StringTup{w}$, must be designed so that when they read an unexpected symbol, they accordingly produce a \noinst of $\Language{A}$ or $\Language{B}$.

(\emph{Hardness}).
Let $\Language{C}$ be a language complete for $\Oracle{\iNExpTime{i}}{\SigmaP{c-1}}$ (resp., $\ComplementPrefixKerned\Oracle{\iNExpTime{i}}{\SigmaP{c-1}}$).
Let $\StringTup{u} = \tup{u_1,\dots,u_r}$ and $\StringTup{v} = \tup{v_1,\dots,v_s}$ be two tuples of strings.
By \zcref{theo_hardness_count_comp_general_2L_2S}, deciding whether $\StringTup{u}$ contains more \yesinsts of $\Language{C}$ than $\StringTup{v}$ is complete for
$\BoundedOracle{\iExpTime{i}}{\SigmaP{c}}{\LogFunctions}
= \BoundedParOracle{\iExpTime{i}}{\SigmaP{c}}{\PolFunctions}
= \BoundedHausdCLASS{\PolFunctions}{\Oracle{\iNExpTime{i}}{\SigmaP{c-1}}}$.

We now exhibit a reduction transforming $\StringTup{u}$ and $\StringTup{v}$ into a tuple of strings $\StringTup{w} = \tup{w_1,\dots,w_n}$ such that the number of \yesinsts of $\Language{C}$ in $\StringTup{u}$ is greater than the number of \yesinsts of $\Language{C}$ in $\StringTup{v}$ iff, for languages $\Language{A}$ and $\Language{B}$, 
the number of \yesinsts of $\Language{A}$ in $\StringTup{w}$ is greater than the number of \yesinsts of $\Language{B}$ in $\StringTup{w}$.

Let $\hashsep$ and $\dollarsep$ be two fresh alphabet symbols.
In what follows, when strings include these fresh symbols, we actually mean suitable re-encodings of these strings over the binary alphabet.
We define the languages $\Language{C}_1$ and $\Language{C}_2$ as follows:
$\Language{C}_1 = \set{\hashsep x \mid x \in \StringUniverse \land x \in \Language{C}}$ and $\Language{C}_2 = \set{\dollarsep x \mid x \in \StringUniverse \land x \in \Language{C}}$.
These two languages are in $\Oracle{\iNExpTime{i}}{\SigmaP{c-1}}$ (resp., $\ComplementPrefixKerned\Oracle{\iNExpTime{i}}{\SigmaP{c-1}}$) by definition.
We define the tuple of strings $\StringTup{w}' = \tup{\hashsep u_1,\dots,\hashsep u_r,\dollarsep v_1,\dots,\dollarsep v_s}$.
It is not hard to verify that the number of \yesinsts of $\Language{C}$ in $\StringTup{u}$ is greater than the number of \yesinsts of $\Language{C}$ in $\StringTup{v}$ iff the number of \yesinsts of $\Language{C}_1$ in $\StringTup{w}'$ is greater than the number of \yesinsts of $\Language{C}_2$ in $\StringTup{w}'$.

Since $\Language{C}_1$ and $\Language{C}_2$ are in $\Oracle{\iNExpTime{i}}{\SigmaP{c-1}}$ (resp., $\ComplementPrefixKerned\Oracle{\iNExpTime{i}}{\SigmaP{c-1}}$) and $\Language{A}$ and $\Language{B}$ are complete for $\Oracle{\iNExpTime{i}}{\SigmaP{c-1}}$ (resp., $\ComplementPrefixKerned\Oracle{\iNExpTime{i}}{\SigmaP{c-1}}$), there exist polynomial reductions $g$ and $h$ from $\Language{C}_1$ and $\Language{C}_2$ to $\Language{A}$ to $\Language{B}$, respectively.
Consider the tuple of strings $\StringTup{w} = \tup{g(\hashsep u_1),\dots,g(\hashsep u_r),h(\dollarsep v_1),\dots,h(\dollarsep v_s)}$.
By definition of $\StringTup{w}$, the number of \yesinsts of $\Language{C}_1$ in $\StringTup{w}'$ is greater than the number of \yesinsts of $\Language{C}_2$ in $\StringTup{w}'$ iff the number of \yesinsts of $\Language{A}$ in $\StringTup{w}$ is greater than the number of \yesinsts of $\Language{B}$ in $\StringTup{w}$.
\end{proof}

\getkeytheorem{NPNexpHardnessTwoMachines}

\begin{proof}
For presentation purposes, let us call \NiExpNjExpGenericProblem{i}{j} the problem defined in the statement.

\emph{(Membership).}
First notice that, for a triple $\tup{\Machine{M}_x,y,z}$, where $\Machine{M}_x$ is a (string encoding a) machine, $y$ is a string, and $z$ is an integer in unary, with $\StringLength{z}$ polynomially\nbdash-bounded in $\StringLength{y}$, deciding whether $\Machine{M}_x$ accepts $y$ within $\iExp{k}{z}$ steps, where $k \geq 0$ is a fixed integer, can be shown in $\iNExpTime{k}$.
The proof is similar to showing that deciding whether $\Machine{M}_x$ accepts $y$ within $z$ steps is in \NPTime (see, e.g., \cite[Theorem~2.9]{Arora2009} or \cite[Theorem~2.19]{Goldreich2008}).

We claim that \NiExpNjExpGenericProblem{i}{j} is in $\Oracle{\iNExpTime{i}}{\iNExpTime{j}}$.
Indeed, an $\iNExpTime{i}$ oracle machine $\Oracle{\Machine{M}}{?}$ can first guess a string $v$ that is \iExponential{i}{}ly-long in~$\StringLength{r}$.
Then, via two $\iNExpTime{j}$ oracle calls ``$\tup{\Machine{M}_\alpha,v,t}$'' and ``$\tup{\Machine{M}_\beta,v,t}$'' to a suitable oracle $\Gamma$, the machine $\Oracle{\Machine{M}}{?}$ can check that~$v$ is accepted by $M_\alpha$ within $\iExp{i+j}{t}$ steps, and $v$ is not accepted by $M_\beta$ within $\iExp{i+j}{t}$ steps.
Observe that $\Gamma$ does \emph{not} need to be a $\iNExpTime{(i+j)}$ oracle, but being a $\iNExpTime{j}$ is enough.
Indeed, the queries that $\Gamma$ receives from $\Oracle{\Machine{M}}{?}$ contain a string~$v$ that is \iExponential{i}{}ly-long in the size of the string in input to $\Oracle{\Machine{M}}{?}$; 
for this reason, although $\Gamma$ is a \iExponential{j}\nbdash-time machine, it can run for $\iExp{i+j}{t}$ steps.

\emph{(Hardness).}
Let $\Language{L}$ be an arbitrary $\Oracle{\iNExpTime{i}}{\iNExpTime{j}}$ language.
We prove $\Language{L} \KarpRed{}$\NiExpNjExpGenericProblem{i}{j} by exhibiting a reduction that, for every string $w$, produces a four\nbdash-tuple $\mathcal{I}_w = \tup{\Machine{M}_{\alpha_w},\Machine{M}_{\beta_w},r_w,t_w}$ such that $w \in \Language{L}$ iff $\mathcal{I}_w$ is a \yesinst of \NiExpNjExpGenericProblem{i}{j}.
Since $\Language{L} \in \Oracle{\iNExpTime{i}}{\iNExpTime{j}}$, by \zcref{theo_summary_nexp_nexp_Hausdorff} $\Language{L} \in \BoundedHausdCLASS{\iExpPolFunctions{i+1}}{\iNExpTime{(i+j)}}$ as well.
Therefore, there are a polynomial $p(n)$ and a $\iNExpTime{(i+j)}$ Hausdorff predicate $\Language{D}$ of length $\iExp{i+1}{p(n)}$ such that, for every string $w$, $w \in \Language{L}$ iff the Hausdorff index $\HausdIndex{w}{\Language{D}}$ of $w$ \Wrt $\Language{D}$ is odd.
Hence,
\begin{equation}
\label{eq_w_belonging_to_L_iff_Hausdorff_pair}
w \in \Language{L} 
    \Leftrightarrow (\exists z \in \NaturalsDomain) \; (1 \leq z \leq \iExp{i+1}{p(\StringLength{w})} \land \text{$z$ is odd} \land \Language{D}(w,z) = 1 \land \Language{D}(w,z+1) = 0). 
\end{equation}

From~$w$, we obtain $\mathcal{I}_w = \tup{\Machine{M}_{\alpha_w},\Machine{M}_{\beta_w},r_w,t_w}$ as follows.
Let $\Omega$ be a machine deciding $\Language{D}$.
Strings $\Machine{M}_{\alpha_w}$ and $\Machine{M}_{\beta_w}$ encode transition functions for machines that, on every binary input string $z$, 
answer as follows:
\begin{equation}
\label{eq_deinition_M_alpha_and_beta}
\begin{aligned}
\Machine{M}_{\alpha_w}(z) = 1 & \Leftrightarrow \underbracket[.5pt]{z \in 1 | ( 1 (0|1)^* 1 )}_{\substack{\text{i.e., $z$ is the canonical binary}\\\text{representation of an odd integer}}} {} \land \Omega(w, z) = 1 \\[-1.35ex]
\Machine{M}_{\beta_w}(z) = 1 & \Leftrightarrow \overbracket[.5pt]{z \in 1 | ( 1 (0|1)^* 1 )}^{\smash{\phantom{\substack{\text{i.e., $z$ is the canonical binary}\\\text{representation of an odd integer}}}}} {} \land \Omega(w, z + 1) = 1.
\end{aligned}
\end{equation}

Consider $\Machine{M}_{\beta_w}$ ($\Machine{M}_{\alpha_w}$ is a simplified version of $\Machine{M}_{\beta_w}$).
Its transition function specifies a three\nbdash-phase computation:
(i)~$\Machine{M}_{\beta_w}$ checks whether the input string $z$ belongs to the regular language $1 | (1 (0|1)^* 1)$, i.e., checks whether $z$ is the canonical binary representation of an odd integer, and writes $z' = z + 1$ on a work tape;
(ii)~it prepends on the work tape the string $w$ and the symbol `$\hashsep$' to $z'$, obtaining $w \hashsep z'$ on the work tape; and
(iii)~it simulates $\Omega$ using the work tape as its input tape.
We claim that (the encoding of) the transition function of $\Machine{M}_{\beta_w}$ can be obtained in polynomial time from~$w$.
Indeed, the parts of the transition function encoding phases~(i) and~(iii) are constant \Wrt~$w$, and the part of the transition function encoding phase~(ii) requires linear time in~$\StringLength{w}$ to be generated.
Similarly, the transition function of~$\Machine{M}_{\alpha_w}$ can be obtained in polynomial time from~$w$.

Numbers $r_w$ and $t_w$ are defined as follows.
The former is simply $r_w = p(\StringLength{w})$.
Notice that, since $p(\cdot)$ is a fixed polynomial as it depends on the Hausdorff predicate $\Language{D}$, $r_w$ can be obtained in polynomial time in the size of~$w$ (see \zcref{sec_maths_complexity}, Polynomials).
Regarding the number $t_w$, it needs to be big enough so that $\Machine{M}_{\alpha_w}$ and $\Machine{M}_{\beta_w}$ can conclude their computations within $\iExp{i+j}{t_w}$ steps on an input of length at most $\iExp{i}{r_w}$.

We can estimate a lower bound value for $t_w$ as follows.
Since $\Language{D}$ is a $\iNExpTime{(i+j)}$ Hausdorff predicate, the machine $\Omega$ deciding $\pair{w,z} \in \Language{D}$ runs in time $\iExp{i+j}{q(\StringLength{w})}$, for some polynomial $q(n)$---%
remember that the running time of a machine deciding a Hausdorff predicate is linked to the size of $w$ alone. 
Let us consider $\Machine{M}_{\beta_w}$ and its computation phases (the running time of $\Machine{M}_{\alpha_w}$ is not higher than that of $\Machine{M}_{\beta_w}$).
Phase~(i) can be carried out in $2 \cdot \StringLength{z} + 2$ steps.
Indeed, we need to scan the input string $z$ once (costing $\StringLength{z}$ steps) and double check that:
either $z = {} \text{``$1$''}$ (which is the canonical binary form of the odd integer $1$), or $z$ starts with `$1$' and it ends with `$1$', so $z$ is a canonical binary representation of some odd integer.
While doing this, we copy $z$ on a work tape;
notice that, at the end of this process, the head of the work tape is on the first blank tape cell on the right of the copy of $z$.
We then need to add $1$ to the copy of $z$ on the work tape and obtain $z'$, and then move the work tape head to the first blank tape cell on the left of the work tape content.
This can be done by overwriting the right-most bits `$1$' of $z$ with `$0$', until we reach the right-most bit `$0$' which is overwritten with `$1$'.
At that point, the head is moved towards the required tape cell.
If there is no bit `$0$', after overwriting all the bits `$1$' with `$0$', an extra bit `$1$' is prepended to the string (and the head is moved one step to the left).
This process takes no more than $2 \cdot \StringLength{z} + 2$ steps.
Phase~(ii) requires to write ``$w\hashsep$'' on the work tape and position the head on the first symbol of $w$.
To do this, we need $\StringLength{w} + 2$ steps (we can write $w$ leftward).
Phase~(iii) requires time $\iExp{i+j}{q(\StringLength{w})}$, as it replicates the execution of $\Omega$ over ``$w \hashsep z'$''.
Since $1 \leq z \leq \iExp{i+1}{p(\StringLength{w})}$ and $z$ must be odd, we have that $z \leq \iExp{i+1}{p(\StringLength{w})} - 1$.
Hence, $\StringLength{z} \leq \iExp{i}{p(\StringLength{w})}$, and the overall running time for $\Machine{M}_{\beta_w}$ over an input string $z$ is bounded by $2 \cdot \iExp{i}{p(\StringLength{w})} + 2 + \StringLength{w} + 2 + \iExp{i+j}{q(\StringLength{w})}$.
The running time of $\Machine{M}_{\alpha_w}$ is bounded by the same function, as the computation of $\Machine{M}_{\alpha_w}$ takes no longer than the computation of $\Machine{M}_{\beta_w}$.
Notice that:
\begin{multline*}
2 \cdot \iExp{i}{p(\StringLength{w})} + \StringLength{w} + 4 + \iExp{i+j}{q(\StringLength{w})}
{} \leq {}
\iExp{i+j}{2 \cdot p(\StringLength{w})} + \iExp{i+j}{\StringLength{w} + 4} + \iExp{i+j}{q(\StringLength{w})}
{} \leq {} \\
\iExp{i+j}{2 \cdot p(\StringLength{w})} \cdot \iExp{i+j}{\StringLength{w} + 4} \cdot \iExp{i+j}{q(\StringLength{w})}
{} \leq {}
\iExp{i+j}{2 \cdot p(\StringLength{w}) + \StringLength{w} + 4 + q(\StringLength{w})}.
\end{multline*}
We hence take $t_w > 2 \cdot p(\StringLength{w}) + \StringLength{w} + 4 + q(\StringLength{w})$, which can be computed in polynomial time in~$\StringLength{w}$, as $p(\cdot)$ and $q(\cdot)$ are fixed polynomial depending on $\Language{D}$ and $\Omega$ (see \zcref{sec_maths_complexity}, Arithmetic and Polynomials).

Since
(a)~$r_w$ bits are enough to represent odd integers up to $\iExp{i}{p(\StringLength{w})}-1$,
(b)~$\iExp{i+j}{t_w}$ steps suffice for $\Machine{M}_{\alpha_w}$ and $\Machine{M}_{\beta_w}$ to complete their computations on inputs of size $r_w$, and
(c)~these machines are defined as in \zcref{eq_deinition_M_alpha_and_beta},
by \zcref{eq_w_belonging_to_L_iff_Hausdorff_pair} we have that
$w \in \Language{L}$ iff $\tup{\Machine{M}_{\alpha_w},\Machine{M}_{\beta_w},r_w,t_w}$ is a \yesinst of \NiExpNjExpGenericProblem{i}{j}.

To conclude, it can be shown that hardness holds even when $v$ must have length exactly $\iExp{i}{r_w}$.
To achieve this, the reduction needs to be slightly amended, and in particular what needs to be changed is the working of the machines $\Machine{M}_{\alpha_w}$ and $\Machine{M}_{\beta_w}$.
Essentially, these machines receive in input binary strings that are not necessarily canonical binary representations of integers, i.e., there can be leading zeros, and need to check in their first computation phase that their input has length exactly $\iExp{i}{r_w}$.
So, in the transition function of these machines, it is hardcoded an initial computation that writes on a work tape a string of length $r_w = p(\StringLength{w})$.
Observe that this part of the transition function can be generated in polynomial time in $\StringLength{w}$.
Then, these machines compute the value $\iExp{i}{r_w}$.
Notice that this can be done in time $O(\iExp{i-1}{r_w})$ (see \zcref{sec_maths_complexity}; Iterated exponentials of polynomials), and hence also in time $O(\iExpPolFunctions{i+j})$, which is the order of magnitude of the running time of these machines.
Then, the machines check to have in input a string of the desired length ending with `$1$', and while doing so they delete the leading zeros of the input string.
After this initial computation phase, the machines carry out a computation similar to the one described in the proof above.
We only need to choose a value for $t_w$ that accounts for this extended computation of the machines.
\end{proof}

\subsection%
{Proofs for \zcref{sec_qbsf_hard_problems}}
\label{sec_detailed_proofs_qbsf_hard_problems}
\silentbookmark[2]{Proofs for Section~\ref*{sec_qbsf_hard_problems}}

\getkeytheorem{ArithmeticEncodedAsBooleanFunctions}

\begin{proof}
Let us assume that the ordered sets of propositional variables $\VarSet{a}$, $\VarSet{a}'$, $\VarSet{a}''$, $\VarSet{b}$, and $\VarSet{b}'$ are:
$\VarSet{a} = \set{a_{u}, \dots, a_1}$, $\VarSet{a}' = \set{a_{u}', \dots, a_1'}$, $\VarSet{a}'' = \set{a_{u}'', \dots, a_1''}$, $\VarSet{b} = \set{b_{u}, \dots, b_1}$, and $\VarSet{b}' = \set{b_{u}', \dots, b_1'}$.

\begin{description}
  \item[Successor:]
    We characterize the successor function $\mi{succ}$ recursively.
    For numbers fitting within a single bit (i.e., their $u-1$ most significant bits are `$0$'):
    $b_1$ is the successor of $a_1$ iff $b_1$ is $\valtrue$ \emph{and} $a_1$ is $\valfalse$.
    Which~is:
    \begin{multline*}
       S_1 \equiv  \mi{succ}(\underbracket[.5pt]{0,\dots,0}_{(u-1)\text{-many}},0;\underbracket[.5pt]{0,\dots,0}_{(u-1)\text{-many}},1)
        \land
        \lnot \mi{succ}(\underbracket[.5pt]{0,\dots,0}_{(u-1)\text{-many}},0;\underbracket[.5pt]{0,\dots,0}_{(u-1)\text{-many}},0)
        \land {} \\
        \lnot \mi{succ}(\underbracket[.5pt]{0,\dots,0}_{(u-1)\text{-many}},1;\underbracket[.5pt]{0,\dots,0}_{(u-1)\text{-many}},0)
        \land
        \lnot \mi{succ}(\underbracket[.5pt]{0,\dots,0}_{(u-1)\text{-many}},1;\underbracket[.5pt]{0,\dots,0}_{(u-1)\text{-many}},1).    
    \end{multline*}
    
    To define the function $\mi{succ}$ for numbers fitting $i > 1$ bits (i.e., their $u-i$ most significant bits are `$0$'), we can simply rely on the definition of $\mi{succ}$ for numbers fitting within $(i-1)$ bits.
    Let us consider the two numbers $\VarSet{a}\langle i \rangle$ and $\VarSet{b}\langle i \rangle$ of $i$ bits represented over the Boolean variables $\set{a_{i},\dots,a_1}$ and $\set{b_{i},\dots,b_1}$, respectively.
    The successor relation between $\VarSet{b}\langle i \rangle$ and $\VarSet{a}\langle i \rangle$ follows the rules below:
    \begin{itemize}[nosep,label=--]
      \item if $b_{i} = a_{i}$, then $\VarSet{b}\langle i \rangle$ is the successor of $\VarSet{a}\langle i \rangle$ iff $\VarSet{b}\langle i {-} 1 \rangle$ is the successor of $\VarSet{a}\langle i {-} 1 \rangle$; and,
      \item if $b_{i} = \valtrue$ and $a_{i} = \valfalse$, then $\VarSet{b}\langle i \rangle$ is the successor of $\VarSet{a}\langle i \rangle$ iff $b_{i-1} = \dots = b_1 = \valfalse$ and $a_{i-1} = \dots = a_1 = \valtrue$; and,
      \item if $b_{i} = \valfalse$ and $a_{i} = \valtrue$, then $\VarSet{b}\langle i \rangle$ is \emph{not} the successor of $\VarSet{a}\langle i \rangle$ (irrespective of what $\VarSet{b}\langle i {-} 1 \rangle$ and $\VarSet{a}\langle i {-} 1 \rangle$ are).
    \end{itemize}
    The above rules can be encoded as:
    \begin{align*}
       S_i \equiv {} &\begin{multlined}[t][.85\columnwidth]
       \Bigl(\mi{succ}(\underbracket[.5pt]{0,\dots,0}_{(u-i)\text{-many}},0,a_{i-1},\dots,a_1;\underbracket[.5pt]{0,\dots,0}_{(u-i)\text{-many}},0,b_{i-1},\dots,b_1) \rightarrow {} \\ \mi{succ}(\underbracket[.5pt]{0,\dots,0}_{(u-i)\text{-many}},1,a_{i-1},\dots,a_1;\underbracket[.5pt]{0,\dots,0}_{(u-i)\text{-many}},1,b_{i-1},\dots,b_1)\Bigr)
       \land{}
       \end{multlined}\displaybreak[0]\\
       &\begin{multlined}[t][.85\columnwidth]
       \Bigl(\lnot \mi{succ}(\underbracket[.5pt]{0,\dots,0}_{(u-i)\text{-many}},0,a_{i-1},\dots,a_1;\underbracket[.5pt]{0,\dots,0}_{(u-i)\text{-many}},0,b_{i-1},\dots,b_1) \rightarrow {} \\ \lnot \mi{succ}(\underbracket[.5pt]{0,\dots,0}_{(u-i)\text{-many}},1,a_{i-1},\dots,a_1;\underbracket[.5pt]{0,\dots,0}_{(u-i)\text{-many}},1,b_{i-1},\dots,b_1)\Bigr)
       \land{}
       \end{multlined}\displaybreak[0]\\
       &\begin{multlined}[t][.85\columnwidth] \mi{succ}(\underbracket[.5pt]{0,\dots,0}_{(u-i)\text{-many}},0,\underbracket[.5pt]{1,\dots,1}_{(i-1)\text{-many}};\underbracket[.5pt]{0,\dots,0}_{(u-i)\text{-many}},1,\underbracket[.5pt]{0,\dots,0}_{(i-1)\text{-many}}) \land{}\\
       \shoveright{\bigwedge\limits_{j = 1}^{i - 1} \lnot \mi{succ}(\underbracket[.5pt]{0,\dots,0}_{(u-i)\text{-many}},0,\underbracket[.5pt]{a_{i-1},\dots,0,\dots,a_1}_{\text{`$0$' in $j$-th position}};\underbracket[.5pt]{0,\dots,0}_{(u-i)\text{-many}},1,b_{i-1},\dots,b_1)
       \land {}}\\
       \bigwedge\limits_{j = 1}^{i - 1} \lnot \mi{succ}(\underbracket[.5pt]{0,\dots,0}_{(u-i)\text{-many}},0,a_{i-1},\dots,a_1;\underbracket[.5pt]{0,\dots,0}_{(u-i)\text{-many}},1,\underbracket[.5pt]{b_{i-1},\dots,1,\dots,b_1}_{\text{`$1$' in $j$-th position}})
       \land {}
       \end{multlined}\displaybreak[0]\\
       & \lnot \mi{succ}(\underbracket[.5pt]{0,\dots,0}_{(u-i)\text{-many}},1,a_{i-1},\dots,a_1;\underbracket[.5pt]{0,\dots,0}_{(u-i)\text{-many}},0,b_{i-1},\dots,b_1).
    \end{align*}
    
    Hence, the successor relation, for numbers represented over $u$\nbdash-many bits, can be encoded as:
    \[
        S(\mi{succ},\VarSet{a},\VarSet{b}) \equiv \bigwedge\limits_{i=1}^{u} S_i,
    \]
    which, by $\alpha \rightarrow \beta \equiv \lnot \alpha \lor \beta$, is a \CNF formula of polynomially\nbdash-many clauses of two literals at most.
    
  \item[Majority:]
    The majority function can be defined based via the successor one and transitivity.
    Specifically, if $b$ is the successor of $a$, then $b$ is greater than $a$;
    and, if $b$ is the successor of $a'$, which in turn is greater than $a$, then $b$ is greater than $a$, as well.
    We also need to impose that, for every pair of integers $a$ and $b$, 
    if $b$ is greater than $a$, then $a$ is \emph{not} greater than $b$ and hence $\lnot \mi{less}(\VarSet{b},\VarSet{a})$ has to be $\valtrue$.
    This can be encoded as:
    \begin{align*}
        M(\mi{less},\mi{succ},\VarSet{a},\VarSet{a}',\VarSet{b}) \equiv {} 
        &\bigl( \mi{succ}(\VarSet{a};\VarSet{b}) \rightarrow \mi{less}(\VarSet{a};\VarSet{b}) \bigr) \land
        \bigl( \mi{less}(\VarSet{a};\VarSet{a}') \land \mi{succ}(\VarSet{a}';\VarSet{b}) \rightarrow \mi{less}(\VarSet{a};\VarSet{b}) \bigr) \land {} \\
        &\bigl(\mi{less}(\VarSet{a};\VarSet{b}) \rightarrow \lnot \mi{less}(\VarSet{b};\VarSet{a})\bigr).
    \end{align*}
    This formula is in \CNF and has three clauses of three literals at most.
    
  \item[Inequality:]
    The inequality function is defined using the majority one.
    Simply, if $b$ is greater than $a$, then $a$ and $b$ are different.
    We also require that an integer is not different from itself.
    This can be encoded as:
    \[
        D(\mi{less},\mi{neq},\VarSet{a},\VarSet{b}) \equiv 
        \bigl( \mi{less}(\VarSet{a};\VarSet{b}) \rightarrow \mi{neq}(\VarSet{a};\VarSet{b}) \bigr) \land
        \bigl( \mi{less}(\VarSet{a};\VarSet{b}) \rightarrow \mi{neq}(\VarSet{b};\VarSet{a}) \bigr) \land
        \lnot \mi{neq}(\VarSet{a};\VarSet{a}).
    \]
    Notice that we do not need to explicitly deal with the case in which $\mi{less}(\VarSet{b},\VarSet{a})$ is $\valtrue$, as the variables $\VarSet{a}$ and $\VarSet{b}$ are universally quantified.
    The two clauses are needed because the inequality relation is symmetric, whereas the majority one is not.
    This formula is in \CNF and has three clauses of two literals at most.
    
  \item[Addition:]
    The addition function is defined via the successor and the inequality ones.
    The formulas encode these rules:
    summing $0$ to $a$ yields $a$, and
    summing the successor of $a'$ to $a$ yields the successor of $a + a'$.
    \begin{align*}
        A(\mi{add},\mi{neq},\mi{succ},\VarSet{a},\VarSet{a}',\VarSet{a}'',\VarSet{b},\VarSet{b}') \equiv {} 
        &\mi{add}(\VarSet{a};\underbracket[.5pt]{0,\dots,0}_{u\text{-many}};\VarSet{a}) \land \bigl(\mi{neq}(\VarSet{a};\VarSet{b}) \rightarrow \lnot\mi{add}(\VarSet{a};\underbracket[.5pt]{0,\dots,0}_{u\text{-many}};\VarSet{b})\bigr) \land {} \\
        &\bigl(\mi{add}(\VarSet{a};\VarSet{a}';\VarSet{b}) \land \mi{succ}(\VarSet{a}';\VarSet{a}'') \land \mi{succ}(\VarSet{b};\VarSet{b}') \rightarrow \mi{add}(\VarSet{a};\VarSet{a}'';\VarSet{b}')\bigr) \land {} \\
        &\bigl(\mi{add}(\VarSet{a};\VarSet{a}';\VarSet{b}) \land \mi{neq}(\VarSet{b};\VarSet{b}') \rightarrow \lnot \mi{add}(\VarSet{a};\VarSet{a}';\VarSet{b}')\bigr).
    \end{align*}
    This formula is in \CNF and has four clauses of four literals at most.
  
\end{description}

\noindent
The overall formula is then:
\begin{multline*}
    \mu^u(\mi{succ},\mi{less},\mi{neq},\mi{add},\VarSet{a},\VarSet{a}',\VarSet{a}'',\VarSet{b},\VarSet{b}') \equiv S(\mi{succ},\VarSet{a},\VarSet{b}) \land M(\mi{less},\mi{succ},\VarSet{a},\VarSet{a}',\VarSet{b}) \land {} \\ D(\mi{less},\mi{neq},\VarSet{a},\VarSet{b}) \land A(\mi{add},\mi{neq},\mi{succ},\VarSet{a},\VarSet{a}',\VarSet{a}'',\VarSet{b},\VarSet{b}').
\end{multline*}
which is in \CNF, and with polynomially-many clauses of four literals at most.
It is not hard to see that, by its definition, the only model of the quantified formula $(\forall \VarSet{a},\VarSet{a}',\VarSet{a}'',\VarSet{b},\VarSet{b}') \PrefixMatrixSeparator \mu^u(\mi{succ},\mi{less},\mi{neq},\mi{add},\VarSet{a},\VarSet{a}',\VarSet{a}'',\VarSet{b},\VarSet{b}')$ must be such that the functions $\mi{succ}$, $\mi{less}$, $\mi{neq}$, and $\mi{add}$, are interpreted as the successor, majority, inequality, and addition, relations between fixed\nbdash-length binary\nbdash-represented integers.
\end{proof}

\getkeytheorem{CodingNEXPMachineHausdorffPredicate}

\begin{proof}
The construction here proposed is inspired by \citeauthor{Cook1971}'s~\cite{Cook1971} reduction proving the \NPTimeh{}ness of \Sat.
The key difference between \citeauthor{Cook1971}'s reduction 
and the construction here proposed is that, instead of employing Boolean propositional variables to encode a polynomial\nbdash-time machine computation, we here employ Boolean function variables to compactly (i.e., polynomially) encode exponential\nbdash-time computations.

Let $\Language{D} \subseteq \HausdPredDomain$ be a $\Oracle{\NExpTime}{\SigmaP{c-1}}$ Hausdorff predicate.
Since $\Language{D} \in \Oracle{\NExpTime}{\SigmaP{c-1}}$, by \zcref{theo_restricted_main_levels_longer_Hausdorff_bound} the length of $\Language{D}$ is doubly exponential at most.
If $\Language{D}$ has doubly exponential (resp., exponential, polynomial) length, the length of $\Language{D}$ is bounded by $2^{2^{\PolynomialLengthHausPredD(n)}} - 1$ (resp., $2^{\PolynomialLengthHausPredD(n)} - 1$, $\PolynomialLengthHausPredD(n) - 1$), for some polynomial $\PolynomialLengthHausPredD(n)$ (see \zcref{theo_Hausd_pred_bounded_complexity_imply_bounded_length}).

Remember also that the problem of deciding whether a pair $\pair{w,z}$ belongs to $\Language{D}$ can be solved by processing a single input string ``$w \hashsep z$'', and the computation time is measured \Wrt $\StringLength{w}$ alone (see \zcref{sec_Hausdorff_reductions_classes}, last paragraph).
By $\Language{D} \in \Oracle{\NExpTime}{\SigmaP{c-1}}$ and the certificate\nbdash-based definition of the \WEHStressedText main levels (see \zcref{sec_prelim_PH_ExpH}), there is a polynomial $\PolynomialSizeCertificates(n)$ and a \emph{deterministic polynomial} $(c{+}1)$\nbdash-ary predicate $R$ such that:
\begin{multline*}
  \pair{w,z} \in \Language{D} \Leftrightarrow 
    (\exists u_1 \in \alphabet^{\leq 2^{\PolynomialSizeCertificates(\StringLength{w})}})
    (\forall u_2 \in \alphabet^{\leq 2^{\PolynomialSizeCertificates(\StringLength{w})}}) \cdots
    (Q_c u_c \in \alphabet^{\leq 2^{\PolynomialSizeCertificates(\StringLength{w})}}) \\
    R(w \hashsep z,u_1,\dots,u_c) = 1,
\end{multline*}
where $Q_c = \exists$ or $Q_c = \forall$ when $c$ is odd or even, respectively.
By known padding techniques (see, e.g., \cite{Arora2009,Lohrey2012}), the quantifications `$(Q_i u_i \in \alphabet^{\leq 2^{\PolynomialSizeCertificates(\StringLength{w})}})$', where $Q_i \in \set{\exists , \forall}$, can be replaced by `$(Q_i u_i \in \alphabet^{2^{\PolynomialSizeCertificates(\StringLength{w})}})$', which means that the length of the certificates $u_i$ can be assumed to precisely be $2^{\PolynomialSizeCertificates(\StringLength{w})}$.
Therefore, we have:
\begin{multline}
\label{eq_condition_NEXP_Hausdorff_language}
  \pair{w,z} \in \Language{D} \Leftrightarrow 
    (\exists u_1 \in \alphabet^{2^{\PolynomialSizeCertificates(\StringLength{w})}})
    (\forall u_2 \in \alphabet^{2^{\PolynomialSizeCertificates(\StringLength{w})}}) \cdots
    (Q_c u_c \in \alphabet^{2^{\PolynomialSizeCertificates(\StringLength{w})}}) \\
    R(w \hashsep z,u_1,\dots,u_c) = 1.
\end{multline}

Since $R$ is deterministic polynomial, there exists a deterministic machine $\Omega$ deciding $R$ in polynomial time in the combined sizes of $w$ and $u_1, \dots, u_c$---the size of $z$ is excluded, as the computation time for Hausdorff predicates is measured \Wrt the size of $w$ alone.
Because $\StringLength{u_i} = 2^{\PolynomialSizeCertificates(\StringLength{w})}$ for all $i$, there exists a polynomial $\PolynomialTimeMachineOmega(n)$ such that $\Omega$'s running time is bounded by $2^{\PolynomialTimeMachineOmega(\StringLength{w})} - 2$, which is a function of~$\StringLength{w}$ alone---%
such a tailored running time is chosen with the only aim of simplifying the presentation below.

If we assume that $\PolynomialLengthHausPredD(n) = a \cdot n^b$, $\PolynomialSizeCertificates(n) = c \cdot n^d$, and $\PolynomialTimeMachineOmega(n) = e \cdot n^f$, we define the polynomial $\PolynomialBoundingEverything(n) = \max\set{a,c,e} \cdot n^{\max\set{b,d,f}}$.
By this, we have $\PolynomialBoundingEverything(n) \geq \PolynomialLengthHausPredD(n)$, $\PolynomialBoundingEverything(n) \geq \PolynomialSizeCertificates(n)$, and $\PolynomialBoundingEverything(n) \geq \PolynomialTimeMachineOmega(n)$.

The formula $\FormulaTMHausdorffWithP$ will encode the working of the machine $\Omega$ on an input string ``$w \hashsep z \dollarsep u_1 \dollarsep \cdots \dollarsep u_c$'', where $u_i$ are the certificates of \zcref{eq_condition_NEXP_Hausdorff_language} separated by the symbols  `$\dollarsep$'.
These certificates will suitably be ``quantified'', to match the characterization of \zcref{eq_condition_NEXP_Hausdorff_language}.
Before presenting the detailed construction, we give an overview, first of the encoding of the computation of $\Omega$, and then of the encoding of the input string on the machine tape.

We assume \Wlog that $\Omega$ is a Turing machine with a single semi\nbdash-infinite read/write input tape, whose extremal\nbdash-left tape cell is marked by the tape symbol~`$\rhd$' (see~\cite{Hopcroft1979,Arora2009}).
Inspired by \Sat{}'s \NPTimeh{}ness proof reported by \citet{GareyJ1979}, we use Boolean \emph{function} variables~$q$, $h$, and~$t$, encoding the (control) state, the tape head position, and the tape content, respectively, at any given step of $\Omega$'s computation.
For this reason, we need to be able to identify the steps of the computation, and the position of the symbols on the tape.
Since we assume that $\Omega$ halts within $2^{\PolynomialTimeMachineOmega(\StringLength{w})} - 2$ steps, the portion of interest of the tape is constituted by $2^{\PolynomialTimeMachineOmega(\StringLength{w})}$ cells, the tape head can be in $2^{\PolynomialTimeMachineOmega(\StringLength{w})}$ different locations, and the number of IDs traversed during a computation is $2^{\PolynomialTimeMachineOmega(\StringLength{w})}$. 
By this, the tape cell positions, needed to locate the symbols and the head on the tape, and the numerical identifiers to distinguish the steps within the computation sequence, can be represented in binary with $\PolynomialTimeMachineOmega(\StringLength{w})$ bits, if we refer to them by $0$\nbdash-based indices.
Representing which symbols are on tape, and which state the machine is in, can be done via indices as well.
These are independent from the input string, hence their number is a constant.
Let us assume the tape symbols to be $\Gamma_\Omega = \set{\alpha_0,\dots,\alpha_{s-1}}$ and the machine states to be $Q_\Omega = \set{q_0,\dots,q_{r-1}}$.
Hence, their indices can be represented with $\lceil \log s \rceil$ and $\lceil \log r \rceil$ bits, respectively.

The intended meaning of the Boolean functions $q$, $h$, and~$t$, will be as follows:
\begin{center}
\begin{tabular}{l p{7.5cm}}
  $q( \underbracket[.5pt]{\overbracket[.5pt]{\bullet,\bullet,\dots,\bullet,\bullet}^{i}}_{\PolynomialTimeMachineOmega(\StringLength{w})\text{-many}}; \underbracket[.5pt]{\overbracket[.5pt]{\circ,\circ,\dots,\circ,\circ}^{k}}_{\lceil \log r \rceil\text{-many}})$ & stating whether, at the $i$-th computation step, $\Omega$'s state is $q_k$, or not; \\[4ex]
  $h( \underbracket[.5pt]{\overbracket[.5pt]{\bullet,\bullet,\dots,\bullet,\bullet}^{i}}_{\PolynomialTimeMachineOmega(\StringLength{w})\text{-many}}; \underbracket[.5pt]{\overbracket[.5pt]{\circ,\circ,\dots,\circ,\circ}^{j}}_{\PolynomialTimeMachineOmega(\StringLength{w})\text{-many}})$ & stating whether, at the $i$-th computation step, $\Omega$'s tape head is at the $j$-th position, or not; and \\[4ex]
  $t( \underbracket[.5pt]{\overbracket[.5pt]{\bullet,\bullet,\dots,\bullet,\bullet}^{i}}_{\PolynomialTimeMachineOmega(\StringLength{w})\text{-many}}; \underbracket[.5pt]{\overbracket[.5pt]{\circ,\circ,\dots,\circ,\circ}^{j}}_{\PolynomialTimeMachineOmega(\StringLength{w})\text{-many}};
  \underbracket[.5pt]{\overbracket[.5pt]{\bullet,\bullet,\dots,\bullet,\bullet}^{\ell}}_{\lceil \log s \rceil \text{-many}})$ & stating whether, at the $i$-th computation step, at the $j$-th position of $\Omega$'s tape the symbol is $\alpha_\ell$,~or~not. \\
\end{tabular}
\end{center}

For example, since we use $0$\nbdash-based indices, if $\Omega$ carries out a computation in which, after two computation steps (hence, in the third ID, whose index is $2$), the fourth tape symbol (hence, position indexed $3$) is the blank $\blanksymbol$ (let us say that the symbol index for $\blanksymbol$ is $0$), we have $t(\boxed{0,\dots,0,1,0} ; \boxed{0,\dots,0,1,1} ; \boxed{0,\dots,0,0,0}) = \valtrue$.

Let us now focus on how we encode in $\FormulaTMHausdorffWithP$ that the input string ``$w \hashsep z \dollarsep u_1 \dollarsep \cdots \dollarsep u_c$'' is on the machine tape at the beginning of the computation.
Since $w$ is given, there will be a part of the formula stating that ``$w \hashsep$'' appears at the left end of the tape.
As anticipated, the formula $\FormulaTMHausdorffWithP$ will rely on a \emph{free} function variable $\mi{ind}^{\PolynomialLengthHausPredD(\StringLength{w})}$ allowing us to ``import'' $z$ from outside $\FormulaTMHausdorffWithP$.
This function variable is a bit\nbdash-indexing function for the value $z$ represented over $2^{\PolynomialLengthHausPredD(\StringLength{w})}$ bits.
The formula $\FormulaTMHausdorffWithP$ will hence have a part devoted to ``copy'' $z$ from $\mi{ind}^{\PolynomialLengthHausPredD(\StringLength{w})}$ to the tape, right after ``$w \hashsep$''.
Notice that $z$ is encoded in $\mi{ind}^{\PolynomialLengthHausPredD(\StringLength{w})}$ via a fixed\nbdash-length binary representation, whereas it has to be encoded on the tape in canonical binary form, as expected by $\Omega$, 
hence parts of $\FormulaTMHausdorffWithP$ will be devoted to this translation.
Additionally, parts of $\FormulaTMHausdorffWithP$ will ``write'' the certificates $u_i$ on the input tape right after ``$w \hashsep z$'', one certificate after the other, separated by the symbol~`$\dollarsep$'.
We will employ Boolean function variables $\mi{cert}_1,\dots,\linebreak[0]\mi{cert}_c$ that, similarly to $\mi{ind}^{\PolynomialLengthHausPredD(\StringLength{w})}$, are bit\nbdash-indexing functions encoding the certificates $u_1,\dots,u_c$, respectively, which are binary strings of length $2^{\PolynomialSizeCertificates(\StringLength{w})}$;
the arity of the functions $\mi{cert}_i$ will hence be $\PolynomialSizeCertificates(\StringLength{w})$.
Following the certificate\nbdash-based characterization of~\zcref{eq_condition_NEXP_Hausdorff_language}, in $\FormulaTMHausdorffWithP$ the function variables $\mi{cert}_i$ will be quantified existentially or universally depending on whether $i$ is odd or even, respectively.
In this way, we obtain the needed ``quantification'' of the certificates of the input string.

\Proofsep

We now detail the formula $\FormulaTMHausdorffWithP$ encoding the working of $\Omega$ over an input string ``$w \hashsep z \dollarsep u_1 \dollarsep \cdots \dollarsep u_c$''.
We start by focusing on an \emph{odd} number $c$ of certificates, and provide a construction yielding a \CompactSigmaKOddFormula{c}.
For even $c$, a naive adaptation would yield a \CompactSigmaKOddFormula{(c+1)};
a more tailored construction for even $c$ producing a \CompactSigmaKEvenFormula{c} will be given later.
For odd $c$, at a high level $\FormulaTMHausdorffWithP$ is a \CompactSigmaKOddFormula{c} of the form:
\begin{multline*}
  \FormulaTMHausdorffWithP(\mi{ind}^{\PolynomialLengthHausPredD(\StringLength{w})}) \equiv {} \\ 
  (\exists \dots,\mi{cert}_1)
  (\forall \mi{cert}_2) \cdots
  (\forall \mi{cert}_{c-1})
  (\exists \mi{cert}_c, \dots)
  (\forall \mi{prop\mhyphen{}vars}) \PrefixMatrixSeparator
  \gamma(\mi{ind}^{\PolynomialLengthHausPredD(\StringLength{w})},\dots),
\end{multline*}
where `$(\exists \dots, \mi{cert}_1)$' means that, together with $\mi{cert}_1$, more function variables appear there;
similarly for `$(\exists \mi{cert}_c, \dots)$'.
By `$(\forall \mi{prop\mhyphen{}vars})$' we mean that propositional variables are quantified there.

We use ordered sets of propositional variables $\VarSet{i} = \set{i_{\PolynomialTimeMachineOmega(\StringLength{w})},\dots,i_1}$, $\VarSet{i}' = \set{i'_{\PolynomialTimeMachineOmega(\StringLength{w})},\linebreak[0]\dots,i'_1}$, $\VarSet{j} = \set{j_{\PolynomialTimeMachineOmega(\StringLength{w})},\linebreak[0]\dots,j_1}$, $\VarSet{j}' = \set{j'_{\PolynomialTimeMachineOmega(\StringLength{w})},\linebreak[0]\dots,j'_1}$, $\VarSet{k} = \set{k_{\lceil \log r \rceil},\linebreak[0]\dots,k_1}$, $\VarSet{k}' = \set{k'_{\lceil \log r \rceil},\linebreak[0]\dots,k'_1}$, $\VarSet{\ell} = \set{\ell_{\lceil \log s \rceil},\linebreak[0]\dots,\ell_1}$, and $\VarSet{\ell}' = \set{\ell'_{\lceil \log s \rceil}\linebreak[0],\dots,\ell'_1}$, to ``store'' the fixed\nbdash-length binary encoding of the indices of the steps ($\VarSet{i}$ and $\VarSet{i}'$), the positions on tape ($\VarSet{j}$ and $\VarSet{j}'$), the (indices of the) machine states ($\VarSet{k}$ and $\VarSet{k}'$), and the (indices of the) tape symbols ($\VarSet{\ell}$ and $\VarSet{\ell}'$), respectively.
We additionally have ordered sets of propositional variables $\VarSet{d} = \set{d_{\PolynomialLengthHausPredD(\StringLength{w})},\dots,d_1}$ and $\VarSet{e} = \set{e_{\PolynomialLengthHausPredD(\StringLength{w})},\dots,e_1}$, which, although not strictly essential, are sometimes employed in formulas to obtain a tidier presentation.
We also use the successor $\mi{succ}$, majority $\mi{less}$, inequality $\mi{neq}$, and addition $\mi{add}$, functions of \zcref{theo_SigmaFormula_arithmetic}.
In the construction below, these functions are defined for binary numbers of $\PolynomialBoundingEverything(\StringLength{w})$ bits.
\emph{However, in the formula we also use these functions over smaller sets of propositional variables;
in these cases, we mean that a padding of $\valfalse$ values is employed for the most significant bits.}
We also have function variables $\mi{in\mhyphen{}range}$ and $\mi{max\mhyphen{}bit}$ aiming at individuating in the fixed\nbdash-length binary representation of $z$ the position of the most significant bit of its canonical binary form.

We now provide the details of the construction.
We build a \CompactSigmaKOddFormula{c} of the form below, where those on first line are second\nbdash-order quantifiers, whereas that on second line is a first\nbdash-order quantifier:
\begin{equation*}
\begin{aligned}[b]
&\FormulaTMHausdorffWithP(\mi{ind}^{\PolynomialLengthHausPredD(\StringLength{w})}) \equiv {} \\
  &\qquad(\exists \mi{succ}, \mi{less}, \mi{neq}, \mi{add}, \mi{in\mhyphen{}range}, \mi{max\mhyphen{}bit}, \mi{cert}_1)
  (\forall \mi{cert}_2) \cdots (\forall \mi{cert}_{c-1}) (\exists \mi{cert}_c, q, h, t) \\
  &\qquad(\forall \VarSet{a},\VarSet{a}',\VarSet{a}'',\VarSet{b},\VarSet{b}',\VarSet{i},\VarSet{i}', \VarSet{j},\VarSet{j}', \VarSet{k},\VarSet{k}', \VarSet{\ell},\VarSet{\ell}',\VarSet{d},\VarSet{e}) \\ &\qquad\qquad\mu^{\PolynomialBoundingEverything(\StringLength{w})}(\mi{succ},\mi{less},\mi{neq},\mi{add},\VarSet{a},\VarSet{a}',\VarSet{a}'',\VarSet{b},\VarSet{b}') \land {} \\
  &\qquad\qquad\mi{WITHIN}(\mi{ind}^{\PolynomialBoundingEverything(\StringLength{w})},\mi{succ},\mi{in\mhyphen{}range,\VarSet{d},\VarSet{e}}) \land \mi{MAX\mhyphen{}BIT}(\mi{succ},\mi{in\mhyphen{}range},\mi{max\mhyphen{}bit},\VarSet{d},\VarSet{e}) \land {} \\
  &\qquad\qquad\psi(\mi{succ},\mi{less},\mi{neq},\mi{add}, \mi{ind}^{\PolynomialLengthHausPredD(\StringLength{w})}, \mi{cert}_1, \dots, \mi{cert}_c, q,h,t, \VarSet{i},\VarSet{i}', \VarSet{j},\VarSet{j}', \VarSet{k},\VarSet{k}', \VarSet{\ell},\VarSet{\ell}',\VarSet{d},\VarSet{e})
\end{aligned}
\end{equation*}
such that an interpretation $\Interpr{I}_z$ of $\mi{ind}^{\PolynomialLengthHausPredD(\StringLength{w})}$ encoding the integer $z$ satisfies $\FormulaTMHausdorffWithP(\mi{ind}^{\PolynomialLengthHausPredD(\StringLength{w})})$ iff $\Language{D}(w,z) = 1$.
This is achieved by designing $\psi$ to be satisfied by $\Interpr{I}_z$ iff $\Omega$ accepts its input string, which is ``$w \hashsep z$'' followed by $c$ certificates separated by `$\dollarsep$' symbols.
Since $\Omega$ decides the predicate $R$ of \zcref{eq_condition_NEXP_Hausdorff_language} \emph{and} the certificates are existentially or universally quantified in $\FormulaTMHausdorffWithP(\mi{ind}^{\PolynomialLengthHausPredD(\StringLength{w})})$ according to the definition of the main level of the \WEHText which $\Language{D}$ belongs to, we obtain $\Interpr{I}_z \models \FormulaTMHausdorffWithP(\mi{ind}^{\PolynomialLengthHausPredD(\StringLength{w})})$ iff $\Language{D}(w,z) = 1$.

Below, we will need to show that $\FormulaTMHausdorffWithP(\mi{ind}^{\PolynomialLengthHausPredD(\StringLength{w})})$ can be obtained in polynomial time \Wrt~$\StringLength{w}$.
The formulas introduced will be defined over functions whose arity is polynomial \Wrt $\StringLength{w}$.
Hence, we will be able to prove that these formulas can be obtained in polynomial time \Wrt $\StringLength{w}$ by simply showing that the number of their literals is polynomial \Wrt $\StringLength{w}$, or even a constant.
To make the meanings of the formulas easier for the reader to understand, we will use implicative forms.
Remember that, e.g., $(\alpha \land \beta) \rightarrow (\gamma \lor \delta) \equiv (\lnot \alpha \lor \lnot \beta \lor \gamma \lor \delta)$, hence such an implication, if $\alpha, \beta, \gamma$, and $\delta$, are literals, is actually a clause of four literals.

Before describing the formula $\psi$, we focus on the subformulas $\mi{WITHIN}$ and $\mi{MAX\mhyphen{}BIT}$, whose aim is identifying the position of the most significant bit of the canonical binary form of $z$.
The subformula $\mi{MAX\mhyphen{}BIT}$ encodes into the function variable $\mi{max\mhyphen{}bit}$ the position of the most\nbdash-significant bit of the canonical form of $z$.
Remember that the indices used as arguments of $\mi{ind}^{\PolynomialLengthHausPredD(\StringLength{w})}$ for the positions of the bits of $z$ are $0$\nbdash-based, and lower indices are for the least significant bits.
For a set of ordered Boolean propositional variables $\VarSet{d} = \set{d_{\PolynomialLengthHausPredD(\StringLength{w})},\dots,d_{1}}$ encoding in binary the integer value $d$, the function $\mi{max\mhyphen{}bit}$ is designed as follows:
\begin{center}
\begin{tabular}{l p{7.5cm}}
  $\mi{max\mhyphen{}bit}(\underbracket[.5pt]{\overbracket[.5pt]{\bullet,\bullet,\dots,\bullet,\bullet}^{d}}_{\PolynomialLengthHausPredD(\StringLength{w})\text{-many}})$ & stating whether the $d$-th bit, starting from the least significant ones, of the fixed\nbdash-length binary representation of $z$ in $\mi{ind}^{\PolynomialLengthHausPredD(\StringLength{w})}$ is the most significant bit `$1$', or not; if $z = 0 $, then $\mi{max\mhyphen{}bit}(\BooleanEncoding{0}) = \valtrue$. \\
\end{tabular}
\end{center}

The definition of the function $\mi{max\mhyphen{}bit}$ in the subformula $\mi{MAX\mhyphen{}BIT}$ is obtained via the helping function $\mi{in\mhyphen{}range}$, which is characterized in turn in the subformula $\mi{WITHIN}$, and whose meaning is:%
\footnote{The function $\mi{max\mhyphen{}bit}$ can also be defined in a simpler way by using \emph{existentially} quantified propositional variables, avoiding the need for the function $\mi{in\mhyphen{}range}$. However, since the other propositional variables in $\FormulaTMHausdorffWithP(\mi{ind}^{\PolynomialLengthHausPredD(\StringLength{w})})$ are \emph{universally} quantified, they cannot all fall under the same first\nbdash-order quantifier. To address this, the definition of $\mi{max\mhyphen{}bit}$ could use $0$\nbdash-arity function variables, which can be quantified existentially under a second\nbdash-order quantifier. With the only aim of having in the presentation a clear distinction between propositional variables and function variables, we prefer to adopt a slightly more involved definition of $\mi{max\mhyphen{}bit}$.}
\begin{center}
\begin{tabular}{l p{7.5cm}}
  $\mi{in\mhyphen{}range}(\underbracket[.5pt]{\overbracket[.5pt]{\bullet,\bullet,\dots,\bullet,\bullet}^{d}}_{\PolynomialLengthHausPredD(\StringLength{w})\text{-many}})$ & stating whether the $d$-th bit, starting from the least significant ones, of the fixed\nbdash-length binary representation of $z$ in $\mi{ind}^{\PolynomialLengthHausPredD(\StringLength{w})}$, has to be included in the canonical binary form of $z$. \\
\end{tabular}
\end{center}

Let us focus first on the subformula $\mi{WITHIN}$ (see below);
notice that the propositional variables appearing in $\mi{WITHIN}$ will universally be quantified in the overall formula.
The first rule states that the first bit (position $0$) is always part of the canonical form.
The second and third rules state that if at position $d$ there is a bit `$1$' in $z$, then position $d$ and lower positions have to be included in the canonical form.
The fourth and fifth rules state that if a position $d$ falls within the canonical form, then there must be a bit `$1$' of $z$ at position $d$ or higher.
\begin{align*}
  \mi{WITHIN} \equiv {} & \mi{in\mhyphen{}range}(\BooleanEncoding{0}) \land {} \\
  & (\mi{ind}^{\PolynomialLengthHausPredD(\StringLength{w})} (\VarSet{d}) \rightarrow \mi{in\mhyphen{}range}(\VarSet{d})) \land {} \\
  & (\mi{in\mhyphen{}range}(\VarSet{d}) \land \mi{succ}(\VarSet{e},\VarSet{d}) \rightarrow \mi{in\mhyphen{}range}(\VarSet{e})) \land {} \\
  & (\mi{in\mhyphen{}range} (\BooleanEncoding{2^{\PolynomialLengthHausPredD(\StringLength{w})} - 1}) \rightarrow \mi{ind}^{\PolynomialLengthHausPredD(\StringLength{w})} (\BooleanEncoding{2^{\PolynomialLengthHausPredD(\StringLength{w})} - 1})) \land {} \\
  & (\mi{in\mhyphen{}range}(\VarSet{d}) \land \mi{succ}(\VarSet{d},\VarSet{e}) \rightarrow \mi{ind}^{\PolynomialLengthHausPredD(\StringLength{w})} (\VarSet{d}) \lor \mi{in\mhyphen{}range}(\VarSet{e})).
\end{align*}
The formula $\mi{WITHIN}$ is a \CNF of five clauses with four literals at most.

Let us now focus on the subformula $\mi{MAX\mhyphen{}BIT}$ (see below);
again, the propositional variables appearing in $\mi{MAX\mhyphen{}BIT}$ will universally be quantified in the overall formula.
The first rule states that if the last bit of the fixed\nbdash-length representation belongs to the canonical form, than that position is also the maximum one of the canonical form.
The second rule states that if a position $d$ is in the canonical form and instead the position $e = d + 1$ is not, then the position $d$ is the maximum one of the canonical form.
The third rule state that if a position $d$ is the maximum one of the canonical form, then all the other positions are not the maximum ones.
\begin{align*}
  \mi{MAX\mhyphen{}BIT} \equiv {} & (\mi{in\mhyphen{}range}(\BooleanEncoding{2^{\PolynomialLengthHausPredD(\StringLength{w})} - 1}) \rightarrow \mi{max\mhyphen{}bit}(\BooleanEncoding{2^{\PolynomialLengthHausPredD(\StringLength{w})} - 1})) \land {} \\
  & (\mi{in\mhyphen{}range}(\VarSet{d}) \land \mi{succ}(\VarSet{d},\VarSet{e}) \land \lnot \mi{in\mhyphen{}range}(\VarSet{e}) \rightarrow \mi{max\mhyphen{}bit}(\VarSet{d})) \land {} \\
  & (\mi{max\mhyphen{}bit}(\VarSet{d}) \land \mi{neq}(\VarSet{d},\VarSet{e}) \rightarrow \lnot \mi{max\mhyphen{}bit}(\VarSet{e})).
\end{align*}
The formula $\mi{MAX\mhyphen{}BIT}$ is a \CNF of three clauses with four literals at most.

We now provide the details of the formula $\psi$ encoding the computation of $\Omega$.
Following Cook's reduction reported in~\cite{GareyJ1979,Hopcroft1979}, $\psi$ will be characterized by four subformulas:
one imposing the consistency of the functions~$q$, $h$, and $t$;
one imposing that~$q$, $h$, and $t$, correctly encode the initial configuration of $\Omega$ over the input string;
one imposing that~$q$, $h$, and $t$, correctly encode a sequence of computation steps by $\Omega$; and
one checking whether the last configuration encoded in~$q$, $h$, and $t$, refers to an accepting configuration for~$\Omega$.
Notice that, in the subformulas below, all the function variables, but $\mi{ind}^{\PolynomialLengthHausPredD(\StringLength{w})}$ which will be free, will be quantified existentially, while all the propositional variables will be quantified universally, in the overall formula.

\begin{itemize}[label=--,left=0pt]
  \item \textbf{Consistency (\bm{$C$})}:
    With this subformula, we impose that~$q$, $h$, and $t$, encode that, at each computation step, the machine is in no more than one state, the tape head is in no more than one position, and each tape cell contains no more than one symbol, respectively.
    Other parts of the formula (see below) will imply that~$q$, $h$, and~$t$, also encode that, at each computation step, the machine is in at least one state, the tape head is in at least one position, and each tape cell contains at least one symbol, respectively.
    These conditions together will impose that~$q$, $h$, and~$t$, encode that, at each computation step, the machine is in exactly one state, the tape head is in exactly one position, and each tape cell contains exactly one symbol, respectively.
    \begin{align*}
    C \equiv {} 
        &\Bigl( q(\VarSet{i};\VarSet{k}) \land \mi{neq}(\VarSet{k};\VarSet{k}') \rightarrow \lnot q(\VarSet{i};\VarSet{k}') \Bigr) \land {} \\
        &\Bigl( h(\VarSet{i};\VarSet{j}) \land \mi{neq}(\VarSet{j};\VarSet{j}') \rightarrow \lnot h(\VarSet{i};\VarSet{j}')\Bigr) \land {} \\
        &\Bigl( t(\VarSet{i};\VarSet{j};\VarSet{\ell}) \land \mi{neq}(\VarSet{\ell};\VarSet{\ell}') \rightarrow \lnot t(\VarSet{i};\VarSet{j};\VarSet{\ell}') \Bigr).
    \end{align*}   
    The formula ${C}$ is a \CNF of three clauses containing three literals.
    
  \item \textbf{Start correct (\bm{$S$})}:
    This subformula constrains $q$, $h$, and $t$, to encode that, at step $0$, the machine is in the initial state $q_0$, the tape head is at position $0$, and the tape content is ``$w \hashsep z \dollarsep u_1 \dollarsep \cdots \dollarsep u_c$'' (preceded by the tape marker `$\rhd$', followed by blanks, and with $z$ in \emph{canonical} binary form).
    We define $S$ via additional subformulas.
    
    The subformula $W$ encodes the string $w$ in the function $t$;
    below, $w[\tau]$ denotes the $\tau$\nbdash-th symbol of the string~$w$%
    ---remember that `$\rhd$' is at position $0$ of the tape, and this is encoded in another part of the formula:
    \[
         W \equiv \bigwedge_{1 \leq \tau \leq n} t(\BooleanEncoding{0};\BooleanEncoding{\tau};\BooleanEncoding{w[\tau]}).
    \]
    The subformula $W$ is a CNF containing a number of clauses which is polynomial \Wrt $\StringLength{w}$, and each of these clauses contains a single literal.
    
    The subformula $Z$ specifies that the string ``$\hashsep  z$'' appears after $w$ on the input tape.
    This is encoded into $t$ by copying $z$ from the unfolding of the function $\mi{ind}^{\PolynomialLengthHausPredD(\StringLength{w})}$.
    Remember that the fixed\nbdash-length representation of $z$ encoded in $\mi{ind}^{\PolynomialLengthHausPredD(\StringLength{w})}$ has to be translated into canonical binary form to be encoded into $t$.
    The tape locations of $z$'s canonical form symbols are obtained as follows.
    Assume that the variables $\VarSet{d}$ encode the index $d$ of the most significant bit of the canonical form of $z$.
    We have that, for a value $e$ such that $0 \leq e \leq d$, the $e$\nbdash-th symbol of the canonical form of $z$, encoded in the Boolean value $\mi{ind}^{\PolynomialLengthHausPredD(\StringLength{w})}(\VarSet{e})$, is at the tape position $j$, obtained by adding $e$ to $z$'s first symbol position, which is $\StringLength{w} + 2$.
    By this, the subformula $Z$ is:
    \begin{align*}
          Z \equiv {} & t(\BooleanEncoding{0};\BooleanEncoding{\StringLength{w}{+}1};\BooleanEncoding{\hashsep}) \land {} \\ 
          & \qquad \Bigl( \mi{max\mhyphen{}bit}(\VarSet{d}) \land \lnot \mi{less}(\VarSet{d},\VarSet{e}) \land \mi{ind}^{\PolynomialLengthHausPredD(\StringLength{w})}(\VarSet{e}) \land \mi{add}(\BooleanEncoding{\StringLength{w}{+}2};\VarSet{e};\VarSet{j}) \rightarrow t(\BooleanEncoding{0};\VarSet{j},\BooleanEncoding{1}) \Bigr) \land {} \\ 
          &\qquad \Bigl( \mi{max\mhyphen{}bit}(\VarSet{d}) \land \lnot \mi{less}(\VarSet{d},\VarSet{e}) \land \lnot \mi{ind}^{\PolynomialLengthHausPredD(\StringLength{w})}(\VarSet{e}) \land \mi{add}(\BooleanEncoding{\StringLength{w}{+}2};\VarSet{e};\VarSet{j}) \rightarrow t(\BooleanEncoding{0};\VarSet{j},\BooleanEncoding{0}) \Bigr).
    \end{align*}
    The formula ${Z}$ is a \CNF of three clauses containing at most five literals.
    
    The subformula $U_i$ specifies that, on the input tape, the string ``$ \dollarsep u_i$'' appears after the string ``$\dollarsep u_{i-1}$'', and the string ``$\dollarsep u_1$'' appears after ``$w \hashsep  z$''.  
    This is encoded into $t$ by copying $u_i$ from the unfolding of the function $\mi{cert}_i$---remember that the certificates are binary strings.
    The tape locations of $u_i$'s symbols are obtained as follows.
    Let $d$ be the \emph{index} of the most significant bit of the canonical form of $z$, which means that the canonical form of $z$ comprises $d + 1$ bits.
    Since all certificates are assumed to have length $2^{\PolynomialSizeCertificates(\StringLength{w})}$, the `$\dollarsep$' symbol preceding $u_i$ is at position $\StringLength{w} + 2 + d + 1 + (i - 1) \cdot (2^{\PolynomialSizeCertificates(\StringLength{w})} + 1)$ on tape.
    The following tape position clearly contains $u_i$'s first symbol.
    The $e$\nbdash-th symbol of $u_i$, encoded in the Boolean value $\mi{cert}_{i}(\VarSet{e})$, where $e$ is such that $0 \leq e \leq 2^{\PolynomialSizeCertificates(\StringLength{w})} - 1$, is at the tape position $j$, obtained by adding $e$ to $u_i$'s first symbol position.
    Hence the subformulas $U_i$ are as follows:    
    \begin{align*}
        U_i \equiv {} 
        &\Bigl( \mi{max\mhyphen{}bit}(\VarSet{d}) \land \mi{add}(\BooleanEncoding{\StringLength{w}{+}3{+}(i{-}1){\cdot}(2^{\PolynomialSizeCertificates(\StringLength{w})}{+}1)};\VarSet{d};\VarSet{j}) \rightarrow t(\BooleanEncoding{0};\VarSet{j};\BooleanEncoding{\dollarsep}) \Bigr) \land {} \\
        &\begin{multlined}[t]
        \Bigl( \mi{max\mhyphen{}bit}(\VarSet{d}) \land \mi{add}(\BooleanEncoding{\StringLength{w}{+}4{+}(i{-}1){\cdot}(2^{\PolynomialSizeCertificates(\StringLength{w})}{+}1)};\VarSet{d};\VarSet{j}') \land \mi{cert}_{i}(\VarSet{e}) \land \mi{add}(\VarSet{j}';\VarSet{e},\VarSet{j}) \rightarrow {} \quad \\
        t(\BooleanEncoding{0};\VarSet{j};\BooleanEncoding{1})
        \Bigr) \land {} 
        \end{multlined}\\
        &\begin{multlined}[t]
        \Bigl(
        \mi{max\mhyphen{}bit}(\VarSet{d}) \land \mi{add}(\BooleanEncoding{\StringLength{w}{+}4{+}(i{-}1){\cdot}(2^{\PolynomialSizeCertificates(\StringLength{w})}{+}1)};\VarSet{d};\VarSet{j}') \land \lnot \mi{cert}_{i}(\VarSet{e}) \land \mi{add}(\VarSet{j}';\VarSet{e},\VarSet{j}) \rightarrow {} \\
        t(\BooleanEncoding{0};\VarSet{j};\BooleanEncoding{0}) \Bigr).
        \end{multlined}
    \end{align*}
    The formula ${U_i}$ is a \CNF of three clauses containing at most five literals.
    
    To conclude, the subformula $B$ specifies that the input tape, beyond the input string, is blank:
    \[
        B \equiv \Bigl(
        \mi{max\mhyphen{}bit}(\VarSet{d}) \land \mi{add}(\BooleanEncoding{\StringLength{w}{+}3{+}c{\cdot}(2^{\PolynomialSizeCertificates(\StringLength{w})}{+}1)};\VarSet{d};\VarSet{j}') \land \mi{less}(\VarSet{j}';\VarSet{j}) \rightarrow t(\BooleanEncoding{0};\VarSet{j};\BooleanEncoding{\blanksymbol}) \Bigr).
    \]
    The formula ${B}$ is a \CNF of one clause containing three literals.
    
    Hence, the subformula $S$ is:
    \[
        S \equiv q(\BooleanEncoding{0};\BooleanEncoding{q_0}) \land h(\BooleanEncoding{0};\BooleanEncoding{0}) \land t(\BooleanEncoding{0};\BooleanEncoding{0};\BooleanEncoding{\rhd}) \land W \land Z \land \bigwedge_{i = 1}^{c} U_i \land B.
    \]
    This subformula also forces $q$, $h$, and~$t$, to encode that, at step $0$, the machine is in at least one state, the head is in at least one position, and each tape cell contains at least one symbol (see the comments above).
    Since $c$ is fixed, $S$ is a CNF with polynomially\nbdash-many clauses \Wrt $\StringLength{w}$, containing each a constant number of literals.
    
    This subformula can hence be built in polynomial time in $\StringLength{w}$ if the values $\StringLength{w} + 3 + (i - 1) \cdot (2^{\PolynomialSizeCertificates(\StringLength{w})} + 1)$, with $1 \leq i \leq c + 1$, and the ones obtained by adding one to the latter, can be computed in polynomial time.
    Computing the binary representation of $\StringLength{w}$ requires linear time in $\StringLength{w}$ (see, e.g., \cite[Chapter~16, ``Incrementing a binary counter'']{CormenLRS2022}). 
    The value $2^{\PolynomialSizeCertificates(\StringLength{w})}$ can be obtained in polynomial time \Wrt $\StringLength{w}$ (see \zcref{sec_maths_complexity}, Iterated exponentials of polynomials).
    Once the value $2^{\PolynomialSizeCertificates(\StringLength{w})}$ has been computed, the rest are simple arithmetic calculations, which can be carried out in polynomial time (see \zcref{sec_maths_complexity}, Arithmetic).
    
  \item \textbf{Next move correct (\bm{$N$})}:
    With this subformula, we impose that~$q$, $h$, and~$t$, encode a legal sequence of computation steps for $\Omega$ over the input string.
    This subformula is in turn constituted by three pieces:
    \begin{itemize}[nosep,label=--]
      \item the ``inertia'' piece ($N^I$) dealing with tape cells that, far from the tape head, do not change their content from one step to the next;
      \item the ``head'' piece ($N^H_\Omega$) dealing with the tape cell that, beneath the tape head, changes its content according to the transition function of $\Omega$; and
      \item the ``padding'' piece ($N^P_\Omega$) propagating a final configuration of $\Omega$, either accepting or not, exactly as it is toward the computation step with maximum index.
    \end{itemize}
    
    The inertia piece can easily be encoded in the following subformula:
     \[{N}^{I} \equiv 
        \Bigl(t(\VarSet{i},\VarSet{j};\VarSet{\ell}) \land \lnot h(\VarSet{i};\VarSet{j}) \land \mi{succ}(\VarSet{i};\VarSet{i}') \Bigr) \rightarrow t(\VarSet{i}';\VarSet{j},\VarSet{\ell}).
    \]
    Notice that, if $\mi{succ}(\VarSet{i};\VarSet{i}')$ is $\valtrue$, then $i < 2^{\PolynomialTimeMachineOmega(n)}-1$, i.e., $i$ is \emph{not} the last step;
    $N^I$ is a clause of four literals.
    
    The head piece encodes all the possible transitions of $\Omega$ (remember that $\Omega$ is \emph{deterministic}):
    \begin{align*}{N}^{H}_{\Omega} \equiv {}
        &\begin{multlined}[t]
        {\bigwedge\limits_{\substack{q,p \in Q_{\Omega},\; \alpha, \beta \in \Gamma_{\Omega}:\\ \delta(q,\alpha) = (p,\beta,\rightarrow)}}}
        \biggl(
        \Bigl(
        q(\VarSet{i};\BooleanEncoding{q}) \land
        h(\VarSet{i};\VarSet{j}) \land
        t(\VarSet{i};\VarSet{j};\BooleanEncoding{\alpha}) \land
        \mi{succ}(\VarSet{i};\VarSet{i}') \land \mi{succ}(\VarSet{j};\VarSet{j}') \Bigr) \rightarrow {} \qquad \\
        \Bigl(
            q(\VarSet{i}';\BooleanEncoding{p}) \land
            h(\VarSet{i}';\VarSet{j}') \land
            t(\VarSet{i}';\VarSet{j};\BooleanEncoding{\beta})
        \Bigr)\biggr) \land {}
        \end{multlined} \displaybreak[0] \\
        &\begin{multlined}[t]
        {\bigwedge\limits_{\substack{q,p \in Q_{\Omega},\; \alpha,\beta \in \Gamma_{\Omega}:\\ \delta(q,\alpha) = (p,\beta,\leftarrow)}}}
        \biggl(
        \Bigl(
        q(\VarSet{i};\BooleanEncoding{q}) \land
        h(\VarSet{i};\VarSet{j}) \land
        t(\VarSet{i};\VarSet{j};\BooleanEncoding{\alpha}) \land\mi{succ}(\VarSet{i};\VarSet{i}') \land \mi{succ}(\VarSet{j}';\VarSet{j}) \Bigr) \rightarrow {} \qquad \\
        \Bigl(
            q(\VarSet{i}';\BooleanEncoding{p}) \land
            h(\VarSet{i}';\VarSet{j}') \land
            t(\VarSet{i}';\VarSet{j};\BooleanEncoding{\beta})
        \Bigr)\biggr).
        \end{multlined}
    \end{align*}
    Notice that, if $\mi{succ}(\VarSet{j};\VarSet{j}')$ (resp., $\mi{succ}(\VarSet{j}';\VarSet{j})$) is $\valtrue$, then $j < 2^{\PolynomialTimeMachineOmega(n)}-1$ (resp., $0 < j$), i.e., $j$ is \emph{not} the right\nbdash-most (resp., left\nbdash-most) tape position, hence the head can move to the right (resp., to the left).
    
    The padding piece can be defined as follows, where we consider the case when $\Omega$ halts in the accepting state $q_F$ (first line) and when $\Omega$ halts in a non\nbdash-accepting state (second line):
    \begin{align*}{N}^{P}_{\Omega} \equiv {}
        &\Bigl(
        q(\VarSet{i};\BooleanEncoding{q_F}) \land
        h(\VarSet{i};\VarSet{j}) \land
        t(\VarSet{i};\VarSet{j};\VarSet{\ell}) \land
        \mi{succ}(\VarSet{i};\VarSet{i}') \Bigr) \rightarrow
        \Bigl(
            q(\VarSet{i}';\BooleanEncoding{q_F}) \land
            h(\VarSet{i}';\VarSet{j}) \land
            t(\VarSet{i}';\VarSet{j},\VarSet{\ell})
        \Bigr) \land {} \\
        &{\bigwedge\limits_{\substack{q \in Q_{\Omega} \setminus \set{q_F},\; \alpha \in \Gamma_{\Omega}:\\ \delta(q,\alpha) \text{ is not defined}}}}
        \Bigl(
        q(\VarSet{i};\BooleanEncoding{q}) \land
        h(\VarSet{i};\VarSet{j}) \land
        t(\VarSet{i};\VarSet{j};\BooleanEncoding{\alpha}) \land
        \mi{succ}(\VarSet{i};\VarSet{i}') \Bigr) \rightarrow
        \Bigl(
            q(\VarSet{i}';\BooleanEncoding{q}) \land
            h(\VarSet{i}';\VarSet{j}) \land
            t(\VarSet{i}';\VarSet{j},\BooleanEncoding{\alpha})
        \Bigr)
    \end{align*}
    
    The formula $N^I$ is always the same (irrespective of $\Omega)$.
    Whereas, $N^{H}_{\Omega}$ and $N^{P}_{\Omega}$ depend on the transition function of~$\Omega$, which are however \emph{fixed} once $\Omega$ is fixed.
    This means that $N^I$, $N^{H}_{\Omega}$, and $N^{P}_{\Omega}$, can be translated into \CNF formulas having a constant number of clauses containing each a constant number of literals.
    
    To conclude, we have that ${N}$ is the conjunction of all the formulas above:
    \[{N} \equiv
        {N}^{I} \land {N}^{H}_{\Omega} \land N^{P}_{\Omega}.
    \]
    
    Since all formulas constituting $N$ can be expressed in \CNF, also $N$ is a \CNF formula made of a constant number of clauses containing each a constant number of literals. 
    Furthermore, via its components, $N$ also imposes that~$q$, $h$, and~$t$, are such that, for every step $i > 0$ and every position $j$, assign $\valtrue$ for at least one machine state, one tape position, and one tape symbol, respectively (see the comments above).
    
  \item \textbf{Finish correct (\bm{$E$})}:
    The last subformula checks whether the instantiations of the functions variables encode an accepting computation of $\Omega$ over its input.
    Thanks to $N^{P}_{\Omega}$ (see above), we simply need to check that, at the step with maximum index, $q$ encodes that $\Omega$ is in its accepting state $q_F$.
    This subformula is clearly in \CNF.
    \[{E} \equiv
        q(\BooleanEncoding{2^{\PolynomialTimeMachineOmega(n)}-1};\BooleanEncoding{q_F}).
    \]
\end{itemize}

\noindent
To conclude, the \CompactSigmaKOddFormula{c} that we build is:
\begin{equation}
\label{eq_high_level_construction_exists_lastSOQ}
\begin{aligned}[b]
&\FormulaTMHausdorffWithP(\mi{ind}^{\PolynomialLengthHausPredD(\StringLength{w})}) \equiv {} \\
  &\qquad(\exists \mi{succ}, \mi{less}, \mi{neq}, \mi{add}, \mi{in\mhyphen{}range}, \mi{max\mhyphen{}bit}, \mi{cert}_1)
  (\forall \mi{cert}_2) \cdots (\forall \mi{cert}_{c-1}) (\exists \mi{cert}_c, q, h, t) \\
  &\qquad(\forall \VarSet{a},\VarSet{a}',\VarSet{a}'',\VarSet{b},\VarSet{b}',\VarSet{i},\VarSet{i}', \VarSet{j},\VarSet{j}', \VarSet{k},\VarSet{k}', \VarSet{\ell},\VarSet{\ell}', \VarSet{d},\VarSet{e}) \\ &\qquad\qquad\mu^{\PolynomialBoundingEverything(\StringLength{w})}(\mi{succ},\mi{less},\mi{neq},\mi{add},\VarSet{a},\VarSet{a}',\VarSet{a}'',\VarSet{b},\VarSet{b}') \land {} \\
  &\qquad\qquad\mi{WITHIN}(\mi{ind}^{\PolynomialLengthHausPredD(\StringLength{w})},\mi{succ},\mi{in\mhyphen{}range,\VarSet{d},\VarSet{e}}) \land \mi{MAX\mhyphen{}BIT}(\mi{succ},\mi{in\mhyphen{}range},\mi{max\mhyphen{}bit},\VarSet{d},\VarSet{e}) \land {} \\
  &\qquad\qquad C \land S \land N \land E.
\end{aligned}
\end{equation}
The formula $\FormulaTMHausdorffWithP(\mi{ind}^{\PolynomialLengthHausPredD(\StringLength{w})})$ is in \CNF, as all its components are in \CNF.
The subformula $\mu^{\PolynomialBoundingEverything(\StringLength{w})}$ has polynomially\nbdash-many clauses (\Wrt $\StringLength{w}$), each of them with a constant number of literals.
The rest of the formula is also made by a polynomial number of clauses with a constant number of literals.
Since all these can be obtained in polynomial time in $\StringLength{w}$, the formula $\FormulaTMHausdorffWithP(\mi{ind}^{\PolynomialLengthHausPredD(\StringLength{w})})$ is obtained in polynomial time in $\StringLength{w}$.
By construction, if $\Interpr{I}_z$ is an interpretation of $\mi{ind}^{\PolynomialLengthHausPredD(\StringLength{w})}$ as the bit\nbdash-indexing function of the integer $z$, then $\Interpr{I}_z \models \FormulaTMHausdorffWithP(\mi{ind}^{\PolynomialLengthHausPredD(\StringLength{w})})$ iff $\Language{D}(w,z) = 1$.

\Proofsep

Let us now consider the case in which the number $c$ of certificates is even.
The construction starts from \zcref{eq_high_level_construction_exists_lastSOQ}.
A first simple rewriting of the latter would produce the \CompactSigmaKOddFormula{(c+1)} below:
\begin{equation}
\label{eq_high_level_construction_forall_lastSOQ_first_attempt}
\begin{aligned}[b]
&\FormulaTMHausdorffWithP(\mi{ind}^{\PolynomialLengthHausPredD(\StringLength{w})}) \equiv {} \\
  &\qquad(\exists \mi{succ}, \mi{less}, \mi{neq}, \mi{add}, \mi{in\mhyphen{}range}, \mi{max\mhyphen{}bit}, \mi{cert}_1)
  (\forall \mi{cert}_2) \cdots (\exists \mi{cert}_{c-1}) (\forall \mi{cert}_c) (\exists q, h, t) \\
  &\qquad(\forall \VarSet{a},\VarSet{a}',\VarSet{a}'',\VarSet{b},\VarSet{b}',\VarSet{i},\VarSet{i}', \VarSet{j},\VarSet{j}', \VarSet{k},\VarSet{k}', \VarSet{\ell},\VarSet{\ell}', \VarSet{d},\VarSet{e}) \\ &\qquad\qquad\mu^{\PolynomialBoundingEverything(\StringLength{w})}(\mi{succ},\mi{less},\mi{neq},\mi{add},\VarSet{a},\VarSet{a}',\VarSet{a}'',\VarSet{b},\VarSet{b}') \land {} \\
  &\qquad\qquad\mi{WITHIN}(\mi{ind}^{\PolynomialLengthHausPredD(\StringLength{w})},\mi{succ},\mi{in\mhyphen{}range,\VarSet{d},\VarSet{e}}) \land \mi{MAX\mhyphen{}BIT}(\mi{succ},\mi{in\mhyphen{}range},\mi{max\mhyphen{}bit},\VarSet{d},\VarSet{e}) \land {} \\
  &\qquad\qquad C \land S \land N \land E,
\end{aligned}
\end{equation}
where those on first line are second\nbdash-order quantifiers, and that on second line is a first\nbdash-order quantifier.
To obtain a \CompactSigmaKEvenFormula{c}, we need to amend the formula $\FormulaTMHausdorffWithP(\mi{ind}^{\PolynomialLengthHausPredD(\StringLength{w})})$ in \zcref{eq_high_level_construction_forall_lastSOQ_first_attempt} so that its last second\nbdash-order quantifier is a universal one, and the first\nbdash-order quantifier is an existential one.

The propositional variables $\VarSet{a},\VarSet{a}',\VarSet{a}'',\VarSet{b},\VarSet{b}',\VarSet{d},\VarSet{e}$ can be moved from the \emph{first\nbdash-order universal} quantifier to the last \emph{second\nbdash-order universal} quantifier, which is the one quantifying the function variable $\mi{cert}_c$.
To this aim, these propositional variables are converted throughout the formula to $0$\nbdash-arity function variables, so that they can appear in the scope of a second\nbdash-order universal quantifier.
In this way, $\VarSet{a},\VarSet{a}',\VarSet{a}'',\VarSet{b},\VarSet{b}',\VarSet{d},\VarSet{e}$ are still universally quantified, and we do not need to amend the subformulas $\mu^{\PolynomialBoundingEverything(\StringLength{w})}$, $\mi{WITHIN}$ ${MAX\mhyphen{}BIT}$, and $\mi{S}$ further (besides the conversion of the propositional variables to $0$\nbdash-arity function variables).
In this way, we obtain:
\begin{equation}
\label{eq_high_level_construction_forall_lastSOQ_second_attempt}
\begin{aligned}[b]
&\FormulaTMHausdorffWithP(\mi{ind}^{\PolynomialLengthHausPredD(\StringLength{w})}) \equiv {} \\
  &\qquad(\exists \mi{succ}, \mi{less}, \mi{neq}, \mi{add}, \mi{in\mhyphen{}range}, \mi{max\mhyphen{}bit}, \mi{cert}_1)
  \cdots
  (\forall \mi{cert}_c, \VarSet{a}, \VarSet{a}', \VarSet{a}'', \VarSet{b}, \VarSet{b}', \VarSet{d}, \VarSet{e}) (\exists q, h, t) \\
  &\qquad(\forall \VarSet{i},\VarSet{i}', \VarSet{j},\VarSet{j}', \VarSet{k},\VarSet{k}', \VarSet{\ell},\VarSet{\ell}') \\ &\qquad\qquad\mu^{\PolynomialBoundingEverything(\StringLength{w})}(\mi{succ},\mi{less},\mi{neq},\mi{add},\VarSet{a},\VarSet{a}',\VarSet{a}'',\VarSet{b},\VarSet{b}') \land {} \\
  &\qquad\qquad\mi{WITHIN}(\mi{ind}^{\PolynomialLengthHausPredD(\StringLength{w})},\mi{succ},\mi{in\mhyphen{}range,\VarSet{d},\VarSet{e}}) \land \mi{MAX\mhyphen{}BIT}(\mi{succ},\mi{in\mhyphen{}range},\mi{max\mhyphen{}bit},\VarSet{d},\VarSet{e}) \land {} \\
  &\qquad\qquad C \land S \land N \land E,
\end{aligned}
\end{equation}
To obtain a \CompactSigmaKEvenFormula{c} from \zcref{eq_high_level_construction_forall_lastSOQ_second_attempt}, we have to spare the last \emph{second\nbdash-order existential} quantifier,
i.e., the one quantifying the function variables~$q$, $h$, and~$t$.
These function variables need hence to be moved toward a preceding second\nbdash-order quantifier.
One may think that, to avoid to change the subformula $C \land S \land N \land E$, we might keep these function variables existentially quantified, and so move them toward a preceding \emph{second\nbdash-order existential} quantifier.
However, this cannot be done, because~$q$, $h$, and~$t$, depend on the function variable $\mi{cert}_c$.
Thus, at most we can move~$q$, $h$, and~$t$, toward the previous second\nbdash-order quantifier, i.e., the one quantifying $\mi{cert}_c$, which is however a universal quantifier.
For this to properly work:
(i)~we replace the \emph{first\nbdash-order universal} quantifier
with an \emph{existential} one, and 
(ii)~we mutate ``$C \land S \land N \land E$'' into an implicative form, and obtain:
\begin{equation}
\label{eq_high_level_construction_forall_lastSOQ}
\begin{aligned}[b]
&\FormulaTMHausdorffWithP(\mi{ind}^{\PolynomialLengthHausPredD(\StringLength{w})}) \equiv {} \\
  &\qquad(\exists \mi{succ}, \mi{less}, \mi{neq}, \mi{add}, \mi{in\mhyphen{}range}, \mi{max\mhyphen{}bit}, \mi{cert}_1)
  \cdots
  (\forall \mi{cert}_c, \VarSet{a}, \VarSet{a}', \VarSet{a}'', \VarSet{b}, \VarSet{b}', \VarSet{d}, \VarSet{e}, q, h, t) \\
  &\qquad(\exists \VarSet{i},\VarSet{i}', \VarSet{j},\VarSet{j}', \VarSet{k},\VarSet{k}', \VarSet{\ell},\VarSet{\ell}') \\ &\qquad\qquad\mu^{\PolynomialBoundingEverything(\StringLength{w})}(\mi{succ},\mi{less},\mi{neq},\mi{add},\VarSet{a},\VarSet{a}',\VarSet{a}'',\VarSet{b},\VarSet{b}') \land {} \\
  &\qquad\qquad\mi{WITHIN}(\mi{ind}^{\PolynomialLengthHausPredD(\StringLength{w})},\mi{succ},\mi{in\mhyphen{}range,\VarSet{d},\VarSet{e}}) \land \mi{MAX\mhyphen{}BIT}(\mi{succ},\mi{in\mhyphen{}range},\mi{max\mhyphen{}bit},\VarSet{d},\VarSet{e}) \land {} \\
  &\qquad\qquad (\lnot E \rightarrow \lnot C \lor \lnot S \lor \lnot N),
\end{aligned}
\end{equation}

We now show that the interpretation $\Interpr{I}_z$ of $\mi{ind}^{\PolynomialLengthHausPredD(\StringLength{w})}$ encoding~$z$ satisfies the implicative form employed in \zcref{eq_high_level_construction_forall_lastSOQ} iff $\Omega$ accepts the string $x = {}$``$w \hashsep z \dollarsep u_1 \dollarsep \cdots \dollarsep u_c$''.
Observe that in \zcref{eq_high_level_construction_forall_lastSOQ} the function variables~$q$, $h$, and~$t$, are universally quantified, and the propositional variables $\VarSet{i},\VarSet{i}', \VarSet{j},\VarSet{j}', \VarSet{k},\VarSet{k}', \VarSet{\ell},\VarSet{\ell}'$ are existentially quantified.
Hence, the satisfaction of the formula above by $\Interpr{I}_z$ mirrors whether $\Omega$ accepts $x$ or not iff the following holds:
when $\Omega$ accepts $x$, then for all instantiations of $q$, $h$, and~$t$, there is a truth assignment for $\VarSet{i},\VarSet{i}', \VarSet{j},\VarSet{j}', \VarSet{k},\VarSet{k}', \VarSet{\ell},\VarSet{\ell}'$ satisfying `$(\lnot E \rightarrow \lnot C \lor \lnot S \lor \lnot N)$'; and
when $\Omega$ rejects $x$, then there is an instantiation of $q$, $h$, and~$t$, such that, for all truth assignments for $\VarSet{i},\VarSet{i}', \VarSet{j},\VarSet{j}', \VarSet{k},\VarSet{k}', \VarSet{\ell},\VarSet{\ell}'$, the implication `$(\lnot E \rightarrow \lnot C \lor \lnot S \lor \lnot N)$' is not satisfied.

Observe that there exists a truth assignment for $\VarSet{i},\VarSet{i}', \VarSet{j},\VarSet{j}', \VarSet{k},\VarSet{k}', \VarSet{\ell},\VarSet{\ell}'$ evaluating:
\begin{itemize}[nosep,label=--,left=0pt]
  \item $\lnot C$ to $\valtrue$, when $q$, $h$, or~$t$, are instantiated to \emph{break} the consistency requirement;
  \item $\lnot S$ to $\valtrue$, when $q$, $h$, or~$t$, are instantiated to \emph{incorrectly} encode the initial configuration of $\Omega$; and
  \item $\lnot N$ to $\valtrue$, when $q$, $h$, or~$t$, are instantiated to \emph{incorrectly} encode the computation of $\Omega$ over $x$.
\end{itemize}

Assume that $\Omega$ accepts $x$.
Given the interpretation $\Interpr{I}_z$, the set of all possible instantiations of~$q$, $h$, and~$t$, can be partitioned into two subsets:
(a)~the instantiations evaluating $E$ to $\valtrue$, and
(b)~the instantiations evaluating $E$ to $\valfalse$ (notice that $E$ does not depend on the propositional variables $\VarSet{i},\VarSet{i}', \VarSet{j},\VarSet{j}', \VarSet{k},\VarSet{k}', \VarSet{\ell},\VarSet{\ell}'$; see above).
For the instantiations~(a), since $\lnot E$ is evaluated to $\valfalse$, the implication `$(\lnot E \rightarrow \lnot C \lor \lnot S \lor \lnot N)$' is trivially satisfied irrespective of the truth values of $C$, $S$, and $N$.
On the other hand, for the instantiations~(b), since $\lnot E$ is evaluated to $\valtrue$, the implication `$(\lnot E \rightarrow \lnot C \lor \lnot S \lor \lnot N)$' is satisfied iff there exists an assignment for $\VarSet{i},\VarSet{i}', \VarSet{j},\VarSet{j}', \VarSet{k},\VarSet{k}', \VarSet{\ell},\VarSet{\ell}'$ evaluating `$\lnot C \lor \lnot S \lor \lnot N$' to $\valtrue$.
Because $\Omega$ is assumed to accept $x$, all instantiations of~$q$, $h$, and~$t$, evaluating $E$ to $\valfalse$ must have an error in the encoding of the computation of $\Omega$ on $x$.
Hence, there must be assignments for $\VarSet{i},\VarSet{i}', \VarSet{j},\VarSet{j}', \VarSet{k},\VarSet{k}', \VarSet{\ell},\VarSet{\ell}'$, satisfying `$\lnot C \lor \lnot S \lor \lnot N$', and hence satisfying the implication.

Assume now that $\Omega$ rejects $x$.
Given the interpretation $\Interpr{I}_z$, consider the instantiation $\Interpr{I}_{\Omega(x)}$ of~$q$, $h$, and~$t$, correctly encoding $\Omega$'s computation on $x$.
Since $\Omega$ rejects $x$, $\Interpr{I}_{\Omega(x)}$ evaluates $\lnot E$ to $\valtrue$ (remember that $E$ does not depend on the propositional variables $\VarSet{i},\VarSet{i}', \VarSet{j},\VarSet{j}', \VarSet{k},\VarSet{k}', \VarSet{\ell},\VarSet{\ell}'$).
Thus, the implication `$(\lnot E \rightarrow \lnot C \lor \lnot S \lor \lnot N)$' is satisfied iff there is an assignment for $\VarSet{i},\VarSet{i}', \VarSet{j},\VarSet{j}', \VarSet{k},\VarSet{k}', \VarSet{\ell},\VarSet{\ell}'$ evaluating `$\lnot C \lor \lnot S \lor \lnot N$' to $\valtrue$.
However, since $\Interpr{I}_{\Omega(x)}$ instantiates~$q$, $h$, and~$t$, to correctly encode $\Omega$'s computation on $x$, all assignments for $\VarSet{i},\VarSet{i}', \VarSet{j},\VarSet{j}', \VarSet{k},\VarSet{k}', \VarSet{\ell},\VarSet{\ell}'$, evaluates `$C \land S \land N$' to $\valtrue$, and hence the implication is not satisfied by any of these assignments.

To conclude, we show that $(\lnot E \rightarrow \lnot C \lor \lnot S \lor \lnot N)$ can be translated to \CNF in polynomial time.
Observe that the implication $(\lnot E \rightarrow \lnot C \lor \lnot S \lor \lnot N)$ is equivalent to $F = (E \lor \lnot C \lor \lnot S \lor \lnot N)$.
Remember that $C$, $S$, and $N$, are \CNF formulas (see above).
Hence, the disjunction of their negations, by De~Morgan's laws, is a \DNF formula.
Moreover, $E$ is a single literal, hence overall $F$ is a \DNF formula.
Observe that $C$ and $N$ have a constant number of clauses containing each a constant number of literals.
Therefore, the disjunction of their negations has a constant number of terms containing each a constant number of literals.
Let us now focus on $S$.
Remember that $S$ is constituted by various subformulas, of which only $W$ has a number of literals depending on the size of the string $w$;
all the other parts of $S$ have a constant number of clauses containing each a constant number of literals.
Thus, except for $W$, the negation of the remaining part of $S$ is a \DNF containing a constant number of terms with a constant number of literals.
If we look at the subformula $W$ (see above), it is a conjunction of $\StringLength{w}$\nbdash-many literals;
hence its negation is a disjunction of $\StringLength{w}$\nbdash-many literals.
Since $F$ is a \DNF formula, $F$ can be rewritten as $F = G \lor \lnot W$, where $G$ collects all the terms coming from all the parts of $F$ but $W$ (which is a subformula of $S$).
Notice that $G$ is made by a constant number of terms containing each a constant number of literals, whereas $\lnot W$ has $\StringLength{w}$\nbdash-many terms of \emph{one} literal.
Therefore, the number of clauses of the \CNF form of $F$ contains a constant number of clauses containing each a linear number of literals.%
\footnote{If $\gamma = \bigwedge_{i = 1}^m D_i$ is a \DNF formula, where $D_i$ are its terms, its \CNF form can easily be obtained by De~Morgan's and distribution laws.
If $|D_i|$ is the number of literals in $D_i$, the \CNF form of $\gamma$ obtained in this way has $\prod_{i = 1}^m |D_i|$ clauses containing each $m$ literals.}

Thus, the entire matrix of \zcref{eq_high_level_construction_forall_lastSOQ} can be translated to a \CNF formula of polynomial size in $\StringLength{w}$.
By this, also the formula in \zcref{eq_high_level_construction_forall_lastSOQ} con be obtained in polynomial time in $\StringLength{w}$.
\end{proof}

The next result states that we can efficiently build from a generic \SigmaKFormula{c} $\Psi$ an equivalent \SimpleSigmaKFormula{c} $\Psi'$ where each set of quantified Boolean function variables is replaced by a single Boolean function variable.

\begin{lemma}[store=SigmaFormulaSubstituteMultipleFunctionsWithOne]
\label{theo_SigmaFormula_substitute_multiple_functions_with_one}
Let $\Psi(\VarSet{g},\VarSet{y}) =
(\exists \VarSet{f}^1)
\cdots
(\SOQ_c \VarSet{f}^c)
(\FOQ_1 \VarSet{x}^1)
\cdots
(\FOQ_m \VarSet{x}^m) \PrefixMatrixSeparator
\psi(\VarSet{f}^1,\dots,\VarSet{f}^c\VarSet{g},\VarSet{x}^1,\dots,\VarSet{x}^m,\VarSet{y})$ be a \SigmaKFormula{c}, with $\psi$ in \CNF (resp., \DNF).
An equivalent \SimpleSigmaKFormula{c} $\Psi'(\VarSet{g},\VarSet{y}) =
(\exists h_1)
\cdots
(\SOQ_c h_c) \linebreak[0]
(\FOQ_1 \VarSet{x}^1)
\cdots
(\FOQ_m \VarSet{x}^m) \linebreak[0] \PrefixMatrixSeparator 
\psi'(h_1,\dots,h_c,\linebreak[0] \VarSet{g},\VarSet{x}^1,\dots,\VarSet{x}^m,\VarSet{y})$, with $\psi'$ in \CNF (resp., \DNF), can be built from $\Psi$ in polynomial time.
\end{lemma}

\begin{proof}
Let $\Psi(\VarSet{g},\VarSet{y}) =
(\exists f^1_1,\dots,f^1_{n_1})
(\forall f^2_1,\dots,f^2_{n_2})
\cdots
(\SOQ_c f^c_1,\dots,f^c_{n_c})
(\FOQ_1 \VarSet{x}^1)
\cdots
(\FOQ_m \VarSet{x}^m) \PrefixMatrixSeparator
\psi(f^1_1,\dots,f^1_{n_1},\linebreak[0] \dots,f^c_1,\dots,f^c_{n_c},\VarSet{g},\VarSet{x}^1,\dots,\VarSet{x}^m,\VarSet{y})$, be a \SigmaKFormula{c}, where $f^i_j$ are Boolean function variables of arity $a_{i,j}$, respectively, for all $i$ with $1 \leq i \leq c$ and all $j$ with $1 \leq j \leq n_i$.
We show below how to aggregate all the Boolean function variables $\set{f^i_1,\dots,f^i_{n_i}}$ of the $i$\nbdash-th quantifier in just one function variable $h_i$.
A similar substitution can be applied to all sets of quantified function variables of $\Psi$, and obtain the rewritten \SimpleSigmaKFormula{c} $\Psi'(\VarSet{g},\VarSet{y}) =
(\exists h_1)
(\forall h_2)
\cdots
(\SOQ_c h_c)
(\FOQ_1 \VarSet{x}^1)
\cdots
(\FOQ_m \VarSet{x}^m) \PrefixMatrixSeparator
\psi'(h_1,\dots,h_c,\linebreak[0] \VarSet{g},\VarSet{x}^1,\dots,\VarSet{x}^m,\VarSet{y})$.

Let $\wt{a}_i = \max_{1 \leq j \leq n_i} (a_{i,j})$.
Consider the Boolean function variable $h_i$ of arity $m_i + \wt{a}_i$, where $m_i = \lfloor \log n_i \rfloor +1$, that we use in $\Psi'$ to replace the functions variables $f^i_1,\dots,f^i_{n_i}$.
The first $m_i$ arguments of $h_i$ are used as an index to distinguish the function $f^i_j$ being replaced, while the remaining $\wt{a}_i$ arguments of $h_i$ accommodate the arguments of the functions $f^i_j$ (for a given $i$ and all $j$ with $1 \leq j \leq n_i$).
Let $f^i_j(z_1,\dots,z_{a_{i,j}})$ be an occurrence in $\psi$ of the Boolean function variable $f^i_j$.
In $\psi'$, we simply replace $f^i_j(z_1,\dots,z_{a_{i,j}})$ with
\[h_i(\BooleanEncoding{j}; \underbracket[.5pt]{0,\dots,0}_{\mathclap{(\wt{a}_i-a_{i,j})\text{-many}}},z_1,\dots,z_{a_{i,j}}),
\]
where `$\BooleanEncoding{j}$' denotes the Boolean encoding over $m_i$ bits of the value $j$.
By replacing in a similar way all the occurrences of the Boolean function variables in $\psi$ we obtain~$\psi'$.
Clearly, this can be carried out in polynomial time in the size of $\Psi$.
Furthermore, since we operate only a substitution of the function literals, the structure of $\psi$ is preserved in $\psi'$, and hence, if $\psi$ is in \CNF (resp., \DNF), then $\psi'$ is in \CNF (resp., \DNF).

Observe that in $\Psi'$ there are no fresh propositional variables, and the arguments of $h_i$, for all $i$, do not include propositional variables besides those appearing as arguments in the occurrences of the functions~$f^i_j$ in~$\Psi$.
For this reason, it is not hard to see that $\Psi$ and $\Psi'$ are equivalent.
\end{proof}

Notice that, in the function mapping introduced in the proof of the \zcref*[typeset=name,nocap]{theo_SigmaFormula_substitute_multiple_functions_with_one} above, the occurrences in $\psi$ of a function $f^i_j()$ of arity~$0$ are always substituted in $\psi'$ by the term $h_i(\BooleanEncoding{j};0,\dots,0)$.
Hence, a single Boolean value is associated with $f^i_j()$'s mapping over $h_i$ in $\psi'$, as expected.  

\getkeytheorem{ComplexityLexMaxFunc}

\begin{proof}
(\emph{Membership}).
We prove \LexMaxFuncSigmaFormula{c} in $\Oracle{\ExpTime}{\SigmaP{c}}$.
Consider first this auxiliary problem, named~$\Language{A}$:
for a pair $\pair{\Psi(\VarSet{g}),\Interpr{I}}$, where $\Psi(\VarSet{g})$ is a \SigmaKFormula{c}, with $\VarSet{g}$ an ordered set of Boolean function variables, and $\Interpr{I}$ is an interpretation for $\VarSet{g}$, decide whether there exists a model $\Interpr{I}'$ of $\Psi(\VarSet{g})$ such that $\Interpr{I}'$ is \emph{not} lexicographically\nbdash-smaller than $\Interpr{I}$.
This problem is easily shown in $\Oracle{\NPTime}{\SigmaP{c-1}} = \SigmaP{c}$.
First notice that the interpretation $\Interpr{I}$, which is part of the input, is of \emph{exponential} size in the size of $\Psi(\VarSet{g})$, because $\Interpr{I}$ has to encode the exponentially\nbdash-long instantiations of the function variables in $\VarSet{g}$.
Therefore, a computation that is \emph{polynomial} in the size of the pair $\pair{\Psi(\VarSet{g}),\Interpr{I}}$, is actually exponential in the size of $\Psi(\VarSet{g})$.
To answer $\Language{A}$, we guess an interpretation $\Interpr{I}'$ for the function variables in $\VarSet{g}$ (feasible in nondeterministic polynomial time), then check that $\Interpr{I}'$ is not lexicographically\nbdash-smaller than $\Interpr{I}$ (feasible in polynomial time), and that $\EvalInterpr{\Psi(\VarSet{g})}{\Interpr{I}'}$ is $\valtrue$.
The latter is feasible in $\Oracle{\NExpTime}{\SigmaP{c-1}}$ in the size of $\Psi(\VarSet{g})$~\cite{Lohrey2012,Luck2016-techrep}.
By what we have observed, this is also in $\Oracle{\NPTime}{\SigmaP{c-1}}$ in the size of the input pair $\pair{\Psi(\VarSet{g}),\Interpr{I}}$.

Let us now focus on \LexMaxFuncSigmaFormula{c}.
Let $\Phi(\VarSet{g})$ be a \SigmaKFormula{c}, where $\VarSet{g} = \set{g_n,\dots,g_1}$ is an ordered set of Boolean function variables of arity $a_n,\dots,a_1$, respectively.
An $\ExpTime$ oracle machine $\Oracle{\Machine{M}}{?}$ can decide \LexMaxFuncSigmaFormula{c} via the aid of a $\SigmaP{c}$ oracle for $\Language{A}$.
Intuitively, $\Machine{M}$ performs a binary search in the space of all the possible instantiations for the Boolean functions $\set{g_n,\dots,g_1}$.
Remember that, for a Boolean function $g$, the unfolding of $g$ is the binary string that in first position has the Boolean value $g(1,\dots,1)$, in second position has the Boolean value $g(1,\dots,1,0)$, and so on.
Therefore, the search space for $\Machine{M}$ is simply the space of all the possible binary strings which are obtained as the juxtaposition of the unfoldings of $g_n,\dots,g_1$ (in this order) for all the possible instantiations of $g_n,\dots,g_1$.
The length of such strings is $\ell = \sum_{i = 1}^{n} 2^{a_i}$, and hence $\Machine{M}$ has to search within a space of $2^{\ell}$ binary strings (of exponential length $\ell$).

By this, $\Machine{M}$ can systematically explore the search space, via a binary search, to individuate the lexicographic\nbdash-maximum model for $\Phi(\VarSet{g})$ via an oracle for $\Language{A}$.
The search space that $\Machine{M}$ explores is doubly exponential in the size of the input to $\Machine{M}$, and hence it can be explored via binary search with exponentially\nbdash-many queries to the oracle.
Once the lexicographic\nbdash-maximum model $\Interpr{I}$ for $\Phi(\VarSet{g})$ has been computed, $\Machine{M}$ can return its answer by looking at whether the Boolean value of $\EvalInterpr{g_1}{\Interpr{I}}(0,\dots,0)$ is $\valtrue$/$\valfalse$.

(\emph{Hardness}).
We now prove that \LexMaxFuncSigmaFormula{c} is $\Oracle{\ExpTime}{\SigmaP{c}}$\HardSuffix by showing that there exists a polynomial reduction from every language $\Language{L} \in \Oracle{\ExpTime}{\SigmaP{c}}$ to \LexMaxFuncSigmaFormula{c}.

Below we exhibit a reduction producing a formula with a single free Boolean function variable $g$, and we show the $\Oracle{\ExpTime}{\SigmaP{c}}$\HardSuffix{}ness of deciding whether $g(0,\dots,0) = \valtrue$ in the lexicographic\nbdash-maximum model of the formula.
Since $\Oracle{\ExpTime}{\SigmaP{c}}$ is closed under complement, the $\Oracle{\ExpTime}{\SigmaP{c}}$\HardSuffix{}ness of deciding whether $g(0,\dots,0)$ is $\valfalse$ in the lexicographic\nbdash-maximum model of the formula will follow.

Let $\Language{L} \in \Oracle{\ExpTime}{\SigmaP{c}}$ be a language.
By \zcref{theo_summary_equivalence_intermediate_levels_Hausdorff} $\Language{L} \in \BoundedHausdCLASS{2^{2^{\PolFunctions}}}{\Oracle{\NExpTime}{\SigmaP{c-1}}}$ and hence $\Language{L}$ is characterized by a $\Oracle{\NExpTime}{\SigmaP{c-1}}$ Hausdorff predicate $\Language{D}$ of doubly exponential length.
By \zcref{theo_coding_nexp_machine_for Hausdorff_predicate}, there are polynomials $\PolynomialBoundingEverything(n)$ and $\PolynomialLengthHausPredD(n)$, with $\PolynomialBoundingEverything(n) \geq \PolynomialLengthHausPredD(n)$, such that $2^{2^{\PolynomialLengthHausPredD(n)}} - 1$ bounds the length of $\Language{D}$ and $\FormulaTMHausdorffWithP(\mi{ind}^{\PolynomialLengthHausPredD(\StringLength{w})})$ encodes whether $\Language{D}(w,z) = 1$, when $\mi{ind}^{\PolynomialLengthHausPredD(\StringLength{w})}$ is interpreted as the bit\nbdash-indexing function of $z$ over $2^{\PolynomialLengthHausPredD(\StringLength{w})}$ bits.

Observe now that the formula $\FormulaTMHausdorffWithP(\mi{ind}^{\PolynomialLengthHausPredD(\StringLength{w})})$ is actually the formula that the reduction has to produce, because, for every string $w$, it holds that $w \in \Language{L}$ iff the lexicografic\nbdash-maximum model $\wt{\Interpr{I}}$ of $\FormulaTMHausdorffWithP(\mi{ind}^{\PolynomialLengthHausPredD(\StringLength{w})})$ instantiates $\mi{ind}^{\PolynomialLengthHausPredD(\StringLength{w})}$ so that $\EvalInterpr{\mi{ind}^{\PolynomialLengthHausPredD(\StringLength{w})}}{\wt{\Interpr{I}}}(0,\dots,0) = \valtrue$.

By \zcref{theo_SigmaFormula_substitute_multiple_functions_with_one} and simple prenex rewriting, we can obtain in polynomial time a \CNF \SimpleCompactSigmaKOddFormula{c} or \SimpleCompactSigmaKEvenFormula{c} equivalent to $\FormulaTMHausdorffWithP(\mi{ind}^{\PolynomialLengthHausPredD(\StringLength{w})})$, according to whether $c$ is odd or even, respectively.
Finally, notice that $\FormulaTMHausdorffWithP(\mi{ind}^{\PolynomialLengthHausPredD(\StringLength{w})})$ is satisfiable, because $\Language{D}(w,0) = 1$ for every string $w$, and hence the interpretation instantiating $\mi{ind}^{\PolynomialLengthHausPredD(\StringLength{w})}$ as the function that evaluates to $\valfalse$ on every input is a model.

We conclude by showing that the hardness of the problem is preserved even under the additional assumption that, for every two interpretations $\Interpr{I}$ and $\Interpr{J}$ of $\VarSet{g}$ such that $\Interpr{I} \lexsucc \Interpr{J}$, it holds that $\Interpr{I} \models \Phi(\VarSet{g})$ implies $\Interpr{J} \models \Phi(\VarSet{g})$.
This follows because $\FormulaTMHausdorffWithP(\mi{ind}^{\PolynomialLengthHausPredD(\StringLength{w})})$ encodes the Hausdorff predicate $\Language{D}$, for which $\Language{D}(w,z-1) \geq \Language{D}(w,z)$ for every $w$ and $z>1$, and $\mi{ind}^{\PolynomialLengthHausPredD}$ encodes in binary the indices $z$ passed to $\Language{D}$.
\end{proof}

\getkeytheorem{ComplexityLexMax}

\begin{proof}
(\emph{Membership}).
By the equivalence $\BoundedOracle{\ExpTime}{\SigmaP{c}}{\PolFunctions} = \Oracle{\NPTime}{\Oracle{\NExpTime}{\SigmaP{c-1}}}$ stemming from \zcref{theo_summary_equivalence_intermediate_levels_Hausdorff,theo_summary_nexp_nexp_Hausdorff}, we prove that \LexMaxSigmaFormula{c} is in $\BoundedOracle{\ExpTime}{\SigmaP{c}}{\PolFunctions}$ by showing that \LexMaxSigmaFormula{c} is in $\Oracle{\NPTime}{\Oracle{\NExpTime}{\SigmaP{c-1}}}$.

Consider first this auxiliary problem, that we name $\Language{A}$:
for a pair $\pair{\Psi(\VarSet{y}),\Interpr{I}}$, where $\Psi(\VarSet{y})$ is a \SigmaKFormula{c}, with $\VarSet{y}$ an ordered set of Boolean propositional variables, and $\Interpr{I}$ is an interpretation for $\VarSet{y}$, decide whether $\Psi(\VarSet{y})$ has a model $\Interpr{I}'$ lexicographically greater than~$\Interpr{I}$.
The problem $\Language{A}$ is easily shown in $\Oracle{\NExpTime}{\SigmaP{c-1}}$:
guess an interpretation $\Interpr{I}'$ for $\VarSet{y}$ (feasible in nondeterministic polynomial time), then check that $\Interpr{I}'$ is lexicographically greater than $\Interpr{I}$ (feasible in polynomial time), and that $\EvalInterpr{\Psi(\VarSet{y})}{\Interpr{I}'}$ is $\valtrue$ (feasible in $\Oracle{\NExpTime}{\SigmaP{c-1}}$, see~\cite{Lohrey2012,Luck2016-techrep});
observe here that, unlike in the membership proof of \zcref{theo_complexity_LexMaxFunc}, the size of the interpretation $\Interpr{I}$ is only polynomial in the size of $\Psi(\VarSet{y})$, because here $\Interpr{I}$ is a truth assignment for the Boolean propositional variables $\VarSet{y}$, and not, like in \zcref{theo_complexity_LexMaxFunc}, a function instantiation of Boolean function variables.

Let us now focus on \LexMaxSigmaFormula{c}, and let $\Phi(\VarSet{y})$ be a \SigmaKFormula{c}.
An $\NPTime$ oracle machine $\Oracle{\Machine{M}}{?}$, aided by a $\Oracle{\NExpTime}{\SigmaP{c-1}}$ oracle for the problem $\Language{A}$, can decide \LexMaxSigmaFormula{c} as follows.
First, $\Machine{M}$ guesses an interpretation $\Interpr{I}$ for $\VarSet{y}$ (feasible in nondeterministic polynomial time).
If $\Interpr{I}$ assigns $\valtrue$ to all the variables in $\VarSet{y}$, then $\Machine{M}$ checks that $\Interpr{I}$ satisfies $\Phi(\VarSet{y})$, and next $\Machine{M}$ answers \yeslbl, because $\EvalInterpr{y_1}{\Interpr{I}} = \valtrue$.
On the other hand, if $\Interpr{I}$ does \emph{not} assign $\valtrue$ to all the variables in $\VarSet{x}$, then $\Machine{M}$ also computes the lexicographic\nbdash-successor $\Interpr{I}^{+}$ of $\Interpr{I}$ (feasible in polynomial time).
After this, $\Machine{M}$ can check that $\Interpr{I}$ is the lexicographic\nbdash-maximum model of $\Phi(\VarSet{y})$ by asking two questions to its oracle:
namely, whether $\pair{\Phi(\VarSet{y}),\Interpr{I}}$ and $\pair{\Phi(\VarSet{y}),\Interpr{I}^{+}}$ are a \yeslbl- and a \noinst, respectively, of the problem $\Language{A}$.
Once $\Machine{M}$ has ascertained that $\Interpr{I}$ is the lexicographic\nbdash-maximum model of $\Phi(\VarSet{y})$, $\Machine{M}$ returns its answer according to the Boolean value assignment for $y_1$ in $\Interpr{I}$.%
\footnote{An $\BoundedOracle{\ExpTime}{\SigmaP{c}}{\PolFunctions}$ procedure for \LexMaxSigmaFormula{c} could compute the lexicographic\nbdash-maximum model of $\Phi(\VarSet{y})$ via binary search.}

(\emph{Hardness}).
We now prove that \LexMaxSigmaFormula{c} is $\BoundedOracle{\ExpTime}{\SigmaP{c}}{\PolFunctions}$\HardSuffix by showing that there exists a polynomial reduction from every language $\Language{L} \in \BoundedOracle{\ExpTime}{\SigmaP{c}}{\PolFunctions}$ to \LexMaxSigmaFormula{c}.
Below we will focus on showing the $\BoundedOracle{\ExpTime}{\SigmaP{c}}{\PolFunctions}$\HardSuffix{}ness of deciding whether $y_1$ is $\valtrue$ in the lexicographic\nbdash-maximum model of $\Phi(\VarSet{y})$.
With this result in place, since $\BoundedOracle{\ExpTime}{\SigmaP{c}}{\PolFunctions}$ is closed under complement, the $\BoundedOracle{\ExpTime}{\SigmaP{c}}{\PolFunctions}$\HardSuffix{}ness of deciding whether $y_1$ is $\valfalse$ in the lexicographic\nbdash-maximum model of $\Phi(\VarSet{y})$ will consequently be proven.

Let $\Language{L}$ be an arbitrary language in $\BoundedOracle{\ExpTime}{\SigmaP{c}}{\PolFunctions}$.
We prove $\Language{L} \KarpRed {}$\LexMaxSigmaFormula{c} by exhibiting a reduction that, for every string $w$, constructs a formula $\Phi^{\Language{L}}_{w}(\VarSet{y})$, with $\VarSet{y}$ an ordered set of Boolean propositional variables, such that $\Phi^{\Language{L}}_{w}(\VarSet{y})$ is a \SimpleCompactSigmaKOddFormula{c} or a \SimpleCompactSigmaKEvenFormula{c} depending on whether $c$ is odd or even, respectively, and $w \in \Language{L}$ iff the lexicographic\nbdash-maximum model of $\Phi^{\Language{L}}_{w}(\VarSet{y})$ assigns $\valtrue$ to the least\nbdash-significant variable in~$\VarSet{y}$.

Since $\Language{L} \in \BoundedOracle{\ExpTime}{\SigmaP{c}}{\PolFunctions}$, by \zcref{theo_summary_equivalence_intermediate_levels_Hausdorff} $\Language{L} \in \BoundedHausdCLASS{2^{\PolFunctions}}{\Oracle{\NExpTime}{\SigmaP{c-1}}}$ and hence $\Language{L}$ is characterized by a $\Oracle{\NExpTime}{\SigmaP{c-1}}$ Hausdorff predicate $\Language{D}$ of exponential length.
By \zcref{theo_coding_nexp_machine_for Hausdorff_predicate}, there are polynomials $\PolynomialBoundingEverything(n)$ and $\PolynomialLengthHausPredD(n)$, with $\PolynomialBoundingEverything(n) \geq \PolynomialLengthHausPredD(n)$, such that $2^{\PolynomialLengthHausPredD(n)} - 1$ bounds the length of $\Language{D}$ and $\FormulaTMHausdorffWithP(\mi{ind}^{\PolynomialLengthHausPredD(\StringLength{w})})$ encodes whether $\Language{D}(w,z) = 1$, when $\mi{ind}^{\PolynomialLengthHausPredD(\StringLength{w})}$ is interpreted as the bit\nbdash-indexing function of $z$ over $2^{\PolynomialLengthHausPredD(\StringLength{w})}$ bits.

Let us now focus on $\Phi^{\Language{L}}_{w}(\VarSet{y})$.
Take the ordered set of propositional variables $\VarSet{y} = \set{y_{\PolynomialLengthHausPredD(\StringLength{w})},\dots,y_1}$;
with these variables, we can represent in binary integers between~$0$ and~$2^{\PolynomialLengthHausPredD(\StringLength{w})}-1$ (remember that $2^{\PolynomialLengthHausPredD(n)} - 1$ bounds the length of $\Language{D}$, hence, for each string $w$, the Hausdorff index 
of $w$ \Wrt $\Language{D}$ is such that $\HausdIndex{w}{\Language{D}} \leq 2^{\PolynomialLengthHausPredD(\StringLength{w})} - 1$). 

Consider first a subformula copying the integer value encoded in the variables $\VarSet{y}$ into the bit\nbdash-indexing function $\mi{ind}^{\PolynomialLengthHausPredD(\StringLength{w})}$, so that this value can be read within $\FormulaTMHausdorffWithP(\mi{ind}^{\PolynomialLengthHausPredD(\StringLength{w})})$.
Let $\VarSet{x} = \set{x_{\PolynomialLengthHausPredD(\StringLength{w})}, \dots, x_1}$ be an additional set of Boolean propositional variables.
The function $\mi{less}$ below is defined in the formula $\mu$ of \zcref{theo_SigmaFormula_arithmetic} and embedded in $\FormulaTMHausdorffWithP(\mi{ind}^{\PolynomialLengthHausPredD(\StringLength{w})})$ (see \zcref{theo_coding_nexp_machine_for Hausdorff_predicate}).
Since its arguments have $\PolynomialBoundingEverything(\StringLength{w})$ bits and $\PolynomialBoundingEverything(n) \geq \PolynomialLengthHausPredD(n)$, the extra most significant bits are set to zero.
\begin{align*}
\mi{copy}(\mi{ind}^{\PolynomialLengthHausPredD(\StringLength{w})},\VarSet{y},\VarSet{x}) \equiv
  &\bigwedge_{1 \leq k \leq \PolynomialLengthHausPredD(\StringLength{w})}
     \Bigl(
        \bigl(y_k \rightarrow \mi{ind}^{\PolynomialLengthHausPredD(\StringLength{w})}(\BooleanEncoding{k {-} 1})\bigr) \land 
        \bigl(\lnot y_k \rightarrow \lnot \mi{ind}^{\PolynomialLengthHausPredD(\StringLength{w})}(\BooleanEncoding{k {-} 1})\bigr)
     \Bigr) \land {} \\
     &\bigl(\lnot \mi{less}(\VarSet{x},\BooleanEncoding{\PolynomialLengthHausPredD(\StringLength{w})}) \rightarrow \lnot \mi{ind}^{\PolynomialLengthHausPredD(\StringLength{w})}(\VarSet{x})\bigr).
\end{align*}
Since the variables $\VarSet{x}$ will be quantified universally in the final formula, the expression ``$\bigl(\lnot \mi{less}(\VarSet{x},\BooleanEncoding{\PolynomialLengthHausPredD(\StringLength{w})}) \rightarrow \linebreak[0] \lnot \mi{ind}^{\PolynomialLengthHausPredD(\StringLength{w})}(\VarSet{x})\bigr)$'' aims at completing the value encoded in the bit\nbdash-indexing function by setting to `$0$' in $\mi{ind}^{\PolynomialLengthHausPredD(\StringLength{w})}$ all bits of $z$ in positions with index strictly greater than $\PolynomialLengthHausPredD(\StringLength{w}){}-1$.
Clearly, this subformula is in \CNF and can be obtained in polynomial time in $\StringLength{w}$.

The overall formula for this reduction is:
\begin{equation*}
  \Phi^{\Language{L}}_{w}(\VarSet{y}) \equiv 
    (\exists \mi{ind}^{\PolynomialLengthHausPredD(\StringLength{w})}) (\forall \VarSet{x}) \PrefixMatrixSeparator \mi{copy}(\mi{ind}^{\PolynomialLengthHausPredD(\StringLength{w})},\VarSet{y},\VarSet{x}) \land 
    \FormulaTMHausdorffWithP(\mi{ind}^{\PolynomialLengthHausPredD(\StringLength{w})}).
\end{equation*}

Observe that $\Phi^{\Language{L}}_{w}(\VarSet{y})$ can be built from $w$ in polynomial time and is in \CNF, since $\FormulaTMHausdorffWithP(\mi{ind}^{\PolynomialLengthHausPredD(\StringLength{w})})$ is in \CNF.
By \zcref{theo_SigmaFormula_substitute_multiple_functions_with_one} and simple prenex rewriting, we can obtain in polynomial time an equivalent \CNF \SimpleCompactSigmaKOddFormula{c} or \SimpleCompactSigmaKEvenFormula{c}, according to whether $c$ is odd or even, respectively.
In the latter case, the propositional variables $\VarSet{x}$ are converted into function variables of arity zero and moved to the last second\nbdash-order universal quantifier.
Moreover, the variables $\VarSet{y}$ do not occur as arguments of function variables.
Finally, $\Phi^{\Language{L}}_{w}(\VarSet{y})$ is satisfiable, because $\Language{D}(w,0) = 1$ for every $w$, and hence the interpretation assigning $\valfalse$ to all variables in $\VarSet{y}$ is a model.

It remains to show that $w \in \Language{L}$ iff the lexicographic\nbdash-maximum model of $\Phi^{\Language{L}}_{w}(\VarSet{y})$ assigns $\valtrue$ to~$y_1$.
By \zcref{theo_coding_nexp_machine_for Hausdorff_predicate}, this follows because $\mi{copy}$ transfers the integer encoded by $\VarSet{y}$ to $\mi{ind}^{\PolynomialLengthHausPredD(\StringLength{w})}$.

We conclude by showing that the hardness of the problem is preserved even under the additional assumption that, for every two interpretations $\Interpr{I}$ and $\Interpr{J}$ of $\VarSet{y}$ such that $\Interpr{I} \lexsucc \Interpr{J}$, it holds that $\Interpr{I} \models \Phi(\VarSet{y})$ implies $\Interpr{J} \models \Phi(\VarSet{y})$.
Indeed, $\mi{copy}(\mi{ind}^{\PolynomialLengthHausPredD(\StringLength{w})},\VarSet{y},\VarSet{x})$ copies the binary string encoded by $\VarSet{y}$ to $\mi{ind}^{\PolynomialLengthHausPredD(\StringLength{w})}$, while $\FormulaTMHausdorffWithP(\mi{ind}^{\PolynomialLengthHausPredD(\StringLength{w})})$ encodes the Hausdorff predicate $\Language{D}$, for which $\Language{D}(w,z-1) \geq \Language{D}(w,z)$ for every $w$ and $z>1$.
\end{proof}

\getkeytheorem{ComplexityMaxSat}

\begin{proof}
(\emph{Membership}).
By the equivalence $\BoundedOracle{\ExpTime}{\SigmaP{c}}{\LogFunctions} = \BoundedParOracle{\ExpTime}{\SigmaP{c}}{\PolFunctions}$ stemming from \zcref{theo_summary_equivalence_intermediate_levels_Hausdorff}, we prove that \MaxSatSigmaFormula{c} is in $\BoundedOracle{\ExpTime}{\SigmaP{c}}{\LogFunctions}$ by showing that \MaxSatSigmaFormula{c} is in $\BoundedParOracle{\ExpTime}{(\Oracle{\NPTime}{\SigmaP{c-1}})}{\PolFunctions} = \BoundedParOracle{\ExpTime}{\SigmaP{c}}{\PolFunctions}$.

Consider first this auxiliary problem, that we name $\Language{A}$:
for a pair $\pair{\Psi(\VarSet{y}),k}$, where $\Psi(\VarSet{y})$ is a \SigmaKFormula{c}, with $\VarSet{y}$ a set of Boolean propositional variables, and $k \geq 0$ is an integer, decide whether there exists an interpretation for $\VarSet{y}$ that is a model of $\Psi(\VarSet{y})$ whose Hamming weight is at least~$k$.
The problem $\Language{A}$ is easily shown in $\Oracle{\NExpTime}{\SigmaP{c-1}}$:
guess an interpretation $\Interpr{I}$ for $\VarSet{y}$ (feasible in nondeterministic polynomial time), then check that the Hamming weight of $\Interpr{I}$ is at least $k$ (feasible in polynomial time), and that $\Interpr{I}$ satisfies $\Psi(\VarSet{y})$ (feasible in $\Oracle{\NExpTime}{\SigmaP{c-1}}$, see~\cite{Lohrey2012,Luck2016-techrep});
like in the membership proof of \zcref{theo_complexity_LexMax}, the size of the interpretation $\Interpr{I}$ is only polynomial in the size of $\Psi(\VarSet{y})$, because here $\Interpr{I}$ is a truth assignment for the Boolean propositional variables $\VarSet{y}$, and not, like in \zcref{theo_complexity_LexMaxFunc}, a function instantiation of Boolean function variables.

Let us now focus on \MaxSatSigmaFormula{c}, and let $\Phi(\VarSet{y})$ be a \SigmaKFormula{c}.
An $\ExpTime$ oracle machine $\ParOracle{\Machine{M}}{?}$, aided by a $\Oracle{\NPTime}{\SigmaP{c-1}}$ oracle for the problem $\Language{A}$, can decide \MaxSatSigmaFormula{c} as follows.
First, $\Machine{M}$ prepares \emph{exponentially\nbdash-padded} queries $\pair{\Phi(\VarSet{y}),k}$, for all $0 \leq k \leq \SetSize{\VarSet{y}}$, to ask in \emph{parallel} to its oracle whether $\Phi(\VarSet{y})$ has a model of Hamming weight (at least) $k$.
Notice that the $\Oracle{\NPTime}{\SigmaP{c-1}}$ oracle can actually answer the query, because it receives a query that is exponentially\nbdash-long in the input of $\Machine{M}$.
By looking at the maximum value $k$ for which the query $\pair{\Phi(\VarSet{y}),k}$ has received a \yesansw, $\Machine{M}$ can accordingly answer depending on whether $k$ is odd/even.%
\footnote{An $\LogOracle{\ExpTime}{\SigmaP{c}}$ procedure for \MaxSatSigmaFormula{c} could compute the maximum Hamming weight of a model of $\Phi(\VarSet{y})$ via binary search.}

(\emph{Hardness}).
We now prove that \MaxSatSigmaFormula{c} is $\BoundedOracle{\ExpTime}{\SigmaP{c}}{\LogFunctions}$\HardSuffix by showing that there exists a polynomial reduction from every language $\Language{L} \in \BoundedOracle{\ExpTime}{\SigmaP{c}}{\LogFunctions}$ to \MaxSatSigmaFormula{c}.
Below we show the $\BoundedOracle{\ExpTime}{\SigmaP{c}}{\LogFunctions}$\HardSuffix{}ness of deciding whether the Hamming weight of the maximum\nbdash-Hamming\nbdash-weight model of $\Phi(\VarSet{y})$ is odd.
The $\BoundedOracle{\ExpTime}{\SigmaP{c}}{\LogFunctions}$\HardSuffix{}ness of deciding whether the Hamming weight of the maximum\nbdash-Hamming\nbdash-weight model of $\Phi(\VarSet{y})$ is even follows from the fact that $\BoundedOracle{\ExpTime}{\SigmaP{c}}{\LogFunctions}$ is closed under complement.

Let $\Language{L}$ be an arbitrary language in $\BoundedOracle{\ExpTime}{\SigmaP{c}}{\LogFunctions}$.
We prove $\Language{L} \KarpRed {}$\MaxSatSigmaFormula{c} by exhibiting a reduction that, for every string $w$, constructs a formula $\Phi^{\Language{L}}_{w}(\VarSet{y})$, with $\VarSet{y}$ a set of Boolean propositional variables, such that $\Phi^{\Language{L}}_{w}(\VarSet{y})$ is a \SimpleCompactSigmaKOddFormula{c} or a \SimpleCompactSigmaKEvenFormula{c}, depending on whether $c$ is odd or even, respectively, and $w \in \Language{L}$ iff the Hamming weight of the maximum\nbdash-Hamming\nbdash-weight model of $\Phi^{\Language{L}}_{w}(\VarSet{y})$ is odd.

Since $\Language{L} \in \BoundedOracle{\ExpTime}{\SigmaP{c}}{\LogFunctions}$, by \zcref{theo_summary_equivalence_intermediate_levels_Hausdorff} $\Language{L} \in \BoundedHausdCLASS{\PolFunctions}{\Oracle{\NExpTime}{\SigmaP{c-1}}}$ and hence $\Language{L}$ is characterized by a $\Oracle{\NExpTime}{\SigmaP{c-1}}$ Hausdorff predicate $\Language{D}$ of polynomial length.
By \zcref{theo_coding_nexp_machine_for Hausdorff_predicate}, there are polynomials $\PolynomialBoundingEverything(n)$ and $\PolynomialLengthHausPredD(n)$, with $\PolynomialBoundingEverything(n) \geq \PolynomialLengthHausPredD(n)$, such that $\PolynomialLengthHausPredD(n) - 1$ bounds the length of $\Language{D}$ and $\FormulaTMHausdorffWithP(\mi{ind}^{\PolynomialLengthHausPredD(\StringLength{w})})$ encodes whether $\Language{D}(w,z) = 1$, when $\mi{ind}^{\PolynomialLengthHausPredD(\StringLength{w})}$ is interpreted as the bit\nbdash-indexing function of $z$ over $2^{\PolynomialLengthHausPredD(\StringLength{w})}$ bits.

Let us now focus on $\Phi^{\Language{L}}_{w}(\VarSet{y})$.
Take the set of propositional variables $\VarSet{y} = \set{y_{\PolynomialLengthHausPredD(\StringLength{w}){}-1},\linebreak[0]\dots,\linebreak[0]y_1}$;
the Hamming weight of assignments for these variables is between $0$ and $\PolynomialLengthHausPredD(\StringLength{w}){}-1$.
Recall that $\PolynomialLengthHausPredD(n) - 1$ bounds the length of $\Language{D}$, hence, for each string $w$, the Hausdorff index of $w$ \Wrt $\Language{D}$ is such that $\HausdIndex{w}{\Language{D}} \leq \PolynomialLengthHausPredD(\StringLength{w}) - 1$.

Consider first a subformula encoding the Hamming weight of an assignment to $\VarSet{y}$ in the bit\nbdash-indexing function $\mi{ind}^{\PolynomialLengthHausPredD(\StringLength{w})}$, so that it can be read within $\FormulaTMHausdorffWithP(\mi{ind}^{\PolynomialLengthHausPredD(\StringLength{w})})$.
We use ``$\mi{encode\mhyphen{}value\mhyphen{}ind}(\BooleanEncoding{k})$'' to denote the conjunction of literals encoding $k$ in the first $\lceil \log \PolynomialLengthHausPredD(\StringLength{w}) \rceil$ bits of $\mi{ind}^{\PolynomialLengthHausPredD(\StringLength{w})}$.
These bits suffice since the maximum Hamming weight is $\PolynomialLengthHausPredD(\StringLength{w}) - 1$.
For example, ``$\mi{encode\mhyphen{}value\mhyphen{}ind}(\BooleanEncoding{5})$'' denotes
\begin{multline*}
\text{``}\bigl(\lnot \mi{ind}^{\PolynomialLengthHausPredD(\StringLength{w})}(\BooleanEncoding{\lceil \log \PolynomialLengthHausPredD(\StringLength{w})\rceil - 1}) \land \dots \land {} \\
\lnot \mi{ind}^{\PolynomialLengthHausPredD(\StringLength{w})}(\BooleanEncoding{3}) \land \mi{ind}^{\PolynomialLengthHausPredD(\StringLength{w})}(\BooleanEncoding{2}) \land \lnot \mi{ind}^{\PolynomialLengthHausPredD(\StringLength{w})}(\BooleanEncoding{1}) \land \mi{ind}^{\PolynomialLengthHausPredD(\StringLength{w})}(\BooleanEncoding{0})\bigr)\text{''}.
\end{multline*}
This encoding has logarithmic length \Wrt~$\StringLength{w}$.

We now look at the subformula encoding the Hamming weight of $\VarSet{y}$ into the bit\nbdash-indexing function.
A part of this subformula imposes a necessary condition for the assignments for $\VarSet{y}$ to satisfy the formula, that is, if $y_i$ is assigned $\valtrue$, then also $y_{i-1}$ must be assigned $\valtrue$.
In this way, the Hamming weight of an assignment for $\VarSet{y}$ equals the largest index $i$ such that $y_i$ is assigned $\valtrue$.
Let $\VarSet{x} = \set{x_{\PolynomialLengthHausPredD(\StringLength{w})}, \dots, x_1}$ be an additional set of Boolean propositional variables.
The function $\mi{less}$ below is defined in the formula $\mu$ of \zcref{theo_SigmaFormula_arithmetic} and embedded in $\FormulaTMHausdorffWithP(\mi{ind}^{\PolynomialLengthHausPredD(\StringLength{w})})$ (see \zcref{theo_coding_nexp_machine_for Hausdorff_predicate}).
Since its arguments have $\PolynomialBoundingEverything(\StringLength{w})$ bits and $\PolynomialBoundingEverything(n) \geq \PolynomialLengthHausPredD(n)$, the extra most significant bits are set to zero.
\begin{align*}
\mi{hw\mhyphen{}enc}(\mi{ind}^{\PolynomialLengthHausPredD(\StringLength{w})},\VarSet{y},\VarSet{x}) \equiv {}
  &\bigwedge_{2 \leq k \leq \PolynomialLengthHausPredD(\StringLength{w}){} - 1} (y_k \rightarrow y_{k-1}) \land {} \\
  &(\lnot y_{1} \rightarrow \mi{encode\mhyphen{}value\mhyphen{}ind}(\BooleanEncoding{0}) ) \land {} \\
  &\bigwedge_{2 \leq k \leq \PolynomialLengthHausPredD(\StringLength{w}){} - 1}
     \bigl(
       (\lnot y_k \land y_{k-1}) \rightarrow \mi{encode\mhyphen{}value\mhyphen{}ind}(\BooleanEncoding{k-1})
     \bigr) \land {} \\
  &(y_{\PolynomialLengthHausPredD(\StringLength{w}){} - 1} \rightarrow \mi{encode\mhyphen{}value\mhyphen{}ind}(\BooleanEncoding{\PolynomialLengthHausPredD(\StringLength{w}){}-1})) \land {} \\
  &\bigl(\lnot \mi{less}(\VarSet{x},\BooleanEncoding{\lceil \log \PolynomialLengthHausPredD(\StringLength{w}) \rceil}) \rightarrow \lnot \mi{ind}^{\PolynomialLengthHausPredD(\StringLength{w})}(\VarSet{x})\bigr).
\end{align*}
Since the variables $\VarSet{x}$ will be quantified universally in the final formula, the part ``$\bigl(\lnot \mi{less}(\VarSet{x},\BooleanEncoding{\lceil \log \PolynomialLengthHausPredD(\StringLength{w}) \rceil}) \rightarrow \lnot \mi{ind}^{\PolynomialLengthHausPredD(\StringLength{w})}(\VarSet{x})\bigr)$'' aims at completing the value encoded in the bit\nbdash-indexing function by setting to `$0$' in $\mi{ind}^{\PolynomialLengthHausPredD(\StringLength{w})}$ all bits of $z$ in positions with index strictly greater than $\lceil \log \PolynomialLengthHausPredD(\StringLength{w}) \rceil{}-1$.
This formula, as it is, is not in \CNF, as ``$\mi{encode\mhyphen{}value\mhyphen{}ind}$'' is a replacement for a conjunction of literals.
However, this formula can be rewritten in \CNF with polynomially many constant\nbdash-size clauses.

The overall formula is:
\begin{equation*}
  \Phi^{\Language{L}}_{w}(\VarSet{y}) \equiv 
    (\exists \mi{ind}^{\PolynomialLengthHausPredD(\StringLength{w})}) (\forall \VarSet{x}) \PrefixMatrixSeparator \mi{hw\mhyphen{}enc}(\mi{ind}^{\PolynomialLengthHausPredD(\StringLength{w})},\VarSet{y},\VarSet{x}) \land 
    \FormulaTMHausdorffWithP(\mi{ind}^{\PolynomialLengthHausPredD(\StringLength{w})}).
\end{equation*}

Observe that $\Phi^{\Language{L}}_{w}(\VarSet{y})$ can be built from $w$ in polynomial time and is in \CNF, as $\FormulaTMHausdorffWithP(\mi{ind}^{\PolynomialLengthHausPredD(\StringLength{w})})$ is in \CNF and $\mi{hw\mhyphen{}enc}$ can be rewritten in \CNF (see above).
By \zcref{theo_SigmaFormula_substitute_multiple_functions_with_one} and simple prenex rewriting, we can obtain in polynomial time an equivalent \CNF \SimpleCompactSigmaKOddFormula{c} or \SimpleCompactSigmaKEvenFormula{c}, according to whether $c$ is odd or even, respectively.
In the latter case, the propositional variables $\VarSet{x}$ are converted to function variables of arity zero and moved to the last second\nbdash-order universal quantifier.
Furthermore, the variables $\VarSet{y}$ do not occur as arguments of any function variable.
Finally, $\Phi^{\Language{L}}_{w}(\VarSet{y})$ is satisfiable, because $\Language{D}(w,0) = 1$ for every string $w$, and hence the interpretation assigning $\valfalse$ to all variables in $\VarSet{y}$ is a model.

It remains to prove that $w \in \Language{L}$ iff the maximum Hamming weight of a model of $\Phi^{\Language{L}}_{w}(\VarSet{y})$ is odd.
By $\mi{hw\mhyphen{}enc}$, every model $\Interpr{I}$ satisfies
$\EvalInterpr{y_i}{\Interpr{I}}=\valtrue \Rightarrow \EvalInterpr{y_{i-1}}{\Interpr{I}}=\valtrue$
for $2 \leq i \leq \PolynomialLengthHausPredD(\StringLength{w})-1$.
Hence, $\HammingWeight{\Interpr{I}}$ is the largest $i$ such that $\EvalInterpr{y_i}{\Interpr{I}}=\valtrue$, or zero if $\EvalInterpr{y_1}{\Interpr{I}}=\valfalse$.
The subformula $\mi{hw\mhyphen{}enc}$ encodes this value in $\mi{ind}^{\PolynomialLengthHausPredD(\StringLength{w})}$.
Together with \zcref{theo_coding_nexp_machine_for Hausdorff_predicate}, this implies that the maximum Hamming weight is the Hausdorff index of $w$ with respect to $\Language{D}$.
Thus, $w \in \Language{L}$ iff the maximum Hamming weight is odd.

We conclude by showing that the hardness of the problem is preserved even under the additional assumption that, for every two interpretations $\Interpr{I}$ and $\Interpr{J}$ of $\VarSet{y}$ such that $\HammingWeight{\Interpr{I}} > \HammingWeight{\Interpr{J}}$, it holds that $\Interpr{I} \models \Phi(\VarSet{y}) \Rightarrow \Interpr{J} \models \Phi(\VarSet{y})$.
Indeed, every model $\Interpr{I}$ of $\Phi^{\Language{L}}_{w}(\VarSet{y})$ assigns $\valtrue$ precisely to an initial segment $y_1,\dots,y_k$ of variables from $\VarSet{y}$, for some $0\leq k\leq\PolynomialLengthHausPredD(\StringLength{w})-1$, and thus has Hamming weight $k$.
The subformula $\mi{hw\mhyphen{}enc}$ encodes $k$ in $\mi{ind}^{\PolynomialLengthHausPredD(\StringLength{w})}$, while $\FormulaTMHausdorffWithP(\mi{ind}^{\PolynomialLengthHausPredD(\StringLength{w})})$ encodes the Hausdorff predicate $\Language{D}$.
Since $\Language{D}(w,z-1)\geq\Language{D}(w,z)$ for every $z>1$, if the model of weight $k$ satisfies the formula, so does every interpretation of smaller weight.
\end{proof}

\subsection%
{Proofs for \zcref{sec_datalog_hard_problems}}
\label{sec_detailed_proofs_datalog_hard_problems}
\silentbookmark[2]{Proofs for Section~\ref*{sec_datalog_hard_problems}}

\getkeytheorem{EncodingQBSFSatisfactionByTGDs}

\begin{proof}
Let $\Phi(\VarSet{y}) = (\exists f)(\forall \VarSet{x}) \PrefixMatrixSeparator \phi(f,\VarSet{x},\VarSet{y})$ be as in the statement, and assume that $\phi$ is a \CNF formula with clauses $c_1,\dots,c_q$.
Let $a$ be the arity of $f$, and assume \Wlog $a \leq n$.
We describe the construction of the ontology $\mathcal M_\Phi$.

\medbreak

\noindent
\emph{(The database).}
First, the database contains facts for the two Boolean values:
\[
\mi{BVal}(\DBvaltrue)
\qquad\text{and}\qquad
\mi{BVal}(\DBvalfalse).
\]
Then, we have a base fact allowing us to jump start the construction of all objects associated with all binary sequences of lengths up to $n$; %
more specifically, there is the base fact associated with the ``empty'' sequence: 
\[
\mi{BinSeq}_0(\epsilon,\epsilon,\epsilon),
\]
where $\epsilon$ is a constant identifying the ``empty'' binary sequence, i.e., the length\nbdash-$0$ sequence.

The database also contains facts associated with the two Boolean functions of arity $0$;
these facts jump start the construction of all objects associated with all possible instantiations of Boolean functions of arity up to $a$: %
\[
\mi{Func}_0(f_0,f_0,f_0), \mi{OutputFunc}_0(f_0,\epsilon,\DBvalfalse)
\qquad \text{and} \qquad
\mi{Func}_0(f_1,f_1,f_1), \mi{OutputFunc}_0(f_1,\epsilon,\DBvaltrue).
\]
Here $f_0$ and $f_1$ are constants identifying the constant\nbdash-false and constant\nbdash-true functions of arity $0$, respectively.
The fact $\mi{OutputFunc}_0(f_1,\epsilon,\DBvaltrue)$ states that the function of arity $0$ identified by $f_1$, when receiving in input the empty sequence, evaluates to $\valtrue$.
The natural meaning is associated also with $\mi{OutputFunc}_0(f_0,\epsilon,\DBvalfalse)$.

\smallskip

\noindent
\emph{(The program).}
A first group of TGDs generates identifiers for all binary sequences of length up to $n$.
Intuitively, from a binary sequence of length $i-1$ we can generate two sequences of length $i$ by appending in last position either $\valtrue$ or $\valfalse$.
Hence, we have rules tracking which sequences have been obtained from which sequences, by relating the identifiers and the Boolean values that have been used in the extension. 
In particular, a rule keeps track of the fact that a length-$i$ sequence identified by $J$ has been obtained from a length-$(i{-}1)$ sequence identified by $I$ by appending the Boolean value $V$.
More specifically, for every $1 \leq i\leq n$, we have a rule:
\[
\mi{BinSeq}_{i-1}(I,\_,\_), \mi{BVal}(V)
\rightarrow
(\exists J) \PrefixMatrixSeparator \mi{BinSeq}_i(J,I,V).
\]

Another group of TGDs associates with the identifiers of the sequences the actual bit values characterizing the sequences.
In particular, there are rules stating that, if a length-$i$ sequence (identified by) $J$ has been obtained from the length-$(i{-}1)$ sequence (identified by) $I$ by appending the Boolean value $V$, then:
\begin{itemize}[nosep,label=--]
  \item $J$ has in position $p_i$ the bit value $V$, where `$p_i$' is a constant; and
  \item if $I$ has in position $P$ the truth value $T$, then also $J$ has in position $P$ the truth value $T$.
\end{itemize}
More specifically, for every $1 \leq i \leq n$, we have:
\begin{align*}
\mi{BinSeq}_i(J,\_,V)
&\rightarrow
\mi{BitValue}_{i}(J,p_i,V),\\
\mi{BinSeq}_i(J,I,\_),\ \mi{BitValue}_{i-1}(I,P,T)
&\rightarrow
\mi{BitValue}_{i}(J,P,T).
\end{align*}

The identifiers of the binary sequences of length $n$ and those of length $a$ are used to identify the assignments for the variables $\VarSet{x}$ and the binary sequences that can be input for the function $f$ (which is of arity $a$), respectively.
We add these rules only to rename predicates and make their later use cleaner.
\begin{align*}
\mi{BinSeq}_n(I,\_,\_)
&\rightarrow
\mi{AssignX}(I),\\
\mi{AssignX}(I), \mi{BitValue}_{n}(I,p_i,T)
&\rightarrow
\mi{XVal}(I,x_i,T), && \text{for every $1 \leq i \leq n$,}\\
\mi{BinSeq}_a(I,\_,\_)
&\rightarrow
\mi{InputFunc}(I),\\
\mi{InputFunc}(I), \mi{BitValue}_{a}(I,p_i,T)
&\rightarrow
\mi{InputFuncVal}(I,p_i,T), && \text{for every $1 \leq i \leq a$.}
\end{align*}
Thus, $\mi{XVal}(I,x_i,T)$ states that the assignment for $\VarSet{x}$ identified by $I$ assigns $T$ to $x_i$, where `$x_i$' is a constant, and $\mi{InputFuncVal}(I,p_i,T)$ states that the input for $f$ identified by $I$ has value $T$ in position $p_i$.

Then, there are TGDs to generate identifiers for all possible instantiations of the Boolean function $f$.
These rules are based on the Shannon decomposition of Boolean functions:
every Boolean function of arity $i$ is uniquely determined by the two Boolean functions of arity $i-1$ obtained by fixing the last argument to $\valfalse$ and $\valtrue$, respectively.
Thus, there is a rule stating that, if $I$ and $J$ identify two functions of arity $i-1$, then there is a function of arity $i$ identified by $N$ obtained by coupling the functions identified by $I$ and $J$.
More specifically, for every $1 \leq i \leq a$, we have:
\[
\mi{Func}_{i-1}(I,\_,\_), \mi{Func}_{i-1}(J,\_,\_)
\rightarrow
(\exists N) \PrefixMatrixSeparator \mi{Func}_i(N,I,J).
\]

Moreover, once we know that the arity-$i$ function (identified by) $N$ has been obtained from the arity-$(i{-}1)$ functions (identified by) $I$ and $J$, we can also generate the instantiation of $N$ by combining the instantiations $I$ and $J$ as follows:
if the last bit in input to $N$ is $\valfalse$ (resp., $\valtrue$), then $N$ behaves as $I$ (resp., $J$).
This is obtained by a rule stating that:
if $\mi{In}$ identifies a length-$(i{-}1)$ binary sequence, and $\mi{In}_0$ and $\mi{In}_1$ identify the two length-$i$ binary sequences obtained by appending to $\mi{In}$ the values $\valfalse$ and $\valtrue$, respectively, then
the output of $N$ on input $\mi{In}_0$ (resp., $\mi{In}_1$) equals the output of $I$ (resp., $J$) on input $\mi{In}$.
Hence, for every $1 \leq i \leq a$, we have the rule:
\begin{multline*}
\mi{Func}_i(N,I,J),\
\mi{BinSeq}_{i-1}(\mi{In},\_,\_),
\mi{OutputFunc}_{i-1}(I,\mi{In},V_I), \mi{OutputFunc}_{i-1}(J,\mi{In},V_J),\\
\mi{BinSeq}_i(\mi{In}_0,\mi{In},\DBvalfalse),
\mi{BinSeq}_i(\mi{In}_1,\mi{In},\DBvaltrue)\\
\rightarrow
\mi{OutputFunc}_i(N,\mi{In}_0,V_I),
\mi{OutputFunc}_i(N,\mi{In}_1,V_J).
\end{multline*}

Finally, we expose the arity-$a$ functions.
As above, these rules just rename the predicates to tidy up their use.
\begin{align*}
\mi{Func}_a(I,\_,\_)
&\rightarrow
\mi{Func}(I),\\
\mi{OutputFunc}_a(I,\mi{In},V)
&\rightarrow
\mi{OutputFunc}(I,\mi{In},V).
\end{align*}

Also, we have TGDs to encode satisfaction of the clauses of $\phi$.
For each clause $c_r$, we have a predicate
\(
\mi{Cl}_r\mi{Sat}(\mi{IdX},\mi{IdY},\mi{IdF}),
\)
whose intended meaning is that the clause $c_r$ is satisfied under the assignment to $\VarSet{x}$ identified by $\mi{IdX}$, the assignment to $\VarSet{y}$ identified by $\mi{IdY}$, and the instantiation of the function $f$ identified by $\mi{IdF}$.
For every positive literal $x_i$ or $y_i$ (resp., negative literal $\lnot x_i$ or $\lnot y_i$) occurring in $c_r$, we add the rule:
\begin{alignat*}{3}
x_i&\colon & \mi{AssignX}(\mi{IdX}), \mi{AssignY}(\mi{IdY}), \mi{Func}(\mi{IdF}), \mi{XVal}(\mi{IdX},x_i,\DBvaltrue)
&\rightarrow
\mi{Cl}_r\mi{Sat}(\mi{IdX},\mi{IdY},\mi{IdF}),\\
\lnot x_i&\colon & \mi{AssignX}(\mi{IdX}), \mi{AssignY}(\mi{IdY}), \mi{Func}(\mi{IdF}), \mi{XVal}(\mi{IdX},x_i,\DBvalfalse)
&\rightarrow
\mi{Cl}_r\mi{Sat}(\mi{IdX},\mi{IdY},\mi{IdF}),\\
y_i&\colon & \mi{AssignX}(\mi{IdX}), \mi{AssignY}(\mi{IdY}), \mi{Func}(\mi{IdF}), \mi{YVal}(\mi{IdY},y_i,\DBvaltrue)
&\rightarrow
\mi{Cl}_r\mi{Sat}(\mi{IdX},\mi{IdY},\mi{IdF}),\\
\lnot y_i&\colon & \mi{AssignX}(\mi{IdX}), \mi{AssignY}(\mi{IdY}), \mi{Func}(\mi{IdF}), \mi{YVal}(\mi{IdY},y_i,\DBvalfalse)
&\rightarrow
\mi{Cl}_r\mi{Sat}(\mi{IdX},\mi{IdY},\mi{IdF}),
\end{alignat*}
which allow us to test whether the clause $c_r$ is satisfied by the Boolean assignment to $x_i$ (resp., $y_i$) given by the assignment for $\VarSet{x}$ (resp., $\VarSet{y}$) identified by $\mi{IdX}$ (resp., $\mi{IdY}$).
Remember that the \AssignYFacts and the $\mi{YVal}$\nbdash-facts are ``externally'' provided to the ontology.

Finally, consider a literal involving the function variable $f$, say $f(z_1,\dots,z_a)$ or its negation, where each $z_j$ is either a variable from $\VarSet{x}$, a variable from $\VarSet{y}$, or a Boolean constant.
For example, for a negative literal $\neg f(x_{i_1},\dots,x_{i_a})$ we add the rule:
\begin{multline*}
\mi{AssignX}(\mi{IdX}), \mi{AssignY}(\mi{IdY}),\mi{Func}(\mi{IdF}),\mi{InputFunc}(\mi{In}),\\ 
\mi{XVal}(\mi{IdX},x_{i_1},T_1),\dots,\mi{XVal}(\mi{IdX},x_{i_a},T_a),\\
\mi{InputFuncVal}(\mi{In},p_1,T_1),\dots,\mi{InputFuncVal}(\mi{In},p_a,T_a),\\
\mi{OutputFunc}(\mi{IdF},\mi{In},\DBvalfalse)
\rightarrow
\mi{Cl}_r\mi{Sat}(\mi{IdX},\mi{IdY},\mi{IdF}).
\end{multline*}
In the general case, each atom $\mi{XVal}(\mi{IdX},x_{i_j},T_j)$ is replaced by the corresponding atom $\mi{YVal}(\mi{IdY},y_\ell,T_j)$, or by fixing $T_j$ to $\DBvaltrue$ or $\DBvalfalse$, whenever the argument is a $y$\nbdash-variable or a Boolean constant.

To test satisfaction of the whole matrix $\phi$, we have a TGD collecting the satisfaction of each clause of $\phi$ into a matrix\nbdash-satisfaction predicate:
\[
\mi{Cl}_1\mi{Sat}(\mi{IdX},\mi{IdY},\mi{IdF}),\dots,\mi{Cl}_q\mi{Sat}(\mi{IdX},\mi{IdY},\mi{IdF})
\rightarrow
\mi{Sat}(\mi{IdX},\mi{IdY},\mi{IdF}).
\]
Notice again that all these TGDs test the satisfaction of the clauses of $\phi$, and of $\phi$, under the assignment for $\VarSet{x}$ identified by $\mi{IdX}$, the assignment for $\VarSet{y}$ identified by $\mi{IdY}$, and the instantiation of $f$ identified by $\mi{IdF}$.

It remains to show the TGDs allowing us to test whether, for a given instantiation $\Interpr{F}$ of $f$ and a given assignment $\Interpr{I}$ for $\VarSet{y}$, the formula $(\forall \VarSet{x}) \PrefixMatrixSeparator \phi(\EvalInterpr{f}{\Interpr{F}},\VarSet{x},\EvalInterpr{\VarSet{y}}{\Interpr{I}})$ is valid (notice the universal quantification over $\VarSet{x}$).
To this aim, we aggregate satisfaction information, related to the various assignments for $\VarSet{x}$, bottom\nbdash-up along the binary tree used to generate the binary sequences associated with the assignments for $\VarSet{x}$.

For every $0 \leq i \leq n$, we have a predicate
\(
  \mi{AllSat}_i(I,\mi{IdY},\mi{IdF}),
\)
whose intended meaning is that \emph{all} assignments for $\VarSet{x}$ whose prefix is the length-$i$ binary sequence identified by $I$ satisfy the matrix $\phi$ of the formula, under the function instantiation identified by $\mi{IdF}$, and the assignment for $\VarSet{y}$ identified by $\mi{IdY}$.

We start with the following rule dealing with the leaves of the binary tree.
If $\mi{IdX}$ identifies a binary sequence of length $n$ that, under the function instantiation (identified by) $\mi{IdF}$, and the assignment for $\VarSet{y}$ (identified by) $\mi{IdY}$, satisfies the matrix $\phi$, we can trivially conclude that all binary sequences extending $\mi{IdX}$ satisfy $\phi$ under the same conditions, as a length-$n$ sequence is already a complete assignment for $\VarSet{x}$.
This is obtained via the rule:
\[
  \mi{Sat}(\mi{IdX},\mi{IdY},\mi{IdF})
  \rightarrow
  \mi{AllSat}_n(\mi{IdX},\mi{IdY},\mi{IdF}).
\]

Then, we have rules dealing with the mid layers of the tree.
Intuitively, there is a rule stating that, given a length-$(i{-}1)$ binary sequence (identified by) $I$, if $I_0$ and $I_1$ identify the two length-$i$ sequences obtained by extending $I$ with one bit, and all assignments for $\VarSet{x}$ with binary prefixes identified by $I_0$ and $I_1$ satisfy $\phi$ under the function instantiation $\mi{IdF}$, and the assignment $\mi{IdY}$ for $\VarSet{y}$, then we can conclude that all assignments for $\VarSet{x}$ whose prefix is $I$ satisfy $\phi$ under the same conditions. 
The rules, for every $1 \leq i \leq n$, are:
\begin{multline*}
  \mi{BinSeq}_i(I_0,I,\DBvalfalse), \mi{BinSeq}_i(I_1,I,\DBvaltrue),\\
  \mi{AllSat}_i(I_0,\mi{IdY},\mi{IdF}), \mi{AllSat}_i(I_1,\mi{IdY},\mi{IdF})
  \rightarrow
  \mi{AllSat}_{i-1}(I,\mi{IdY},\mi{IdF}).
\end{multline*}

Finally, a TGD captures the validity of $(\forall \VarSet{x}) \PrefixMatrixSeparator \phi(f,\VarSet{x},\VarSet{y})$.
When $\mi{AllSat}_0(\epsilon,\mi{IdY},\mi{IdF})$ holds true, it means that all assignments for $\VarSet{x}$ whose prefix is the empty sequence, and hence all the assignments for $\VarSet{x}$, satisfy $\phi$ under the function instantiation $\mi{IdF}$, and the assignment $\mi{IdY}$ for $\VarSet{y}$.
Hence, we can conclude that the interpretation $\Interpr{I}$ for $\VarSet{y}$ identified by $\mi{IdY}$ is a model of $\Phi(\VarSet{y})$ and that the function instantiation identified by $\mi{IdF}$ is a witness of the validity of $(\exists f)(\forall \VarSet{x}) \PrefixMatrixSeparator \phi(f,\VarSet{x},\EvalInterpr{\VarSet{y}}{\Interpr{I}})$.
We express this in the fact $\mi{ModelWitness}(\mi{IdY},\mi{IdF})$ meaning that the instantiation of the function $f$ identified by $\mi{IdF}$ satisfies $\phi$ for every assignment to the universally quantified variables $\VarSet{x}$, under the assignment to the free variables $\VarSet{y}$ identified by $\mi{IdY}$.
This is encoded in the rule:
\[
  \mi{AllSat}_0(\epsilon,\mi{IdY},\mi{IdF})
  \rightarrow
  \mi{ModelWitness}(\mi{IdY},\mi{IdF}).
\]

\smallskip

We conclude the construction by observing that the ontology is acyclic because the predicates
$\mi{BinSeq}_i$,
$\mi{BitValue}_i$,
$\mi{Func}_i$,
$\mi{OutputFunc}_i$,
and
$\mi{AllSat}_i$,
are stratified by the index $i$, and can be ordered by strata.
Furthermore, predicate arities are bounded by a constant.

\medbreak

\noindent
We now prove the correctness of the construction.

First, by induction on $i$, the facts $\mi{BinSeq}_i$ identify exactly all the possible binary sequences of length $i$, and the facts $\mi{BitValue}_i$ correctly record their bits.
The base case $i = 0$ follows from the fact $\mi{BinSeq}_0(\epsilon,\epsilon,\epsilon)$.
For the induction step, each length\nbdash-$(i{-}1)$ sequence is extended once with $\valfalse$ and once with $\valtrue$, and the bit\nbdash-propagation rules copy the previous bits while recording the newly appended bit.
Therefore, the $\mi{AssignX}$\nbdash-objects are exactly associated with the assignments for $\VarSet{x}$, and the $\mi{InputFunc}$\nbdash-objects are exactly associated with the binary input sequences of the function $f$.

Analogously, again by induction on $i$, the facts $\mi{Func}_i$ identify exactly all the possible instantiations of Boolean functions of arity $i$, and $\mi{OutputFunc}_i$ encodes their truth tables.
For $i = 0$, the database contains exactly the constant\nbdash-false and constant\nbdash-true functions.
For the induction step, the Shannon decomposition gives a bijection between Boolean functions of arity $i$ and ordered pairs of Boolean functions of arity $i-1$.
The rule generating $\mi{Func}_i(N,I,J)$ realizes this pairing, and the output\nbdash-propagation rule defines the truth table of the function identified by $N$ from the truth tables of the functions identified by $I$ and $J$, according to the Shannon decomposition.
Hence, the facts $\mi{Func}(\mi{IdF})$ identify exactly the Boolean functions of arity $a$, and $\mi{OutputFunc}(\mi{IdF},\mi{In},V)$ correctly records their values.

Now, let $E_\Interpr{I}$ be a set of facts ``properly encoding'', via the facts $\mi{AssignY}$ and $\mi{YVal}$, a single interpretation $\Interpr{I}$ of the variables $\VarSet{y}$;
assume that the object identifying $\Interpr{I}$ is $\mi{id}_\Interpr{I}$, i.e., $\mi{AssignY}(\mi{id}_\Interpr{I}) \in E_{\Interpr{I}}$.
Let $\mi{id}_{\bar x}$ be an identifier for an assignment for the variables $\VarSet{x}$ and let $\mi{id}_{\Interpr{F}}$ be an identifier for an instantiation $\Interpr{F}$ of the function $f$.
By construction of the clause rules, for each clause $c_r$, the fact
\(
\mi{Cl}_r\mi{Sat}(\mi{id}_{\bar x},\mi{id}_\Interpr{I},\mi{id}_{\Interpr{F}})
\)
is derived from $\KBDef{\DB_\Phi \cup E_\Interpr{I}, \Dep_\Phi}$ iff at least one literal of $c_r$ is true under the assignment for $\VarSet{x}$ identified by $\mi{id}_{\bar x}$, the assignment for $\VarSet{y}$ identified by $\mi{id}_\Interpr{I}$, and the instantiation of the function $f$ identified by $\mi{id}_{\Interpr{F}}$.
Therefore,
\(
\mi{Sat}(\mi{id}_{\bar x},\mi{id}_\Interpr{I},\mi{id}_{\Interpr{F}})
\)
is derived iff $\phi$ is true under these interpretations of the variables.
It remains to prove that $\KBDef{\DB_\Phi \cup E_{\Interpr{I}},\Dep_\Phi} \models \mi{ModelWitness}(\mi{id}_{\Interpr{I}},\mi{id}_{\Interpr{F}})$ iff $\Interpr{I} \models \Phi(\VarSet{y})$ and the instantiation of $f$, identified by $\mi{id}_{\Interpr{F}}$, witnesses the validity of $(\exists f)(\forall \VarSet{x}) \PrefixMatrixSeparator  \phi(f,\VarSet{x},\EvalInterpr{\VarSet{y}}{\Interpr{I}})$.

We claim that the predicates $\mi{AllSat}_i$ correctly implement the universal quantification over $\VarSet{x}$.
Indeed, let $\Interpr{I}$, $E_{\Interpr{I}}$, $\mi{id}_{\Interpr{I}}$, and $\mi{id}_{\Interpr{F}}$, be as above.
We prove by downward induction on $i$ that, for every length-$i$ binary\nbdash-sequence (identified by) $I$, the fact $\mi{AllSat}_i(I,\mi{id}_{\Interpr{I}},\mi{id}_{\Interpr{F}})$ is derived from $\KBDef{\DB_\Phi \cup E_{\Interpr{I}},\Dep_\Phi}$ iff, under the given interpretations of $f$ and $\VarSet{y}$, every assignment for $\VarSet{x}$ having $I$ as prefix satisfies $\phi$.

The base case $i = n$ follows from the rule `$\mi{Sat}(\mi{IdX},\mi{IdY},\mi{IdF}) \rightarrow \mi{AllSat}_n(\mi{IdX},\mi{IdY},\mi{IdF})$', and the fact that $\mi{Sat}(\mi{id}_{\bar x},\mi{id}_\Interpr{I},\mi{id}_{\Interpr{F}})$ is derived iff $\phi$ is true under the considered interpretations of the variables.

For the induction step, the rule defining $\mi{AllSat}_{i-1}$ derives $\mi{AllSat}_{i-1}(I,\mi{id}_{\Interpr{I}},\mi{id}_{\Interpr{F}})$
iff the corresponding facts $\mi{AllSat}_{i}(I_0,\mi{id}_{\Interpr{I}},\mi{id}_{\Interpr{F}})$ and $\mi{AllSat}_{i}(I_1,\mi{id}_{\Interpr{I}},\mi{id}_{\Interpr{F}})$
hold for the two children $I_0$ and $I_1$ of $I$, where $I_0$ and $I_1$ identify the binary sequences obtained from the sequence (identified by) $I$ by appending $\valfalse$ and $\valtrue$, respectively.
Since every assignment for $\VarSet{x}$ extending $I$ extends either $I_0$ or $I_1$, the claim follows.

Thus, $\mi{AllSat}_0(\epsilon,\mi{id}_{\Interpr{I}},\mi{id}_{\Interpr{F}})$ is derived iff $(\forall \VarSet{x}) \PrefixMatrixSeparator \phi(\EvalInterpr{f}{\Interpr{F}},\VarSet{x},\EvalInterpr{\VarSet{y}}{\Interpr{I}})$ is valid.
Hence, $\mi{ModelWitness}(\mi{id}_{\Interpr{I}},\mi{id}_{\Interpr{F}})$ is entailed iff $\Interpr{I} \models \Phi(\VarSet{y})$, and $\mi{id}_{\Interpr{F}}$ identifies an instantiation of $f$ witnessing the validity of $(\exists f)(\forall \VarSet{x}) \PrefixMatrixSeparator \phi(f,\VarSet{x},\EvalInterpr{\VarSet{y}}{\Interpr{I}})$.

To conclude, observe that the argument above is given for a single interpretation $\Interpr{I}$ of $\VarSet{y}$.
However, if $E_{\mathscr I}$ ``properly encodes'' a set
$\mathscr{I} = \set{\Interpr{I}_1,\dots,\Interpr{I}_p}$ of interpretations of $\VarSet{y}$, the same argument applies independently to each identifier $\mi{id}_j$ associated with $\Interpr{I}_j$.
Indeed, all predicates whose derivation depends on the values assigned to $\VarSet{y}$ carry as an argument the TGD variable $\mi{IdY}$, which is instantiated with the object identifying the interpretation of $\VarSet{y}$.
Thus, facts encoding different interpretations cannot be mixed in a derivation of $\mi{ModelWitness}(\mi{id}_j,\mi{id}_{\Interpr{F}})$.
It follows that $\KBDef{\DB_\Phi \cup E_{\mathscr I},\Dep_\Phi} \models \mi{ModelWitness}(\mi{id}_j,\mi{id}_{\Interpr{F}})$ iff $\mi{id}_j$ identifies an interpretation
$\Interpr{I}_j\in\mathscr I$ that is a model of $\Phi(\VarSet{y})$, and $\mi{id}_{\Interpr{F}}$ identifies an instantiation of $f$ witnessing the validity of $(\exists f)(\forall \VarSet{x}) \PrefixMatrixSeparator \phi(f,\VarSet{x},\EvalInterpr{\VarSet{y}}{\Interpr{I}_j})$.
\end{proof}

We report below a different proof for \zcref{thm_minex-rel-card_A_ba_hardness}---it will be restated below as \zcref{thm_minex-rel-card_A_ba_hardness-RESTATED}. %

This different hardness proof rely on a reduction from a variant of the tiling problem, that we now recall.
A \emph{tiling system} $\TilingSyst{t} = \tup{T,V,H}$ consists of a set of tile-types $T = \{\tau_1,\dots,\tau_k\}$, together with the vertical and horizontal adjacency rules $V \subseteq T \times T$ and $H \subseteq T \times T$, respectively~\cite{Furer83}---the rules are represented by pairs of tile\nbdash-types that may be placed contiguously vertically or horizontally, respectively, on the surface.
A \emph{tiling} for $\TilingSyst{t}$ is a function $f\colon \{1,\dots,i,\dots\} \times \{1,\dots,j,\dots\} \mapsto T$ placing a tile of the specified type at position $(i,j)$ on the surface,%
\footnote{The finiteness of the function's domain depends on whether a finite or infinite surface has to be tiled.}
such that $\pair{f(i,j),f(i+1,j)} \subseteq H$ and $\pair{f(i,j),f(i,j+1)} \subseteq V$;
i.e., $f$ complies with the horizontal and vertical adjacency rules.
A tiling problem asks, for a tiling system, whether it admits a tiling covering the entirety of a given surface.
A tiling problem is said \emph{bounded} if the surface to cover is finite, \emph{unbounded} otherwise, and it is said \emph{constrained} if the covering must comply with some tiles already allocated (like, e.g., on the first row), \emph{unconstrained} otherwise.

The following problem is \NExpTimec~\cite{Furer83}:
for an instance $\tup{\TilingSyst{t},\tau,n}$, where $\TilingSyst{t} = \tup{T,V,H}$ is a tiling system, $\tau \in T$ is a tile\nbdash-type, and $n$ is an integer in \emph{unary} notation, decide whether $\TilingSyst{t}$ admits a tiling of the exponential square $2^n \times 2^n$ (i.e., a finite surface) that places a tile of type $\tau$ at the first position of the first row---the \NExpTimeh{}ness of this problem holds even when no initial tile\nbdash-type is imposed~\cite{Furer83}.

The \emph{Exponential Tiling Problem} (\ExpTilProb) \cite[Theorem~15]{EiterLP2016} and \cite[Lemma~3.2]{LukasiewiczMMMP2022} here considered is a variation of the one above, from which \ExpTilProb inherits its \NExpTimeh{}ness:
for an instance $\tup{\TilingSyst{t},s,n}$, where $\TilingSyst{t}$ is a tiling system, $s$ is an initial tiling condition, which is a sequence of tile\nbdash-types with $s[i]$ denoting the $i$\nbdash-th element of $s$, and $n$ is an integer in \emph{unary} notation, decide whether $\TilingSyst{t}$ admits a tiling of the exponential square $2^n \times 2^n$ that places a tile of type $s[i]$ at the $i$\nbdash-th position of the first row.

Given these preliminaries, we can now close the hardness results of the above mentioned problems, and we start by showing the \PNExpLogh{}ness of \MinExRelMinCard{}($\DLFrag{A}$).

\begin{theorem}
\label{thm_minex-rel-card_A_ba_hardness-RESTATED}
\getkeytheorem[body]{thmMinexRelACardBAHardness}
\end{theorem}

\begin{proof}
Consider the following problem: 
given a tuple $\StringTup{t} = \tup{t_1,\dots,t_m}$ of $m$ (independent) instances $t_i = \tup{\TilingSyst{t}_i,s_i,n_i}$ of \ExpTilProb, decide whether the number of \yesinsts of \ExpTilProb in $\StringTup{t}$ is odd.
By \zcref{theo_theta_level_equals_polHausdorff,theo_hardness_oddity_general}, this problem is complete for $\LogOracle{\ExpTime}{\NPTime} = \PNExpLog$, and the hardness holds even for tuples $\StringTup{t}$ such that, if $t_{i+1}$ is a \yesinst of \ExpTilProb, then also $t_i$ is a \yesinst of \ExpTilProb, for all $i \geq 1$. 

To prove the theorem, we exhibit a polynomial reduction from the problem of deciding whether the number of \yesinsts of \ExpTilProb in $\StringTup{t}$ is odd to \MinExRelMinCard{}($\DLFrag{A}$).
In particular, from $\StringTup{t}$, we build an instance $\tup{\KB,\Query,f}$ of \MinExRelMinCard{}($\DLFrag{A}$), where $\KB = \KBDef{\DB,\Dep}$ is a knowledge base, $\Query$ is a \BCQ, and $f\in \DB$ is a fact.

To encode the instances $t_i$ in $\KB$, we use the TGD and database encoding presented in \cite[Theorem~15]{EiterLP2016}, \cite[Lemma~23]{EiterTechRep2016}, and \cite[Lemma~3.2]{LukasiewiczMMMP2022}.
More specifically, we have $m$ sets of bounded\nbdash-arity acyclic TGDs $\Dep_{\TilingSyst{t}_i,n_i,|s_i|}$, for $i = 1, \dots, m$, to encode the structure of the tiling task, that is, these rules describe how a correct tiling looks like when it has to cover the exponential square $2^{n_i} \times 2^{n_i}$ and has an initial condition of \emph{length} $|s_i|$;
and there are $m$ sets of facts $\DB_{\TilingSyst{t}_i}$, for $i = 1, \dots, m$, to encode the adjacency rules of $\TilingSyst{t}_i$.
There are also $m$ sets of facts $\DB_{s_i}$, for $i = 1, \dots, m$, encoding the initial tiling conditions $s_i$.
The $m$ sets of TGDs and facts for each instance $t_i$ of \ExpTilProb are made disjoint by using disjoint predicates (we index them with a superscript~$i$);
that is, we have
$\Dep^i = {\Dep_{\TilingSyst{t}_i,n_i,|s_i|}}^{i}$, and $\DB^i = {\DB_{\TilingSyst{t}_i}}^{i} \cup {\DB_{s_i}}^{i}$, for $i = 1, \dots, m$.
The rules $\Dep^i$ and the database $\DB^i$, for $i = 1, \dots, m$, are such that the tiling system $\TilingSyst{t}_i$ admits a tiling of the exponential square $2^{n_i} \times 2^{n_i}$ with initial tiling condition $s_i$ iff $\DB^i$ entails the fact $\mi{yes}^i()$ via $\Dep^i$ \cite{EiterLP2016,EiterTechRep2016,LukasiewiczMMMP2022}.
By inspection of the reductions there proposed, the construction can easily be amended so that the fact entailed is $\mi{yes}(i)$;
e.g., this can be achieved by adding the constant `$i$' in the facts encoding the initial condition for $\TilingSyst{t}_i$, so that this constant can be propagated toward $\mi{yes}(i)$.
We can now provide the details of the reduction.

\smallskip

\noindent
\emph{(The database).}
The database is $\DB = \bigcup_{i = 1}^{m} \DB^i \cup \bigcup_{i = 1}^{m} \set{\mi{dummy\mhyphen{}yes}(i),\mi{aux}(i)} \cup \set{\mi{odd}()} \cup \set{\mi{even}(i) \mid 2 \leq i \leq m \land i \text{ is even}} \cup \set{\mi{succ}(i,i+1) \mid 1 \leq i \leq m-1}$, where $i$ and $i+1$ in the atoms are numerical constants.

\smallskip

\noindent
\emph{(The program).}
The program contains the TGDs $\Dep^i$, for all $1 \leq i \leq m$, plus the following rules:
\begin{align*}
\mi{dummy\mhyphen{}yes}(X) \land \mi{aux}(X) & \ra \mi{check}(X) \\
\mi{yes}(X) \land \mi{odd}() & \ra \mi{check}(X) \\
\mi{yes}(X) \land \mi{even}(X) & \ra \mi{check}(X) \\
\mi{yes}(X) \land \mi{even}(X) \land \mi{succ}(Y,X) & \ra \mi{check}(Y).
\end{align*}

\smallskip

\noindent
\emph{(The query).}
The query is $\Query = \DB^1 \land \dots \land \DB^m \land \mi{Even} \land \mi{Succ} \land \mi{check}(1) \land \dots \land \mi{check}(m)$, where with $\DB^i$ in the query we mean the conjunction of all database facts from $\DB^i$, and $\mi{Even}$ and $\mi{Succ}$ denote the conjunction of all the $\mi{even}$\nbdash-facts and $\mi{succ}$\nbdash-facts in the database, respectively. 

\smallskip

\noindent
\emph{(The distinguished fact).}
We take $f = \mi{odd}()$.

\smallskip

Observe that the program has bounded arity and is acyclic, that the reduction is such that $\KBDef{\DB,\Dep} \models \Query$, because $(\bigcup_{i = 1}^{m} \DB^i \cup \linebreak[0] \bigcup_{i = 1}^{m} \set{\mi{dummy\mhyphen{}yes}(i),\linebreak[0]\mi{aux}(i)} \cup \set{\mi{even}(i) \mid 2 \leq i \leq m \land i \text{ is even}} \cup \set{\mi{succ}(i,i+1) \mid 1 \leq i \leq m-1}) \subseteq \DB$, and the reduction can be computed in polynomial time.

\medskip

\noindent
We now prove that the number of \yesinsts of \ExpTilProb in $\StringTup{t}$ is odd iff there exists a $\leq$\nbdash-\Minex for $\KBDef{\DB,\Dep} \models \Query$ containing $f = \mi{odd}()$.
Let $\hat z$ denote the maximum index $i$ for which $t_i$ is a \yesinst of \ExpTilProb.
Since we are assuming that when $t_{i+1}$ is a \yesinst of \ExpTilProb, also $t_i$ is a \yesinst of \ExpTilProb, the number of \yesinsts of \ExpTilProb in $\StringTup{t}$ is odd iff $\hat z$ is odd (if none of the $t_i$ is \yesinst of \ExpTilProb, then we assume $\hat z = 0$).
We start by highlighting a property of minimal\nbdash-cardinality explanations.

\medbreak

\begin{subproperty}
\label{subprop_facts_in_explanations}
Let $\Expl$ be a $\leq$\nbdash-\Minex for $\Query$.
Then, for all $j$ with $1 \leq j\leq \hat z$, the facts $\mi{dummy\mhyphen{}yes}(j)$ and $\mi{aux}(j)$ do not belong to $\Expl$, instead, for all $j$ with $\hat z < j \leq m$, the facts $\mi{dummy\mhyphen{}yes}(j)$ and $\mi{aux}(j)$ belong to $\Expl$.
\end{subproperty}

\begin{subproof}
Let us consider the case in which $j \leq \hat z$.
By assumption on the instances $\StringTup{t}$, we have that $t_j$ is a \yesinst of \ExpTilProb.
There are two cases: either (a)~$j$ is even, or (b)~$j$ is odd.

Consider Case~(a).
Since $j$ is even, $\mi{even}(j)$ is in $\Expl$, for otherwise $\Expl$ would not entail $\Query$.
Because $t_j$ is a \yesinst of \ExpTilProb, the fact $\mi{check}(j)$ can be obtained via the TGD `$\mi{yes}(X) \land \mi{even}(X) \ra \mi{check}(X)$', as $\mi{yes}(j)$ is entailed by the TGDs in $\Sigma^j$.
Hence, $\mi{dummy\mhyphen{}yes}(j)$ and $\mi{aux}(j)$ are not needed in $\Expl$ to entail $\mi{check}(j)$.
Thus, they do not belong to $\Expl$, because $\Expl$ is a cardinality\nbdash-minimal explanation for $\Query$.

Consider now Case~(b).
There are two subcases: either (i)~$j < \hat z$, or (ii)~$j = \hat z$.

Let us focus on Case~(i).
Since $j < \hat z$, we have that $j + 1 \leq \hat z$ and $j + 1$ is even (as $j$ is odd).
By $j + 1 \leq \hat z$, $t_{j+1}$ is a \yesinst of \ExpTilProb.
Moreover, $\Expl$ must contain $\mi{even}(j+1)$ (as $j+1$ is even) and $\mi{succ}(j,j+1)$, for otherwise $\Expl$ would not imply $\Query$.
Thus, $\mi{check}(j)$ is entailed by the TGD `$\mi{yes}(X) \land \mi{even}(X) \land \mi{succ}(Y,X) \ra \mi{check}(Y)$', because $\mi{yes}(j+1)$ is entailed by the TGDs in $\Sigma^{j+1}$.
Hence, $\mi{dummy\mhyphen{}yes}(j)$ and $\mi{aux}(j)$ are not needed in $\Expl$, which therefore does not contain them (see above).

For Case~(ii), we have $j = \hat z$, hence we refer to the index $\hat z$.
Since $t_{\hat z}$ is a \yesinst of \ExpTilProb, $\mi{yes}(\hat z)$ is entailed by the TGDs in $\Sigma^{\hat z}$.
We now show that neither $\mi{dummy\mhyphen{}yes}(\hat z)$ nor $\mi{aux}(\hat z)$ is in $\Expl$.

Let us assume by contradiction that at least one of the two facts $\mi{dummy\mhyphen{}yes}(\hat z)$ or $\mi{aux}(\hat z)$ is in $\Expl$.
There are two cases: either exactly one of them is in $\Expl$, or both of them are in $\Expl$.
Remember that $\Expl$ is an explanation for $\Query$, hence $\KBDef{\Expl,\Dep} \models \Query$.
If only one of the two facts, let us name it $g$, is in $\Expl$, then $\mi{check}(\hat z)$ must be entailed by the TGD `$\mi{yes}(X) \land \mi{odd}() \ra \mi{check}(X)$'.
Thus, $g$, although in $\Expl$, is not involved in the entailment of $\mi{check}(\hat z)$, hence $g$ is not necessary in $\Expl$.
By this, $\Expl$ is not cardinality\nbdash-minimal: a contradiction.

If both $\mi{dummy\mhyphen{}yes}(\hat z)$ and $\mi{aux}(\hat z)$ are in $\Expl$, we can obtain from $\Expl$ a smaller explanation for $\Query$.
Indeed, let us define $\Expl' = (\Expl \setminus \set{\mi{dummy\mhyphen{}yes}(\hat z),\mi{aux}(\hat z)}) \cup \set{\mi{odd}()}$.
Since $\Expl \supseteq \set{\mi{dummy\mhyphen{}yes}(\hat z),\mi{aux}(\hat z)}$, we have that $\SetSize{\Expl'} \leq (\SetSize{\Expl} - 2) + 1$: a contradiction, as $\Expl$ is assumed to be cardinality\nbdash-minimal.

\medskip

Let us now consider the case in which $j > \hat z$.
By assumption on $\StringTup{t}$, the instance $t_j$ is a \noinst of \ExpTilProb.
Thus, $\mi{yes}(j)$ is not entailed by the TGDs $\Sigma^j$, and hence $\mi{dummy\mhyphen{}yes}(j)$ and $\mi{aux}(j)$ need to be in $\Expl$ in order for $\mi{check}(j)$ to be entailed via the TGD `$\mi{dummy\mhyphen{}yes}(X) \land \mi{aux}(X)  \ra \mi{check}(X)$'.
\end{subproof}

\medbreak

We now prove that $\hat z$ is odd iff there is a $\leq$\nbdash-\Minex for $\Query$ including $\mi{odd}()$.
First, we show that, regardless of the actual value of $\hat z \geq 0$, there always is a single cardinality\nbdash-minimal explanation $\Expl$ for $\Query$.
Indeed, all explanations for $\Query$ must include all facts from $\bigcup_{i = 1}^{m} D^i \cup \set{\mi{succ}(i,i+1) \mid 1 \leq i \leq m-1} \cup \set{\mi{even}(i) \mid 2 \leq i \leq m \land i \text{ is even}}$.
Moreover, by \zcref{subprop_facts_in_explanations}, for every $\hat z \geq 0$, each $\leq$\nbdash-\Minex for $\Query$ does not include $\mi{dummy\mhyphen{}yes}(j)$ and $\mi{aux}(j)$, for $j \leq \hat z$, and includes $\mi{dummy\mhyphen{}yes}(j)$ and $\mi{aux}(j)$, for $j > \hat z$.
So, if there were more than one cardinality\nbdash-minimal explanation, there would be two $\leq$\nbdash-\Minexs differing only for the presence or not of $\mi{odd}()$ in the explanation, because $\mi{odd}()$ is the \emph{only} other fact besides those in $\DB^i$, the $\mi{succ}$-, $\mi{even}$-, $\mi{dummy\mhyphen{}yes}$-, and $\mi{aux}$\nbdash-facts.
However, this would contradict that both these sets are cardinality\nbdash-minimal explanations.

If $\hat z$ is odd, then $\hat z \geq 1$.
By \zcref{subprop_facts_in_explanations}, $\mi{dummy\mhyphen{}yes}(\hat z)$ and $\mi{aux}(\hat z)$ are not in $\Expl$, hence $\mi{check}(\hat z)$ is entailed via the TGD `$\mi{yes}(X) \land \mi{odd}() \ra \mi{check}(X)$'.
By this, $\mi{odd}()$ must be in the single $\leq$\nbdash-\Minex $\Expl$.

If $\hat z$ is even, there are two cases:
either $\hat z = 0$, or $\hat z \geq 2$.
When $\hat z = 0$, none of the $t_i$ is a \yesinst of \ExpTilProb, hence none of the $\mi{yes}$\nbdash-facts is entailed by the TGDs $\Sigma^i$.
Thus, all the $\mi{check}$\nbdash-facts needed to entail the query are obtained via the TGD `$\mi{dummy\mhyphen{}yes}(X) \land \mi{aux}(X) \ra \mi{check}(X)$', and not via the TGD `$\mi{yes}(X) \land \mi{odd}() \ra \mi{check}(X)$'.
By this, the only $\leq$\nbdash-\Minex $\Expl$ for $\Query$ does not contain $\mi{odd}()$.
On the other hand, when $\hat z \geq 2$, the entailment of the fact $\mi{check}(\hat z)$ is via the TGD `$\mi{yes}(X) \land \mi{even}(X) \ra \mi{check}(X)$', which does not require the presence of the fact $\mi{odd}()$ in $\Expl$.
\end{proof}

\resumetoc

\section{Some Basic Notions in Computational Complexity}
\label{sec_general_preliminaries}

\paragraph{Asymptotic notation.}\hspace{0pt}\newline
We use the asymptotic notation $O(\cdot)$, $o(\cdot)$, $\Omega(\cdot)$, and $\omega(\cdot)$, with its usual meaning (see, e.g., \cite{CormenLRS2022}).
The next remarks given for $O(\cdot)$ are extended to $o(\cdot)$, $\Omega(\cdot)$, and $\omega(\cdot)$.
Following~\cite{BalcazarDG1995}, $O(f(n))$ also denotes the \emph{set} of functions that are $O(f(n))$.
With ``$g(n) \in O(f(n))$'' we mean that $g(n)$ is~$O(f(n))$.
For a function $g \colon \NaturalsDomain \to \NaturalsDomain$, and a family $F$ of functions $f \colon \NaturalsDomain \to \NaturalsDomain$, $g(F)$ is the \emph{set} of functions obtained by composing $g$ with each function $f(n)$ in $F$, i.e., $g(F) = \bigcup_{f(n) \in F} \set{g(f(n))}$.
For a family $F$ of functions, we let $O(F) = \bigcup_{f(n) \in F} O(f(n))$.

\paragraph{Alphabets, strings, and languages.}\hspace{0pt}\newline
In this paper, we consider languages over the \defin{binary} alphabet $\alphabet = \set{0,1}$.
We denote by $\StringUniverse$ the set of all strings over $\alphabet$, including the empty string $\emptystring$.
For a string $w \in \StringUniverse$, $\StringLength{w}$ is $w$'s \defin{length}.
For an integer $n \geq 0$, $\alphabet^{n}$ and $\alphabet^{\leq n}$ are the sets of strings over $\alphabet$ of length \emph{exactly} $n$ and \emph{at most} $n$, respectively.
The sets $\alphabet^{\geq n}$, $\alphabet^{< n}$, and $\alphabet^{> n}$ are defined in the natural way.
A language $\Language{L}$ over an alphabet $\alphabet$ is a subset of $\StringUniverse$.
For a language $\Language{L} \subseteq \StringUniverse$ and a string $w$, $\Language{L}(w)$ is the Boolean predicate associated with $\Language{L}$ such that $\Language{L}(w) = 1$ if $w \in \Language{L}$, i.e., if $w$ is a \yesinst of $\Language{L}$, and $\Language{L}(w) = 0$ if $w \notin \Language{L}$, i.e., if $w$ is a \noinst of $\Language{L}$.
For a language $\Language{L}$, the \defin{(language) complement to $\Language{L}$} is $\compl{\Language{L}} = \StringUniverse \setminus \Language{L}$.
For an integer $k \geq 2$, a \defin{$k$\nbdash-ary relation $\Relation{R}$ over $\StringUniverse$} is a subset of the $k$\nbdash-fold cartesian product of $\StringUniverse$.
Since a suitable encoding can represent every tuple $\StringTup{w} = \tup{w_1,\dots,w_n}$ of strings with just one string $y_{\StringTup{w}}$, a language $\Language{L}_{\Relation{R}}$ can be associated with a relation $\Relation{R}$ such that $\StringTup{w} \in \Relation{R} \Leftrightarrow y_{\StringTup{w}} \in \Language{L}_{\Relation{R}}$.
By this, we regard relations as languages;
e.g.,~we will say that a tuple $\StringTup{w} = \tup{w_1,\dots,w_k}$ of strings is a \yesinst or a \noinst of $\Relation{R}$, and that $\Relation{R}(w_1,\dots,w_k)$ is a predicate whose value is $1/\valtrue$ or $0/\valfalse$, respectively.

\paragraph{Turing machines and their computations.}\hspace{0pt}\newline
Unless differently stated, we consider multi-tape Turing machines with bidirectional\footnote{Bidirectional here means that the head of the tapes can move to the left and to the right.} semi\nbdash-infinite tapes with a (single) \emph{read-only} input tape and (possibly several) \emph{read/write} work tapes.
We here consider only \emph{recursive decision} problems, hence our machines do not have an output tape and accept, or reject, their input string by \emph{halting} in an accepting, or non\nbdash-accepting, state, respectively.
A Turing machine can be represented via a binary string through a suitable encoding of its transition function (see, e.g., \cite{Hopcroft1979}).

Remember that the computation of a deterministic machine is always characterized by a single next step at any given moment, whereas the computation of a \emph{non}\/deterministic machine may have several possible next steps at some points.
In this paper, all machines are deterministic, unless otherwise specified.
The computation that a deterministic (resp., a nondeterministic) machine~$\Machine{M}$ performs (resp., may perform) over an input string~$w$ can be described by the sequence of configurations that~$\Machine{M}$ traverses (resp., may traverse) when executing on~$w$.

A \defin{configuration}, or \defin{instantaneous description (ID)}, of a (non)deterministic machine $\Machine{M}$ is a comprehensive snapshot of $\Machine{M}$'s execution at a particular moment.
An ID of $\Machine{M}$ includes the current (control) state, the current content of all tapes, and the current positions of all heads on the tapes.
From an ID and the knowledge of the machine's transition function, for a deterministic (resp., nondeterministic) machine $\Machine{M}$ it is possible to know what the next action performed by $\Machine{M}$ is (resp., what the possible next actions available to $\Machine{M}$ are).

The IDs that a \emph{non}\/deterministic machine $\Machine{M}$ may traverse when executing on input~$w$ can be arranged into a \defin{computation tree}.
The starting ID of $\Machine{M}$ on $w$ is the root of this tree, and its edges go from an ID (node) $\alpha$ to an ID $\beta$ iff $\beta$ is a legal next configuration of $\alpha$, according to $\Machine{M}$'s transition function.
A \defin{partial computation} of $\Machine{M}$ on~$w$, or for $\Machine{M}(w)$, is a sequence of IDs forming a contiguous path within the computation tree of $\Machine{M}$ on~$w$.
A \defin{computation} for $\Machine{M}(w)$ is one of its partial computations starting from the root and ending in a leaf of the computation tree of $\Machine{M}$ on~$w$.
A computation is \defin{accepting} or \defin{rejecting} if its last ID contains an accepting or non-accepting state, respectively.
A machine $\Machine{M}$ accepts an input $w$ iff there exists an accepting computation in the computation tree of $\Machine{M}$ on~$w$.
Computation trees of deterministic machines are actually lists.
Since a (partial) computation $\pi$ for $\Machine{M}(w)$ is a sequence of IDs, $\pi$ can be represented as a string over the binary alphabet via a suitable encoding;
such an encoding is characterized by a linear overhead only (see, e.g., \cite{Hopcroft1979}).

For a (non)deterministic machine $\Machine{M}$, we denote by $\Machine{M}(w)$ the \defin{output of the machine $\Machine{M}$ on input $w$}, where $\Machine{M}(w) = 1$, if $\Machine{M}$ accepts $w$, otherwise $\Machine{M}(w) = 0$;
we denote by $\LanguageOf{\Machine{M}}$ the \defin{language decided by the machine $\Machine{M}$};
we say that $\Machine{M}$ decides/solves a language/problem $\Language{L}$ iff $\Language{L} = \LanguageOf{\Machine{M}}$.

\paragraph{Resource-bounded computations.}\hspace{0pt}\newline
The \defin{computation time} (resp., \defin{computation space}) \defin{of a machine $\Machine{M}$ on input $w$} is the length of the longest computation (resp., is the maximum number, on all $\Machine{M}$'s \emph{work} tapes, of distinct cells scanned in any computation) in the computation tree of $\Machine{M}$ over $w$.
For a strictly positive~\cite{BalcazarDG1995} nondecreasing~\cite{LewisSH1965,HartmanisS1965,StearnsHL1965} function $t\colon \NaturalsDomain \to \NaturalsDomain$ (resp., $s\colon \NaturalsDomain\to\NaturalsDomain$), the machine $\Machine{M}$ has \defin{running time} $t(n)$ (resp., \defin{running space} $s(n)$) iff, on all but finitely-many inputs $w$, the computation time (resp., space) of $\Machine{M}$ over $w$ does not exceed $t(\StringLength{w})$ (resp., $s(\StringLength{w})$)~\cite{Kozen2006};
$t(n)$ (resp., $s(n)$) is also named \defin{time function} (resp., \defin{space function}).
A time function $t(n)$ (resp., space function $s(n)$) is \defin{time constructible} (resp., \defin{space constructible}) iff there is a Turing machine $\Machine{M}$ which, on every input of \emph{length} $n$, halts in exactly $t(n)$ steps (resp., marks off $s(n)$ work tape cells and halts, never using more than $s(n)$ space in the process)~\cite{HartmanisS1965,StearnsHL1965,BalcazarDG1995,Kozen2006}.
A different characterization of time and space constructibility can be proven equivalent to the one above.
A time function $t(n)$ such that, for all but finitely many values $n$, $t(n) \geq (1{+}\epsilon) \cdot n$, where $\epsilon > 0$, (resp., a space function $s(n)$) is time (resp., space) constructible iff there is a Turing machine that on an input of \emph{length} $n$ outputs the binary\nbdash-represented value $t(n)$ (resp., $s(n)$) in time $O(t(n))$ (resp., in space $O(s(n))$)---see, e.g., \cite{Kannan1982,Kobayashi1985,BalcazarDG1995}.
A machine \defin{$\Machine{M}$ decides a language $\Language{L}$ in time (resp., space) $f(n)$} iff $\Language{L} = \LanguageOf{\Machine{M}}$ and $\Machine{M}$ has running time (resp., space) $f(n)$.

\paragraph{Complexity classes.}\hspace{0pt}\newline
A \defin{complexity class} is a set of \emph{languages} that can be decided by machines of a specific sort (i.e., either deterministic or nondeterministic) within a given bound of computational resources.
Classically, computational resources characterizing complexity classes are computation time~\cite{HartmanisS1965}, work space~\cite{StearnsHL1965}, and, as we will see below, the possibility to have access a computation oracle.
For a complexity class $\ComplexityClass{C}$, \ComplementPrefix$\ComplexityClass{C}$ is the class of languages whose complements are in~$\ComplexityClass{C}$.
With a slight abuse of terminology, we say that a machine $\Machine{M}$ belongs to a complexity class $\ComplexityClass{C}$ if $\Machine{M}$ is of the sort and uses the amount of computational resources characterizing~$\ComplexityClass{C}$.

By $\DTime{t(n)}$ and $\DSpace{s(n)}$ (resp., $\NTime{t(n)}$ and $\NSpace{s(n)}$) we denote the class of languages decided by deterministic (resp., nondeterministic) machines in time $t(n)$ and space $s(n)$, respectively.
The functions $t(n)$ and $s(n)$ are generally assumed to be time and space constructible, respectively, as the complexity classes defined via these kinds of functions enjoy some desirable properties (see, e.g., \cite{Hopcroft1979,BalcazarDG1995}).
For a family of functions $F$, we let $\DTime{F} = \bigcup_{f(n) \in F} \DTime{f(n)}$;
a similar notation is also defined for $\mathrm{NTIME}$, $\mathrm{DSPACE}$, and $\mathrm{NSPACE}$.
For every constant $c \geq 1$, it clearly holds that $\DTime{f(n)} \subseteq \DTime{c \cdot \! f(n)}$, and for $\mathrm{NTIME}$, $\mathrm{DSPACE}$, and $\mathrm{NSPACE}$ as well.
Interestingly, the linear speed-up~\cite{HartmanisS1965} and space compression~\cite{StearnsHL1965} theorems show that, for each constant $c \geq 1$, if $t(n)$ is a time function such that $t(n) \geq n + 1$, and $s(n)$ is a space function such that $s(n) \geq 1$, we have $\DTime{c \cdot t(n)} \subseteq \DTime{t(n)}$, $\NTime{c \cdot t(n)} \subseteq \NTime{t(n)}$, $\DSpace{c \cdot s(n)} \subseteq \DSpace{s(n)}$, and $\NSpace{c \cdot s(n)} \subseteq \DSpace{s(n)}$.%
\footnote{More precisely, $\DTime{c \cdot t(n)} \subseteq \DTime{t(n)}$ and $\NTime{c \cdot t(n)} \subseteq \NTime{t(n)}$ hold when $t(n) \in \omega(n)$.
When $t(n) = c \cdot n$, it holds that, for every $\epsilon > 0$, $\DTime{t(n)} \subseteq \DTime{(1 {+} \epsilon) \cdot n}$ and $\NTime{t(n)} \subseteq \NTime{(1 {+} \epsilon) \cdot n}$, which intuitively means that the constant $c$ can be made arbitrarily close to~$1$. See, e.g., \cite{HartmanisS1965,StearnsHL1965,Hopcroft1979,Papadimitriou1994,BalcazarDG1995,Kozen2006}.}
This implies that $\DTime{f(n)} = \DTime{O(f(n))}$, that is, constant factors can be ignored when analyzing machines' running times.
For example, if $\Language{L}$ is a language in $\DTime{O(2^{n^k})}$, there is a machine deciding $\Language{L}$ in time $2^{n^k}$.
Similar remarks holds for $\mathrm{NTIME}$, $\mathrm{DSPACE}$, and $\mathrm{NSPACE}$ as well.

The classes \PTime and \ExpTime (resp., \NPTime and \NExpTime) contain the languages decided by deterministic (resp., nondeterministic) machines in polynomial and exponential time, respectively.
The classes \LogSpace, \PSpace, and \ExpSpace, contain the languages decided by deterministic machines in logarithmic, polynomial, and exponential, space, respectively.
More precise definitions of these, and other, classes are given in \zcref{sec_prelim_central_complexity_classes}.

When we will say that a language, problem, relation, or predicate, is a polynomial(\nbdash-time) one, i.e., without mentioning whether it is deterministic or not, we mean that it can be decided by a \emph{deterministic} polynomial-time machine.
We will explicitly mention when we refer to \emph{non}\/deterministic classes.

\paragraph{Oracle Turing machines.}\hspace{0pt}\newline
Intuitively, oracles are subroutines that can be invoked by machines.
An \defin{oracle Turing machine $\Oracle{\Machine{M}}{?}$}, is a (non)deterministic Turing machine $\Machine{M}$ that, during its computation, may ask to an oracle to answer membership queries at \emph{unit} cost;
i.e., $\Oracle{\Machine{M}}{?}$ may ask whether some strings belong or not to the oracle's language.
The definition of $\Oracle{\Machine{M}}{?}$ is \emph{independent} from its oracle, and the symbol ``$?$'' indicates that oracles for different languages can be ``attached'' to $\Machine{M}$~\cite{Papadimitriou1994}.
For a language $\Language{A}$, by $\Oracle{\Machine{M}}{\Language{A}}$ we mean that the oracle attached to $\Oracle{\Machine{M}}{?}$ decides~$\Language{A}$.
For an oracle machine $\Oracle{\Machine{M}}{?}$ and a language $\Language{A}$, $\Oracle{\Machine{M}}{\Language{A}}(w)$ denotes the output of the oracle machine on input $w$ when the oracle decides $\Language{A}$;
$\LanguageOf{\Oracle{\Machine{M}}{\Language{A}}}$ denotes the language decided by the oracle machine $\Oracle{\Machine{M}}{?}$ with an oracle for~$\Language{A}$.

More specifically, an oracle machine $\Oracle{\Machine{M}}{?}$ is equipped with an additional \emph{write-only} \emph{unidirectional}\footnote{Unidirectional here means that the head of this tape can only move to the right.}
work tape, called the \defin{query tape}, and has three specific states, $\querystate$ (the \defin{query state}), and $\yesanswerstate$ and $\noanswerstate$ (the \defin{answer states}), to interact with the oracle.
The computation of an oracle machine $\Oracle{\Machine{M}}{?}$ proceeds like in an Turing machine, except for the transitions from the query state $\querystate$.
Once entering in $\querystate$, thanks to the oracle's ``counsel'', $\Oracle{\Machine{M}}{?}$ moves to either $\yesanswerstate$ or $\noanswerstate$ depending on whether the current string on the query tape belongs to the oracle's language or not.
We assume that, when $\Oracle{\Machine{M}}{?}$ moves to one of the answer states, at the same time the content of the query tape is deleted.
The answer state which $\Oracle{\Machine{M}}{?}$ has moved to allows $\Oracle{\Machine{M}}{?}$ to use the oracle's answer in its subsequent computation.
In this paper, to streamline the notation, if it is clear from the context that a machine $\Oracle{\Machine{M}}{?}$ is actually an oracle Turing machine, we may refer to it just by~$\Machine{M}$.
Since we will always consider oracles for recursive languages, there will always be machines for such languages.
Hence, for presentation purposes, we could sometimes refer to the oracle as to an additional machine which the oracle machine (i.e., the caller) can ask questions to.
For this reason, we might sometimes refer to the computation carried out by the oracle;
nonetheless, the time cost paid by the caller for the computation carried out by the oracle is always \emph{one} step.

Like for non\nbdash-oracle machines, oracle machine computations can be described via sequences of IDs.
Oracle machine IDs are very similar to the IDs of non-oracle machines:
the query tape is simply an additional tape, whose content and head's position are part of the ID.
Notice that this makes the oracle machine IDs here considered different from those sometimes considered in the literature, where the query tape is \emph{not} included in the IDs (see, e.g., \cite{LadnerL76,Hemaspaandra1994}).
Also the IDs that a (non)deterministic oracle machine $\Oracle{\Machine{M}}{?}$ with oracle $\Language{A}$ may traverse when executing on input~$w$ can be arranged into computation trees.
These computation trees are a natural extension of those for non-oracle machines:
nodes are again associated with the possible IDs for $\Oracle{\Machine{M}}{?}$ when executing on $w$, and the edges connect IDs accounting for legal next configurations, both due to $M$'s own transition function and due to the oracle's answer(s) for the specific language~$\Language{A}$.

\paragraph{Bounded query oracles.}\hspace{0pt}\newline
For a strictly positive nondecreasing function $f\colon \NaturalsDomain\to\NaturalsDomain$ (resp., for a constant integer $k$), $\BoundedOracle{\Machine{M}}{?}{f(n)}$ (resp., $\BoundedOracle{\Machine{M}}{?}{k}$) denotes an oracle machine allowed to issue at most $f(n)$ (resp., at most $k$) queries to its oracle, where $n$ is the input string size.
The notation $\ParOracle{\Machine{M}}{?}$ state that the oracle machine is allowed to issue a \emph{single} round of \defin{parallel queries}, i.e.,
the queries that $\ParOracle{\Machine{M}}{?}$ submits to its oracle are collected and asked all at once.
This way of querying the oracle is also called \defin{nonadaptive}, in contrast with the standard sequential way of asking queries, which instead can be \defin{adaptive}~\cite{Book1988}.
Oracle machines issuing parallel queries have an additional answer tape where they receive the answers from their oracles.
When present, this tape, together with its head position, are part of the IDs.
For parallel queries, we might relax the constraint on the \emph{single} round of parallel queries;
for a function $f\colon \NaturalsDomain\to\NaturalsDomain$ (resp., a constant integer $k$), $\ParBoundedOracle{\Machine{M}}{?}{f(n)}$ (resp., $\ParBoundedOracle{\Machine{M}}{?}{k}$) denotes an oracle machine allowed to issue at most $f(n)$ (resp., at most $k$) rounds of parallel queries to its oracle, where $n$ is the input string size.
The constraints can also be combined:
for two strictly positive nondecreasing functions $f(n)$ and $g(n)$, $\DoubleBoundedParOracle{\Machine{M}}{?}{f(n)}{g(n)}$ denotes an oracle machine allowed to ask $g(n)$ rounds of parallel queries, and in each round at most $f(n)$ queries can be asked;
the notation $\BoundedParOracle{\Machine{M}}{?}{f(n)}$ has the natural meaning, and the functions $f(n)$ and $g(n)$ can be replaced in the notation by constant values, with the natural meanings.

\paragraph{Oracle complexity classes.}\hspace{0pt}\newline
If $\ComplexityClass{C}$ is a (non)deterministic \emph{time}\nbdash-bounded complexity class and $\Language{A}$ is a language, $\ComplexityClass{C}^\Language{A}$ is the class of languages decided by oracle machines in $\ComplexityClass{C}$ querying an oracle for $\Language{A}$;
for a complexity class $\ComplexityClass{D}$, we denote by $\ComplexityClass{C}^\ComplexityClass{D}$ the class of languages decided by oracle machines in $\ComplexityClass{C}$ querying an oracle for a language in~$\ComplexityClass{D}$.
In the following, when we say that an oracle machine $\Oracle{\Machine{M}}{?}$ queries an oracle in~$\ComplexityClass{D}$, 
we mean that $\Machine{M}$ queries an oracle for a $\ComplexityClass{D}$\CompleteSuffix language;
such an oracle is capable of deciding any language in~$\ComplexityClass{D}$ (see the paragraph on reductions and hardness below).
The notation introduced to denote constraints over the queries is naturally extended to oracle complexity classes, e.g., $\DoubleBoundedParOracle{\ComplexityClass{C}}{\ComplexityClass{D}}{f(n)}{g(n)}$.
For a family $F$ of functions, $\BoundedOracle{\ComplexityClass{C}}{\ComplexityClass{D}}{F}$ (resp., $\BoundedParOracle{\ComplexityClass{C}}{\ComplexityClass{D}}{F}$) is the class of languages decided by oracle machines in $\ComplexityClass{C}$ querying an oracle in $\ComplexityClass{D}$ at most $f(n)$\nbdash-many times (resp., with at most $f(n)$\nbdash-many parallel queries), with $f(n) \in F$;
more formally, $\BoundedOracle{\ComplexityClass{C}}{\ComplexityClass{D}}{F} = \bigcup_{f(n) \in F} \BoundedOracle{\ComplexityClass{C}}{\ComplexityClass{D}}{f(n)}$ and $\BoundedParOracle{\ComplexityClass{C}}{\ComplexityClass{D}}{F} = \bigcup_{f(n) \in F} \BoundedParOracle{\ComplexityClass{C}}{\ComplexityClass{D}}{f(n)}$.

Based on these definitions of oracle complexity classes, hierarchies of complexity classes can also be defined.
The classical \PHText and \WEHStressedText are introduced in \zcref{sec_prelim_PH_ExpH}.

\paragraph{Space-bounded oracle Turing machines.}\hspace{0pt}\newline
Defining \emph{space}-bounded oracle machines is not straightforward~\cite{Hartmanis1988}, as the space used on the query tape might, or might not, contribute toward the computation space of the oracle machine.
Three major definitions were proposed:
the bounded query model~\cite{Simon1977,Book1979}, the unrestricted query model~\cite{LadnerL76}, and the deterministic query model~\cite{RuzzoST84} (see also \cite{JBuss1988,Hartmanis1988,Michel1992} for additional comments and references).

In the \defin{bounded query model}, defined by \citet{Simon1977} and \citet{Book1979}, the space used on the query tape is counted toward the oracle machine's computation space.
An issue of this definition is that there are languages $\Language{A}$ for which $\Language{A} \notin \Oracle{\LogSpace}{\Language{A}}$ (i.e., $\Language{A}$ is not logspace Turing reducible to itself---see below the definition of Turing reduction);
intuitively, this happens because all queries issued by the logspace-bounded oracle machine must be of logarithmic size, and hence the whole input string cannot be passed to the oracle, if needed.

In the \defin{unrestricted query model}, defined by \citet{LadnerL76}, the space used on the query tape is \emph{not} counted toward the computation space.
In this model, an oracle machine $\Oracle{M}{?}$ with an $O(f(n))$ space-bound (on the work tapes) can ask queries of size $2^{O({f(n)})}$, because such a machine can run for $2^{O({f(n)})}$ time.
Moreover, if such an oracle machine were nondeterministic, it would be able to overall generate $2^{2^{O(f(n))}}$ different queries throughout the entire computation tree of all its possible computations.
\Citet*{RuzzoST84} suggested that this might be too much for a space-bounded oracle machine, and could explain the unexpected behavior of nondeterministic logspace-bounded oracle machines recounted by \citet{LadnerL76}.

To overcome some of these non-intuitive consequences of the unrestricted query model, the \defin{deterministic query model} was introduced by \citet*{RuzzoST84} as an intermediate model.
In the deterministic query model, the space used on the query tape is \emph{not} counted toward the computation space, as in the unrestricted model, however, it is required that the oracle machine acts \emph{deterministically} while writing on the query tape.
More specifically, the oracle machine, for every query asked, since when it writes the first query's symbol on the tape, until the moment in which the query is actually submitted to the oracle, must carry out a deterministic computation (even if the oracle machine were nondeterministic).
In this way, a space-bounded oracle machine $\Oracle{M}{?}$ can still ask questions that are exponentially-longer than $\Oracle{M}{?}$'s space-bound, however the number of different queries in the computation tree of a nondeterministic oracle machine is not any more double-exponential.%
\footnote{Even the deterministic query model has some non-desirable consequences. For this reason, \citet{JBuss1988} proposed an additional query model for space-bounded oracle machines.}

\paragraph{Alternating Turing machines.}\hspace{0pt}\newline
\defin{Alternating Turing machines}~\cite{ChandraKS81} are a kind of nondeterministic machines whose (control) states are partitioned in two types:
the \emph{existential} and the \emph{universal} states.
Computations of alternating machines are defined via computation trees as for standard nondeterministic machines;
what changes is the acceptance condition, which is defined according to the following rules.
Let $\Machine{M}$ be an alternating machine, let~$w$ be a string, and let $\mathcal{T}$ be the computation tree of $\Machine{M}$ on~$w$.
An ID of $\Machine{M}$ is existential or universal iff its state is existential or universal, respectively.
The leaf IDs of $\mathcal{T}$ are (labelled as) accepting or rejecting iff the state in the ID is accepting or not, respectively.
All other IDs $\alpha$ of $\mathcal{T}$ are inductively labelled as accepting or rejecting as follows:
\begin{itemize}[nosep,label=--]
  \item if $\alpha$ is existential, $\alpha$ is accepting iff at least one of $\alpha$'s successor IDs in $\mathcal{T}$ is accepting; and
  \item if $\alpha$ is universal, $\alpha$ is accepting iff all $\alpha$'s successor IDs in $\mathcal{T}$ are accepting.
\end{itemize}
Given these rules, $\Machine{M}$ accepts~$w$ iff the root of $\mathcal{T}$ is labelled as accepting.
The initial state of $\Machine{M}$ is \emph{not} required to be existential.
A (standard) nondeterministic machine is an alternating machine whose states are all existential.

Running time and computation space for alternating machines are defined as for standard nondeterministic machines.
We denote by $\ATime{f(n)}$ (resp., $\ASpace{f(n)}$) the class of languages that can be decided by alternating machines in time (resp., in space) $O(f(n))$.
Interestingly, another ``computational resource'' can be bounded on alternating machines, which is the number of allowed alternations between existential and universal states during the computations.
More precisely, the number of alternations that an alternating machine $\Machine{M}$ performs in a computation on~$w$ is the number of transitions between existential and universal states during that computation \emph{plus one}---we add one to account for the first state type that the machine starts with.
An alternating machine $\Machine{M}$ runs with at most $k$ alternations iff, on \emph{every} string~$w$, \emph{all} the computations of~$\Machine{M}$ on~$w$, in the computation tree of $\Machine{M}$ on~$w$, have at most $k$ alternations.
When the number of allowed alternations is fixed, starting in an existential or in a universal state can make a difference.
If $\Machine{M}$ is an alternating machine which runs with at most $k$ alternations and starts in an existential (resp., a universal) state, we say that $\Machine{M}$ is a $\Sigma_k$\nbdash-alternating (resp., a $\Pi_k$\nbdash-alternating) machine.
We denote by $\BoundedExATime{k}{f(n)}$ (resp., $\BoundedUnATime{k}{f(n)}$) the class of languages that can be decided by $\Sigma_k$\nbdash-alternating (resp., by $\Pi_k$\nbdash-alternating) machines in time $O(f(n))$.

\paragraph{Reductions.}\hspace{0pt}\newline
The concept of reduction is an interesting tool introduced to study the relative complexity between problems.
Various notions of reduction have been introduced~\cite{Ladner1975a,Ladner1975}.

Let $\Language{A}$ and $\Language{B}$ be two languages, $\Language{A}$ is (\defin{many-one}, or \defin{Karp}) \defin{reducible} to $\Language{B}$, denoted by $\Language{A} \KarpRed[] \Language{B}$, iff there exists a \emph{computable} function $f\colon \StringUniverse \to \StringUniverse$ such that, for every string $w$, it holds that $w \in \Language{A} \Leftrightarrow f(w) \in \Language{B}$~\cite{Karp1972}.
The function $f$ is said to be a (\defin{many-one}, or \defin{Karp})
\defin{reduction} from $\Language{A}$ to $\Language{B}$, and is polynomial(-time) if $f$ is computable in polynomial time;
polynomial (many-one/Karp)
reductions are denoted by $\KarpRed[p]$.
For a complexity class $\ComplexityClass{C}$ such that $\ComplexityClass{C} \supseteq \NPTime$, or $\ComplexityClass{C} \supseteq \CoNPTime$,
a language $\Language{L}$ is \defin{$\ComplexityClass{C}$-hard} iff, for all languages $\Language{L}' \in \ComplexityClass{C}$, it holds that $\Language{L}' \KarpRed \Language{L}$.
A language $\Language{L}$ is \defin{$\ComplexityClass{C}$-complete} iff $\Language{L} \in \ComplexityClass{C}$ and $\Language{L}$ is $\ComplexityClass{C}$-hard.
By this, a Turing machine deciding a $\ComplexityClass{C}$-complete language can decide, modulo a prior polynomial re-encoding of the input strings, every language in~$\ComplexityClass{C}$.

Let $\Language{A}$ and $\Language{B}$ be two languages, $\Language{A}$ is \defin{Turing-reducible} to $\Language{B}$, denoted by $\Language{A} \TuringRed[] \Language{B}$, iff there exists an oracle machine $\Oracle{\Machine{M}}{?}$ deciding $\Language{A}$ by having access to an oracle for $\Language{B}$, i.e., $\Language{A} = \LanguageOf{\Oracle{\Machine{M}}{\Language{B}}}$.
This essentially means that there exists an algorithm deciding $\Language{A}$ by means of (multiple) calls to a subroutine for $\Language{B}$.
This algorithm is said to be a \emph{Turing reduction} from $\Language{A}$ to $\Language{B}$.
A Turing reduction is polynomial(-time) if the running time of $\Oracle{\Machine{M}}{?}$ is bounded by a polynomial;
this is denoted by $\TuringRed[p]$.
Polynomial Turing reductions are also called \defin{Cook reductions} (defined in~\cite{Cook1971}).
Oracle machines characterizing Turing reductions can be subject to the constraints seen above over the issued queries;
i.e., bounding the number of allowed queries and/or sequential vs.\ parallel queries.

Let $\Language{A}$ and $\Language{B}$ be two languages, $\Language{A}$ is \defin{truth-table reducible} to $\Language{B}$, denoted by $\Language{A} \TTRed[] \Language{B}$, iff there exists a \emph{computable} function $f\colon \StringUniverse \to \StringUniverse$ such that, for every string $w$, it holds that the evaluation of $f$ over $w$ outputs a tuple of strings $\tup{\Circuit{C}_w,y_1,\dots,y_{k(w)}}$, where $\Circuit{C}_w$ is (the string representing) a Boolean circuit with $k(w)$ input gates (where $k(w)$ depends on $w$), and $y_1, \dots, y_{k(w)}$ are strings, for which
$w \in \Language{A} \Leftrightarrow \CirctuitValue{\Circuit{C}_w}{\Language{B}(y_1),\dots,\Language{B}(y_{k(w)})} = 1$,
where $\CirctuitValue{\Circuit{C}}{b_1,\dots,b_n}$ denotes the output value of $\Circuit{C}$ when receiving in input the Boolean values $b_1, \dots, b_n$.
The function $f$ is said to be a \defin{truth-table reduction}, or \defin{tt-reduction}, from $\Language{A}$ to $\Language{B}$, and is polynomial(-time) if $f$ is computable in polynomial time;
polynomial tt-reductions are denoted by $\TTRed$.

Polynomial tt-reductions are known to be equivalent to polynomial Turing reductions with the additional constraint that the oracle machine can perform only a single round of parallel queries to its oracle~\cite{Ladner1975}.
Hence, for two languages $\Language{A}$ and $\Language{B}$, it holds that $\Language{A} \TTRed \Language{B}$ iff there exists a polynomial-time oracle machine $\ParOracle{\Machine{M}}{?}$ such that $\Language{A} = \LanguageOf{\ParOracle{\Machine{M}}{B}}$.
Therefore, if we denote by ${\TTRed[p]}(\ComplexityClass{C})$ the class of languages that can be polynomially tt-reduced to a language in the complexity class~$\ComplexityClass{C}$, it holds that the complexity class $\TTRedCLASS{\ComplexityClass{C}}$ equals~$\ParOracle{\PTime}{\ComplexityClass{C}}$.
 
\end{appendices}

\endgroup

\emergencystretch=3em

\urlstyle{same}

\printbibliography[heading=bibintoc]

\end{document}